\documentclass[mnsc,nonblindrev]{informs4rad} 

\usepackage{diagbox}
\newcommand{\revcolorm}[1]{{\color{black}#1}}

\newif\ifMS
\MStrue

\OneAndAHalfSpacedXI

\newif\iffigdraft

\figdraftfalse

\TheoremsNumberedThrough 
\EquationsNumberedThrough

\usepackage{xfrac}
\usepackage{xcolor}
\usepackage{bbm}
\usepackage[normalem]{ulem}
\usepackage{wasysym}
\usepackage{subcaption}
\usepackage{wrapfig}
\usepackage[english]{babel}
\usepackage[skip=2pt]{caption}
\usepackage{thm-restate}
\usepackage[ruled,vlined,linesnumbered]{algorithm2e}
\SetNoFillComment

\usepackage{scalefnt}

\usepackage[breakable,skins]{tcolorbox}
\newtcolorbox{myboxx}[3][]
{
  colframe = white,
  colback  = #2!10,
  left=0pt,
  right=0pt,
  top=0pt,
  bottom=0pt,
  enlarge left by=0mm,
  boxsep=10pt,
  arc=0pt,outer arc=0pt,
  #1,
}

\colorlet{shadecolor}{gray!50}

\usepackage{csquotes}

\makeatletter
\renewenvironment*{displayquote}
  {\begingroup\setlength{\leftmargini}{0.1cm}\csq@getcargs{\csq@bdquote{}{}}}
  {\csq@edquote\endgroup}
\makeatother

\makeatletter
\renewenvironment*{displayquote}
  {\begingroup\setlength{\leftmargini}{0cm}\csq@getcargs{\csq@bdquote{}{}}}
  {\csq@edquote\endgroup}
\makeatother

\DeclareRobustCommand{\mybox}[2][gray!15]{
\begin{tcolorbox}[ 
        colback=white,      
        colframe=gray,  
        boxrule=0.2pt,      
        arc=2pt,outer arc=2pt,
        left=12pt,
        right=12pt,
        top=5pt,
        bottom=5pt,
        width=1.07\linewidth,
        enlarge left by=-0.55cm,
        before upper=\renewcommand{\baselinestretch}{1.3}\selectfont,
        after upper=\normalfont
        ]
 #2
 \end{tcolorbox}
}

\usepackage{mathtools}

\definecolor{cornellred}{rgb}{0.7, 0.11, 0.11}
\definecolor{maroon}{rgb}{0.52, 0, 0}
\definecolor{dgreen}{rgb}{0.0, 0.5, 0.0}
\definecolor{ballblue}{rgb}{0.13, 0.67, 0.8}
\definecolor{royalblue(web)}{rgb}{0.25, 0.41, 0.88}
\definecolor{bleudefrance}{rgb}{0.19, 0.55, 0.91}
\definecolor{royalazure}{rgb}{0.0, 0.22, 0.66}

\usepackage[hypertexnames=false]{hyperref}
\hypersetup{
	colorlinks = true,
	linkcolor=royalazure,
	urlcolor=royalazure,
	citecolor=royalazure,
	linkbordercolor = {white}
}

\usepackage{cleveref}

\usepackage{enumitem}
\usepackage{multirow}

\usepackage{pgf,tikz,pgfplots}
\pgfplotsset{compat=1.15}

\usetikzlibrary{arrows,calc}
\usepackage{pgfplots}
\usetikzlibrary{patterns}
\usepgfplotslibrary{fillbetween}
\usetikzlibrary{intersections}

\usetikzlibrary{arrows, decorations.markings}

\tikzstyle{vecArrow} = [thick, decoration={markings,mark=at position
	1 with {\arrow[semithick]{open triangle 60}}},
	double distance=1.4pt, shorten >= 5.5pt,
	preaction = {decorate},
postaction = {draw,line width=1.4pt, white,shorten >= 4.5pt}]
\tikzstyle{innerWhite} = [semithick, white,line width=1.4pt, shorten >= 4.5pt]

\ifMS
\else
\newtheorem{theorem}{Theorem}[section]
\newtheorem{lemma}[theorem]{Lemma}
\newtheorem*{lemma*}{Lemma}
\newtheorem{proposition}[theorem]{Proposition}

\newtheorem{claim}{Claim}
\newtheorem{remark}{Remark}
\theoremstyle{definition}
\newtheorem{definition}{Definition}[section]
\theoremstyle{definition}
	\newtheorem{assumption}[definition]{Assumption} 
		\newtheorem{example}[definition]{Example} 
\fi

	\usepackage{accents}

	\usepackage[authoryear,round]{natbib}

\usepackage[showdeletions]{color-edits}

\addauthor{VM}{blue}
\addauthor{MA}{brown}
\addauthor{RN}{maroon}

\newcommand*{\rom}[1]{\expandafter\romannumeral #1}
\newcommand{\Rom}[1]{\uppercase\expandafter{\romannumeral #1\relax}}

	\newcommand{\reals}{\mathbb{R}}

\ifMS
\else
	\DeclareMathOperator{\argmax}{argmax}
	\DeclareMathOperator{\argmin}{argmin}
\fi

\newcommand{\prob}[2][]{\text{\bf Pr}\ifthenelse{\not\equal{}{#1}}{_{#1}}{}\!\left[{\def\givenn{\middle|}#2}\right]}
\newcommand{\expect}[2][]{\text{\bf E}\ifthenelse{\not\equal{}{#1}}{_{#1}}{}\!\left[{\def\givenn{\middle|}#2}\right]}

\newcommand{\indicator}[1]{{\mathbbm{1}\left\{ #1 \right\}}}

	\DeclarePairedDelimiterX{\set}[1]\{\}{#1}
	\let\Pr\relax
	\DeclarePairedDelimiterXPP{\Pr}[1]{\mathbb{P}}[]{}{#1}
	\DeclarePairedDelimiterXPP{\Ex}[1]{\mathbb{E}}[]{}{#1}

\newcolumntype{P}[1]{>{\centering\arraybackslash}c{#1}}

	\newcommand{\RN}[1]{\textcolor{maroon}{\bf (Rad: #1)}}  
	
	\newcommand{\MA}[1]{\textcolor{brown}{\bf (Mohammad-Reza: #1)}}  
	\newcommand{\medit}[1]{\textcolor{purple}{#1}}

\iffigdraft
\usepackage[allfiguresdraft]{draftfigure}
\fi

\newcommand{\MC}{\mathcal{G}}
\newcommand{\states}{\mathcal{S}}
\newcommand{\transition}{{A}}
\newcommand{\altnum}{n}
\newcommand{\reward}{{R}}
\newcommand{\vecreward}{\mathbf{R}}
\newcommand{\terminals}{\mathcal{T}}

\newcommand{\dimension}{d}
\newcommand{\barP}{H_p}
\newcommand{\alloc}{\mathbf{p}}
\newcommand{\FeasibleAlloc}{\mathcal{P}}
\newcommand{\NumAffine}{m_{a}}
\newcommand{\NumConvex}{m_{c}}
\newcommand{\AffineVector}{\boldsymbol{\theta}}
\newcommand{\AffineConstant}{b}

\newcommand{\LagVector}{\lambda}
\newcommand{\ConvexFun}{F}
\newcommand{\GradFun}{\nabla\ConvexFun}

\newcommand{\FairSet}{\mathcal{F}_{\textsc{cons}}}
\newcommand{\OptFairPolicy}{\policy^*_{\textsc{cons}}}
\newcommand{\OptFair}{\textrm{OPT}_{\textsc{cons}}}
\newcommand{\InRate}{\gamma_{\textsc{I}}}
\newcommand{\OutRatelambda}{\gamma_{\textsc{O},\lambda}}
\newcommand{\OutRatebeta}{\gamma_{\textsc{O},\beta}}

\newcommand{\InnerNum}{K_{\textsc{I}}}
\newcommand{\OuterNum}{K_{\textsc{O}}}
\newcommand{\DualConstAffine}{\lambda}
\newcommand{\DualConstAffineVec}{\boldsymbol{\DualConstAffine}}
\newcommand{\DualConstConvex}{\beta}
\newcommand{\DualConstConvexVec}{\boldsymbol{\DualConstConvex}}
\newcommand{\DualConvex}{\boldsymbol{\mu}}
\newcommand{\Lagrange}{\mathcal{L}_{\textsc{JMS-cons}}}
\newcommand{\barLagrange}{\overline{\mathcal{L}}_{\textsc{JMS-cons}}}

\newcommand{\AdjustedReward}{\widetilde{\vecreward}}

\newcommand{\DualLagrangeSelect}{\mathcal{G}_{\textsc{FS}}}

\newcommand{\Kap}{\kappa}
\newcommand{\CurrentIndex}{o}

\newcommand{\policy}{\pi}
\newcommand{\PolicySpace}{\Pi}
\newcommand{\select}{\mathbb{A}}
\newcommand{\inspect}{\mathbb{I}}
\newcommand{\reserve}{\sigma}
\newcommand{\values}{\mathcal{V}}

\newcommand{\ManSet}{\mathcal{X}}
\newcommand{\WomanSet}{\mathcal{Y}}
\newcommand{\option}{o}
\newcommand{\tiebreak}{\tau}
\newcommand{\slack}{\Delta}

\newcommand{\candidates}{\mathcal{C}}

\newcommand{\tierule}{\tau}

\renewcommand{\qed}{\hfill\Halmos}

\newcommand{\DualConvexkl}{\DualConvex^{(m,\ell)}}
\newcommand{\DualConstConvexk}{\DualConstConvex^{(m)}}
\newcommand{\DualConstAffinek}{\DualConstAffine^{(m)}}

\newcommand{\DualConvexTemp}{\boldsymbol{\omega}}
\newcommand{\DualConvexTempkl}{\DualConvexTemp^{(m,\ell)}}

\newcommand{\commentcolor}{blue}
\newcommand{\normtwo}[1]{\lVert#1\rVert_2}
\newcommand{\upperDualConstConvex}{H_{\beta}}
\newcommand{\upperDualConvex}{H_{\mu}}
\newcommand{\upperDualAffine}{H_{\lambda}}
\newcommand{\Baralloc}{\overline{\alloc}}
\newcommand{\BarDualConvex}{\overline{\DualConvex}}
\newcommand{\BarDualConstAffine}{\overline{\DualConstAffine}}
\newcommand{\BarDualConstConvex}{\overline{\DualConstConvex}}

\newcommand{\lowerp}{L_p}
\newcommand{\collapse}{\mathcal{C}}
\newcommand{\nonCollapsedStates}{\mathcal{NC}}

\newcommand{\numaffine}{m}
\newcommand{\origin}{\mathbf{0}}
\newcommand{\vertices}{{V}}
\newcommand{\cone}{C}
\newcommand{\dualcone}{\cone^*}
\newcommand{\polarcone}{\cone^\circ}
\newcommand{\argmaxCone}{U}
\newcommand{\ellipse}{E}

\renewenvironment{proof}[1]{%
  \Trivlist
  \item[\hskip\labelsep {{#1}}]\ignorespaces
}{%
  \endTrivlist\addvspace{0pt}
}

\usepackage[showdeletions]{color-edits}

\addauthor{VM}{purple}
\addauthor{RN}{red} 
\addauthor{MA}{orange} 

\newcommand{\revcolor}[1]{{\color{black}#1}}

\MANUSCRIPTNO{MS-DSC-2023-04079.R1}

\MANUSCRIPTNO{}

\begin{document}



\RUNAUTHOR{Aminian, Manshadi, Niazadeh}

\RUNTITLE{Markovian Search with Ex-Ante Constraints}

\TITLE{Markovian Search with Ex-Ante Constraints:\\ {\scalefont{0.8}Theory and Applications to Socially Aware Algorithmic Hiring}}

\ARTICLEAUTHORS{%
\AUTHOR{Mohammad Reza Aminian}
\AFF{The University of Chicago, Booth School of Business, Chicago, IL, \EMAIL{maminian@chicagobooth.edu}}
\AUTHOR{Vahideh Manshadi}
\AFF{Yale School of Management, New Haven, CT, \EMAIL{vahideh.manshadi@yale.edu }}
\AUTHOR{Rad Niazadeh}
\AFF{The University of Chicago, Booth School of Business, Chicago, IL, \EMAIL{rad.niazadeh@chicagobooth.edu}}
 } 

\newpage
\ABSTRACT{%
\revcolorm{
We study and develop an algorithmic framework to incorporate ``ex-ante'' constraints---that is, constraints on outcomes that hold only on average---into stateful sequential search problems with costly inspection. Our framework encompasses the classical Weitzman's Pandora's box~\citep{weitzman1979optimal} as well as its extensions to joint Markovian scheduling~\citep{dumitriu2003playing,gittins1979bandit}, which models richer processes such as multistage search with multiple layers of inspection.  Ex-ante constraints in this context are particularly motivated by social considerations in algorithmic hiring, where they can adjust outcome distributions to promote equity and access.  While most work in the algorithmic fairness literature in computer science and economics has focused on incorporating such constraints into machine learning tasks like classification and regression, much less attention has been given to operational problems such as sequential search, with their unique intricacies. Our work aims to bridge this gap. 
 
Building on the optimality of index-based policies in the unconstrained versions of these problems, we show that optimal policies under a single ex-ante constraint (e.g., demographic parity) retain an index-based structure but require (i) dual-based adjustments of the indices and (ii) randomization between two such adjustments via a ``tie-breaking rule,'' both easy to compute and economically interpretable. We then extend our results to handle multiple affine constraints by reducing the problem to a variant of the exact Carathéodory problem and providing a polynomial-time algorithm to construct an optimal randomized dual-adjusted index-based policy that satisfies all constraints simultaneously.  For general affine and convex constraints, we develop a primal-dual algorithm that randomizes over a polynomial number of dual-based adjustments, yielding a near-feasible, near-optimal policy. All of these results rely on the key observation that a suitable relaxation of the Lagrange dual function for these constrained problems admits index-based policies akin to those in the unconstrained setting. Finally, through a numerical study, we investigate the implications of imposing various socially aware ex-ante constraints, in particular the utilitarian loss (price of fairness), and examine whether these constraints achieve their intended socially desirable outcomes.}

}



\pagenumbering{gobble} 

\maketitle
\setcounter{page}{1}
\newpage
\pagenumbering{arabic}

\section{Introduction}
\label{sec:intro}


Decisions that significantly impact people's lives, such as employment and hiring, have historically exhibited discrimination against certain minority groups. For example, \citet{bertrand2004emily} found that applicants with African-American names received fewer interview callbacks than those with white-sounding names but otherwise identical resumes, highlighting deep-rooted explicit biases in this context. Such issues have contributed to persistent disparities in representation, with long-term negative economic and societal consequences~\citep{becker2010economics}. Despite efforts to address them, progress has been limited: a meta-analysis by \citet{quillian2017meta} shows that racial discrimination in hiring has seen little improvement in recent decades.

\revcolorm{
The rise of \emph{algorithmic hiring}---the practice of using data-driven algorithms for candidate search and selection---offers a promising way to break this pattern. These algorithms are not only faster and more efficient than human decision-makers, making them attractive alternatives to traditional hiring methods, but they are also more transparent and do not introduce explicit bias by design. However, the risk of implicit bias remains if their ``input''---whether prior data or screening instruments---is itself biased~\citep{liebkind2016ethnic,kleinberg2018algorithmic}. Marginalization can also persist regardless of whether decisions are made by humans or algorithms. For example, underprivileged individuals may lack access or financial means to apply for certain jobs, leading to their continued underrepresentation~\citep{gaddis2013influence,chetty2020income}. These concerns have fueled the development of \emph{socially aware} algorithmic tools aimed at reducing such disparities and promoting access and inclusion by adjusting algorithmic outcomes~\citep{garr2019diversity, kleinberg2018selection}.
}



\revcolorm{
A well-established (and effective) adjustment approach for achieving socially desirable outcomes in the presence of disparities is to guide an algorithm's decisions by imposing constraints on its outcome distribution. These constraints---typically enforced on average and referred to as \emph{ex-ante constraints}---can capture various notions of fairness and inclusion in decision-making. The subfield of algorithmic fairness in computer science and economics has explored several such socially aware criteria in core machine learning settings such as classification and regression, along with methods for enforcing these constraints and reasons on why they could be effective; see \citet{kleinberg2018algorithmic} for a detailed discussion. In summary, these constraints, roughly speaking, ensure that the outcomes generated by machine learning algorithms used in certain high-stake decisions (e.g., loan assignments) do not exhibit strong statistical evidences of discrimination.}

Similarly, employing companies may wish to incorporate inclusion-promoting or other socially aware interventions at various stages of their hiring processes by imposing certain ex-ante constraints to adjust the hiring outcome distribution. In the U.S., these efforts often focus on increasing opportunities for candidates from minority groups, whereas in other countries, practices closer to \emph{demographic parity} or \emph{quota}---two well-studied ex-ante constraints in the algorithmic fairness literature~\citep{kleinberg2018algorithmic}---are also common. 

However, modern hiring processes often involve overlooked operational intricacies that make them fundamentally different from standard classification or regression tasks. Consider the hiring of high-skilled workers, such as software engineers in tech companies, as a leading example: companies like Google employ well-defined multistage evaluation processes---often aided by algorithmic tools---to screen, interview, and ultimately hire candidates~\citep{bock2015work}. 
Given the \emph{cost} of these \emph{inspection} processes and the need to proactively search for the best talent, candidates are typically considered \emph{sequentially}, with the search process adapted based on the outcomes of intermediate steps.\footnote{This sequential aspect makes the process more efficient than a non-adaptive, batched approach and is especially relevant in scenarios with flexible hiring timelines, such as hiring software engineers throughout the year or promoting employees internally. For more details, see Chapter~4 of \cite{bock2015work}, ``Searching for the Best.'' Additionally, it is useful in cases where screening is resource-intensive, such as searching for a CEO, where the sequential search minimizes unnecessary evaluations~\citep{ryan2004attracting}.} Lastly, the search usually targets a limited number of positions, imposing a \emph{capacity} constraint on the number of hires. In all these cases, imposing such ex-ante constraints in complex hiring processes is fundamentally different from applying them to, say, a classification problem.\footnote{Ex-ante constraints are not only instrumental in avoiding statistical discrimination, they  are also natural and motivated in contexts where the search process repeats over many instantiations, e.g., tech firms repeatedly hiring software engineers, or other settings where an algorithm is repeatedly used for pre-employment screenings (see \citet{kleinberg2018selection} for discussions and examples).}

\revcolorm{
Motivated by above applications in socially aware algorithmic hiring, the overarching goal of our work is to bridge the aforementioned gap between the algorithmic fairness and the  sequential search literature. In particular, we aim to develop an algorithmic/computational framework to incorporate ex-ante constraints into sequential search and selection processes, given their unique operational aspects. We focus on a general framework to model a broad class of sequential search and selection problems, which we refer to as \emph{Markovian search}. In simple terms, a Markovian search models a stateful process in which candidates transition between states (e.g., uninspected, inspected, selected) according to a Markov chain.\footnote{We defer the formal definition to \Cref{sec:general}, where we focus on absorbing Markov chains with one or multiple terminal states.} These transitions occur as the decision maker sequentially takes (often costly) actions---such as inspecting or selecting a candidate---to maximize the net utility of the search and selection process. Given this setting, we ask the following research question:
}

\begin{displayquote}
\revcolorm{\mybox{\em {Given a Markovian search process with specified ex-ante constraints, how can we design (and efficiently compute) optimal or near-optimal policies that satisfy these constraints, either exactly or approximately?}}}
\end{displayquote}

\revcolorm{Before outlining our results, we clarify that we do not take a stance for or against ex-ante constraints that capture various socially aware notions such as fairness, nor do we engage in the legal debates surrounding such criteria (e.g., see \cite{ho2020affirmative} for discussions on the legal aspects of algorithmic fairness). Rather, our focus is on understanding how the imposition of such constraints introduces computational nuances and alters the structure of the optimal policy. Furthermore, while socially aware algorithmic hiring is the primary motivating application in this paper, we emphasize that ex-ante constraints in sequential Markovian search can also capture a broader class of operational constraints---such as long-run resource capacity limitations for specific parts of the search---which extend beyond socially aware applications.}



\smallskip


\revcolor{\noindent{\textbf{Basic Model -- Pandora's Box (\Cref{sec:pandora}).} 
To address this question, we first extend the classical \emph{Pandora's box problem}~\citep{weitzman1979optimal} by adding a simple \emph{ex-ante affine constraint}. In the original model, a decision maker selects a subset of candidates (or boxes), each with an independent stochastic reward for selection, within a capacity limit. Initially, only distributional information about rewards is known.\footnote{\revcolor{These distributions serve as input data to the algorithm, which may be ``biased.'' For example, if candidates come from different demographic groups, the prior distributions of their quality could be biased against minority groups. While we do not explicitly model this bias, similar to the approach taken in the algorithmic fairness literature~\citep{kleinberg2018algorithmic}, we account for it in our numerical simulations to gain insights. See \Cref{sec:numerical} for more details.}} However, inspecting a box (e.g., interviewing a candidate) reveals the actual reward at a cost. The main challenge is balancing inspection costs with the search for better alternatives in terms of rewards.  As shown in \cite{weitzman1979optimal} and later extended for multiple selections in \cite{singla2018price}, the optimal unconstrained policy is \emph{index-based}. Such a policy computes a polynomial-time computable index for each box at each time and greedily inspects or selects the one with the highest nonnegative index until capacity is reached or there exists no non-negative index. These indices are essentially \emph{Gittins indices}~\citep{gittins1979bandit} adapted to the Pandora's box problem.

In our variant of the above problem, we introduce a single general affine constraint on the marginal probabilities of selection and inspection of candidates, and also, as a generalization, on the same probabilities conditional on the candidates' qualities.
As a concrete example, consider candidates from two demographic groups, with the decision maker aware of this attribution upfront. To address implicit biases in reward distributions or cost disparities between these groups, the decision maker can impose demographic parity, ensuring an equal number of selections or inspections (in expectation) between the two groups~\citep{kleinberg2018algorithmic}. Alternatively, to enhance access and inclusion, they might apply a quota, ensuring a minimum fraction of the expected number of selections or inspections from the minority group. Lastly, constraints can also focus on high-quality candidates, targeting the fairness constraint towards highest values or quantiles of reward distributions, to avoid token selections. Other special cases of this constraint can also help with different forms of operational feasibility, such as satisfying a budget constraint on average. Our general class of ex-ante affine constraint for marginal probabilities of selection and inspection is defined in \Cref{sec:single-affine}, and its refined version for probabilities conditional on candidates' qualities is defined in \Cref{sec:quantile-constraint}, covering value-specific constraints and other scenarios.

Our first main result characterizes the optimal constrained policy that \revcolorm{\emph{exactly} satisfies such a general ex-ante affine constraint}; see \Cref{thm:const} for the basic version and \Cref{prop:general-constraint} for the refined version. We show that the optimal policy, called the \emph{Randomized Dual-adjusted Index Policy} (RDIP), has a remarkably simple structure: it randomizes between at most two deterministic index-based policies, both using the same indices but differing in tie-breaking. Furthermore, these indices can be obtained by simple adjustments to the Gittins indices in the original model.  \revcolorm{It is intriguing that adding our general affine constraint does not break the optimality of index-based policies, which is generally known to be brittle in several other problem variations}. The reason behind this is the existence of a relaxed Lagrange dual function, which transforms the constrained problem into an unconstrained Pandora’s box problem for a \emph{dual-based adjusted instance}. Later, we build on this insight in our general model.

From a computational perspective, we also show that both the adjusted indices and the two relevant tie-breaking rules are computable in polynomial time, even though there are exponentially many possible tie-breaking rules for an adjusted index-based policy. In particular, the ``correct'' adjustment can be found by solving a specific convex program related to the relaxed Lagrange dual function. Additionally, we show that the tie-breaking rules have an intuitive closed-form and correspond to two perturbations of the adjusted index-based policy, aiming to maximize slack in positive and negative directions, respectively; see \Cref{def:tie:extreme} (and \Cref{def:refined-ext-tie} for refined constraints) for details.

Our proposed dual-based adjustments have notable economic implications in our application. For example, in the case of demographic parity in selection, it is sufficient to uniformly adjust the \emph{rewards} of all candidates in each group by the same amount but in opposite directions, without changing inspection costs. Interestingly, this adjustment preserves the search order within each group, which is generally desirable, as highlighted in discussions of algorithmic fairness (see \citet{kleinberg2018algorithmic}).
In contrast, to meet an inspection quota or parity (for example, for candidates of high quality in the minority group, to avoid token interviews), we adjust \emph{inspection costs} instead. Specifically, inspection costs for minority candidates are reduced, while those for other candidates are increased, with the rewards unchanged. Unlike parity in selection, this adjustment can change the order within each group, indicating that inspection quotas may come at the expense of distorting within-group rankings; see \Cref{sec:lagrangian-pandora}, \Cref{sec:insights-dem-parity} and \Cref{sec:beyond-group} for further discussion.}

Later, as our second main result, we extend the efficient computation of the exact constrained optimal policy to the case with multiple \revcolorm{ex-ante} affine constraints in \Cref{sec:caratheodory} and \Cref{app:mutiple-affine}. Somewhat surprisingly, we show that, in contrast to the single-constraint case, randomizing among ``corner'' policies that in some sense try to maximize or minimize the slack for different constraints \emph{does not} suffice. Instead, our result is based on a reduction to a variant of the \emph{exact algorithmic Carathéodory problem}~\citep{caratheodory1911variabilitatsbereich} for a certain polytope (with possibly exponentially many vertices). We design a novel algorithm to solve this reduced problem given the structural properties of this polytope, in particular, being amenable to polynomial-time linear optimization. Although this algorithm is slow (yet polynomial-time), it serves as a proof of concept for generalizing our single-constraint result to multiple constraints. Our algorithm may also be of independent interest in other contexts where one only has access to a polytope (with exponentially many vertices) via a linear optimization oracle---see \Cref{sec:caratheodory} for details~\citep{cai2012optimal,alaei2014bayesian,dughmi2021bernoulli}.

\smallskip

\noindent\textbf{General Model -- Joint Markov Scheduling (\Cref{sec:general}).} While the Pandora's box problem has served as a cornerstone for the study of sequential search, real-world hiring decisions typically involve richer and more complex search processes, for example, with multiple stages of screening or several rounds of communication between the decision-maker and the candidates. Motivated by studying such search processes, we model the general ``stateful'' sequential search for hiring candidates as the \emph{joint Markov scheduling (JMS) problem}~\citep{gittins1979bandit,dumitriu2003playing}, which mathematically extends the classic Pandora's box problem to richer sequential search and hiring settings.  

In the JMS model, each candidate is represented by an absorbing Markov chain (MC) with terminal states. When we interact with a candidate, the corresponding MC undergoes a state transition. A candidate is ``selected'' when its MC reaches a terminal state. Non-terminal states typically have negative rewards (representing inspection costs), while terminal states offer positive rewards (representing selection gains). The search process involves sequentially inspecting these MCs, and ends when a subset of candidates is selected up to the available capacity (or earlier, leaving some capacity unfilled). In the unconstrained model, the goal is to choose and inspect MCs to maximize the expected net reward. We formally introduce this setting in \Cref{sec:JMS-setting}.

Equipped with the JMS setting for stateful sequential search, we generalize our framework in \Cref{sec:pandora} even further by incorporating a broad range of \revcolorm{ex-ante constraints}, beyond a single affine ex-ante constraint. In particular, we allow for \emph{multiple affine or convex constraints} on the vector of ex-ante outcomes of the search, defined as the \emph{expected visit numbers} for each state of each candidate under a given policy. This flexible approach captures both group and individual notions of fairness, as well as additional operational constraints. We formalize these constraints---and elaborate on their applications for fairness and inclusion---in \Cref{sec:JMS-setting}. \revcolorm{Importantly, by using ideas similar to the Pandora's box setting, we can obtain randomized dual-adjusted index-based policies that satisfy a single or multiple affine constraints exactly (as in \Cref{sec:single-affine-policy} and \Cref{sec:caratheodory}, with more details in \Cref{app:mutiple-affine}). However, this approach fails in the presence of convex constraints.}

As our third main result, we present a fully polynomial-time approximation scheme (FPTAS) called the \emph{Generalized Randomized Dual-adjusted Index Policy} (G-RDIP) \revcolorm{for the general JMS problem with general ex-ante constraints (affine and convex).} Given  constants $\epsilon, \delta>0$, G-RDIP computes a randomized policy in time polynomial in $\frac{1}{\epsilon}$, $\frac{1}{\delta}$, and the input size. The policy achieves an expected objective value within an additive error $\epsilon$ of the constrained optimal solution while satisfying all constraints within an additive error $\delta$  (see \Cref{alg:RAI} and \Cref{thm:RAI}). Unlike in \Cref{sec:pandora}, G-RDIP handles multiple ex-ante constraints by reformulating the problem as a \emph{minimax Lagrangian game}, where the decision maker (primal player) selects a randomized policy and the dual player chooses the dual variables. To approximately find the equilibrium---and thus the optimal constrained policy---we aim to use a standard primal-dual method~\citep{arora2012multiplicative}, with the dual player running online learning and the primal player best-responding in each iteration.

The above approach faces two critical challenges. First, a key component of G-RDIP is the best-response procedure, which maximizes the Lagrangian for a given set of dual variables. This problem can be viewed as an unconstrained JMS with a specific regularizer in the objective. When only affine constraints are involved, it reduces to an unconstrained JMS with an adjusted instance, similar to that in \Cref{sec:pandora}. Prior work has shown that JMS admits an optimal Gittins index-based policy under certain assumptions on state rewards~\citep{gittins1979bandit,dumitriu2003playing}. However, these assumptions are violated after dual adjustments, so we cannot directly use this result. To resolve this, in \Cref{sec:optimal-JMS-arbitrary} and \Cref{app:JMS-general} we refine these results to show that the index-based structure of the optimal policy remains valid for \emph{arbitrary} positive or negative state rewards after an intricate polynomial-time preprocessing step (\Cref{thm:JMS-optimal-reduction-informal}). Thus, the optimal adjusted policy remains polynomial-time and index-based. Second, with convex constraints, the Lagrangian relaxation is no longer linear in the ex-ante outcome vector. This means that it is not associated with an adjusted instance, and it is unclear whether there is a polynomial-time optimizer for the best response. To address this technical barrier, we further relax the Lagrangian using Fenchel's weak duality on the concave terms (reviewed in \Cref{sec:convex}). This results in a relaxed game that, while still non-linear, becomes bilinear with respect to the two sets of dual variables: one for the constraints and the other for the Fenchel conjugate functions.

Our policy, G-RDIP, employs a simple two-layer iterative learning algorithm in this relaxed game, which updates the two sets of duals separately in its outer and inner layers.  An overview of this design is provided in \Cref{sec:RAI}. It also leverages the structure of our problem to find an approximate equilibrium and provides a certificate that the relaxed game is tight, with small additive errors compared to the original game. The algorithm randomizes among polynomially many deterministic index-based policies, each optimal for an adjusted JMS instance, with adjustments based on the dual estimates from a given round of the two-layer learning algorithm (see Line 4 of \Cref{alg:RAI}). Notably, we are unaware of any prior work offering an FPTAS for a JMS problem with concave rewards, making our results of independent interest.

\smallskip
\noindent\textbf{Numerical Simulations \&  Insights (\Cref{sec:numerical}, \Cref{apx:numerical-main}, \Cref{sec:numerical_uniform}, and \Cref{sec:numerical_JMS})} We complement our theoretical framework with numerical simulations on synthetic data, focusing on demographic parity and quota in the Pandora's box setting. These simulations provide insight into the potential costs and benefits of imposing such constraints. We consider scenarios with no inherent statistical asymmetry between the two groups in terms of true candidate qualities, aside from natural population heterogeneity. \revcolorm{We then add one extra layer on how our input instance is generated: we assume that true qualities are \emph{unobservable} by the decision maker, and instead she can observe ``signals'' as proxies of the true qualities through inspections. As before, we assume that the decision maker only has prior distributional knowledge of these signals (but not the true qualities). These signals---which are the analogue of the ``values''  in our base model (\Cref{sec:pandora})---may be \emph{biased} downward for one group, meaning that they are smaller than the true qualities.\footnote{\revcolorm{By distinguishing between true qualities and the signals, we consider a richer setting that enables us to investigate the impact of socially aware constraints in terms of both immediately observable quantities (i.e., signals) and the unobservable ones (i.e., true qualities); however, focusing only on the observable signals, we are exactly simulating our base model.\label{footnote:biased observatoin model}}}. We consider a specific multiplicative bias where the signals are scaled versions of the true qualities based on a bias factor. This modeling approach aligns with prior work \citet{kleinberg2018selection}---while supported by empirical evidence driven by data in similar contexts (see \cite{wenneras2010nepotism,faenza2020reducing})---and it reflects the bias in input data, which is a common practical concern as discussed before.} We then explore the effects of varying bias levels and provide the following intriguing insights:}

\revcolorm{
\smallskip
\noindent\emph{(i)}~We first consider the utility of the searcher when candidates' qualities are perceived to be their signals (and not true qualities), which we term as \emph{short-term utility}. We find that even with a moderate bias in the signals and strong fairness constraints like demographic parity, the ``price of fairness''---the relative utilitarian loss due to imposing the constraint---is small. This is surprising, as the optimal unconstrained policy can result in significant disparities between the two groups. The key implication is that imposing a parity constraint can deliver substantial egalitarian benefits without significantly impacting the utilitarian outcomes (\Cref{sec:num-short-term}).}

A compelling perspective from the theory of downstream hiring outcomes in labor economics (see \cite{becker2010economics, canay2020use}) suggests that candidates’ true qualities are often unobservable at the time of hiring and only revealed in the \emph{long-term} once they are given a chance. For example, although interview performance or resumes are important signals of future job success, they can underestimate the potential of candidates from disadvantaged backgrounds due to limited access to professional training or resources.\footnote{\revcolor{A study by \cite{deortentiis2022different} shows that candidates from higher social classes often perform better in interviews because of greater access to preparatory resources and increased confidence.}} However, once hired and provided with equal opportunities, these candidates can perform as well as their privileged peers. Therefore, applying constraints such as demographic parity in hiring can improve long-term outcomes by giving diverse candidates with hidden potential an equal chance to succeed. We examine this hypothesis in our numerics.


\smallskip
\revcolorm{
\noindent\emph{(ii)}~We consider \emph{long-term utility} where the candidates' values are replaced with their true, unobserved qualities. Our numerical results show that imposing socially aware constraints, such as demographic parity or quota, can even make the search \emph{more efficient} in terms of long-term utilities. This has a key implication: even though the decision maker selects candidates based on biased signals, ex-ante parity constraints help ``calibrate'' selections to (partially) correct the bias, meaning that selections will be more balanced between the two groups. Thus, this approach can outperform an unconstrained policy that ignores this balancedness (see \Cref{sec:num-long-term}).}

\smallskip
\noindent\emph{(iii)}~When the bias in the signals is significantly high, we find that imposing a strict constraint like demographic parity may lead to \emph{unintended inefficiencies}. In such cases, the decision maker might leave part of the capacity \emph{unallocated} to ensure parity. Therefore, it may be more practical to consider lenient alternatives, such as quotas with carefully chosen parameters, which can provide adequate representation to minority groups without causing underallocation. (see \Cref{sec:numerical-unintended}).

We also extend our simulations to a multistage screening scenario---such as a hiring process with phone interviews followed by on-site interviews---which is a special case of our JMS model. We then run our near-optimal, near-feasible G-RDIP algorithm to incorporate multiple socially aware constraints simultaneously. In summary, our numerical insights listed above for the Pandora's box model carry over to the JMS setting. Moreover, our convergence analysis indicates that the G-RDIP algorithm is fast, underscoring its practical relevance; see \Cref{sec:numerical_JMS} for details of the simulation setup and numerical results.

Lastly, we study the effect of various forms of resource augmentation (e.g., increasing capacity) and constraint adjustments (e.g., tuning the fraction in a quota constraint) in our simulations. In summary, our results suggest that these small and simple changes can go a long way in terms of improving the search utility, both with respect to signals and also true values. \revcolor{We also check the robustness of our results to the choice of value distributions.} See \Cref{sec:numerical-parity}, 
 \Cref{sec:numerical-quota}, \Cref{sec:numerical-budget}, and \Cref{sec:numerical_uniform} for more details. 

\smallskip
\revcolorm{
\noindent\textbf{Algorithmic \& Managerial Takeaway.} A key managerial takeaway from our work is that algorithmic decision-makers can develop optimal or near-optimal policies for sequential search and selection processes, which satisfy a wide range of (socially aware or operational) ex-ante constraints through carefully applied randomization and simple, often interpretable, adjustments to the original unconstrained problem.}

\revcolor{We end this introduction by highlighting that our work is related to various lines of work in operations research, computer science, and economics. We postpone the discussion of further related work to \Cref{sec:further-related-work}.}

\newcommand{\policyplus}{\policy^{+}}
\newcommand{\policyminus}{\policy^{-}}
\newcommand{\NumSelect}{k}
\newcommand{\Instance}{\mathcal{I}}
\newcommand{\util}{\textsc{Utility}}
\newcommand{\OptUnconstrained}{\textsc{OPT}_{\textsc{uc}}}
\newcommand{\OptConstrained}{\textsc{OPT}_{\textsc{cons}}}
\newcommand{\LagrangeConst}{\mathcal{L}_{\textsc{cons}}}
\newcommand{\DualLagrangeConst}{\mathcal{G}_{\textsc{cons}}}

\section{\revcolorm{Pandora's Box with Affine Constraints}}
\label{sec:pandora}
We start by revisiting the canonical sequential search model known as the {\em Pandora's box problem}, introduced by \cite{weitzman1979optimal}, under a single ex-ante affine constraint. Our main goal in this section is to characterize and compute an optimal policy that \emph{exactly} satisfies such a constraint. \revcolorm{Later in this section we show extensions to a more general version of the single ex-ante affine constraint and to multiple ex-ante affine constraints.}

\subsection{Setting and Notations}
\label{sec:pandora-setting}

Consider the following setting, known as the ``Pandora's box problem with multiple selections'': a decision maker is presented with $n$ alternatives (or boxes) indexed by $[\altnum] = {1, 2, \ldots, n}$, and aims to eventually select at most $\NumSelect \in [n]$ of them. Each box $i \in [n]$ is associated with an independent stochastic reward $v_i \sim F_i$, also referred to as the \emph{value} of box $i$, where $F_i$ is a prior probability distribution with finite and bounded support $\values_i \subset \mathbb{R}$.\footnote{\revcolor{We consider finite bounded support mostly for the simplicity of technical expositions and consistency with \Cref{sec:general}. Our results can be extended to general value distributions with appropriate technical modifications, omitted for brevity.}}
Initially, the decision maker only knows the prior distributions. To learn the actual reward of box $i$, she must \emph{inspect} (or open) it at a known cost $c_i$. Upon opening box $i$, she observes its reward. We also assume that the inspection costs are bounded for technical reasons. Although rewards and costs are typically non-negative in applications, we allow them to be negative or zero for reasons pertinent to our setting (explained later). At any time, the decision maker decides whether to stop or continue the search; if she decides to stop, she can choose to \emph{select} up to $\NumSelect$ opened boxes. Otherwise, she decides which unopened box (if any) to open next. We represent an instance of our problem by $\Instance = \displaystyle\left\{\left( \values_i, F_i, c_i \right) \mid i \in [n] \right\}$. We also denote the outside option by the index $0$, which is perceived as a dummy box with $v_0 = c_0 = 0$.

An \emph{admissible} policy $\policy$ in the above setting is a (possibly randomized and adaptive) rule that, at each time, given the history, decides whether to inspect a box or to stop and make selections as described above.  The goal is to maximize the expected \emph{utility of the search}, defined as the sum of the rewards of the finally selected boxes minus all inspection costs incurred throughout the search. Given an instance $\Instance$ and an admissible policy $\policy$, for each box $i$, we define the indicator random variables $\select_i^{\policy} \in \{0,1\}$ and $\inspect_i^{\policy} \in \{0,1\}$, representing whether box $i$ is selected and inspected under policy $\policy$, respectively. Note that $\inspect_i^\policy \geq \select_i^\policy$ in every sample path, as inspection is obligatory before selection. The expected utility of a policy $\policy$ in an instance $\Instance$ can be expressed as follows:
\begin{equation}
\label{eq:utility}
\util(\pi;\Instance)\triangleq\expect{\sum_{i\in[\altnum]} \left(\select_i^{\policy}  v_i - \inspect_i^{\policy}  c_i\right)}~.
\end{equation}
Using this notation, the \emph{unconstrained} optimization problem of finding the optimal policy for an instance $\mathcal{I}$ is formulated as the following stochastic program over the space of policies:
\begin{equation}
\label{eq:opt-unconstrained}
\tag{\textsc{OPT-uc}}
\begin{aligned}
\OptUnconstrained\triangleq & \quad \max_{\policy\in\PolicySpace}~\util(\pi;\Instance)~,
\end{aligned}
\end{equation}
where $\PolicySpace$ is the set of all admissible policies.\footnote{As a minor technical detail, we note that the set of deterministic policies is finite. This is because there are finitely many mappings from the history---which is finite due to the discrete rewards--to the set of possible actions: stopping and selection, or inspecting the next box (which is also finite since we have finitely many boxes).}

\subsection{Affine Ex-ante Constraint: Parity, Quota, and Budget}
\label{sec:single-affine}
Now, consider adding the following ex-ante (i.e., in expectation) affine constraint to \eqref{eq:opt-unconstrained}, which can have the form of either an equality or an inequality constraint:
\begin{equation}
\label{eq:affine-constraint}
\expect{\sum_{i\in[n]}\theta^{S}_i{\select_i^{\policy}}+ \sum_{i\in[n]}\theta^{I}_i{\inspect_i^{\policy}}}=b ~~\left(\textrm{or} \leq b\right)
\end{equation}
where $\AffineVector=\left[\theta_i^S; ~\theta_i^I\right]_{i\in[\altnum]}\in\mathbb{R}^{2\altnum}$ and $b\in[-1,1]$ (after normalization). \revcolor{Note that Constraint~$\eqref{eq:affine-constraint}$ can be alternatively interpreted as an affine constraint on the marginal probabilities of selection and inspection under policy $\policy$.} In the resulting constrained problem, the goal is to maximize the expected net utility of the search while satisfying this constraint. Formally, we have the following stochastic program:
\begin{equation}
\label{eq:opt-constrained}
\tag{\textsc{OPT-cons}}
\begin{aligned}
\OptConstrained\triangleq & \max_{\policy\in\PolicySpace(\AffineVector,\AffineConstant)}~\util(\pi;\Instance)\\
=&~~\max_{\policy\in\PolicySpace} ~~\expect{\sum_{i\in[\altnum]} \left(\select_i^{\policy}  v_i - \inspect_i^{\policy}  c_i\right)}, 
~\textrm{s.t.}~~\expect{\sum_{i\in[n]}\theta^{S}_i{\select_i^{\policy}}+ \sum_{i\in[n]}\theta^{I}_i{\inspect_i^{\policy}}}=b~~(\textrm{or} \leq b), 
\end{aligned}
\end{equation}
where $\PolicySpace(\AffineVector,\AffineConstant)$ is the set of all admissible policies that satisfy Constraint~\ref{eq:affine-constraint}. 

As mentioned earlier, various special cases of Constraint~\ref{eq:affine-constraint} can be used to satisfy socially aware or operational criteria. \revcolorm{For example, the following socially aware notions can be encoded by such a constraint:}
\begin{itemize}[leftmargin=*]
\revcolor{
\item\textbf{Group demographic parity:} Suppose each candidate belongs to one of two non-overlapping demographic groups, $\ManSet$ and $\WomanSet$ (e.g., male or female), and the decision maker observes each candidate's group. To promote access and equality, a common approach is to balance the expected number of ``successful'' outcomes (e.g., selection or inspection) between the two groups. This can be achieved by imposing this equality constraint:
\begin{equation} \label{eq:parity} \tag{\textsc{Parity}} \mathbb{E}\left[ \sum_{i \in \ManSet} \select_i^{\policy} \right] = \mathbb{E}\left[ \sum_{i \in \WomanSet} \select_i^{\policy} \right] \quad \left( \textrm{or} \quad \mathbb{E}\left[ \sum_{i \in \ManSet} \inspect_i^{\policy} \right] = \mathbb{E}\left[ \sum_{i \in \WomanSet} \inspect_i^{\policy} \right] \right) \end{equation}
This constraint ensures \emph{demographic parity} by equalizing the expected number of successful outcomes between the two groups. It also aligns with similar constraints studied in the algorithmic fairness literature for group fairness in classification (e.g., \cite{kleinberg2018algorithmic}).
\smallskip

\item\textbf{Group quota:}  Under the same setting as above, alternatively, one can promote inclusion and equality of opportunity by ensuring a minimum representation of a minority group $\WomanSet$ (e.g., female candidates or people of color), we can impose a \emph{quota} constraint. This constraint is commonly used in affirmative action policies in hiring (cf. \cite{welch1976employment}). Let the parameter $\theta \in [0,1]$ represent the desired minimum proportion of successful outcomes from group $\WomanSet$. The constraint is formulated as:
\begin{equation} \label{eq:quota} \tag{\textsc{Quota}} \mathbb{E}\left[ \sum_{i \in \WomanSet} \select_i^{\policy} \right] \geq \theta \cdot \mathbb{E}\left[ \sum_{i \in [\altnum]} \select_i^{\policy} \right] \quad \left( \textrm{or} \quad \mathbb{E}\left[ \sum_{i \in \WomanSet} \inspect_i^{\policy} \right] \geq \theta \cdot \mathbb{E}\left[ \sum_{i \in [\altnum]} \inspect_i^{\policy} \right] \right) 
\end{equation}}
    \item \textbf{Average budget constraint}: 
    In repeated hiring scenarios with limited resources, the decision maker might face an average \emph{budget} constraint on total interview costs, or aim to keep the average number of hires below a threshold lower than the actual capacity. Alternatively, if $\WomanSet$ represents underprivileged candidates who cannot afford hiring or interviewing expenses (e.g., relocation fees or application costs), we might want to subsidize their expenses within an average budget. Given a budget $B \geq 0$ and expenses $e_i \geq 0$ for hiring (or interviewing) each candidate $i$, these constraints are formulated as follows.
    \begin{equation}
    \label{eq:budget}
        \tag{\textsc{Budget}}
        \expect{\sum_{i \in [\altnum]} e_i\select_i^{\policy} } \leq B~~\left(\textrm{or}~~ \expect{\sum_{i \in [\altnum]} e_i\inspect_i^{\policy} } \leq B\right)
    \end{equation}
\end{itemize}

\smallskip
\begin{remark}[Feasible vs. Infeasible]
\label{rem:feasible}
    Problem~\ref{eq:opt-constrained} may be infeasible; that is, $\PolicySpace(\AffineVector, \AffineConstant) = \emptyset$. However, as the marginal probabilities of selection and inspection are variables set by the policy---and noting that inspection is obligatory before selecting any box---the problem is feasible if and only if the following polytope is non-empty: $\mathbf{x},\mathbf{y}\in[0,1]^\altnum$ such that (i)~$x_i\leq y_i,~i\in[\altnum]$, (ii)~$\sum_{i\in[\altnum]}x_i\leq \NumSelect$, and (iii)~$\sum_{i\in[\altnum]}\theta^{S}_i{x_i}+ \sum_{i\in[n]}\theta^{I}_i{y_i}\leq b$. Thus, to verify the feasibility of Problem~\ref{eq:opt-constrained}, we only need to check whether this simple polytope is non-empty. Constraints such as \ref{eq:parity}, \ref{eq:quota}, and \ref{eq:budget} always result in a feasible problem, since the trivial policy that ``does nothing'' (i.e., selects and inspects no boxes) satisfies the constraint.
\end{remark}
\begin{remark}[Equality vs. Inequality Constraint]
\label{rem:tightness}
\revcolorm{When Constraint~\ref{eq:affine-constraint} is an inequality, we can determine whether to drop the constraint or replace it with its equality form (i.e., make it binding) through a simple check. First, select \emph{any} optimal solution of Problem~\ref{eq:opt-unconstrained}. If this solution also satisfies Constraint~\ref{eq:affine-constraint}, then it is also optimal for Problem~\ref{eq:opt-constrained}, and Constraint~\ref{eq:affine-constraint} can be dropped. Otherwise, we show formally in \Cref{app:binding} (\Cref{lemma:binding}) that without loss of generality we can assume that the optimal constrained policy exactly `uses up' the constraint, and hence we can replace Constraint~\ref{eq:affine-constraint} with its equality form.}

    
    
\end{remark}
\revcolor{Based on Remarks~\ref{rem:feasible} and \ref{rem:tightness}, we can first check the feasibility and whether the constraint is binding as a pre-processing step; Therefore, without loss of generality, we assume the following in the remainder of this section.}
 \begin{assumption}
 \label{ass:pandora_single_equality_feasibility}
    Problem~\ref{eq:opt-constrained} is feasible and Constraint~\ref{eq:affine-constraint} is an equality constraint.
 \end{assumption}

\subsection{Optimal Policy for the Constrained Problem}
\label{sec:single-affine-policy}

To design optimal policies for the constrained problem, we first revisit the optimal algorithm for the unconstrained version in \Cref{subsub:base:pandora}. After highlighting the nonuniqueness of the optimal solution, we provide a \emph{refined} version that offers more flexibility in selecting outcomes, allowing us to incorporate the ex-ante affine constraint. Building on this refinement in \Cref{sec:lagrangian-pandora} and \Cref{sec:randomized-tie}, we present a remarkably simple optimal policy for the constrained version.


\subsubsection{Pandora's Box Optimal Policy: Review and Refinement}
\label{subsub:base:pandora}

In his seminal work, \cite{weitzman1979optimal} presented an elegant index-based policy for the unconstrained Pandora's box problem with non-negative rewards and costs, which works as follows: (i)~For each box $i$, calculate an index $\reserve_i$ such that $\reserve_i\in\left\{\reserve\in\mathbb{R}:\expect{\left(v_i - \reserve\right)^+} = c_i\right\}$ (set $\reserve_0=0$ for the outside option); (ii)~Begin inspecting boxes one by one in decreasing order of their indices $\reserve_i$, observing their rewards upon inspection; (iii)~After inspecting each box, stop if the maximum realized reward among the inspected boxes exceeds the maximum index among the unopened boxes (i.e., those not yet inspected); (iv)~Upon stopping, select the opened box with the highest realized reward. For the case of multiple selections with $\NumSelect>1$, \cite{kleinberg2016descending,singla2018price} show that a simple modification of this policy, called \emph{(frugal) greedy index-based policy}, is optimal: In step~(iii), stop if the 
$\NumSelect^{\textrm{th}}$ highest reward in the inspected boxes (considered zero if fewer than $\NumSelect$ boxes are inspected) exceeds the maximum index among unopened boxes, and in step~(iv), select the $\NumSelect$ inspected boxes with the highest realized rewards.

Importantly, the above description of the optimal policy lacks crucial details when we consider generic instances of the problem. For instance, when \( c_i < 0 \), the index \( \reserve_i \) becomes ill-defined. Additionally, when \( c_i = 0 \), the choice of \( \reserve_i \) is not unique since any \( \reserve_i \in [\max\{ v : v \in \values_i \}, +\infty) \) is valid. This means the optimal policy can either open any box \( i \) with \( c_i = 0 \) sooner by selecting a higher \( \reserve_i \) within that interval or delay opening box \( i \) until the largest index among unopened boxes is lower than \( \max\{ v : v \in \values_i \} \). Moreover, if the \( \NumSelect^{\textrm{th}} \) highest reward among the inspected boxes is negative or zero at any point, it is unclear from the description whether the algorithm should stop or continue. Finally, there may be ties in the order of inspections (step~(ii)), the stopping decision (step~(iii)), and the selection decision (step~(iv)). In the basic problem, these details can be overlooked because rewards and costs are typically non-negative, and any feasible choice of \( \reserve_i \) and tie-breaking rules for steps (ii)--(iv) yields an optimal policy. However, making the ``right choices'' becomes extremely important when satisfying our ex-ante constraint, as we see later in this section.


In light of these considerations, we present a refined version of the greedy index-based policy, described in \Cref{alg:Pandora}, which specifies the previously undefined components as follows:

\begin{itemize}[leftmargin=*]
    \item \textbf{Redefinition of indices}: For each \( i \in [n] \), we redefine the index \( \reserve_i \) as:
    \begin{equation}
    \label{eq:base:reserve}
    \reserve_i \triangleq 
    \begin{cases}
    + \infty & \text{if } c_i < 0, \\ 
    \inf \left\{ \sigma \in \mathbb{R} : \mathbb{E}\left[ (v_i - \sigma)^+ \right] = c_i \right\} & \text{otherwise.}
    \end{cases}
    \end{equation}

    \item \textbf{Option values \& candidates}: We define an \emph{option value} $ \option_i $ for each $ i \in [n] $, initialized to $\reserve_i$ and updated to $v_i$  once box $i$ is inspected ($o_0 = 0$ for the outside option). Using the option values $\{ \option_i \}$, the algorithm maintains the set of \emph{candidates} $\candidates \subseteq [n]$ for the next inspection or selection at each step, consisting of unselected boxes whose option values are currently maximum and unopened boxes with zero inspection cost. The algorithm also tracks the set \( \mathcal{O} \subseteq [n] \) of opened boxes.

    \item \textbf{Tie-Breaking rule}: The algorithm takes a (possibly adaptive or randomized) tie-breaking rule \( \tierule \) as input, which, at each step, maps the history of the search process to a member of the current set of candidates \( \candidates \).
\end{itemize}

\begin{algorithm}
\caption{Refined Optimal Policy for Pandora's Box (with Multiple Selections)}
\label{alg:Pandora}
    \SetKwInOut{Input}{input}
    \SetKwInOut{Output}{output}
 \Input{instance $\displaystyle\left\{\left(\values_i, F_i, c_i\right)| i \in [n]\right\}$; target number of selections $\NumSelect$; tie-breaking rule $\tiebreak$}
 
 \Output{set of opened boxes $\mathcal{O}$; set of selected boxes $\mathcal{S}\subseteq \mathcal{O}$ with $\lvert\mathcal{S}\rvert \leq \NumSelect$}
 
 \vspace{2mm}
 
 Initialize $\mathcal{O} \leftarrow \{0\}$, $\mathcal{S}\leftarrow \emptyset$,  and $\textsc{terminate}\leftarrow\texttt{No}$~~{\color{\commentcolor}\tcc{outside option, with index $0$, is a dummy box with ${v_0=c_0=0}$ that is opened at the beginning}}
 
 \For{$i\in[n]$}{
 Compute $\reserve_i$ as per \Cref{eq:base:reserve} and initialize the \emph{option value}  $\option_i \leftarrow \reserve_i$}
 
 Initialize option value of the outside option $o_0\leftarrow 0$

 \vspace{1mm}
 \While{$\textsc{terminate} == \emph{\texttt{No}}$}
{
 Pick the set of candidates $\candidates\leftarrow \Big(\underset{i\in\left([n]\setminus\mathcal{S}\right)\cup\{0\}}{\argmax}~\option_i\Big) \cup \left\{i\in[n]\setminus \mathcal{O}:~c_i = 0\right\}$ {\color{\commentcolor}~~\tcc{the candidate set $\candidates$ may include the outside option ${i=0}$}}
 
 Choose  $i^* \in \candidates$ based on the tie-breaking rule $\tiebreak$ \\
 \If{$i^* \notin \mathcal{O}$}
  {
 Inspect box $i^*$ and observe $v_{i^*}$
 
 Add $i^*$ to $ \mathcal{O}$ and set $\option_{i^*} \leftarrow v_{i^*}$ 
 }
 \Else
 {
 Add $i^*$ to $\mathcal{S}$
 
 \If{$i^*=0$ or $\lvert\mathcal{S}\rvert=\NumSelect$}
 {
 Set $\textsc{terminate}\leftarrow \texttt{Yes}$
 }
 }
 }
\end{algorithm}

\revcolor{We remark that \Cref{alg:Pandora} with $\NumSelect = 1$ implements the same ordering and stopping rule as in \cite{weitzman1979optimal}, up to tie-breaking. Similarly, for $\NumSelect > 1$, this algorithm is exactly equivalent, again up to tie-breaking, to the optimal policy of \cite{kleinberg2016descending,singla2018price} described earlier. For more details on why this is the case, see \Cref{prop:pandora-refinement} in \Cref{app:pandora-details}.}




Now, consider a generic instance $\displaystyle\left\{(\values_i, F_i, c_i) \mid i \in [n] \right\}$ of the problem, possibly with negative rewards or costs. In such an instance, various kinds of ties can occur as described earlier. Specifically, during the execution of \Cref{alg:Pandora}, there may be multiple boxes with the maximum option value at any step, or there might be one or more unopened boxes available for free inspection.\footnote{\revcolor{One might think these ties only happen in degenerate cases when distributions and costs are not in general position. However, as we will see later, ties can easily arise after dual adjustments of any instance. In fact, our dual adjustments sometimes lead to ties even when the original instances are in general position to incorporate our ex-ante constraints. Mathematically speaking, ``no ties'' occur if no value is adjusted to zero (otherwise, there can be a tie between that value and the outside option), and the minimum of the piecewise-linear convex function $\DualLagrangeConst$ (defined later in \Cref{sec:lagrangian-pandora}) occurs at a non-breakpoint. The first condition can be violated in any instance. For the second condition, from a polyhedral geometric perspective, this is equivalent to an entire face of a polytope being optimal for a certain linear optimization over this polytope—which should not happen for instances in general position. We defer the details to later in this section; see also \Cref{example:FS}.}} Therefore, there may be multiple candidates to choose for the next step. The advantage of the refined presentation in \Cref{alg:Pandora} is that it involves only a \emph{single tie-breaking decision (in Line 7)}, rather than separate decisions for ordering, stopping, and selection. As we will show in \Cref{sec:lagrangian-pandora}, this single  rule is sufficient to implement the optimal policy for the constrained Pandora's box problem with an ex-ante affine constraint, after properly adjusting the problem instance.\footnote{For the special case of $\NumSelect = 1$, this algorithm is rich enough to cover all (possibly randomized) optimal policies for the unconstrained Pandora's box problem. A proof of this fact is presented in \Cref{app:pandora-details}, Lemma~\ref{lemma:full-cover}. However, this is not true for general $\NumSelect > 1$. Nevertheless, it suffices for our purposes when $\NumSelect > 1$.}

\subsubsection{Dual-based Adjustments}
\label{sec:lagrangian-pandora}
Equipped with these preliminaries, we now focus to Problem~\ref{eq:opt-constrained}, and construct an optimal policy for this problem. We start by  ``Lagrangifying'' the Constraint~\ref{eq:affine-constraint} in the objective of this stochastic program. In particular, given policy $\policy$ and dual variable $\lambda$, define the Lagrangian relaxation of the problem as the following:
\begin{align}
\LagrangeConst (\pi;\LagVector) & \triangleq \expect{\sum_{i \in[\altnum]} \left(\select_i^{\policy}   v_i - \inspect_i^{\policy}   c_i\right)} - \LagVector  \expect{\sum_{i\in[n]}\theta^{S}_i{\select_i^{\policy}}+ \sum_{i\in[n]}\theta^{I}_i{\inspect_i^{\policy}}}+ \LagVector\cdot b\\ 
& = \expect{\sum_{i \in[\altnum]} \select_i^{\policy}(v_i-\LagVector\cdot\theta^{S}_i)-\sum_{i\in[n]}\inspect_i^{\policy}   (c_i+\LagVector\cdot\theta^{I}_i)} +\LagVector\cdot b~.
\label{eq:largrange}
\end{align}
We then define the Lagrange dual function as:
\begin{align}
\label{eq:dual}
\DualLagrangeConst (\LagVector) \triangleq \max_{\policy \in \PolicySpace} ~\LagrangeConst (\pi;\LagVector)~.
\end{align}
By fixing $\lambda$ and ignoring the constant term $\LagVector\cdot b$, the maximization problem in \eqref{eq:dual} has \emph{exactly} the same structure as the original Pandora's box problem but with adjusted instance parameters. Specifically, we define the \emph{adjusted rewards} $\widetilde{v_i}$ and \emph{adjusted costs} $\widetilde{c_i}$ for each $i \in [n]$ as:
\begin{equation}
\label{eq:adjusted-instance}
\widetilde{v_i}\triangleq v_i-\LagVector\cdot\theta^{S}_i~~~,~~~\widetilde{c_i}\triangleq c_i+\LagVector\cdot\theta^{I}_i~.
\end{equation}

To gain more insight into this adjustment, let us examine some special cases. First,
consider \ref{eq:parity} in  selection given the two demographic groups $\ManSet$ and $\WomanSet$. In this case, the values are adjusted as follows (with costs remaining unchanged):
\begin{equation}
\label{eq:adjusted-value}
\begin{aligned}
\widetilde{v_i}\triangleq 
\begin{cases}
 v_i+\LagVector &  i \in \WomanSet~,\\ 
 v_i-\LagVector & i\in \ManSet~.\\
\end{cases}
\end{aligned}
\end{equation}
Thus, the values for one group are increased by $\LagVector$, while those for the other group are decreased by the same amount. The sign of $\LagVector$ determines which group gains more representation and which loses. By choosing an appropriate $\LagVector$, we can favor the underrepresented group (i.e., the group with a lower expected number of selections in the unconstrained problem) and reduce the advantage of the overrepresented group, thereby equalizing their expected number of selections. 

As another example, consider \ref{eq:quota} in inspection given the minority group $\WomanSet$ and the majority group $\ManSet$. Here, the costs are adjusted as follows (with values remaining unchanged):

\begin{equation}
\label{eq:adjusted-quota}
\begin{aligned}
\widetilde{c_i}\triangleq 
\begin{cases}
 c_i+ (\theta-1)\LagVector &  i \in \WomanSet~,\\ 
 c_i+\theta \LagVector & i\in \ManSet~.\\
\end{cases}
\end{aligned}
\end{equation} 
In this adjustment, the costs for the minority group decrease by $(1-\theta)\LagVector$ and for the majority group increase by $\theta \LagVector$. By selecting an appropriate $\LagVector$, the policy provides more advantage to the minority group by lowering their inspection costs, thus ensuring a certain level of representation for this group during the inspection process.


Importantly, we note that for all \(\lambda \in \mathbb{R}\), we have \(\DualLagrangeConst(\lambda) \geq \OptConstrained\). Therefore, we can solve Problem~\ref{eq:opt-constrained} if we can find a policy $\policy$ maximizing the Lagrangian relaxation function $\LagrangeConst (\pi;\LagVector)$ for some choice of \(\lambda\) that also satisfies Constraint~\ref{eq:affine-constraint}. The rest of this section is dedicated to constructing such a policy. To this end, we begin by establishing some structural properties of the function \(\DualLagrangeConst\), summarized in \Cref{lemma:Const:DualLagrangeProperties} (proved in \Cref{app:sec:pandora}). To facilitate its proof, for any policy \(\policy\), we first define its corresponding \emph{constraint slack}:
\begin{align}
\label{eq:slack}
\slack_{\textsc{cons}}^{{\policy}} \triangleq b-\expect{\sum_{i\in[n]}\theta^{S}_i{\select_i^{\policy}}+ \sum_{i\in[n]}\theta^{I}_i{\inspect_i^{\policy}}} 
\end{align}

\revcolor{
\begin{proposition}[\textbf{Properties of $\boldsymbol{\DualLagrangeConst}$}]
\label{lemma:Const:DualLagrangeProperties}
The Lagrange dual function \(\DualLagrangeConst\) (\cref{eq:dual}) satisfies the following:

\begin{enumerate}[label=(\roman*),leftmargin=0.22in]
    \item  \(\DualLagrangeConst\) is a bounded and piecewise-linear convex function.
    
    \item There exists a minimizer \(\LagVector^* \in \arg\min_{\LagVector \in \mathbb{R}} \DualLagrangeConst(\LagVector)\), and this minimizer is bounded in absolute value by an instance-dependent constant.
    
    \item For every \(\LagVector\), and for any optimal policy \(\policy^{\LagVector}\) in the corresponding adjusted instance (defined formally in \Cref{eq:adjusted-instance}), the constraint slack \(\slack_{\textsc{cons}}^{\policy^{\LagVector}}\) is a subgradient of \(\DualLagrangeConst\) at \(\LagVector\).
\end{enumerate}


    
\end{proposition}
}

 \begin{figure}[htbp]
    \centering
    \footnotesize
    \begin{subfigure}{0.48\textwidth}
        \centering
        \includegraphics[width=\linewidth]{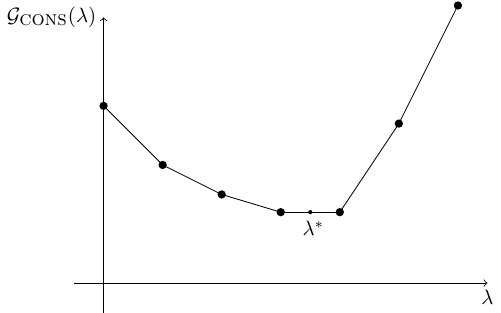}
        \caption{\footnotesize{Degenerate case, differentiable at $\LagVector^*$.}}
        \label{fig:first_sub}
    \end{subfigure}
    \hfill
    \begin{subfigure}{0.48\textwidth}
        \centering
        \includegraphics[width=\linewidth]{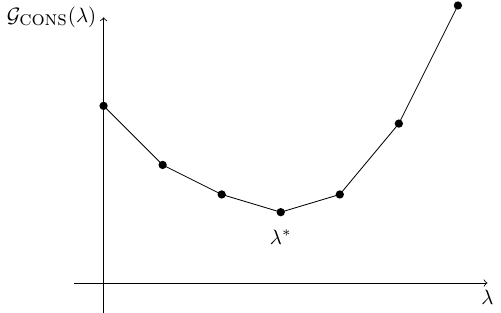}
        \caption{\footnotesize{Non-degenerate case, breakpoint at $\LagVector^*$.}}
        \label{fig:second_sub}
    \end{subfigure}
    \caption{The Lagrange dual function $\boldsymbol{\DualLagrangeConst}$ as a function of $\boldsymbol{\LagVector}$}
    \label{fig:both_figs}
\end{figure}

The above proposition also implies that a global minimum \(\LagVector^*\) of \(\DualLagrangeConst\) can be efficiently computed using binary search or standard convex optimization methods such as gradient descent~\citep{bubeck2015convex} \footnote{\revcolorm{Even though, theoretically speaking, one should be able to find $\lambda^*$ \emph{exactly} in polynomial time using convex minimization (for any input instance with bounded bit complexity), in practice one may suffice to instead find a $\tilde{\lambda}$ that is $\varepsilon$-close to $\lambda^*$, then its easy to show we can still obtain a randomized policy that satisfies the ex-ante constraint exactly, but now may have an additional $\mathcal{O}(\varepsilon)$ error in the objective versus the optimal constrained policy.}\label{footnote:error_from_approximating_optimal_dual}}. Since \(\DualLagrangeConst\) is a piecewise-linear function, if it is differentiable at \(\LagVector^*\) (which occurs in the degenerate case when \(\LagVector^*\) is not a breakpoint; see \Cref{fig:first_sub}), then it must have a slope of zero. By part~(iii) of the above proposition (a simple application of the envelope theorem), the slope of \(\DualLagrangeConst\) at \(\LagVector^*\) is given by \(\slack_{\textsc{cons}}^{\policy^{\LagVector^*}}\), where \(\policy^{\LagVector^*}\) denotes an optimal policy for the adjusted instance corresponding to \(\LagVector^*\). In this case, we are done because \(\policy^{\LagVector^*}\) satisfies Constraint~\ref{eq:affine-constraint} and is thus a solution to \eqref{eq:opt-constrained}. However, if \(\DualLagrangeConst\) is not differentiable at \(\LagVector^*\) (which is typically the case when the instance is in general position; see~\Cref{fig:second_sub}), then there exist multiple optimal policies \(\policy^{\LagVector^*}\) for the adjusted instance corresponding to \(\LagVector^*\), each associated with a different tie-breaking rule and yielding different constraint slacks \(\slack_{\textsc{cons}}^{\policy^{\LagVector^*}}\) (positive or negative). Finding the optimal policy now involves randomizing over these different tie-breaking rules.

Due to the nature of our problem, there may be exponentially many tie-breaking rules to randomize over, since each deterministic rule is a (possibly adaptive) total ordering over boxes, making it challenging to compute the optimal policy. Nevertheless, we show that we only need to consider \emph{two} specific tie-breaking rules to solve our problem, and these can be computed in polynomial-time.

\subsubsection{Randomized Tie-Breaking}
\label{sec:randomized-tie}
To design a randomized optimal policy following the recipe suggested earlier, we first define  ``extreme'' deterministic tie-breaking rules as follows, which turns out to play a critical role in our final policy.
\smallskip
\begin{definition}[\textbf{Extreme Tie-Breaking Rules}]
\label{def:tie:extreme}
Given any set of candidates $\candidates$ for breaking ties at any point during the execution of \Cref{alg:Pandora} (Line 7), the \emph{negative-extreme rule}, 
denoted by $\tierule^{-}$, assigns a \emph{tie-breaking score} $s^-_i\in\mathbb{R}$ to each $i\in\candidates$ as follows (here, given the set of selected boxes $\mathcal{S}$ and option values $\{o_i\}$ at this point in the execution of the algorithm, $o_{\textrm{max}}\triangleq \underset{i\in([\altnum]\setminus\mathcal{S})\cup\{0\}}{\max}~o_i$):
\begin{itemize}[leftmargin=*]
    \item For $i\displaystyle \in\mathcal{C}\setminus\mathcal{O}$:
    \begin{itemize}
        \item If $\displaystyle c_i<0$, set $\displaystyle s^-_i\leftarrow +\infty$.
        \item If $\displaystyle c_i\geq 0$ and $ \displaystyle\theta_i^{I}\geq 0$, set $\displaystyle s^-_i\leftarrow \theta_i^S+\frac{\theta_i^I}{\prob{v_i>o_{\textrm{max}}}}$ (set $s^-_i=+\infty$ if $\prob{v_i> o_{\textrm{max}}}=0$).
        \item If $\displaystyle c_i\geq 0$ and $\displaystyle \theta_i^{I}<0$, set $\displaystyle s^-_i\leftarrow \theta_i^S+\frac{\theta_i^I}{\prob{v_i\geq o_{\textrm{max}}}}$ (set $s^-_i=-\infty$ if $\prob{v_i\geq o_{\textrm{max}}}=0$).
    \end{itemize}
       \item For $\displaystyle i\in \candidates\cap\mathcal{O}$:
       \begin{itemize}[leftmargin=*]
           \item If $i\neq 0$, set $\displaystyle s^-_i\leftarrow \theta_i^S$, and if $i=0$ (that is, outside option), set $s^-_i\leftarrow 0$.
       \end{itemize}
    \end{itemize} 
Similarly, the counterpart rule, calling it {\em positive-extreme rule} and denote it by $\tierule^{+}$, assigns a \emph{tie-breaking score} $s^+_i\in\mathbb{R}$ to each $i\in\candidates$ as follows: 
\begin{itemize}[leftmargin=*]
    \item For $i\displaystyle \in\mathcal{C}\setminus\mathcal{O}$:
    \begin{itemize}[leftmargin=*]
        \item If $\displaystyle c_i<0$, set $\displaystyle s^+_i\leftarrow +\infty$.
        \item If $\displaystyle c_i\geq 0$ and $ \displaystyle\theta_i^{I}\leq 0$, set $\displaystyle s^+_i\leftarrow -\theta_i^S-\frac{\theta_i^I}{\prob{v_i>o_{\textrm{max}}}}$ (set $s^+_i=+\infty$ if $\prob{v_i> o_{\textrm{max}}}=0$).
        \item If $\displaystyle c_i\geq 0$ and $\displaystyle \theta_i^{I}>0$, set $\displaystyle s^+_i\leftarrow -\theta_i^S-\frac{\theta_i^I}{\prob{v_i\geq o_{\textrm{max}}}}$ (set $s^+_i=-\infty$ if $\prob{v_i\geq o_{\textrm{max}}}=0$).
    \end{itemize}
       \item For $\displaystyle i\in \candidates\cap\mathcal{O}$:
       \begin{itemize}[leftmargin=*]
           \item If $i\neq 0$, set $\displaystyle s^+_i\leftarrow -\theta_i^S$, and if $i=0$ (that is, outside option), set $s^+_i\leftarrow 0$.
       \end{itemize}
    \end{itemize} 
Then, the rule $\tierule^{-}$ (resp. $\tierule^{+}$) breaks the ties in favor of scores $\{s^-_i\}_{i\in\candidates}$ (resp. $\{s^+_i\}_{i\in\candidates}$), that is, it returns any $\displaystyle i^*\in \underset{i\in\candidates}{\argmax}~s^-_i$ (resp. any $\displaystyle i^*\in \underset{i\in\candidates}{\argmax}~s^+_i$). 
\end{definition}

\smallskip

To gain more intuition about these tie-breaking scores, consider the special case  \ref{eq:parity} in selection, assuming there are no boxes with negative or zero costs. Given two demographic groups, \(\ManSet\) and \(\WomanSet\), each extreme tie-breaking rule corresponds to assigning a score of \(+1\) to one group and \(-1\) to the other, thereby breaking all ties entirely in favor of one group over the other. This approach maximizes the probability of selection from the preferred group as much as possible.

As another example, consider \ref{eq:budget} in selection, when \(e_i = 1\) for all candidates \(i\) in a special group \(\WomanSet\) and \(e_i = 0\) otherwise. In this case, each extreme tie-breaking rule assigns a score of \(+1\) to boxes in \(\WomanSet\) and \(0\) to all other boxes, or assigns \(-1\) to boxes in \(\WomanSet\) and \(0\) to all others. The outside option always has a score of \(0\); therefore, in situations with ties, each extreme rule either always prefers boxes not in \(\WomanSet\) over those in \(\WomanSet\), or always prefers boxes in \(\WomanSet\) over the others.


Building on this intuition, we formally show that these two extreme tie-breaking rules correspond to the policies that minimize and maximize the constraint slack \(\Delta_{\textsc{cons}}\) among all optimal policies for any given adjustment \(\LagVector\). Consequently, we can find an optimal policy for Problem~\ref{eq:opt-constrained} by first adjusting the rewards and costs using \(\LagVector^*\) as in \eqref{eq:adjusted-instance}, and then randomizing over \emph{only two} index-based optimal policies. These policies are obtained by running \Cref{alg:Pandora} on the adjusted instance, with \(\tierule^{+}\) and \(\tierule^{-}\) as the tie-breaking rules. We refer to these policies as \(\policyplus\) and \(\policyminus\), respectively.



\begin{proposition}[\textbf{Slack Signs for Extreme Rules}]  
\label{lem:const:slack}
For the two extreme tie-breaking rules $\tierule^{+}$ and $\tierule^{-}$ (as in \Cref{def:tie:extreme}), and their corresponding index-based optimal policies $\policyplus$ and $\policyminus$ for the adjusted instance (as defined in \eqref{eq:adjusted-instance}) with $\LagVector^*\in\underset{\LagVector\in\mathbb{R}}{\argmin}~\DualLagrangeConst(\LagVector)$, we have $\slack_{\textsc{cons}}^{\policyplus} \geq 0 \geq \slack_{\textsc{cons}}^{\policyminus}$.
\end{proposition}
\begin{proof}{\emph{Proof sketch.}}
The proof consists of two main steps. First, we show that \(\DualLagrangeConst(\LagVector^*)\) admits an optimal policy with nonpositive (resp. nonnegative) slack. This policy is also the optimal policy used in the problem of computing \(\DualLagrangeConst(\LagVector^* - \varepsilon)\) (resp. \(\DualLagrangeConst(\LagVector^* + \varepsilon)\)) for a sufficiently small perturbation \(\varepsilon > 0\). This step relies on the properties of \(\DualLagrangeConst\) established in \Cref{lemma:Const:DualLagrangeProperties}, particularly its piecewise linearity and convexity. In the second step, we show that for an \emph{infinitesimal} \(\varepsilon > 0\) in the reward-adjusted problem with adjustment \(\LagVector^* - \varepsilon\) as defined in \eqref{eq:adjusted-instance}, the ordering of adjusted option values produced by an optimal index-based policy (implemented by \Cref{alg:Pandora}) directly determines a corresponding tie-breaking rule for the dual-adjusted problem with adjustment \(\LagVector^*\). If we assume that this policy breaks ties among zero-cost boxes by prioritizing those with \(\theta_i^I \geq 0\) and treating the rest as normal boxes, the resulting tie-breaking rule \emph{exactly} matches \(\tierule^{-}\) in \Cref{def:tie:extreme}. Similarly, applying the same reasoning with the perturbed dual adjustment \(\LagVector^* + \varepsilon\) for an infinitesimal \(\varepsilon > 0\) recovers the other extreme tie-breaking rule \(\tierule^{+}\). We defer all proof details to \Cref{app:sec:pandora}.\qed
\end{proof}
\begin{remark}
\label{remark:extreme-cons}
As mentioned above, it turns out that we can even establish a stronger statement than \Cref{lem:const:slack}: Among all the optimal policies $\pi^{\lambda^*}$ for the dual-adjusted instance with $\lambda^*$, the optimal policy $\policyminus$ (resp. $\policyplus$) with tie-breaking rule $\tierule^{-}$ (resp. $\tierule^{+}$) has the minimum (resp. maximum) amount of the constraint slack equal to $\slack_{\textsc{cons}}^{\policyminus}$ (resp. $\slack_{\textsc{cons}}^{\policyplus}$). See the proof in \Cref{app:sec:pandora}.
\end{remark}

We highlight the important implication of \Cref{lem:const:slack}: By properly randomizing between the two extreme tie-breaking rules, we can construct an optimal policy with zero slack. We formalize this construction in \Cref{alg:const}. We now arrive at the main result of this section:

\begin{algorithm}
\caption{Randomized Dual-adjusted Index Policy (RDIP)}
\label{alg:const}
    \SetKwInOut{Input}{input}
    \SetKwInOut{Output}{output}
\Input{instance $\displaystyle\left\{\left(\values_i, F_i, c_i\right)| i \in [n]\right\}$}

 Compute  $\LagVector^*\in \underset{{\LagVector\in\mathbb{R}}}{\argmin}~\DualLagrangeConst (\LagVector)$  {\color{\commentcolor}\tcc{single-dimensional convex optimization}}
 
  \vspace{1mm}
 Define the adjusted instance $\{({\widetilde{\values}_i}, \widetilde{F_i}, \widetilde{c}_i)| i \in [n]\}$  with $\widetilde{v_i} = v_i-\LagVector^*\cdot\theta_i^S,$ and $\widetilde{c_i}=c_i+\LagVector^*\cdot\theta_i^I$ \\

 \vspace{1mm}

\For{$\texttt{sign}\in\{+,-\}$}{
\vspace{1mm}
Run \Cref{alg:Pandora} with inputs $\displaystyle\{({\widetilde{\values}_i}, \widetilde{F_i}, \widetilde{c}_i)| i \in [n]\}$  as the instance and $\displaystyle\tierule^{\texttt{sign}}$ (from \Cref{def:tie:extreme}) as the tie-breaking rule; call this algorithm $\policy^{\texttt{sign}}$.

Calculate the slack of Constraint~\ref{eq:affine-constraint}
 $\displaystyle\slack_{\textsc{cons}}^{\policy^\texttt{sign}}$ (defined in \cref{eq:slack}).\\
}






 \vspace{2mm}
{\bf Policy:} If $\slack_{\textsc{cons}}^{\policy^+}=\slack_{\textsc{cons}}^{\policy^-}=0$, then run $\policy^{-}$. Otherwise, with probability $\slack_{\textsc{cons}}^{\policy^+}/\left({\slack_{\textsc{cons}}^{\policy^+}-\slack_{\textsc{cons}}^{\policy^-}}\right)$ run $\policy^{-}$ and with probability $-\slack_{\textsc{cons}}^{\policy^-}/\left({\slack_{\textsc{cons}}^{\policy^+}-\slack_{\textsc{cons}}^{\policy^-}}\right)$ run $\policy^{+}$. 

\end{algorithm}
\begin{theorem}[\textbf{Optimal Policy for Constrained Problem}]
\label{thm:const}
The policy RDIP (presented in \Cref{alg:const}) is an optimal policy for the constrained Pandora's box problem with multiple selection, defined in \eqref{eq:opt-constrained}, under an ex-ante affine constraint. 
\end{theorem}

\revcolor{We defer the proof of the above theorem, which builds on the earlier propositions, to \Cref{app:sec:pandora}. Instead, we conclude with a few remarks on managerial insights of our results:

\begin{itemize}[leftmargin=*]
    \item As discussed earlier in \Cref{sec:lagrangian-pandora}, our proposed dual adjustment for \ref{eq:parity} in selection is both intuitive and economically interpretable. Specifically, compared to the optimal unconstrained policy, this adjustment increases the selection probability for the under-represented group while maintaining the within-group ordering of candidates (see \Cref{app:pandora:managerial} for more details).
    \item A delicate primitive of our proposed policy is interleaving inspections between the two groups based on dual adjustments and implementing a specific randomized tie-breaking rule. Both the adjustment and the tie-breaking are crucial; alternative methods for either would lead to an optimality gap (see \Cref{example:FS} in \Cref{sec:insights-dem-parity} and \Cref{example:tie} in \Cref{sec:beyond-group} for details).
\end{itemize}}

\subsection{Extensions}
\label{sec:pandora-extensions}
Going beyond a single affine constraint on marginal probabilities, we extend our results to settings with (i) single affine constraint on probabilities conditional on candidate values and (ii) multiple affine constraints. Similar to \Cref{sec:single-affine-policy}, our goal is to characterize and compute an optimal policy for the constrained problem that \emph{exactly} satisfies these new ex-ante affine constraints. We briefly overview the settings in this section below, and defer the full details to the electronic supplement (\Cref{app:sec:pandora-extension} and \Cref{app:mutiple-affine}).\footnote{Notably, both of these settings are encompassed by our more general model in \Cref{sec:general}; the primary difference is that here we seek an \emph{exact} optimal constrained policy.}


\subsubsection{Value-specific Constraints} 
\label{sec:quantile-constraint}
In contexts of fairness and inclusion in hiring, decision-makers may want to fine-tune ex-ante constraints to account for the heterogeneity in candidates' values and inspection costs, rather than applying a blanket approach. To capture this, we generalize our earlier ex-ante affine constraint (Constraint~\ref{eq:affine-constraint}) by allowing the coefficients \(\theta_i^S\) and \(\theta_i^I\) for each candidate \(i\) to be \emph{arbitrary functions} of their reward \(v_i\) and inspection cost \(c_i\) (assuming the constraint is again an equality without loss of generality):
\begin{equation}
\label{eq:general-constraint}
\expect{\sum_{i\in[n]}\theta^{S}_i(v_i,c_i){\select_i^{\policy}}+ \sum_{i\in[n]}\theta^{I}_i(v_i,c_i){\inspect_i^{\policy}}}=b~~
\end{equation}
Constraint~\eqref{eq:general-constraint} can be viewed as an affine constraint on the probabilities of selection and inspection conditional on each candidate's specific values. These value-specific constraints are motivated by scenarios where value-independent constraints fail to achieve their intended purpose. For example, suppose a firm aims to hire one candidate and must respect demographic parity in inspections, meaning the expected number of interviews from both groups must be equal. If the minority group comprises both high-quality candidates with high inspection costs and low-quality candidates with low costs, enforcing parity without considering values could lead to ``token" interviews---only interviewing low-quality, low-cost minority candidates to satisfy the constraint. While this maintains parity, it fails to provide equal opportunity. By imposing value-dependent constraints, we ensure that only high-quality minority candidates are counted toward achieving parity, aligning the constraint with the goal of ``real'' equal opportunity. \revcolor{All of our results in \Cref{sec:single-affine-policy} extend to this setting after proper non-trivial adaptations. We postpone all technical details to \Cref{app:sec:pandora-extension}; in particular, see \Cref{eq:general-adjusted-instance} for the definition of refined version of our dual-adjusted instance, \Cref{def:refined-ext-tie} for the extension of our extreme tie-breaking rules to this setting, and \Cref{prop:general-constraint} for the characterization of the optimal constrained policy as a dual-adjusted index-based policy with randomization over two extreme tie-breaking rules.

\subsubsection{Multiple Affine Constraints \& Connections to Algorithmic Carathéodory}
\label{sec:caratheodory}
In certain applications, it may be desirable to satisfy multiple affine constraints. For example, one might combine a socially aware affine constraint such as \eqref{eq:parity} in selection, with a \eqref{eq:quota} in inspection to ensure a minimum on the expected number of interviews from the minority group for inclusion in screening. Motivated by such applications, here we study a generalization of \eqref{eq:opt-constrained}, this time with $m>1$ affine constraints.
\smallskip

\noindent\textbf{Sketch of our approach:}
To characterize the optimal constrained policy, we follow a similar approach to our earlier investigation. We begin by Lagrangifying all ex-ante affine constraints into the objective, defining the Lagrangian relaxation/dual function as before. As in \Cref{sec:single-affine-policy}, we demonstrate that the policy maximizing the Lagrangian relaxation corresponds to an optimal policy for an adjusted problem instance using Lagrangian duals (see, e.g., \Cref{subsection: reduction of equality constrained to Caratheodory}). We then find the optimal set of dual variables \(\boldsymbol{\lambda}^* = \{\lambda^*_i\}_{i \in m}\) using convex optimization, given oracle access to the Lagrangian dual function and its sub-gradient via computing dual-adjusted optimal policies. Following a similar line of reasoning, the optimal constrained policy is a dual-adjusted index-based policy with a randomized tie-breaking rule, thus effectively a convex combination of deterministic dual-adjusted optimal policies. The key remaining question is whether we can identify a polynomial number of these dual-adjusted optimal policies such that an appropriate randomization among them satisfies all affine constraints \emph{exactly}.
\smallskip


\noindent\textbf{Failure of extreme tie-breaking rules:}
Based on our previous results, one might consider randomizing over \(2^m\) policies obtained by perturbing \(\boldsymbol{\lambda}^*\) with infinitesimal perturbations \(\vec{\varepsilon} = [\pm \varepsilon]_{i \in [m]}\)---a natural extension of our earlier ``extreme tie-breaking rules'' to multiple constraints. This approach would require a convex combination of exponentially many policies. However, even ignoring computational complexity of this approach, we prove in \Cref{app:extreme-failure-multiple-affine} that this method fails by providing a simple example with two constraints where no convex combination of the resulting (possibly) four policies achieves zero slack for both constraints.

\smallskip

 
\noindent\textbf{Reduction to algorithmic Carathéodory:} 
Despite the negative result, we address the key question by reducing our problem to a specific instance of the classical algorithmic Carathéodory problem~\citep{caratheodory1911variabilitatsbereich}. \revcolorm{In particular, we propose a novel  polynomial-time algorithm for \emph{exact} Carathéodory in a polytope with potentially exponentially many vertices, given oracle access to an algorithm that can solve linear optimization over this polytope. For any point in the polytope, the algorithm finds a polynomial-size convex combination of vertices that equals that point. In a nutshell, our algorithm runs a variant of the standard Ellipsoid algorithm~\citep{grotschel1981ellipsoid} in the ``dual space'' by searching for a small set of directions, such that maximizing vertices along those directions cover the original point in their convex hull; see \Cref{alg:Ellipsoid Exact Caratheodory} in \Cref{app:sec:exact-caratheodory-algorithm}.} We apply this algorithm to our setting via a reduction in which the vertices represent the (constraint slacks of) dual-adjusted index-based policies for the Pandora's box problem, and the oracle corresponds to computing an index-based optimal policy for a general Pandora's box instance (which can be computed in polynomial-time).

\revcolorm{By combining this reduction with our exact Carathéodory algorithm, we obtain a polynomial-time procedure to compute the exact optimal policy under multiple affine constraints. 
Importantly, while we describe this result in the context of the Pandora's box problem for simplicity of exposition, our approach is in fact more general and can be readily extended to the broader Joint Markovian Scheduling (JMS) problem, which we formally define and study next in \Cref{sec:general}. Roughly speaking, the key component of the Carathéodory-based argument is the ability to solve the dual-adjusted problem using a polynomial-time computable policy---a property that also holds in the JMS setting. As a result, one can design randomized dual-adjusted index-based policies in the JMS setting, by following the exact same recipe, that satisfies one or multiple affine constraints \emph{exactly}. For technical details and formal statements, see \Cref{app:mutiple-affine}.}

We note that our approach may also be of independent interest for other applications of exact Carathéodory, where the only access to a polytope (with potentially exponentially many vertices) is through a linear optimization oracle. \revcolorm{For instance, this technique can be used in computing the optimal revenue Bayesian Incentive Compatible (BIC) mechanism for multidimensional types, by first computing the exact BIC reduced-form allocation rule of this mechanism (derivable via polynomial-time LPs if social welfare maximization is polynomial-time computable), and then decomposing this reduced-form allocation rule into a distribution over ex-post feasible deterministic BIC allocation rules to obtain an \emph{exact} BIC mechanism at the end with the same expected revenue. Importantly, the linear optimization oracle in this context is equivalent to (virtual) social welfare maximization and in many settings is polynomial-time computable~\citep{cai2012optimal,alaei2014bayesian}.}}

\newcommand{\capacity}{k}

\section{Markovian Sequential Search with General Ex-ante Constraints}
\label{sec:general}

\revcolor{The Pandora's box problem is a simplified abstraction of sequential search and selection in the real world. In fact, many real-world search processes are more complex and involve multiple screening stages, as well as various rounds of communication with candidates. Here is an example.}

\revcolor{
\smallskip
\begin{example}[Multi-stage Search with Rejection]
\label{ex:reject}
The hiring process for many jobs involves two stages of inspection. First, there is a low-cost stage, such as a phone interview, which provides basic pass/fail information. The second stage, which is more expensive, typically involves an on-site visit and gives a detailed assessment of the candidate's quality. At any point, the hiring firm has three options: initiate the first stage for a new candidate, proceed to the second stage for a candidate who passed the first, or extend an offer to a candidate who has undergone both stages. The offer may be declined with a certain probability, in which case the firm resumes the search.
\end{example}}
\smallskip

\revcolor{This is an example of a more general \emph{``stateful''} sequential search process---beyond the Pandora's box model studied in \Cref{sec:pandora}---where the state of a candidate evolves after each interaction, possibly in a stochastic fashion. This state captures where the candidate is in the search process, which identifies the cost of further inspection or the realized reward of selection if the candidate is ready to be hired. The stateful nature of such processes leads to highly complex policies, which may introduce disparities at different stages of the search or in the final outcomes. \revcolorm{To extend our investigation of imposing (socially aware or operational) ex-ante} constraints on search outcomes to such more complex stateful search scenarios, such as the example above, we adopt the Joint Markov Scheduling (JMS) model of sequential search~\citep{dumitriu2003playing}, which generalizes Pandora's box. Furthermore, to capture a more comprehensive notion of ex-ante constraints, we consider scenarios with \emph{multiple affine or convex constraints} on the \emph{visit frequencies of different states}. It is important to note that the setting discussed in this section is general and encompasses both the value-specific constraint setting in \Cref{sec:quantile-constraint} and the multiple-affine constraints setting in \Cref{sec:caratheodory} as special cases. The main difference is that in this section, we aim to compute \emph{near-optimal} and \emph{near-feasible} policies, rather than exactly optimal and feasible policies. We will formally define this setting and the constraints next.}

\subsection{Setting and Notations}
\label{sec:JMS-setting}
We consider a Markovian system with finitely many alternatives indexed by $[\altnum]\triangleq\{1,2,\ldots,n\}$ and an outside alternative indexed by $0$. Each alternative $i\in[\altnum]$ is modeled as a finite Markov reward process $\MC_i=(\states_i,\terminals_i,\transition_i,\reward_i)$, where $\states_i$ is the finite set of states, $\terminals_i\subset \states_i$ is a special subset of states called terminal states, $\transition_i:\states_i\times\states_i\rightarrow[0,1]$ is the transition matrix, and $\reward_i$ is the vector of all state-rewards for alternative $i$. In particular, $\reward_i(s)\in\mathbb{R}$ for the non-terminal state $s\in\states_i\setminus\terminals_i$ is the reward of making alternative $i$ to exit state $s$, and for the terminal state $s\in\terminals_i$ is the reward of entering $s$. Note that any terminal state $s\in\terminals_i$ is absorbing, that is, there is no transition from $s$ to any other state in $\states_i$. For simplicity, let $\vecreward\triangleq [\reward_i(s)]_{i\in[\altnum],s\in\states_i}\in\mathbb{R}^{\dimension}$ denote the concatenation of the state-reward vectors of all alternatives, where the (finite) dimension $d$ is defined as $\dimension\triangleq\sum_{i\in[\altnum]}\lvert{\states_i}\rvert$. We also occasionally index the set of all states $\cup_{i\in[\altnum]}\states_i$ by $[\dimension]=\{1,2,\ldots,\dimension\}$. Further, we assume the rewards are bounded, and therefore without the loss of generality are normalized such that  $\lVert\vecreward\rVert_{+\infty}\leq 1$.\footnote{We only impose this assumption on the original rewards; Our dual-adjusted rewards, introduced later in the section, do not need to be bounded or normalized between $[-1,1]$.}

\revcolor{Starting with an initial configuration of states $(s^{(0)}_1,\ldots,s^{(0)}_n)$ for the alternatives, a decision maker interacts with the Markovian system in discrete rounds. In each round $t = 1, 2, \ldots$, she chooses to either inspect an alternative $i_t \in [n]$ or select the outside alternative $i_t = 0$. If she inspects $i_t$, she collects a reward $\reward_{i_t}(s)$ based on the current state $s$ of the Markov chain $\MC_{i_t}$, which then undergoes a probabilistic transition (according to the transition matrix $\transition_{i_t}$) to a new state $s'$. If the new state $s'$ is a terminal state, that is, $s' \in \terminals_{i_t}$, the decision maker adds $i_t$ to the final set of selected alternatives and collects an additional reward $R_{i_t}(s')$. We consider the case where the decision maker has a capacity $\capacity \in \mathbb{N}$, meaning that at most $\capacity$ alternatives can be in the selected set (or equivalently, their Markov chains be in terminal states) at any time.\footnote{\revcolor{Our results extend straightforwardly to a more general matroid environment, where $[n]$ is the ground set of a matroid, and the decision maker must ensure that in each round, the set of alternatives in terminal states is an independent set of this matroid. For simplicity, we focus on the case of the $\capacity$-uniform matroid, where the capacity is $\capacity$.}} The search process terminates when the decision maker selects the outside alternative $0$ or reaches the capacity $\capacity$ for selected alternative. Otherwise, the process proceeds to the next round, and the decision maker selects a new alternative to inspect.}

\revcolor{The goal of the decision maker is to maximize the expected accumulated reward before the process ends. A policy $\policy$ for the decision maker is a mapping that, at each time, assigns the \emph{history}---the sequence of previous actions and realized states of all Markov chains up to the current time---to one of the unselected alternatives in $[\altnum]$ or the outside alternative. A deterministic stationary policy $\pi$ is a fixed mapping from the current state configuration $(s_1,\ldots,s_{\altnum})$ to $[\altnum]\cup{0}$. While policies can be non-stationary, we focus on the set $\PolicySpace$ of (possibly randomized) stationary policies that ensure no more than $\capacity$ alternatives are selected in any sample path.\footnote{A stationary optimal policy always exists for the basic JMS without ex-ante constraints~\citep{dumitriu2003playing,gupta2019markovian}. With ex-ante constraints, as shown in our analysis, this restriction is without loss because a (near-optimal, near-feasible) solution exists in this class if the instance is feasible (Assumption~\ref{assumption:feasible-fair}). We omit details for brevity and refer the reader to the discussion in \Cref{sec:RAI}.}}


\revcolor{For a given policy $\policy \in \PolicySpace$, let $R_\policy$ denote the accumulated reward realized by $\policy$ until termination. We denote by $\alloc_\policy \in \mathbb{R}_{\geq 0}^\dimension$ the vector of expected ``number of visits'' to different states in $\cup_{i \in [\altnum]} \states_i$ before termination under $\policy$. By convention, for non-terminal states, we count the number of times we exit the state as its number of visits; for terminal states, we count a visit when we enter the state (since the corresponding Markov chain is selected). Because of the linearity of expectations, we then have:}
\begin{equation}   
\expect{R_\pi}=\expect{\sum_{i\in[\altnum]}\sum_{s\in \states_i}(\textrm{$\#$ of visits of state $s$ under $\policy$})\times\reward_i(s)}=\vecreward\cdot\alloc_\pi~.
\end{equation}


We further assume that each Markov chain $\MC_i$ is absorbing---that is, it has at least one absorbing (terminal) state, and from each non-terminal state, there is a path with nonzero probability to a terminal state. Since the Markov chains are finite and absorbing, it follows that there exists a constant $\barP \in \mathbb{R}{+}$ such that, for every policy $\policy$, the expected number of visits to each state $s \in \cup_{i\in[\altnum]}\states_i$ before absorption is bounded above by $\barP$~\citep{resnick1992adventures}. Let $\FeasibleAlloc \subseteq [0, \barP]^{\dimension}$ \revcolor{denote the space of implementable expected visit numbers by admissible stationary policies, that is,}
\begin{equation}
\label{eq:admissible}
    \FeasibleAlloc\triangleq \left\{\alloc\in [0,\barP]^{\dimension} \mid\exists \policy\in\PolicySpace: \alloc_\pi=\alloc\right\}~.
\end{equation}
Note that $\FeasibleAlloc$ is compact and convex as randomized policies are allowed.\footnote{The space of deterministic stationary policies for JMS is finite, as each stationary policy is a mapping from current states of MCs to an index. Therefore, $\mathcal{P}$ becomes a polytope with finitely many vertices and hence compact.}

We finally highlight that JMS is an extensive and general model. For example, see how primitives of JMS help us model both the Pandora's box problem (\Cref{fig:pandora-jms}) and the multi-stage hiring with rejection in \Cref{ex:reject}~(\Cref{fig:MC-hiring}). With the unconstrained JMS problem explained, we next move on to the general type of ex-ante constraints that we aim to capture in this paper.

\begin{figure}[htb]
	\centering
	\begin{subfigure}[b]{0.25\textwidth}
            \includegraphics[width=\textwidth]{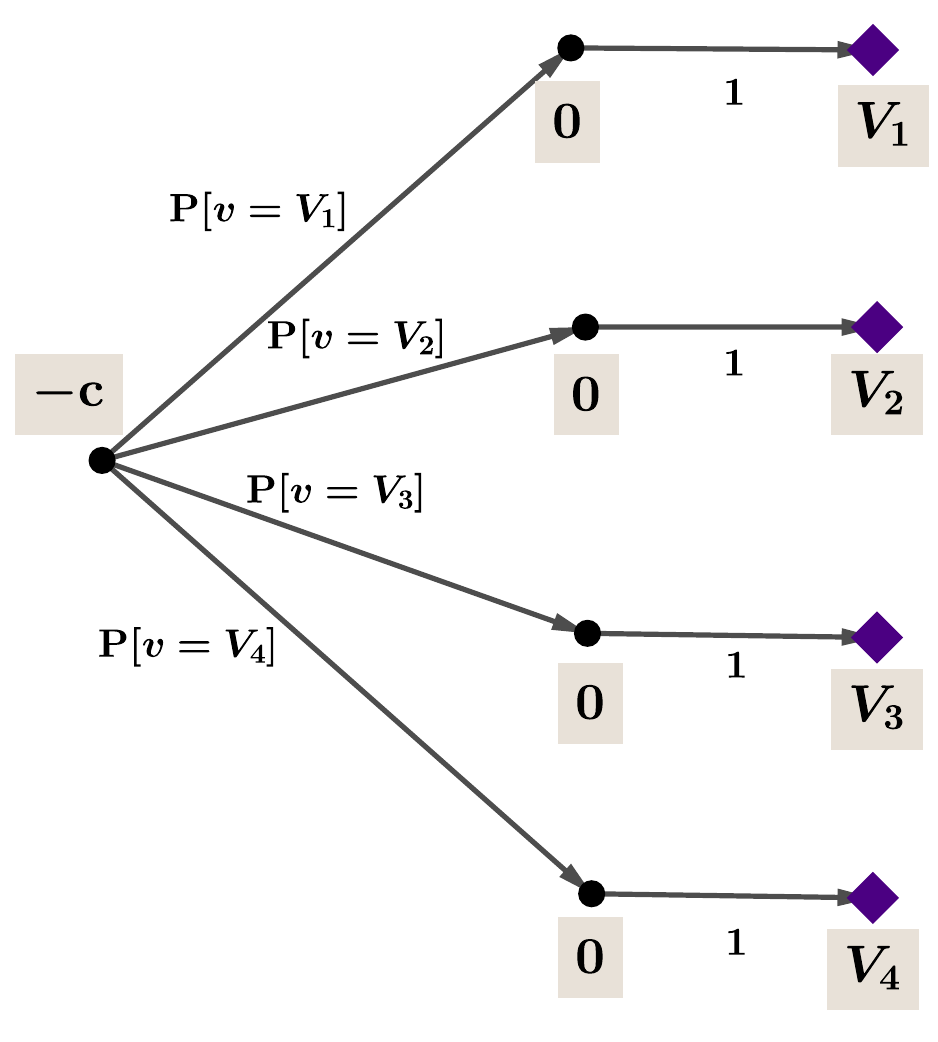}
             \caption{\label{fig:pandora-jms}}
	\end{subfigure}
	\begin{subfigure}[b]{0.68\textwidth}
	  \includegraphics[width=\textwidth]{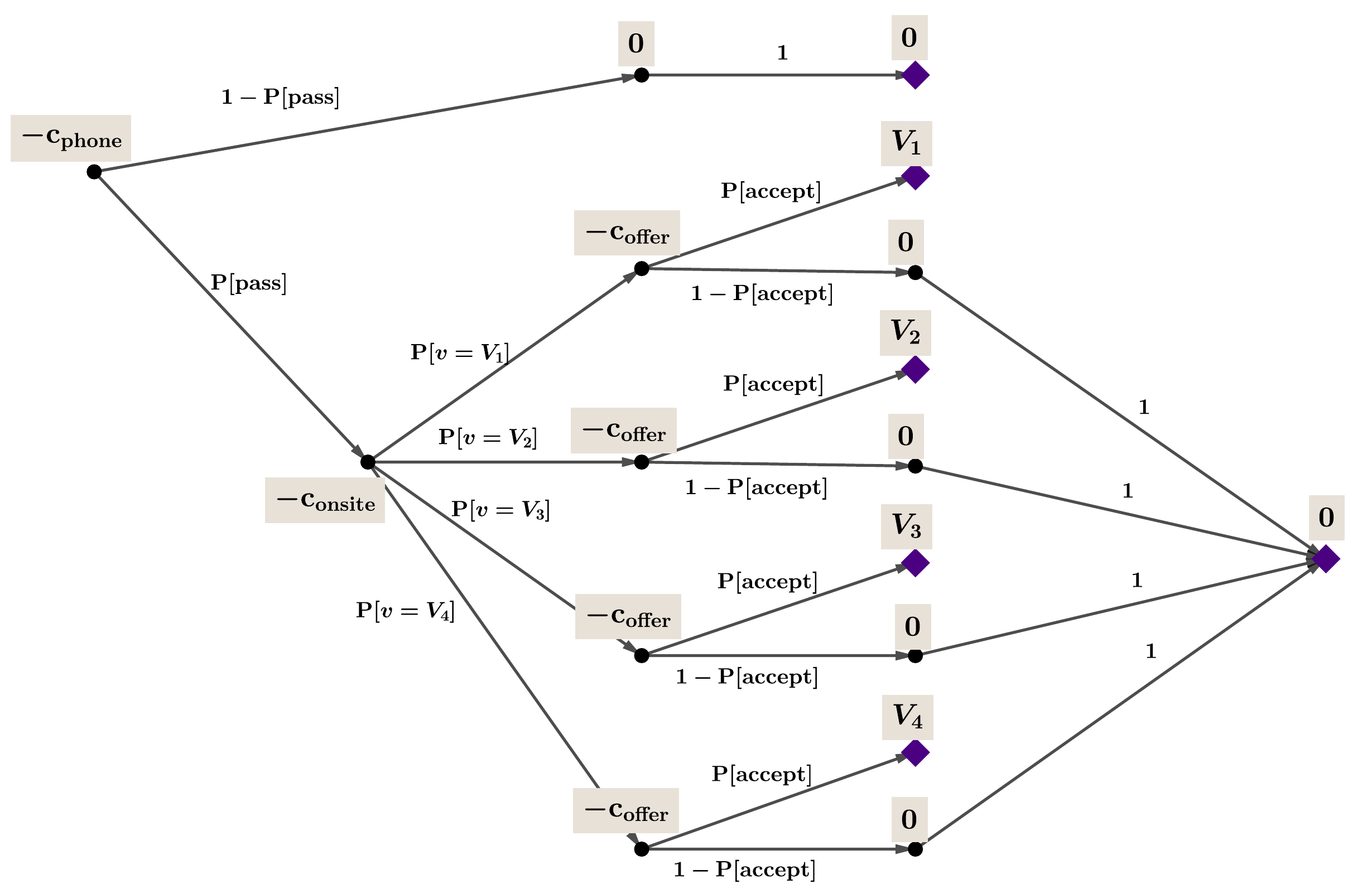}
	  \caption{\label{fig:MC-hiring}}
	\end{subfigure}
	\vspace{1mm}
	\caption{\revcolor{Candidates as Markov reward processes in JMS: The numbers on the states represent rewards, and those on the edges are transition probabilities; $\CIRCLE$ and $\blacklozenge$ denote non-terminal and terminal states, respectively; (a)~\emph{Pandora's box problem}~\citep{weitzman1979optimal}: After inspection, the value is realized from $\boldsymbol{\{V_i\}_{i=1}^4}$, and then the box can be selected;  (b) \emph{Multi-stage search with rejection (\Cref{ex:reject}}): The candidate passes the phone interview with probability $\textbf{P}\left[\boldsymbol{\texttt{pass}}\right]$. If successful, there is an onsite interview, after which the her value is realized from $\boldsymbol{\{V_i\}_{i=1}^4}$. Lastly, if an offer is made, the candidate accepts it with probability $\textbf{P}\left[\boldsymbol{\texttt{accept}}\right]$}. 
 }
\end{figure}
\vspace{-1mm}
\label{sec:general-ex-ante-constraints}
\revcolor{As alluded to earlier in \Cref{sec:caratheodory}, there are scenarios in the JMS setting where multiple constraints are needed simultaneously, and not always these constraints are affine. For example, in the multi-stage search (\Cref{ex:reject}), consider imposing an average quota constraint on the final selections of a certain minority group of candidates, alongside demographic parity between males and females for the phone interview. As another example, in certain hiring contexts, candidates may benefit from advancing in the process even if not selected---for instance, by gaining experience or visibility. By modeling the probabilities of achieving these benefits as utilities, the decision maker can aim to improve a ``welfare function'' of these utilities by adding a constraint while simultaneously respecting demographic parity. Notably, many of the commonly used welfare functions are not necessarily affine, but are typically convex functions of the utilities.


Motivated by these examples, we consider two general categories of ex-ante constraints on the vector of expected numbers of visits to different states. In particular, we 
allow $\NumAffine\in \mathbb{Z}_{\geq 0}$ affine and $\NumConvex\in \mathbb{Z}_{\geq 0}$ convex constraints that can be applied to all or any subset of states---to incorporate ex-ante constraints in both the final selection (at the time of termination) and the inspection phase (during the search process before termination):
\smallskip

\begin{itemize}[leftmargin=*]
    \item \emph{Affine constraints}: each affine constraint $j\in[\NumAffine]$ is defined by the halfspace $\AffineVector_j\cdot\alloc\leq \AffineConstant_j$ for some vector $\theta\in\mathbb{R}^\dimension$ and constant $\AffineConstant_j\in\mathbb{R}$. 
    By proper normalization, without the loss of generality, we assume $\lvert\AffineConstant_j\rvert\leq 1$ and for every $\alloc\in[0,\barP]^{\dimension}$, $\lvert\AffineVector_j\cdot\alloc\rvert\leq 1$.\footnote{\revcolor{We remark that with only affine constraints, the JMS setting is quite similar to the setting studied in \Cref{sec:caratheodory}; the main difference is the linear optimization oracle, as after adjustments we have to solve a dual-adjusted JMS instance (with arbitrary rewards, possibly negative or positive). As we show in \Cref{sec:optimal-JMS-arbitrary} and \Cref{app:JMS-general}, any general instance of JMS can be solved in polynomial-time.}}
    \smallskip
    
   \noindent\underline{\emph{Applications:}} similar to \Cref{sec:pandora}, this category captures various group fairness criteria, such as demographic parity and quota for disadvantaged groups. These criteria can be implemented at the selection level or at any intermediate stage of the search process. It can also capture various forms of individual fairness constraints. For example, we can add multiple affine constraints, one for each candidate, setting lower bounds on the expected number of visits of a particular search states (including terminal states). Finally, such constraints can have operational implications, e.g., capture budget constraints at different stages of the search. \revcolorm{See \Cref{subsec: affine Constraint Formulations} for the formal representation of such constraints.}
   
    
    \smallskip
    \item \emph{Convex constraints:} each constraint $i\in[\NumConvex]$ is defined by the convex set $\ConvexFun_i(\alloc)\leq 0$, where $\ConvexFun_i:\mathbb{R}^\dimension\rightarrow \mathbb{R}$ is a strictly convex function and admits continuous first partial derivatives. We further assume that $\ConvexFun_i$ is bounded in $[0,\barP]^{\dimension}$, the gradient of $\ConvexFun_i$ diverges to infinity, i.e., $\lim_{\lVert\alloc\rVert_{\infty}\to \infty}\lVert \nabla\ConvexFun_i(\alloc)\rVert_{\infty}=+\infty$, and that it is bounded in norm infinity over $[0,\barP]^{\dimension}$, i.e., $\sup_{\alloc\in[0,\barP]^\dimension}\lVert\GradFun_i(\alloc)\rVert_{\infty}<\infty$. By proper normalization, without loss of generality, we assume $\lvert\ConvexFun_i(\alloc)\rvert\leq 1$ and $\lVert\GradFun_i(\alloc)\rVert_{\infty}\leq \upperDualConvex $ for every $\alloc\in[0,\barP]^{\dimension}$, for some $\upperDualConvex>0$. Because of diverging gradient, there also exists a threshold $\lowerp$ such that $\forall \alloc: \lVert\alloc\rVert_{\infty} > \lowerp$ we have $\lVert\GradFun_i(\alloc)\rVert_{\infty}> \upperDualConvex$.
    
  \smallskip
  \noindent\underline{\emph{Applications:}} Thinking of individual candidates' utilities from the search, as described earlier, we can use convex constraints to improve the \emph{egalitarian welfare} of the search process, which naturally leads to more diverse or fair outcomes. In fact, egalitarian welfare is typically captured by concave symmetric functions of these utilities, using notions such as Nash social welfare, negative entropy, or the Hölder mean of the utilities~\citep{kaneko1979nash,dwork2012fairness}. We can then add a lower bound constraint on egalitarian welfare. These convex constraints can be customized to apply at the level of demographic groups of candidates (basically, thinking of each group as a ``meta agent''  whose utility equals to the summation of utilities of the candidates within that group), capturing group notions of fairness. They can also be used at the level of individual candidates, capturing individual notions of fairness. This flexibility is one of the appeals of our general set of constraints. \revcolorm{See \Cref{subsec: convex Constraint Formulations} for a formal demonstration of these constraints.}
\end{itemize}

\smallskip
Given the general ex-ante constraints as described above, a policy $\policy$ for the joint Markov scheduling problem is said to be \emph{ex-ante feasible} if $\alloc_\policy\in\FairSet$, where 
\begin{equation}
    \FairSet\triangleq \left\{\alloc\in[0,\barP]^{\dimension}\mid \forall j\in[\NumAffine]:\AffineVector_j\cdot\alloc\leq \AffineConstant_j, \forall i\in[\NumConvex]:\ConvexFun_i(\alloc)\leq 0\right\}.
\end{equation}}
\noindent\textbf{Optimal constrained policy.} We then define an \emph{optimal constrained policy} $\pi^*\in\Pi$ as any solution to the following stochastic program:
\begin{equation}
\label{eq:opt-fair}
\tag{\textsc{OPT-JMS-cons}}
\begin{aligned}
\OptFair\triangleq & \quad \max_{\policy\in\PolicySpace} ~~\expect{\reward_\policy}\quad\textrm{s.t.}
 \quad & \alloc_\policy\in\FairSet &
\end{aligned}
\end{equation}
\revcolorm{We make the following assumption in the rest of this section about our problem instance.
\begin{assumption}[JMS Feasibility]
\label{assumption:feasible-fair}
Problem~\ref{eq:opt-fair} is feasible, that is,  $\FairSet\cap\FeasibleAlloc\neq \emptyset$, where $\FeasibleAlloc$ is defined as in \cref{eq:admissible} for the underlying JMS instance $\{\MC_i\}_{i\in[\altnum]}$ (for selecting $\capacity$ alternatives) in this problem.
\end{assumption}}



\revcolor{
\subsection{Near-optimal Near-feasible Constrained Policy: Sketch of the Approach}
\label{sec:RAI}

Before we start, we refer the reader to a premier on \emph{Fenchel convex duality} in \Cref{sec:convex}, where we also provide a related simple lemma on properties of the strict convex functions we use in our general ex-ante constraints (\Cref{lemma:convex-conjugate}); see also \cite{bubeck2015convex} for more details. We use these constructs in our technical framework below and in our analysis.

At a high level, our goal is to follow the approach in \Cref{sec:pandora} to obtain a dual characterization of the optimal policy. To start, we define the Lagrangian relaxation of \eqref{eq:opt-fair} as follows:
\begin{align}
\label{eq:lagrange-fair-1}
    \Lagrange(\alloc;\DualConstAffineVec,\DualConstConvexVec)\triangleq \vecreward\cdot\alloc +\sum_{j\in[\NumAffine]}\DualConstAffine_j\left(\AffineConstant_j-\AffineVector_j\cdot\alloc\right)-\sum_{i\in[\NumConvex]}\DualConstConvex_i\ConvexFun_i(\alloc)~.
\end{align}
Clearly, for any ex-ante feasible policy $\pi$ (with $\alloc_\policy\in\FairSet$),  $\expect{R_\pi}=\vecreward\cdot\alloc_\policy\leq \Lagrange(\alloc_\policy;\DualConstAffineVec,\DualConstConvexVec)$ for any $\DualConstAffineVec,\DualConstConvexVec\geq 0$. Now, for any choice of $\alloc\in[0,\barP]^{\dimension}$ and $\DualConvex_i\in[-\upperDualConvex,\upperDualConvex]^{\dimension}$ for $i\in[\NumConvex]$, we can further relax $\Lagrange$ to the linearized version of the Lagrangian, denoted by $\barLagrange$, by applying the Fenchel weak duality:
\begin{align}
\label{eq:lagrange-fair-2}
 \Lagrange(\alloc;\DualConstAffineVec,\DualConstConvexVec)\leq \barLagrange(\alloc;\DualConstAffineVec,\DualConstConvexVec,\DualConvex)&\triangleq \vecreward\cdot\alloc +\sum_{j\in[\NumAffine]}\DualConstAffine_j\left(\AffineConstant_j-\AffineVector_j\cdot\alloc\right)-\sum_{i\in[\NumConvex]}\DualConstConvex_i(\DualConvex_i\cdot\alloc-\ConvexFun_i^*(\DualConvex_i))\nonumber\\
&=\AdjustedReward(\DualConstAffineVec,\DualConstConvexVec,\DualConvex)\cdot\alloc +\sum_{j\in[\NumAffine]}\DualConstAffine_j\AffineConstant_j+\sum_{i\in[\NumConvex]} \DualConstConvex_i \ConvexFun_i^*(\DualConvex_i)~,
\end{align}
where $\ConvexFun^*_i$ is the convex conjugate of $\ConvexFun$ (as in \Cref{def:convex-conjugate}) and the \emph{adjusted reward vector}, denoted by $\AdjustedReward(\DualConstAffineVec,\DualConstConvexVec,\DualConvex)$, is defined as
\begin{equation}
\label{eq:adj:reward}
\AdjustedReward(\DualConstAffineVec,\DualConstConvexVec,\DualConvex)\triangleq \vecreward-\sum_{j\in[\NumAffine]}\DualConstAffine_j\AffineVector_j-\sum_{i\in[\NumConvex]}\DualConstConvex_i\DualConvex_i.
\end{equation}

To help design a candidate policy that is approximately ex-ante feasible and optimal, we consider two min-max games based on the above relaxations. In the first game, the max-player selects a randomized policy in $\PolicySpace$---or equivalently, a vector of expected visit numbers $\alloc$ in $\FeasibleAlloc$—to maximize the game payoff defined by $\Lagrange(\alloc; \DualConstAffineVec, \DualConstConvexVec)$. Meanwhile, the min-player chooses non-negative vectors $\DualConstAffineVec$ and $\DualConstConvexVec$ to minimize the game payoff. The second game is similar to the first, but with the payoff function  relaxed to $\barLagrange(\alloc; \DualConstAffineVec, \DualConstConvexVec, \DualConvex)$. In addition to the non-negative vectors $\DualConstAffineVec$ and $\DualConstConvexVec$, the min-player also selects a matrix $\DualConvex \in [-\upperDualConvex, \upperDualConvex]^{\dimension \times \NumConvex}$.


To see the connection between these games and the optimal ex-ante feasible policy, observe that the first game is indeed a convex-concave game. The function $\Lagrange(\alloc; \DualConstAffineVec, \DualConstConvexVec)$ is concave in $\alloc$ and linear in both $\DualConstAffineVec$ and $\DualConstConvexVec$. Moreover, since randomization is allowed, the set $\FeasibleAlloc \subseteq [0, \barP]^{\dimension}$ is compact and convex. Therefore, by applying Sion's minimax theorem~\citep{sion1958general}, the game admits equilibrium strategies $(\alloc^*; \DualConstAffineVec^*, \DualConstConvexVec^*)$ such that:
\begin{equation}
    \displaystyle\underset{\displaystyle\alloc\in\FeasibleAlloc}{\max}\left(\underset{\displaystyle\DualConstConvexVec,\DualConstAffineVec\geq 0}{\min} \Lagrange(\alloc;\DualConstAffineVec,\DualConstConvexVec)\right)=\displaystyle\underset{\displaystyle\displaystyle\DualConstConvexVec,\DualConstAffineVec\geq 0}{\min} \left(\underset{\displaystyle\alloc\in\FeasibleAlloc}{\max}~\Lagrange(\alloc;\DualConstAffineVec,\DualConstConvexVec)\right)\equiv\Lagrange(\alloc^*;\DualConstAffineVec^*,\DualConstConvexVec^*)~.
\end{equation}
Since $\alloc^*$ is also a Stackelberg equilibrium in the game when the max-player moves first, we conclude that $\alloc^* \in \FairSet$; otherwise, the min-player could drive the payoff to $-\infty$. Furthermore, as stated earlier, for any ex-ante feasible policy $\policy$:
$$\expect{\reward_\pi}\leq \Lagrange(\alloc_\pi;\DualConstAffineVec^*,\DualConstConvexVec^*)\leq \left(\underset{\displaystyle\alloc\in\FeasibleAlloc}{\max}~\Lagrange(\alloc;\DualConstAffineVec^*,\DualConstConvexVec^*)\right)=\Lagrange(\alloc^*;\DualConstAffineVec^*,\DualConstConvexVec^*)\overset{(1)}=\vecreward\cdot\alloc^*~,$$ 
where equality (1) holds because $(\DualConstAffineVec^*, \DualConstConvexVec^*)$ is a best response to $\alloc^*$. Thus, if $\theta_j \cdot \alloc^* < \AffineConstant_j$, then $\DualConstAffine^*_j = 0$, and if $\ConvexFun_i(\alloc^*) < 0$, then $\DualConstConvex^*_i = 0$ (i.e., complementary slackness holds). Hence, the policy $\policy^*$ that implements $\alloc^*$ is an optimal ex-ante feasible policy.

However, the main challenge lies in how one can \emph{efficiently compute} both $\alloc^*$ and $\policy^*$, since even the best-response problem from the perspective of the max-player seems quite complicated. This problem is equivalent to a non-linear version of the joint Markov scheduling problem when the objective function is concave in terms of the expected visit numbers $\alloc$. To the best of our knowledge, this problem has not been studied prior to our work, and no polynomial-time solution is known.

To overcome this challenge, we switch to the second min-max game, which is a relaxation of the first game. By similar arguments, if an equilibrium $(\alloc^*; \DualConstAffineVec^*, \DualConstConvexVec^*, \DualConvex^*)$ exists, then $\alloc^*$ corresponds to an optimal ex-ante feasible policy $\policy^*$.\footnote{\revcolor{Although $\barLagrange(\alloc; \DualConstAffineVec, \DualConstConvexVec, \DualConvex)$ is not jointly convex in $(\DualConstConvexVec, \DualConvex)$ for a given $\alloc$ and $\DualConstAffineVec$, for any fixed $\DualConstConvexVec$ it is convex in $\DualConvex$, and vice versa. As we will clarify later in our proofs, this property, combined with Fenchel duality, is sufficient to establish strong duality and the existence of an equilibrium. However, we do not rely on this existence in our argument.}} More importantly, the best-response problem of the max-player, given a strategy $(\DualConstAffineVec, \DualConstConvexVec, \DualConvex)$ of the min-player, has a simpler structure. It reduces to solving a \emph{modified instance of the joint Markov scheduling problem}, where the rewards are replaced by (possibly negative or positive) adjusted rewards $\AdjustedReward(\DualConstAffineVec, \DualConstConvexVec, \DualConvex)$ as defined in \eqref{eq:adj:reward}. Targeting the relaxation, the next goal is solving this dual-adjusted JMS problem.

\subsubsection{Index-based Optimal Policy for JMS with Arbitrary Rewards}
\label{sec:optimal-JMS-arbitrary}
Viewing the max-player's best-response optimization as a subproblem, we aim to solve it in polynomial time. We draw on previous work studying the JMS problem with linear rewards. These results assume that intermediate states incur negative rewards (i.e., costs) and only terminal states earn positive rewards---see, e.g.,  \cite{dumitriu2003playing,gupta2019markovian}.\footnote{There is slightly a more general condition called \emph{No Free Lunch (NFL)} assumption on the state-reward structure of the Markov chains, under which a similar analysis extends (see, e.g., \cite{gittins1979bandit,kleinberg2017tutorial}).} Under this assumption, they established the optimality of \emph{Gittins index policy}~\citep{gittins1979bandit,dumitriu2003playing}, which is a generalization of the optimal index-based policy of \citeauthor{weitzman1979optimal} for the Pandora's box problem: Given an instance $\{\MC_i\}_{i\in[\altnum]}$, there exists an index mapping $\sigma: \bigcup_{i\in[\altnum]} \states_i \rightarrow \mathbb{R}$ such that, at each time, given the current states $\{s_i\}_{i\in[\altnum]}$, choosing to inspect the Markov chain $\MC_i$ with the maximum index $\sigma(s_i)$ is optimal. This process continues until either $\capacity$ Markov chains enter terminal states or all remaining indices become non-positive, at which point the process terminates. See \Cref{app:JMS-general} for details on Gittins indices and the structure of optimal policy. Here, we only highlight that these indices can be computed in polynomial time.


However, the above approach fails when computing the best response in our problem since the adjusted rewards $\AdjustedReward(\DualConstAffineVec, \DualConstConvexVec, \DualConvex)$ can take both positive or negative values. Nevertheless, as we show in \Cref{app:JMS-general}, there exists a \emph{refinement} of the Gittins index policy (by proper pre-processing of the Markov chains) that solves the linear optimization over the space of randomized policies $\PolicySpace$ in polynomial time for arbitrary positive or negative reward vectors $\vecreward = [R_i(s)]_{i\in[\altnum], s\in\states_i}$. This result, which is based on an intricate reduction, may be of independent interest. We defer the details  to \Cref{app:JMS-general}. From now on, we assume access to an oracle solving the general JMS problem in polynomial time.
}

\subsubsection{Generalized Randomized Dual-adjusted Index Policy}
\label{sec:learning-two-layer}
Our main algorithm for finding a randomized approximate optimal policy is summarized in \Cref{alg:RAI}. At a high level, this algorithm is an iterative primal-dual method that aims to solve both of the above games simultaneously. In each round, the primal player essentially plays a best response based on the payoff of the \emph{second game}, i.e., $\barLagrange$, by selecting an index-based optimal policy for adjusted rewards $\AdjustedReward(\DualConstAffineVec, \DualConstConvexVec, \DualConvex)$. In response, the dual player runs a two-layer coordinated gradient descent (CGD) algorithm to exploit the structure of the second game's payoff (i.e., $\barLagrange$ is convex in each coordinate but not jointly convex with respect to $(\DualConstConvexVec, \DualConvex)$). This simple online learning algorithm helps in finding the Stackelberg equilibrium strategy of the min-player and learning the optimal dual values. In particular, CGD uses the payoff function of the \emph{first game}, i.e., $\Lagrange$, in gradient computations needed for updating $\DualConstAffineVec$ and $\DualConstConvexVec$, and uses the payoff function of the \emph{second game}, i.e., $\barLagrange$, in gradient computations needed for updating $\DualConvex$. When CGD concludes, our final (randomized) policy is the uniform distribution over all the best-response policies computed by the max player during the run of CGD.


\begin{algorithm}[htb]
\caption{Generalized Randomized Dual-adjusted Index Policy (G-RDIP)}
\label{alg:RAI}
    \SetKwInOut{Input}{input}
    \SetKwInOut{Output}{output}
 \Input{learning rates $\InRate,~\OutRatelambda,~\OutRatebeta>0$,  $\#$ of inner iterations $\InnerNum>0$, $\#$  of outer iterations $\OuterNum>0$, oracle access to general JMS solver, upper-bounds $\upperDualAffine,\upperDualConstConvex\geq 0$.}
 
 
 initialize $\forall j\in[\NumAffine]: \DualConstAffine_j^{(1)}\in[0,\upperDualAffine]$~;~$\forall i\in[\NumConvex]: \DualConstConvex_i^{(1)}\in[0,\upperDualConstConvex]$ and $\DualConvex_i^{(1,1)}\in[-\upperDualConvex,\upperDualConvex]^{\dimension}$.
 
 \For{$k=1: \OuterNum$}
 {
 \For{$\ell=1: \InnerNum$}
 {
  \smallskip

  {\color{\commentcolor}\tcc{computing primal player's best-response (general JMS problem)}}
 Let $\policy^{(m,\ell)}\in\underset{\policy\in\PolicySpace}{\argmax}\left(\vecreward-\sum_{j\in[\NumAffine]}\DualConstAffine_j^{(m)}\AffineVector_j-\sum_{i\in[\NumConvex]}\DualConstConvex_i^{(m)}\DualConvex_i^{(m,\ell)}\right)\cdot \alloc_\policy$~ {\color{\commentcolor}\tcp{break the ties arbitrarily (if any)}}


{\color{\commentcolor}\tcc{inner-coordinates gradient update \& projection  of conjugate dual}}
$\forall i \in[\NumConvex]:~~\DualConvexTempkl_i\leftarrow \DualConvexkl_i -\InRate\times\DualConstConvexk_i\left(\nabla\ConvexFun^*_i(\DualConvexkl_i)-\alloc_{\policy^{(m,\ell)}}\right)$

$\forall i \in [\NumConvex]:~~\DualConvex^{(m,\ell+1)}_i\leftarrow \underset{\DualConvex\in[-\upperDualConvex,\upperDualConvex]^{\dimension}}{\textrm{argmin}}~\lVert\DualConvex - \DualConvexTempkl_i\rVert_2$

  }
Let $\displaystyle{\Baralloc}^{(m)}\leftarrow \frac{1}{\InnerNum}\sum_{\ell\in[\InnerNum]}\alloc_{\policy^{(m,\ell)}}$~~~{\color{\commentcolor}\tcp{also, let $\BarDualConvex_i^{(m)}\leftarrow\frac{1}{\InnerNum}\sum_{\ell\in[\InnerNum]}\DualConvexkl_i$}}

  {\color{\commentcolor}\tcc{outer-coordinates gradient update and projection for duals of affine and convex constraints}}
$\forall j \in[\NumAffine]:~~\DualConstAffine^{(m+1)}_j\leftarrow \min\left(\upperDualAffine,\max\left(0,\DualConstAffinek_j -\OutRatelambda\times\left(\AffineConstant_j-\AffineVector_j\cdot\Baralloc^{(m)}\right)\right)\right)$
 
$\forall i \in[\NumConvex]:~~\DualConstConvex^{(m+1)}_i\leftarrow \min\left(\upperDualConstConvex,\max\left(0,\DualConstConvexk_i +\OutRatebeta\times \ConvexFun_i\left(\Baralloc^{(m)}\right)\right)\right)$


 }

\Return{$\displaystyle\hat{\policy}\sim \textrm{unif}\left\{\policy^{(m,\ell)}:(m,\ell)\in[\OuterNum]\times[\InnerNum]\right\}$ and $\displaystyle\alloc_{\hat{\policy}}=\frac{1}{\OuterNum\InnerNum}\sum_{m\in [\OuterNum]}\sum_{\ell\in [\InnerNum]}\alloc_{\policy^{(m,\ell)}}$}

\end{algorithm}

\revcolor{
\begin{theorem}[\textbf{Approximate Ex-ante Feasibility and Optimality}] 
\label{thm:RAI}
Given any $\delta,\varepsilon>0$, the Generalized Randomized Dual-adjusted Index policy $\hat{\policy}$ (\Cref{alg:RAI}) with parameters set as (i) $\upperDualAffine=\upperDualConstConvex=\mathcal{O}\left(\frac{1}{\delta}\right)$, (ii) $\InnerNum=\mathcal{O}\left(\frac{1}{\delta^2\epsilon^2}\right)$ and $\OuterNum=\mathcal{O}\left(\frac{1}{\delta^2\epsilon^2}\right)$, and (iii) $\InRate=\mathcal{O}\left(\delta^2\epsilon\right)$, $\OutRatelambda=\mathcal{O}\left(\delta^2\epsilon\right)$, and $\OutRatebeta=\mathcal{O}\left(\delta^2\epsilon\right)$, satisfies:

\medskip
\begin{itemize}
    \item \textbf{Approximate optimality:} $\expect{R_{\hat{\policy}}} \geq \OptFair -\varepsilon$
    \smallskip
     \item \textbf{Approximate ex-ante feasibility:} for all affine constraint $j\in[\NumAffine]$, $\AffineVector_j\cdot\alloc_{\hat{\policy}}\leq \AffineConstant_j+\delta$, and for all convex constraints $i\in[\NumConvex]$, $\ConvexFun_i\left(\alloc_{\hat{\policy}}\right)\leq \delta$.
\end{itemize}
Furthermore, the resulting policy is randomized, obtains a distribution over $\mathcal{O}\left(\frac{1}{\delta^4\epsilon^4}\right)$ deterministic policies, and runs in polynomial time in $\left(\dimension,\altnum,\frac{1}{\varepsilon},\frac{1}{\delta}\right)$, where the running time dependency on $\varepsilon$ and $\delta$ is $\mathcal{O}\left(\frac{1}{\delta^4\epsilon^4}\right)$.
\end{theorem} 

}

We postpone the analysis of G-RDIP to \Cref{app:jms}. At a high level, our analysis follows the game-theoretic view of this algorithm and how it essentially solves a relaxation game, as sketched in \Cref{sec:RAI}. For more details, see the proof of \Cref{thm:RAI} in \Cref{app:jms}.


\section{Numerical Simulations}
\label{sec:numerical}
In this section, we supplement our theoretical analysis in \Cref{sec:pandora} with numerical simulations using synthetic data. We examine the Pandora's box model with multiple selections from a population with two demographic groups, and empirically compare the optimal constrained and unconstrained policies across a wide range of model primitives. \revcolor{For numerical simulations for the JMS setting with multiple constraints, see \Cref{sec:numerical_JMS}.}

\smallskip
\noindent\textbf{Short-term vs. long-term effects:}
\revcolor{As discussed in the introduction, to empirically assess how adding socially aware constraints affects the decision maker's utility, we distinguish between \emph{observable signal}s for candidate qualities and their \emph{unobservable true values} for candidate qualities. This will enable us to study both
(i) \emph{short-term effects} and (ii) \emph{long-term effects} of the algorithm's outcomes.}

For short-term effects, consistent with our theoretical analysis, we assume that the decision maker has access to prior distributions of observable signals about candidate quality, which may be subject to implicit bias, and can observe these signals through costly inspections. These signals represent the primary means of candidate assessment in the short-term, and \revcolorm{the decision maker's utility is evaluated by incorporating these observable signals as candidate values into the objective function described in \cref{eq:utility}, \Cref{sec:pandora}.}

For long-term effects, we further assume that candidates possess unobservable true qualities at the time of hiring, representing their genuine downstream quality after being given the opportunity. These true qualities may significantly differ from the observable signals. We also assume that the decision maker does not have access to the true qualities or their distributions during the search process and relies solely on the signals and their distributions to conduct the search and make selections. \revcolorm{The long-term utility is thus evaluated by substituting the unobservable true qualities as values into the objective function described in \cref{eq:utility}, \Cref{sec:pandora}---instead of plugging in the observable signals.} Throughout our study, we maintain that there is no inherent bias in the true qualities between demographic groups, though the signals about these qualities may be biased.\footnote{For more context, see our discussion in the Introduction.}

\smallskip
\noindent\textbf{Basic simulation setup:}  
We construct randomly generated instances of the Pandora's box problem (\Cref{sec:pandora}), where each instance comprises \(n = 60\) candidates evenly divided between the groups \(\WomanSet\) and \(\ManSet\). \revcolor{Inspection costs for these candidates are independently drawn from a uniform distribution over \([c_l, c_h]\) and are fixed thereafter, with \(c_l = 3\) and \(c_h = 6\) in our simulations.\footnote{This distribution was selected for clarity, but our results remain robust across different cost choices.} The values \(\{v_i\}_{i \in [n]}\) in the Pandora's box model represent observable signals about the candidates' quality or skills. To capture the natural heterogeneity of quality in the population, we generate prior value distributions \(\{F_i\}_{i \in \WomanSet \cup \ManSet}\) by first sampling \(n = 60\) ``unbiased" mean values \(\{\bar{\mu}_i\}_{i \in [n]}\) independently from a log-normal distribution with parameters \(\mu = 0\) and \(\sigma = 1\), scaled by a factor of 10 and shifted by +20.}\footnote{The choice of a log-normal distribution is inspired by the U.S. Bureau of Labor Statistics report in February 2017, which used nationally representative data on specific skills required for individual jobs to study how wage and skill distributions vary across different sectors (\Cref{fig-apx:v1_Histogram}~(d)). See \cite{bls2017wage} for more context and justification on fit of a log-normal distribution.}
We then define each value distribution \(F_i\) as a normal distribution \(\mathcal{N}(\mu_i, \sigma_i)\), where \(\mu_i = \bar{\mu}_i \times \rho_i\) and \(\sigma_i = 10 \times \rho_i\). Here, \(\rho_i \in [0, 1]\) is the \emph{bias factor} for the observable quality signal \(v_i\) of candidate \(i\). This normal distribution captures the uncertainty in the candidate's quality, which is revealed upon inspection. To model the unobservable true qualities, we define the true value as \(v_i^{\dagger} = v_i / \rho_i\). It is easy to see that \(v_i^{\dagger} \sim \mathcal{N}(\bar{\mu}_i, 10)\); therefore, these true qualities do not have an inherent bias, as posited. Furthermore, to model possible group-wise bias in quality signals—considering \(\ManSet\) as the majority/privileged group and \(\WomanSet\) as the minority/under-privileged group—we set \(\rho_i = 1\) for all \(i \in \ManSet\) and \(\rho_i = \rho\) for all \(i \in \WomanSet\), where \(\rho \in \{0.1, 0.2, \dots, 0.9, 1\}\).  Note that a smaller bias factor \(\rho\) implies a higher disparity. The effect of varying \(\rho\) on the underlying distributions is illustrated in \Cref{fig-apx:v1_Histogram} \footnote{\revcolor{To check the robustness of our numerical results to this model primitive, we have also studied non-multiplicative forms of bias. We do not report the exact results for brevity and coherence, but all of our qualitative results and insights remained unchanged.}}. We consider various values of capacity \(k\) ranging from 1 to 20. Lastly, using Monte Carlo simulations, we examine various summary statistics of our policies. \revcolorm{Given this setup, as mentioned earlier, we use the \emph{same} utility formula (see \cref{eq:utility} in \Cref{sec:pandora}) for both short-term and long-term evaluations: in the short-term case, we compute utilities using the observable quality signals \(v_i\), while in the long-term case, we use the unobservable true qualities \(v_i^{\dagger}\).}

\begin{figure}[htb]
    \centering
    \begin{subfigure}[b]{0.32\textwidth}
        \centering
        \includegraphics[width=\textwidth]{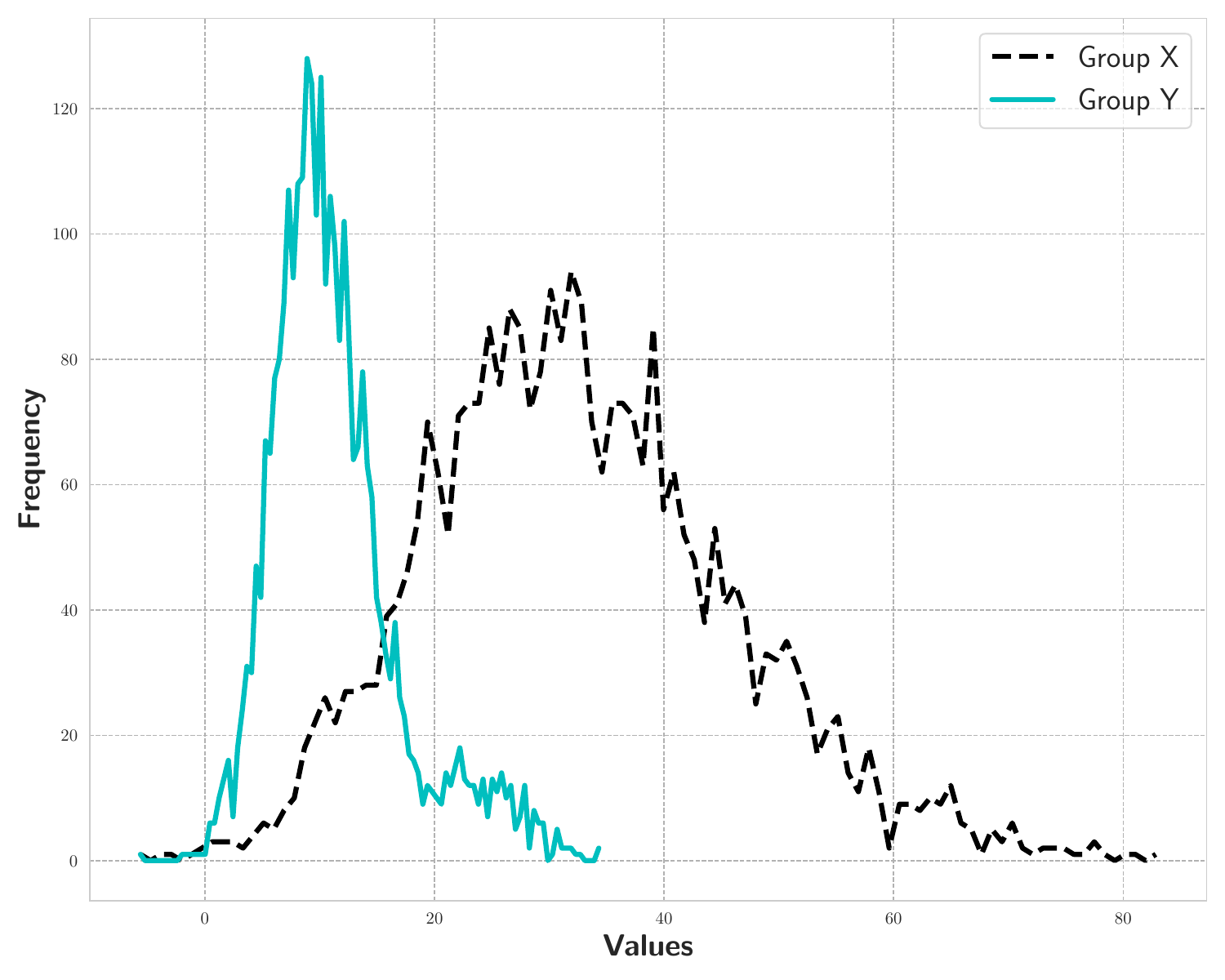}
        \caption{Bias factor = 0.3}
    \end{subfigure}
    \hfill
        \begin{subfigure}[b]{0.32\textwidth}
        \centering
        \includegraphics[width=\textwidth]{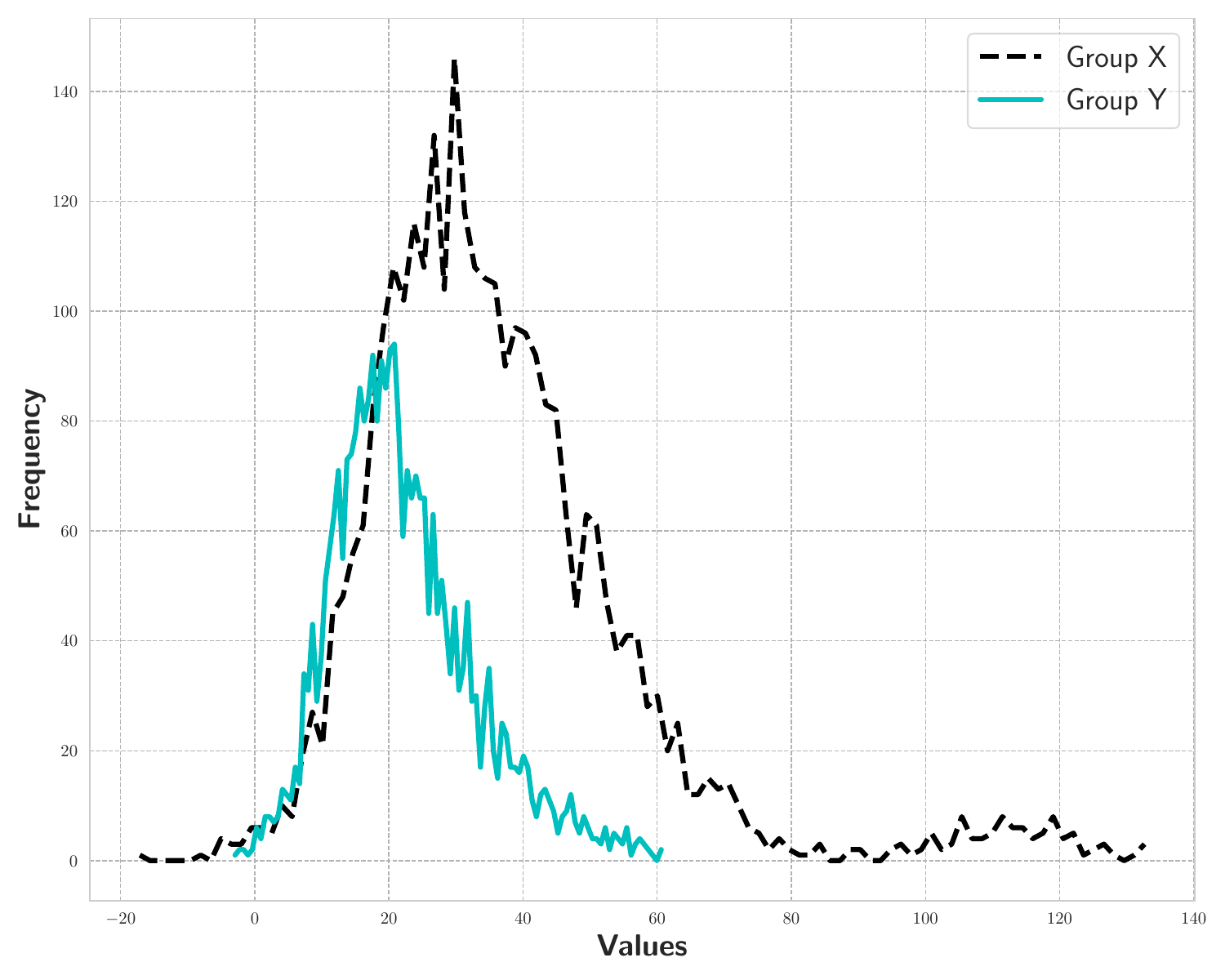}
        \caption{Bias factor = 0.6}
    \end{subfigure}
    \hfill
    \begin{subfigure}[b]{0.32\textwidth}
        \centering
        \includegraphics[width=\textwidth]{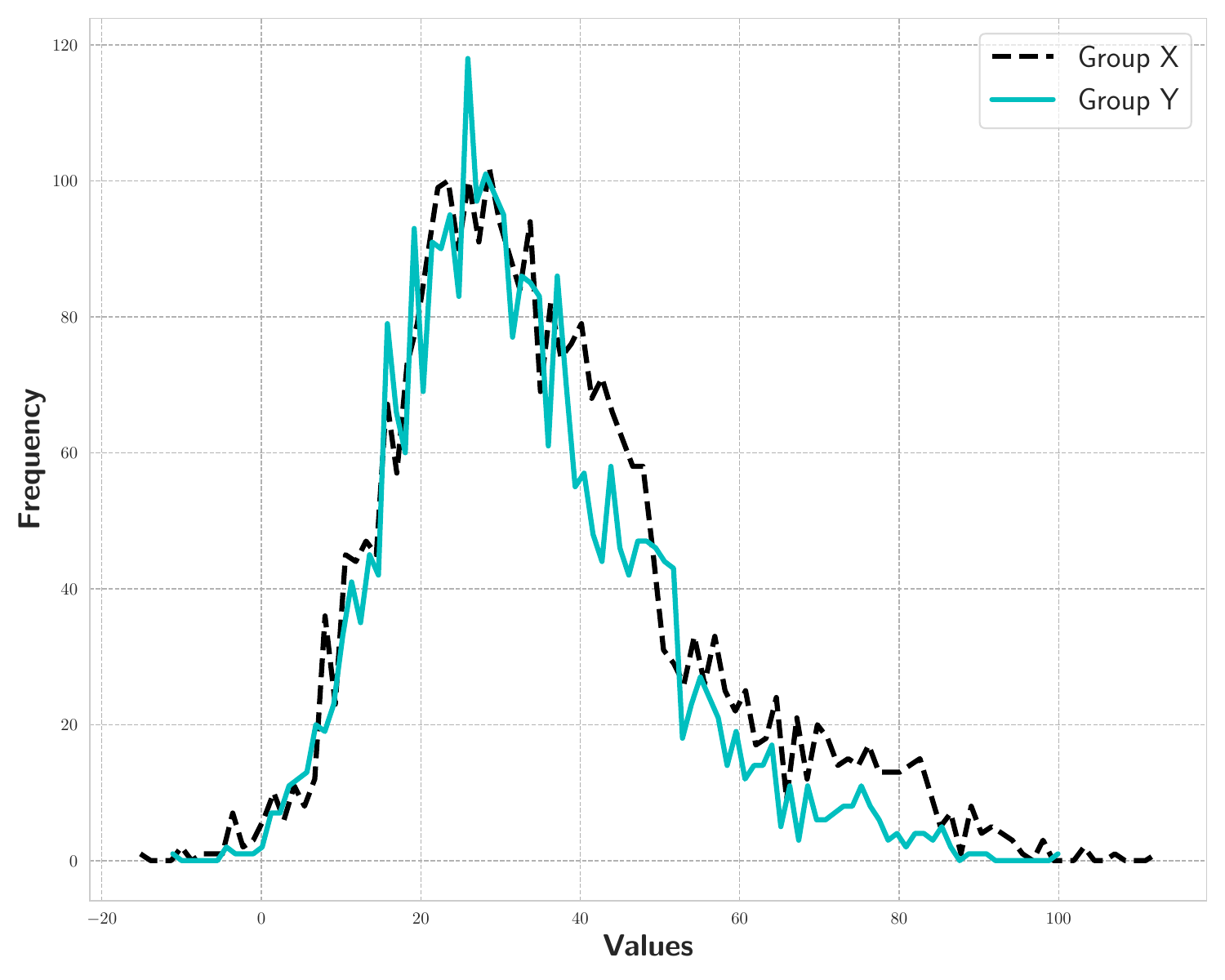}
        \caption{Bias factor = 0.9}
    \end{subfigure}
    \begin{subfigure}[b]{0.5\textwidth}
        \centering
        \includegraphics[width=\textwidth]{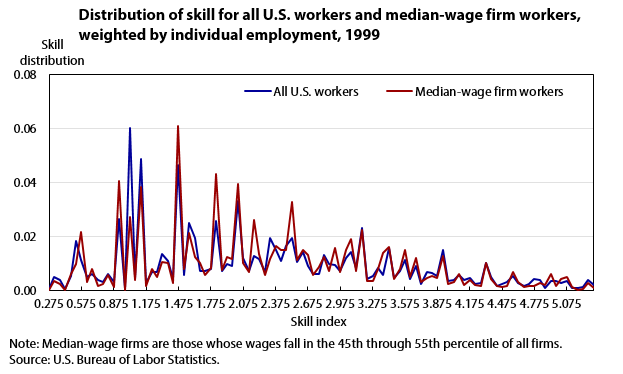}
        \caption{Aggregate skill score distribution}
    \end{subfigure}
    \caption{(a), (b), and (c): Sample histograms of the generated values $\boldsymbol{\{{v}_i\}_{i\in[1:60]}}$ for the groups $\boldsymbol{\WomanSet}$ (cyan) and $\boldsymbol{\ManSet}$ (black); (d) Distribution of skill scores for all U.S. workers in 1999~\citep{bls2017wage}.
     }
    \label{fig-apx:v1_Histogram}
\end{figure}


In the remainder of this section, we examine the effects of demographic parity in selection (\ref{eq:parity}). We perform a numerical comparison between the optimal unconstrained policy (solution to \eqref{eq:opt-unconstrained}) and the optimal constrained policy (solution to \eqref{eq:opt-constrained}).  We then highlight the key numerical findings and discuss managerial insights. \revcolorm{More specifically, in \Cref{sec:num-short-term} we compare the performance of the constrained and unconstrained policies using short-term utilities, then in \Cref{sec:num-long-term} we do the same but this time measuring their long-term utilities, and lastly in \Cref{sec:num-sim-quota} we study both short-term and long-term effects under the \ref{eq:quota} constraint. For more comprehensive numerical results and additional scenarios---encompassing a long list including (i) considering other notions of socially aware ex-ante constraints such as the average budget for subsidization \eqref{eq:budget}, (ii) studying the unintended consequences of socially aware constraints, (iii) checking robustness of the result to the choice of model primitives, (iv) studying the effects of resource augmentations such as increased capacity or budget subsidies on the price of fairness, and (v) extending our simulations to the JMS setting with multiple constraints---refer to \Cref{apx:numerical-main}. }

\subsection{Short-term Outcomes: (Surprisingly) Small Utilitarian Loss}
\label{sec:num-short-term}

We begin by comparing the short-term performance of two optimal policies. In \Cref{fig-rad:short-term}~(a), we plot the expected utilities of these policies as a function of the bias factor $\rho$ for a fixed capacity $k = 20$. In \Cref{fig-rad:short-term}~(b), we illustrate the price of fairness (the ratio of the two expected utilities) as a function of $\rho$ for capacities $k = 8$, $15$, and $20$. We observe that as $\rho$ decreases from $1$ (unbiased signals) to approximately $0.5$ (moderately biased), the price of fairness decreases gradually. For example, when $\rho = 0.7$, the drop in utility is less than $6\%$. 

One might speculate that the utility loss is minimal because the unconstrained solution was not very ``unfair'' to the minority group; in other words, the optimal unconstrained policy was relatively balanced between the two groups. To explore this, we also plot the normalized constraint slack under the optimal unconstrained policy, $\Delta_{\textsc{Cons}}^{\pi^*_{\textsc{UC}}}/k$, in \Cref{fig-rad:short-term}~(b). Contrary to expectation, for $\rho = 0.7$ and $k = 20$, the normalized slack is around $0.5$, implying that without the parity constraint, an average of $10$ more candidates from the majority group would be selected. Overall, our observations in \Cref{fig-rad:short-term}~(b) indicate that for moderate bias values, imposing ex-ante demographic parity leads to only a small utilitarian loss while yielding a significant egalitarian gain in selection. These findings are consistent across different parameter choices in our simulations. For more numerical results and additional scenarios, refer to \Cref{apx:numerical-main}, particularly \Cref{fig:v1_both_fair_unfair}, \Cref{fig:v1_diff_and_CR_fair_unfair}, and \Cref{fig:v1_slack_lambda_fair_unfair}.

\begin{figure}[htb]
    \centering
    \begin{subfigure}[b]{0.48\textwidth}
        \centering
        \includegraphics[width=\textwidth]{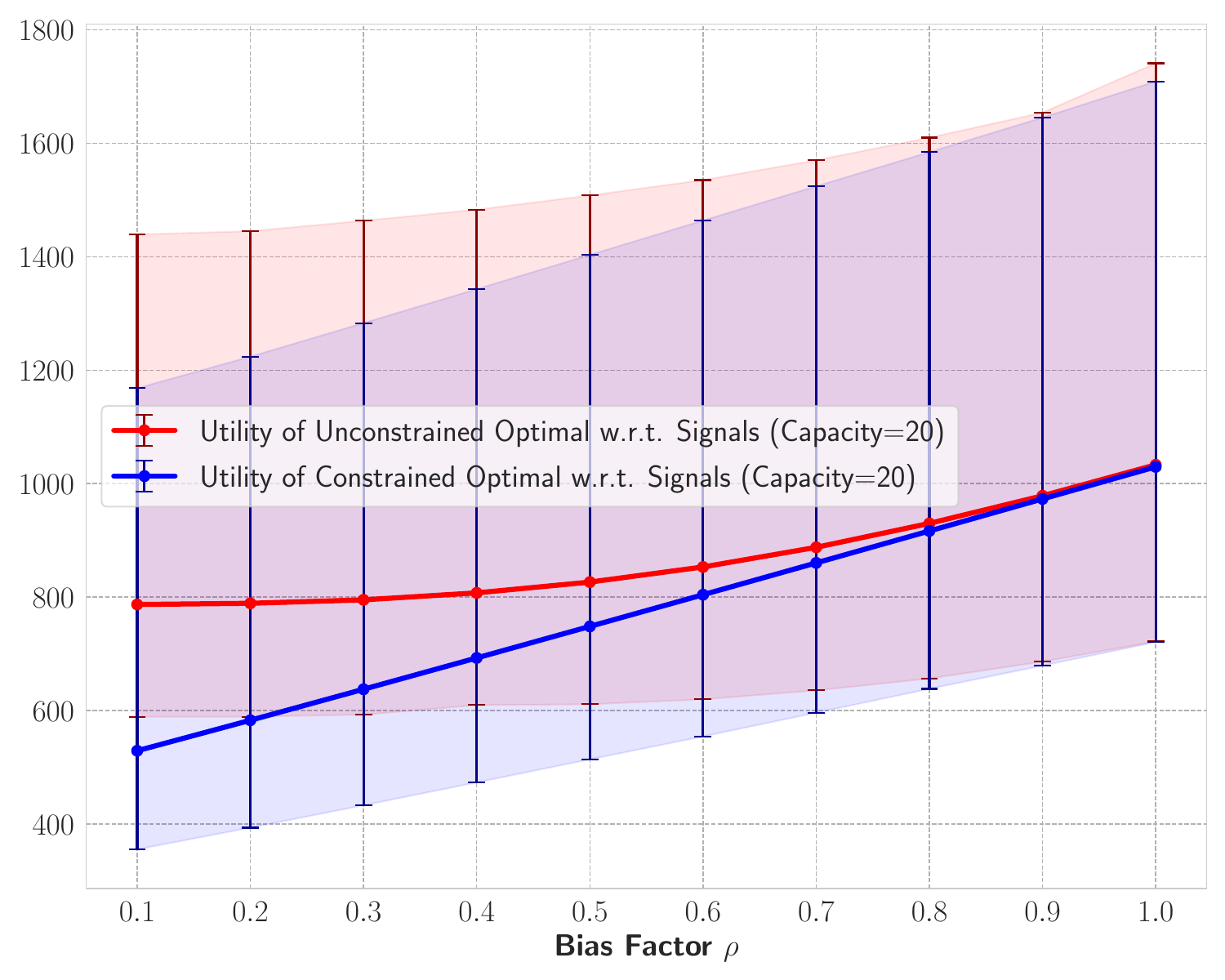}
        \caption{Expected utilities calculated based on signals $\{v_i\}_{i\in[n]}$ for the unconstrained optimal policy (red) and the constrained optimal policy (blue).}
    \end{subfigure}
    \hfill
        \begin{subfigure}[b]{0.45\textwidth}
        \centering
        \includegraphics[width=\textwidth]{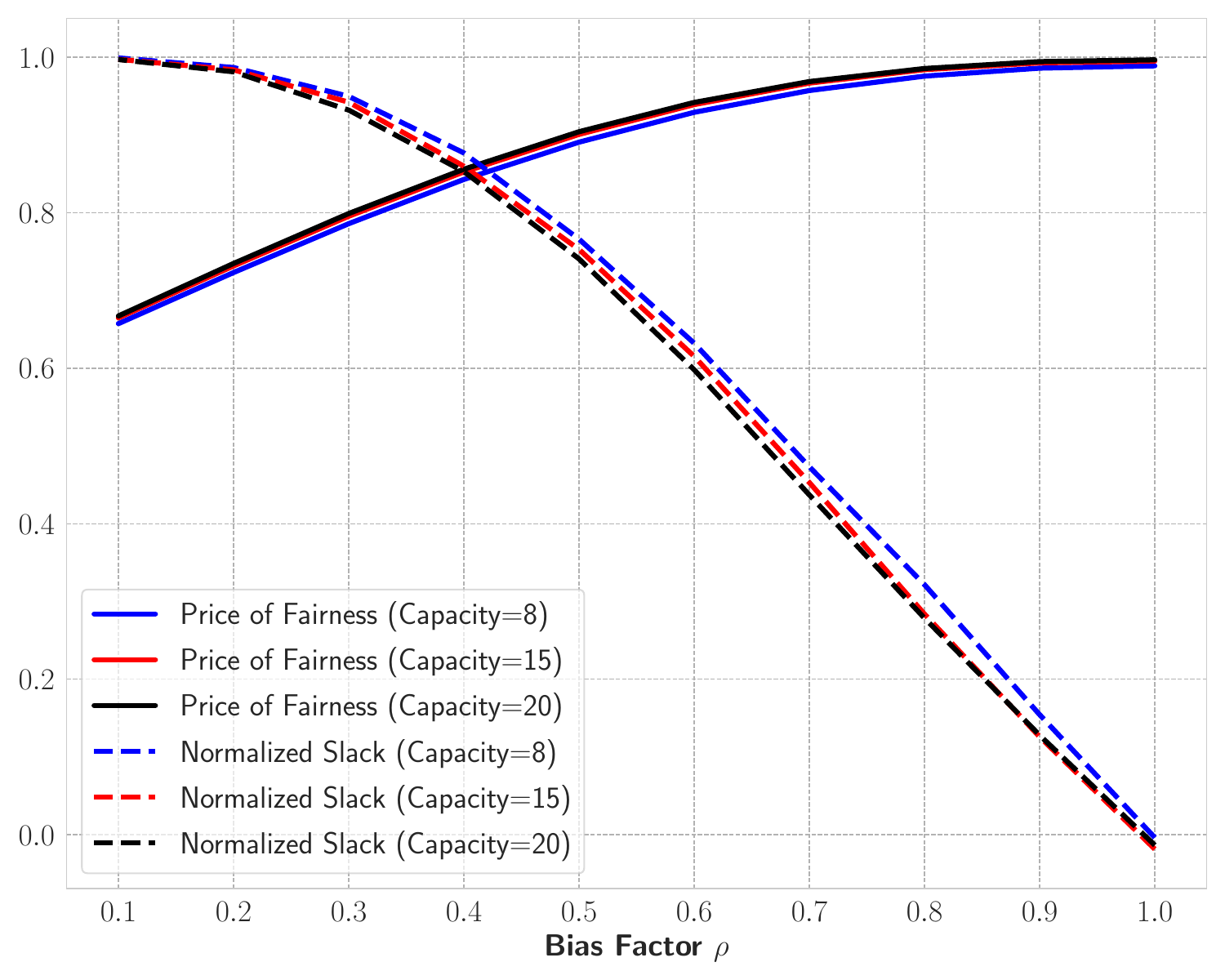}
        \caption{Price of fairness ratio calculated based on signals $\{v_i\}_{i\in[n]}$ (solid lines) and the normalized constraint slack of unconstrained optimal policy (dashed lines).}
    \end{subfigure}
    \caption{\revcolor{Comparing the short-term outcomes of unconstrained and constrained optimal policies.}}
    \label{fig-rad:short-term}
\end{figure}

\subsection{Long-term Outcomes: Potential Utilitarian Gain}
\label{sec:num-long-term}
We now examine the long-term performance of the two optimal policies. We run them as before, but this time measure their expected utility based on the true unbiased values \(\{v_i^{\dagger}\}_{i \in [n]}\) instead of the biased signals \(\{v_i\}_{i \in [n]}\). Our goal is to empirically assess whether imposing ex-ante demographic parity can also lead to utilitarian gains---alongside its significant egalitarian benefits---when utility is measured by true values.  Specifically, in \Cref{fig-rad:long-term}~(a), we plot the expected utilities of both policies as a function of the bias factor \(\rho\) for a fixed capacity \(k = 20\). In \Cref{fig-rad:long-term}~(b), we show the price of fairness (i.e., the ratio of the two expected utilities) and the normalized constraint slack of the optimal unconstrained policy as functions of \(\rho\) for capacities \(k = 8\), \(15\), and \(20\). From both graphs, we observe that imposing demographic parity results in a long-term utilitarian gain. In \Cref{fig-rad:long-term}~(a), as \(\rho\) decreases from \(1\) (unbiased) to \(0\) (significantly biased), the performance of the optimal constrained policy remains nearly constant. In contrast, the performance of the optimal unconstrained policy deteriorates rapidly. Similarly, in \Cref{fig-rad:long-term}~(b), the price of fairness with respect to the true qualities \(\{v_i^{\dagger}\}_{i \in [n]}\) remains above \(1\) for all \(\rho\) values. This result is robust across all our simulation parameter choices. For additional numerical results and scenarios see \Cref{apx:numerical-main}, particularly \Cref{fig:v3_both_fair_unfair} and \Cref{fig:v3_diff_and_CR_fair_unfair}.

\begin{figure}[htb]
    \centering
    \begin{subfigure}[b]{0.48\textwidth}
        \centering
        \includegraphics[width=\textwidth]{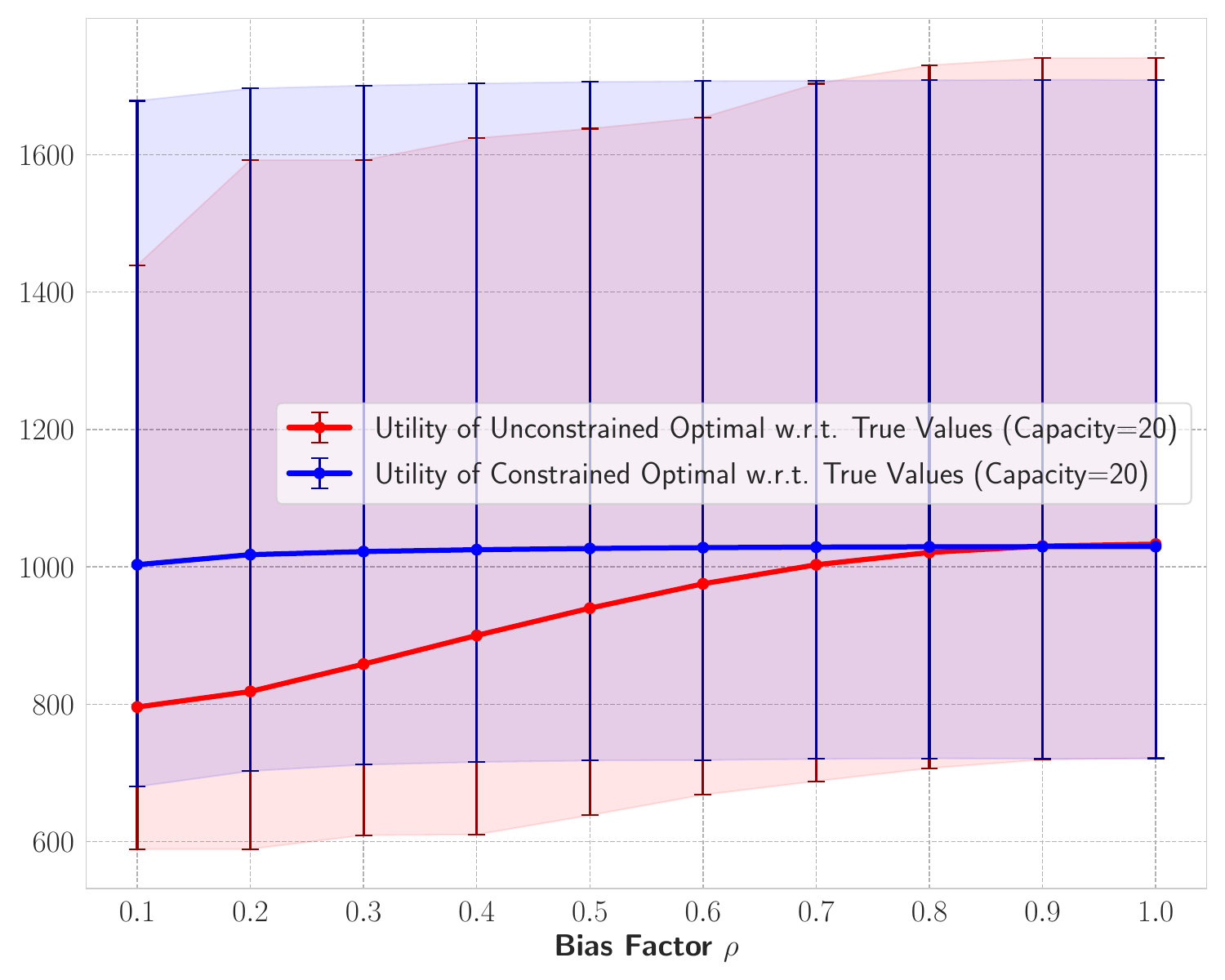}
        \caption{Expected utilities calculated based on true values $\{{v}^\dagger_i\}_{i\in[n]}$ for the unconstrained optimal policy (red) and the constrained optimal policy (blue).}
    \end{subfigure}
    \hfill
        \begin{subfigure}[b]{0.48\textwidth}
        \centering
        \includegraphics[width=\textwidth]{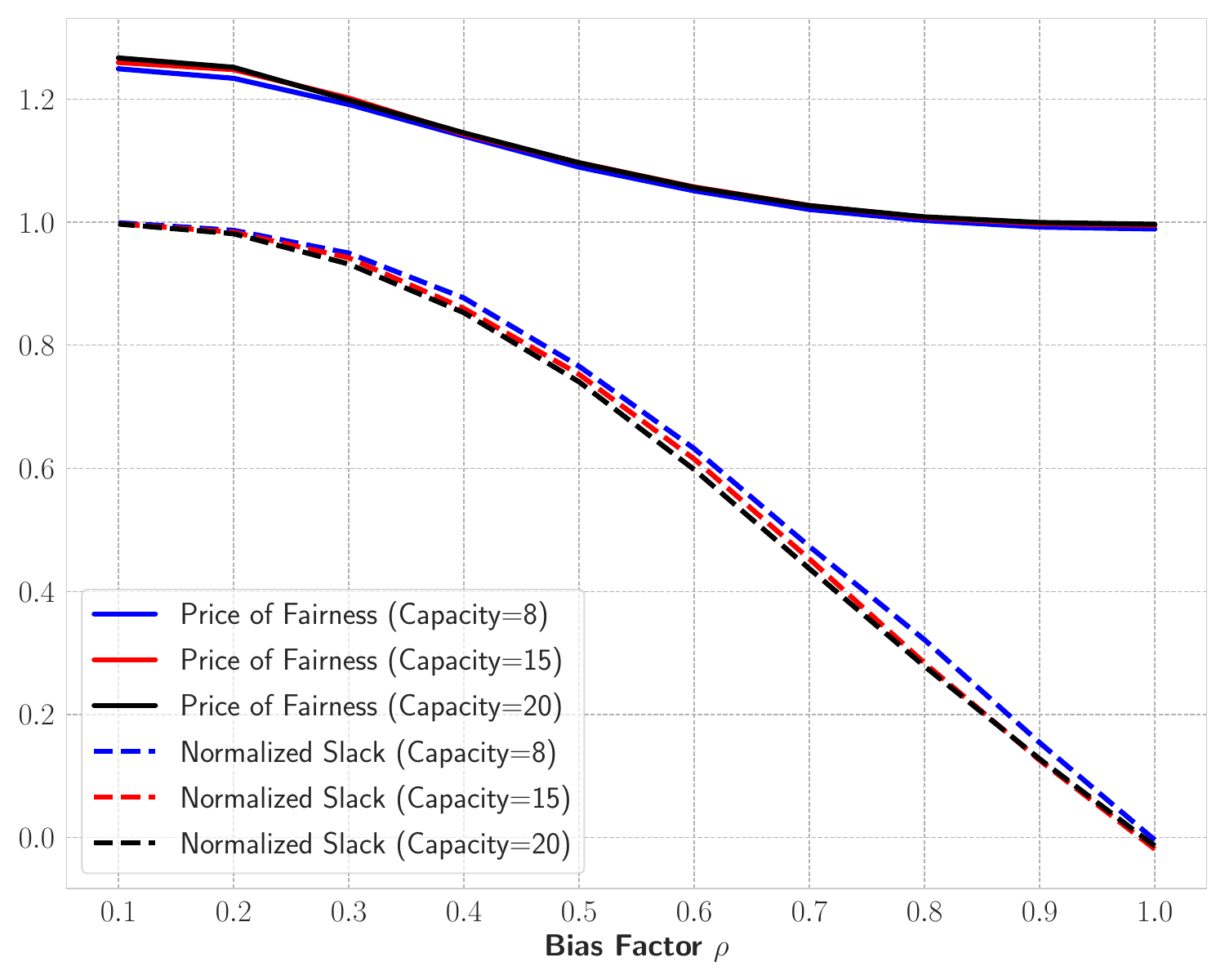}
        \caption{Price of fairness ratio calculated based on true values $\{v^\dagger_i\}_{i\in[n]}$ (solid lines) and the normalized constraint slack of unconstrained optimal policy (dashed lines).}
    \end{subfigure}
    \caption{\revcolor{Comparing the long-term outcomes of unconstrained and constrained optimal policies.}}
    \label{fig-rad:long-term}
\end{figure}

\subsection{Refining Fairness with Adjustable Quotas}
\label{sec:num-sim-quota}
When \(\rho\) is excessively small, it becomes necessary to refine the notion of fairness defined by Constraint~\ref{eq:parity} to achieve improved performance. To address this, we focus on Constraint~\ref{eq:quota} in selection and consider adjusting the parameter \(\theta\) based on the degree of disparity or bias in the signals. In \Cref{fig-rad:quota}, we plot the price of fairness (calculated using both biased signals and true values) as a function of \(\theta\). Notably, \(\theta = 0.5\) corresponds to demographic parity (i.e., Constraint~\ref{eq:parity}), and the constraint becomes more relaxed as \(\theta\) decreases. As observed in \Cref{fig-rad:quota}~(a), decreasing \(\theta\) reduces the short-term utility loss compared to the optimal unconstrained policy across all values of \(\rho\). This effect is particularly pronounced for smaller \(\rho\) values (e.g., \(\rho \in \{0.1, 0.2, 0.3\}\)) compared to larger ones (e.g., \(\rho \in \{0.8, 0.9\}\)). Additionally, in \Cref{fig-rad:quota}~(b), we see that the long-term utility gain relative to the optimal unconstrained policy for small \(\rho\) values increases significantly as \(\theta\) decreases from \(0.5\). The gain peaks at a certain point (around $0.5$) and approaches \(1\) as \(\theta\) approaches \(0\).

This exercise yields an important managerial insight: To avoid excessive short-term utilitarian loss and potentially achieve long-term utilitarian gains, decision makers should carefully set the parameter \(\theta\). A ``good'' choice of the parameter $\theta$ for short-term utility highly depends on $\rho$, or in general, the degree of bias in the signals. On the other hand, the optimal choice for long-term utility does not depend on $\rho$ and is mainly dependent on the size of the groups (in this case, $\theta=0.5$ as the groups have the same size). Our observations are robust across various parameter choices in all simulations. For additional numerical results and scenarios, see \Cref{sec:numerical-quota}, in particular, \Cref{fig:v5_both_fair_unfair} and \Cref{fig:v6_both_fair_unfair} for simulations with alternatives parameters.

\begin{figure}[htb]
    \centering
    \begin{subfigure}[b]{0.48\textwidth}
        \centering
        \includegraphics[width=\textwidth]{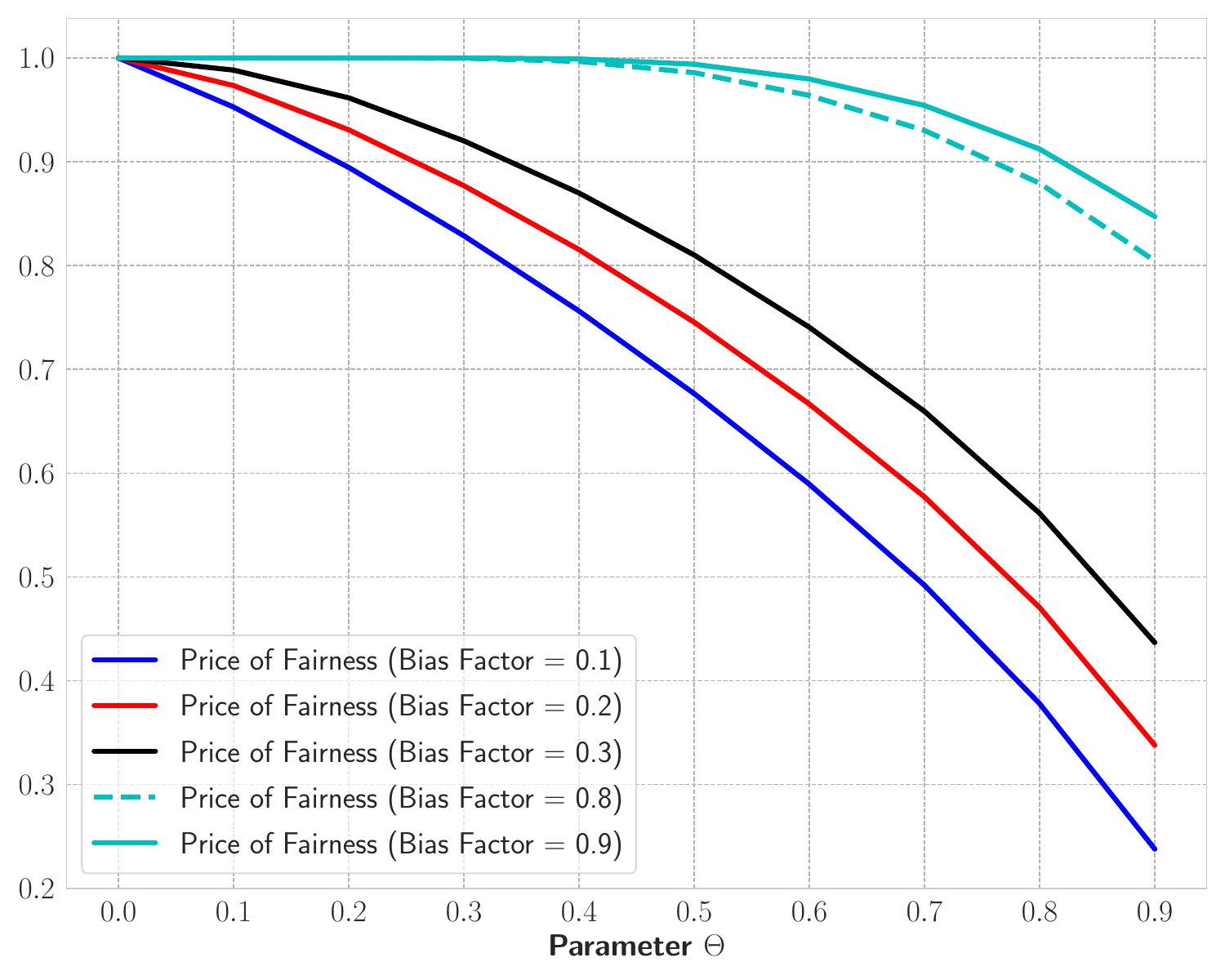}
        \caption{Short-term price of fairness}
    \end{subfigure}
    \hfill
        \begin{subfigure}[b]{0.48\textwidth}
        \centering
        \includegraphics[width=\textwidth]{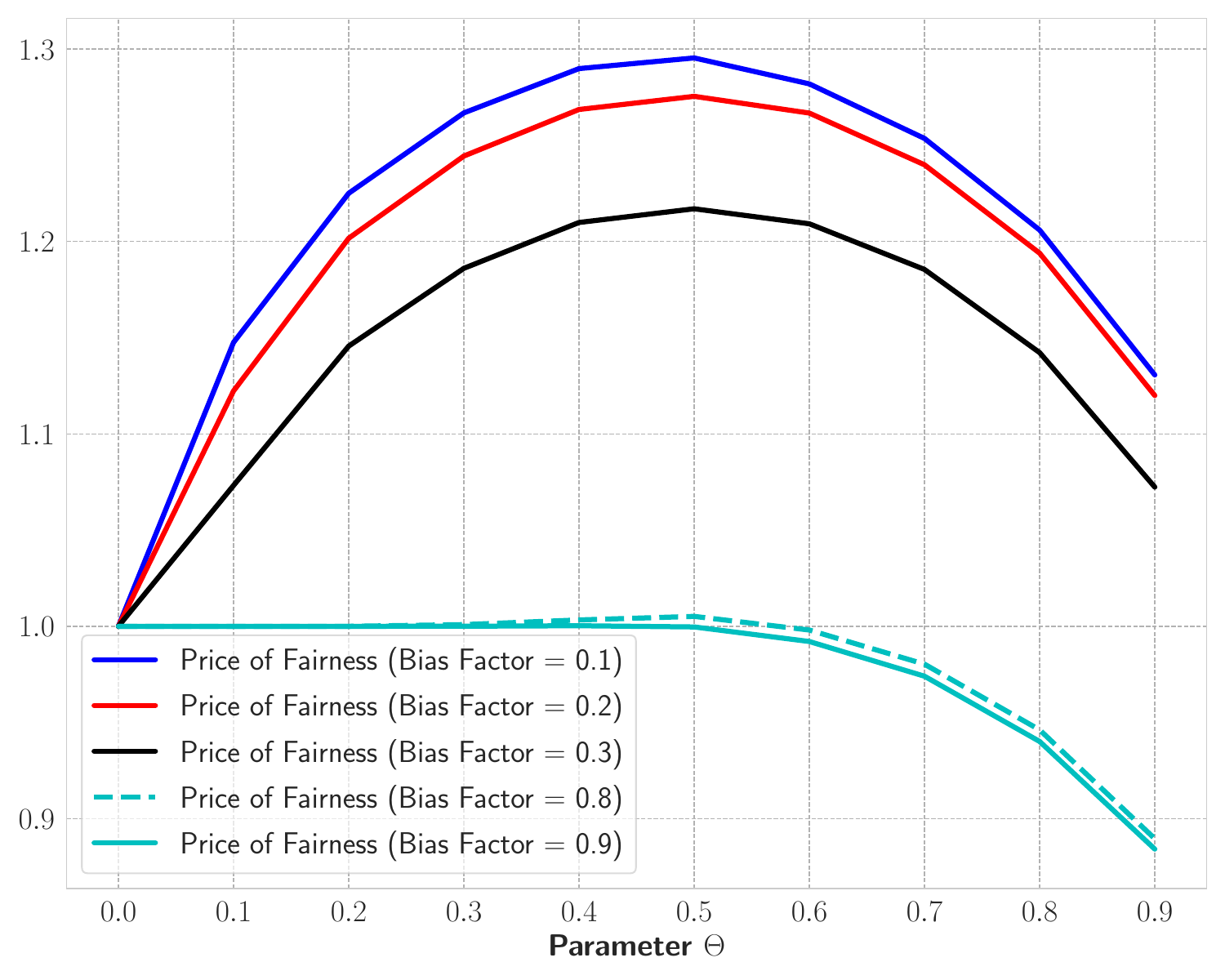}
        \caption{Long-term price of fairness}
    \end{subfigure}
    \caption{\revcolor{The performance of optimal constrained policy for \ref{eq:quota} in selection with parameter $\boldsymbol{\theta}$ ($\textbf{k=20}$).}}
    \label{fig-rad:quota}
\end{figure}


\vspace{-3mm}
\section{Conclusion and Future Directions}
\label{sec:conclusion}
\revcolor{We initiate the study of adding ex-ante constraints to Markovian sequential search with costly information acquisition. Our study is motivated by the rise of algorithmic hiring and the interest in improving measures of inclusion and access. We encode many of such measures as general ex-ante constraints and design optimal constrained policies (or nearly optimal and near feasible ones) for the underlying search processes. We view our work as a building block for understanding the efficiency–fairness trade-off, a direction we plan to pursue.

Although we focus on ex-ante fairness---natural in settings where the search repeats many times (e.g., pre-employment screening)---studying ex-post fairness across different stages of a single, high-stakes search is another promising direction. For instance, how can we ensure an ex-post balance in interviews for a top executive hire? Additionally, many search processes are \emph{delegated} in principal–agent frameworks, raising questions about how to induce fairness when only the principal values it. Also, a future direction, ex-ante constraints can capture several operational limitations and be studied in other applications of sequential search beyond hiring, such as facility location with costly inspection (e.g., for placing windmills in different locations).

Finally, we adopt the classic Pandora's box assumptions of known, independent reward distributions. Recent work has begun to relax these assumptions—e.g., by considering unknown distributions or correlated boxes—albeit without fairness considerations \citep{chawla2020pandora, gatmiry2022bandit}. Revisiting these richer settings with well-defined notions of fairness is an appealing avenue for future research.}

\newcommand{\newblock}{}
\setlength{\bibsep}{0.0pt}
\bibliographystyle{plainnat}
\OneAndAHalfSpacedXI
{\footnotesize
\bibliography{refs}}


\newpage
\ECSwitch
\ECDisclaimer
\renewcommand{\theHchapter}{A\arabic{chapter}}
\renewcommand{\theHsection}{A\arabic{section}}
\section{Further related work}
\label{sec:further-related-work}

\paragraph{\revcolorm{Fairness, sequential decision making, \& information acquisition.}}

Fairness has been extensively studied in machine learning problems such as classification; see \citet{barocas2017fairness} for an overview. Fairness in offline allocation with limited resources has also been considered, starting with the seminal work of \citet{baruah1993proportionate}, \citet{kumar2000fairness}, and \citet{bertsimas2012efficiency}. More recently, attention has turned to studying dynamic and operational considerations in fairness, such as sequential decisions and costly information acquisition; see surveys by \citet{finocchiaro2021bridging} and \citet{nashed2023fairness}. To the best of our knowledge, our work is the first to study socially aware constraints in sequential search under costly inspection.

Several related strands are still worth noting. A number of papers study fairness in online selection problems for indivisible goods, such as \revcolorm{secretary problems and prophet inequalities~\citep{buchbinder2009secretary,correa2021fairness,arsenis2022individual,salem2024secretary}.} Others focus on dynamic allocation of divisible goods, designing policies to maximize egalitarian welfare~\citep{lien2014sequential,sinclair2020sequential,manshadi2021fair}. Fair division and market equilibrium problems (e.g., notions such as envy-freeness) have also been explored in dynamic settings~\citep{walsh2011online,kash2014no,aleksandrov2015online,peysakhovich2023implementing,gao2023infinite}. \revcolorm{There are several other works that consider fairness in modern variants of online resource allocation, for example the work of \citet{liao2022nonstationary} analyzing fairness in allocating sequentially arriving items under non-stationarity, and \citet{bateni2022fair} addressing fairness-efficiency trade-offs in online resource allocation with applications to online advertising. The interplay between information acquisition and fairness/discrimination has also been studied. The closest to us is the work of \citet{cai2020fair}, that revolves around achieving fairness through information acquisition. However, the setting differs from ours as it is concerned with targeted screening to improve information quality for a subset of individuals. Another  conceptually related work---especially to our numerical study of long-term utilities and downstream outcomes--is \citet{baek2023feedback} that highlights feedback loops exacerbating disparities in hiring, by showing that even a small initial difference in how well a firm can evaluate candidates from different groups can snowball into long-lasting hiring disparities due to a feedback loop.}

Finally, another line of research that is conceptually related to us investigates fairness in bandit settings, \revcolorm{introducing concepts such as fairness in exploration and hindsight~\citep{cayci2020group,baek2021fair,schumann2022group,li2020hiring,gupta2021individual}. Another example is the work of \cite{komiyama2024statistical} that take a multi-armed bandit approach, showing how temporary affirmative-action subsidies when firms lack data about minorities can mitigate persistent statistical discrimination.}

\smallskip
\paragraph{Fairness in operations and revenue management.}
There has been a growing literature studying various notions of fairness in different operational settings. For example, \citet{cohen2022price} explores the challenges and trade-offs involved in implementing price fairness constraints for different customer groups in the context of discriminatory pricing. Other examples study fairness-aware online price discrimination~\citep{chen2021fairness}, fairness criteria in network revenue management with demand learning~\citep{chen2021fairness,chen2022fairness}, fairness in assortment planning~\citep{lu2023simple,chen2022fair}, fair and dynamic rationing of scarce resources~\citep{manshadi2021fair}, group fairness in stochastic matching~\citep{ma2022group}, group fairness in offline and online combinatorial optimization~\citep{asadpour2023sequential,golrezaei2024online,tang2023beyond,niazadeh2023online}, refugee resettlement~\citep{freund2023group}, \revcolorm{fair incentives in repeated engagements~\citep{freund2025fair}, individual fairness in revenue management~\citep{arsenis2022individual,jaillet2024grace}, and limitations of Rooney rules~\citep{farajollahzadeh2025rooney} in contexts such as employment or admission.}

\smallskip
\paragraph{\revcolorm{Extensions of sequential search: Pandora's box \& beyond.}} Moving beyond fairness considerations, our work contributes to the rich literature on sequential search. 
Building on the seminal work of \citet{weitzman1979optimal}, numerous papers study sequential search in richer settings different from ours. Recent work in this direction includes studying settings with the option of selecting a box without inspection \citep{beyhaghi2019pandora, alaei2021revenue,doval2018whether,aouad2020pandora}, uncertainty in the availability of a box for inspection or selection \citep{brown2022sequential}, \revcolorm{ correlated or unknown reward distributions that should be learned or estimated via samples~\citep{chawla2020pandora, gatmiry2022bandit,agarwal2024semi}, and designing prior-independent search policies in the absence of prior information about candidate values~\citep{brown2025data}.} 
Modeling the stateful search as JMS enables us to capture several well-motivated variants of the Pandora's box problem while retaining the simple structure of the optimal policy, which is not necessarily the case in most of these other extensions. \revcolor{Another interesting extension of Pandora's box is when the ordering under which the boxes should be inspected is restricted by a given partial ordering, which is a model studied in \cite{boodaghians2020pandora}. Despite this restriction, they show a polynomial-time computable index-based optimal policy for this extension.} \revcolorm{Beyond Pandora's box model, other models for sequential search and hiring have been studied that are conceptually related, e.g., \cite{epstein2024selection} consider a variant of the stochastic probing problem to model the ordering and selection in hiring pipelines.}

\revcolor{
\smallskip
\paragraph{Restless bandits and weakly-coupled MDPs.}
Also related to us, mostly in terms of the philosophy in designing algorithms, is the growing literature on Restless multi-armed bandits (RMAB) and weakly coupled MDPs. In the RMAB setting, even the arms that are not played may evolve according to potentially different transition probability kernels. In its most general form, finding an optimal policy is computationally hard. More specifically, \cite{papadimitriou1999complexity} prove that even under known transition kernels and infinite-horizon average reward, finding the optimal policy is PSPACE-hard. In response, several works aim to provide approximate optimality results. The technique of relaxing the ex-post constraints of the problem to hold in expectation over the horizon, and then Lagrangifying this relaxed constraint into the objective is predominantly employed to obtain approximate or heuristic policies in this literature (and also in the literature on weakly coupled MDPs, e.g., \cite{hawkins2003langrangian,adelman2008relaxations}). This technique, at a high-level, has similarities to how we address ex-ante constraints in both the Pandora's box setting and JMS. The most prevalent heuristic related to this technique is an index policy, called the \emph{Whittle-index}, proposed by \cite{whittle1988restless}, which is only well-defined under certain indexability conditions and may still be intractable to compute in some cases \citep{nino2007dynamic}. Nevertheless, \cite{weber1990index} show the asymptotic optimality of the Whittle-index policy in infinite-horizon average reward under certain assumptions. Other works include showing an unbounded gap for the Whittle-index under finite-horizon and discounted infinite-horizon settings—even if the other assumptions in \cite{weber1990index} hold—and instead providing alternative LP-based heuristics with sublinear regret bounds \citep{zhang2021restless,ghosh2023indexability}; online learning of the Whittle-index when transition kernels are unknown \citep{wang2023optimistic}; and recent studies on incorporating fairness notions, such as a minimum fraction of times to pull each arm \citep{li2019combinatorial,wang2024online}---for which they show sublinear regret---or bounding the maximum time since the last pull of any arm \citep{li2022efficient}. The literature on RMAB is massive and growing fast, and we refer the interested reader to this recent survey \cite{nino2023markovian}.

While at the surface level there seem to be some similarities between our framework and the RMAB/Weakly-coupled MDP setting, the two settings are both semantically and mathematically quite different. In fact, there are some fundamental distinctions between our Pandora's box and JMS models and the RMAB problem. For details, see the discussion in \Cref{app:discussion-rmab}.
}

\revcolor{
\smallskip
\paragraph{Primal-dual methods for learning in games and applications in fairness.} The idea of using primal-dual method and online learning to solve games and linear programs goes to back to the seminal work of \cite{plotkin1995fast}  for solving fractional packing and covering LPs, with its roots in the classic work of \cite{blackwell1956analog} and \cite{dantzig1956primal}. In this framework, the problem, after Lagrangifying the constraints, can be reinterpreted as a min-max game  played between the primal player (who proposes a feasible solution) and the dual player (who picks dual variables corresponding to the constraints). See \cite{arora2012multiplicative} for a survey. This technique has also manifested in various forms in the literature and has given rise to iterative primal-dual algorithms  that rely on different first-order methods, e.g., \cite{lyu2019capacity, jiang2019achieving}, or Blackwell approachability, e.g., \cite{zhong2018resource} and \cite{niazadeh2023online}. More recently, a similar approach has proven to be highly useful for near-optimally solving constrained online linear and convex programming problems~\citep{agrawal2014dynamic, agrawal2014fast,balseiro2023best}.

In terms of applications, such methods have been utilized in various applications in operations research and computer science, including bandit problems with knapsack constraints~\citep{badanidiyuru2018bandits}, resource allocation problems~\citep{devanur2011near,balseiro2023best,agrawal2014dynamic,agrawal2014fast}, inventory pooling and capacity allocation problems with service level constraints~\citep{lyu2019capacity,jiang2019achieving}, packing and covering problems~\citep{plotkin1995fast}, dynamic matching for refugee resettlement~\citep{bansak2024dynamic}, and classification problems or combinatorial optimization problems with subgroup fairness constraints~\citep{kearns2018preventing,golrezaei2024online}. Lastly, the primal-dual method has also proved to be essential in designing online resource allocation and matching algorithms under adversarial arrival~\citep{karp1990optimal,mehta2007adwords,buchbinder2009online,huang2019online,feng2024batching,ekbatani2025online} and in different applications of these models~\citep{golrezaei2014real,ma2020algorithms,gong2022online,feng2021robustness,delong2023online,udwani2024stochastic,ekbatani2023online,feng2024two}. See \cite{mehta2013online} for a comprehensive survey on this topic. 

We highlight that our primal-dual approach in \Cref{sec:general}, which leads to a near-optimal near-feasible solution, at a high-level, is built on this standard framework. However, there are important distinctions that make our algorithmic results novel. See the discussion in  \Cref{sec:discussion-primal-dual} for more details. 

\smallskip
\paragraph{Exact and approximate algorithmic Carathéodory.} The ``Carathéodory problem'' is a fundamental problem in geometry and polyhedral optimization, which dates back to the classic work of \cite{caratheodory1911variabilitatsbereich}. The basic proof of the Carathéodory theorem is constructive and based on an algorithm that has access to various types of oracles (separation oracle, optimization oracle, validity oracle, and membership oracle in each face of the polytope). See \cite{grotschel1981ellipsoid,grotschel2012geometric,vishnoi2021algorithms} for details. More recently, this fundamental problem has been revisited through the lens of  approximations, and several fast algorithms are designed by leveraging tools from online adversarial learning that provide approximate solutions~\citep{mirrokni2017tight,combettes2023revisiting}. We remark that our algorithm (\Cref{alg:Ellipsoid Exact Caratheodory}) in \Cref{app:mutiple-affine} is essentially an algorithm for \emph{exact} algorithmic Carathéodory problem that \emph{only} uses (linear) optimization oracle. Compared to earlier methods mentioned, our method has the advantage of being simpler, using linear optimization oracle more efficiently, and not requiring any other oracle access to the polytope.

}

\smallskip
\paragraph{Computational aspects of Bayesian sequential decision making.} To the best of our knowledge, our paper is the first to offer an FPTAS for JMS with on-average constraints. However, related computational settings have been explored in the literature. \citet{segev2023efficient} present a framework for efficiently approximating fundamental stochastic combinatorial optimization problems, such as stochastic probing, improving approaches for various non-adaptive settings. \citet{aouad2020pandora} explore a generalization of the Pandora's box problem with sequential inspection, balancing information acquisition and cost efficiency, and offering near-optimal approximate schemes. Additionally, \cite{anari2019nearly} study Bayesian online allocations with laminar matroid constraints and provide an FPTAS when the matroid's depth is constant. Similar computational questions about the tractability of various Markov decision processes in Bayesian allocations have been studied in \cite{papadimitriou1987complexity,papadimitriou2021online}. Conceptually, there are some similarities between our approach and the framework in \cite{liu2018delayed}, which examines the delayed impact of fair machine learning, though the two settings are distinct and not directly comparable.

\revcolor{
\subsection{Discussion on the Differences Between our Pandora/JMS Framework and Restless Bandits}
\label{app:discussion-rmab}

Although our work employs index-based policies and uses Lagrangification technique to incorporate ex-ante constraints, our framework (in both the Pandora’s box and JMS models) differs fundamentally from the Restless Multi-Armed Bandits (RMAB) framework. Below, we highlight the key distinctions:

\begin{itemize}
\smallskip
    \item \textbf{Model primitives:} the two models are quite different in important ways, which drastically changes both the computational landscapes of these models, as well as the algorithm design principles: 
    \begin{enumerate}[label=(\roman*)]
        \item \textbf{Non-restless vs. restless arms:} In our Pandora’s box and JMS settings, an arm (or Markov chain) \emph{does not} evolve if we do not pull (inspect) it. The state changes occur only when we actively decide to inspect an arm and the global state is fully known at each step. By contrast, in RMAB problems, each arm’s state continues to evolve (restlessly) even if it is not selected, which complicates the state space or leads to partial observability.
        \item \textbf{Endogenous stochastic horizon vs. fixed/infinite exogenous horizon:} RMAB formulations consider an \emph{exogenous} and \emph{fixed} finite horizon $T$ or an infinite horizon with discounting. Our JMS and Pandora’s box settings, on the other hand, have a \emph{stochastic} horizon by stopping once a Markov chain reaches a terminal state (or once the set of Markov chains at terminal states satisfies an ex-post matroid constraint, such as a capacity $k$). As a result, the horizon is \emph{endogenous} to the algorithm.
        \item \textbf{Different computational complexities:} General RMAB problems are known to be PSPACE-hard~\citep{papadimitriou1999complexity,papadimitriou1987complexity}; a well-known index-based policy in this setting based on \emph{Whittle indices} is often just a (clever) heuristic, which can be  optimal (or even near-optimal or approximately optimal) only in special cases or under strong assumptions and asymptotic scenarios. By contrast, under mild assumptions, our JMS and Pandora’s box models admit a polynomial-time optimal index-based solution ( a variant of the Gittins index-based policy) under matroid ex-post constraints on termination (see \Cref{sec:optimal-JMS-arbitrary} and \Cref{app:JMS-general} for details; also see \citep{dumitriu2003playing}). Even after we add some ex-ante constraints on visit frequencies of states (which is an extra restriction on top of the ex-post termination constraint and the constraint to pull one arm at a time), we still maintain polynomial-time computability via dual-adjustments and specialized tie-breaking (\Cref{sec:pandora}), Carathéodory-type decompositions (\Cref{sec:caratheodory}), or FPTAS algorithms in case of convex constraints (\Cref{sec:general}).           
    \end{enumerate}
    We re-iterate that the key difference in computational tractability, together with quite different mathematical semantics of the two models as described above, underscore how the two frameworks, though superficially similar in their use of “index-like” ideas, are inherently different (in fact, JMS/Pandora settings are more like \emph{Bayesian bandits} setting in terms of algorithmic landscapes than RMAB).

    \smallskip
    \item \textbf{Nature of the Lagrangian approach:} while there are superficial similarities between the way that the Lagrangification techniques are used in our framework (for incorporating ex-ante constraints) and in the RMAB framework, in fact they are quite distinct and serve completely different purposes:
    \begin{enumerate}[label=(\roman*)]
        \item \textbf{Exact ex-post constraint vs. relaxed constraint}: In RMAB framework the ex-post constraint of “pull at most one arm each time” often gets relaxed into a single in-expectation constraint---e.g., “pull a total of $T$ arms over $T$ rounds in expectation.” This relaxation allows one to apply Lagrangification, which \emph{decouples} the problem across arms, leading to the well-known Whittle index approach. However, Whittle-index-based relaxation policies typically only guarantee feasibility in expectation (where all arms with a non-negative index are pulled at each time); additional rounding or scheduling heuristics (such as pulling highest index arm at each time, as in the classic Whittle index heuristic policy) are required to restore ex-post feasibility, often yielding \emph{approximate} guarantees or \emph{asymptotic near-optimality} guarantees~\citep{guha2010approximation})
        In contrast, we \emph{always} pull one arm in each time (the arm with the highest Gittins index), and also never relax our ex-post constraint on termination (e.g., a capacity on how many terminal states can be accepted). Our capacity (or matroid) constraint is enforced on \emph{every sample path}, and we only Lagrangify \emph{extra ex-ante constraints} (like demographic parity or quota) that are meant to hold in expectation. As a result, our method yields a fully feasible optimal constrained policy for the original problem with the given ex-post capacity-like constraint and the extra ex-ante constraints, without requiring \emph{any} further post-processing.
        \item \textbf{Post-Lagrangification behavior (decoupling vs. adjusted instance):} In RMAB, once you Lagrangify the relaxed constraint, the problem ``decouples'' into single-armed subproblems---one for each arm. By solving each single-armed problem, one can define and calculated the Whittle index for each arm. By contrast, when we Lagrangify our ex-ante constraints, the resulting ``adjusted'' problem remains a single, integrated instance of JMS or Pandora’s box with modified rewards or costs. It \emph{does not} decouple into multiple subproblems, and we still need to solve that adjusted instance. As we show, the solution to this problem is a dual-adjusted index-based policy, along with a randomized tie-breaking to guarantee exact feasibility (or near-feasibility in case of convex constraints). 
    \end{enumerate}
    The above distinctions in the way that Lagrangification helps with designing algorithms capture essential differences between the two models, both in terms of the type of the constraints (i.e., relaxation of an ex-post constraint vs. an original ex-ante constraint) and also how they are handled via this technique. It also serves as another evidence why the RMAB’s Whittle approach is not mathematically connected to the Pandora's box or JMS settings, with or without ex-ante constraints, and does not carry over there.

    \smallskip
    \item\textbf{Multiple or more complex constraints:} RMAB literature usually focuses on a single ex-post constraint (e.g., at most one arm pulled per round) or a relaxed version of that constraint. In our setting, we can incorporate \emph{multiple} ex-ante affine or convex constraints on the expected number of visits to any states (which can more complex than just a single affine constraint at the arm-level). We then solve the resulting constrained problem in polynomial time \emph{exactly} (for affine constraints) or via a \emph{near-optimal policy/FPTAS} (for convex constraints)---all while maintaining ex-post feasibility with respect to capacity or other structural constraints (like matroids). This level of generality and exact satisfaction of constraints distinguishes our framework even further from RMAB.
\end{itemize}

In conclusion, while there some similarities in the algorithmic philosophy used in our paper and this framework, the two models are mathematically and semantically different. As a result, to the best of our knowledge, there is no reduction or formal connection between our results and this literature.

}
\revcolor{
\subsection{Discussion on Distinctions of G-RDIP (\Cref{alg:RAI}) from the Standard Primal-Dual Method} 
\label{sec:discussion-primal-dual}
Using primal-dual ideas from learning-in-games to solve convex-concave saddle-point problems---where a primal player best-responds iteratively and a dual player runs an adversarial online learning algorithm to find the optimal dual that satisfies complementary slackness---is quite standard. In fact, as mentioned earlier, this approach dates back to seminal work on fractional packing and covering LPs (e.g., \cite{plotkin1995fast}; see also \cite{arora2012multiplicative} for a survey) and classic works of \cite{blackwell1956analog} and \cite{dantzig1956primal}. A similar algorithmic philosophy has also been extended to the online setting under i.i.d. stochastic arrivals (or variants such as random order or almost-i.i.d.), as in the work of \cite{agrawal2014dynamic, devanur2011near, agrawal2014fast}. It is important to note that typically in this framework, given the optimal dual variables, the primal player’s best response is simple, straightforward and computationally easy.


Although, at a high level, our G-RDIP approach ({\Cref{alg:RAI}}) is built on this standard framework---and indeed, we employ Fenchel duality for reasons similar to those in the above papers---we believe that our work has the following distinguishing aspects, which highlight its novelty:

\begin{itemize}
    \item \textbf{Computing the best response for the primal player:} Suppose that we only have ex-ante affine constraints. After Lagrangifying these constraints into the objective, the resulting best-response problem (i.e., maximizing the Lagrangian for a fixed set of dual variables) is equivalent to solving an unconstrained JMS problem with adjusted rewards. However, even this unconstrained JMS problem is nontrivial in the general setting, where some state rewards may be positive and others negative---situations that readily occur after dual adjustments. Previous work has analyzed this setting under the ``No Free Lunch (NFL)'' assumption on state rewards (see \Cref{app:JMS-general} and \Cref{def:NFL}), an assumption that can be violated after adjustments. To overcome this, in \Cref{app:JMS-general} we introduce a novel ``collapsing reduction'' that reduces an instance of JMS with arbitrary rewards to one that satisfies NFL in polynomial time. Consequently, we obtain a polynomial-time computable index-based algorithm for this more general version of the JMS problem, a result that we believe is of independent interest.

    \item \textbf{Two-Layer online learning and Fenchel duality for convex constraints:} Although we show how to compute a polynomial-time algorithm for the JMS problem with arbitrary rewards, this alone does not suffice when convex constraints are present. After Lagrangifying a convex constraint, the best-response problem becomes equivalent to solving an extension of the JMS problem where the objective is a concave function of the state visiting frequencies rather than a linear one. To our knowledge, no prior work addresses this specific problem, and it was not even known before our work that one could obtain a near-optimal solution for this problem. Inspired by the use of Fenchel duality in \cite{agrawal2014fast} and related work such as \cite{balseiro2023best}, we replace our convex constraints with their relaxations using Fenchel duals, introducing another set of dual variables. However, this leads to a technical challenge: the Lagrangian becomes linear in state frequencies but nonlinear (and indeed non-convex) in terms of the dual variables, so naively running online learning to minimize the Lagrangian would not work. We observe, however, that if we fix the Fenchel duals, the Lagrangian is linear in the remaining dual variables, and if we fix those, it is convex in the Fenchel duals. This observation suggests using ``two layers of online learning'' for the dual player and ``best response'' for the primal player to obtain a near-optimal solution. Accordingly, our final algorithm comprises an inner layer, where the Fenchel duals are learned, and an outer layer, where the remaining dual variables are learned. For more details, please refer to \Cref{sec:learning-two-layer} and \Cref{sec:RAI}, and see the analysis of G-RDIP in \Cref{app:jms} (proof of \Cref{thm:RAI}).
\end{itemize}
}

\revcolor{
\section{Missing Details of \texorpdfstring{\Cref{sec:single-affine}}{} and \texorpdfstring{\Cref{sec:single-affine-policy}}{}}

\subsection{Missing Technical Details of Section~\texorpdfstring{\ref{sec:single-affine}}{}}
\label{app:binding}

\begin{lemma}[\textbf{Checking for a Binding Constraint}]
\label{lemma:binding}
Suppose Problem~\ref{eq:opt-constrained} is feasible. If there exists an optimal solution $\policy^*$ for Problem~\ref{eq:opt-unconstrained} such that:
$$
\expect{\sum_{i\in[n]}\theta^{S}_i{\select_i^{\policy^*}}+ \sum_{i\in[n]}\theta^{I}_i{\inspect_i^{\policy^*}}}>b~,
$$
then there should exist an optimal solution $\hat\policy$ for Problem~\ref{eq:opt-constrained} for which Constraint~\ref{eq:affine-constraint} is binding, that is, 
$$
\expect{\sum_{i\in[n]}\theta^{S}_i{\select_i^{\hat\policy}}+ \sum_{i\in[n]}\theta^{I}_i{\inspect_i^{\hat\policy}}}=b~,
$$
\end{lemma}
\begin{proof}{\emph{Proof.}}
We prove the claim by contradiction. Suppose that the claim does not hold. Then, for any optimal policy $\policy'$ of Problem~\ref{eq:opt-constrained}---which exists due to the feasibility of Problem~\ref{eq:opt-constrained}---we have:
$$
\expect{\sum_{i\in[n]}\theta^{S}_i{\select_i^{\policy'}}+ \sum_{i\in[n]}\theta^{I}_i{\inspect_i^{\policy'}}}<b~.
$$

Because $\policy^*$ and $\policy'$ have negative and positive slacks in Constraint~\ref{eq:affine-constraint}, respectively, there exists a proper convex combination of these two policies for some $q\in(0,1)$, that is, a randomized policy $\tilde{\policy}$ that with probability $q$ runs $\policy'$ and with probability $1-q$ runs $\policy^*$, such that:

$$
q\left(b-\expect{\sum_{i\in[n]}\theta^{S}_i{\select_i^{\policy'}}+ \sum_{i\in[n]}\theta^{I}_i{\inspect_i^{\policy'}}}\right)+(1-q)\left(b-\expect{\sum_{i\in[n]}\theta^{S}_i{\select_i^{\policy^*}}+ \sum_{i\in[n]}\theta^{I}_i{\inspect_i^{\policy^*}}}\right)=0~.
$$
Therefore, applying the linearity of the expectation, the resulting randomized policy $\tilde{\policy}$ satisfies Constraint~\ref{eq:affine-constraint} in its equality form. Moreover, 
$$
\util(\tilde\policy;\Instance)=q \cdot\util (\policy';\Instance)+(1-q)\cdot\util (\policy^*;\Instance)\geq \util (\policy';\Instance)~,
$$
where $\Instance$ is the problem instance under consideration. The last inequality holds because $\util(\pi^*; \mathcal{I})\geq \util(\pi';\mathcal{I})$, since adding an ex-ante affine constraint to Problem~\ref{eq:opt-unconstrained} can only lower the objective value. Thus, $\policy'$ is also an optimal solution of Problem~\ref{eq:opt-constrained} for which Constraint~\ref{eq:affine-constraint} is binding, a contradiction. 
\qed
\end{proof}}

\revcolor{
\subsection{Missing Technical Details of \texorpdfstring{\Cref{sec:single-affine-policy}}{}}
\label{app:pandora-details}
\begin{proposition}
    \label{prop:pandora-refinement}
     \Cref{alg:Pandora} with $\NumSelect = 1$ implements the same ordering and stopping rule, up to tie-breaking, as in the optimal index-based policy in \cite{weitzman1979optimal}. Moreover, for $\NumSelect > 1$, this algorithm is exactly equivalent, again up to tie-breaking, to the optimal greedy (frugal) index-based policy in \cite{kleinberg2016descending,singla2018price}.
\end{proposition}
\begin{proof}{\emph{Proof.}}
To see why, for a moment, suppose that there are no boxes with zero or negative cost or reward, there are no ties in the indices, and the model primitives are such that $\reserve_i \neq v_j$, for all $i,j \in [n]$. Then \Cref{alg:Pandora} does the following: at each time if there is an opened box $i^*$ in the set $\mathcal{C}$ (which only happens if $v_{i^*}$ is larger than the indices of all the unopened boxes and the realized rewards of all the opened boxes so far), the algorithm stops and selects $i^*$; otherwise, it continues by opening an unopened box with the largest non-negative index. This algorithm is exactly equivalent to the optimal policy of \cite{weitzman1979optimal} described earlier in \Cref{subsub:base:pandora}. 

Similarly, for $k>1$,  \Cref{alg:Pandora} at each time greedily considers the unselected box $i^*$  in $\mathcal{C}$ with the maximum option value $o_i$. If the box is already open, then it should be among the top $k$ rewards in the set of opened boxes at the time of termination (as the algorithm terminates once the $k^{\textrm{th}}$ largest reward among opened boxes is larger than all the remaining indices) and therefore will be selected. If the capacity $k$ is reached after selection, then the algorithm terminates (as the $k^{\textrm{th}}$ largest reward among opened boxes is now larger than all of the remaining indices); otherwise, it continues by selecting more opened boxes or opening an unopened box with the largest index. Again, this algorithm is exactly equivalent to the optimal policy of \cite{kleinberg2016descending,singla2018price}, as described earlier in \Cref{subsub:base:pandora}. 
\qed
    
\end{proof}
}
\begin{lemma}[\textbf{Universality of the Refined Policy}]
\label{lemma:full-cover}
Any optimal policy for the Pandora's box problem (i.e., the special case of the unconstrained problem \ref{eq:opt-unconstrained} when $\NumSelect=1$) can be implemented by \Cref{alg:Pandora} with a proper choice of tie-breaking rule $\tau$.
\end{lemma}
\begin{proof}{\emph{Proof.}}
We start by recalling the definition of a ``non-exposed'' policy. 
A policy is said to be non-exposed if it is forced to eventually select any box $i$   that the policy has opened and observed that $v_i>\sigma_i$. Note that any optimal policy $\policy$ for the (the unconstrained version) of the Pandora's box problem must be non-exposed as shown in \cite{kleinberg2016descending}. Also, we remark that as stated in \citet{kleinberg2016descending, armstrong2017ordered}, \cite{singla2018price}, the Pandora's box problem (with multiple selections) can be viewed as a static discrete choice problem (with multiple selections) in which this optimal policy \emph{always} selects at most one (resp. $\NumSelect$) of the boxes with the largest non-negative realized ``random utility'' defined as 
\begin{equation}
\label{eq:capped:value}
\begin{aligned}
\Kap_i \triangleq \min\{\reserve_i, v_i\}
\end{aligned}
\end{equation}
Therefore, for the special case of $k=1$, any optimal policy should eventually select the box with the maximum non-negative $\kappa_i$ (if any) in every sample path. See \cite{kleinberg2016descending} for more details.
Now compare any optimal policy $\policy$ that satisfies the above properties with \Cref{alg:Pandora}. First of all, notice that any optimal policy has to inspect all of the negative-cost boxes. As a consequence, we can assume that policy $\policy$
inspects them first without changing the rest of the search process. 
Now, under this convention, we fix a realization of all the random variables in our instance and run both policies $\policy$ and \Cref{alg:Pandora}. Consider the first time step by which the policy $\policy$ decides to choose a box $i'$ that does not belong to the candidate set $\candidates$ (determined in line 5 of \Cref{alg:Pandora}), that is, it satisfies the following two conditions: (i) Box $i'$ does not have the maximum option value, that is, $o_{i^*}>o_{i'}$ at that time, and also (ii) $i'\notin \{i\in[n]\setminus \mathcal{O}:~c_i = 0\}$. 

\medskip 
Let us consider two cases separately:

\medskip  
{\bf Case (a):} Box $i'$ has a negative or zero cost or it is already open.
If box $i'$ has negative cost, then, by our convention, it is open. If box $i'$ has cost zero, then by condition (ii)  -- stated above -- it is also open. 
Thus, choosing box $i'$ implies that $\policy$ selects $i'$ with probability one. (Box $i'$ can also be the outside option which is by default open and again choosing it means selecting it.)  

{\bf Case (b):} Box $i'$ is not yet opened and has a positive cost. Again, we show that $\policy$ will select $i'$ with positive probability. Note that after $i'$ is inspected by $\policy$, which occurs in the next step, the policy still would have enough capacity for selection (as it had before).  Since $c_{i'} > 0$,  there is a positive probability that the realization of $v_{i'}$ after inspecting box $i'$ satisfies $v_{i'}>\sigma_{i'}$. Furthermore, by the non-exposedness property of $\policy$, this policy is forced to select $i'$ in that case.

\medskip 
In summary, until now, we showed that there is a positive probability that $\policy$ will be forced to eventually select box $i'$. In such sample paths (with nonzero measure), a few cases might arise:

{\bf Case (1):} $i^*$ has already been inspected when $i'$ is chosen by $\policy$. First, this implies $v_{i^*}=o_{i^*}>o_{i'}$. In this case, if $v_{i^*}>\sigma_{i^*}$, then it implies that $\policy$ is forced to also select $i^*$ (by its non-exposed property). This results in a contradiction because $\policy$ cannot select both $i'$ and $i^*$. If $v_{i^*}\leq\sigma_{i^*}$, then $\kappa_{i^*}=v_{i^*}>o_{i'}\geq\kappa_{i'}$, which is again a contradiction since it implies that the selected box $i'$ by $\policy$ cannot be the box with a maximum non-negative value of $\kappa_i$ in certain sample paths.

{\bf Case (2):} $i^*$ is not inspected when $i'$ is chosen by $\policy$. 
First, this implies $\sigma_{i^*}=o_{i^*}>o_{i'}$. Note also that $c_{i^*}\geq 0$, and hence there is a positive probability that $v_{i^*}\geq \sigma_{i^*}$. In that case, $\kappa_{i^*}=\sigma_{i^*}>o_{i'}\geq\kappa_{i'}$, which is again a contradiction since it implies that the selected box $i'$ by $\policy$ cannot be the box with a maximum non-negative value of $\kappa_i$ in certain sample paths.

Putting everything together, all possible cases result in a contradiction, and hence we show that $\policy$ cannot be an optimal policy, as desired.\qed



\end{proof}



\subsection{Missing Proofs of Section~\texorpdfstring{\ref{sec:single-affine-policy}}{}}
\label{app:sec:pandora}
\begin{proof}{\emph{Proof of \Cref{lemma:Const:DualLagrangeProperties}.}}
We provide separate proofs for three parts of this lemma

\smallskip
\begin{enumerate}[leftmargin=*,label=(\roman*)]
    \item To prove this part, note that for any $\LagVector$ there exists always a deterministic policy that maximizes the Lagrangian $\LagrangeConst (\LagVector, \pi)$.  Therefore, w.l.o.g., we can restrict ourselves only to the set of all deterministic policies when computing $\DualLagrangeConst (\LagVector)$. Now note that there are finitely many deterministic policies, simply because the number of boxes $\altnum$ and the number of possible value realizations are both finite (due to the discreteness assumption on the values). Fixing a deterministic policy $\policy$, the Lagrangian function is linear in $\LagVector$. Therefore, the function $\DualLagrangeConst (\LagVector)$ is the maximum over a finite number of linear functions, which implies that $\DualLagrangeConst (\LagVector)$ is piecewise linear and convex.

\smallskip

    \item Recall that $\DualLagrangeConst (\LagVector)$ is always an upper-bound on the objective value of any feasible policy $\policy$ in the primal problem, i.e., Problem~\ref{eq:opt-constrained}. Because this problem is assumed to be feasible (see the discussion after \Cref{rem:feasible} and \Cref{rem:tightness} in \Cref{sec:single-affine}), $\DualLagrangeConst$ is bounded from below. Given that $\DualLagrangeConst$ is piecewise linear and convex, it should always have a bounded minimizer $\LagVector^*$.
    
    \item For any given policy $\policy$, by taking partial derivative of $\LagrangeConst (\LagVector;\pi)$ with respect to $\LagVector$, we have:
    $$
    \frac{\partial\LagrangeConst}{\partial\LagVector}=b-\expect{\sum_{i\in[n]}\theta^{S}_i{\select_i^{\policy}}+ \sum_{i\in[n]}\theta^{I}_i{\inspect_i^{\policy}}} =\slack_{\textsc{cons}}^{\policy}~.
    $$
 Therefore, by a simple application of the envelope lemma, we have:

 $$
 \slack_{\textsc{cons}}^{\policy^{\lambda}}\in \partial\left(\underset{\policy\in\PolicySpace}{\max}~\LagrangeConst(\cdot; \pi)\right)(\LagVector)=\partial\left(\LagrangeConst(\cdot; \policy^{\cdot})\right)\left(\LagVector\right)=\partial\DualLagrangeConst\left(\LagVector\right)~,
 $$
 where $\policy^{\lambda}$ is any policy maximizing the Lagrangian for a given $\LagVector$, i.e., $\policy^{\LagVector}\in\underset{\policy\in\PolicySpace}{\argmax}~\LagrangeConst(\LagVector;\policy)$, and $\partial f(x)$ is the set of all subgradients of $f$ at point $x$.

\end{enumerate}
\qed
\end{proof}

\begin{proof}{\emph{Proof of \Cref{lem:const:slack}.}}
\noindent As mentioned in the proof sketch, this proof involves two steps. 

\medskip  
{\bf \underline{Step 1}:} 
We show that there exists an optimal policy with a non-negative constraint slack. As we established in \Cref{lemma:Const:DualLagrangeProperties}, the function $\DualLagrangeConst (\LagVector)$ is the maximum of linear functions and, hence, is a convex piecewise-linear function. This piecewise-linear function has breakpoints at certain values, each corresponding to a $\LagVector$ at which the slope $\frac{\partial}{\partial\LagVector}\DualLagrangeConst(\LagVector)$ changes. Using part~(iii) of \Cref{lemma:Const:DualLagrangeProperties}, the slope of any differentiable line segment of $\DualLagrangeConst$ is equal to the slack $\slack^{\pi^{\LagVector}}_\textsc{cons}$, where $\LagVector$ is an arbitrary point lying in that line segment and $\pi^\LagVector$ is an arbitrary optimal policy for the adjusted instance with respect to $\LagVector$.
Also, the space of possible policies (due to the finite and discrete support assumption on the rewards) is finite, and hence there are finitely many breaking points in this piecewise linear function. 

Now, fixing any $\LagVector^*\in\underset{\LagVector}{\argmin}~\DualLagrangeConst(\LagVector)$, there must exist a sufficiently small positive $\varepsilon$ such that both $\LagVector^*$ and $\LagVector^*+\varepsilon$ lie in the same line segment, or equivalently, the function $\DualLagrangeConst$ is a line in the interval $[\LagVector,\LagVector^*+\varepsilon$]. For this choice of $\varepsilon$, we show that any optimal policy $\pi^{\LagVector^*+\varepsilon}$ for $\max_{\policy \in \PolicySpace} ~\LagrangeConst(\LagVector^*+\varepsilon, \pi)$ is also an optimal policy for $\max_{\policy \in \PolicySpace} ~\LagrangeConst (\LagVector^*, \pi)$. To see this, first note that $\pi^{\LagVector^*+\varepsilon}$ is also an optimal policy for $\max_{\policy \in \PolicySpace} ~\LagrangeConst (\LagVector, \pi)$ for any choice of $\lambda\in(\LagVector^*,\LagVector^*+\varepsilon]$ (hence, $\DualLagrangeConst(\LagVector)=\LagrangeConst(\LagVector,\pi^{\LagVector^*+\varepsilon})$ for any $\lambda\in(\LagVector^*,\LagVector^*+\varepsilon]$). This last statement holds because the optimal policy for the adjusted instance with any $\lambda\in(\LagVector^*,\LagVector^*+\varepsilon]$ has the same slack $\slack_{\textsc{cons}}^{\pi^{\LagVector}}=\slack_{\textsc{cons}}^{\pi^{\LagVector^*+\varepsilon}}$ regardless of the choice of $\LagVector$. Therefore, if $\pi^{\LagVector^*+\varepsilon}$ is not optimal for an adjusted instance with some $\lambda\in(\LagVector^*,\LagVector^*+\varepsilon)$, we should have:
$$
\LagrangeConst(\lambda,\pi^\lambda)>\LagrangeConst(\lambda,\pi^{\LagVector^*+\varepsilon})~\textrm{and}~\slack_{\textsc{cons}}^{\pi^{\LagVector}}=\slack_{\textsc{cons}}^{\pi^{\LagVector^*+\varepsilon}}~~\Longrightarrow~~\util(\pi^\lambda;\Instance)>\util(\pi^{\LagVector^*+\varepsilon};\Instance)~.
$$
However, this is in contradiction to the optimality of $\pi^{\LagVector^*+\varepsilon}$ for the adjusted instance with $\LagVector^*+\varepsilon$, simply because we can show:
$$
\util(\pi^\lambda;\Instance)>\util(\pi^{\LagVector^*+\varepsilon};\Instance )~\textrm{and}~\slack_{\textsc{cons}}^{\pi^{\LagVector}}=\slack_{\textsc{cons}}^{\pi^{{\LagVector^*+\varepsilon}}}~~\Longrightarrow~~\LagrangeConst(\lambda^*+\varepsilon,\pi^\lambda)>\LagrangeConst(\lambda^*+\varepsilon,\pi^{\LagVector^*+\varepsilon})~.
$$
Second, note that the function $\DualLagrangeConst$ is continuous and therefore we have:
$$
\max_{\policy \in \PolicySpace} ~\LagrangeConst (\LagVector^*, \pi)=\DualLagrangeConst(\LagVector^*)=\underset{\delta\rightarrow 0^+}{\lim}~{\DualLagrangeConst(\LagVector^*+\delta)}
$$
Now, for $\delta\in(0,\varepsilon)$, we can replace $\DualLagrangeConst(\LagVector^*+\delta)$ with $\LagrangeConst(\LagVector^*+\delta,\pi^{\LagVector^*+\varepsilon})$ as stated above. So we have:
$$
\max_{\policy \in \PolicySpace} ~\LagrangeConst (\LagVector^*, \pi)=\underset{\delta\rightarrow 0^+}{\lim}~{\LagrangeConst(\LagVector^*+\delta,\pi^{\LagVector^*+\varepsilon})}=\LagrangeConst(\LagVector^*,\pi^{\LagVector^*+\varepsilon})~,
$$
where the last equality is retained due to the continuity of the Lagrangian function $\LagrangeConst(\lambda,\pi)$ with respect to $\lambda$ for a fixed $\pi$ (it is indeed a linear function). Therefore, $\pi^{\LagVector^*+\varepsilon}$ is also an optimal policy for the adjusted instance with $\LagVector^*$.

Furthermore, by convexity, all subgradients of $\DualLagrangeConst$ at $\LagVector^*+\varepsilon$ must be nonnegative, implying that the constraint slack $\slack_{\textsc{cons}}^{\pi^{\LagVector^*+\varepsilon}}$ of policy $\pi^{\LagVector^*+\varepsilon}$ is nonnegative. Therefore, the policy $\pi^{\LagVector^*+\varepsilon}$ is an adjusted optimal policy (with respect to $\LagVector^*$) with nonnegative slack. Similarly, by repeating the same argument for $\LagVector^*-\varepsilon$ for small enough $\varepsilon$,  we conclude that there exists an adjusted optimal policy $\pi^{\LagVector^*-\varepsilon}$ (with respect to $\LagVector^*$) with nonpositive constraint slack, which completes the proof of this step.

\medskip  

{\bf \underline{Step 2}:}
In this step, we show $\slack_{\textsc{cons}}^{\policy^{+}} \geq 0$ (resp. $0\geq \slack_{\textsc{cons}}^{\policy^{-}}$). To show this, we will look at the tie-breaking rule that arises from the perturbed adjusted problem with $\LagVector = \LagVector^*-\varepsilon$ for an \emph{infinitesimal} $\varepsilon>0$ when we run \Cref{alg:const}. More formally, we show that there exists a run of \Cref{alg:const} on the adjusted instance with $\LagVector^*-\varepsilon$ that is exactly equivalent to running policy $\pi^-$, that is, a run of \Cref{alg:const} on the adjusted instance $\LagVector^*$ with the negative-extreme tie-breaking rule $\tau^-$ (which uses tie-breaking scores $\{s^-_i\}_{i\in\candidates}$ as in \Cref{def:tie:extreme}).

First, suppose that $\varepsilon$ is small enough so that the perturbation will not change any strict order among the possible realizations of the adjusted values $\widetilde{v_i}$ and the adjusted indices $\widetilde{\reserve_i}$ (and henceforth the adjusted option values $\widetilde{o_i}$), where the adjustment is with respect to $\LagVector^*$.  Moreover, we let $\varepsilon$ be small enough so that the sign of no adjusted cost $\widetilde{c_i}$ changes. Therefore, if an adjusted cost $\widetilde{c_i}$ with respect to $\LagVector^*$ is strictly positive (resp. strictly negative), it remains strictly positive (resp. strictly negative) after adjustment with respect to $\LagVector^*-\varepsilon$,  regardless of the sign of $\theta_i^I$. It is also important to note that if the adjusted cost $\widetilde{c_i}$ with respect to $\LagVector^*$ is exactly equal to zero, if we have $\theta^I_i<0$, then the adjusted cost with respect to $\LagVector^*-\varepsilon$ turns out to be strictly positive (but infinitely close to zero), and if $\theta^I_i>0$, then the adjusted cost with respect to $\LagVector^*-\varepsilon$ turns out to be strictly negative. Furthermore, if $\theta^I_i=0$ and the adjusted cost $\widetilde{c_i}$ with respect to $\LagVector^*$ is zero, it remains zero in the adjusted instance with respect to $\LagVector^*-\varepsilon$. In the rest of the proof, we use the notation $\widetilde{v_i}^-$, $\widetilde{\reserve_i}^-$, $\widetilde{o_i}^-$, and $\widetilde{c_i}^-$ to denote adjusted values, adjusted reservation values (or indices), adjusted option values, and adjusted costs corresponding to the adjustment $\LagVector^*-\varepsilon$.

Next, consider an optimal policy $\pi^{\LagVector^*-\varepsilon}$ for the adjusted instance with $\LagVector^*-\varepsilon$ for infinitesimal $\varepsilon$. Importantly, for any choice of tie-breaking rule for this policy, we have $\slack_{\textsc{cons}}^{\pi^{\LagVector^*-\varepsilon}}\leq 0$.
Now to compare this policy with $\pi^-$, we run both of these policies by using \Cref{alg:const} with the same instance with different adjustments as input, that is, we run $\pi^-$ on the adjusted instance with $\LagVector^*$ and run $\pi^{\LagVector^*-\varepsilon}$ on the adjusted instance with $\LagVector^*-\varepsilon$. We also couple the sample path realizations of the two runs. Now suppose inductively that the two policies have made exactly the same decisions up to some iteration of \Cref{alg:const}. We then show that they can continue making the same decision while maintaining valid runs of both policies, which completes the proof. More precisely, suppose that policy $\pi^{-}$ picks $i\in[\altnum]$ from the set of candidates $\candidates$ to inspect (if $i$ is not yet open) or add to the selection set (if $i$ is already open) in the current iteration.  We show that picking $i$ in this iteration would be a valid choice for policy $\pi^{\LagVector^*-\varepsilon}$.

To show the above claim, we consider the following cases:
\begin{itemize}
\item If box $i$ is not yet open and $\widetilde{c_i}<0$ (which implies $\widetilde{c_i}^-<0$), then $\widetilde{\sigma_i}^-=\widetilde{o_i}^-=+\infty$ and hence this box will be among the boxes in $[\altnum]\setminus\mathcal{S}$ (i.e., set of unselected boxes) with the maximum option value in the current iteration of $\pi^{\LagVector^*-\varepsilon}$. Accordingly, the box $i$ can be chosen by $\pi^{\LagVector^*-\varepsilon}$ in this iteration, as desired.

\item If box $i$ is not yet open, $\widetilde{c_i}=0$ and $\theta_i^I\geq 0$, then we either have $\widetilde{c_i}^-=0$ (when $\theta_i^I=0$) or $\widetilde{c_i}^-<0$ (when $\theta_i^I>0$). In the former case, box $i$ will be among the candidate boxes in the current iteration of $\pi^{\LagVector^*-\varepsilon}$, as it is an unopened zero-cost box. In the latter case, $\widetilde{\sigma_i}^-=\widetilde{o_i}^-=+\infty$ and this box will be among the boxes in $[\altnum]\setminus\mathcal{S}$ with the maximum option value and, therefore, among the candidate boxes in the current iteration of $\pi^{\LagVector^*-\varepsilon}$. When the two cases are combined, we conclude that box $i$ can be chosen by policy $\pi^{\LagVector^*-\varepsilon}$ in this iteration, as desired.

\item If box $i$ is not yet open and either $\widetilde{c_i}>0$, or $\widetilde{c_i}=0$ and $\theta_i^I<0$,  we first show that it should be among the boxes in $([\altnum]\setminus\mathcal{S})\cup\{0\}$  with the maximum option value in the run of $\pi^-$, and hence $\widetilde{\sigma_i}=\widetilde{o}_{\textrm{max}}$, where
$$\displaystyle\widetilde{o}_{\textrm{max}}=\underset{i'\in ([\altnum]\setminus\mathcal{S})\cup\{0\}}{\max}~\tilde{o_{i'}}$$
To prove this statement, note that if $\widetilde{c_i}>0$ this statement is clearly true, as we know that $i$ is among the candidate boxes of $\pi^-$ in the current iteration. If $\widetilde{c_i}=0$ and $\theta_i^I<0$, note that $\widetilde{\sigma_i}$ is set to the upper-support $\bar{v}$ of the distribution of $\widetilde{v_i}$. If $\bar{v}=\widetilde{\sigma_i}<\widetilde{o}_{\textrm{max}}$ then $\prob{\widetilde{v_i}\geq \widetilde{o}_{\textrm{max}}}=0$. This is a contradiction, as according to \Cref{def:tie:extreme} the tie-breaking score $s_i^-$ of box $i$ should be set to $s^-_{i}=-\infty$ by $\pi^-$, so $\pi^-$ is not allowed to choose $i$ among the candidate boxes in $\mathcal{C}$ (note that there is always at least one box with a bounded tie-breaking score in $\mathcal{C}$). 

Next, we show that this box $i$ should also be among the boxes in $[\altnum]\setminus\mathcal{S}\cup\{0\}$ with the maximum option value in $\pi^{\LagVector^*-\varepsilon}$. This statement implies that  box $i$ is in the set of candidates in the current iteration of $\pi^{\LagVector^*-\varepsilon}$ and therefore can be chosen by this policy in this iteration, as desired. First, observe that $\widetilde{c_i}^->0$, as either $\widetilde{c_i}>0$, or $\widetilde{c_i}=0$ and $\theta_i^I<0$. Second, 
observe that $\widetilde{o_i}^-=\widetilde{\sigma_i}^-$, and we have:
$$
\expect{\left(\widetilde{v_i}^--\widetilde{\sigma_i}^-\right)^+}=\expect{\left(\widetilde{v_i}+\varepsilon\theta_i^S-\widetilde{\sigma_i}^-\right)^+}=\widetilde{c_i}^-=\widetilde{c_i}-\theta_i^I\varepsilon
$$
Under our conditions in this case, if $\theta_i^I\geq 0$ (and hence $\widetilde{c_i}>0$), we have $\widetilde{\sigma_i}^-=\widetilde{\sigma_i}+\varepsilon\theta_i^S+\delta$ for some infinitesimal $\delta\geq 0$. This simply holds because $\varepsilon>0$ is infinitesimal. Furthermore, given that $\widetilde{c_i}>0$, we have $\prob{\widetilde{v_i}>\widetilde{\sigma_i}}>0$, and therefore:
$$
\expect{\left(\widetilde{v_i}^--\widetilde{\sigma_i}^-\right)^+}=\expect{\left(\widetilde{v_i}-\widetilde{\sigma_i}\right)^+}-\delta \cdot\prob{\widetilde{v_i}>\widetilde{\sigma_i}}=\widetilde{c_i}-\theta_i^I\varepsilon~~\Rightarrow ~~\delta=\varepsilon\frac{\theta_i^I}{\prob{\widetilde{v_i}>\widetilde{\sigma_i}}}=\varepsilon\frac{\theta_i^I}{\prob{\widetilde{v_i}>\widetilde{o}_{\textrm{max}}}}
$$
Similarly, if $\theta_i^I< 0$ (and hence $\widetilde{c_i}\geq 0$), we have $\widetilde{\sigma_i}^-=\widetilde{\sigma_i}+\varepsilon\theta_i^S-\delta$ for some infinitesimal $\delta\geq 0$. Again, this holds because $\varepsilon>0$ is infinitesimal. Moreover, $\prob{\widetilde{v_i}\geq \widetilde{\sigma_i}}>0$. This is true because either we have $\widetilde{c_i}>0$, or $\widetilde{c_i}=0$ and $\widetilde{\sigma_i}$ is set to the upper-support $\bar{v}$ of the distribution of $\widetilde{v_i}$ and hence $\prob{\widetilde{v_i}\geq \widetilde{\sigma_i}}=\prob{\widetilde{v_i}= \bar{v}}>0$. Therefore, we have:
$$
\expect{\left(\widetilde{v_i}^--\widetilde{\sigma_i}^-\right)^+}=\expect{\left(\widetilde{v_i}-\widetilde{\sigma_i}\right)^+}+\delta \cdot\prob{\widetilde{v_i}\geq\widetilde{\sigma_i}}=\widetilde{c_i}-\theta_i^I\varepsilon~~\Rightarrow~~ \delta=-\varepsilon\frac{\theta_i^I}{\prob{\widetilde{v_i}\geq \widetilde{\sigma_i}}}=\varepsilon\frac{\theta_i^I}{\prob{\widetilde{v_i}\geq\widetilde{o}_{\textrm{max}}}}
$$
Putting the pieces together, the following holds for any unopened box with $\widetilde{c_i}>0$, or $\widetilde{c_i}=0$ and $\theta_i^I<0$:
\begin{equation}
\label{eq:rank}
\widetilde{\sigma_i}^-=\widetilde{\sigma_i}+\varepsilon\cdot s^-_i~,
\end{equation}
where $s^-_i$ is the tie-breaking score of positive-extreme rule $\tau^-$ as in \Cref{def:tie:extreme}. Now consider another box $i'\in [\altnum]\setminus\mathcal{S}\cup\{0\},~i'\neq i$. If $\widetilde{o_i}>\widetilde{o_{i'}}$, then $\widetilde{o_i}^->\widetilde{o_{i'}}^-$ as $\varepsilon$ is infinitesimal. Now suppose that $\widetilde{o_i}=\widetilde{o_{i'}}$ (and hence $i'$ is also among the boxes with the maximum option value in $[\altnum]\setminus\mathcal{S}\cup\{0\}$). First, note that if $\widetilde{c_{i'}}<0$, or $\widetilde{c_i}=0$ and $\theta_{i'}^I\geq 0$, then the tie-breaking score $s^-_{i'}$ of $i'$ is set to $s^-_{i'}=+\infty$, so $i'$ should be favored over $i$ by $\pi^-$ as $s^-_i<+\infty=s^-_{i'}$, a contradiction to the fact that $\pi^-$ has picked $i$ among the candidate boxes in $\mathcal{C}$. Therefore, $\widetilde{c_{i'}}> 0$, or $\widetilde{c_{i'}}=0$ and $\theta_{i'}^I < 0$. Now consider two cases:
\begin{itemize}
    \item If box $i'$ is not open yet, we have: 
$$
\widetilde{o_i}^-=\widetilde{\sigma_i}^-=\widetilde{\sigma_i}+\varepsilon \cdot s^-_i\overset{(a)}{\geq}\widetilde{\sigma_i}+ \varepsilon \cdot s^-_{i'}\overset{(b)}{=}\widetilde{\sigma_{i'}}+ \varepsilon \cdot s^-_{i'}\overset{(c)}{=}\widetilde{\sigma_{i'}}^-=\widetilde{o_{i'}}^-~,
$$
where the inequality~(a) holds as $i$ has the maximum tie-breaking score $s^-_i$ in the set of candidates $\mathcal{C}$, equality~(b) holds as $\widetilde{\sigma_i}=\widetilde{o_i}=\widetilde{o_{i'}}=\widetilde{\sigma_{i'}}$ as $i'$ is not open, and equality~(c) holds due to \Cref{eq:rank} applied to box $i'$ (as we proved earlier, this equation holds if $\widetilde{c_{i'}}> 0$, or $\widetilde{c_{i'}}=0$ and $\theta_{i'}^I < 0$).
\smallskip
\item If box $i'$ is already open, or $i'=0$ (outside option) we have: 
$$
\widetilde{o_i}^-=\widetilde{\sigma_i}^-=\widetilde{\sigma_i}+\varepsilon \cdot s^-_i\overset{(a)}{\geq}\widetilde{\sigma_i}+ \varepsilon \cdot s^-_{i'}\overset{(b)}{=}\widetilde{v_{i'}}+ \varepsilon \cdot s^-_{i'}\overset{(c)}{=}\widetilde{v_{i'}}^-=\widetilde{o_{i'}}^-~,
$$
where the inequality~(a) holds as $i$ has the maximum tie-breaking score $s^-_i$ in the set of candidates $\mathcal{C}$, equality~(b) holds as $\widetilde{\sigma_i}=\widetilde{o_i}=\widetilde{o_{i'}}=\widetilde{v_{i'}}$ as $i'$ is open, and equality~(c) holds because for an opened box $i'$, $s^-_{i'}=\theta_i^S$ and $\widetilde{v_{i'}}^-=\widetilde{v_{i'}}+\varepsilon\cdot\theta_i^S$ (as a convention, set $\theta_{0}^S=0$ for the outside option).

\end{itemize}

\item If box $i$ is already open or $i=0$ (outside option), then it should be among the boxes in $[\altnum]\setminus\mathcal{S}\cup\{0\}$ with the maximum option value in the run of $\pi^-$. We now show that this box will also be among the boxes in $[\altnum]\setminus\mathcal{S}\cup\{0\}$ with the maximum option value in the run of $\pi^{\LagVector^*-\varepsilon}$. Consider another box $i'\in [\altnum]\setminus\mathcal{S}\cup\{0\},~i'\neq i$. Similar to the previous case, if $\widetilde{o_i}>\widetilde{o_{i'}}$, then $\widetilde{o_i}^->\widetilde{o_{i'}}^-$ as $\varepsilon$ is infinitesimal. Now suppose that $\widetilde{o_i}=\widetilde{o_{i'}}$ (and hence $i'$ is also among the boxes with the maximum option value in $[\altnum]\setminus\mathcal{S}\cup\{0\}$). Similar to the previous case, the fact that $i$ is picked over $i'$ implies that $\widetilde{c_{i'}}> 0$, or $\widetilde{c_{i'}}=0$ and $\theta_{i'}^I < 0$. Now we have two cases (as a convention, set $\theta_{0}^S=0$ for the outside option):
\begin{itemize}
    \item If box $i'$ is not open yet, we have: 
$$
\widetilde{o_i}^-=\widetilde{v_i}^-\overset{(a)}{=}\widetilde{v_i}+\varepsilon \cdot s^-_i\overset{(b)}{\geq}\widetilde{v_i}+ \varepsilon \cdot s^-_{i'}\overset{(c)}{=}\widetilde{\sigma_{i'}}+ \varepsilon \cdot s^-_{i'}\overset{(d)}{=}\widetilde{\sigma_{i'}}^-=\widetilde{o_{i'}}^-~,
$$
where the equality~(a) holds as for the opened box $i$ we have $\widetilde{v_i}^-=\widetilde{v_i}+\varepsilon\cdot\theta_i^S$ and $s^-_i=\theta_i^S$, inequality~(b) holds as $i$ has the maximum tie-breaking score $s^-_i$ in the set of candidates $\mathcal{C}$, equality~(c) holds as $\widetilde{v_i}=\widetilde{o_i}=\widetilde{o_{i'}}=\widetilde{\sigma_{i'}}$ as $i'$ is not open, and equality~(d) holds due to \Cref{eq:rank} applied to box $i'$ (as we proved earlier, this equation holds if $\widetilde{c_{i'}}> 0$, or $\widetilde{c_{i'}}=0$ and $\theta_{i'}^I < 0$).
\smallskip
\item If box $i'$ is already open, or $i'=0$ (outside option) we have: 
$$
\widetilde{o_i}^-=\widetilde{v_i}^-\overset{(a)}{=}\widetilde{v_i}+\varepsilon \cdot s^-_i\overset{(b)}{\geq}\widetilde{v_i}+ \varepsilon \cdot s^-_{i'}\overset{(c)}{=}\widetilde{v_{i'}}+ \varepsilon \cdot s^-_{i'}\overset{(d)}{=}\widetilde{v_{i'}}^-=\widetilde{o_{i'}}^-~,
$$
where the equality~(a) holds as for the opened box $i$ we have $\widetilde{v_i}^-=\widetilde{v_i}+\varepsilon\cdot\theta_i^S$ and $s^-_i=\theta_i^S$, inequality~(b) holds as $i$ has the maximum tie-breaking score $s^-_i$ in the set of candidates $\mathcal{C}$, equality~(c) holds as $\widetilde{v_i}=\widetilde{o_i}=\widetilde{o_{i'}}=\widetilde{v_{i'}}$ as $i'$ is open, and equality~(d) holds because for an opened box $i'$, $s^-_{i'}=\theta_i^S$ and $\widetilde{v_{i'}}^-=\widetilde{v_{i'}}+\varepsilon\cdot\theta_i^S$.
\end{itemize}
 Putting the above cases together, we have $\widetilde{o_i}^-\geq \widetilde{o_{i'}}^-$ for any box $i'\in [\altnum]\setminus\mathcal{S}\cup\{0\},~i'\neq i$, as desired.
\end{itemize}

The proof of the counterpart statement of
$\slack_{\textsc{cons}}^{\policy^{+}}\geq 0$ follows a similar line of argument (and with exactly the same case analysis), which we omit for brevity. \qed
\end{proof}

\smallskip

\if false 

\begin{proof}[{Proof of \Cref{remark:extreme-FS}.}]
\RN{this proof has not been revised yet}
First of all notice that any adjusted optimal policy will select exactly one of the alternatives with the largest $\widetilde{\Kap_i}$ (including $i=0$), at each realization of the sample path. Now to show that the mentioned policy that always gives highest priority to $\WomanSet$, then to $i=0$, and lastly lowest priority to $\ManSet$ when choosing $i^*$ - let's call one such policy as $\pi_\WomanSet$ - will have the largest slack, its enough to show that it will get the largest possible slack in any realization of the sample path. In other words, its enough to prove that in any realization, if the policy $\pi_\WomanSet$ select an alternative that belongs to $\ManSet$, then any other adjusted optimal policy will also end up selecting from $\ManSet$, and if $\pi_\WomanSet$ selects $i=0$, then no optimal policy can select someone from $\WomanSet$. To prove this, let us first define $\CurrentIndex_i^t$ as the value of $\CurrentIndex_i$ at iteration t of the while loop in \cref{alg:Pandora}. Now first suppose $\pi_\WomanSet$ selects alternative j, which belongs to $\ManSet$, at iteration T. Now there are 2 cases, i) $\widetilde{v_j}\leq \widetilde{\reserve_j}$, ii)  $\widetilde{v_j} > \widetilde{\reserve_j}$:

\begin{enumerate}[label=(\roman*)]
\item if $\widetilde{v_j}\leq \widetilde{\reserve_j}$: \\
Since at iteration T the policy has chosen j, we know that there was not any alternative from $\WomanSet\cup\{0\}$ available at T among the argmax. Additionally, notice that $\forall i,t$: $\CurrentIndex_i^t$ is equal to either $\widetilde{v_i}$ or $\widetilde{\reserve_i}$. Therefore, we can conclude that $\forall i,t: ~\widetilde{\Kap_i} = min(\widetilde{v_i},\widetilde{\reserve_i}) \leq \CurrentIndex_i^t$.
Now from the fact that there were no alternative from $\WomanSet\cup\{0\}$ present in the argmax at T, we can deduce that $\forall i\in \WomanSet\cup\{0\}: \CurrentIndex_i^T<\CurrentIndex_j^T$. Finally notice that in this case we have $\CurrentIndex_j^T=\widetilde{v_j}=\widetilde{\Kap_j}$. Hence, $\forall i\in \WomanSet\cup\{0\}: \widetilde{\Kap_i}\leq \CurrentIndex_i^T<\CurrentIndex_j^T = \widetilde{\Kap_j}$. This means that no one from $\WomanSet\cup\{0\}$ has largest $\widetilde{\Kap_i}$, which means any optimal policy cannot choose from $\WomanSet$ in this realization. This is exactly the result we wanted to show.

\item if $\widetilde{v_j} > \widetilde{\reserve_j}$: \\
First of all notice that the implication of this case is $c_j>0$. Call $T^{'}$ the time that alternative j is first opened. This means that at the beginning of iteration $T{'}$, $\CurrentIndex_j=\widetilde{\reserve_j}$, j was among the argmax, and no one from group $\WomanSet\cup\{0\}$ was present in the argmax. Accordingly, similar to the previous case, we will have the following inequality chain:
$\forall i\in \WomanSet\cup\{0\}: \widetilde{\Kap_i}\leq \CurrentIndex_i^{T^{'}}<\CurrentIndex_j^{T^{'}} = \widetilde{\reserve_j} = \widetilde{\Kap_j}$,
which again means that there is no way for any adjusted optimal policy to select from $\WomanSet\cup\{0\}$ in this realization.

\end{enumerate}

The remaining situation to analyze is when $j=0$. Nonetheless, the analysis of this situation is almost the same as case (i) above, just replace $\WomanSet\cup\{0\}$ with $\WomanSet$ in it. This will conclude the proof of this remark.\qed
\end{proof}

\fi 

\begin{proof}{\emph{Proof of \Cref{remark:extreme-cons}.}}
Let $\lambda^*$ be a minimizer of $\DualLagrangeConst$ and $\pi^*$ be any optimal policy for the adjusted instance with $\lambda^*$ (where the adjustment is based on \eqref{eq:adjusted-value}).  As stated in the proof of \Cref{lem:const:slack}, any optimal policy $\pi^{\lambda^*+\varepsilon}$ for an adjusted instance with adjustment $\lambda^*+\varepsilon$ for infinitesimal $\varepsilon>0$ is an optimal policy for the adjusted instance with adjustment $\lambda^*$. Note that $\DualLagrangeConst$ is differentiable at $\lambda^*+\varepsilon$ for infinitesimal $\varepsilon$. Then, applying the envelope theorem on $\DualLagrangeConst$ similar to part~(iii) of \Cref{lemma:Const:DualLagrangeProperties}, we conclude that the constraint slack of $\pi^{\lambda^*+\varepsilon}$ is equal to the slope of the convex piecewise linear function $\DualLagrangeSelect$ at point $\lambda^*+\varepsilon$. At the same time, as we showed in \Cref{lem:const:slack}, this quantity is equal to $\slack_{\textsc{cons}}^{\pi^+}$, that is, the constraint slack of the policy $\pi^+$ with positive extreme positive tie-breakining rule $\tau^+$ defined in \Cref{def:tie:extreme}. Now, note that due to the convexity of $\DualLagrangeConst$, this slope is not lower than any subderivative / subgradient of $\DualLagrangeConst$ at $\lambda^*$. At the same time, the constraint slack $\slack_{\textsc{cons}}$of the optimal policy $\pi^*$ for the adjusted instance with $\lambda^*$ is equal to one of the subderivatives of $\DualLagrangeConst$ at $\lambda^*$. Therefore:
$$
\slack_{\textsc{cons}}^{\pi^{+}}\geq \slack_{\textsc{cons}}^{\pi^*}~,
$$
as desired. The counterpart argument $\slack_{\textsc{cons}}^{\pi^{-}}\leq \slack_{\textsc{cons}}^{\pi^*}$ can be proved in a similar fashion
\qed
\end{proof}

\begin{proof}{\emph{Proof of \Cref{thm:const}}.}
First, note that the resulting randomized policy from \Cref{alg:const}, which we denote by $\hat{\policy}$, constitutes an optimal solution of $\DualLagrangeConst (\LagVector^*)$, as it randomizes over two such optimal solutions $\policyplus$ and $\policyminus$. The only part left to prove is that our randomized policy obtains an ex-ante constraint slack $\slack_{\textsc{cons}}^{\hat\policy}$ of exactly equal to zero. By construction,
$$
\slack_{\textsc{cons}}^{\hat{\policy}} = b-\expect{\sum_{i\in[n]}\theta^{S}_i{\select_i^{\hat\policy}}+ \sum_{i\in[n]}\theta^{I}_i{\inspect_i^{\hat\policy}}}  =\frac{\slack_{\textsc{cons}}^{\policyplus}}{\slack_{\textsc{cons}}^{\policyplus}-\slack_{\textsc{cons}}^{\policyminus}}\cdot\slack_{\textsc{cons}}^{\policyminus}-\frac{\slack_{\textsc{cons}}^{\policyminus}}{\slack_{\textsc{cons}}^{\policyplus}-\slack_{\textsc{cons}}^{\policyminus}}\cdot \slack_{\textsc{cons}}^{\policyplus}=0~,
$$
as desired, hence finishing the proof.
\qed
\end{proof}




\revcolor{
\section{More Intuitions and Managerial Insights from \Cref{sec:pandora}}
\label{app:pandora:managerial}

In this section, we explore the implications of our results from \Cref{sec:pandora}, providing several managerial interpretations. We also present illustrative examples that demonstrate the necessity of our specific dual adjustments and randomized tie-breaking rules to achieve the optimal policy. These examples highlight that, although alternative policies or tie-breaking methods may be optimal in certain cases, they generally lead to suboptimal or in-feasible solutions.

\subsection{Implications for Demographic Group Fairness in Selection}
\label{sec:insights-dem-parity}
Focusing on the special case of \ref{eq:parity} in selection, we have the following managerial observations:

\begin{enumerate}[label=(\roman*),leftmargin=0.22in]
\item \textbf{More advantage to the under-represented group:} Ignoring tie-breaking, the adjustment in constructing the instance \(\{(\widetilde{\values}_i, \widetilde{F}_i, c_i) \mid i \in [n]\}\) based on Equation~\eqref{eq:adjusted-value} is both intuitive and economically interpretable. To illustrate, consider the optimal solution of the unconstrained problem. If the ex-ante number of selections from both groups is equal, the policy also satisfies parity in selection. Otherwise, suppose that group \(\WomanSet\) has a lower ex-ante number of selections, making it the \emph{under-represented} group. In this case, \(\LagVector^* > 0\).\footnote{For any \(\LagVector < 0\), we have \(\DualLagrangeConst(\LagVector) > \DualLagrangeConst(0)\) in this special case, thus \(\LagVector\) cannot be a minimizer.} This implies that our adjustment (i) uniformly increases the rewards of those in \(\WomanSet\) by \(\LagVector^*\), (ii) uniformly decreases the rewards of those in \(\ManSet\) by \(\LagVector^*\), and (iii) does not adjust any costs.


\item \textbf{Preserving within-Group order and interleaving between groups:} In the adjusted instance, the indices are uniformly shifted: \(\widetilde{\reserve_i} = \reserve_i + \LagVector^*\) for all \(i \in \WomanSet\) and \(\widetilde{\reserve_i} = \reserve_i - \LagVector^*\) for all \(i \in \ManSet\). Ignoring tie-breaking, a key structural property of the optimal policy is that the within-group order of candidates is preserved after this adjustment. The only change in the search process pertains to the interleaved inspection order between the two groups.\footnote{As seen from our general adjustments in \Cref{eq:adjusted-instance} and the tie-breaking rules in \Cref{def:tie:extreme}, the optimal policy for \eqref{eq:opt-constrained} employs a non-trivial adaptive ordering over the boxes in the general case, due to the non-linear relationships between costs and indices defined in \Cref{eq:base:reserve}. For instance, in the case of \ref{eq:parity} in inspection, the optimal policy does not necessarily preserve the within-group order, unlike the optimal policy for \ref{eq:parity} in selection.} Interleaving the relative inspection orderings of the two groups is a delicate aspect of our optimal policy. For example, consider a naive policy that achieves parity in selection by randomizing with probability \(1/2\) between two search processes, each exclusively searching within one group and utilizing all available capacity. This policy does not interleave the inspection orderings of the groups and consequently suffers from an optimality gap, as illustrated in the example below.

\begin{example}
\label{example:FS}
Consider instance $\mathcal{I}$ for selecting one out of four candidates. Candidates 1 and 2 belong to $\ManSet$, and 3 and 4 belong to $\WomanSet$. All inspection costs are normalized to be $1$, and,
\begin{equation*}
    v_1=\begin{cases}
        10&\textrm{w.p.}~1/2\\
        4&\textrm{w.p.}~1/2
    \end{cases},~~~
    v_2=\begin{cases}
        9&\textrm{w.p.}~1/2\\
        3&\textrm{w.p.}~1/2
    \end{cases},~~~
    v_3=\begin{cases}
        8&\textrm{w.p.}~1/2\\
        2&\textrm{w.p.}~1/2
    \end{cases},~~~
    v_4=\begin{cases}
        7&\textrm{w.p.}~1/2\\
        1&\textrm{w.p.}~1/2
    \end{cases}~.
\end{equation*}
Note that in such an example, $\sigma_1=8$, $\sigma_2=7$, $\sigma_3=6$ and $\sigma_4=5$. The optimal unfair policy $\policy_{\textsc{unfair}}$ inspects the candidates in the order $1\rightarrow 2 \rightarrow 3 \rightarrow 4$, and $\util\left(\policy_{\textsc{unfair}};\mathcal{I}\right)\approx 7.06$. However it is unfair: it selects from group $\WomanSet$ with probability $\approx 0.1875$. The naive fair policy $\policy_{\textsc{naive}}$ (defined above) flips a coin to decide which group to consider and then inspects $\ManSet$ in the order $1\rightarrow 2$ and $\WomanSet$ in the order $3\rightarrow 4$. This policy selects exactly with probability $0.5$ from each group, but only generates $\util\left(\policy_{\textsc{naive}};\mathcal{I}\right)\approx 5.75$. Finally, our optimal fair policy $\policy_\textsc{fair}$, based on \Cref{alg:const}, inspects the candidates in the order of $1\rightarrow 3\rightarrow 2 \rightarrow 4$ with probability $0.5$, and in the order of $3\rightarrow 1\rightarrow 4 \rightarrow 2$ otherwise. It not only selects from each group with probability exactly $0.5$, but also generates $\util\left(\policy_{\textsc{fair}};\mathcal{I}\right)\approx 6.5625$.
\end{example}
\end{enumerate}

\subsection{The Necessity of Going Beyond Group-level Tie-breaking Rules}
\label{sec:beyond-group}
As discussed earlier in \Cref{sec:pandora}, for the special case of \ref{eq:parity} in selection, restricting to simple and intuitive group-level tie-breaking rules was sufficient to obtain two rules with opposite slack signs. This condition is both necessary and sufficient for constructing a randomized tie-breaking rule that exactly satisfies the ex-ante constraint. However, this simplifying property does not hold for all constraints, including \ref{eq:parity} in inspection. In this section, we provide a simple counterexample to illustrate this limitation. This example underscores the necessity of moving beyond group-level tie-breaking rules by specifying precise within-group order, to be able to satisfy the ex-ante constraint.


\begin{example}
\label{example:tie}
Consider instance $\mathcal{I}$ for selecting one out of four candidates. Candidates 1 and 2 belong to $\ManSet$, and 3 and 4 belong to $\WomanSet$. All inspection costs are normalized to be $1$, and,
\begin{align*}
    v_1 =
    \begin{cases}
        7.5 & \textrm{w.p.}~2/3\\
        4 & \textrm{w.p.}~1/3
    \end{cases},~~~
    v_2 = 
    \begin{cases}
        9 & \textrm{w.p.}~1/3\\
        4 & \textrm{w.p.}~2/3,
    \end{cases},~~~
    v_3 =
    \begin{cases}
        10 & \textrm{w.p.}~1/4\\
        4 & \textrm{w.p.} ~3/4,
    \end{cases},~~~
    v_4 =
    \begin{cases}
        7 & \textrm{w.p.}~1/2\\
        4 & \textrm{w.p.}~1/2,
    \end{cases}
\end{align*}
\end{example}

Consider the problem of finding the optimal constrained policy subject to \ref{eq:parity} in inspection for this instance. It is easy to verity that $\lambda^*=0$, indicating that there exists an optimal policy for the constrained problem that is also optimal for the unconstrained problem. By simple calculations, we find that  $\reserve_1=\reserve_2=\reserve_3=6$ and $\reserve_4=5$. Moreover, no value realization of any candidate can equal these reservation values. Therefore, ties can only occur in the inspection order of boxes $1$,$2$, and $3$.

Suppose the optimal policy fixes the within-group inspection order in $\ManSet$ such that candidate $2$ is inspected before candidate $1$. Recall the definition of the constraint slack:
\begin{align}
\label{eq:slack_inspection_tiebreak}
\slack_{\textsc{cons}}^{i_1,i_2,i_3} \triangleq \expect{\sum_{i\in \WomanSet}{\inspect_i^{\policy^{i_1,i_2,i_3}}} - \sum_{i\in \ManSet}{\inspect_i^{\policy^{i_1,i_2,i_3}}}},
\end{align}
where $i_1,i_2,i_3$ is a permutation of candidates $\{1,2,3\}$ and $\policy^{i_1,i_2,i_3}$ is the optimal policy that breaks the ties in the order $i_1 \succ i_2 \succ i_3$. Simple calculations show that:
$$
\slack_{\textsc{cons}}^{3,2,1},~\slack_{\textsc{cons}}^{2,3,1},~\slack_{\textsc{cons}}^{2,1,3} < 0~,
$$
indicating that no optimal policy (deterministic or randomized) with this fixed within-group order in $\ManSet$ can satisfy the constraint exactly. However, by occasionally changing the within-group order in $\ManSet$ to have candidate $1$ inspected before candidate $2$, we observe that:
$$-\frac{1}{12}=\slack_{\textsc{cons}}^{3,2,1} < 0 <\slack_{\textsc{cons}}^{3,1,2}=\frac{1}{6}~,$$
which implies that randomizing between the two orders \(3 \succ 2 \succ 1\) and \(3 \succ 1 \succ 2\)---which have different within-group orders for boxes in $\ManSet$---allows us to satisfy the constraint exactly, as expected.

}


\section{Technical Details of \texorpdfstring{\Cref{sec:quantile-constraint}}{}: Exact Optimal Policy for Pandora's Box with Value-Specific Ex-ante Affine Constraint}
\label{app:sec:pandora-extension}
In this section, we provide all the technical details needed to extend our framework in \Cref{sec:single-affine-policy} to incorporate value-specific constraints, as defined in Constraint~\ref{eq:general-constraint}. We start by providing some applications of this category of constraints. We then elaborate on how to generalize dual-based adjustments and extreme tie-breaking rules to this setting. We finish by providing the main result of this section, which is a characterization of the optimal constrained policy.

\subsection{Various Applications of value-specific constraints} Consider a threshold-based refinement of \eqref{eq:parity}, for selection or inspection, in which we set: 
$$
\theta_i^S=\begin{cases}
    -\mathbb{I}\{v_i\geq \underbar{v}_{\WomanSet}\} & i\in\WomanSet\\
    +\mathbb{I}\{v_i\geq \underbar{v}_\ManSet\} & i\in\ManSet
\end{cases}~, \theta_i^I=0~~\left(\textrm{or}~~
\theta_i^I=\begin{cases}
    -\mathbb{I}\{v_i\geq \underbar{v}_\WomanSet\} & i\in\WomanSet\\
    +\mathbb{I}\{v_i\geq \underbar{v}_\ManSet\} & i\in\ManSet
\end{cases}~, \theta_i^S=0
\right)~,
$$
where $\underbar{v}_{\WomanSet}\in\mathbb{R}_{\geq}$ (resp. $\underbar{v}_{\ManSet}\in\mathbb{R}_{\geq}$) is a threshold defining 
``acceptable'' values for group $\WomanSet$ (resp. $\ManSet$). 
Typically, we would like to set the thresholds $\underbar{v}_{\WomanSet},\underbar{v}_{\ManSet}$ high enough to exclude low-quality candidates and avoid issues such as token interviews as mentioned earlier. Alternatively, we can also consider a threshold-specific refinement of \eqref{eq:quota}, again for both selection and inspection, in which we set:
$$
\theta_i^S=\begin{cases}
    (\theta-1)\cdot\mathbb{I}\{v_i\geq \underbar{v}_{\WomanSet}\} & i\in\WomanSet\\
    \theta \cdot \mathbb{I}\{v_i\geq \underbar{v}_\ManSet\} & i\in\ManSet
\end{cases}~, \theta_i^I=0~~\left(\textrm{or}~~\theta_i^I=\begin{cases}
    (\theta-1)\cdot\mathbb{I}\{v_i\geq \underbar{v}_{\WomanSet}\} & i\in\WomanSet\\
    \theta \cdot \mathbb{I}\{v_i\geq \underbar{v}_{\ManSet}\} & i\in\ManSet
\end{cases}~, \theta_i^S=0\right)~.
$$
This focus on higher values achieves multiple objectives. First, it signals that opportunities (e.g., being interviewed or hired in the context of search and hiring) are accessible regardless of the demographic group, as long as the individual is considered as a top performer. Second, it promotes outcomes that are \emph{truly} fair by eliminating the need for token interviews, as elaborated earlier. 

Another significant application of this refined approach to fairness arises in scenarios involving high-cost minority candidates. For example, candidates residing in geographically challenging or inaccessible locations may incur higher inspection costs for the decision-maker. By incorporating a fairness constraint tailored to these high-cost individuals, a guaranteed level of opportunity---be it in the form of interviews or job offers---can be ensured for this group. For example, we can formulate a refinement of \eqref{eq:quota} for selection or inspection, in which we set:
$$
\theta_i^S=\begin{cases}
    (\theta-1)\cdot\mathbb{I}\{c_i\geq \underbar{c}\} & i\in\WomanSet\\
    \theta  & i\in\ManSet
\end{cases}~~,\theta_i^I=0 ~~\left(\textrm{or}~~\theta_i^I=\begin{cases}
    (\theta-1)\cdot\mathbb{I}\{c_i\geq \underbar{c}\} & i\in\WomanSet\\
    \theta  & i\in\ManSet
\end{cases}~~,\theta_i^S=0\right)~,
$$
where $\underbar{c}$ is the defining lower-limit of the cost for high-cost minority group. We note that this refinement mitigates the risk that these candidates are categorically overlooked due to cost considerations, thus adding another layer of nuance to fairness in hiring and search processes. 

\subsection{Dual-based Adjustments \& Extreme Tie-breaking Rules for Value-specific Constraints}
To handle the refined Constraint~\eqref{eq:general-constraint}, we first observe that for any adaptive feasible policy $\pi$, the indicator random variable $\inspect^\policy_i$ for inspecting box $i$ is independent from the value $v_i$ of the box. Therefore, by following exactly the same recipe as in \Cref{sec:single-affine-policy} (i.e., Lagrangifying the constraint and re-arranging the terms in the Lagrangian function $\LagrangeConst$) and applying the law of iterated expectations, 
we get the following equivalent form for the Lagrangian function:
\begin{equation*}
\LagrangeConst (\pi;\LagVector) 
 = \expect{\sum_{i \in[\altnum]} \select_i^{\policy}(v_i-\LagVector\cdot\theta^{S}_i(v_i,c_i))-\sum_{i\in[n]}\inspect_i^{\policy}   \left(c_i+\LagVector\cdot\mathbf{E}_{v_i\sim F_i}\left[\theta^{I}_i(v_i,c_i)\right]\right)} +\LagVector\cdot b~,
\label{eq:largrange-general}  
\end{equation*}
which in turn suggests the following refined dual-adjustment of the values and the costs given  $\lambda$ (cf. the earlier dual adjustment in \eqref{eq:adjusted-instance}):
\begin{equation}
\label{eq:general-adjusted-instance}
\widetilde{v_i}\triangleq v_i-\LagVector\cdot\theta^{S}_i(v_i,c_i)~~~,~~~\widetilde{c_i}\triangleq c_i+\LagVector\cdot\mathbf{E}_{v_i\sim F_i}\left[\theta^{I}_i(v_i,c_i)\right]~.
\end{equation}
As before, the Lagrange dual function $\DualLagrangeConst$ can be defined as the minimizer of the Lagrangian function over all feasible policies. Moreover, by solving an adjusted instance based on the adjustment in \eqref{eq:general-adjusted-instance}, we obtain query access to $\DualLagrangeConst$ (through the optimal objective value of the adjusted instance) and $\frac{\partial\DualLagrangeConst}{\partial \lambda}$ (through the corresponding constraint slack $\slack_{\textsc{cons}}$ of the optimal policy after adjustments). Finally, given the minimizer $\lambda^*$ of $\DualLagrangeConst$, the two optimal policies $\policy^+$ and $\policy^{-}$ corresponding to the perturbed adjusted instances with respect to $\lambda^*+\epsilon$ and $\lambda^*-\epsilon$, respectively, (i) will still be optimal for an instance with adjustment corresponding to $\lambda^*$, and (ii) will define the two extreme tie-breaking rules $\tau^+$ and $\tau^-$ that guarantee positive and negative slacks, respectively (similar to \Cref{lem:const:slack}).







Before explicitly characterizing these two extreme tie-breaking rules $\tau^+$ and $\tau^-$, let us first provide the required technical notation and setup. Consider the optimal adjusted instance, as defined in \eqref{eq:general-adjusted-instance} where $\LagVector=\LagVector^*$. Recall the definition of the maximum adjusted option value $\widetilde{o_{\textrm{max}}}\triangleq \underset{i\in([\altnum]\setminus\mathcal{S})\cup\{0\}}{\max}~\widetilde{o_i}$, defined in any round in the execution of  \Cref{alg:Pandora} on the adjusted instance. Consider any candidate $i$ and let $\{V_{i,1}, V_{i,2}, ... ,V_{i,L} \}$, for some $L\in \mathbb{N}\cup\{0\}$, be all the values in $\values_i$ (i.e., the support of $F_i$)  that after adjustment have all became equal to $\widetilde{o_{\textrm{max}}}$, i.e., $\widetilde{V_{i,\ell}}\triangleq V_{i,\ell}-\LagVector\cdot\theta^{S}_i(V_{i,\ell},c_i)=\widetilde{o_{\textrm{max}}}, \forall \ell \in [L]$.   Without loss, suppose that these values are sorted in decreasing order according to their $\theta^{S}_i(V_{i,\ell},c_i)$, that is, $\theta^{S}_i(V_{i,\ell},c_i)\geq \theta^{S}_i(V_{i,\ell'},c_i)$ if $\ell\leq \ell'$. With this in mind, consider two nested sequences $\mathcal{E}^-_0\subseteq \mathcal{E}^-_1\subseteq ... \subseteq\mathcal{E}^-_L$ and $\mathcal{E}^+_0\subseteq \mathcal{E}^+_1\subseteq ... \subseteq\mathcal{E}^+_L$  of subsets of support of $F_i$ defined below:
\begin{align}
\label{eq:events_positive}
    \mathcal{E}^-_\ell & \triangleq \left\{V\in\values_i: V-\LagVector\cdot\theta^{S}_i(v,c_i)>\widetilde{o_{\textrm{max}}} \right\} \cup \left \{V_{i,j}\right\}_{1 \leq j \leq \ell} , \quad \quad & \forall \ell: 0 \leq \ell \leq L \\
    \label{eq:events_positive-2}
    \mathcal{E}^+_\ell & \triangleq \left\{V\in\values_i: V-\LagVector\cdot\theta^{S}_i(v,c_i)>\widetilde{o_{\textrm{max}}} \right\} \cup \left \{V_{i,j} \right \}_{L+1-\ell \leq j \leq L} , \quad \quad & \forall \ell: 0 \leq \ell \leq L 
\end{align}
With these two sequences of subsets defined, we provide the exact characteristics of our two extreme tie-breaking rules in the following definition. To simplify the notation, we also slightly abuse the notation and just use $\theta^{S}_i$ (resp. $\theta^{I}_i$) rather than $\theta^{S}_i(v_i,c_i)$ (resp. $\theta^{I}_i(v_i,c_i)$), while keeping in mind that these numbers are random variables for boxes that are yet to be opened.

\smallskip
\begin{definition}[\textbf{Refined Extreme Tie-Breaking Rule}]
\label{def:refined-ext-tie}
Given any set of candidates $\candidates$ for breaking ties at any point during the execution of \Cref{alg:Pandora} (Line 7), the \emph{negative-extreme rule}, 
denoted by $\tierule^{-}$, assigns a \emph{tie-breaking score} $s^-_i\in\mathbb{R}$ to each $i\in\candidates$ as follows (here, $\mathcal{E}^-_0\subseteq \mathcal{E}^-_1\subseteq ... \subseteq\mathcal{E}^-_L$ and $\mathcal{E}^+_0\subseteq \mathcal{E}^+_1\subseteq ... \subseteq\mathcal{E}^+_L$ are two nested sequence of subsets of $\values_i$ at this point in the execution of the algorithm on adjusted instance, as defined in \cref{eq:events_positive} and \cref{eq:events_positive-2}):
\begin{itemize}
    \item For $i\displaystyle \in\mathcal{C}\setminus\mathcal{O}$:
    \begin{itemize}[leftmargin=*]
        \item If $\displaystyle c_i<0$, set $\displaystyle s^-_i\leftarrow +\infty$.
        \item If $\displaystyle c_i\geq 0$, set $\displaystyle s^-_i\leftarrow \max_{0\leq \ell\leq L} \left\{ \mathbf{E}_{v_i\sim F_i}\left[\theta^{S}_i | \mathcal{E}^-_\ell \right] + \frac{\mathbf{E}_{v_i\sim F_i}\left[\theta^{I}_i\right]}{\prob{\mathcal{E}^-_\ell}} \right\}$ 
        
       \noindent
       \begingroup\renewcommand*{\arraystretch}{1.5}
       $\left(\begin{matrix*}[l]
            s^-_i=+\infty \quad \textrm{if}~\prob{\mathcal{E}^-_0}=0 ~\textrm{and}~ \mathbf{E}_{v_i\sim F_i}\left[\theta^{I}_i\right]\geq 0\\ 
             s^-_i=-\infty \quad \textrm{if}~\prob{\mathcal{E}^-_L}=0 ~\textrm{and}~\mathbf{E}_{v_i\sim F_i}\left[\theta^{I}_i\right]< 0.
             \end{matrix*}\right)$
        \endgroup

    \end{itemize}
       
   
    \item For $\displaystyle i\in \candidates\cap\mathcal{O}$:
    \begin{itemize}[leftmargin=*]
        \item If $i\neq 0$, set $\displaystyle s^-_i\leftarrow \theta^{S}_i$, and if $i=0$ (that is, outside option), set $s^-_i\leftarrow 0$.
    \end{itemize}
\end{itemize} 
Similarly, the counterpart rule, calling it {\em positive-extreme rule} and denote it by $\tierule^{+}$, assigns a \emph{tie-breaking score} $s^+_i\in\mathbb{R}$ to each $i\in\candidates$ as follows: 
\begin{itemize}[leftmargin=*]
    \item For $i\displaystyle \in\mathcal{C}\setminus\mathcal{O}$:
    \begin{itemize}[leftmargin=*]
        \item If $\displaystyle c_i<0$, set $\displaystyle s^+_i\leftarrow +\infty$.
        \item If $\displaystyle c_i\geq 0$, set $\displaystyle s^+_i\leftarrow \max_{0\leq \ell\leq L} \left\{ -\mathbf{E}_{v_i\sim F_i}\left[\theta^{S}_i | \mathcal{E}^+_\ell \right] - \frac{\mathbf{E}_{v_i\sim F_i}\left[\theta^{I}_i\right]}{\prob{\mathcal{E}^+_\ell}} \right\}$ 
      
 \noindent
       \begingroup\renewcommand*{\arraystretch}{1.5}
       $\left(\begin{matrix*}[l]
            s^+_i=+\infty \quad \textrm{if}~\prob{\mathcal{E}^+_0}=0 ~\textrm{and}~ \mathbf{E}_{v_i\sim F_i}\left[\theta^{I}_i\right]\leq 0\\ 
             s^+_i=-\infty \quad \textrm{if}~\prob{\mathcal{E}^+_L}=0 ~\textrm{and}~\mathbf{E}_{v_i\sim F_i}\left[\theta^{I}_i\right]> 0.
             \end{matrix*}\right)$
        \endgroup
        
    \end{itemize}
       \item For $\displaystyle i\in \candidates\cap\mathcal{O}$:
       \begin{itemize}[leftmargin=*]
           \item If $i\neq 0$, set $\displaystyle s^+_i\leftarrow -\theta^{S}_i$, and if $i=0$ (that is, outside option), set $s^+_i\leftarrow 0$.
       \end{itemize}
    \end{itemize} 
Then, the rule $\tierule^{-}$ (resp. $\tierule^{+}$) breaks the ties in favor of scores $\{s^-_i\}_{i\in\candidates}$ (resp. $\{s^+_i\}_{i\in\candidates}$), that is, it returns any $\displaystyle i^*\in \underset{i\in\candidates}{\argmax}~s^-_i$ (resp. any $\displaystyle i^*\in \underset{i\in\candidates}{\argmax}~s^+_i$). 

\end{definition}
\medskip

Given the above definitions of (i) dual-adjusted problem instance in \Cref{eq:general-adjusted-instance} and (ii) extreme tie-breaking rules in \Cref{def:refined-ext-tie}, we are ready to state and prove our main result for this section, which is \Cref{prop:general-constraint}.

\begin{theorem}[Optimal Policy for Value-specific Constrained Problem]
\label{prop:general-constraint} 
    Consider a modified version of policy RDIP (described in \Cref{alg:const}) in which:
    \begin{itemize}[leftmargin=*]
        \item in line~(2), the adjusted instances $\{({\widetilde{\values_i}}, \widetilde{F_i}, \widetilde{c_i})| i \in [n]\}$ is defined based on \eqref{eq:general-adjusted-instance}, i.e., $\widetilde{v_i}\triangleq v_i-\LagVector^*\cdot\theta^{S}_i(v_i,c_i)$ and $\widetilde{c_i}\triangleq c_i+\LagVector^*\cdot\mathbf{E}_{v_i\sim F_i}\left[\theta^{I}_i(v_i,c_i)\right]$, 
        \item in line~(4), the extreme tie-breaking rules $\{\tau^{+},\tau^{-}\}$ are defined based on the scoring rules introduced in \Cref{def:refined-ext-tie} in \Cref{app:sec:pandora-extension}.
    \end{itemize}
    Then this modified policy is optimal for the constrained Pandora's box problem with multiple selection, defined in \eqref{eq:opt-constrained}, under a value-specific ex-ante affine constraint as in \Cref{eq:general-constraint}.
\end{theorem}

\begin{proof}{\emph{Proof of \Cref{prop:general-constraint}.}}
In general the proof of \Cref{prop:general-constraint} is very similar to that of \Cref{lem:const:slack} and \Cref{thm:const}, as such we only provide the parts that have a non-trivial analog. In particular, the first step of \Cref{lem:const:slack} and the proof of \Cref{thm:const} can also be used here. Therefore, the only part in which we need to provide details is the second step in the proof of \Cref{lem:const:slack}. We show it here only for the negative-extreme rule, but the proof of the positive-extreme rule would be exactly the same (since the only difference is that instead of $-\varepsilon$ we have $\varepsilon$, which will just change the signs of all perturbations). Furthermore, all opened boxes, as well as all degenerate unopened boxes for which the score will be set to $+\infty$ or $-\infty$ will also be treated in the same way. As a result, the remaining part is to show that the score $\max_{0\leq \ell\leq L} \left\{ \mathbf{E}_{v_i\sim F_i}\left[\theta^{S}_i | \mathcal{E}^-_\ell \right] + \frac{\mathbf{E}_{v_i\sim F_i}\left[\theta^{I}_i\right]}{\prob{\mathcal{E}^-_\ell}} \right\}$ is in fact the correct amount for an unopened box whose perturbed adjusted cost $\widetilde{c_i}^-$, by perturbing $\LagVector$ with $-\varepsilon$, is positive and also is among $\underset{i\in([\altnum]\setminus\mathcal{S})\cup\{0\}}{\argmax} \{ \widetilde{o_i}\}$, which means $\widetilde{\reserve_i}=\widetilde{o_i}=\widetilde{o_{\textrm{max}}}$.

First, recall that the total number of possible deterministic policies is finite, indicating that there exists an $\Bar{\varepsilon}>0$, such that the optimal policy for the perturbed instance $\LagVector^*-\varepsilon$ remains the same for all $0 \leq \varepsilon \leq \Bar{\varepsilon}$. This shows that the ordering among all perturbed values and reservation values ($\widetilde{\reserve_i}$) will remain exactly the same during the entire perturbation interval $\varepsilon \in (0, \Bar{\varepsilon})$.

Knowing that this ordering will remain unchanged over a sufficiently small interval, it is easy to verify that (i) the change in both $\widetilde{v_i}$ and $\widetilde{c_i}$ is linear in $\varepsilon$, and (ii) as a result, the change in $\widetilde{\reserve_i}$ (= $\widetilde{o_i}$) is linear in $\varepsilon$. Let $\widetilde{\reserve_i}^-(\varepsilon)$ be the adjusted reservation value of box $i$ after perturbation $-\varepsilon$ (hence $\widetilde{\reserve_i}^-(0)=\widetilde{\reserve_i}$ and $\widetilde{\reserve_i}^-(\varepsilon)$ is linear in $\varepsilon$ for $\varepsilon\in(0,\bar{\varepsilon}$).
Given these linear functions $\widetilde{\reserve_i}^-$ for different boxes, among all boxes in $\candidates$ (those with a tie), the box $i$ with the highest slope would be the one with the highest $\widetilde{o_i}$ after the perturbation, as all these boxes have the same adjusted reservation value $\widetilde{\reserve_i}=\widetilde{o_i}=\widetilde{o_{\textrm{max}}}$ before the perturbation.

Thus, the only remaining part of the proof is to find an explicit formula for the slope of $\widetilde{\reserve_i}^-$ for such boxes in $\candidates$, for which we have $\widetilde{\reserve_i} = \widetilde{o_i}=\widetilde{o_{\textrm{max}}}$. Let $\widetilde{\reserve_i}^-=\widetilde{\reserve_i} + \varepsilon \times \delta$, where $\delta$ is the slope of the linear function $\widetilde{\reserve_i}^-$.  Denoting the adjusted cost and adjusted values of box $i$ after perturbation $-\varepsilon$ (according to \Cref{eq:general-adjusted-instance}) by $\widetilde{c_i}^-$ and $\widetilde{v_i}^-$, respectively, the following equation should hold:
\begin{equation} 
\label{eq:delta-1}
    \widetilde{c_i} - \varepsilon  \times \mathbf{E}_{v_i\sim F_i}\left[\theta^{I}_i(v_i,c_i)\right] \overset{(1)}{=} \widetilde{c_i}^- \overset{(2)}{=}\mathbf{E}_{v_i\sim F_i}\left[(\widetilde{v_i}^- - \widetilde{\reserve_i}^-)^+ \right]= \displaystyle\sum_{V_i\in \textrm{supp}(F_i):\widetilde{V_i}^{-}\geq \widetilde{\reserve_i}^-} (\widetilde{V_i}^- - \widetilde{\reserve_i}^-)\prob{V_i}~,
\end{equation} 
where $\widetilde{V_i}^-\equiv \widetilde{V_i}+\varepsilon\times \theta^{S}_i(V_{i},c_i)$, Equation~(1) holds due to the definition of adjusted cost after perturbation in \Cref{eq:general-adjusted-instance}, and Equation~(2) holds due to the definition of the reservation value. Note that for a value $V_i$ in the support of $F_i$, if $\widetilde{V_i}>\widetilde{\reserve_i}$ then $\widetilde{V_i}^{-}\geq\widetilde{\reserve_i}^-$ for sufficiently small $\varepsilon$. First, suppose that there is no value $V_{i}$ in the support of box $i$ such that $\widetilde{V_{i}}\equiv V_{i}-\LagVector\cdot\theta^{S}_i(V_{i},c_i)=\widetilde{\reserve_i}$. In this case, taking the derivative with respect to $\varepsilon$ of both sides of \Cref{eq:delta-1} and rearranging the terms, it is easy to show that $\delta= \mathbf{E}_{v_i\sim F_i}\left[\theta^{S}_i | \mathcal{E}^-_0 \right] + \frac{\mathbf{E}_{v_i\sim F_i}\left[\theta^{I}_i\right]}{\prob{\mathcal{E}^-_0}}$ (similar to the way we calculated $\delta$ in the proof of step~2 in \Cref{lem:const:slack}). 

Now, assume that there are $L\geq 1$ values $\{V_{i,1}, V_{i,2}, ... ,V_{i,L} \}$ in the support of $F_i$ that satisfy $\widetilde{V_{i,\ell}}\equiv V_{i,\ell}-\LagVector\cdot\theta^{S}_i(V_{i,\ell},c_i)=\widetilde{\reserve_i}$.
To find a similar characterization for $\delta$ using \Cref{eq:delta-1}, we have to find values in $\{V_{i,1}, V_{i,2}, ... ,V_{i,L} \}$ for which $\widetilde{V_{i,\ell}}^{-}\equiv\widetilde{V_{i,\ell}}+\varepsilon\times \theta^{S}_i(V_{i,\ell},c_i)$ is no smaller than $\widetilde{\reserve_i}^-\equiv\widetilde{\reserve_i} + \varepsilon \times \delta$. Note that $\widetilde{V_{i,\ell}}=\widetilde{\reserve_i}$ for all $\ell\in[1:L]$, and therefore $\widetilde{V_{i,\ell}}^{-}\geq\widetilde{\reserve_i}^-$ if and only if $\theta^{S}_i(V_{i,\ell},c_i)\geq\delta$. Also, recall that the values $\{V_{i,1}, V_{i,2}, ... ,V_{i,L} \}$ are sorted in the decreasing order of the slopes $\theta^{S}_i(V_{i,\ell},c_i)$. As a result, there should exist a unique $0\leq\ell\leq L$ such that $\theta^{S}_i(V_{i,\ell'},c_i)\geq \delta$ if and only if $0\leq \ell'\leq \ell$, or equivalently $\theta^{S}_i(V_{i,\ell},c_i) \geq \delta >\theta^{S}_i(V_{i,\ell+1},c_i)$. Putting everything together, the set of values $V_i$ in the support of $F_i$ whose adjustment after perturbation would be higher than  $\widetilde{\reserve_i}^-$ (adjusted reservation value $\widetilde{\reserve_i}$ after perturbation) is \emph{exactly} the subset $\mathcal{E}^-_\ell$.
We can now find the slope $\delta$ using \Cref{eq:delta-1}. More precisely, the slope $\delta$ should satisfy the following chain of equations:
\begin{align*}
    \widetilde{c_i} - \varepsilon  \times \mathbf{E}_{v_i\sim F_i}\left[\theta^{I}_i(v_i,c_i)\right] & = \widetilde{c_i}^- \\ \nonumber
    & = \mathbf{E}_{v_i\sim F_i}\left[(\widetilde{v_i}^- - \widetilde{\reserve_i}^-)^+ \right] \\ \nonumber
    & = \sum_{V_i\in \mathcal{E}^-_\ell} (\widetilde{V_i}^- - \widetilde{\reserve_i}^-)\prob{V_i} \\ \nonumber
    & = \sum_{V_i\in \mathcal{E}^-_\ell} \left( \left( \widetilde{V_i} - \widetilde{\reserve_i}\right) \times \prob{V_i} + \varepsilon\times \left(\theta^{S}_i(V_i,c_i)-\delta \right) \times \prob{V_i} \right) \\ \nonumber
    & = \mathbf{E}_{v_i\sim F_i}\left[(\widetilde{v_i} - \widetilde{\reserve_i})^+ \right] + \varepsilon \times \prob{\mathcal{E}^-_\ell} \left(\mathbf{E}_{v_i\sim F_i}\left[\theta^{S}_i(v_i,c_i) | \mathcal{E}^-_\ell \right]  - \delta \right) \\ \nonumber
    & = \widetilde{c_i} + \varepsilon \times \prob{\mathcal{E}^-_\ell} \left[\mathbf{E}_{v_i\sim F_i}\left(\theta^{S}_i(v_i,c_i) | \mathcal{E}^-_\ell \right) - \delta \right].
\end{align*} 
If we cancel $\widetilde{c_i}$ from both RHS and LHS and then divide by $\varepsilon \times \prob{\mathcal{E}^-_\ell}$, we get
$$\delta = \mathbf{E}_{v_i\sim F_i}\left[\theta^{S}_i(v_i,c_i) | \mathcal{E}^-_\ell \right] + \frac{\mathbf{E}_{v_i\sim F_i}\left[\theta^{I}_i(v_i,c_i)\right]}{\prob{\mathcal{E}^-_\ell}}.$$

For simplicity, let us define $d_\ell \triangleq \mathbf{E}_{v_i\sim F_i}\left[\theta^{S}_i(v_i,c_i) | \mathcal{E}^-_\ell \right] + \frac{\mathbf{E}_{v_i\sim F_i}\left[\theta^{I}_i(v_i,c_i)\right]}{\prob{\mathcal{E}^-_\ell}}$ for all $\ell, 0\leq \ell \leq L$. With this, the problem reduces to a verification problem, wherein we should verify that for which $\ell$ the following inequalities hold:
\begin{align}
\label{eq: rank_verification}
   \theta^{S}_i(V_{i,\ell},c_i) \geq d_\ell > \theta^{S}_i(V_{i,\ell+1},c_i). 
\end{align}
Importantly, it turns out that $\ell$ satisfies \Cref{eq: rank_verification} if and only if $d_\ell=\max_{0\leq s \leq L} d_s$. This can be easily derived from the combination of the following four properties; and thus, we skip the rest of the details for the sake of brevity.

1) By definition, the sequence $\theta^{S}_i(V_{i,\ell},c_i)$ is a (weakly) decreasing sequence w.r.t. $\ell$.

2) $d_{\ell+1}$ is a convex combination of $d_\ell$ and $\theta^{S}_i(V_{i,\ell+1},c_i)$.

3) Combining 1 and 2, we get that the sequence $d_\ell$ is a (weakly) increasing sequence up until some $\hat{\ell}$, and then it will become a (weakly) decreasing sequence. This also tells us that $d_{\hat{\ell}} = \max d_\ell$.

4) $\hat{\ell}$, and also any $\ell$ for which $d_\ell$ still remains equal to $d_{\hat{\ell}}$, are the only $\ell$'s that satisfy \eqref{eq: rank_verification}. More specifically, any smaller $\ell$ does not satisfy the second inequality, any larger $\ell$ does not satisfy the first inequality, and all $\ell\in \argmax_s d_s$ do satisfy both of the inequalities in \Cref{eq: rank_verification}.

With this we immediately conclude that the correct slope $\delta$ would be:
$$\delta = \max_\ell d_\ell =  \max_{0\leq \ell\leq L} \left\{ \mathbf{E}_{v_i\sim F_i}\left[\theta^{S}_i | \mathcal{E}^-_\ell \right] + \frac{\mathbf{E}_{v_i\sim F_i}\left[\theta^{I}_i\right]}{\prob{\mathcal{E}^-_\ell}} \right\},$$
which is exactly the amount that we set to our score $s_i^-$ in such scenarios. The rest would again be quite similar to what we did in \Cref{lem:const:slack}, and we show that this scoring rule enables us to run exactly the optimal policy $\policy^-$ corresponding to the perturbed adjusted instance by $\LagVector-\varepsilon$. Hence, we conclude the proof. 
\qed
\end{proof}

\newcommand{\epsilonVector}{\overrightarrow{\varepsilon}}
\newcommand{\polytope}{\mathcal{P}}
\newcommand{\slackvec}{\boldsymbol{\slack}_{\textsc{cons}}^{\policy}}
\newcommand{\slackconstj}{{\slack}^{\policy}_{\textsc{cons},j}}
\newcommand{\lambdavec}{\boldsymbol{\lambda}}
\newcommand{\extoracle}{\textsc{Ext-Lin-Oracle}}
\newcommand{\linoracle}{\textsc{Lin-Oracle}}
\newcommand{\ellips}{{E}}
\newcommand{\ellipsc}{\mathbf{e}}
\newcommand{\projset}{{\Omega}}
\newcommand{\direction}{\boldsymbol{\omega}}
\newcommand{\projspac}{\mathcal{W}}
\newcommand{\projorth}{\projspac^{\perp}}
\newcommand{\facep}{\mathcal{A}}
\newcommand{\tertime}{\tau}

\revcolor{

\section{Technical Details of \texorpdfstring{\Cref{sec:caratheodory}}{}: Exact Optimal Policy with Multiple Ex-ante Affine Constraints}
\label{app:mutiple-affine}
\revcolorm{In this section, we provide all the technical details for the results promised in \Cref{sec:caratheodory}. In particular, we show how to ``properly'' generalize our approach from \Cref{sec:pandora} and \Cref{sec:general} to handle multiple ex-ante affine constraints \emph{exactly}, that is, without any slack---resulting in a polynomial-time algorithm that computes an optimal policy satisfying all the ex-ante affine constraints with no additive error.}

Our generalized approach involves reducing the problem to a variant of the classical (algorithmic) \emph{exact Carathéodory problem}~\citep{caratheodory1911variabilitatsbereich}. We first explain this reduction in \Cref{subsection: reduction of equality constrained to Caratheodory}. \revcolorm{We then introduce specific oracle algorithms in \Cref{app:sec:linear-oracle} that are polynomial-time computable within both our Pandora's box setting and its generalization to joint Markovian scheduling. Next, in \Cref{app:sec:exact-caratheodory-algorithm}, we demonstrate how to solve the reduced exact Carathéodory problem in polynomial time, given access to these oracles in a blackbox manner.} Lastly, in \Cref{app:extreme-failure-multiple-affine}, we present a simple example showing that the natural extension of ``extreme tie-breaking rules'' from \Cref{sec:pandora} fails, even when applied to settings with two ex-ante affine constraints, indicating that our reduction to exact algorithmic Carathéodory is crucial for solving the problem with multiple affine ex-ante constraints.

\revcolorm{In the remainder of this section, we focus on the Pandora's box problem with multiple selections (described in \Cref{sec:pandora-setting}) and its generalization to the joint Markovian scheduling problem with  multiple selections or a matroid constraint (described in \Cref{sec:JMS-setting}). We assume we are given $m \in \mathbb{N}$ ex-ante affine constraints, analogous to Constraint~\ref{eq:affine-constraint} for the Pandora's box problem or the ex-ante affine constraints defined in \Cref{sec:general-ex-ante-constraints} for the joint Markovian scheduling problem. We first focus on the special case where all constraints are equalities.} Later, in \Cref{subsec: reduction of inequality to equality}, we demonstrate how to reduce the problem with $m$ general ex-ante affine constraints, where some are equalities and others are inequalities, to a problem with $m$ ex-ante affine equality constraints.

\subsection{Reduction to the Exact Carathéodory Problem}
\label{subsection: reduction of equality constrained to Caratheodory}

For some notation throughout this section,  given an admissible policy $\policy$, we denote the constraint slack vector by $\slackvec = \left(\slackconstj\right)_{j \in [m]}$. Here, $\slackconstj$ represents the slack of the $j^{\textrm{th}}$ ex-ante affine constraint---for example, for the case of the Pandora's box problem, it is defined in \Cref{eq:slack}. \revcolorm{For the joint Markovian scheduling problem with $m$ ex-ante affine constraints, this slack vector can be defined similarly.} Note that for now we have assumed all constraints are in the equality form. The main objective of this section is to compute a randomized admissible policy $\policy^*$ that maximizes the expected utility of the search while ensuring that $\boldsymbol{\slack}_{\textsc{cons}}^{\policy^*} = \textbf{0} \in \mathbb{R}^m $. Notably, if the randomized optimal policy $\policy^*$ is a convex combination (or equivalently, a randomization) of finitely many deterministic admissible policies $\{\policy^{(i)}\}_{i\in S}$ for a finite set $S$, then $\boldsymbol{\slack}_{\textsc{cons}}^{\policy^*}$ will be the same convex combination of constraint slack vectors $\{\boldsymbol{\slack}^{\policy^{(i)}}_{\textsc{cons}}\}_{i\in S}$, and therefore:
$$\boldsymbol{\slack}_{\textsc{cons}}^{\policy^*}=\mathbf{0}\in \textrm{Conv}\left(\{\boldsymbol{\slack}^{\policy^{(i)}}_{\textsc{cons}}\}_{i\in S}\right)~,$$
where $\textrm{Conv}(\cdot)$ denotes the convex hull of its input argument.

\revcolorm{Now, let us focus on the Pandora's box setting first. The argument for the case of the joint Markovian scheduling is exactly identical and omitted for brevity.} Following the same approach as in the case of the single affine constraint, given the vector of dual variables $\lambdavec\in \mathbb{R}^m$, we define the Lagrangian relaxation function $\LagrangeConst$ and the Lagrangian dual function $\DualLagrangeConst$ as follows:
\begin{align}
\label{eq:dual_multiple}
\DualLagrangeConst (\lambdavec) \triangleq \max_{\policy \in \PolicySpace} ~\LagrangeConst (\pi;\lambdavec)~.
\end{align}
The function $\DualLagrangeConst$ will have same properties as before (such as being a piece-wise affine convex function). Moreover, it continues to hold that an optimal index-based policy (similar to \Cref{alg:Pandora}) in the Lagrangian adjusted version of the problem would be the maximizer solution in \Cref{eq:dual_multiple}, providing us with polynomial-time access to both the value and sub-gradients of $\DualLagrangeConst (\lambdavec)$, as before. By applying standard methods in convex optimization, we can efficiently find the vector of optimal dual variables $\lambdavec^*$ minimizing the Lagrangian dual function $\DualLagrangeConst$. However, to find the optimal constrained policy we essentially need to find the ``saddle point''--- a randomized policy $\policy^*$ that maximizes the Lagrangian relaxation (in expectation) against the worst-case choice of $\lambdavec$, that is, 
$$
\policy^*\in \underset{\pi\in\Delta(\Pi)}{\argmax}~ \underset{\lambdavec}{\min}~\mathbf{E}\left[\LagrangeConst (\pi;\lambdavec)\right]~.
$$
By applying strong-duality (i.e., a weaker version of Sion's minimax theorem~\citep{sion1958general}), the resulting randomized policy $\policy^*$ would be a convex combination of (deterministic)  maximizer policies in \Cref{eq:dual_multiple} when $\lambdavec\leftarrow\lambdavec^*$, and satisfies $\boldsymbol{\slack}_{\textsc{cons}}^{\policy^*} = \textbf{0}$. However, it remains a challenge to compute this convex combination, as there may be exponentially many such maximizer policies.

Let $N$ be the number of these deterministic maximizer policies denoted by $\{\policy^{(i)}\}_{i\in[N]}$.\footnote{As mentioned earlier in \Cref{sec:pandora-setting}, there are finitely many index-based policies in the Pandora's box with multiple selections when value distributions have finite discrete support. Moreover, as we mentioned in \Cref{sec:JMS-setting}, there are also finitely many index-based policies in the JMS problem, simply because we assume the underlying MCs have finite state spaces. All of our results in this section extend to the setting with continuous distributions through proper adjustments and formalizations, which we omit for the sake of simplicity.}  Each policy $\policy^{(i)}$ is an optimal dual-adjusted index-based policy with respect to $\lambdavec^*$, corresponding to a certain deterministic tie-breaking rule $\tiebreak^{(i)}$ and associated with a particular slack vector $ \boldsymbol{\slack}_{\textsc{cons}}^{\policy^{(i)}}\in \reals^\numaffine$. The goal here is to select a handful of these policies in a computationally efficient way, so that by randomizing over them, we can achieve slack of $\origin$. In other words, we would like to find a small subset $S \subseteq [N]$,  such that:
$$\mathbf{0}\in \textrm{Conv}\left(\{\boldsymbol{\slack}^{\policy^{(i)}}_{\textsc{cons}}\}_{i\in S}\right)~.$$
Note that because we assume the problem is feasible, there should exist a saddle point solution, or equivalently, a randomized optimal constrained policy $\policy^*$. Therefore, we already know that we can obtain a slack of $\origin$  by randomizing over all of these maximizer policies, i.e., 
$$\mathbf{0}\in \textrm{Conv}\left(\{\boldsymbol{\slack}^{\policy^{(i)}}_{\textsc{cons}}\}_{i\in N}\right)~.$$

With this formulation of our problem, one can think of the $\numaffine$-dimensional polytope $\polytope=\textrm{Conv}\left({V}\right)$, where ${V}\triangleq\{\boldsymbol{\slack}^{\policy^{(i)}}_{\textsc{cons}}\}_{i\in N}\subset\reals^\numaffine$. Now our problem of finding $S$ as described above is, in fact, an instance of the \emph{exact algorithmic Carathéodory problem}: given the polytope $\polytope$ that contains $\origin$, find a ``small'' subset of points $\vertices'\subseteq\vertices$ in polynomial-time such that $\origin\in \textrm{Conv}\left(\vertices'\right)$.\footnote{An alternative way of defining the goal in the algorithmic Carathéodory problem is identifying a subset of \emph{extreme points (i.e., vertices)} of $\polytope$ that their convex hull includes the target point. Note that not all the points in ${V}$ are the vertices of $\polytope$. However, the two versions of the problem are mathematically equivalent, as long as the oracles the algorithm uses always return a vertex, which is without loss of generality by applying standard arguments (see \Cref{app:sec:exact-caratheodory-algorithm} for more details).}

We recall that the polytope $\polytope$ described above can have exponentially many vertices in the parameters of the problem. Even though the classical Carathéodory theorem~\citep{caratheodory1911variabilitatsbereich} implies that there should exist $m+1$ vertices of this $m$-dimensional polytope $\polytope$ that cover $\origin$ (their convex hull includes $\origin$), it is not even clear whether we can find a polynomial number of points in $\vertices$ that can cover $\origin$.
If we can find such a set of points (and therefore their corresponding policies and constraint slack vectors), then by using linear programming we can find the desired convex combination to satisfy the slack of $\origin$, and therefore we will have a randomized optimal policy for our problem (i.e., a randomization over policies uncovered, with the resulting convex combination obtained through solving a feasibility LP)  that satisfies all the constraints exactly.


\revcolorm{In what follows, we provide an affirmative answer by showing a polynomial-time algorithm for our specific instance of the exact Carathéodory problem. In particular, in \Cref{app:sec:linear-oracle} we show that linear optimization over polytope $\polytope$ is equivalent to finding the dual-adjusted index-based policy corresponding to a certain perturbation of $\lambdavec^*$, which can be done in polynomial-time (as we showed in \Cref{sec:single-affine-policy} for the Pandora's box problem, and in \Cref{sec:optimal-JMS-arbitrary} and \Cref{app:JMS-general} for the joint Markovian scheduling problem). Having blackbox access to this polynomial-time oracle, in \Cref{app:sec:exact-caratheodory-algorithm} we show how to solve the exact algorithmic Carathéodory problem.}

\subsection{Linear Optimization Oracle: Basic and Extended}
\label{app:sec:linear-oracle}
Consider the polytope $\polytope\subset \reals^\numaffine$ defined earlier in \Cref{subsection: reduction of equality constrained to Caratheodory}, and an arbitrary direction $\boldsymbol{\omega}\in\reals^\numaffine$. The goal of this section is to implement a ``linear optimization oracle'' over $\polytope$, denoted by $\linoracle$, in polynomial time. This simple oracle is formally defined as follows. 

\begin{definition}[Linear Optimization Oracle]
\label{def:lin-oracle}
Given the polytope $\polytope=\textrm{Conv}(\vertices)\subset \reals ^\numaffine$, the oracle $\linoracle(\cdot ; \polytope)$ is defined by the following input-output relationship:
\begin{itemize}
    \item \textbf{input:} a direction $\boldsymbol{\omega}$ in $\reals^\numaffine$.
    \item \textbf{output:} a point $\mathbf{v}\in \vertices$ such that $\mathbf{v}\in \underset{\mathbf{u}\in\polytope}{\argmax}~\boldsymbol{\omega}\cdot\mathbf{u}$. 
\end{itemize}
By convention, if the input vector is empty, the oracle returns an arbitrary point $\mathbf{v}\in\vertices$. 
\end{definition}

In order to implement the above linear optimization oracle for our polytope $\polytope$, we use the structure of this polytope. More specifically, given the optimal dual variables $\lambdavec^*$, we show that we can find a policy $\hat{\policy}$ in polynomial time such that: (i) the policy $\hat{\policy}$ is an optimal dual-adjusted index-based policy corresponding to $\lambdavec^*$, and (ii) among such policies, it maximizes $\boldsymbol{\omega}\cdot \boldsymbol{\slack}^{\hat\policy}_{\textsc{cons}}$. Formally speaking, we have the following proposition. 
\begin{proposition}
Let the polytope $\polytope$ be as defined in \Cref{subsection: reduction of equality constrained to Caratheodory}. For any given direction $\boldsymbol{\omega}\in\mathbb{R}^\numaffine$, the output of the oracle $\linoracle(\boldsymbol{\omega}; \polytope)$ can be computed in polynomial time.
\label{prop:linearoracle}
\end{proposition}
\begin{proof}{\emph{Proof.}}
Consider perturbing the vector of optimal dual variables $\lambdavec^*$ by a perturbation vector $\varepsilon \boldsymbol{\omega}$, where $\varepsilon > 0$ is an infinitesimal scalar. \revcolorm{For any admissible policy $\policy$ for the Pandora's box problem with multiple selections (or similarly, for any admissible policy for the JMS problem), we have:}
\begin{equation}
\label{eq:oracle-perturb}
\LagrangeConst(\policy; \lambdavec^*+\varepsilon\boldsymbol{\omega}) = \LagrangeConst(\policy; \lambdavec^*) + \varepsilon \boldsymbol{\omega}\cdot\slackvec.
\end{equation}

Let ${\policy}^{(\varepsilon)} \in \underset{\policy\in\PolicySpace}{\argmax}\ \LagrangeConst(\policy; \lambdavec^*+\varepsilon\boldsymbol{\omega})$. First, ${\policy}^{(\varepsilon)}$ will be a dual-adjusted index-based optimal policy corresponding to $\lambdavec^*+\varepsilon\boldsymbol{\omega}$, and thus it is polynomial-time computable. Second, since ${\policy}^{(\varepsilon)}$ maximizes the right-hand side of \eqref{eq:oracle-perturb} for an infinitesimal $\varepsilon$, it must maximize $\LagrangeConst(\policy; \lambdavec^*)$. Moreover, it should be the policy $\policy$ that maximizes $\boldsymbol{\omega}\cdot\slackvec$ among all policies $\policy$ in $\underset{\policy'\in\PolicySpace}{\argmax}\ \LagrangeConst(\policy'; \lambdavec^*)$.

Combining these observations, for sufficiently small $\varepsilon > 0$, ${\policy}^{(\varepsilon)}$ is a dual-adjusted index-based optimal policy corresponding to $\lambdavec^*$, and among such policies, which differ in their tie-breaking rules, it uses a (deterministic) tie-breaking rule that maximizes $\boldsymbol{\omega}\cdot\slackvec$.\footnote{More specifically, we can find a closed-form tie-breaking rule for any perturbation of the form $\varepsilon\direction$  using our extreme tie-breaking rules defined in \Cref{def:tie:extreme}. This can be done by simply considering a single ex-ante affine constraint  corresponding to a linear combination of our $m$ constraints with coefficients $\{\omega_i\}_{i\in[m]}$.}  Hence:
$$
\boldsymbol{\slack}_{\textsc{cons}}^{{\policy}^{(\varepsilon)}} \in \underset{\mathbf{u}\in \polytope}{\argmax}\ \boldsymbol{\omega}\cdot\mathbf{u}~~\textrm{and}~~\boldsymbol{\slack}_{\textsc{cons}}^{{\policy}^{(\varepsilon)}}\in \vertices,
$$
allowing us to implement $\linoracle(\boldsymbol{\omega}; \polytope)$ by returining $\boldsymbol{\slack}_{\textsc{cons}}^{{\policy}^{(\varepsilon)}}$ (and its corresponding policy ${\policy}^{(\varepsilon)}$) in polynomial time, as required. \qed
\end{proof}

Before proceeding to the next part, we also introduce the notion of an ``extended linear optimization oracle,'' denoted by $\extoracle$, which slightly generalizes the standard oracle $\linoracle$ that solves linear optimization over the polytope $\polytope$. Later, we show that this oracle is not a strict generalization and is indeed \emph{equivalent} to $\linoracle$ through a simple polynomial-time reduction. Consequently, if linear optimization over $\polytope$ can be solved in polynomial time, then $\extoracle$ can also be implemented as a polynomial-time oracle algorithm.

\begin{definition}[Extended Linear Optimization Oracle]
\label{def: Oracle_caratheodory}
Given the polytope $\polytope=\textrm{Conv}(\vertices)\subset \mathbb{R}^\numaffine$, the oracle $\extoracle(\cdot ; \polytope)$ is defined by this input-output relationship:
\begin{itemize}
    \item \textbf{Input:} A tuple of $k$ directions $(\boldsymbol{\omega}_i)_{i\in [k]} = (\boldsymbol{\omega}_1,\ldots,\boldsymbol{\omega}_k)$ for some $k \in \mathbb{N}$, where each $\boldsymbol{\omega}_i \in \mathbb{R}^\numaffine$.
    \item \textbf{Output:} A point $\mathbf{v} \in \facep_k\cap V$, where $\facep_k \subseteq \facep_{k-1} \subseteq \cdots \subseteq \facep_0 \triangleq \polytope$, and for each $i \in [k]$:
    $$\facep_i \triangleq \underset{\mathbf{u}\in \facep_{i-1}}{\argmax}~\boldsymbol{\omega}_i \cdot \mathbf{u}.$$
    By convention, if the input tuple is empty, the oracle returns an arbitrary point $\mathbf{v} \in \vertices$.
\end{itemize}
\end{definition}

\smallskip
We note that if we give a single direction $\boldsymbol{\omega} \in \mathbb{R}^\numaffine$ as input to the oracle $\extoracle(\boldsymbol{\omega}; \polytope)$, then it returns a point $\mathbf{v} \in \mathbb{R}^\numaffine$ such that:
$$
\mathbf{v} \in \underset{\mathbf{u} \in V}{\argmax}~\boldsymbol{\omega}\cdot\mathbf{u}\equiv\left(\underset{\mathbf{u} \in \polytope}{\argmax}~\boldsymbol{\omega}\cdot\mathbf{u}\right)\cap V.$$
Therefore, it can implement the linear optimization oracle $\linoracle$ over the polytope $\polytope$ as a special case. The following lemma shows that the oracle $\extoracle$ is in fact (computationally) equivalent to the linear optimization oracle $\linoracle$.

\begin{lemma}
    \label{lem:oracle-equivalenc} 
    Given the polytope $\polytope=\textrm{Conv}(\vertices)$, for any tuple of directions $(\boldsymbol{\omega}_i)_{i\in [k]} = (\boldsymbol{\omega}_1,\ldots,\boldsymbol{\omega}_k)$, the output of the oracle $\extoracle((\boldsymbol{\omega}_i)_{i\in[k]}; \polytope)$ can be computed by a single query to $\linoracle(\cdot; \polytope)$.
\end{lemma}

\begin{proof}{\emph{Proof.}}
Given the $k$ directions $(\boldsymbol{\omega}_i)_{i \in [k]}$, consider a single direction $\boldsymbol{\omega}^{(\varepsilon)} \triangleq \sum_{i \in [k]} \varepsilon^{(i-1)} \boldsymbol{\omega}_i$, where $\varepsilon > 0$ is an infinitesimal scalar. Let $\mathbf{v}\in V$ be the output of $\linoracle(\boldsymbol{\omega}^{(\varepsilon)}; \polytope)$, i.e.,
$$
\mathbf{v} \in \underset{\mathbf{u} \in \polytope}{\argmax}~\boldsymbol{\omega}^{(\varepsilon)} \cdot \mathbf{u} \equiv \underset{\mathbf{u}\in \polytope}{\argmax}~\boldsymbol{\omega}_1 \cdot \mathbf{u} + \varepsilon(\boldsymbol{\omega}_2 \cdot \mathbf{u}) + \varepsilon^2(\boldsymbol{\omega}_3 \cdot \mathbf{u}) + \cdots + \varepsilon^{k-1}(\boldsymbol{\omega}_k \cdot \mathbf{u}).
$$
For sufficiently small $\varepsilon$, if $\mathbf{v}$ maximizes $\boldsymbol{\omega}^{(\varepsilon)}\cdot\mathbf{u}$ over $\polytope$, it must also maximize $\boldsymbol{\omega}_1\cdot\mathbf{u}$ over $\polytope$. Let $\facep_1$ be the set of all such maximizers. Then:
$$
\mathbf{v} \in \underset{\mathbf{u}\in \facep_1}{\argmax}~\boldsymbol{\omega}_2 \cdot \mathbf{u} + \varepsilon(\boldsymbol{\omega}_3 \cdot \mathbf{u}) + \varepsilon^2(\boldsymbol{\omega}_4 \cdot \mathbf{u}) + \cdots + \varepsilon^{k-2}(\boldsymbol{\omega}_k \cdot \mathbf{u}).
$$

Applying a similar argument recursively, for small enough $\varepsilon$, $\mathbf{v}$ also maximizes $\boldsymbol{\omega}_2 \cdot \mathbf{u}$ within $\facep_1$. Recalling $\facep_0 = \polytope$ and $\facep_i = \underset{\mathbf{u}\in \facep_{i-1}}{\argmax}~\boldsymbol{\omega}_i\cdot \mathbf{u}$ in Definition~\ref{def: Oracle_caratheodory}, we conclude that for all  $i \in [k]$
$$
\mathbf{v} \in  \facep_i~.
$$
Thus, a single call to $\linoracle$ suffices to implement $\extoracle$ for any tuple of directions, as desired. \qed
\end{proof}

In the remainder of this section, we assume blackbox access to the oracle $ \extoracle $. Based on our earlier discussion, if an algorithm uses $ \extoracle $ in a computationally efficient manner, it can be implemented in polynomial time due to \Cref{prop:linearoracle} and \Cref{lem:oracle-equivalenc}.

\subsection{The Exact Algorithmic Carathéodory Problem with Oracle Access: Formal Statement \& Solution}
\label{app:sec:exact-caratheodory-algorithm}

We are now ready to formally state the problem we aim to solve:

\smallskip
\begin{displayquote}
\label{problem: exact caratheodory}
{\em
{\textbf{Problem Statement (Exact Algorithmic Carathéodory):}} Given blackbox access to the extended linear optimization oracle $\extoracle$ (as in \Cref{def: Oracle_caratheodory}) for a polytope $\polytope=\textrm{Conv}(\vertices) \subset \reals^\numaffine$, and knowing that $\polytope$ includes the origin $\origin \in \reals^\numaffine$, find a polynomial-size subset $V'\subseteq V$ of points (returned by the oracle), in polynomial time, such that $\origin \in \textrm{Conv}(\vertices')$.}
\end{displayquote}

\smallskip

\noindent\textbf{Review of basic concepts:} We start by reviewing some basic concepts and definitions in polyhedral geometry and linear algebra that we will use throughout the remainder of this section.

\begin{definition}[Cone, Dual Cone, Polar Cone]
\label{def: cones}
Let $V'\neq \emptyset$ be a bounded subset of $\mathbb{R}^m$. Then, we have the following definitions:
\begin{itemize}
    \item \textit{Conic hull of $V'$}, denoted by $\textrm{Cone}(V')$:
    $$
    \textrm{Cone}(V') \triangleq \left\{\sum_{i\in [k]}\alpha_i\mathbf{v}_i : \forall i,~\mathbf{v}_i\in V',~\alpha_i\in \mathbb{R}_{\geq 0}, ~k\in\mathbb{N}\right\}.
    $$


    \item \textit{Dual cone of $V'$}, denoted by $\textrm{Dual-Cone}(V')$:
    $$
    \textrm{Dual-Cone}(V') \triangleq \{\mathbf{u}\in \mathbb{R}^\numaffine : \mathbf{u}\cdot\mathbf{v}\geq 0 \text{ for all } \mathbf{v}\in V'\}.
    $$

    \item \textit{Polar cone of $V'$}, denoted by $\textrm{Polar-Cone}(V')$:
    $$
    \textrm{Polar-Cone}(V') \triangleq \{\mathbf{u}\in \mathbb{R}^\numaffine : \mathbf{u}\cdot\mathbf{v}\leq 0 \text{ for all } \mathbf{v}\in V'\} = -\textrm{Dual-Cone}(V').
    $$
\end{itemize}
By convention, we also set $\textrm{Cone}(\emptyset)=\{\origin\}$ and $\textrm{Dual-Cone}(\emptyset)=\textrm{Polar-Cone}(\emptyset)= \reals^\numaffine$.
\end{definition}

We also use the abbreviated notation $\cone_{V'}$, $\dualcone_{V'}$, and $\polarcone_{V'}$ to denote the conic hull, the dual cone, and the polar cone of $V'$, respectively. When it is clear from the context, we may also drop the subscript $V'$ from this notation. See \Cref{fig:enter-label} for a geometric visualization of these cones.

\begin{figure}
    \centering

\begin{tikzpicture}[line cap=round, line join=round, scale=0.9]

\tikzset{
  coneLine/.style={black, line width=0.8pt},
  coneLineaux/.style={black, line width=0.8pt,dotted},
  coneFillC/.style={fill=blue!90, opacity=0.4},     
  coneFillPolar/.style={fill=maroon!90, opacity=0.4},
  coneFillinv/.style={fill=blue!30, opacity=0.3},
  coneFilldual/.style={fill=maroon!30, opacity=0.3},
}

\coordinate (O) at (0,0);

\coordinate (Ctop) at (2,4);
\coordinate (Cright) at (4.5,0);

\coordinate (Cpolar_up) at (-4,2);
\coordinate (Cpolar_down) at (0,-4.5);

\coordinate (Cinv_up) at (-4.5,0);
\coordinate (Cinv_down) at (-2,-4);

\coordinate (Cdual_up) at (4,-2);
\coordinate (Cdual_down) at (0,4.5);

\path[coneFillC] (O) -- (Ctop) -- (Cright) -- cycle;
\draw[coneLine] (O) -- (Ctop);
\draw[coneLine] (O) -- (Cright);

\path[coneFillinv] (O) -- (Cinv_up) -- (Cinv_down) -- cycle;
\draw[coneLineaux] (O) -- (Cinv_up);
\draw[coneLineaux] (O) -- (Cinv_down);

\path[coneFillPolar] (O) -- (Cpolar_up) -- (Cpolar_down) -- cycle;
\draw[coneLine] (O) -- (Cpolar_up);
\draw[coneLine] (O) -- (Cpolar_down);

\path[coneFilldual] (O) -- (Cdual_up) -- (Cdual_down) -- cycle;
\draw[coneLineaux] (O) -- (Cdual_up);
\draw[coneLineaux] (O) -- (Cdual_down);
\begin{scope}
  \draw[coneLine] (O) -- ++(0.3,0) -- ++(0,-0.3) -- ++(-0.3,0) -- cycle;
  \draw[coneLine] (O) -- ++(0.134,0.268) -- ++(-0.268,0.134) -- ++(-0.134,-0.268) -- cycle;  
\end{scope}

\coordinate (P1) at (0.5,1);
\coordinate (P2) at (1,1.5);
\coordinate (P3) at (1.7,1.4);
\coordinate (P4) at (2,1);
\coordinate (P5) at (1.7,0.3);
\coordinate (P6) at (1,0);
\coordinate (P7) at (0.5,0.5);

\foreach \i in {1,2,...,7}
    \fill[black] (P\i) circle (2pt);

\fill[yellow, opacity=0.5] (P1) -- (P2) -- (P3) -- (P4) -- (P5) -- (P6) -- (P7) -- cycle;

\node[color=black, font=\footnotesize] at (1.2,0.8) {$\textrm{Conv}(V')$};

\node[color=black, font=\large] at (-0.3,-0.2) {$\mathbf{0}$};
\node[font=\large] at (2.5,1.6) {$C$};
\node[font=\large] at (-2.5,-1.6) {$-C$};
\node[color=maroon, font=\large] at (2.5,-0.7) {$C^*$};
\node[color=maroon, font=\large] at (-2.5,0.7) {$C^\circ$};

\end{tikzpicture}

    \caption{\revcolor{Cone $\boldsymbol{C}$ (dark blue), polar cone $\boldsymbol{C^\circ}$ (dark red), negative cone $\boldsymbol{-C}$ (light blue), and dual cone $\boldsymbol{C^*}$ (light red) of subset of points $\boldsymbol{V'}$, with the convex hull $\boldsymbol{\textrm{Conv}(V')}$ (yellow). \label{fig:enter-label}}}
    
\end{figure}


Also, recall definitions of the \emph{linear span} of a set $Y\subseteq \reals^\numaffine$: 
$$\textrm{Span}(Y)\triangleq \left\{\sum_{i\in [k]}\alpha_i\mathbf{v}_i : \forall i,~\mathbf{v}_i\in Y,~\alpha_i\in \mathbb{R},~k\in\mathbb{N}\right\},$$
the \emph{orthogonal complement} of a linear subspace $X$:
$$
X^{\perp}\triangleq \left\{\mathbf{u}\in\reals^\numaffine: \mathbf{u}\cdot\mathbf{v}=0 \text{ for all } \mathbf{v}\in X\right\},
$$
and the \emph{orthogonal projection} of a set $Y\subseteq \reals^\numaffine$ onto a linear subspace $X\subseteq \reals^\numaffine$:
\[
\textrm{Proj}_X(Y)\triangleq \left\{ \mathbf{u} \in X : \mathbf{u} = \mathbf{y} + \mathbf{v}, \, \mathbf{v} \in X^\perp,~ \mathbf{y}~ \in Y \right\}.
\]

\subsubsection{The Algorithm} Before describing our main algorithm in this section, we prove the following key technical lemma, which is crucial in both the design and analysis of our algorithm. 
\begin{lemma}
\label{lemma:covering}

Suppose $ \origin \in \polytope$ for a given polytope $\polytope\subset \reals^\numaffine$. For any given subset of points $\hat{V} \subseteq \polytope$, let $ {U} \subseteq \polytope $ be the set of all points $\left\{\mathbf{u} \in \polytope:\exists~ \direction\in \textrm{Proj}_{\textrm{Span}(\polytope)}(\polarcone_{\hat{V}})\setminus\{\origin\},~\mathbf{u}\in \argmax_{\mathbf{u}'\in \polytope}~\direction\cdot\mathbf{u}'\right\}$, that is, the convex hull of all points in the polytope $ \polytope $ that are maximizers along some non-zero direction in the projection of the polar cone $ \polarcone_{\hat{V}}$ onto the linear subspace spanning $\polytope$. Then we have:
$$
\origin \in \mathrm{Conv}({U} \cup \hat{V}).
$$
\end{lemma}


\begin{proof}{\emph{Proof.}}
We assume $\hat{V}\neq \emptyset,~U\neq \emptyset$, $\origin \notin \textrm{Conv}(\hat{\vertices})$, and $\origin \notin \textrm{Conv}(U)$, as otherwise we have:
\begin{enumerate}[label=(\roman*)]
       \item if $\hat{\vertices}=\emptyset$, then $\cone_{\hat{V}}=\{\origin\}$ and $\dualcone_{\hat{V}}=\polarcone_{\hat{V}}=\reals^\numaffine$. Hence, $\textrm{Conv}(U)=U=\polytope\ni \origin$, and we are done. 
      \item if $\argmaxCone=\emptyset$, then $\textrm{Proj}_{\textrm{Span}(\polytope)}(\polarcone_{\hat{V}})=\origin$. Denoting $\textrm{Span}(\polytope)^\perp$ by $\projspac$ (which implies $\textrm{Span}(\polytope)=\projorth$), we conclude that $\polarcone_{\hat{V}}\subseteq \projspac$. We now claim that $\cone_{\hat{V}}=\projorth$. First, note that $\cone_{\hat{V}}\subseteq\projorth$, as $\hat{V}\subseteq \polytope\subset \projorth$. Moreover, if $\mathbf{v}\in \projorth$, then $\mathbf{v}\cdot \mathbf{u}=0$ for all $\mathbf{u}\in \polarcone_{\hat{V}}$, and therefore, $\mathbf{v}\in \textrm{Polar-Cone}(\polarcone_{\hat{V}})=\cone_{\hat{V}}$. This implies that $\projorth\subseteq \cone_{\hat{V}}$, which proves our claim. Now, if $\hat{V}=\{\origin\}$ we are done. Otherwise, pick an arbitrary $\mathbf{v}\in\hat{V},\mathbf{v}\neq \origin$. Note that $\mathbf{v}\in \cone_{\hat{V}},~ \cone_{\hat{V}} =\projorth$, and $\projorth$ is a linear subspace. Therefore we should have $-\mathbf{v}\in \cone_{\hat{V}}$, and hence we should be able to write $-\mathbf{v}$ as a conic combination of vectors in $\hat{V}$. Note that $\mathbf{v}+(-\mathbf{v})=\origin$. Therefore, there exists a non-zero conic combination of vectors in $\hat{V}$ that is equal to $\origin$. By normalizing the corresponding non-negative coefficients to sum up to $1$, we have $\origin\in \textrm{Conv}(\hat{V})$ and hence we are done.
      \item if $\origin\in \textrm{Conv}(\hat{\vertices})$ or $\origin \in \textrm{Conv}(\argmaxCone)$, then clearly $\origin\in \textrm{Conv}(U\cup \hat{\vertices})$ and we are done. 
     
\end{enumerate}
Given these assumptions, we obtain the following equivalent condition for the statement of the lemma that we want to prove:
\begin{align}
\label{eq:covering_equivalency}
        \origin \in \textrm{Conv}(\argmaxCone\cup \hat{\vertices}) \iff (-\cone_{\hat{V}}) \cap \textrm{Conv}(\argmaxCone) \neq \emptyset.
\end{align}
To see the $\Rightarrow$ direction of this equivalence, note that if $\origin \in \textrm{Conv}(\hat{\vertices}\cup \argmaxCone)$, then there exists a non-zero conic combination of points in $U$ that can be written as the negative of a non-zero conic combination of points in $\hat{V}$, simply because $\hat{V}, U\neq \emptyset$ and $\origin \notin \textrm{Conv}(\hat{\vertices}),\textrm{Conv}(U)$. Therefore, after normalization, there exists a convex combination of points in $U$ that can be written as a conic combination of points in $-\hat{V}$, hence $(-\cone_{\hat{V}}) \cap \textrm{Conv}(\argmaxCone) \neq \emptyset$. To see the $\Leftarrow$ direction of the equivalence, note that if $(-\cone_{\hat{V}}) \cap \textrm{Conv}(\argmaxCone) \neq \emptyset$, then there exists a convex combination of points in $U$ that is equal to a non-zero conic combination of points in $-{\hat{V}}$, as $\origin\notin\textrm{Conv}(U)$. Therefore, there exists a non-zero conic combinations of points in $U\cup\hat{V}$ that is equal to $\origin$, and hence after normalization, there exists a convex combination of points in $U\cup\hat{V}$ that is equal to $\origin$. Therefore, $\origin \in \textrm{Conv}(\argmaxCone\cup\hat{\vertices})$. 

Having the above equivalence, to finish the proof of the lemma, we prove that the RHS of \Cref{eq:covering_equivalency} holds by contradiction. Suppose $(-\cone_{\hat{V}}) \cap \textrm{Conv}(\argmaxCone) = \emptyset$. Then there exists a strict separating hyperplane $H$ that separates $-\cone_{\hat{V}}$ and $\textrm{Conv}(\argmaxCone)$, as both of them are non-empty closed convex sets and $\textrm{Conv}(\argmaxCone)$ is compact. Note that $\origin \notin \textrm{Conv}(\argmaxCone)$ and $\origin$ is the only vertex of the cone $-\cone_{\hat{V}}$. Therefore, without loss of generality we can assume that $H$ passes through $\origin$. Let $\direction\neq \origin$ be the normal vector of $H$, pointing to the side of $H$ that includes $-\cone_{\hat{V}}$ (i.e., the upper-half). Therefore, for any $\mathbf{v}\in -\cone_{\hat{V}}$, we have that $\direction\cdot\mathbf{v}\geq 0$, and for any $\mathbf{u}\in \textrm{Conv}(U)$ we have that $\direction\cdot\mathbf{u}<0$. 

Now decompose $\direction$ into $\direction=\direction_{\mathcal{W}}+\direction_{\mathcal{W}^{\perp}}$, where $\direction_{\mathcal{W}}\in \mathcal{W}$ and $\direction_{\mathcal{W}^{\perp}}\in\mathcal{W}^{\perp}$. 
 Observe that $U\subseteq \polytope \subset \projorth$ and $\hat{V}\subseteq \polytope \subset \projorth$, hence $-\cone_{\hat{V}}\subseteq \projorth$  and $\textrm{Conv}(U)\subset \projorth$, implying that: 
 $$
 \forall\,\mathbf{v}\in -\cone_{\hat{V}}:~\direction\cdot \mathbf{v}=\direction_{\projorth}\cdot \mathbf{v}~~~~,~~~~\forall\,\mathbf{u}\in \textrm{Conv}(U):~\direction\cdot \mathbf{u}=\direction_{\projorth}\cdot \mathbf{u}
 $$
 As a result, the hyperplane $H'$ passing through $\origin$ with normal vector $\direction_{\projorth}\neq \origin$ should also be a strict separating hyperplane that separates $-\cone_{\hat{V}}$ and $\textrm{Conv}(\argmaxCone)$, with $-\cone_{\hat{V}}$ being in the upper-half. 

 We now consider a point ${\mathbf{u}}^*\in \argmax_{\mathbf{u}'\in \polytope} \direction_{\projorth}\cdot \mathbf{u}'$. As $\mathbf{u}^*,\origin\in \polytope$, we should have:
$$\direction_{\projorth}\cdot\mathbf{u}^* \geq \direction_{\projorth}\cdot\origin=0~.$$
At the same time, because $-\cone_{\hat{V}}$ is in the upper-half of hyperplane $H'$, $\direction_{\projorth}\cdot \mathbf{v} \geq 0$ for all the points $\mathbf{v}\in -\cone_{\hat{V}}$, and therefore we have $\direction_{\projorth} \in \polarcone_{\hat{V}}$. We conclude that $\direction_{\projorth}\in \textrm{Proj}_{\projorth}(\polarcone_{\hat{V}})=\textrm{Proj}_{\textrm{Span}(\polytope)}(\polarcone_{\hat{V}})$ and $\direction_{\projorth}\neq \origin$, and therefore $\mathbf{u}^*\in U$ (by the definition of the set $U$). Because $U$ is on the opposite side of hyperplane $H$ compared to $-\cone_{\hat{V}}$, i.e., in the lower-half of hyperplane $H'$, and $H'$ is a strict separating hyperplane, we have 
$$\direction_{\projorth}\cdot\mathbf{u}^*<0~,$$
a contradiction, which finishes the proof of the lemma. \qed

\end{proof}

Now that we have proved \Cref{lemma:covering}, we will formally present our algorithm, named \emph{Ellipsoid-based Exact Carathéodory (EEC)}, in \Cref{alg:Ellipsoid Exact Caratheodory}. Given blackbox oracle access to $\extoracle$ (which can be implemented in polynomial-time due to \Cref{lem:oracle-equivalenc} and \Cref{prop:linearoracle}), this algorithm recovers a subset of points $\hat{\vertices}\subseteq\vertices$ returned by the oracle such that $\origin\in \textrm{Conv}(\hat{\vertices})$, by only sending polynomial-number of queries to the oracle $\extoracle$ and some additional polynomial-time computation. Intuitively speaking, inspired by our key technical lemma in \Cref{lemma:covering}, the algorithm is designed to identify a face of the polytope $\facep\ni\origin$ and a subset of points $\hat{V}\subseteq V\cap \facep$, such that if we invoke \Cref{lemma:covering} on $\facep$ and $\hat{V}$, we have $U=\emptyset$---and hence we can conclude that $\origin\in \textrm{Conv}(\hat{\vertices})$. We formalize this statement in \Cref{sec:caratheodory-analysis}. 

Before proceeding further, we remark on the connection between our method and the celebrated \emph{``ellipsoid method''} \citep{khachiyan1979polynomial}, which we use in our analysis. 

\begin{remark}
\label{remark:ellipsoid}
At a high level, our iterative \textrm{EEC} algorithm is based on the classical ellipsoid method. We use certain mathematical properties of this method in both the algorithm and its analysis, in particular, maintaining an ellipsoid $E_t$ as search space in each iteration $t$ and updating $E_t$ by (i) identifying its center, (ii) slicing $E_t$ using a hyperplane $H$ passing through the center with a normal vector $\boldsymbol{\omega}$, and (iii) explicitly computing the next ellipsoid $E_{t+1}$ so that it is the \emph{minimal ellipsoid} containing one of the two slices produced by cutting $E_t$ with $H$. Moreover, the ellipsoid method guarantees that the volume of the resulting minimal ellipsoid shrinks exponentially fast. Specifically, if the dimension of $E_t$ (and $E_{t+1}$) is $d$, then
$$
\frac{\mathrm{Vol}(E_{t+1})}{\mathrm{Vol}(E_t)} \leq e^{-\frac{1}{2(d+1)}}.
$$
For additional details, we refer the reader to \cite{vishnoi2021algorithms}.
\end{remark}

\begin{algorithm}[ht]
\caption{\revcolor{Ellipsoid-based Exact Carathéodory (EEC)}}
\label{alg:Ellipsoid Exact Caratheodory}
    \SetKwInOut{Input}{input}
    \SetKwInOut{Output}{output}
\Input{$\numaffine$-dimensional polytope $\polytope=\textrm{Conv}(V)\ni \origin$, oracle access to $\extoracle$.}
\Output{set of points $\hat{V}\subseteq \vertices$ such that $\origin\in\textrm{Conv}(\hat{V})$.}

\smallskip
Initialize $t\leftarrow 0$, direction tuple $\projset\leftarrow \emptyset$, 
ellipsoid $\ellips_0\leftarrow \mathcal{B}_2^\numaffine(1) \triangleq \textrm{unit $\ell_2$-ball in $\reals^\numaffine$}$, $\ellips_{-1}\leftarrow \{\origin\}$, $\vertices_0\leftarrow \emptyset$. {\color{\commentcolor}\tcc{Let $\cone_{V_0} \leftarrow \textrm{Cone}(\vertices_{0})=\{\origin\}$, $\polarcone_{V_0} \leftarrow \textrm{Polar-Cone}(V_0) = \reals^\numaffine$}}






 
\vspace{1mm}

\While{$\ellipse_t \neq \{\origin\}$}{
{\color{\commentcolor}\tcc{\textbf{Outer loop}}}


Set $\projspac\leftarrow \textrm{Span}(\projset)^\perp$ and $\ellipse_0\leftarrow \textrm{Proj}_{\projspac}(\mathcal{B}_2^\numaffine(1))$
{\color{\commentcolor}\tcc{i.e., factoring out directions in $\projset$, hence $E_0\perp \projset$; the search zooms in on a face $\mathcal{A}_i\subset \projorth$ (\Cref{lemma: lower dimension faces include 0})}}

Set $t\leftarrow 0$


\While{$\ellipse_t \neq \{\origin\}$ and $\ellipse_t \neq \ellipse_{t-1}$}{
{\color{\commentcolor}\tcc{\textbf{Inner loop}}}

\If{center of ellipsoid $\ellipse_t\neq \origin$}{
Set $\direction\leftarrow$center of ellipsoid $\ellipse_t$
}\Else{
Set $\direction\leftarrow$arbitrary non-zero direction in $\ellipse_t$. 
} {\color{\commentcolor}\tcc{the direction $\direction$ is always in $\ellipse_t$ and hence orthogonal to $\projset$}}


\smallskip
Find $\mathbf{v}_t=\extoracle\left((\projset,\direction); \polytope\right)$ and set $V_{t+1}\leftarrow V_{t}\cup \{\mathbf{v}_t\}$. 

{\color{\commentcolor}\tcc{the point $\mathbf{v}_t$ is always in the face $\mathcal{A}_i\subset\projorth$ (\Cref{lemma: lower dimension faces include 0}).}}

\smallskip
\If{$\mathbf{v}_t=\origin$}{
\textbf{return} $\hat{\vertices}=\{\origin\}$. {\color{\commentcolor}\tcc{the algorithm terminates as $\origin\in \hat{\vertices}$.}}
} 
\Else{
Set $H_t\leftarrow$ hyperplane  with normal vector $\mathbf{v_t}$ passing through the center of $\ellipse_t$.

Set $H^-_t\leftarrow$negative half-space of $H_t$ (i.e., in the opposite direction of $\mathbf{v}_t$).

{\color{\commentcolor}\tcc{Let $\cone_{V_{t+1}}\leftarrow\textrm{Cone}(V_{t+1}),$ $\polarcone_{V_{t+1}}\leftarrow \textrm{Polar-Cone}(V_{t+1})=\polarcone_{V_t}\cap \{\mathbf{u}:\mathbf{u}\cdot\mathbf{v}_t\leq 0\}$ }}




Set $\ellipse_{t+1}\leftarrow$minimal ellipsoid containing $\ellipse_t\cap H_t^-$. {\color{\commentcolor}\tcc{By following the exact construction as in the ``ellipsoid method'' \citep{khachiyan1979polynomial}.}}


$t\leftarrow t+1$.


}
} 

\smallskip
\If{$\ellipse_t = \ellipse_{t-1}$}{Set $\projset\leftarrow \left(\projset, \direction\right)$.~~~~{\color{\commentcolor}\tcc{reducing the dimension of the search space.}}}

\smallskip
} 

\vspace{2mm}

\Return{$\hat{\vertices}=\vertices_t$.}
\end{algorithm}

\smallskip
\subsubsection{Analysis of the EEC Algorithm}
\label{sec:caratheodory-analysis}
The EEC algorithm always maintains a tuple of orthogonal directions $\projset$. This tuple remains unchanged during the iterations of the inner loop and is only updated at the end of the inner loop (equivalently, when the algorithm goes to the next outer loop iteration). Let $\projspac$ denote the linear subspace generated by spanning the directions in $\projset$ and let $\projorth$ denote its orthogonal complement. Each time the inner loop ends, the algorithm adds the last direction $\direction$ to the current set $\projset$ and updates $\projspac$ and $\projorth$ accordingly, unless $\ellipse_t = \{\origin\}$, in which case the algorithm terminates. As an important invariant, the EEC algorithm should maintain a subset $V_t\subseteq V$ of points at each iteration $t$ of the inner loop, such that $V_t\subset \projorth$ (as we show next).



Consider the tuple $\projset=(\direction_1,\direction_2,\ldots,\direction_i)$ at some point during the execution of the algorithm, where $i$ is the size of $\projset$, together with the corresponding linear subspace $\projspac$ and its orthogonal complement $\projorth$. These components help the algorithm identify a suitable (lower-dimensional) face of the polytope $\polytope$, namely, the face $\facep_i$ derived in \Cref{def: Oracle_caratheodory} given the directions in $\Omega$. As another important invariant of the algorithm, this face should contain the origin $\origin$ and satisfy $\facep_i\subset \projorth$.
The algorithm then ``zooms in'' on $\facep_i$ and the search is restricted to finding a subset of points $\hat{V}\subseteq V\cap\facep_i$ such that $\origin\in\textrm{Conv}(\hat{V})$---in other words, the algorithm re-starts the search starting from $\facep_i$ as if $\facep_i$ was the initial polytope. We note that these face polytopes are nested, that is, $\facep_i\subset \facep_{i-1}\subset \ldots \facep_0\equiv \polytope$ at any point during the execution of the algorithm. 
The following lemma (\Cref{lemma: lower dimension faces include 0}) formalizes this connection and shows they satisfy such desired properties mentioned.


\begin{lemma}
\label{lemma: lower dimension faces include 0}
    Given a tuple of directions $\projset=(\direction_1,\direction_2,\ldots,\direction_i)$ for any $i\in\mathbb{Z}_{\geq 0}$ at any point during the execution of \Cref{alg:Ellipsoid Exact Caratheodory}, and its corresponding orthogonal complement space $\projorth_i$, we have
    $$\origin \in \facep_i~~\textrm{and}~~\facep_i\subset \projorth_i~,$$
    where $\facep_i$ is defined as in \Cref{def: Oracle_caratheodory} for $\extoracle\left(\projset; \polytope\right)$.
\end{lemma}

\begin{proof}{\emph{Proof.}}
    First, we prove $\origin\in \facep_i$ by induction on $i$, the size of $\projset$. If $i=0$ (that is, the tuple is empty, which happens at the beginning of the execution of the algorithm), we know $\origin \in \polytope = \facep_0$. Now suppose that $\origin\in \facep_i$ with $\projset=(\direction_1,\direction_2,\ldots,\direction_i)$ at the beginning of some outer iteration of the algorithm. At the end of this outer iteration, we will either terminate or update the set of directions $\projset$ to $(\direction_1,\ldots,\direction_i, \direction)$, where $\direction$ corresponds to the direction identified in lines 6-9 of \Cref{alg:Ellipsoid Exact Caratheodory} in the last iteration of the inner loop. If this update occurs, it has to be the case that $\ellipse_{t+1} = \ellipse_{t}$, which only happens if the hyperplane $H_t$ contains the entire $\ellipse_t$ rather than cutting it through. Consequently, $\mathbf{v}_t$ should be orthogonal to $\ellipse_t$, and in particular to $\direction\in \ellipse_t$, implying $\direction\cdot\mathbf{v}_t=0$. Furthermore, by applying the induction hypothesis, we already know $\origin \in \facep_i$. Therefore, by construction of $\mathbf{v}_t$, that is, $\mathbf{v}_t=\extoracle\left((\projset,\direction); \polytope\right)$, we can write
    \begin{align}
    \label{eq:max}
        \mathbf{v}_t \in \facep_{i+1} = \argmax_{\mathbf{u}'\in \facep_{i}}~ \direction\cdot\mathbf{u}' \ \rightarrow \ \max_{\mathbf{u}'\in \facep_{i}} ~\direction\cdot\mathbf{u}'=\direction\cdot \mathbf{v}_t = 0 = \direction\cdot\origin~.
    \end{align}
    This shows $\origin \in \facep_{i+1}$, which completes the induction proof.

    Second, we also prove $\facep_i\subset \projorth_i$ by induction. As for the base of induction, $\facep_0 =\polytope\subset \reals^\numaffine = \projorth_0$. Now assume $\facep_i\subset \projorth_i$, and we show that $\facep_{i+1}\subset \projorth_{i+1}$. Note that $\facep_{i+1}\subset \facep_i\subset \projorth_i$. Therefore, it is enough to show that $\forall \mathbf{v}\in \facep_{i+1}$, $\direction\cdot \mathbf{v}=0$, and hence $\facep_{i+1}\subset \projorth_{i+1}$. Now, if $\mathbf{v}\in \facep_{i+1}$, then $\direction\cdot \mathbf{v}=\max_{\mathbf{u'}\in \facep_i}~\direction\cdot\mathbf{u}'=0$ (as we showed earlier in \cref{eq:max}), finishing the induction proof. 
    \qed 
\end{proof}

\smallskip
Suppose that upon termination of the algorithm, we have $\projset=(\direction_1,\direction_2,\ldots,\direction_k)$ for some $k\in \mathbb{Z}_{\geq 0}$. Let $\hat{V}=V_\tertime \subseteq \vertices \cap \facep_k$ be the final set of points returned by the algorithm, where $t=\tertime$ is the last index of the last inner loop before termination. By applying the key \Cref{lemma:covering} to the face polytope $\facep_k$, which includes $\origin$ as stated in \Cref{lemma: lower dimension faces include 0} we have
$$
\origin \in \textrm{Conv}(V_\tertime\cup U_\tertime),
$$
where $U_\tertime\subseteq \facep_k$ is defined as in \Cref{lemma:covering}, i.e., the convex hull of all points in the face polytope $\facep_k$ that are maximizers along some non-zero direction in the polar cone $\polarcone_{V_{\tertime}}$ of $V_\tertime$ projected onto the linear span of $\facep_k$. If we manage to show that (i)~the EEC algorithm (\Cref{alg:Ellipsoid Exact Caratheodory}) terminates after a polynomial number of iterations, and (ii)~at termination $\textrm{Proj}_{\textrm{Span}(\facep_k)}(\polarcone_{V_\tertime})=\{\origin\}$ and therefore $U_\tertime=\emptyset$, then we have a polynomial-time algorithm recovering a subset $V_\tertime\subseteq V$ such that $\origin\in\textrm{Conv}(V_\tertime)$.


In order to show property~(ii) above, it is enough to show that $\polarcone_{V_{\tertime}} \subseteq \textrm{Span}(\facep_k)^\perp$ and therefore $\textrm{Proj}_{\textrm{Span}(\facep_k)}(\polarcone_{V_\tertime})=\{\origin\}$.  To establish this claim, we present and prove two simple lemmas.

The first lemma (\Cref{lemma:inner-loop}), intuitevly speaking, controls the ``projected volume'' of the polar cone $\polarcone_{V_\tertime}$, which turns out to be crucial for establishing the claim.

\begin{lemma}
\label{lemma:inner-loop}
At any iteration $t$ of any of the inner loops of \Cref{alg:Ellipsoid Exact Caratheodory}, we have $\polarcone_{V_{t}}\cap \ellipse_0\subseteq \ellipse_t$. 
\end{lemma}
\begin{proof}{\emph{Proof.}}
Fix an inner loop of the algorithm (corresponding to a particular outer iteration). We prove the claim by induction on $t$. As for the base of the induction, for $t=0$ we clearly have $\polarcone_{V_{0}}\cap \ellipse_0\subseteq \ellipse_0$. Now suppose $\polarcone_{V_{t}}\cap \ellipse_0\subseteq \ellipse_{t}$ at iteration $t$, and we show that $\polarcone_{V_{t+1}}\cap \ellipse_0\subseteq \ellipse_{t+1}$.

To see this,  let $\mathbf{c}_t$ denote the center of $\ellipse_t$. First of all, if $\mathbf{c}_t=\origin$, then
\begin{equation}
\label{eq:equal-halfspace}
H_t^-=\{\mathbf{u}:\mathbf{u}\cdot \mathbf{v}_t\leq 0 \}.
\end{equation}
If $\mathbf{c}_t\neq\origin$, then the algorithm sets the direction $\direction$ in that inner iteration to $\mathbf{c}_t$. At the same time, if $\projset=(\direction_1,\ldots,\direction_i)$ at the beginning of this inner loop, then based on \Cref{lemma: lower dimension faces include 0} we know $\origin\in \facep_i$ (recall the definition of $\facep_i$ in \Cref{def: Oracle_caratheodory}). As a result, according to the definition of $\mathbf{v}_t$, we have $\mathbf{c}_t\cdot\mathbf{v}_t\geq \mathbf{c}_t\cdot\origin=0$, and therefore
\begin{equation}
\label{eq:nested-halfspace}
H^{-}_t=\{\mathbf{u}:\mathbf{u}\cdot\mathbf{v}_t\leq \mathbf{c}_t\cdot \mathbf{v}_t\}\supseteq \{\mathbf{u}: \mathbf{u}\cdot\mathbf{v}_t\leq 0\} 
\end{equation}

By combining the induction hypothesis $\polarcone_{V_{t}}\cap \ellipse_0\subseteq \ellipse_t$ with \cref{eq:equal-halfspace} (or \cref{eq:nested-halfspace}), noting that  $\polarcone_{V_{t+1}}=\polarcone_{V_{t}}\cap \{\mathbf{u}:\mathbf{u}\cdot \mathbf{v}_t\leq 0 \}$ and $E_{t+1} \supseteq E_t\cap H^{-}_t$, we have 
$$
\polarcone_{V_{t+1}}\cap E_0=\polarcone_{V_{t}}\cap \ellipse_0\cap \{\mathbf{u}: \mathbf{u}\cdot\mathbf{v}_t\leq 0\}\subseteq \ellipse_t\cap H^{-}_t \subseteq E_{t+1}~,
$$
and therefore $\polarcone_{V_{t+1}}\cap \ellipse_0\subseteq \ellipse_{t+1}$, as desired.\qed
\end{proof}

\smallskip
Using \Cref{lemma:inner-loop} and the fact that algorithm EEC only terminates when the ellipsoid of the last iteration satisfies $\ellipse_\tertime=\{\origin\}$, we conclude that:
\begin{align}
\label{eq:empty_proj_polar_cone}
    \polarcone_{V_\tertime}\cap \ellipse_0= \ellipse_\tertime = \{\origin\}
\end{align}
We now have our second lemma (\Cref{lemma:orthogonal}) that builds on this conclusion to show that $\polarcone_{V_\tertime}\subseteq \mathcal{W}$. In fact, we prove a slightly stronger claim.

\begin{lemma}
\label{lemma:orthogonal}
    If $\polarcone_{V_\tertime}\cap E_0=\{\origin\}$, then $\polarcone_{V_\tertime}= \textrm{Span}(\facep_k)^\perp =\mathcal{W}$.
\end{lemma}
\begin{proof}{\emph{Proof.}}
    First, note that  $\vertices_\tertime \subseteq \facep_k \subset \projorth \equiv \textrm{Span}(\ellipse_0)$ by construction of $\vertices_\tertime$ and \Cref{lemma: lower dimension faces include 0}; therefore, for any $\direction'\in \mathcal{W}$, we have: 
    $$\forall \,\mathbf{v}\in V_{\tertime}:~\direction'\cdot \mathbf{v}=0~,$$
    and hence $\direction'\in \polarcone_{V_\tertime}$. This implies $\projspac \,\subseteq \textrm{Span}(\facep_k)^\perp \,\subseteq \polarcone_{V_\tertime}$. 
    
    Second, we prove $\mathcal{W}\supseteq \polarcone_{V_\tertime}$ by contradiction. Suppose that there exists a (non-zero) direction $\direction'\in \polarcone_{V_\tertime}\setminus\projspac$. Then $\direction'$ can be decomposed into $\direction'=\direction'_{\projspac}+\direction'_{\projorth}$, where $\direction'_{\projspac}\in\projspac$, $\direction'_{\projorth}\in\projorth$, and $\direction'_{\projorth}\neq\origin$. Now, for any $\mathbf{v}\in V_\tertime$, noting that $V_\tertime\subset \projorth$ and thus $\direction'_{\projspac} \cdot \mathbf{v} = 0$, we have:
    $$
    \direction'_{\projorth}\cdot \mathbf{v}=\direction'\cdot\mathbf{v}\leq 0~,
    $$
    where the last inequality holds due to the definition of the polar cone (\Cref{def: cones}). Consequently,  $\direction'_{\projorth}\in\polarcone_{V_\tertime}$, and hence $\direction'_{\projorth}/\lVert\direction'_{\projorth}\rVert\in\polarcone_{V_\tertime}$. Moreover, the vector $\direction'_{\projorth}/\lVert\direction'_{\projorth}\rVert$ clearly belongs to the projection of the unit $\ell_2$-ball onto the linear subspace $\projorth$ (which is $E_0$). Therefore, we have $\polarcone_{V_\tertime}\cap E_0\ni\direction'_{\projorth}/\lVert\direction'_{\projorth}\rVert\neq \origin$, a contradiction with $\polarcone_{V_\tertime}\cap E_0=\{\origin\}$. \qed
\end{proof}

\smallskip
Putting everything together---in particular, having Lemmas~\ref{lemma: lower dimension faces include 0}, \ref{lemma:inner-loop}, and \ref{lemma:orthogonal}---we are now ready to show the following main theorem of this section. 

\begin{theorem}
\label{thm: exact caratheodory alg}
    Algorithm EEC (\ref{alg:Ellipsoid Exact Caratheodory}) terminates in polynomial time w.r.t. the size of the problem instance. Furthermore, if $\hat{\vertices}\subseteq V$ is the final set of points returned by the algorithm, we have:
    \begin{align}
        \origin \in \textrm{Conv}(\hat{\vertices}).
    \end{align}
\end{theorem}

\begin{proof}{\emph{Proof.}}
    To show that the algorithm terminates in polynomial time, notice that we are essentially following the update rule of the ``ellipsoid method" to obtain $\ellipse_{t+1}$ from  $\ellipse_t$ at each iteration $t$ of each inner loop of the algorithm (which corresponds to a fixed outer iteration). As such, we know that the volume of the ellipsoid will be geometrically shrinking (except only in the last inner iteration). Therefore, for every inner-loop, the number of iterations of the algorithm in that inner loop is upper bounded by $\mathcal{O}\left(\log\frac{1}{\delta}\right)$, where $\delta=2^{-b}$ is the numerical precision of the input instance and $b$ is the bit complexity of the input instance. Hence, the  number of iterations in each inner-loop is polynomial in $b$. Also, when each inner loop terminates, the hyperplane $H_t$ contains the entire $\ellipse_t$ rather than cutting it through (and therefore, $\ellipse_t=\ellipse_{t+1}$)---unless $\ellipse_t = \{\origin\}$, in which case the algorithm would terminate. This former case only happens when $\ellipse_t$ is in lower dimension and $\mathbf{v}_t$ is orthogonal to it. As the dimension of $\ellipse_0$ goes down by exactly $1$ at each outer iteration, and the starting dimension $\textrm{dim}(\polytope)=\numaffine$, we can only have at most $\numaffine$ number of outer iterations. Putting these pieces together, the algorithm terminates after sending polynomial number of queries to the oracle $\extoracle$ and polynomial-time extra computation, as desired.

To show that $\origin \in \textrm{Conv}(\hat{\vertices})$, as mentioned earlier, we first invoke the key technical lemma (\Cref{lemma:covering}) for the face polytope $\facep_k\ni \origin$ (\Cref{lemma: lower dimension faces include 0}) and the final set of points $\hat{V}=V_\tertime \subseteq \vertices \cap \facep_k$. We note that $\polarcone_{V_\tertime}= \textrm{Span}(\facep_k)^\perp =\mathcal{W}$ (\Cref{lemma:inner-loop}, \cref{eq:empty_proj_polar_cone}, and \Cref{lemma:orthogonal}), and therefore $\textrm{Proj}_{\textrm{Span}(\facep_k)}(\polarcone_{\vertices_\tertime})=\{\origin\}$. As a result $\argmaxCone_\tertime=\emptyset$, and hence because of the covering guarantee of \Cref{lemma:covering} we should have:
$$
\origin \in \textrm{Conv}(\vertices_\tertime \cup \argmaxCone_\tertime)=\textrm{Conv}(\hat{\vertices})~,
$$
which finishes the proof of the theorem.   \qed

\end{proof}


We conclude this section by remarking that although one could potentially solve this variant of the exact algorithmic Carathéodory problem in polynomial time using standard (but more involved) reductions from optimization to separation (and vice versa)---since algorithmic Carathéodory via a combination of separation and membership oracles is well-known---the resulting algorithm would be quite complicated and would not match the running time of our {\Cref{alg:Ellipsoid Exact Caratheodory}}. Moreover, our algorithm is able to recover an almost-linear \(\tilde{\mathcal{O}}(m)\) number of policies in the resulting randomized tie-breaking rule (thinking of the encoding bit complexity of the problem instance as a constant $\epsilon$, to reduce the volume of the initial ellipsoid in each inner-iteration of  \Cref{alg:Ellipsoid Exact Caratheodory} to ${\epsilon}^m$, we need $\mathcal{O}\left(m\log(1/\epsilon)\right)$ number of iterations), or in other words, the size of the uncovered convex combination by our algorithm as a function of the dimension $m$ is almost linear. Note that based on the Carathéodory theorem, this is almost the best possible, as every point can be represented by a convex combination of at most $m+1$ vertices in an $m$-dimensional polytope (and this is tight). Also, our algorithm is simple, structured, and interpretable---which is completely in contrast to any other known method for solving these types of Carathéodory problems in the literature, e.g., \cite{grotschel1981ellipsoid,grotschel2012geometric}.

\subsection{Extension to Multiple General Affine Constraints}
\label{subsec: reduction of inequality to equality}

So far we have assumed that all affine constraints are equalities. We now show how to reduce the problem with $m$ general affine constraints, some of which are inequalities, to a problem with only equality affine constraints. 

Suppose that our original problem has $\numaffine_n$ inequality and $\numaffine_e$ equality constraints ($\numaffine=\numaffine_n+\numaffine_e$). For each constraint slack vector $\slackvec=(\slackconstj)_{j\in[m]}\in\reals^\numaffine$, let the first $\numaffine_n$ coordinates correspond to the inequality constraints, and the rest correspond to the equality constraints. In the presence of inequality constraints, the goal of our problem is to find a small subset of policies $\{\policy^{(i)}\}_{i\in S}$, with $S\subseteq [N]$, such that: 
$$\textrm{Conv}\left(\{\boldsymbol{\Delta}^{\policy^{(i)}}_\textsc{cons}\}_{i \in S}\right)\cap Q \neq \emptyset,$$
where $Q\triangleq\{\mathbf{v}\in\reals^\numaffine:\forall i\in[m_n], \mathbf{v}_i\geq 0,~\forall j\in[m_e], \mathbf{v}_{m_n+j}=0\}$. We also have the guarantee that $\polytope\cap Q\neq \emptyset$ before finding the set $S$, where $\polytope=\textrm{Conv}(\{\boldsymbol{\Delta}^{\policy^{(i)}}_\textsc{cons}\}_{i \in [N]})$, as before. We remark that in the special case with $m_n=0$, we have $Q=\{\origin\}$ and this problem becomes the exact Carathéodory problem in \Cref{subsection: reduction of equality constrained to Caratheodory}. In what follows, we basically show that this new problem is not a strict generalization, and there is a bi-directional polynomial-time reduction from this problem to the exact Carathéodory problem.




To see this reduction, consider adding $m_n$ dummy vectors $\{-\mathbf{e}^{(i)}\}_{i\in[m_n]}$ to the original set of slack vectors $\{\boldsymbol{\Delta}^{\policy^{(i)}}_\textsc{cons}\}_{i \in [N]}$, where $\mathbf{e}^{(i)}\in\reals^\numaffine$ is the standard unit vector for coordinate $i$. Define 
$$\bar{\polytope}=\textrm{Conv}\left(\{\boldsymbol{\Delta}^{\policy^{(i)}}_\textsc{cons}\}_{i \in [N]}\cup\{-\mathbf{e}^{(i)}\}_{i\in[m_n]}\right)~.$$
Note that $Q=\textrm{Cone}(\{\mathbf{e}^{(i)}\}_{i\in[m_n]})$. Therefore, we have the following equivalence:
$$\polytope\cap Q \neq \emptyset  \quad \Longleftrightarrow \quad \origin \in \bar{\polytope}.$$

Moreover, given oracle access to the linear optimization oracle $\linoracle(\cdot;\polytope)$ for the polytope $\polytope=\textrm{Conv}(V)$, we can easily solve linear optimization over polytope $\bar{\polytope}$ along some direction $\direction$ by first calling $\linoracle(\direction;\polytope)$ to return $\mathbf{v}^*\in V$ and then comparing $\direction\cdot \mathbf{v}^*$ with $\direction\cdot \medit{-} \mathbf{e}^{(i)}$ for all $i\in[m_n]$ and then returning the one with the maximum dot product. Using the equivalence of linear optimization oracle and the extended linear optimization oracle as in \Cref{lem:oracle-equivalenc}, we can construct the extended linear optimization oracle $\extoracle(\cdot; \bar{\polytope})$ for $\bar{\polytope}$.

    




Putting all the pieces together, our reduction is as follows: we know $\polytope\cap Q \neq \phi$, so we know $\origin \in \bar{\polytope}$. Now, by using $\extoracle(\cdot;\bar{\polytope})$, we can efficiently find a polynomial-sized subset of points $\vertices' \subseteq V\cup \{-\mathbf{e}^{(i)}\}_{i\in[m_n]}$, such that $\origin \in \textrm{Conv}(\vertices')$. Nonetheless, we can partition $\vertices'=\hat{V}\cup \{-e^{(i_{\ell})}\}_{\ell}$, where $\hat{V} \subseteq V$. Note that $\origin\in \textrm{Conv}(\vertices')$. Therefore, there exists a convex combination of points in $\hat{V}$ that is equal to a conic combination of points $\{e^{(i_{\ell})}\}_{\ell}\subseteq \textrm{Cone}(\{\mathbf{e}^{(i)}\}_{i\in[m_n]}) = Q$, implying $\textrm{Conv}(\hat{V})\cap Q\neq \emptyset$ and completing our reduction.




 \revcolor{

\subsection{Failure of the Extreme Tie-breaking Rules for Multiple Constraints}
\label{app:extreme-failure-multiple-affine}
In this section, we provide an illustrative example demonstrating why a simple extension of our extreme tie-breaking rules for \revcolorm{a single constraint in the Pandora's box setting may not work even in the case with $m=2$ ex-ante affine constraints---highlighting the importance of our earlier approach by solving the problem via a reduction to the exact algorithmic Carathéodory problem with oracle access.}

\smallskip
\subsubsection{Overview of the Suggested Approach} 
Recall that when we had a single affine constraint ($\numaffine=1$), we observed in \Cref{sec:randomized-tie} that the two extreme tie-breaking rules induced by the perturbed optimal dual variables $\LagVector^* - \varepsilon$ and $\LagVector^* + \varepsilon$ indeed achieve the two extreme slack values (highest and lowest) in the constraint among all possible tie-breaking rules. Consequently, one slack should be non-negative while the other should be non-positive (as the problem is feasible), thus enabling zero slack by randomizing over them. 

Given the success of extreme tie-breaking rules in the special case of $m=1$, it might seem reasonable to generalize this idea to settings with $\numaffine > 1$ constraints. To this end, we consider $2^m$ dual-adjusted optimal policies corresponding to specific perturbed versions of vector $\boldsymbol{\lambda^*}\in \mathbb{R}^m$, that is, vectors of the form $\boldsymbol{\lambda^*}+\boldsymbol{\varepsilon}$, where the perturbation vector has the form $\boldsymbol{\varepsilon}\in \{-\varepsilon,+\varepsilon\}^m$ for an infinitesimal scalar $\varepsilon>0$. One might hope that with a proper randomization over this set, we can make the slack of all constraints zero. However, this approach fails due to the following two reasons.





First, even if this approach can yield zero slacks for all the constraints simultaneously, the computational complexity of this method is significant. The running time is exponential with respect to $\numaffine$, making it impractical for real-world implementations when $\numaffine$ is large. 

 Second, ignoring its computational complexity, there is a deeper issue with this approach and can fail accordingly. In the remainder, we explain the issue first, and then show a simple example in which it arises.


\smallskip
\revcolorm{
\subsubsection{A Geometric Interpretation of the Issue}
\label{app:bad_example_extreme_not_enough}
Recall the definition of polytope $\polytope=\textrm{Conv}(V)$ from \Cref{subsection: reduction of equality constrained to Caratheodory}. Each point $\mathbf{v}\in V$ corresponds to the slack vector of a dual-adjusted optimal policy $\policy^{(i)}$ with some tie-breaking rule.  When we focus only on the extreme tie-breaking rules discussed above, then we essentially have a subset of size $2^m$ of the points in $V$, corresponding to the slacks of tie-breaking rules derived from perturbations $\lambdavec^*+\boldsymbol{\varepsilon}$ for $\boldsymbol{\varepsilon}\in\{\pm\varepsilon\}^m$. 

For $m>1$ affine equality constraints, we can show that the above subset may \emph{not} contain $\origin$ in its convex hull, indicating that there is no randomization over the $2^m$ corresponding dual-adjusted optimal policies that can satisfy all constraints exactly. The following illustrates this phenomenon with a simple and concrete example.}

\medskip
\noindent\textbf{A simple counterexample:} 
In the following, we present a simple parametric example that rigorously demonstrates the phenomenon described earlier. This example has only $3$ candidates, where candidates $1$ and $2$ belong to $\ManSet$ and candidate $3$ belongs to $\WomanSet$. Each candidate/box has a binary random value as follows:

\revcolorm{
\begin{align}
\label{app:bad_example_carathodory_values}
    v_1 =
    \begin{cases}
        H_1 = 6 \quad & \textrm{w.p.} \ 1\\
        L_1 = 1 \quad & \textrm{w.p.} \ 0
    \end{cases}
    ,~v_2 =
    \begin{cases}
        H_2 = 14 \quad & \textrm{w.p.} \ p_2,\\
        L_2 = 2 \quad & \textrm{w.p.} \ 1-p_2.
    \end{cases}
    ,~v_3 =
    \begin{cases}
        H_3 = 9 \quad & \textrm{w.p.} \ p_3\\
        L_3 = 3 \quad & \textrm{w.p.} \ 1-p_3
    \end{cases}.
\end{align}
Also let the inspection costs be $c_1=2$, $c_2=10p_2$, and $c_3=5p_3$, respectively. This implies $\reserve_1 = \reserve_2 = \reserve_3 = 4$, where $\reserve_i$ is the reservation value of candidate $i$, as defined in \Cref{eq:base:reserve}. Hence, it holds that:
\begin{align}
\label{eq:example-cara}
    \min(H_1,H_2,H_3) > \reserve_1 = \reserve_2 = \reserve_3 > L_3 > \max (L_1,L_2).
\end{align}
We note that our counterexample remains valid if \eqref{eq:example-cara} holds, which can be satisfied by several other choices of numerical values as well. 
}
Moreover, \eqref{eq:example-cara} ensures that an optimal policy would stop if and only if (i) it sees a high value and selects it, or (ii) it has already inspected all three candidates (all with low values) and then selects the third candidate.


As for the constraints, suppose the decision maker wants to simultaneously satisfy (normalized versions of) demographic parity in both selection and inspection, that is, 

\begin{equation}
\label{eq:normalized parity}
        \expect{\select_1^{\policy} + \select_2^{\policy} } = 2 ~ \expect{\select_3^{\policy}}~~ \textrm{and}~~~ \expect{\inspect_1^{\policy} + \inspect_2^{\policy} } = 2 ~ \expect{\inspect_3^{\policy}}.
\end{equation}

With this setting in mind, we highlight that in this example the loss due to fairness is designed to be $0$, which means that $\LagVector^*=\origin$ (this can easily be verified). In other words, one of the (randomized) optimal unconstrained policies is also feasible. Thus, the remaining question is to find exactly which tie-breaking rules to choose.

Because all $\sigma_i$'s are equal, the policy can inspect the candidates in any of the $3!=6$ possible permutations over  $\{1,2,3\}$ (and that is the only source of tie in this example).  After simple calculations for each of these permutations, the resulting slack in \ref{eq:normalized parity} for parity in inspection and selection, denoted by $\slack_I$ and $\slack_S$, respectively, are as follows 
(here, we choose $p_2=0.1$ and $p_3=0.8$):

\begin{itemize}
    \item $1 \rightarrow 2 \rightarrow 3$:
    $\slack_I=1, \slack_S=1$

    \item $2 \rightarrow 1 \rightarrow 3$:
    $\slack_I=2-p_2=1.9, \slack_S=1$
    
    \item $3 \rightarrow 1 \rightarrow 2$: 
    $\slack_I=-(1+p_3)=-1.8, \slack_S=1-3p_3=-1.4$
    
    \item $3 \rightarrow 2 \rightarrow 1$:
    $\slack_I=p_2p_3-p_2-2p_3=-1.62, \slack_S=1-3p_3=-1.4$
    
    \item $1 \rightarrow 3 \rightarrow 2$:
    $\slack_I=1, \slack_S=1$
    
    \item $2 \rightarrow 3 \rightarrow 1$:
    $\slack_I=p_2p_3+p_2-p_3=-0.62, \slack_S=1-3p_3+3p_2p_3=-1.16$
\end{itemize}


    
    
    
    


Nevertheless, it turns out that if we have $0<p_2,p_3<1$,
then permutations $2\rightarrow 1\rightarrow 3$ and $3 \rightarrow 1 \rightarrow 2$ are the only ones that can be obtained by the four perturbations $(\pm\varepsilon,\pm\varepsilon)$ corresponding to the extreme tie-breaking rules. 
Moreover, it is easy to verify that for a general setup of parameters $p_2,p_3 \in (0,1)$, having only the second and third permutations are not enough to cover $\origin$, but once we also include the slack of other permutations, then we will be able to cover $\origin$. \Cref{fig:6 point example} illustrates this point.

\begin{figure}[htb]
    \centering
    \begin{subfigure}[b]{0.5\textwidth}
    \centering
        \begin{tikzpicture}[scale=0.9]
            \begin{axis}[
                axis lines = middle,
                xlabel = {$\slack_I$},
                ylabel = {$\slack_S$},
                xmin=-2, xmax=2,
                ymin=-2, ymax=2,
                grid=both,
            ]
            \addplot[
                only marks,
                mark=*,
                color=blue,
                ] coordinates {
                (1, 1) (1.9, 1) (-0.62, -1.16) (-1.62, -1.4) (-1.8, -1.4) (1, 1) 
            };
            \addplot[
                only marks,
                mark= *,
                color=red,
                ] coordinates {
                (1.9, 1) (-1.8, -1.4) 
            };
            \addplot[
                color=blue,
                thin
                ] coordinates {
                (1, 1) (1.9, 1) (-0.62, -1.16) (-1.62, -1.4) (-1.8, -1.4) (1, 1) 
            };
            \addplot[
                color=red,
                thick
                ] coordinates {
                (-1.8, -1.4) (1.9, 1)
            };
            \end{axis}
        \end{tikzpicture}
    \end{subfigure}%
    \caption{\revcolor{Configuration of the slacks for all 6 orders, under $\boldsymbol{p_2=0.1, p_3=0.8}$. The top right and bottom left belong to the two distinct tie-breaks achievable by $\boldsymbol{\{-\varepsilon, +\varepsilon\}^2}$. The fact that red line does not pass through origin shows the insufficiency of this approach. \label{fig:6 point example}}}
\end{figure}
}

}
\revcolorm{

\section{Constraint Formulations, Examples}
\label{sec: Constraint Formulations}
In this section we exemplify some of the applications we mentioned in \Cref{sec:general-ex-ante-constraints} for the types of constraints that we consider. In order to provide better intuition, we focus on the Markov chains corresponding to our \Cref{ex:reject} as illustrated in \Cref{fig:MC-hiring}. 

\subsection{Affine Constraints}
\label{subsec: affine Constraint Formulations}
For simplicity of exposition, let us first define $S_i^{\textrm{phone}}, S_i^{\textrm{onsite}}, S_i^{\textrm{offer}}$ as the set of all vertices (i.e., states) for ${\MC_i}$ (candidate $i$'s Markov chain) corresponding to phone, onsite, and offer stages, respectively. Note that for every candidate $i$ both $S_i^{\textrm{phone}}$ and $S_i^{\textrm{onsite}}$ consist of only a single state in this particular example (see \Cref{fig:MC-hiring}), while $S_i^{\textrm{offer}}$ includes multiple states. Additionally, we denote by $\AffineVector_{i,v}$, the coefficient in an affine constraint corresponding to the node $v\in\MC_i$.

In the following, we formally elaborate some of the mentioned affine constraints in the main body.

\textbf{Minority group quota on offer.} In this case, we add a single constraint to set a quota on the minority group $\WomanSet$ candidates' expected number of offers: 
\begin{align}
    \AffineVector_{i,v} = 
    \begin{cases}
        (\theta-1) \quad & \textrm{if} ~ i\in \WomanSet ~ \& ~ v\in S_i^{\textrm{offer}} \\ 
        \theta \quad & \textrm{if} ~ i\in \ManSet ~ \& ~ v\in S_i^{\textrm{offer}} \\ 
        0 \quad & \textrm{o.w.}
    \end{cases},
\end{align}
with the constraint being 
$\AffineVector \cdot \alloc \leq 0.$

\textbf{Onsite interview budget.} Here, again we add a single constraint on all candidates' onsite interview stage: 
\begin{align}
    \AffineVector_{i,v} = 
    \begin{cases}
        1 \quad \textrm{if} ~ v\in S_i^{\textrm{onsite}} \\ 
        0 \quad \textrm{o.w.}
    \end{cases},
\end{align}
with the constraint being: $\AffineVector \cdot \alloc \leq \AffineConstant$
where $\AffineConstant$ is our desired average budget for onsite interview (in terms of number of people).

\textbf{Individual fairness in phone interview (lower bounding the minimum opportunity).} 
For this application, we should add a constraint for every candidate $j \in [\altnum]$, denoted by $\AffineVector^j$, such that only the coefficients corresponding to the $j^{\textrm{th}}$ candidate's phone interview stage would appear in it: 
\begin{align}
    \forall j \in [\altnum], \quad \AffineVector_{i,v}^j = 
    \begin{cases}
        1 \quad \textrm{if} ~ i=j ~ \& ~ v\in S_i^{\textrm{phone}} \\ 
        0 \quad \textrm{o.w.}
    \end{cases},
\end{align}
with the constraints being:
$\AffineVector^j \cdot \alloc \geq \AffineConstant, \forall j \in [\altnum],$
where $\AffineConstant$ is some desired lower bound for each individual's phone interview chance.

\subsection{Convex Constraints}
\label{subsec: convex Constraint Formulations}

Here we formulate some examples of convex constraints that our model is able to capture, including but not limited to Nash social welfare, negative entropy, and the generalized mean (a.k.a H\"{o}lder mean), which are standard notions for capturing various forms of egalitarian welfare, e.g., see \cite{roughgarden2010algorithmic}.

\textbf{Individual-level Nash social welfare on phone interview.} We need a single constraint on all candidates' phone interview stage: 
$$\AffineConstant \leq \ConvexFun(\alloc) \triangleq \sqrt[\altnum]{\prod_{i=1}^\altnum (\sum_{v\in S_i^{\textrm{phone}}} \alloc_{i,v})}.$$

\textbf{Group-level negative entropy on onsite interview.} We add a single constraint on all groups' onsite interview stage: 
$$\AffineConstant \geq \ConvexFun(\alloc) \triangleq \alloc_{\ManSet}^{\textrm{onsite}} \log{(\alloc_{\ManSet}^{\textrm{onsite}})} + \alloc_{\WomanSet}^{\textrm{onsite}} \log{(\alloc_{\WomanSet}^{\textrm{onsite}})},
$$
where $\alloc_{\ManSet}^{\textrm{onsite}} \triangleq \sum_{i\in \ManSet,v\in S_i^{\textrm{onsite}}} \alloc_{i,v}$ and $\alloc_{\WomanSet}^{\textrm{onsite}} \triangleq \sum_{i\in \WomanSet,v\in S_i^{\textrm{onsite}}} \alloc_{i,v}$.

\textbf{Group-level generalized mean on offer.} We consider a single constraint on all groups' offer stage: 
$$\ConvexFun(\alloc) \triangleq \left(\frac{(\alloc_{\ManSet}^{\textrm{onsite}})^q + (\alloc_{\WomanSet}^{\textrm{onsite}})^q}{2}\right)^{1/q},$$
where $\alloc_{\ManSet}^{\textrm{onsite}} \triangleq \sum_{i\in \ManSet,v\in S_i^{\textrm{onsite}}} \alloc_{i,v}$ and $\alloc_{\WomanSet}^{\textrm{onsite}} \triangleq \sum_{i\in \WomanSet,v\in S_i^{\textrm{onsite}}} \alloc_{i,v}$. Note that for $q<1$, $\ConvexFun(\alloc)$ is a concave function---which includes geometric mean ($q=0$) and maximally egalitarian welfare ($q=-\infty$). Therefore, we would set a lower bound for it in the constraint to make sure enough fairness in the allocation of utilities.

Note that for $q>1$, $\ConvexFun(\alloc)$ is a convex function---which includes max functions ($q=\infty$)---and it is not playing the role of an equalitarian welfare anymore. However, such a function now somewhat captures the amount of ``unfairness'' in the utilities, that is, how far the utilities are from all being equal. Therefore, we would want to set an upper bound for it in order to encourage fairness. Our framework is general enough that can capture such constraints as well. 

Finally, note that one can always add a small quadratic function to the above convex/concave functions with the correct sign, so as to make the resulting function strongly convex/concave as well---which is a requirement we need in our near-optimal near-feasible approximation scheme. 


}
\section{A Premier on Fenchel Duality and its Implications}
\label{sec:convex}
In this supplemental section, we provide more details regarding Fenchel duality and provide a lemma that is crucial in our analysis in proof of \Cref{thm:RAI} in  \Cref{app:jms}.
\begin{definition}[\textbf{Fenchel Conjugate~\citep{bubeck2015convex}}]
\label{def:convex-conjugate}
Given a convex function $\ConvexFun:\mathbb{R}^{\dimension}\to \mathbb{R}$, the Fenchel conjugate function $\ConvexFun^*:\mathbb{R}^{\dimension}\to\mathbb{R}$ is defined as:
$$
\forall \DualConvex\in \mathbb{R}^{\dimension}: \ConvexFun^*(\DualConvex)\triangleq
\underset{\alloc\in\mathbb{R}^{\dimension}}{\sup}~\left(\DualConvex\cdot\alloc-\ConvexFun(\alloc)\right)
$$
\end{definition}

\begin{lemma}[\textbf{an adaptation of a similar lemma in \cite{bubeck2015convex}}]
\label{lemma:convex-conjugate}
Suppose (i) $\ConvexFun:\mathbb{R}^{\dimension}\to\mathbb{R}$ is strictly convex, (ii) admits continuous first partial derivatives, (iii)~$\lim_{\lVert\alloc\rVert\to \infty}\lVert \nabla\ConvexFun_j(\alloc)\rVert=+\infty$, and (iv)~there exists constants $\upperDualConvex,\lowerp>0$ such that $\lVert\GradFun(\alloc)\rVert_{\infty}\leq \upperDualConvex $ for every $\alloc\in[0,\barP]^{\dimension}$ and  $\forall \alloc: \lVert\alloc\rVert_{\infty} > \lowerp$ we have $\lVert\GradFun(\alloc)\rVert_{\infty}> \upperDualConvex$. Then we have:
\begin{enumerate}[label=(\Roman*)]
    \item The Fenchel conjugate function $F^*$ is strictly convex with continuous first partial derivatives.
    \item The conjugate of $\ConvexFun^*$ is the function $\ConvexFun$ itself, i.e. $(F^*)^*(\alloc)=\ConvexFun(\alloc)$ for all $\alloc\in\mathbb{R}^\dimension$.
    \item (envelop theorem) $\nabla \ConvexFun^*(\boldsymbol\mu)=\alloc^*(\DualConvex)$, where $\alloc^*(\DualConvex)\triangleq\underset{\alloc\in\mathbb{R}^{\dimension}}{\textrm{argmax}}~\left(\DualConvex\cdot\alloc-\ConvexFun(\alloc)\right)$, and $\nabla \ConvexFun(\alloc)=\DualConvex^*(\alloc)$, where $\DualConvex^*(\alloc)\triangleq\underset{\DualConvex\in\mathbb{R}^{\dimension}}{\textrm{argmax}}~\left(\DualConvex\cdot\alloc-\ConvexFun^*(\DualConvex)\right)$.
    \item The gradient map $\nabla\ConvexFun:\mathbb{R}^\dimension\to\mathbb{R}^\dimension$  is a bijection (i.e., an invertible and surjective map) and $(\nabla\ConvexFun)^{-1}=\nabla\ConvexFun^*$. Moreover, when the map is restricted to the domain $[0,\barP]^{\dimension}$, its image is a subset of $[-\upperDualConvex,\upperDualConvex]^{\dimension}$, and for any point $\boldsymbol\mu\in [-\upperDualConvex,\upperDualConvex]^{\dimension}$, $\nabla\ConvexFun^*(\boldsymbol\mu)\in[-\lowerp,\lowerp]^{\dimension}$.
\end{enumerate}
\end{lemma}
\begin{proof}{\emph{Proof.}}
Our assumptions (i), (ii) and (iii) guarantee that $\ConvexFun$ is a Legendre map/mirror map, and hence satisfies (I) and (II), and the first part of (IV). See Definition~1 and Lemma~1 in \cite{audibert2014regret}. (III) is a simple consequence of applying envelop theorem for high-dimensional differentiable functions, applied to $F$ and $F^*$. Finally, the second part of (IV) holds as $\lVert\GradFun(\alloc)\rVert_{\infty}\leq \upperDualConvex $ for every $\alloc\in[0,\barP]^{\dimension}$ due to (iv), and last part of (IV) holds as  if  $\lVert\GradFun(\alloc)\rVert_{\infty}\leq \upperDualConvex$ we have $\forall \alloc: \lVert\alloc\rVert_{\infty} \leq \lowerp$ due to (iv). 
\qed
\end{proof}




\revcolor{
\section{Optimal policy for a General JMS}
\label{app:JMS-general}
The first step in solving \eqref{eq:opt-fair} is solving the same problem with no ex-ante constraints, that is, finding a policy $\policy$ that maximizes $\expect{\reward_\policy}$. Importantly, we allow the rewards $\vecreward=[R_i(s)]_{i\in[\altnum],s\in\states_i}$  to take negative or positive rewards in the JMS instance, which proves to be crucial for incorporating ex-ante constraints, as we have already seen in \Cref{sec:pandora} and we will also see later when we define dual-adjusted rewards (\Cref{sec:RAI}) for JMS.  We sketch how to devise a polynomial-time algorithm for this problem, even in such an instance.


Past work studying the JMS problem with linear rewards characterize the optimal policy by either assuming negative rewards (i.e., costs) for intermediate states and only allowing positive rewards for the terminal states (see, e.g.,  \cite{dumitriu2003playing,gupta2019markovian}), or considering the more general so called \emph{No Free Lunch (NFL)} assumption on the state-reward structure of the Markov chains (see, e.g., \cite{gittins1979bandit,kleinberg2017tutorial}) and showing a similar analysis extends.

\begin{definition}[\textbf{NFL \citep{kleinberg2017tutorial}}]
\label{def:NFL}
An alternative $\MC$ satisfies NFL if for any state $s\in\states$ with $\reward(s)>0$, there exists a terminal state $t\in\terminals$ such that $\transition(s,t)>0$.
\end{definition}
Intuitively speaking, the NFL assumption implies that there shall be no opportunity to receive a positive reward from an intermediary state without risking a transition to a terminal state, thereby terminating the search.

Under NFL assumption, the earlier work established the optimality of the \emph{Gittins index policy}~\citep{gittins1979bandit,dumitriu2003playing}, which is a generalization of the optimal index-based policy of \citeauthor{weitzman1979optimal} for the Pandora's box problem: Given an instance $\left\{\MC_i\right\}_{i\in[\altnum]}$, there exists an index mapping $\sigma:\cup_{i\in[\altnum]}\states_i\rightarrow \mathbb{R}$ such that choosing the Markov chain $\MC_i$ with maximum $\reserve_i(s_i)$ to inspect given states $\{s_i\}_{i\in[\altnum]}$ at each time, until either $\capacity$ number of the Markov chains enter a terminal state or all remaining indices become non-positive (hence termination), is an optimal policy. 

For completeness, in the following, we revisit how the Gittins indices are defined in this more general model under NFL assumption. For each state $s$ in the MC $i$ (satisfying NFL), we define the $\reserve_i(s)$ as the smallest real number such that the following property holds: Consider a new JMS problem that only subsumes MC $i$, starting from state $s$, as well as another Markov chain that consists of only 2 states, both with zero rewards, an initial state $s^{'}$ and the terminal state $T^{'}$ with a transition probability of 1 from $s^{'}$ to $T^{'}$. Now if we subtract the amount $\reserve_i(s)$ from the rewards of all of the terminal states in MC $i$, then there exists no policy that can achieve positive expected reward for this new instance of JMS.

These amounts are, in fact, the Gittins indices of the corresponding states in the JMS instance. We highlight that these indices can be computed in polynomial time using backward induction, as shown in \cite{gittins1979bandit,kleinberg2017tutorial}. Consequently, they proved this theorem:

\begin{theorem}[Gittins Index Policy for JMS under NFL]
\label{thm:JMS:refinement}
The index-based policy, which at each time $t$ inspects the Markov chain whose $\reserve_i(s_i^t)$ is the highest across all Markov chains, until either $\capacity$ number of the Markov chains enter a terminal state or all remaining indices become non-positive (hence termination),  is an optimal policy for the JMS instance satisfying NFL assumption in \Cref{def:NFL}.
\end{theorem}

As we show in the remainder of this section, still a refinement of the Gittins index policy above (after proper pre-processing on the Markov chains) can solve the linear optimization over the space of randomized policies $\PolicySpace$ in polynomial-time for arbitrary reward vectors $\vecreward$, where this time $\{\reward_i(s)\}$ can be arbitrarily positive or negative. In fact, we show how to reduce the problem in polynomial-time to the special case satisfying NFL by introducing the idea of a \emph{collapsed instance}.

\begin{theorem}[\textbf{Optimal Policy for JMS with Arbitrary Rewards}]
\label{thm:JMS-optimal-reduction-informal}
Given any instance $\left\{\MC_i\right\}_{i\in[\altnum]}$ of the JMS problem, there exists a polynomial-time reduction that: (i) generates a new instance of the JMS problem satisfying NFL, called ``collapsed instance'', and (ii) by computing the Gittins indices of the collapsed instance, it returns a new set of indices $\sigma$ such that the index-based policy corresponding to $\sigma$ is optimal for the original instance of the JMS with arbitrary rewards.
\end{theorem}

In the following subsection, we elaborate on the above discussion and the statement of \Cref{thm:JMS-optimal-reduction-informal}. We then provide proof of \Cref{thm:JMS-optimal-reduction-informal}.

\subsection{Collapsing Reduction and Analysis of \texorpdfstring{\Cref{thm:JMS-optimal-reduction-informal}}{}}
In this subsection, we characterize the optimal policy for a JMS problem with general rewards, and thus we  prove \Cref{thm:JMS-optimal-reduction-informal}. To that end, we start by formally defining  ``free-lunch'', or ``FL'', states as follows: For a given MC, any state $s\in \states$ is FL iff it violates the NFL condition given in \Cref{def:NFL}, i.e., $\reward(s)>0$ and there does not exist a terminal state $t\in\terminals$ such that $\transition(s,t)>0$.

Given any problem instance $\left\{\MC_i\right\}_{i\in[\altnum]}$ that may also include some FL states, we now present a reduction, called ``Collapsing'', which results in a new JMS instance denoted by $\left\{\MC_i^\collapse\right\}_{i\in[\altnum]}$ where each $\MC_i^\collapse$ contains no FL state.

\begin{definition}[{\bf Collapsed MC and JMS}]
\label{def:Collapsing}
For any $\MC_i$ with general rewards, we construct its collapsed version, $\MC_i^\collapse$, by iteratively collapsing FL states until there remains no FL state in the resulting MC, $\MC_i^\collapse$. 
In particular, at each iteration, select a FL state, say $s$. Let $\states^p$ ({resp.} $\states^c$) be the set of all parents ({resp.} children) states of $s$ in the current MC, excluding $s$ itself.\footnote{Note that $s$ can have a self-loop and thus be its own parent and child}

\begin{enumerate}

    \item {\bf State elimination:} Remove state s from the current MC.
    
    \item {\bf Updating reward:} For any $s^p \in \states^p$, add $\frac{\reward_i(s)}{1-\transition_i(s,s)}$ to $\reward_i(s^p)$.
    
    \item {\bf Updating transition probabilities:} For any $s^p \in \states^p$ and $s^c \in \states^c$, add $\frac{\transition_i(s^p,s)\transition_i(s,s^c)}{1-\transition_i(s,s)}$ to $\transition_i(s^p,s^c)$.
    
\end{enumerate}

After completing this iterative process for each MC, we arrive at the JMS $\left\{\MC_i^\collapse\right\}_{i\in[\altnum]}$, the collapsed version of $\left\{\MC_i\right\}_{i\in[\altnum]}$. This process will end in at most $d$ iterations, as we are removing one state at each iteration.


\end{definition}

Based on this process, for any $\MC_i$, $i\in[\altnum]$, we define the set of {``non-collapsed''} states, denoted by $\nonCollapsedStates_i \subseteq \states_i$, as the set of all states of $\MC_i^\collapse$. 
We call states in $\states_i \setminus \nonCollapsedStates_i $, ``collapsed'' states.
With these definitions, we next establish an equivalence between stationary policies for $\left\{\MC_i^\collapse\right\}_{i\in[\altnum]}$ and the class of stationary ``efficient'' policies for the original JMS instance $\left\{\MC_i\right\}_{i\in[\altnum]}$, as defined below:

\begin{definition}[{\bf Efficient Policy}]
We call a policy ``efficient'', if it never terminates when (i) it has remaining capacity for selection and (ii) there exists at least one $\MC_i$ whose current state is in $\states_i \setminus \nonCollapsedStates_i$. In other words, as long as there exists 
Markov chains whose current states are collapsed states,
an efficient policy will always inspect one of such MCs as long as it has not run out of capacity $\capacity$ for slection.
\end{definition}

Note that there exists an optimal policy of $\left\{\MC_i\right\}_{i\in[\altnum]}$ that is efficient. 
To see why, first note that if there exists an MC in a collapsed state, it is always strictly better to inspect such an MC to accrue its positive expected reward before terminating. Next, notice that the order 
of inspecting MCs at collapsed states do not impact the expected reward,
because (i) all of them have to eventually be inspected, and (ii) inspecting a Markov chain at a collapsed state
will not result in terminating the search process, as it does not cause any of the MCs to go to a terminal state. We state the aforementioned equivalence in the following claim.

\begin{claim}
\label{clm:equivalency}
Consider any general instance of JMS, $\left\{\MC_i\right\}_{i\in[\altnum]}$, with starting states $(s^{(0)}_1,\ldots,s^{(0)}_\altnum)$ where $\forall i \in [\altnum]: s^{(0)}_i \in \nonCollapsedStates_i$. Then for any stationary efficient policy $\policy$ of $\left\{\MC_i\right\}_{i\in[\altnum]}$ with $(s^{(0)}_1,\ldots,s^{(0)}_\altnum)$, there exists a stationary policy for $\left\{\MC_i^\collapse\right\}_{i\in[\altnum]}$ with $(s^{(0)}_1,\ldots,s^{(0)}_\altnum)$, that achieves the same expected reward, and vice versa.
\end{claim}

\begin{proof}{\emph{Proof.}}
    Define $\policy^\collapse$ as the restriction of $\policy$ on only the non-collapsed states $\prod_{i\in [\altnum]} \nonCollapsedStates_i$. In other words, policy $\policy^\collapse$ for $\left\{\MC_i^\collapse\right\}_{i\in[\altnum]}$ at any state will make the exact same decision as  $\policy$ does in that state of the original JMS $\left\{\MC_i\right\}_{i\in[\altnum]}$. To see why the expected rewards under policy $\policy$ (for $\left\{\MC_i\right\}_{i\in[\altnum]}$)  and $\policy^\collapse$ 
    (for $\left\{\MC_i^\collapse\right\}_{i\in[\altnum]}$) are the same, note that if there is a MC in $\left\{\MC_i\right\}_{i\in[\altnum]}$  at a collapsed state, $\policy$ will inspect that (by definition of being efficient). Further, as noted above, the order of inspecting MCs at collapsed states does not impact the expected reward. As such, the expected reward accrued during inspection of MCs at collapsed states will be the same as the increase in the rewards of non-collapsed states determined in the reduction of  $\left\{\MC_i\right\}_{i\in[\altnum]}$ to $\left\{\MC_i^\collapse\right\}_{i\in[\altnum]}$ (as in \Cref{def:Collapsing}).
    For the reverse direction, we define $\policy$ for $\left\{\MC_i\right\}_{i\in[\altnum]}$ as the policy which makes the same decision as $\policy^\collapse$ does, if every $\MC_i$ is at a state  in $\nonCollapsedStates_i$. Otherwise, it will inspect a $\MC_i$ whose state is not in $\nonCollapsedStates_i$. By a similar line of reasoning, the expected reward under  the newly-constructed $\policy$ (for $\left\{\MC_i\right\}_{i\in[\altnum]}$) will be the same as that under $\policy^\collapse$ 
    (for $\left\{\MC_i^\collapse\right\}_{i\in[\altnum]}$).
    \qed
\end{proof}

Building on Claim \ref{clm:equivalency}, in the next claim we complete the proof of \Cref{thm:JMS-optimal-reduction-informal} by giving an optimal index-based policy for the original JMS, $\left\{\MC_i\right\}_{i\in[\altnum]}$.

\begin{claim}
\label{claim:OptimalIndexPolicy}
Let  $\reserve_i^\collapse(s)$, $\forall i\in [\altnum], \forall s\in \nonCollapsedStates_i$, be the Gittins indices defined in \cite{kleinberg2017tutorial,gupta2019markovian} for the JMS, $\left\{\MC_i^\collapse\right\}_{i\in[\altnum]}$, which satisfies the NFL condition. Then, the index-based (greedy) policy for selecting at most $\capacity$ number of MCs based on the following indices is an optimal policy for the original JMS, $\left\{\MC_i\right\}_{i\in[\altnum]}$.

\begin{equation}
\begin{aligned}
\reserve_i(s) \triangleq 
\begin{cases}
\reserve_i^\collapse(s) & s \in \nonCollapsedStates_i ~,\\ 
+ \infty & \textrm{o.w.}
\end{cases}
\end{aligned}
\end{equation}
\end{claim}

\begin{proof}{\emph{Proof.}}
To prove this claim, consider any stationary efficient optimal policy $\policy^*$ for $\left\{\MC_i\right\}_{i\in[\altnum]}$.\footnote{See \cite{dumitriu2003playing} for existence of an optimal stationary policy.} Since it is efficient, by Claim~\ref{clm:equivalency}, its equivalent ``collapsed'' policy $\policy^\collapse$ for $\left\{\MC_i^\collapse\right\}_{i\in[\altnum]}$ achieves the same expected reward. Now consider the index-based policy based on the $\reserve_i(s)$ introduced above; we call it $\widetilde{\policy}$. First, notice that this policy is also an efficient policy: by definition of indices, if there is a MC at a collapsed state (thus with index $+\infty$), then this policy will inspect such a MC.
Since $\widetilde{\policy}$ is stationary and efficient, again by Claim~\ref{clm:equivalency}  its equivalent ``collapsed'' policy $\widetilde{\policy}^\collapse$ for $\left\{\MC_i^\collapse\right\}_{i\in[\altnum]}$ achieves the same expected reward. Second, notice that $\widetilde{\policy}^\collapse$ was nothing but the optimal Gittins index policy for $\left\{\MC_i^\collapse\right\}_{i\in[\altnum]}$, implying that its expected reward cannot be less than $\policy^\collapse$. Hence, we can conclude that for any starting states $(s^{(0)}_1,\ldots,s^{(0)}_\altnum)$, where $\forall i \in [\altnum]: s^{(0)}_i \in \nonCollapsedStates_i$, the expected reward of $\widetilde{\policy}$ is at least that of $\policy^*$. Finally, suppose there are some Markov chains whose starting states are collapsed states. Then, since 
both $\widetilde{\policy}$ and $\policy^*$ are efficient, both will inspect those Markov chains until they reach a state $(s_1,\ldots,s_\altnum)$, where $\forall i \in [\altnum]: s_i \in \nonCollapsedStates_i$.
As a result, from any starting state the expected reward of $\widetilde{\policy}$ would be at least that of $\policy^*$, implying that $\widetilde{\policy}$ is also an optimal policy. This will conclude the proof of this claim.
\qed
\end{proof}


In the last part of this section, for the sake of completeness, we restate the definition of the Gittins indices, $\reserve_i^\collapse(s)$, for the collapsed JMS $\left\{\MC_i^\collapse\right\}_{i\in[\altnum]}$, which satisfied the NFL assumption \citep{dumitriu2003playing,kleinberg2017tutorial}.

\begin{definition}[{\bf Gittins indices of $\left\{\MC_i^\collapse\right\}_{i\in[\altnum]}$}]
For any Markov chain $\MC_i^\collapse$ and state $s\in \nonCollapsedStates_i$, we define the $\reserve_i^\collapse(s)$ as the smallest real number such that this property holds: \emph{Consider a new JMS problem with only two Markov chains $\left\{\widetilde{\MC}_j\right\}_{j\in[2]}$, where $\widetilde{\MC}_1\triangleq \MC_i^\collapse$, and $\widetilde{\MC}_2\triangleq (\states=\left\{ 
s',t' \right\},\terminals=\left\{ 
t' \right\},\transition(s_1, s_2)= \indicator{s_2 = t'}
 ,\reward=(0,0))$.
 Now if we subtract $\reserve_i^\collapse(s)$ from the rewards of all of the terminal states in $\widetilde{\MC}_1$, then there exists no policy that can achieve positive expected reward for this new instance of JMS.}


\end{definition}

We conclude by noting that we can establish the polynomial-time computability of the Gittins indices defined above by a simple adjustment to the polynomial-time algorithm given in \cite{dumitriu2003playing}. We omit the details for the sake of brevity.
}

\if false 

\MA{notations, please skip the following. Just read the above!}

\MA{old version}

For now, suppose that non of our Markov Chains contains a positive reward state that also has a transition to a terminal state (these states are named ``Free Lunch'' or ``FL'' states). We call these subset of all JMS problems as ``Basic JMS'' problems.

\cite{dumitriu2003playing} and \cite{kleinberg2017tutorial} have provided an optimal policy to solve Basic JMS, which is revisited in the following. 

For each state $s$ in the Basic MC $i$, we define the $\reserve_i(s)$ as the smallest real number such that the following property holds:

Consider a new JMS problem that only subsumes MC $i$, starting from state $s$, as well as another Markov chain that consists of only 2 states, both with zero rewards, an initial state $s^{'}$ and the terminal state $T^{'}$ with a transition probability of 1 from $s^{'}$ to $T^{'}$. Now if we subtract the amount $\reserve_i(s)$ from the rewards of all of the terminal states in MC $i$, then there exists no policy that can achieve positive expected reward for this new instance of JMS.

These amounts are, in fact, the Gittins indices of the corresponding states in the Basic JMS. Consequently, the aforementioned references have proved the following theorem:

\begin{theorem}
\label{thm:JMS:refinement}
The index-based policy, which at each time $t$ inspects the Markov chain whose $\reserve_i(s_i^t)$ is the highest across all Markov chains, is an optimal policy for the Basic JMS.
\end{theorem}

Now that we have revisited the literature, we will state our reduction result on how we can get an optimal policy for a general JMS, using the above result for Basic JMS.

For any general JMS, we will define its corresponding ``Collapsed'' JMS as the following. Notice that the Collapsed JMS is indeed a Basic JMS. 

\begin{definition}[Collapsed MC and JMS]
\label{def:Collapsed MC, JMS, policy}
The Collapsed MC of a MC with general rewards will be derived by iteratively collapsing the FL nodes of the MC in the following manner until no FL node remains in the final MC:

    $\bullet$ At each iteration $t$, we select one of the FL nodes, say $i$, and then connect all of its parents, say $p$, to all of its children, say $c$, with probability $A(p,i)A(i,c)$, if the edge already exists, we will add this number to its transition probability. Furthermore, for any of its parents $p$, we will also add $R(i)A(p,i)$, i.e. the effective added reward to $p$ by $i$, to $R(p)$. Lastly, we just delete the node $i$ from the MC, and the resulting MC will be the MC at the next iteration.

Notice that this algorithm will finish in less than $\dimension$ iteration, since we are collapsing one state at each iteration.
Additionally, if we collapse all of the MCs in our JMS, we will get the ``Collapsed'' JMS version of the original JMS.
\end{definition}

In the end, we will state our theorem for this reduction.

\begin{theorem}
\label{theorem:generalJMSoptPolicy}
Given an instance $\left\{\MC_i\right\}_{i\in[\altnum]}$ of the JMS problem, and its corresponding collapsed instance (as in Definition~\ref{xxx}), consider the following assignment $\sigma:\cup_{i\in[\altnum]}\states_i\rightarrow \mathbb{R}$: (i) for any collapsed state $s$, $\reserve(s)$ is set to $+\infty$ and (ii) for any non-collapsed state $s$, $\reserve(s)$ is set to the Gittins index of state $s$ in $\left\{\MC_i\right\}_{i\in[\altnum]}$ (as in Definition~\ref{xxx}). Then the resulting index-based policy with indices $\sigma$, i.e., selecting Markov chain $\MC_i$ with maximum $\reserve_i(s_i)$ given states $\{s_i\}_{i\in[\altnum]}$ , is optimal.
\end{theorem}

\begin{proof}{Proof.}
Please\RN{Hmmm, not sure if typically we use please!} \MA{;D} read \Cref{subsec:JMSOptimalpolicy} for the proof of the theorem.
\end{proof}

\MA{above from main text}

\begin{definition}[Efficient Policy]
\label{def:Efficient Policy}
We call a policy Efficient, if it will never terminate when the state of one of its MCs is a FL node. In other words, an Efficient policy always makes sure it inspects an alternative enough times such that its state becomes a non-FL node, before inspecting another alternative whose state has a possibility of going to a terminal state after inspection.
\end{definition}

Now please note the following observations: 

\begin{enumerate}
    \item Any optimal policy for the original JMS has to be an efficient policy.
    \item For any Efficient policy for the original JMS there is an equivalent policy for the Collapsed JMS that has the same expected reward (provided that they are both starting from same initial state, which has to be a non-collapsed state for each MC), and vice versa.
\end{enumerate}

To find the equivalent policy we define the following two process on policies. First, for any Efficient policy for the original JMS, the equivalent Collapsed policy for the Collapsed JMS, over the collapsed states, will make exactly the same decision as the Efficient policy. Second, for any policy for the Collapsed JMS, the equivalent DeCollapsed policy for the original JMS will always choose a MC whose state is a collapsed node, if none such MC existed, then all MC are at non-collapsed nodes. In this situation, it will do exactly what the policy for Collapsed JMS was doing. 

As a result of these observations, in order to get an optimal policy for the original JMS, we can first find an optimal policy for the Collapesd JMS, and then ``DeCollapse'' it to get its equivalent policy for the original JMS. As a consequence of the aforementioned two observations, this DeCollapsed policy has to be an optimal policy for the original JMS. Furthermore, notice that this DeCollapsed policy is exactly doing the same thing as the index policy stated in \Cref{theorem:generalJMSoptPolicy}.

\MA{please somewhere mention that the gittins indices for the Basic JMS can be computed in poly time with the algorithm introduced in \cite{dumitriu2003playing}}

\fi 
\section{Missing Proofs of \texorpdfstring{\Cref{sec:general}}{}}
\label{app:jms}

\begin{proof}{\emph{Proof of \Cref{thm:RAI}}.}
\revcolor{Let the parameters be chosen as (i) $\upperDualAffine=\upperDualConstConvex=\frac{d\barP+\varepsilon}{\delta}$, (ii) $\InnerNum=\left(\frac{6\dimension\upperDualConvex(\barP+\lowerp)\upperDualConstConvex\NumConvex}{\varepsilon}\right)^2=\mathcal{O}\left(\frac{1}{\delta^2\epsilon^2}\right)$ and $\OuterNum=\left(\frac{3\max\left(2\upperDualAffine\NumAffine,\upperDualConstConvex\NumConvex\right)}{\varepsilon}\right)^2=\mathcal{O}\left(\frac{1}{\delta^2\epsilon^2}\right)$, and (iii) $\InRate=\frac{2\upperDualConvex}{(\barP+\lowerp)\upperDualConstConvex}\frac{1}{\sqrt{\InnerNum}}$, $\OutRatelambda=\frac{\upperDualAffine}{2\sqrt{\OuterNum}}$, and $\OutRatebeta=\frac{\upperDualConstConvex}{\sqrt{\OuterNum}}$.}
We start by considering the inner-loop of the G-RDIP policy (\Cref{alg:RAI}). Fix an iteration $(m,\ell)$ of the inner loop. For any $i\in[\NumConvex]$ and $\DualConvex_i\in[-\upperDualConvex,\upperDualConvex]^{\dimension}$ we have:
\begin{align}
\DualConstConvexk_i(\DualConvexkl_i-\DualConvex_i)&\cdot \left(\nabla\ConvexFun^*_i(\DualConvexkl_i)-\alloc_{\policy^{(m,\ell)}}\right)=\frac{1}{\InRate}(\DualConvexkl_i-\DualConvex_i)\cdot (\DualConvexkl_i-\DualConvexTempkl_i)\nonumber\\
&\overset{(1)}{=}\frac{1}{2\InRate}\left(\normtwo{\DualConvexkl_i-\DualConvex_i}^2+\normtwo{\DualConvexkl_i-\DualConvexTempkl_i}^2-\normtwo{\DualConvexTempkl_i-\DualConvex_i}^2\right)\nonumber\\
&\overset{(2)}{\leq} \frac{1}{2\InRate}\left(\normtwo{\DualConvexkl_i-\DualConvex_i}^2+\normtwo{\DualConvexkl_i-\DualConvexTempkl_i}^2-\normtwo{\DualConvex^{(m,\ell+1)}_i-\DualConvex_i}^2\right)\nonumber\\
\label{eq:GDsummation1}
&\overset{(3)}{\leq} \frac{1}{2\InRate}\left(\normtwo{\DualConvexkl_i-\DualConvex_i}^2-\normtwo{\DualConvex^{(m,\ell+1)}_i-\DualConvex_i}^2+\dimension \InRate^2\upperDualConstConvex^2(\barP+\lowerp)^2\right),
\end{align}
where equality~$(1)$ holds due to the Pythagorean's lemma, inequality~$(2)$ holds by the fact that $\DualConvex^{(m,\ell+1)}_i$ is the projection of $\DualConvexTempkl_i$ onto $[-\upperDualConvex,\upperDualConvex]^{\dimension}$, and inequality~$(3)$ holds as $\lVert\alloc_{\policy^{(m,\ell)}}-\nabla\ConvexFun_i^*(\DualConvexkl_i)\rVert_{\infty}\leq (\barP+\lowerp)$, sbecause $\lVert\nabla\ConvexFun_i^*(\DualConvexkl_i)\rVert_{\infty} \leq \lowerp$ for all $\DualConvexkl_i\in[-\upperDualConvex,\upperDualConvex]^\dimension$ by applying \Cref{lemma:convex-conjugate}. By averaging both hand sides of \eqref{eq:GDsummation1} over $\ell \in [\InnerNum]$, rearranging the terms, and finally setting  $\InRate=\frac{2\upperDualConvex}{(\barP+\lowerp)\upperDualConstConvex}\frac{1}{\sqrt{\InnerNum}}$ we have:
\begin{align}
\label{eq:GDsummation2}
    \frac{1}{\InnerNum}\sum_{\ell \in [\InnerNum]}\DualConstConvexk_i(\DualConvexkl_i-\DualConvex_i)\cdot \left(\nabla\ConvexFun^*_i(\DualConvexkl_i)-\alloc_{\policy^{(m,\ell)}}\right)&\leq \frac{1}{2\InnerNum\InRate}\normtwo{\DualConvex_i^{(m,1)}-\DualConvex_i}^2-\frac{1}{2\InnerNum\InRate}\normtwo{\DualConvex_i^{(m,\InnerNum+1)}-\DualConvex_i}^2\nonumber\\
    &~~~~~~~~~~~~~~~~~~~~~~~~~~~~+\frac{\InRate}{2}\dimension(\barP+\lowerp)^2\upperDualConstConvex^2\nonumber\\
    &\leq \frac{2\dimension\upperDualConvex^2}{\InnerNum\InRate}+\frac{\InRate}{2}\dimension(\barP+\lowerp)^2\upperDualConstConvex^2=\frac{2\dimension\upperDualConvex(\barP+\lowerp)\upperDualConstConvex}{\sqrt{\InnerNum}}
\end{align}

Now, denote by $\BarDualConvex_i^{(m)}$ the average of $\DualConvexkl_i$ during the outer iteration $k$, i.e., $ \BarDualConvex_i^{(m)}=\frac{1}{\InnerNum}\sum_{\ell \in [\InnerNum]}\DualConvexkl_i$ and denote by $\Baralloc^{(m)}$ the average expected visit numbers vector of $\policy^{(m,\ell)}$ during the outer iteration $k$, i.e., $\Baralloc^{(m)}=\frac{1}{\InnerNum}\sum_{\ell \in [\InnerNum]}\alloc_{\policy^{(m,\ell)}}$. Using these notations, we can further  inequality~\eqref{eq:GDsummation2}. To do so, by incorporating the convexity of the conjugate function $\ConvexFun_i^*$ (\Cref{lemma:convex-conjugate}) we have:
\begin{align}
  \label{eq:GD-summation3}
     \frac{1}{\InnerNum}\sum_{\ell\in[\InnerNum]}\DualConstConvexk_i(\DualConvexkl_i-\DualConvex_i)\cdot \nabla\ConvexFun^*_i(\DualConvexkl_i)&\geq \frac{1}{\InnerNum}\sum_{\ell\in[\InnerNum]}\DualConstConvexk_i\left(\ConvexFun^*_i(\DualConvexkl_i)-\ConvexFun^*_i(\DualConvex_i)\right)\nonumber\\
     &\geq \DualConstConvexk_i\left(\ConvexFun^*_i(\BarDualConvex_i^{(m)})-\ConvexFun^*_i(\DualConvex_i)\right)
\end{align}
Also, we have:
\begin{equation}
    \label{eq:GD-summation4}
     \frac{1}{\InnerNum}\sum_{\ell\in[\InnerNum]}\DualConstConvexk_i(\DualConvexkl_i-\DualConvex_i)\cdot \alloc_{\policy^{(m,\ell)}}=\frac{1}{\InnerNum}\sum_{\ell\in[\InnerNum]}\DualConstConvexk_i\DualConvexkl_i\cdot \alloc_{\policy^{(m,\ell)}}-\DualConstConvexk_i\DualConvex_i\cdot\Baralloc^{(m)}
\end{equation}
By combining \eqref{eq:GDsummation2}, \eqref{eq:GD-summation3}, and \eqref{eq:GD-summation4}, and summing up the terms for all $i\in[\NumConvex]$, we have: 
\begin{align}
\label{eq:dual-Best-Response-Inner}
&\sum_{i\in[\NumConvex]} \DualConstConvexk_i\left(\ConvexFun^*_i(\BarDualConvex_i^{(m)})-\ConvexFun^*_i(\DualConvex_i)\right) + \sum_{i\in[\NumConvex]}\DualConstConvexk_i\DualConvex_i\cdot\Baralloc^{(m)}-\sum_{i\in[\NumConvex]}\frac{1}{\InnerNum}\sum_{\ell \in [\InnerNum]}\DualConstConvexk_i\DualConvexkl_i\cdot\alloc_{\policy^{(m,\ell)}}\nonumber\\
\tag{\textsc{Dual-Best-Response-Inner}}
\leq &\frac{2\dimension\upperDualConvex(\barP+\lowerp)\upperDualConstConvex\NumConvex}{\sqrt{\InnerNum}}~,
\end{align}
for any $\DualConvex_i\in[-\upperDualConvex,\upperDualConvex]^{\dimension}$. On the other hand, our algorithm also selects the optimal policy $\policy^{(m,\ell)}\in\PolicySpace$ for the adjusted rewards $\AdjustedReward^{(m,\ell)}=\vecreward-\sum_{j\in[\NumAffine]}\DualConstAffine_j^{(m)}\AffineVector_j-\sum_{i\in[\NumConvex]}\DualConstConvex_i^{(m)}\DualConvex_i^{(m,\ell)}$ at any iteration $(m,\ell)$. Hence:
\begin{equation}
\label{eq:inner-primal-bestresponse}
   \frac{1}{\InnerNum}\sum_{\ell\in [\InnerNum]}\AdjustedReward^{(m,\ell)}\cdot\alloc\leq \frac{1}{\InnerNum}\sum_{\ell\in [\InnerNum]}\AdjustedReward^{(m,\ell)}\cdot\alloc_{\policy^{(m,\ell)}}~,
\end{equation}
\vspace{-2mm}
or equivalently:
\begin{equation}
\label{eq:primal-best-response-inner}
\tag{\textsc{Primal-Best-Response-Inner}}
   \vecreward\cdot\alloc -\sum_{j\in[\NumAffine]}\DualConstAffine_j^{(m)}\AffineVector_j\cdot\alloc-\sum_{i\in[\NumConvex]}\DualConstConvex_i^{(m)}\BarDualConvex_i^{(m)}\cdot\alloc
  \leq \vecreward\cdot\Baralloc^{(m)} -\sum_{j\in[\NumAffine]}\DualConstAffine_j^{(m)}\AffineVector_j\cdot\Baralloc^{(m)}- \sum_{i\in[\NumConvex]}\frac{1}{\InnerNum}\sum_{\ell\in[\InnerNum]}\DualConstConvexk_i\DualConvexkl_i\cdot\alloc_{\policy^{(m,\ell)}}~,
\end{equation}
where the above inequalities hold for any $\alloc\in\FeasibleAlloc$. Finally, by combining the  two inequalities in  \eqref{eq:dual-Best-Response-Inner} and ~\eqref{eq:primal-best-response-inner}, and rearranging the terms, for any set of vectors $\alloc\in\FeasibleAlloc$ and $ \DualConvex_i\in[-\upperDualConvex,\upperDualConvex]^{\dimension}$, and for $i\in[\NumConvex], m\in [\OuterNum]$ we have: 
\vspace{-2mm}
\begin{myboxx}{shadecolor}{}
\vspace{-4mm}
\begin{align}
    \label{eq:inner-final}
    \tag{\textsc{Inner-Approximate-Equilibrium}}
    \barLagrange\left(\alloc;\{\DualConstAffinek_j\},\{\DualConstConvexk_i\},\{\BarDualConvex^{(m)}_i\}\right)-\barLagrange\left(\Baralloc^{(m)};\{\DualConstAffinek_j\},\{\DualConstConvexk_i\},\{\DualConvex_i\}\right)\leq \frac{2\dimension\upperDualConvex(\barP+\lowerp)\upperDualConstConvex\NumConvex}{\sqrt{\InnerNum}} \equiv \varepsilon_1.
\end{align}
\vspace{-4mm}
\end{myboxx}
\vspace{-4mm}

Next, we look at the outer loop of G-RDIP policy. Fix an iteration $k$ of the outer loop, and consider the way $\DualConstConvexk_i$ and $\DualConstAffinek_j$ are updated in this iteration for each $i\in[\NumConvex]$ and $j\in[\NumAffine]$. First, define $\forall i\in[\NumConvex], ~\hat{\beta}_i^{k}:=\DualConstConvexk_i+\OutRatebeta\ConvexFun_i(\Baralloc^{(m)})$. Then for any $i\in[\NumConvex]$ and $\DualConstConvex_i\in[0,\upperDualConstConvex]$ we have:
\begin{align}
    -\left(\DualConstConvexk_i-\DualConstConvex_i\right)\ConvexFun_i\left(\Baralloc^{(m)}\right)&=\frac{1}{\OutRatebeta}\left(\DualConstConvexk_i-\DualConstConvex_i\right)\left(\DualConstConvexk_i-\hat{\DualConstConvex}_i^{(m)}\right)\nonumber\\
    &=\frac{1}{2\OutRatebeta}\left(\left(\DualConstConvexk_i-\DualConstConvex_i\right)^2+\left(\DualConstConvexk_i-\hat\DualConstConvex_i^{(m)}\right)^2-\left(\hat\DualConstConvex^{(m)}_i-\DualConstConvex_i\right)^2\right)\nonumber\\
    \label{eq:outer-loop-beta}
    &\leq \frac{1}{2\OutRatebeta}\left(\left(\DualConstConvexk_i-\DualConstConvex_i\right)^2-\left(\DualConstConvex^{(m+1)}_i-\DualConstConvex_i\right)^2+\OutRatebeta^2\right)~.
\end{align}

By averaging both hand sides of \eqref{eq:outer-loop-beta} for $m\in [\OuterNum]$, rearranging the terms, and finally setting $\OutRatebeta=\frac{\upperDualConstConvex}{\sqrt{\OuterNum}}$ we obtain the following inequality:
\begin{align}
     -\frac{1}{\OuterNum}\sum_{m\in[\OuterNum]}\left(\DualConstConvexk_i-\DualConstConvex_i\right)\ConvexFun_i\left(\Baralloc^{(m)}\right)&\leq \frac{1}{2\OuterNum\OutRatebeta}\left(\DualConstConvex_i^{(1)}-\DualConstConvex_i\right)^2+\frac{\OutRatebeta}{2}
     \leq\frac{\upperDualConstConvex^2}{2\OuterNum\OutRatebeta}+\frac{\OutRatebeta}{2}=\frac{\upperDualConstConvex}{\sqrt{\OuterNum}}~.\label{eq:beta-GD1}
\end{align}
Similarly, by considering the update equation of $\DualConstAffinek_j$ for any $j\in[\NumAffine]$, following exactly the same lines as in the above argument, and finally setting $\OutRatelambda=\frac{\upperDualAffine}{2\sqrt{\OuterNum}}$, for any $\DualConstAffine_j\in[0,\upperDualAffine]$ we have:
\begin{align}
     \frac{1}{\OuterNum}\sum_{m\in[\OuterNum]}\left(\DualConstAffinek_j-\DualConstAffine_j\right)\left(\AffineConstant_j-\AffineVector_j\cdot\Baralloc^{(m)}\right)&\leq \frac{1}{2\OuterNum\OutRatelambda}\left(\DualConstAffine_j^{(1)}-\DualConstAffine_j\right)^2+2\OutRatelambda
     \leq\frac{\upperDualAffine^2}{2\OuterNum\OutRatelambda}+2\OutRatelambda=\frac{2\upperDualAffine}{\sqrt{\OuterNum}}~ \label{eq:lambda-GD1}
\end{align}
where in the first inequality we have used the assumption that $\lvert\AffineConstant_j\rvert\leq 1$, and that for every $\alloc\in[0,\barP]^{\dimension}$ we have $\lvert\AffineVector_j\cdot\alloc\rvert\leq 1$. 
Now, $\forall i \in [\NumConvex]$ denote by $\BarDualConstConvex_i$ the average of $\DualConstConvexk_i$ over all outer iterations, i.e., $ \BarDualConstConvex_i=\frac{1}{\OuterNum}\sum_{m\in[\OuterNum]}\DualConstConvexk_i$, and $\forall j \in[\NumAffine]$ denote by $\BarDualConstAffine_j$ the average of $\DualConstAffinek_j$ over all outer iterations, i.e., $ \BarDualConstAffine_j=\frac{1}{\OuterNum}\sum_{m\in[\OuterNum]}\DualConstAffinek_j$. Also, denote by $\Baralloc$ the average of all vectors of expected visit numbers across all iterations of our algorithm, i.e., $\Baralloc=\frac{1}{\OuterNum}\sum_{m\in[\OuterNum]}\Baralloc^{(m)}\equiv \frac{1}{\OuterNum\InnerNum}\sum_{m\in[\OuterNum]}\sum_{\ell\in[\InnerNum]}\alloc_{\policy^{(m,\ell)}}$. Note that due to the convexity of $\ConvexFun_i$, $\ConvexFun_i(\Baralloc)\leq \frac{1}{\OuterNum}\sum_{m\in[\OuterNum]}\ConvexFun_i(\Baralloc^{(m)})$. Using this fact, and by summing up both hand sides of \eqref{eq:beta-GD1} for $i\in[\NumConvex]$, we obtain this inequality for any choice of $\DualConstConvex_i\in[0,\upperDualConstConvex]$ and for $i\in[\NumConvex]$: 
\vspace{-2mm}
\begin{myboxx}{shadecolor}{}
\vspace{-4mm}
\begin{equation}
    \label{eq:dual-beta-bestresponse}
    \tag{\textsc{Dual-Best-Response-Outer-I}}
    \sum_{i\in[\NumConvex]}\DualConstConvex_i\ConvexFun_i(\Baralloc)-\sum_{i\in[\NumConvex]}\frac{1}{\OuterNum}\sum_{m\in[\OuterNum]}\DualConstConvexk_i\ConvexFun_i(\Baralloc^{(m)})\leq \frac{\upperDualConstConvex\NumConvex}{\sqrt{\OuterNum}}\equiv\varepsilon_2~.
    \vspace{-2mm}
\end{equation}
\end{myboxx}
\vspace{-2mm}
Similarly, by summing up both hand sides of \eqref{eq:lambda-GD1} for $j\in[\NumAffine]$, we obtain the following inequality for any choice of $\DualConstAffine_j\in[0,\upperDualAffine]$ and for $j\in[\NumAffine]$: 
\vspace{-2mm}
\begin{myboxx}{shadecolor}{}
\vspace{-0mm}
\begin{equation}
    \label{eq:dual-lambda-bestresponse}
    \tag{\textsc{Dual-Best-Response-Outer-II}}
    \sum_{j\in[\NumAffine]}\frac{1}{\OuterNum}\sum_{m\in[\OuterNum]}\DualConstAffinek_j\left(\AffineConstant_j-\AffineVector_j\cdot\Baralloc^{(m)}\right)-\sum_{j\in[\NumAffine]}\DualConstAffine_j\left(\AffineConstant_j-\AffineVector_j\cdot\Baralloc\right)\leq \frac{2\upperDualAffine\NumAffine}{\sqrt{\OuterNum}}\equiv\varepsilon_3~.
\end{equation}
\vspace{-4mm}
\end{myboxx}
\vspace{-2mm}
To put all the pieces together and obtain the final result, first note that at any iteration $k$, one can consider the assignment $\forall i \in [\NumConvex], ~ \DualConvex_i \leftarrow \nabla\ConvexFun_i(\Baralloc^{(m)})\in[-\upperDualConvex,\upperDualConvex]^{\dimension},$ in \eqref{eq:inner-final}. Due to \Cref{lemma:convex-conjugate}, 
$\ConvexFun_i(\Baralloc^{k})=\DualConvex_i\cdot\Baralloc^{(m)}-\ConvexFun_i^*(\DualConvex_i)$, and therefore:
$$
\barLagrange\left(\Baralloc^{(m)};\{\DualConstAffinek_j\},\{\DualConstConvexk_i\},\{\nabla\ConvexFun_i(\Baralloc^{(m)})\}\right)=\Lagrange\left(\Baralloc^{(m)};\{\DualConstAffinek_j\},\{\DualConstConvexk_i\}\right)~.
$$

Therefore, we obtain the following inequality
\begin{equation}
\label{eq:inner-final-relaxed}
\forall m\in[\OuterNum],~\alloc\in\FeasibleAlloc:~~ \barLagrange\left(\alloc;\{\DualConstAffinek_j\},\{\DualConstConvexk_i\},\{\BarDualConvex^{(m)}_i\}\right)-\Lagrange\left(\Baralloc^{(m)};\{\DualConstAffinek_j\},\{\DualConstConvexk_i\}\right)\leq \varepsilon_1
\end{equation}

Recall the definition of the optimal fair policy $\OptFairPolicy$ in \eqref{eq:opt-fair}. Such a policy exists as the Markovian game instance $\{\MC_i\}_{i\in[\altnum]}$ is feasible (Assumption~\ref{assumption:feasible-fair}). Let $\OptFair\triangleq \expect{R_{\OptFairPolicy}}=\vecreward\cdot\alloc_{\OptFairPolicy}$. By setting $\alloc=\alloc_{\OptFairPolicy}$ in inequality~\eqref{eq:inner-final-relaxed}, and using the fact that $\barLagrange$ is a relaxation of the optimal policy for any \emph{feasible} choice of dual variables, i.e., $\DualConstAffineVec,\DualConstConvexVec\geq 0$, and $\forall i\in[\NumConvex]: ~\DualConvex_i \in \mathbb{R}^{\dimension}$, we have:
\begin{equation}
\label{eq:inner-final-relaxed-2} 
\forall m\in[\OuterNum]:~~ \OptFair-\Lagrange\left(\Baralloc^{(m)};\{\DualConstAffinek_j\},\{\DualConstConvexk_i\}\right)\leq \varepsilon_1
\end{equation}
Now, by averaging over all iterations $m\in [\OuterNum]$ in \eqref{eq:inner-final-relaxed-2}, we obtain the following inequality:
\begin{equation}
    \label{eq:inner-final-relaxed-3}
    \OptFair-\vecreward\cdot\Baralloc-\sum_{j\in[\NumAffine]}\frac{1}{\OuterNum}\sum_{m\in [\OuterNum]}\DualConstAffinek_j\left(\AffineConstant_j-\AffineVector_j\cdot\Baralloc^{(m)}\right)+\sum_{i\in[\NumConvex]}\frac{1}{\OuterNum}\sum_{m\in[\OuterNum]}\DualConstConvexk_i\ConvexFun_i(\Baralloc^{(m)})  \leq \varepsilon_1
\end{equation}

To conclude, we add up both hand sides of inequalities ~\eqref{eq:inner-final-relaxed-3}, \eqref{eq:dual-beta-bestresponse}, and \eqref{eq:dual-lambda-bestresponse}, so that we obtain the following final inequality (which holds for \emph{any assignment} of $\DualConstAffine_j\in[0,\upperDualAffine],j\in[\NumAffine]$ and $\DualConstConvex_i\in[0,\upperDualConstConvex],i\in[\NumConvex]$):
\begin{equation}
    \label{eq:final}
    \OptFair-\vecreward\cdot\Baralloc-\sum_{j\in[\NumAffine]}\DualConstAffine_j\left(\AffineConstant_j-\AffineVector_j\cdot\Baralloc\right)+\sum_{i\in[\NumConvex]}\DualConstConvex_i\ConvexFun_i(\Baralloc)  \leq \varepsilon_1+\varepsilon_2+\varepsilon_3\leq\varepsilon~,
\end{equation}
Now, by setting $\DualConstAffine_j=0$ for all $j\in[\NumAffine]$ and $\DualConstConvex_i=0$ for all $i\in[\NumConvex]$, the expected reward objective of the G-RDIP policy $\hat{\policy}$ (returned by \Cref{alg:RAI}) is bounded below by:
$$
\expect{R_{\hat{\policy}}}=\frac{1}{\OuterNum\InnerNum}\sum_{m\in[\OuterNum]}\sum_{\ell\in[\InnerNum]}\vecreward\cdot \alloc_{\policy^{(m,\ell)}}=\vecreward\cdot\Baralloc\geq \OptFair -\varepsilon
$$
At the same time, notice that $\lvert\OptFair-\vecreward\cdot\Baralloc\rvert\leq d\barP$. Hence, for all $j\in[\NumAffine]$ we should have that $\AffineConstant_j-\AffineVector_j\cdot\Baralloc\geq -\delta$, because if the converse holds for some $j$ then we can set $\DualConstAffine_j=\frac{\dimension\barP+\varepsilon}{\delta}\leq\upperDualAffine$ (and all other $\DualConstAffine_{j'}$'s and $\DualConstConvex_{i}$'s are set to zero), which violates \eqref{eq:final}. Similarly, for all $i\in[\NumConvex]$ we should have $\ConvexFun_i(\Baralloc)\leq \delta$, because  if the converse holds for some $i$ then we can set $\DualConstConvex_i=\frac{\dimension\barP+\varepsilon}{\delta}\leq\upperDualConstConvex$ (and all other $\DualConstConvex_{i'}$'s and $\DualConstAffine_{j}$'s are set to zero), which violates \eqref{eq:final}. This completes the proof of the first part of the theorem.

\revcolor{Regarding running time, $K_I=\mathcal{O}(d^4)$ and $ K_O=\mathcal{O}(d^2)$, and therefore the total number of iterations of our algorithm is $\mathcal{O}(K_I K_O)= \mathcal{O}(d^6)$. We note that our algorithm needs to solve a JMS instance for the primal player and compute the gradient for the dual player in each iteration. In general, solving an instance of JMS requires an extra polynomial-time computation. This extra computation depends on two factors: (i) the amount of time it takes to compute the indices for each individual arm---which is polynomial-time; if the Markov chain $i$ has $\ell_i$ number of nodes/states, the running time is at most $\mathcal{O}\left(\ell_i^5\right)$ to solve for the indices using dynamic programming~\citep{dumitriu2003playing} (note that $\sum_i \ell_i=d$). (ii) the amount of time it takes to compute the gradient of $\bar{\mathcal{L}}_\textsc{JMS-const}$, which requires computing the expected number of visits of different states under the optimal index-based policy (computed earlier). The latter quantity depends on the absorption time of the underlying kernels of the Markov chains, which is $\mathcal{O}(d)$, as it is assumed that the Markov chain is finite and absorbing, as there is a constant $H_p$ such that the expected number of visits of each state before absorption is bounded above by $H_p$. We note that the eventual running time will be polynomial in $d$, $n$, $\frac{1}{\varepsilon}$ and $\frac{1}{\delta}$, as desired.
}
\qed
\end{proof}

\section{Supplemental Numerical Simulations}
\label{apx:numerical-main}

\revcolorm{
In this section, we provide supplementary materials for the numerical study in \Cref{sec:numerical}. In particular, we include the following discussions and simulations:
\begin{itemize}
    \item Further discussions and simulations on unintended consequences of our socially-aware constraints (\Cref{sec:numerical-unintended}); additional experiments related to demographic parity in selection (\Cref{sec:numerical-parity}); additional numerical results for the average quota constraint~\ref{eq:quota} (\Cref{sec:numerical-quota}); the average budget for subsidization constraint~\ref{eq:budget} (\Cref{sec:numerical-budget}); and several more detailed performance metrics (\Cref{apx:numerical});
    \item Robustness checks of our numerical findings under uniform distributions for candidates' values
    (instead of normal), examining both short-term effects (\Cref{sec:num-short-term_uniform}) and long-term impacts (\Cref{sec:num-long-term_uniform}) under constraint~\ref{eq:parity}, as well as under constraint~\ref{eq:quota} (\Cref{sec:num-long-term_uniform_quota});
    \item Experiments on the JMS model for \Cref{ex:reject} under multiple affine constraints, specifically imposing one constraint of type~\ref{eq:parity} at each stage of the search process (\Cref{sec:numerical_JMS}).
\end{itemize}
}

\revcolor{
\subsection{Unintended Consequences of Demographic Parity: a Dichotomy}
\label{sec:numerical-unintended}

Regardless of considering biased observable signals or unbiased unobservable true values to evaluate performances,  in scenarios that the cost of search is high, we may observe a certain type of \emph{unintended inefficiency} in the performance of optimal constrained policy for excessively small values of $\rho$. To see this, we consider the same setup as before but with the only difference being that we increase the inspection costs to be drawn independently from a uniform distribution over \([c_l=25, c_h=35]\) rather than \([c_l=3, c_h=6]\): in both \Cref{fig-rad:high_cost_short-term} and \Cref{fig-rad:high_cost_long-term} (analogs of \Cref{fig-rad:short-term} and \Cref{fig-rad:long-term}, respectively), for sufficiently small values of $\rho$, say $\rho\in [0,0.3]$, the expected utility of optimal constrained policy drops drastically as $\rho$ becomes smaller, while the expected utility of optimal unconstrained policy remains almost unchanged. 

\begin{figure}[htb]
    \centering
    \begin{subfigure}[b]{0.48\textwidth}
        \centering
        \includegraphics[width=\textwidth]{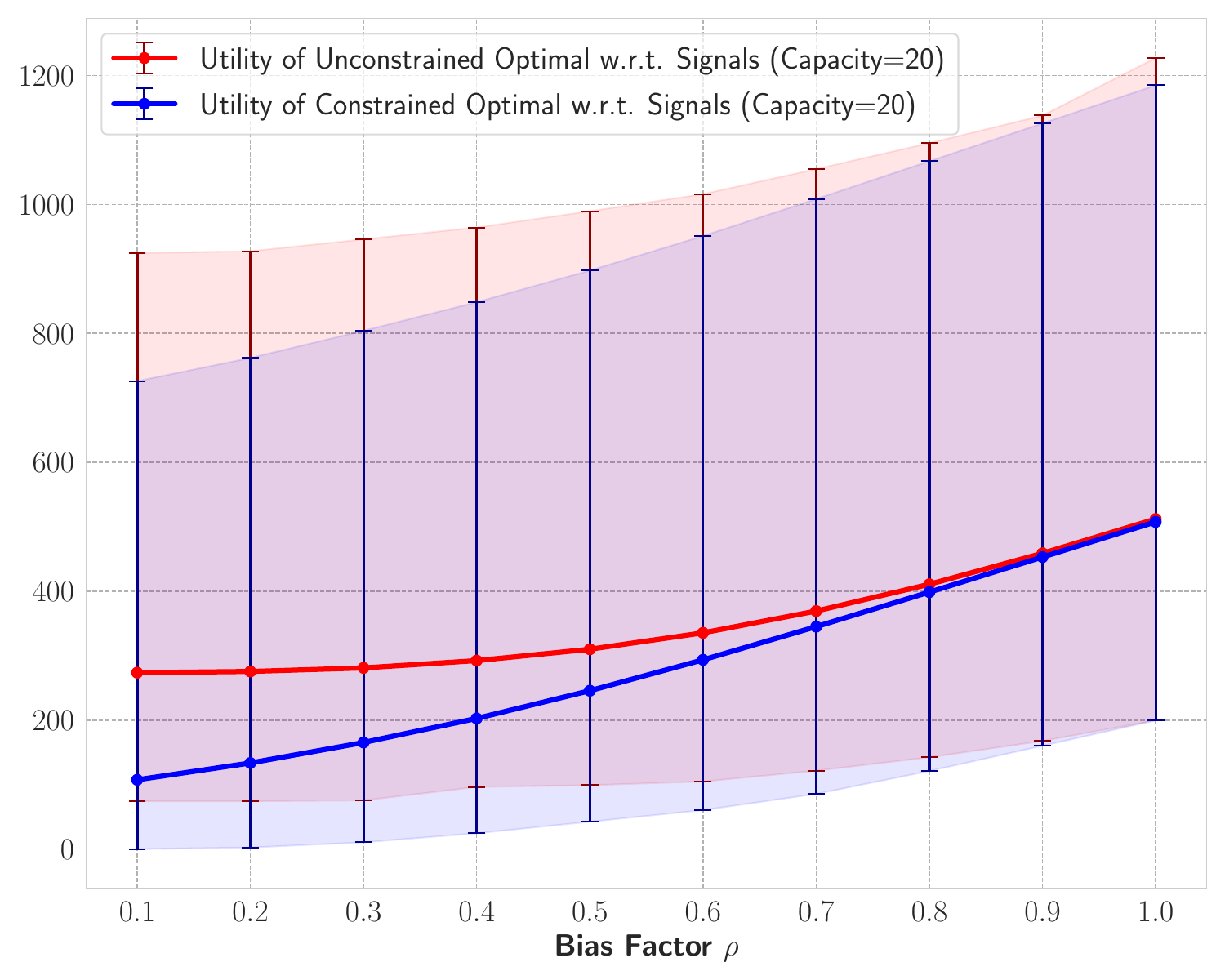}
        \caption{Expected utilities calculated based on signals $\{v_i\}_{i\in[n]}$ for the unconstrained optimal policy (red) and the constrained optimal policy (blue).}
    \end{subfigure}
    \hfill
        \begin{subfigure}[b]{0.48\textwidth}
        \centering
        \includegraphics[width=\textwidth]{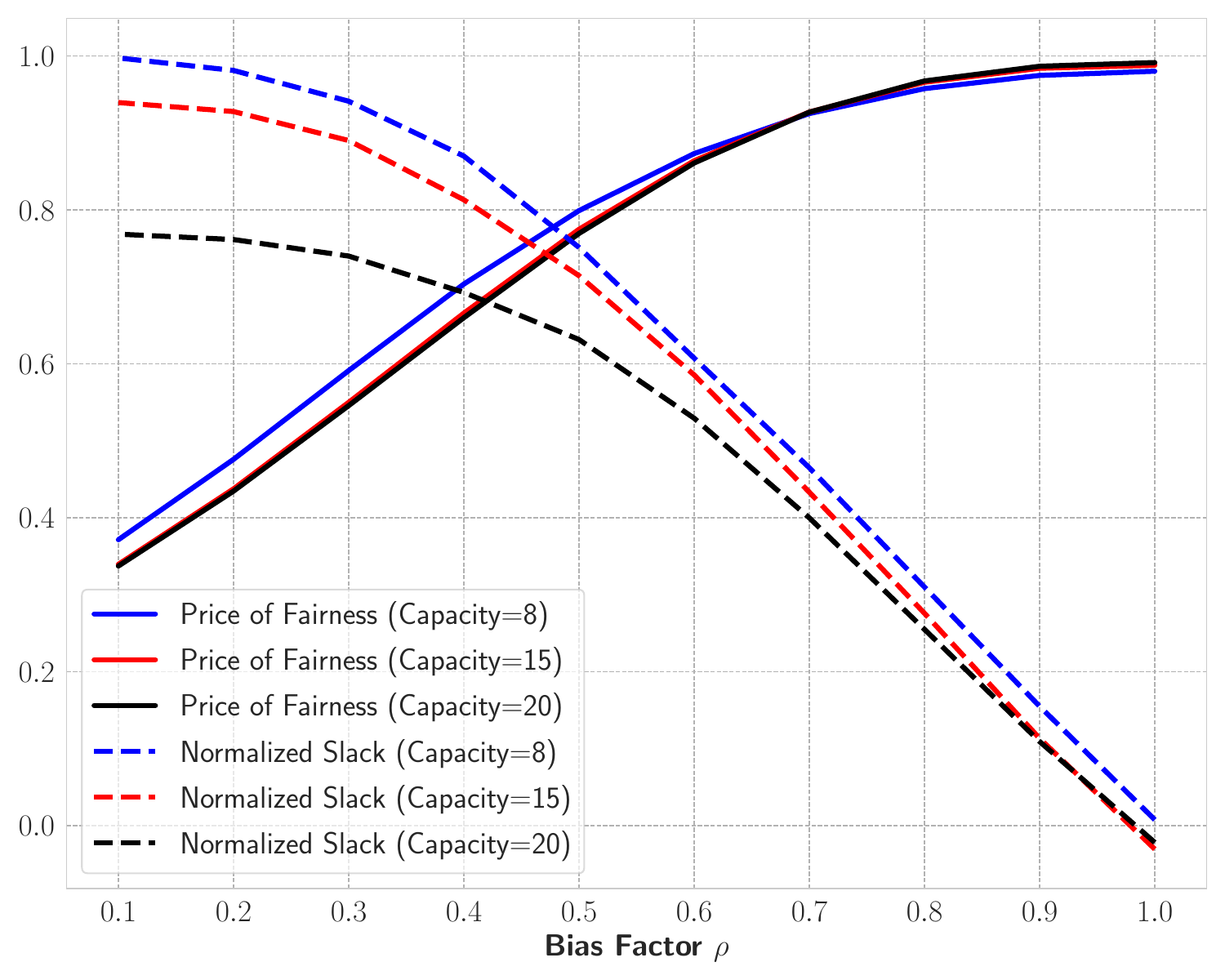}
        \caption{Price of fairness ratio calculated based on signals $\{v_i\}_{i\in[n]}$ (solid lines) and the normalized constraint slack of unconstrained optimal policy (dashed lines).}
    \end{subfigure}
    \caption{Comparing the short-term outcomes of unconstrained and constrained optimal policies, when inspection costs are high.}
    \label{fig-rad:high_cost_short-term}
\end{figure}

\vspace{-1cm}

\begin{figure}[htb]
    \centering
    \begin{subfigure}[b]{0.48\textwidth}
        \centering
        \includegraphics[width=\textwidth]{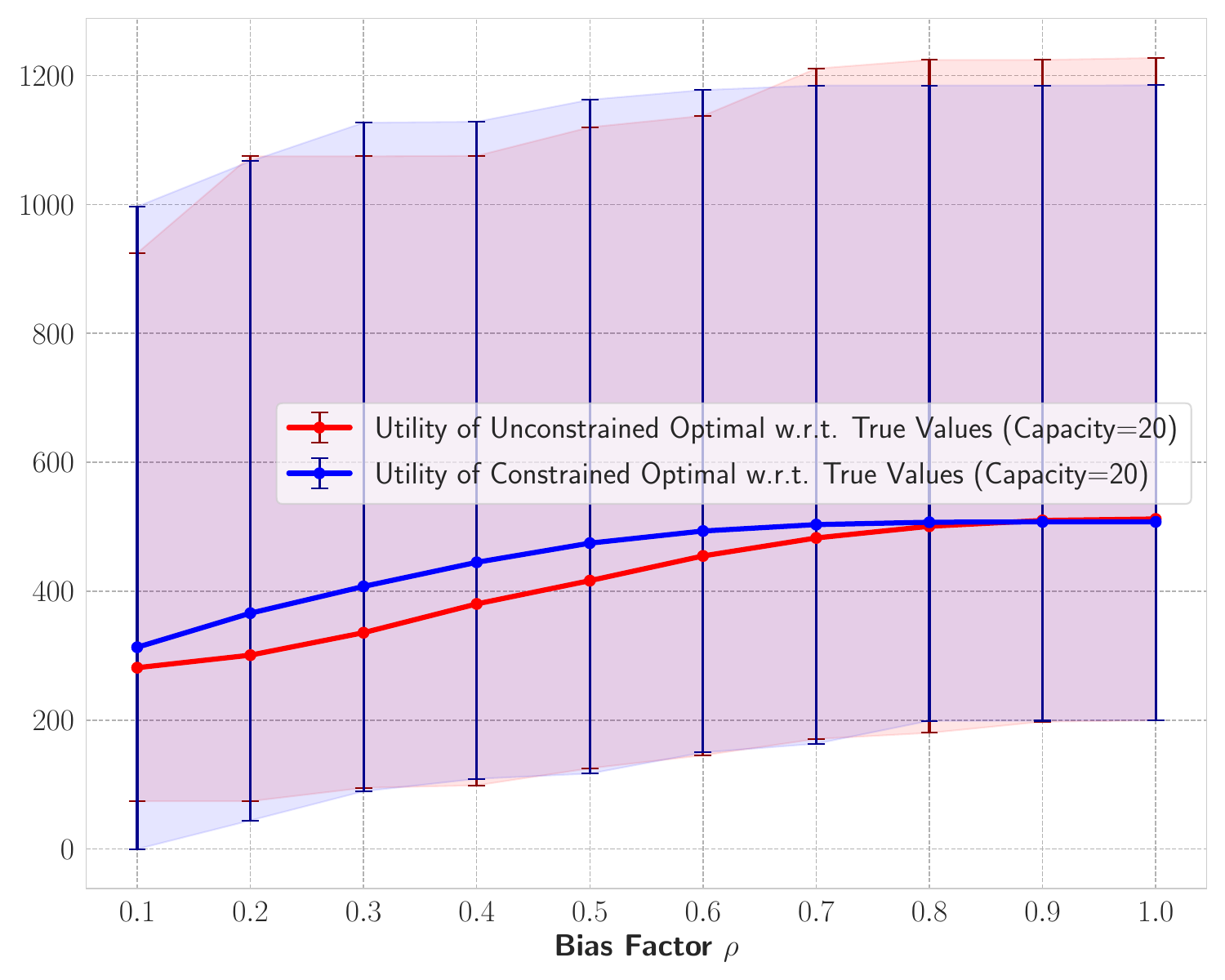}
        \caption{Expected utilities calculated based on true values $\{{v}^\dagger_i\}_{i\in[n]}$ for the unconstrained optimal policy (red) and the constrained optimal policy (blue).}
    \end{subfigure}
    \hfill
        \begin{subfigure}[b]{0.48\textwidth}
        \centering
        \includegraphics[width=\textwidth]{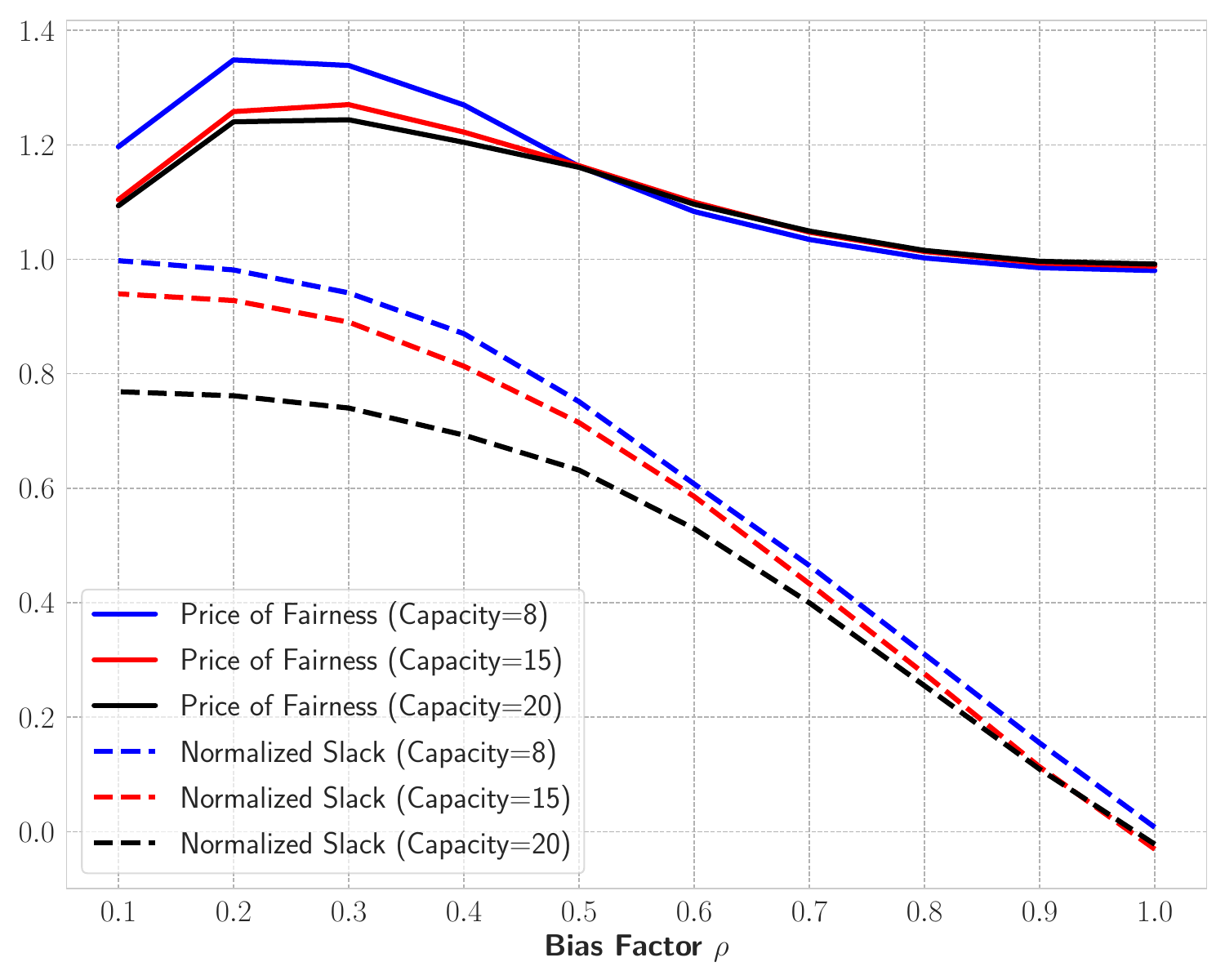}
        \caption{Price of fairness ratio calculated based on true values $\{v^\dagger_i\}_{i\in[n]}$ (solid lines) and the normalized constraint slack of unconstrained optimal policy (dashed lines).}
    \end{subfigure}
    \caption{Comparing the long-term outcomes of unconstrained and constrained optimal policies, when inspection costs are high.}
    \label{fig-rad:high_cost_long-term}
\end{figure}

Why would imposing the parity constraint have a different ``calibrating effect'' for bias factors $0.7$ and $0.3$? 
It turns out that the optimal constrained policy ends up not filling the entire capacity when inspection costs are high and $\rho$ is small. 
This under-allocation is because the (observable and biased) signal distributions suggest that if demographic parity is enforced, it is less costly to leave the capacity unused than inspecting and then hiring ``seemingly'' low-quality high-cost candidates.
In \Cref{fig-rad:unintended}, we plot the fraction of unallocated capacity by the optimal constrained policy as parameters $\rho$ and $k$ vary, which clearly shows the existence of this unintended effect for small values of $\rho$ (and that it intensifies for larger values of $k$). Lastly, we note that this is in contrast with the behavior of the optimal unconstrained policy, for it continues to fill most of its capacity even if $\rho=0.1$, as can be seen from its normalized slack in both \Cref{fig-rad:high_cost_short-term} and \Cref{fig-rad:high_cost_long-term}.

\vspace{-4mm}

\begin{figure}[htb]
    \centering
    \begin{subfigure}[b]{0.47\textwidth}
        \centering
        \includegraphics[width=\textwidth]{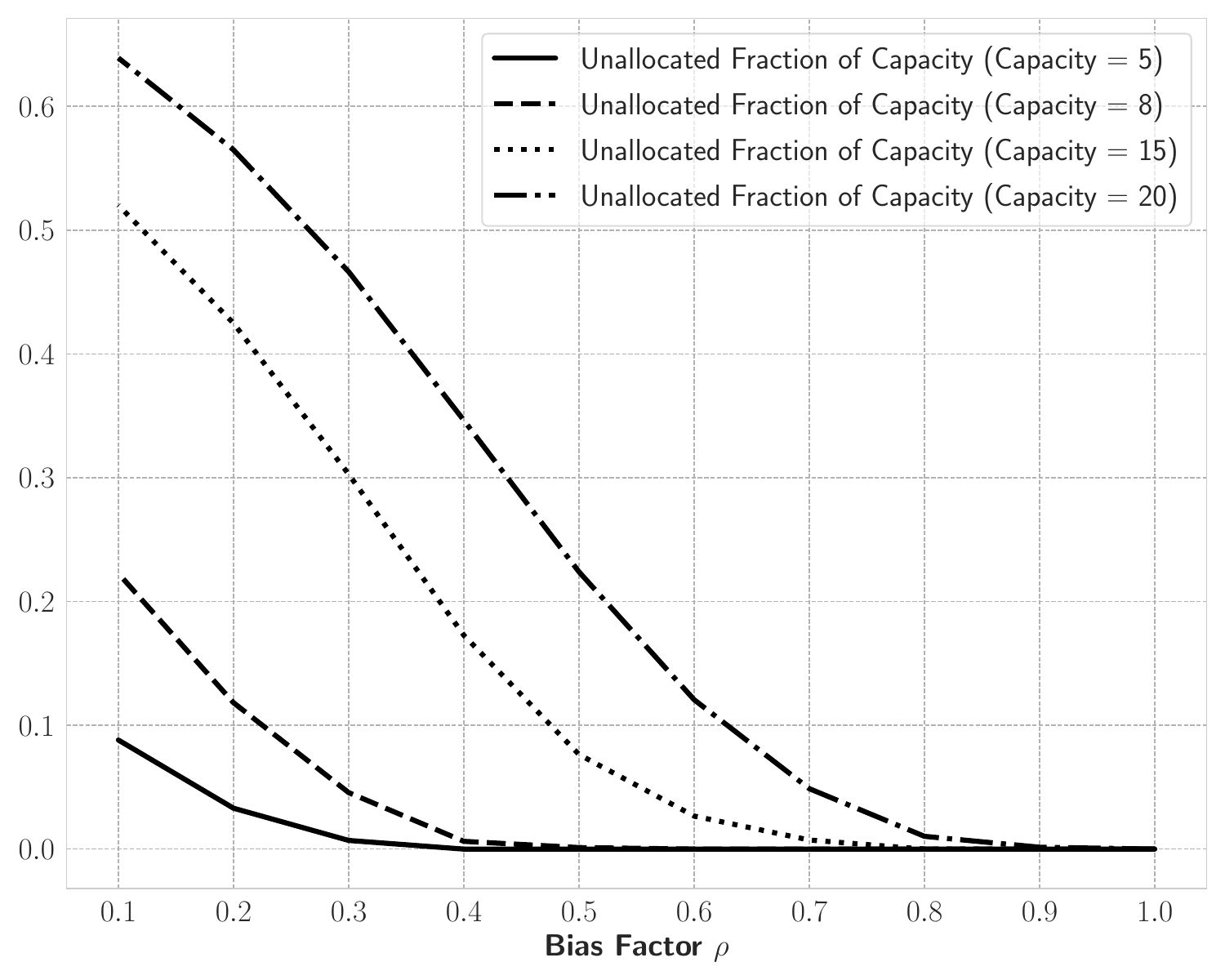}
    \end{subfigure}
    \hfill
        \begin{subfigure}[b]{0.47\textwidth}
        \centering
        \includegraphics[width=\textwidth]{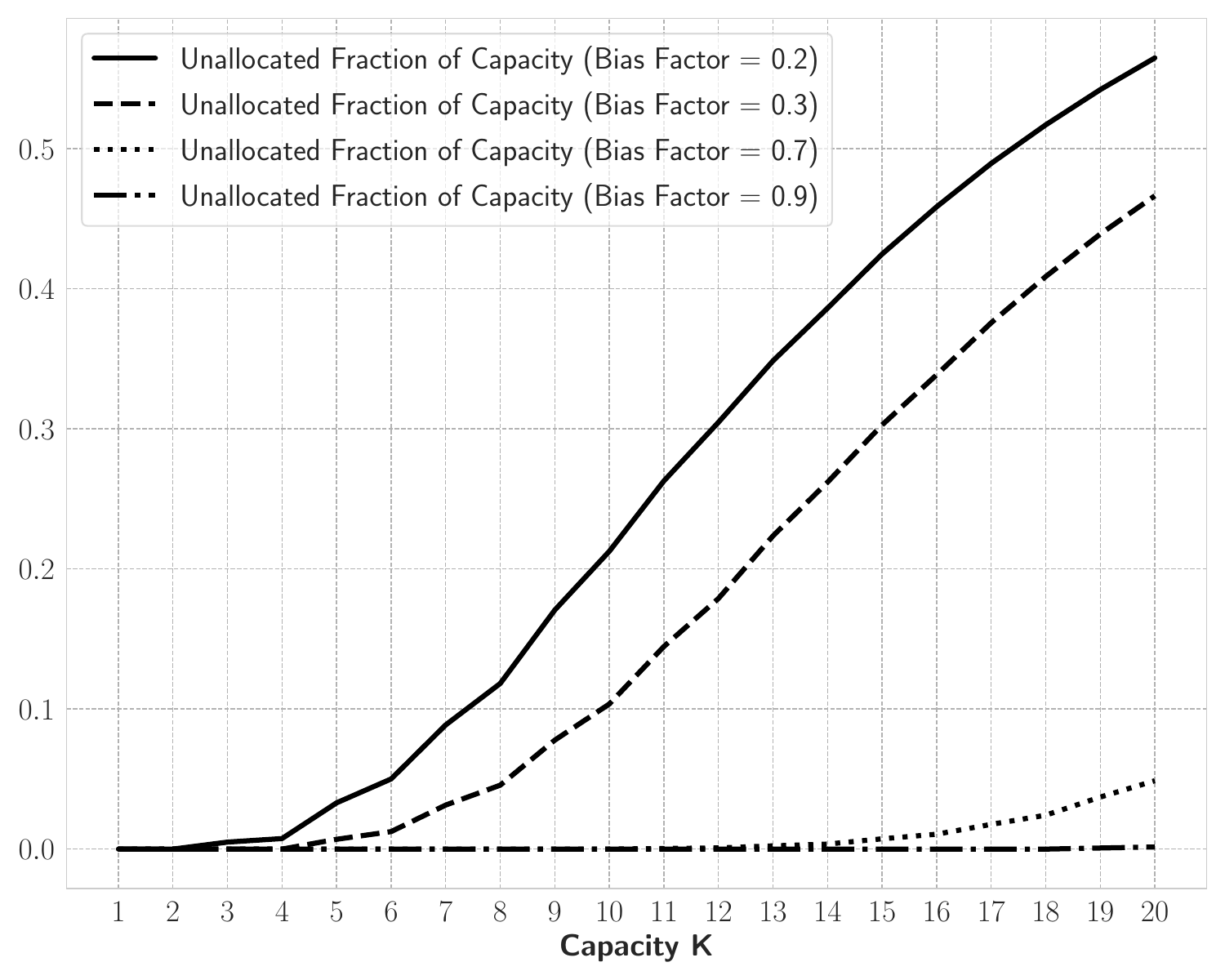}
    \end{subfigure}
    \caption{\vspace{-2mm}The unallocated fraction of the capacity by the optimal constrained policy for \ref{eq:parity} in selection.}
    \label{fig-rad:unintended}
\end{figure}
}

\vspace{-4mm}

\subsection{Demographic parity in selection}
\label{sec:numerical-parity}
We first study the effect of demographic parity in selection, that is, the constraint \eqref{eq:parity} for selection. In \Cref{sec:numerical-changing-cap}, we consider the short-term effect of imposing demographic parity. To do so, we use observable signals (which are biased for the minority group $\WomanSet$) as the only surrogate for true values and measure the utilities using these signals. We then consider the long-term effects of imposing demographic parity in \Cref{sec:numerical-parity-true}, by measuring the utilities with respect to true values and not biased signals. Note that in our setting, the true values of $\WomanSet$ are not statistically different from the true values of $\ManSet$. In both settings, we compare the optimal unconstrained policy, that is, the solution to \eqref{eq:opt-unconstrained}, with the optimal constrained policy, that is, the solution to \eqref{eq:opt-constrained}, where both policies have only access to observable biased signals upon inspection. Finally, we study the effect of increasing capacity on the ``price of fairness'' in \Cref{sec:numerical-few-pos}. In \Cref{apx:numerical}, we also report the running time of these optimal policies; see \Cref{fig-apx:v1_Running_Time}.

\subsubsection{Short-term performance -- the effect of changing capacity and bias factor.}
\label{sec:numerical-changing-cap}
In our first scenario, we compare the two optimal policies as capacity $k\in[1:20]$ and bias factor $\rho\in\{0.1,0.2,\cdots,1\}$ vary. In \Cref{fig:v1_both_fair_unfair}, we plot the net short-term utilities (calculated based on biased observable signals $\{v_i\}_{i\in[n]}$) as a function of $k$ for different bias factors and as a function of $\rho$ for different capacities. Furthermore, in \Cref{fig:v1_diff_and_CR_fair_unfair}, we plot the ratio of the net utility of the optimal constrained policy over that of the optimal unconstrained policy, again as a function of both $k$ and $\rho$. To see a similar plot for utility differences, refer to \Cref{fig-apx:v1_diff_and_CR_fair_unfair} in \Cref{apx:numerical}. Lastly, in \Cref{fig:v1_slack_lambda_fair_unfair}, we plot the constraint slack of the optimal unconstrained policy, as a function of $k$ and also as a function of $\rho$. See \Cref{fig-apx:v1_slack_lambda_fair_unfair} in \Cref{apx:numerical} for the graph of the dual adjustment $\lambda^*$ required to fix the disparity (see \Cref{eq:adjusted-value}), as a function of capacity $k$ and bias factor $\rho$. Next, we discuss some managerial insights that are derived from these simulations.

\begin{figure}[htb]
    \centering
    \begin{subfigure}[b]{0.32\textwidth}
        \centering
        \includegraphics[width=\textwidth]{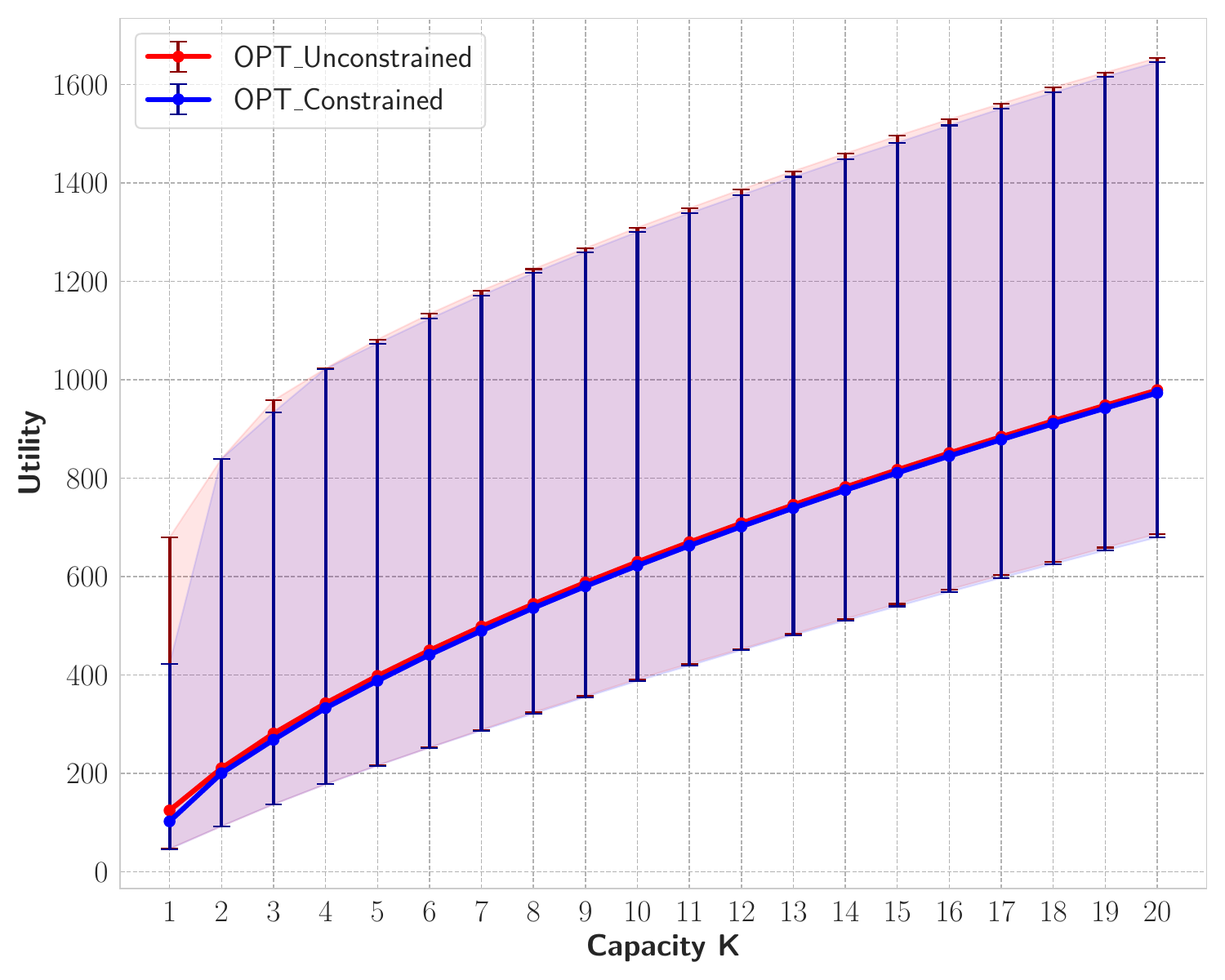}
        \caption{  Bias factor = 0.9}
    \end{subfigure}
    \begin{subfigure}[b]{0.32\textwidth}
        \centering
        \includegraphics[width=\textwidth]{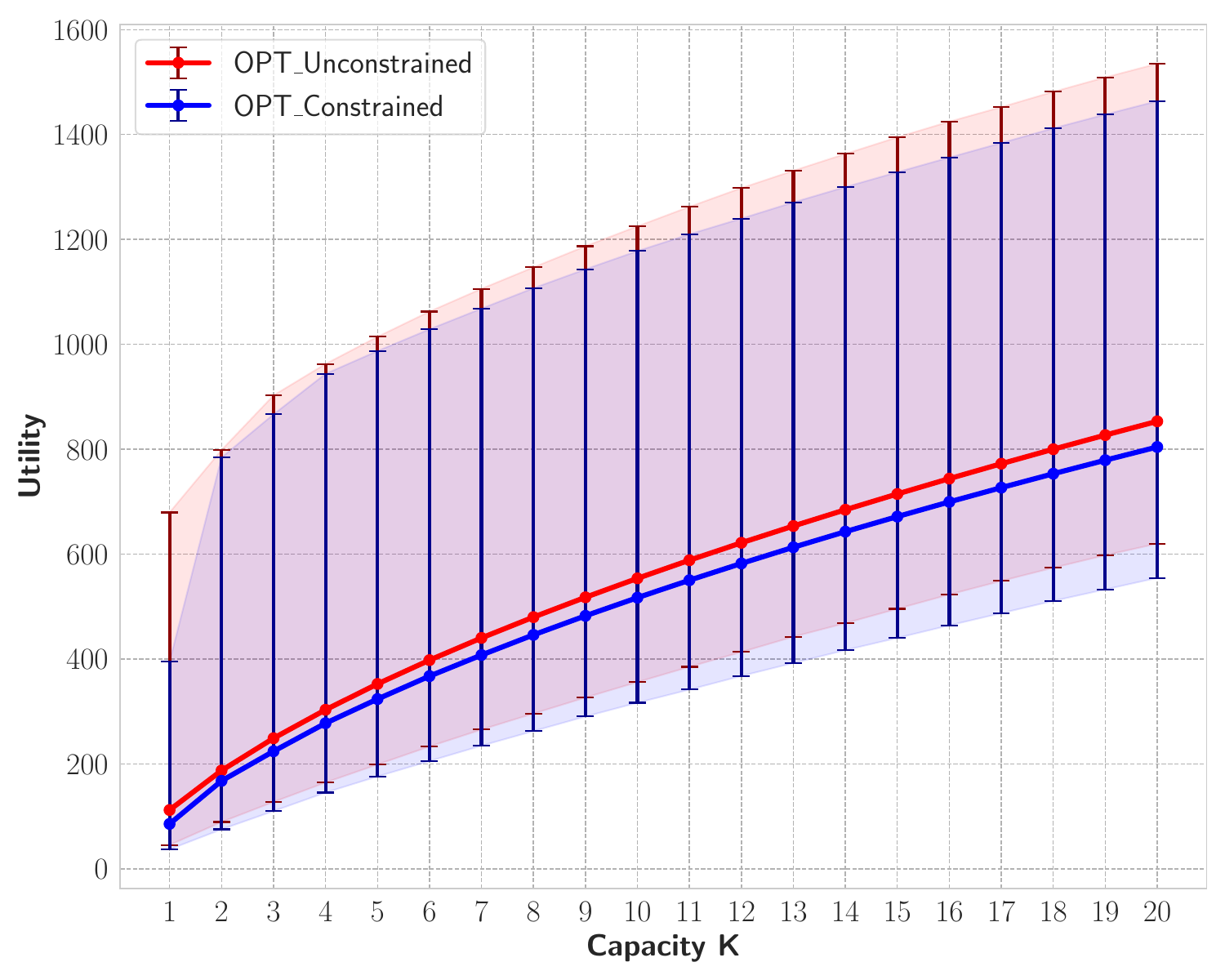}
        \caption{  Bias factor = 0.6}
    \end{subfigure}
    \begin{subfigure}[b]{0.32\textwidth}
        \centering
        \includegraphics[width=\textwidth]{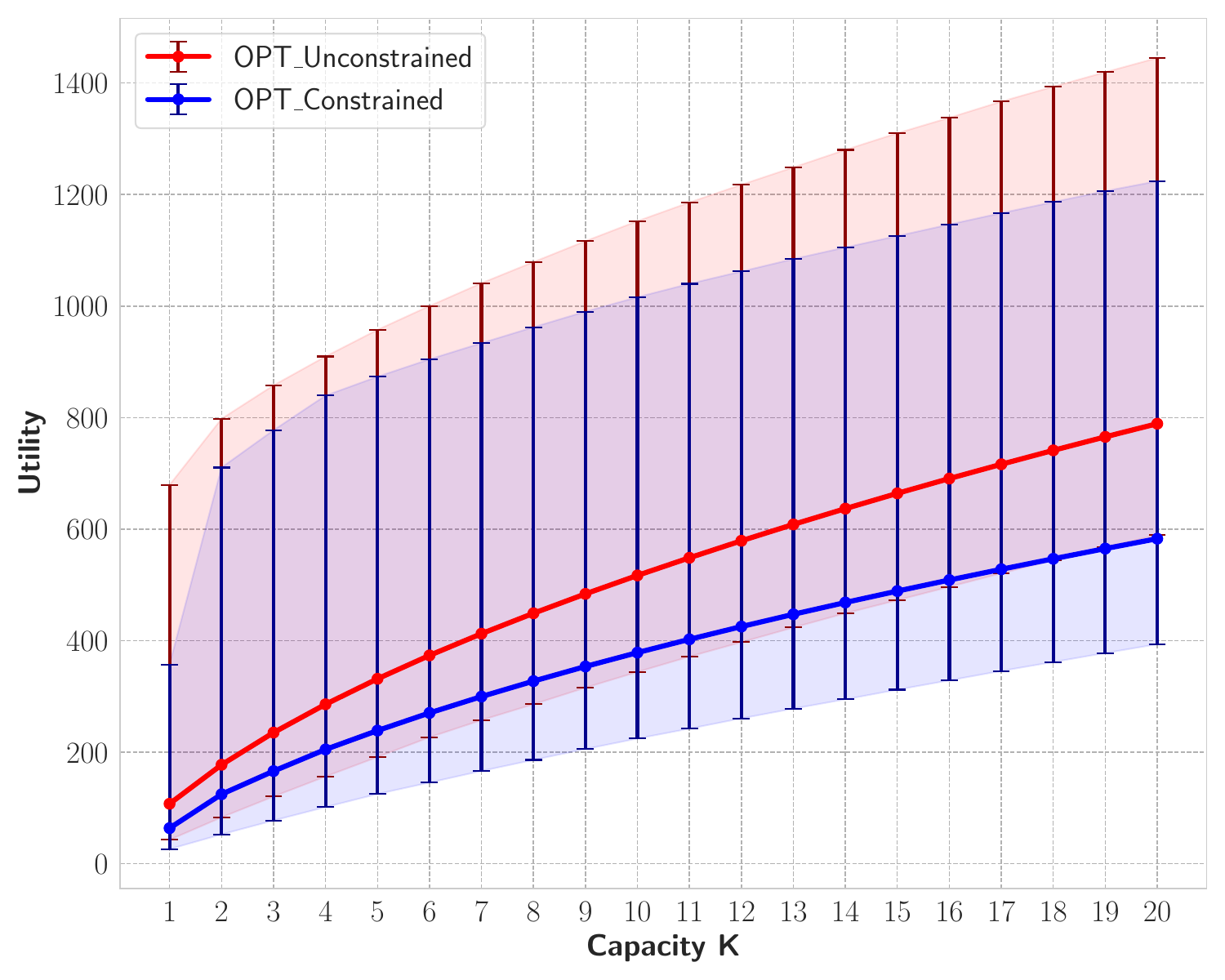}
        \caption{  Bias factor = 0.2}
    \end{subfigure}
    \begin{subfigure}[b]{0.32\textwidth}
        \centering
        \includegraphics[width=\textwidth]{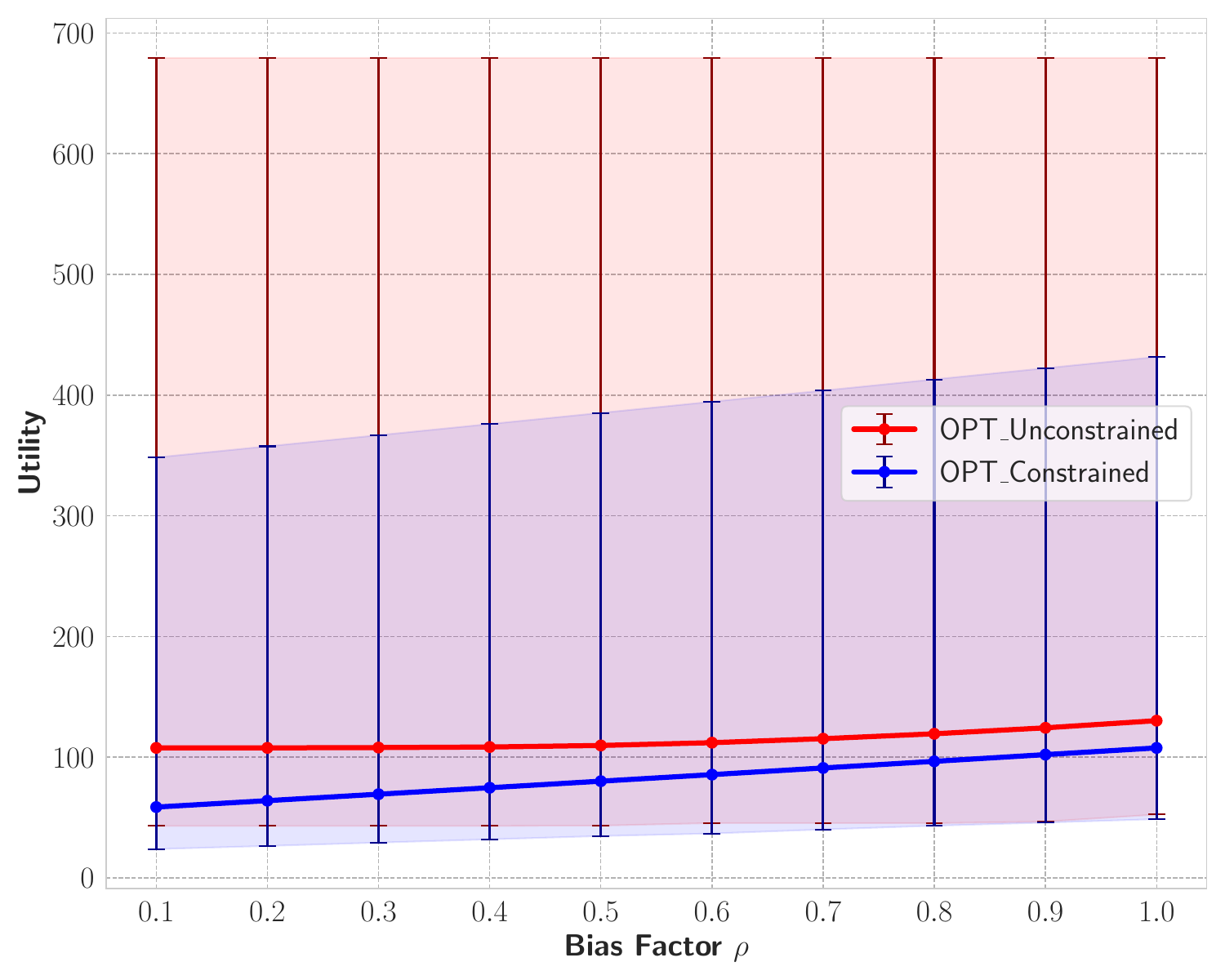}
        \caption{  Capacity = 1}
    \end{subfigure}
    \begin{subfigure}[b]{0.32\textwidth}
        \centering
        \includegraphics[width=\textwidth]{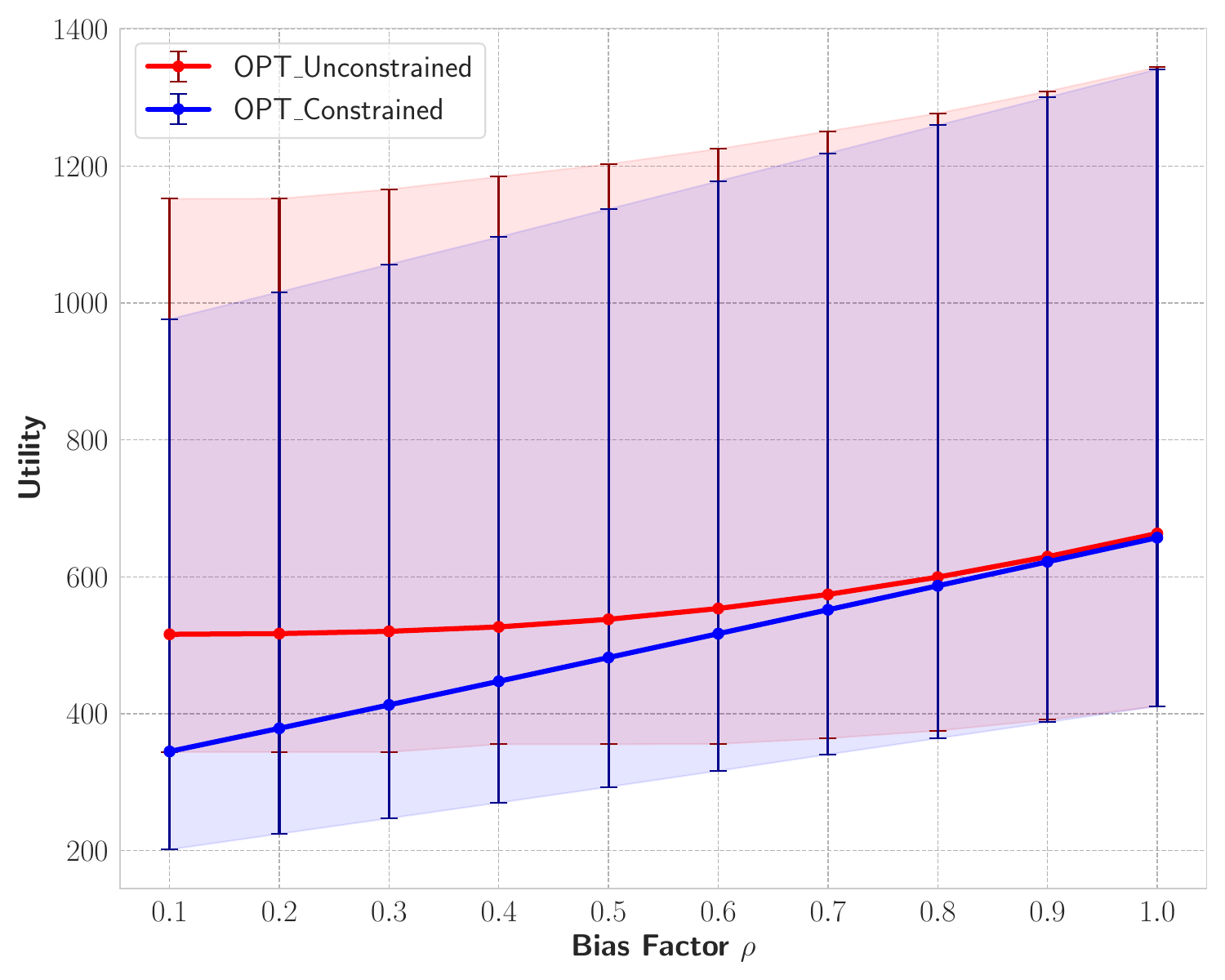}
        \caption{  Capacity = 10}
    \end{subfigure}
    \begin{subfigure}[b]{0.32\textwidth}
        \centering
        \includegraphics[width=\textwidth]{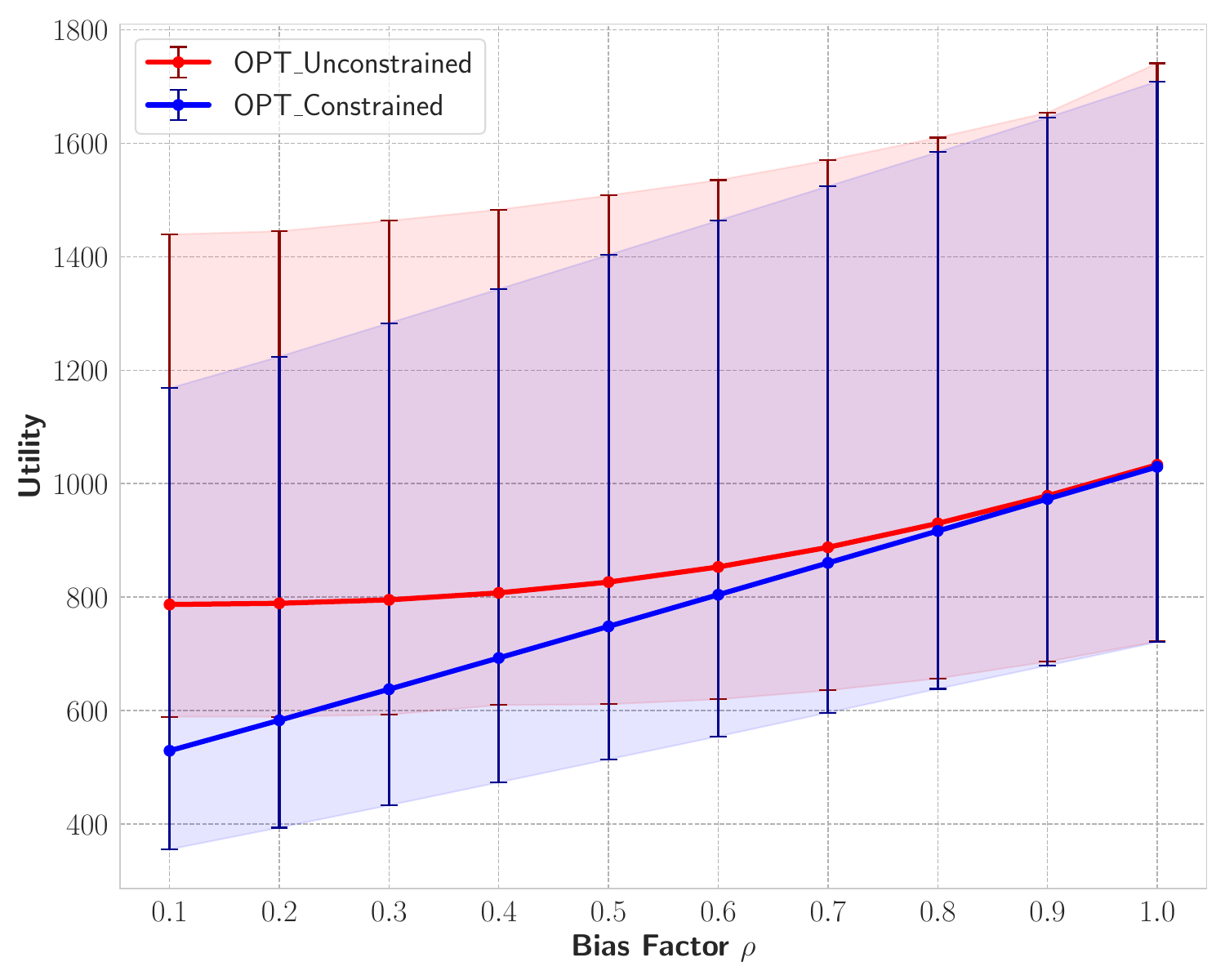}
        \caption{  Capacity = 20}
    \end{subfigure}
    \caption{Comparing the short-term utilities of optimal unconstrained (red) and  constrained (blue) policies for demographic parity in selection, for different capacities $\boldsymbol{k}$ and bias factors $\boldsymbol{\rho}$.}
    \label{fig:v1_both_fair_unfair}
\end{figure}
\begin{figure}[htb]
    \centering
    \begin{subfigure}[b]{0.46\textwidth}
        \centering
        \includegraphics[width=\textwidth]{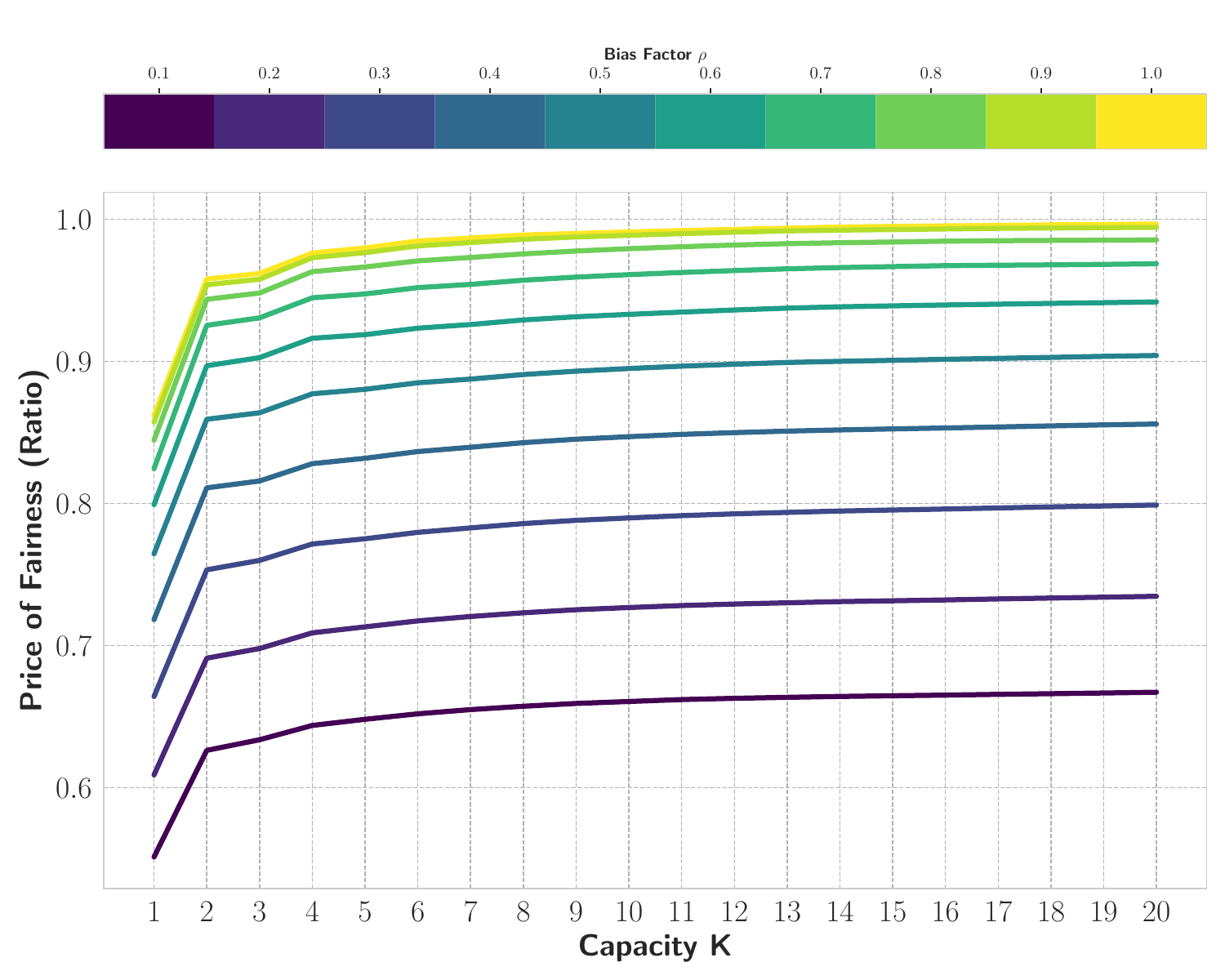}
        \caption{}
    \end{subfigure}
     \quad\quad\quad
    \begin{subfigure}[b]{0.46\textwidth}
        \centering
        \includegraphics[width=\textwidth]{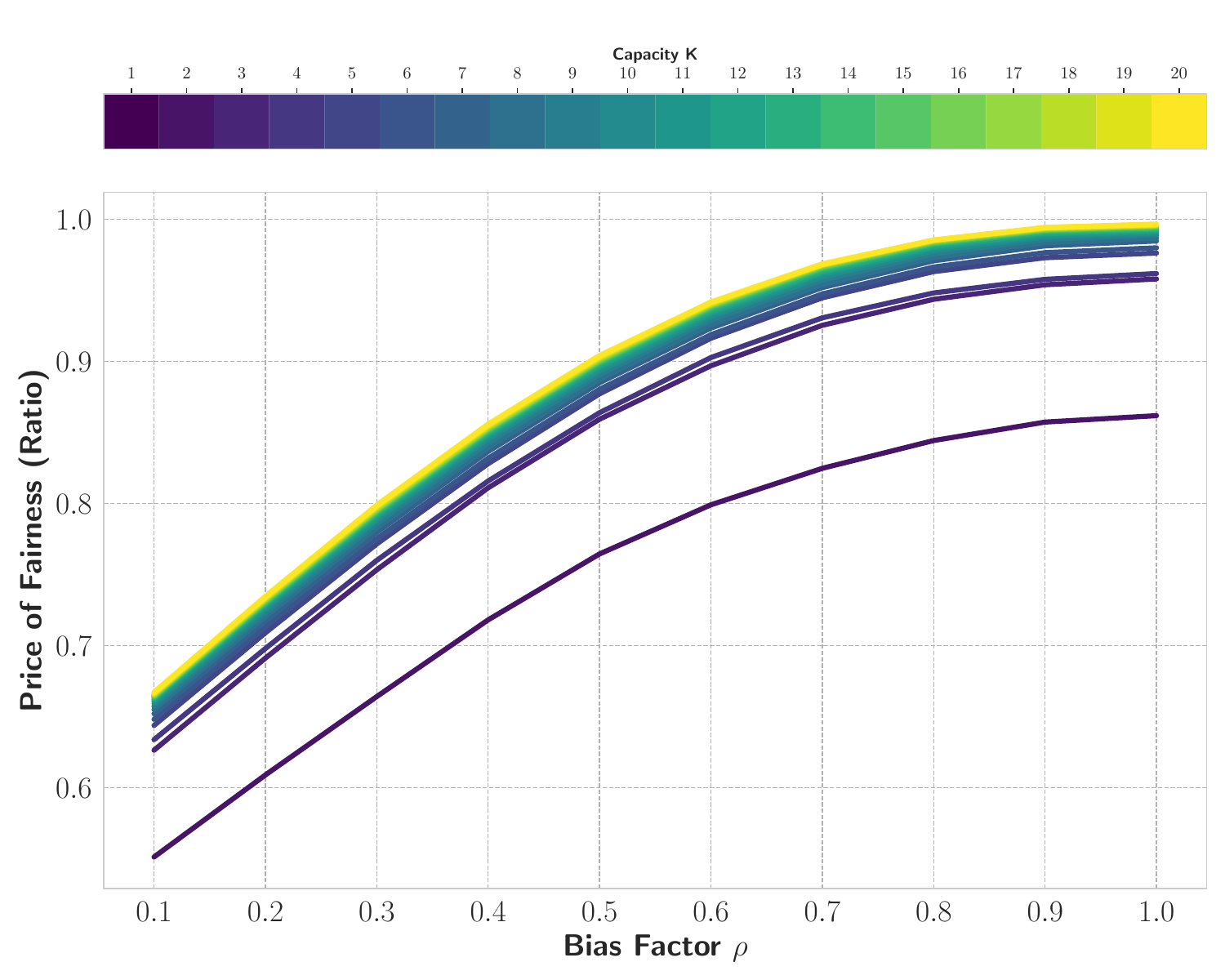}
        \caption{}
    \end{subfigure}
    \caption{Short-term price of fairness of demographic parity in selection in terms of utility ratio; (a) as a function of capacity $\boldsymbol{k}$, and (b) as a function of bias factor $\boldsymbol{\rho}$.}
    \label{fig:v1_diff_and_CR_fair_unfair}
\end{figure}

\begin{figure}[htb]
    \centering
    \begin{subfigure}{0.46\textwidth}
        \centering
        \includegraphics[width=\textwidth]{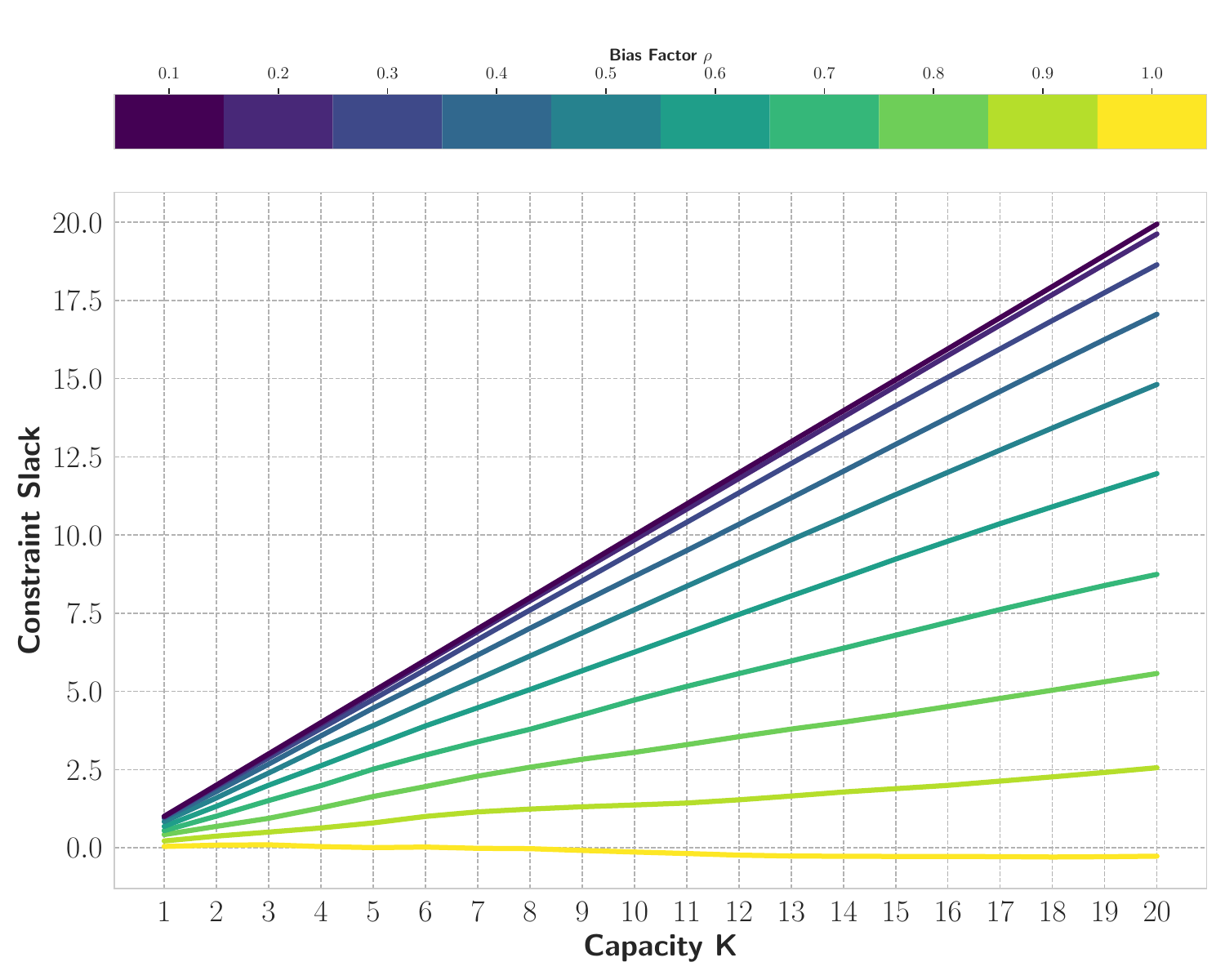}
        \caption{}
    \end{subfigure}
   \quad\quad\quad
    \begin{subfigure}{0.46\textwidth}
        \centering
        \includegraphics[width=\textwidth]{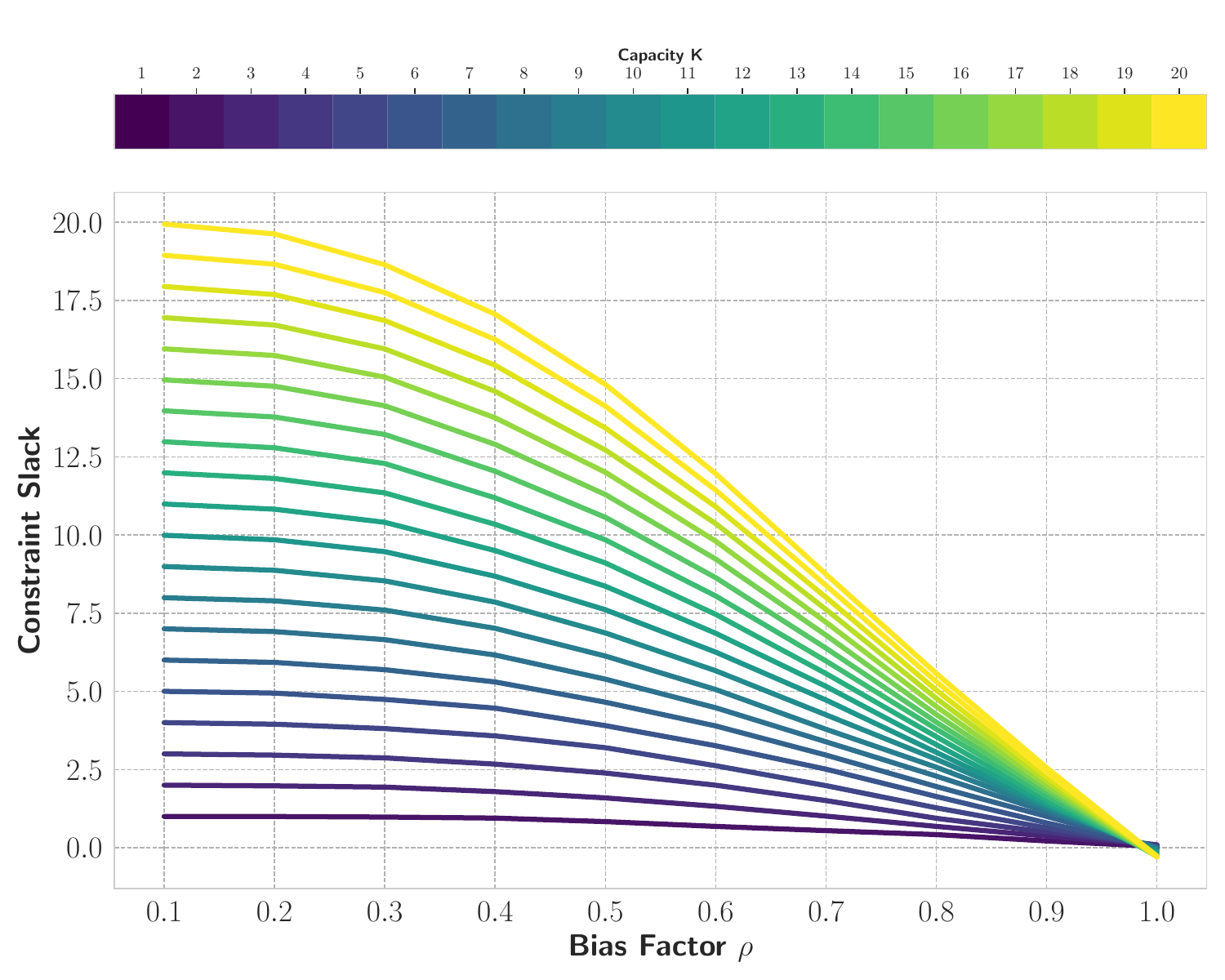}
        \caption{}
    \end{subfigure}
       \quad\quad\quad
    \caption{The constraint slack of the unconstrained optimal policy in demographic parity in selection; (a) as a function of capacity $\boldsymbol{k}$, and (b) as a function of bias factor $\boldsymbol{\rho}$.}
    \label{fig:v1_slack_lambda_fair_unfair}
\end{figure}

 First, clearly the short-term utility gap between the two policies increases as $\rho$ decreases (which means more bias) and the optimal adjustment $\lambda^*$ increases; nevertheless, for a moderate value of $\rho$, say $\rho\in [0.7,1]$, the performance gap is quite small, while the constraint slack of the optimal unconstrained policy is still quite considerable. We have investigated this managerial insight in detail in \Cref{sec:num-short-term}.

Second, for large enough values of $\rho$, say $\rho\in [0.3, 1]$, the utility of both optimal policies increases as $k$ increases; nevertheless, for small $\rho$, the performance of optimal constrained policy becomes constant after some $k\approx 6$  as it begins to suffer from a new form of inefficiency: due to the significant difference between the two groups, this policy decides \emph{not to fill} its capacity to satisfy \eqref{eq:parity}. We have already investigated this source of inefficiency in detail in \Cref{sec:numerical-unintended}.

\subsubsection{Long-term performance -- the effect of changing capacity and bias factor}
\label{sec:numerical-parity-true}
Now, we compare the expected long-term utilities of the optimal constrained and the optimal unconstrained policies in \Cref{fig:v3_both_fair_unfair}, where the long-term utilities are calculated based on the true values. The price of fairness with respect to the true values, in terms of the ratio of optimal constrained to optimal unconstrained and also their difference, is reported in \Cref{fig:v3_diff_and_CR_fair_unfair}.  Interestingly, we observe that the true utility of the optimal constrained policy \emph{dominates} that of the optimal unconstrained policy, as long as the bias factor is not very small (e.g., $\rho\geq 0.07$ for $k=5$, $\rho\geq 0.18$ for $k=10$ and $\rho\geq 0.28$ for $k=20$). We have investigated this managerial insight in \Cref{sec:num-long-term}.

For small bias factors, even when the true values are unbiased, the constrained optimal policy might decide \emph{not to fill} its capacity to satisfy \eqref{eq:parity} -- hence it suffers from a similar form of inefficiency as mentioned earlier. See more details in \Cref{sec:numerical-unintended}.

\begin{figure}[htb]
    \centering
    \begin{subfigure}[b]{0.32\textwidth}
        \centering
        \includegraphics[width=\textwidth]{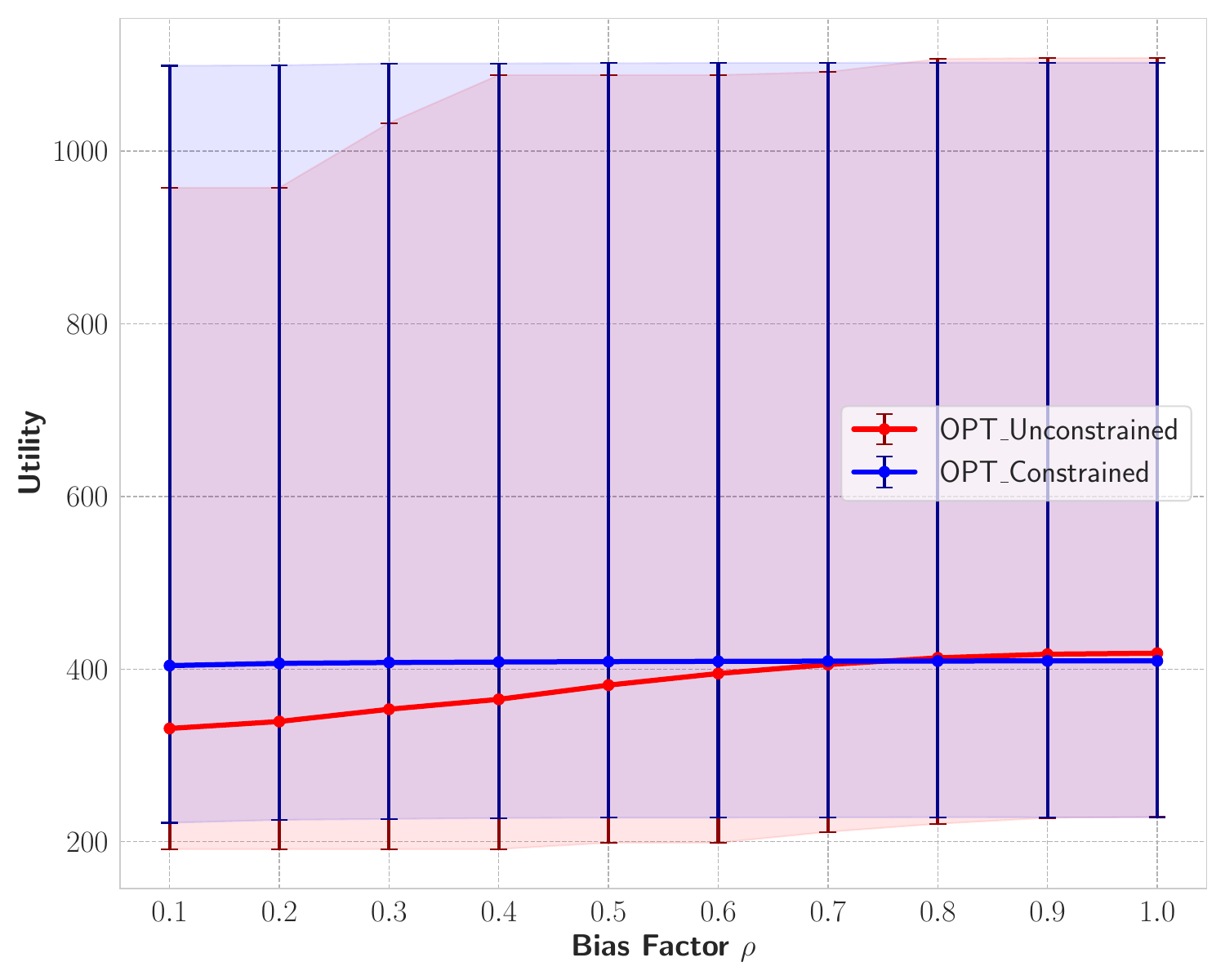}
        \caption{  Capacity = 5}
    \end{subfigure}
    \hfill
    \begin{subfigure}[b]{0.32\textwidth}
        \centering
        \includegraphics[width=\textwidth]{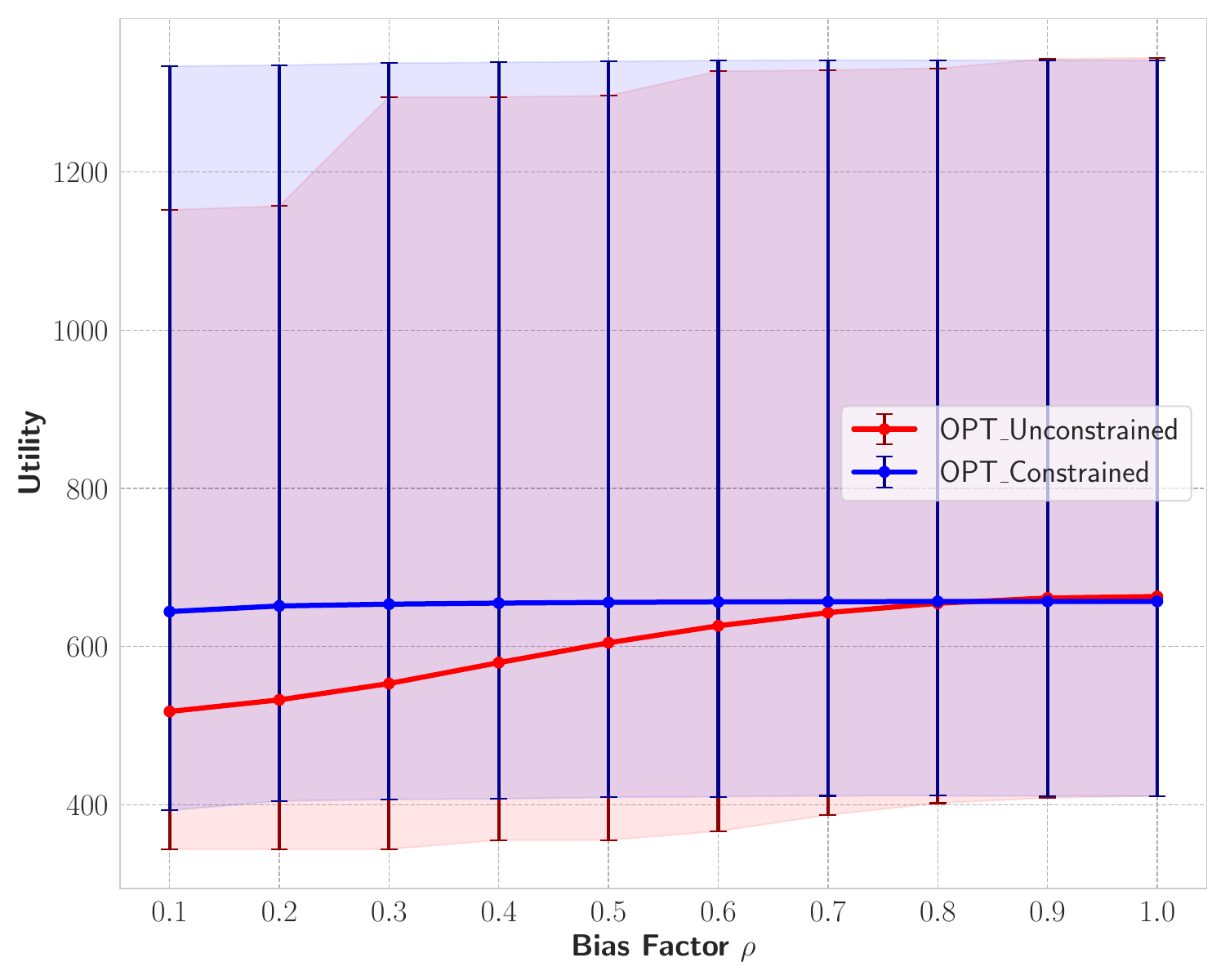}
        \caption{  Capacity = 10}
    \end{subfigure}
    \hfill
    \begin{subfigure}[b]{0.32\textwidth}
        \centering
        \includegraphics[width=\textwidth]{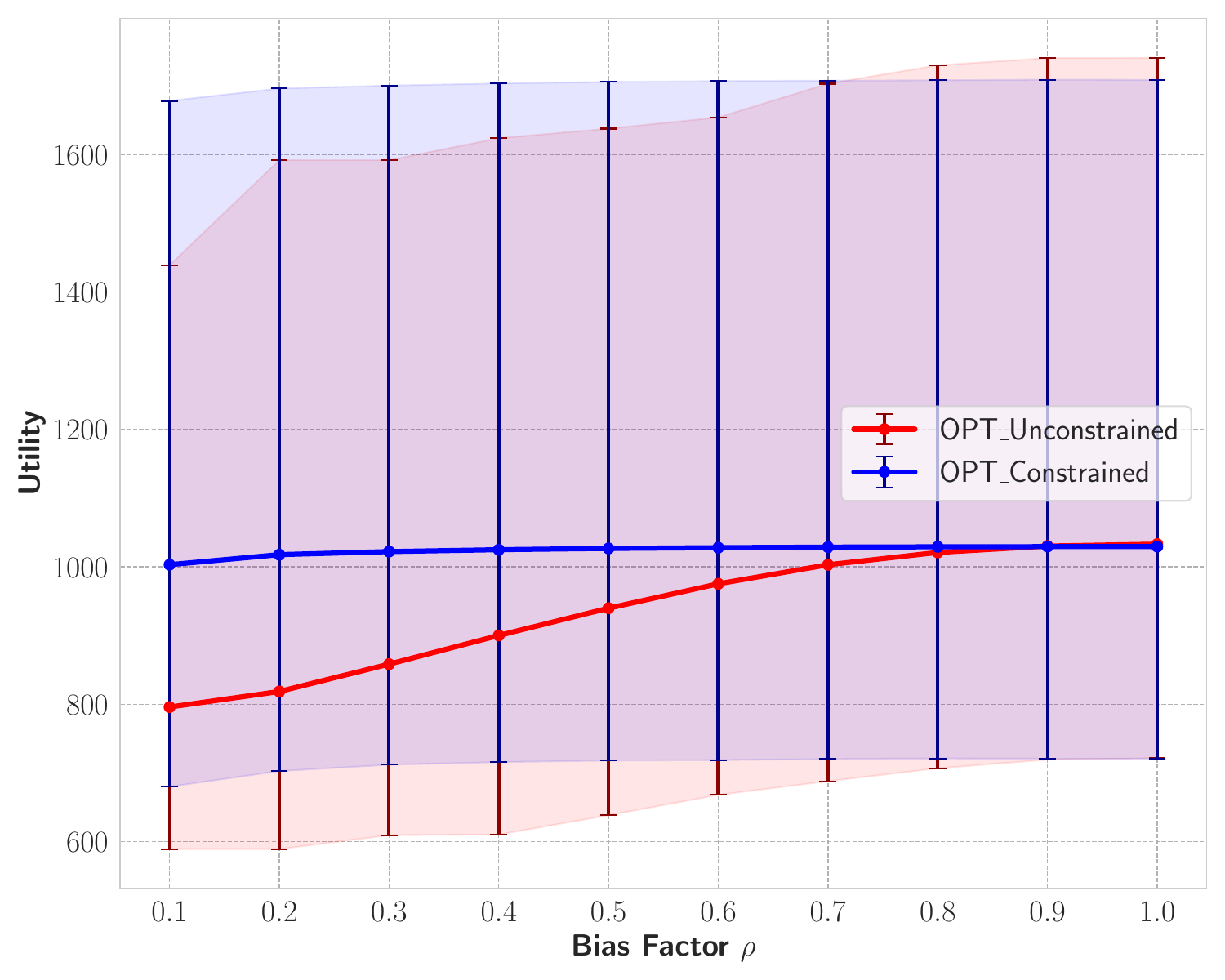}
        \caption{  Capacity = 20}
    \end{subfigure}
    \hfill
    \caption{Comparing the long-term utilities of optimal unconstrained (red) and  constrained (blue) policies for demographic parity in selection, for different capacities $\boldsymbol{k}$ and bias factors $\boldsymbol{\rho}$.}
    \label{fig:v3_both_fair_unfair}
\end{figure}

\begin{figure}[htb]
    \centering
    \begin{subfigure}[b]{0.45\textwidth}
        \centering
        \includegraphics[width=\textwidth]{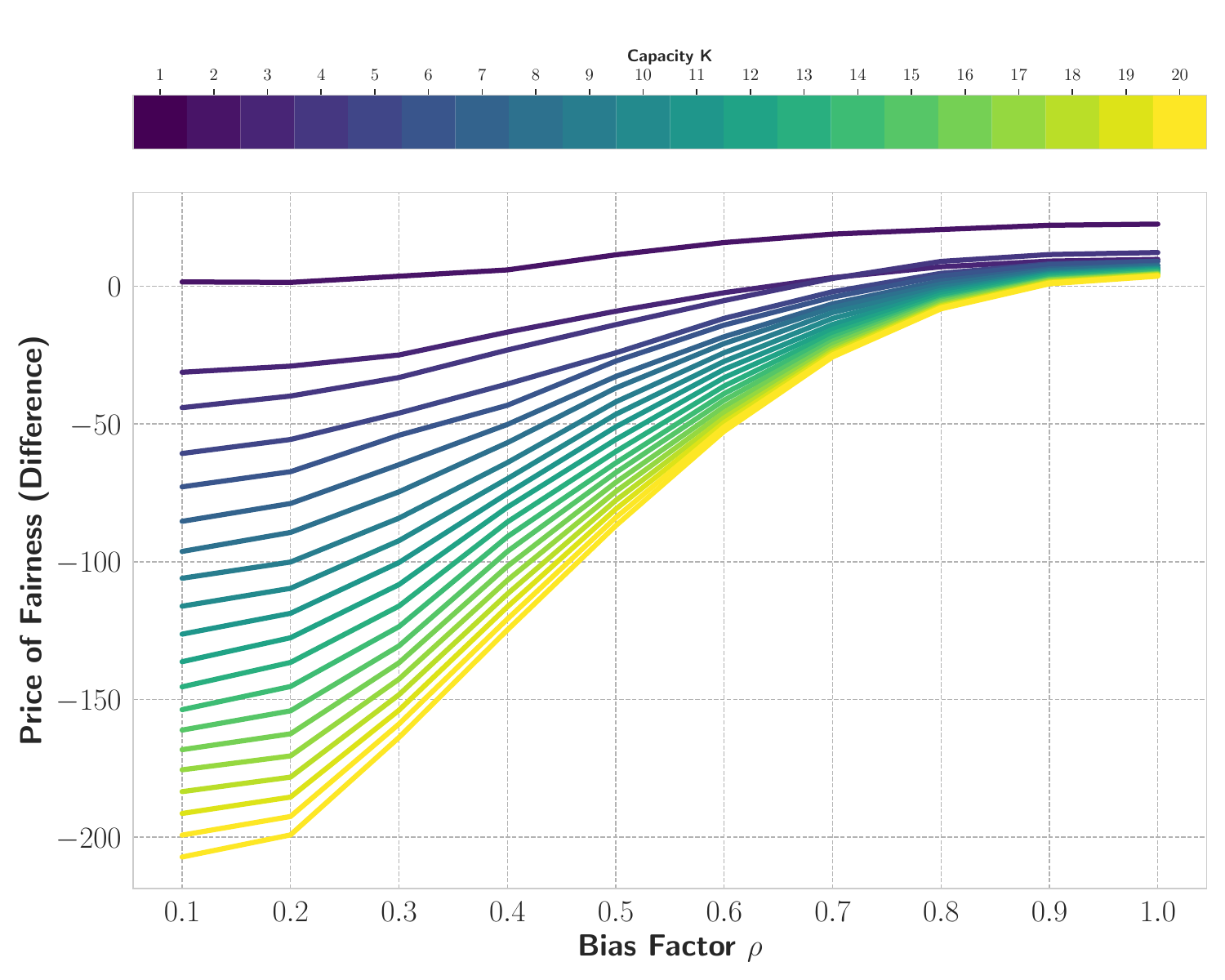}
        \caption{  Difference in utilities}
    \end{subfigure}
    \quad\quad\quad
    \begin{subfigure}[b]{0.45\textwidth}
        \centering
        \includegraphics[width=\textwidth]{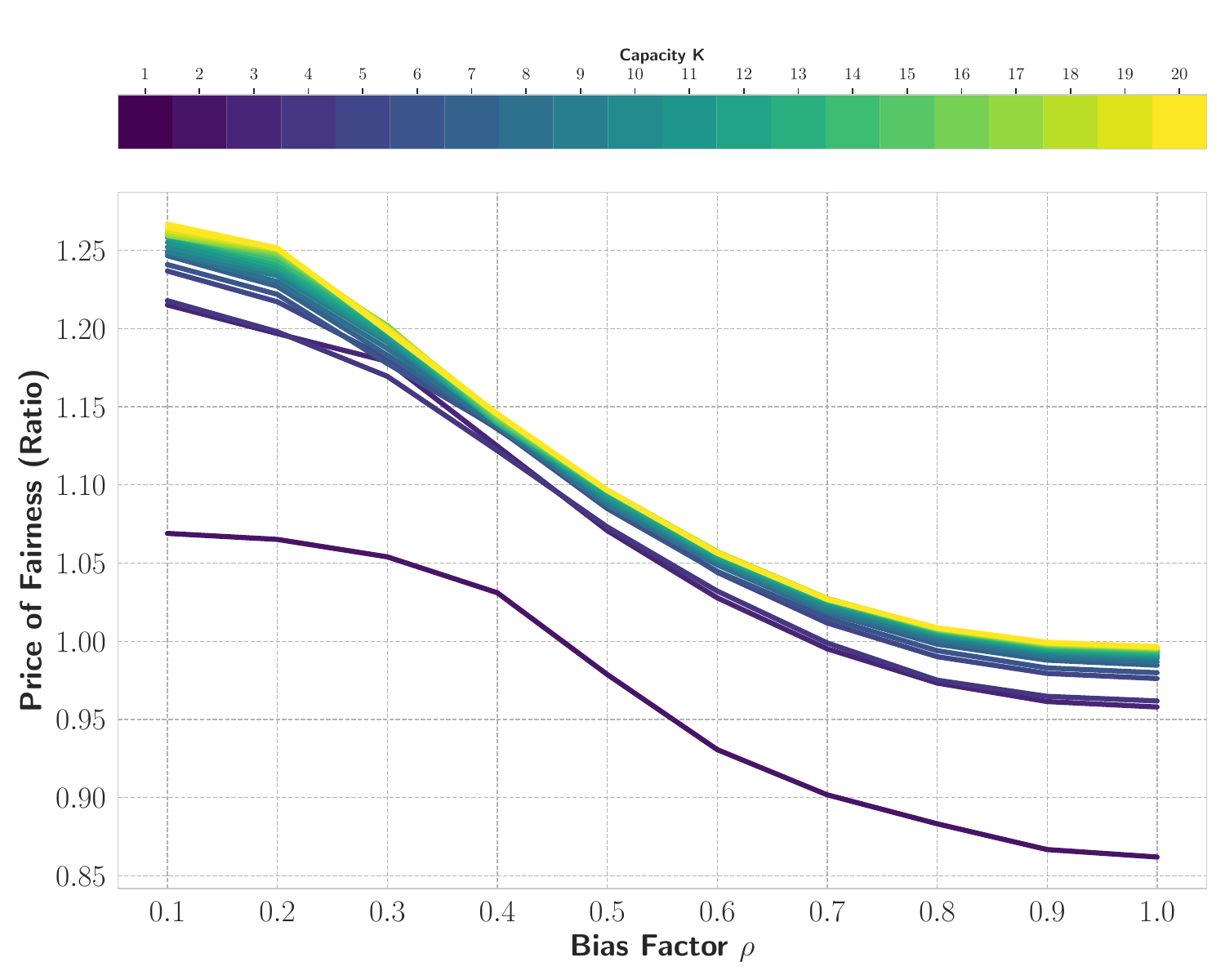}
        \caption{  Ratio of utilities}
    \end{subfigure}
    \caption{Long-term price of fairness of demographic parity in selection (with biased signals and unbiased true values) as a function of bias factor $\boldsymbol{\rho}$ for different capacities $\boldsymbol{k}$ in terms of (a) utility differences (b) utility ratios.}
    \label{fig:v3_diff_and_CR_fair_unfair}
\end{figure}

\subsubsection{A few positions more.}
\label{sec:numerical-few-pos}
\begin{figure}[htb]
    \centering
    \begin{subfigure}[b]{0.45\textwidth}
        \centering
        \includegraphics[width=\textwidth]{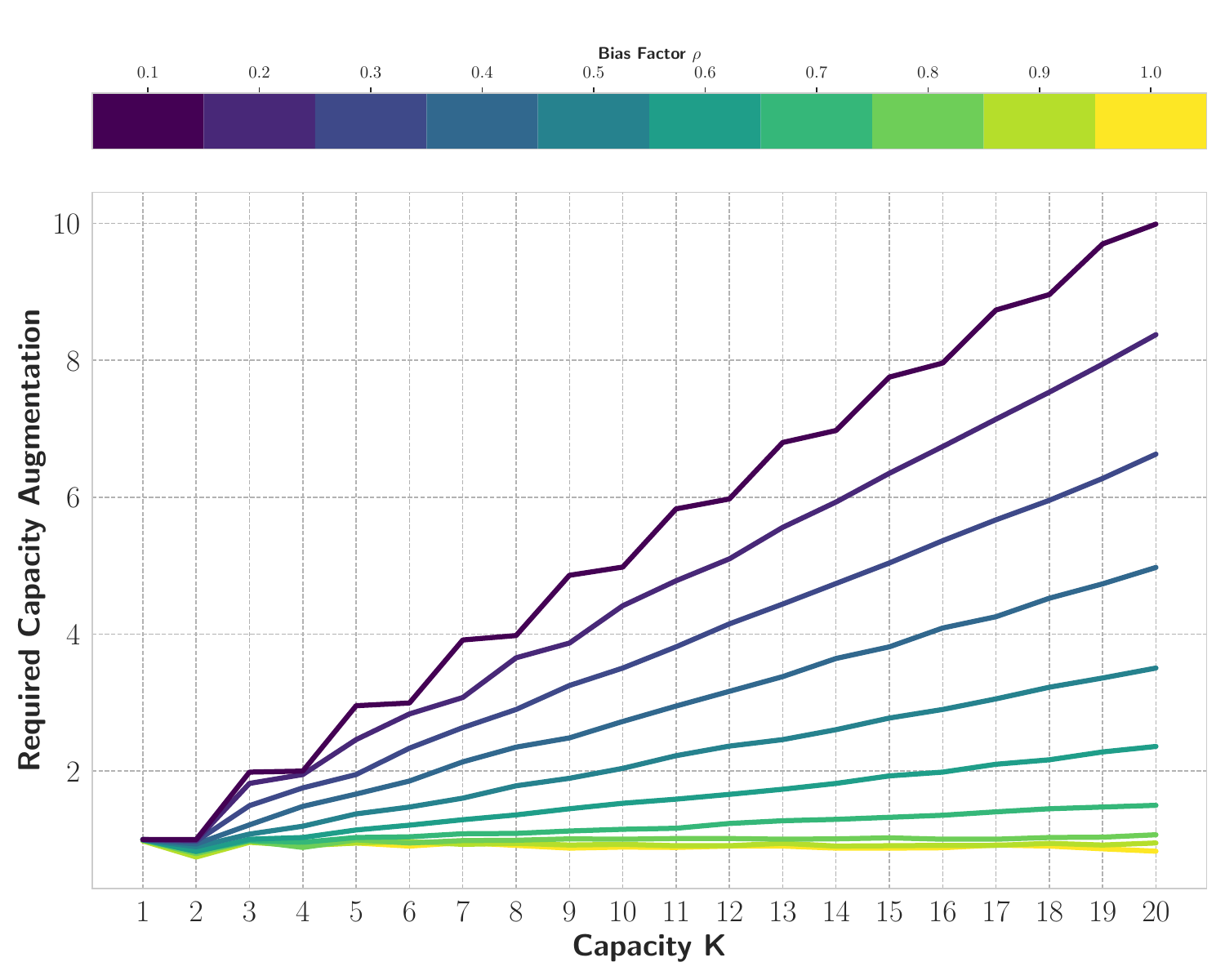}
        \caption{}
    \end{subfigure}
    \quad\quad\quad
    \begin{subfigure}[b]{0.45\textwidth}
        \centering
        \includegraphics[width=\textwidth]{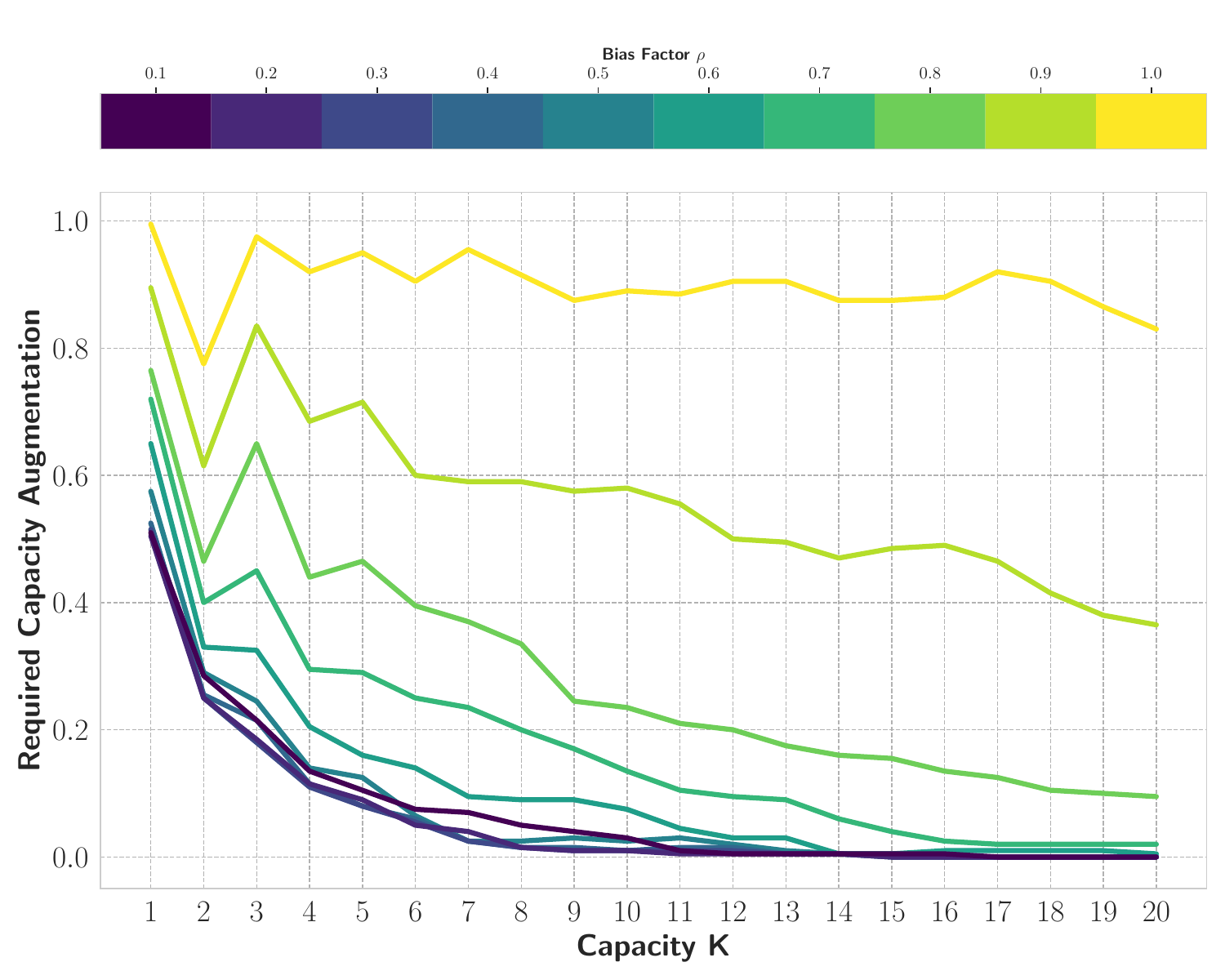}
        \caption{}
    \end{subfigure}
    \caption{The required extra capacity as a function of capacity $\boldsymbol{k}$ for different bias factors $\boldsymbol{\rho}$, to compensate for the (a) short-term utility reduction, and (b) long-term utility reduction.}
    \label{fig:v1_dominated_capacity}
\end{figure}
Given the previous investigation, an intriguing question can be asked: how many additional units of capacity should be used to impose demographic parity in selection without any loss in short-term or long-term utility? To answer this question, in \Cref{fig:v1_dominated_capacity}, we plot the number of extra units $k'$ used for each capacity $k$, so that the optimal constrained policy with capacity $k$ has at least the same net utility as the optimal unconstrained policy with capacity $k-k'$. We study both settings: short-term utilities calculated using biased signals (part (a)) and long-term utilities calculated using unbiased true values (part (b)). 

Interestingly, we observe in \Cref{fig:v1_dominated_capacity}(a) that for a moderate bias factor, say $\rho\approx 0.75$, around $23\%$ extra capacity can ensure that demographic parity in selection would not harm short-term utility at all. This percentage decreases to less than $5\%$ for $\rho=0.9$ and increases to approximately $50\%$ (with a sharp increase) when $\rho=0.6$. This sharp increase, combined with the inefficiency caused by the unused capacity mentioned earlier in this section, suggests that when there is a significant bias in the signals of one of the groups, the decision maker might be better off focusing on more relaxed notions of fairness than \eqref{eq:parity}, for example \eqref{eq:quota} with $\theta \ll 0.5$. We further investigate this phenomenon in \Cref{sec:numerical-quota}. See also \Cref{sec:numerical-unintended} for a more in-depth discussion on how/why to adjust the quota parameter $\theta$ as a function of bias factor $\rho$. 

Switching to the case of long-term utilities, which are calculated based on unbiased true values, the earlier observation that a few more positions can drastically help with the price of fairness becomes amplified: for a wide range of bias factors (e.g., $\rho\in[0.18,1]$ for $k=10$), the optimal constrained policy dominates the optimal unconstrained policy in terms of long-term utility. Furthermore, under significantly biased signals where this domination does not occur (e.g., $\rho=0.1$ or $\rho=0.2$), increasing $k$ by a small amount goes a long way: \Cref{fig:v1_dominated_capacity}(b) suggests that increasing the capacity by $11\%$ for $\rho=0.2$ and by $38\%$ for $\rho=0.1$ increases the true utility of the optimal constrained policy to more than that of the optimal unconstrained policy. We investigate how these percentages change as we switch to more relaxed notions of fairness, for example, \eqref{eq:quota} with $\theta \ll 0.5$., in \Cref{sec:numerical-quota}.

\subsubsection{Additional Notes.}
\label{sec:numeric_pandora_parity_Additional Notes}

\revcolor{
\Cref{fig-apx:normal_v1_randomization} emphasizes on the significance of randomization, by showing that the optimal policy does, indeed, randomize over 2 extreme tie-breaking rules in majority of the instances. \Cref{fig-apx:normal_v1_expost slack} demonstrates the histogram of the normalized ex-post slack in the constraint, \revcolorm{measured by the formula:
$$\frac{\sum_{i \in \ManSet} \select_i^{\policy}-\sum_{i \in \WomanSet} \select_i^{\policy}}{\sum_{i \in \ManSet} \select_i^{\policy}+\sum_{i \in \WomanSet} \select_i^{\policy}}.$$ 
}
As can be seen in the plot, there is a fast decay in the tail of the distribution. This implies, even though our optimal policy is designed to only satisfy the ex-ante constraint, its ex-post slack is also very close to 0 in most of the practical instances.}

\begin{figure}[htb]
    \centering
    \begin{subfigure}[b]{0.45\textwidth}
        \centering
        \includegraphics[width=\textwidth]{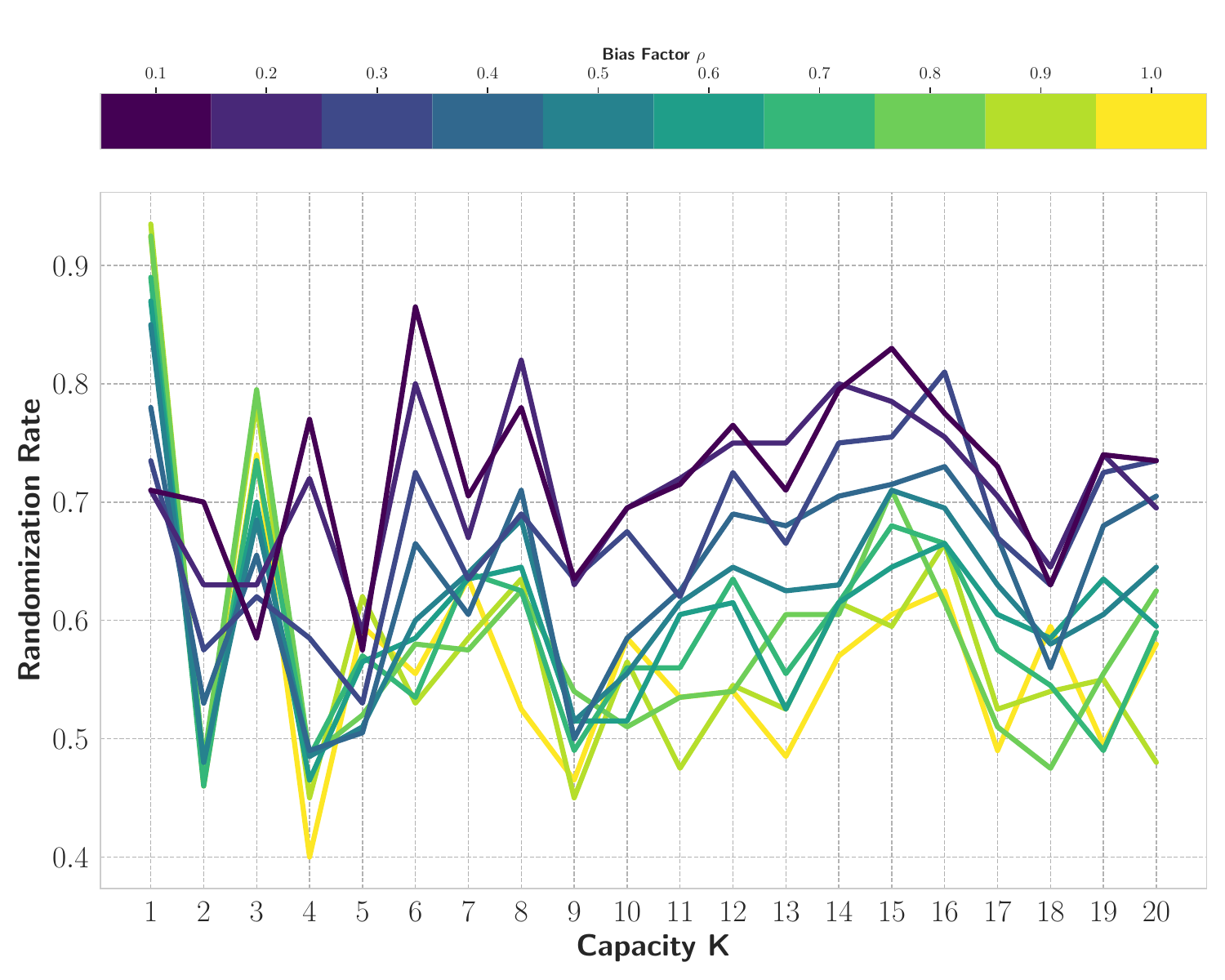}
        \caption{Fraction of instances that randomization was needed in the optimal policy.}
        \label{fig-apx:normal_v1_randomization}
    \end{subfigure}
    \hfill
    \begin{subfigure}[b]{0.45\textwidth}
        \centering
        \includegraphics[width=\textwidth]{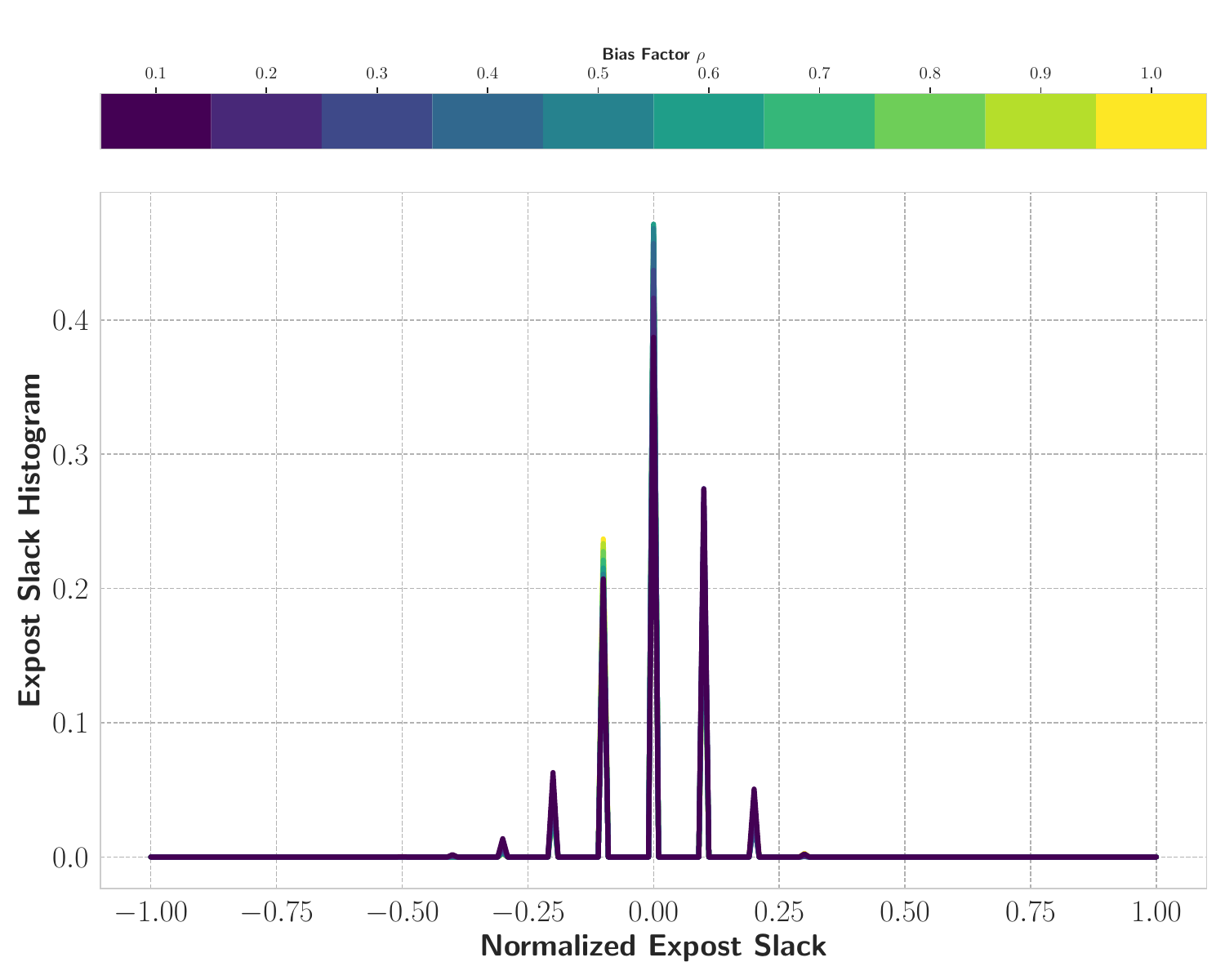}
        \caption{Histogram of normalized ex-post slack, Capacity = 20}
        \label{fig-apx:normal_v1_expost slack}
    \end{subfigure}
    \caption{\revcolor{Measuring the necessity for randomness and the ex-post statistics of the constraint slack. (a) fraction of instances where randomization was employed, and (b) empirical distribution of ex-post slack.
    \label{fig-apx:normal_v1_randomization and expost slack}}}
\end{figure}

\subsection{Average quota in selection}
\label{sec:numerical-quota}
Next, we study the average quota constraint in selection, that is, \eqref{eq:quota} for selection with parameter $\theta\in[0,1]$. Importantly, $\theta=0.5$ corresponds to demographic parity, while $\theta\in[0,0.5)$ (resp., $\theta\in(0.5,1]$) is more relaxed (resp. more restricting) than demographic parity. We repeat the same simulation scenarios as before in \Cref{sec:numerical-quota-change}, \Cref{sec:numerical-quota-true}, and \Cref{sec:numerical-quota-more-cap}. 

\subsubsection{Short-term performance -- the effect of changing capacity and bias factor.}
\label{sec:numerical-quota-change}
In \Cref{fig:v5_both_fair_unfair}, we plot short-term utilities (calculated based on biased observable signals) as a function of the quota parameter $\theta$ for different values of capacity $k$ and bias factor $\rho$. In \Cref{apx:numerical}, we also plot the short-term price of fairness ratio (\Cref{fig-apx:v5_CR_fair_unfair}) and the optimal dual adjustment $\lambda^*$ (\Cref{fig-apx:v5_lambda_fair_unfair}) as a function of $\theta$. First, we observe that the short-term utility gap between optimal constrained and unconstrained policies is increasing in $\theta$, as expected. However, we also observe that for smaller values of bias factor, for example $\rho\in[0,0.3]$, the utility decreases dramatically as $\theta$ increases. This observation suggests that when there is a significant asymmetry between the two groups, a smaller choice of $\theta\ll 0.5$ is a better choice from the perspective of short-term utility. On the other hand, for higher values of $\rho$, higher values of $\theta$ are admissible to obtain the same short-term price of fairness. See \Cref{sec:numerical-unintended} for more details on the choice of $\theta$ as a function of $\rho$ to mitigate the unintended under-allocations mentioned earlier.

\begin{figure}[htb]
    \centering
    \begin{subfigure}[b]{0.32\textwidth}
        \centering
        \includegraphics[width=\textwidth]{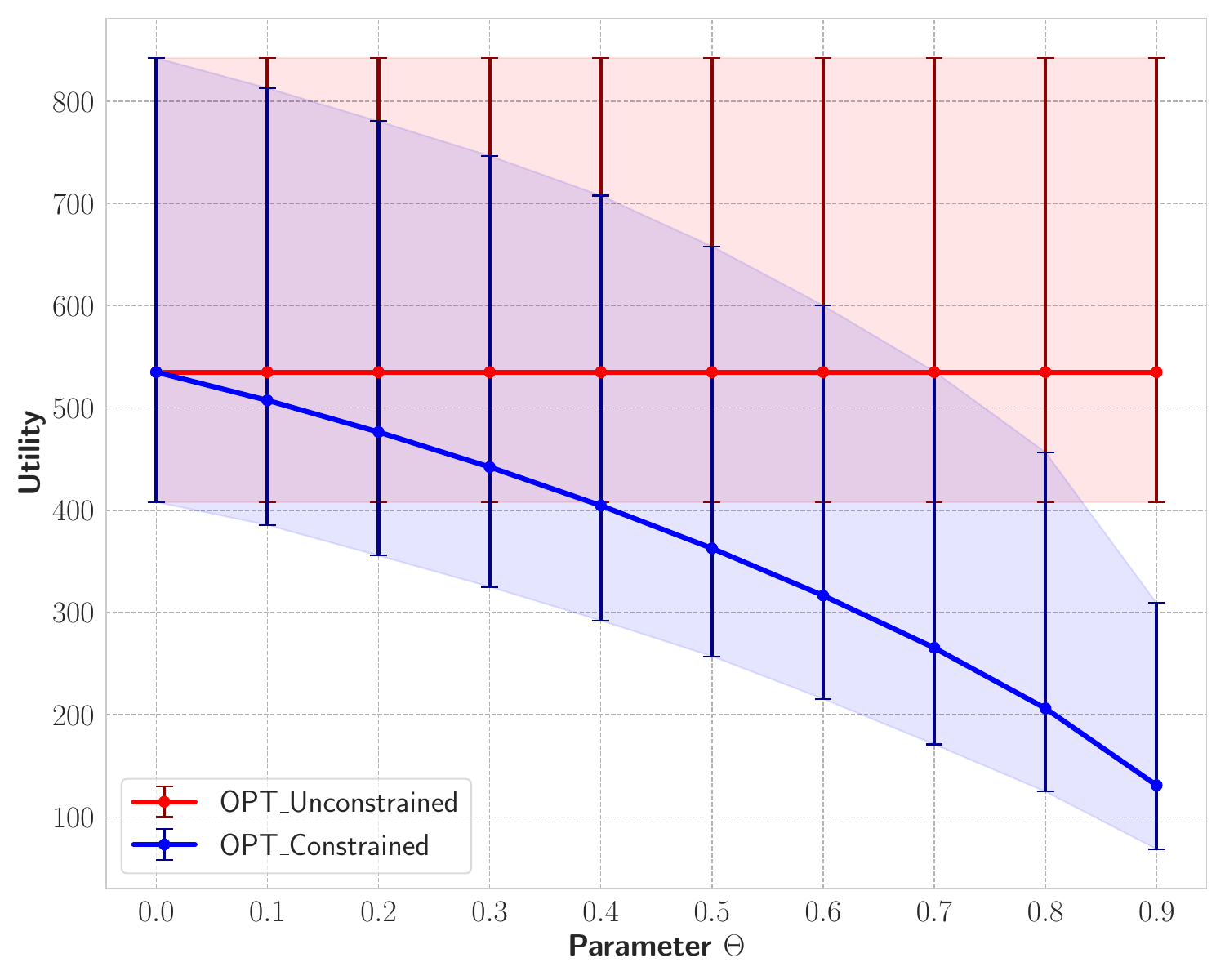}
        \caption{  Capacity = 10, bias factor = 0.1}
    \end{subfigure}
    \hfill
    \begin{subfigure}[b]{0.32\textwidth}
        \centering
        \includegraphics[width=\textwidth]{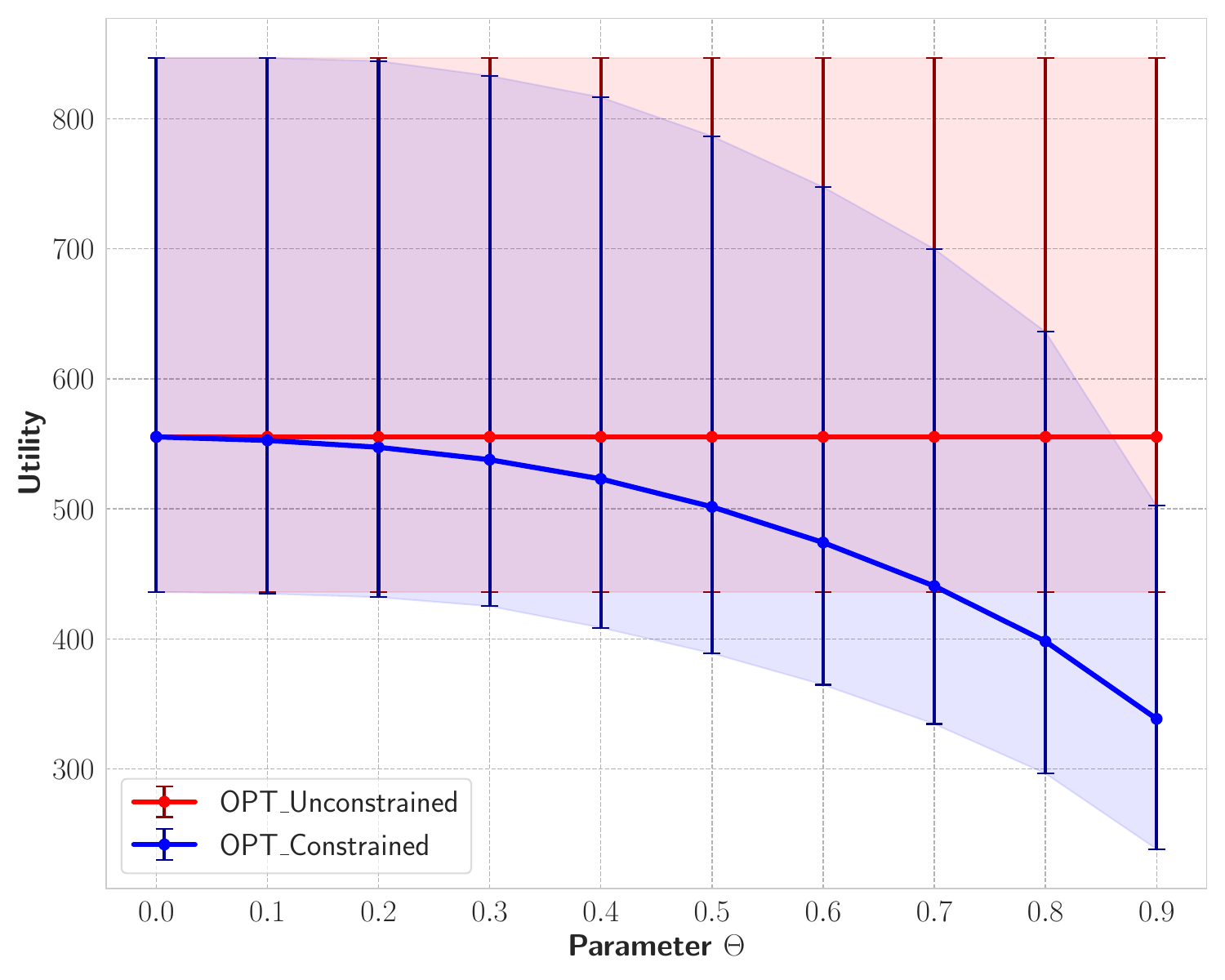}
        \caption{  Capacity = 10, bias factor = 0.5}
    \end{subfigure}
    \hfill
    \begin{subfigure}[b]{0.32\textwidth}
        \centering
        \includegraphics[width=\textwidth]{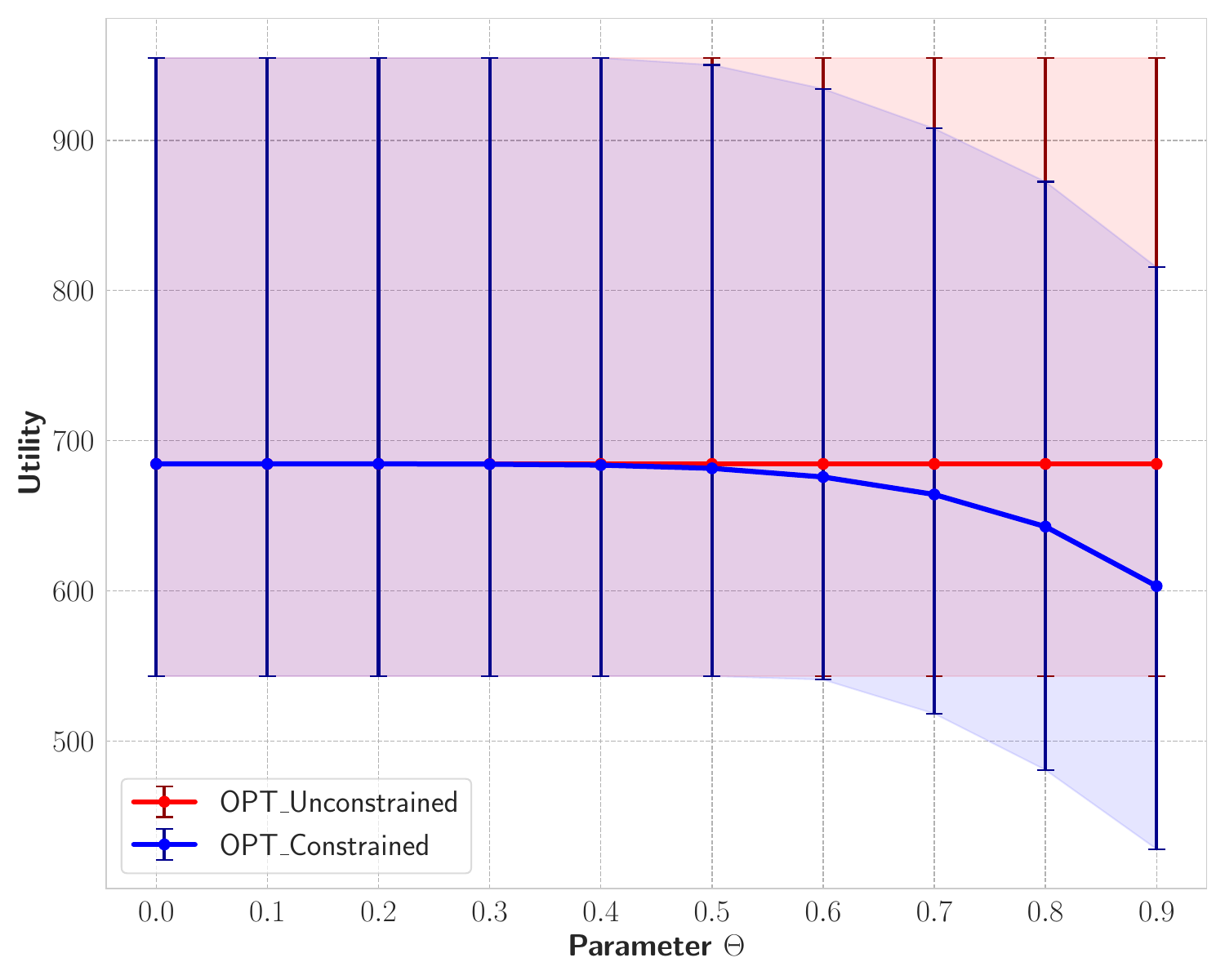}
        \caption{  Capacity = 10, bias factor = 1.0}
    \end{subfigure}
    \hfill
    \begin{subfigure}[b]{0.32\textwidth}
        \centering
        \includegraphics[width=\textwidth]{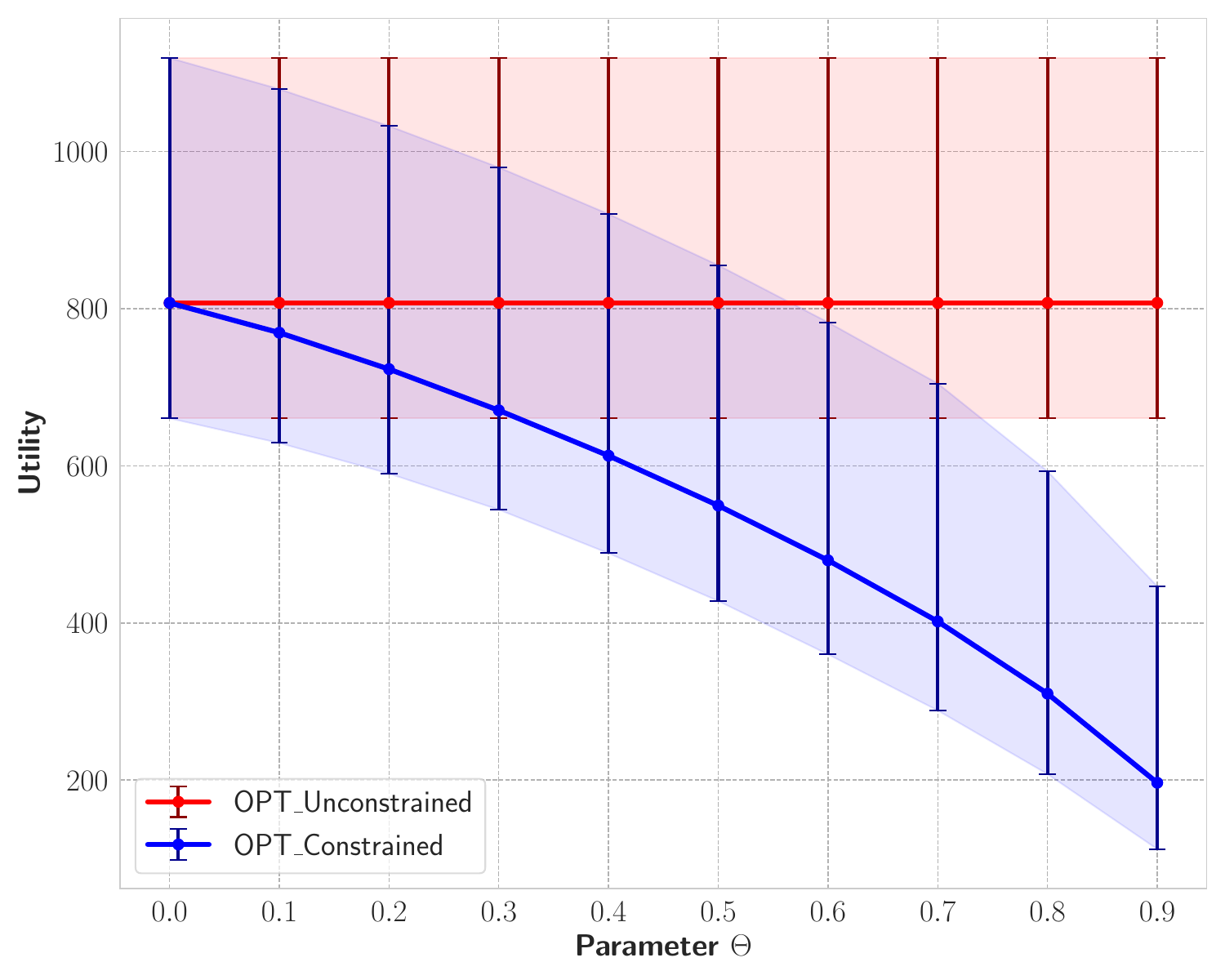}
        \caption{  Capacity = 20, bias factor = 0.1}
    \end{subfigure}
    \hfill
    \begin{subfigure}[b]{0.32\textwidth}
        \centering
        \includegraphics[width=\textwidth]{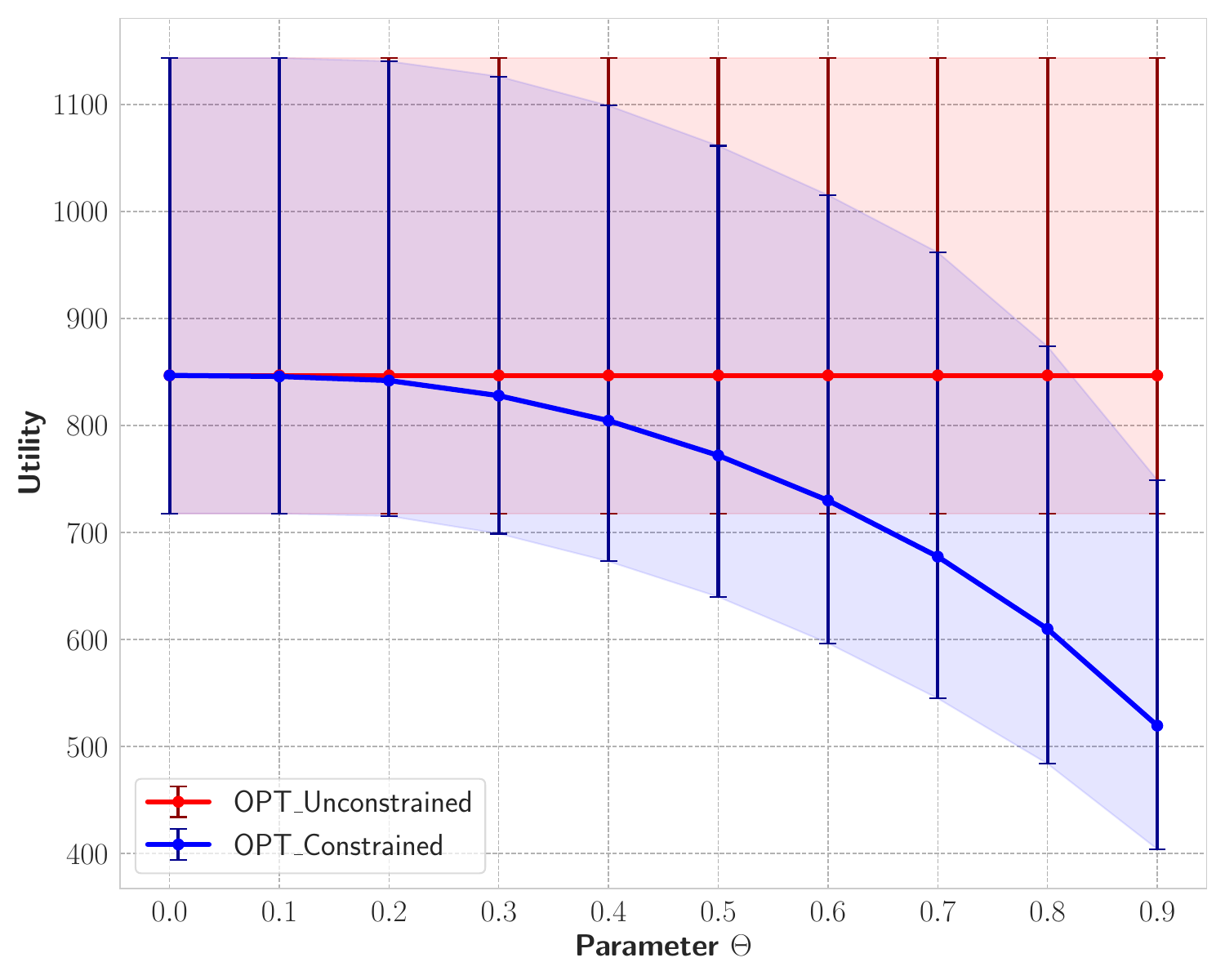}
        \caption{  Capacity = 20, bias factor = 0.5}
    \end{subfigure}
    \hfill
    \begin{subfigure}[b]{0.32\textwidth}
        \centering
        \includegraphics[width=\textwidth]{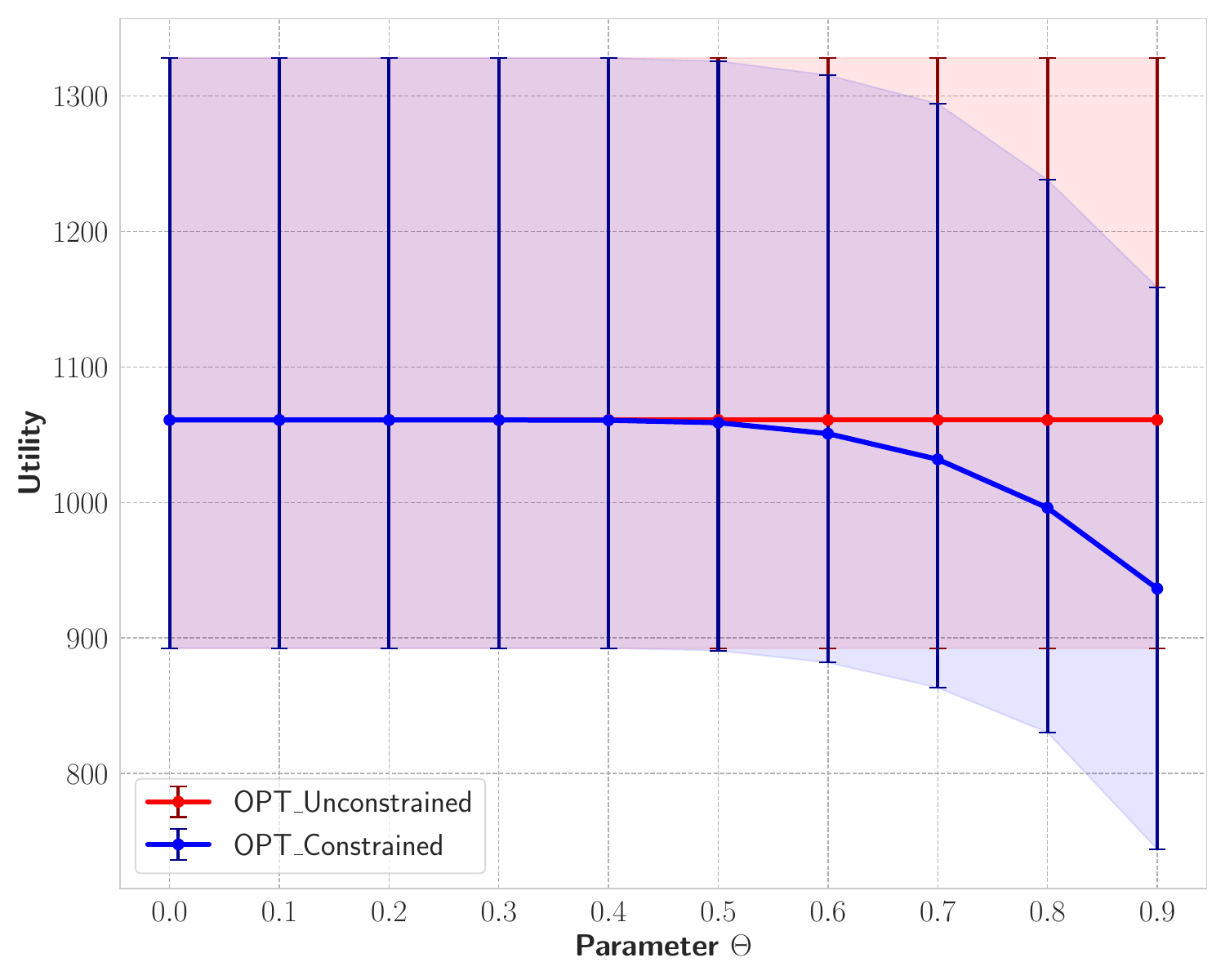}
        \caption{  Capacity = 20, bias factor = 1.0}
    \end{subfigure}
    \hfill
    \vspace{-1mm}
    \caption{Comparing the short-term utilities of optimal unconstrained (red) and constrained (blue) policies for the average quota in selection, for different capacities $\boldsymbol{k\in[1:20]}$ and bias factors $\boldsymbol{\rho\in\{0.1,\cdots,1\}}$; both curves are functions of the quota parameter $\boldsymbol{\theta\in[0,1]}$ in \eqref{eq:quota} ($\boldsymbol{\theta=0.5}$ corresponds to demographic parity).}
    \label{fig:v5_both_fair_unfair}
\end{figure}

\subsubsection{Long-term performance -- the effect of changing capacity and bias factor.}
\label{sec:numerical-quota-true} 
We now consider a setting similar to \Cref{sec:numerical-parity-true} with biased signals and unbiased values and study the long-term performance of our policies. See \Cref{fig:v6_both_fair_unfair} for a comparison of long-term utilities (calculated based on the true values) under optimal constrained and unconstrained policies. 
In \Cref{apx:numerical}, we further plot the price of fairness ratio with respect to the true values as a function of $\theta$ (\Cref{fig-apx:v6_diff_CR_fair_unfair}). As before, the optimal constrained policy dominates the optimal unconstrained policy with respect to the true values in a wide range of parameters. Moreover, as can be seen in all these graphs, if the bias in the signals decreases (that is, the bias factor $\rho$ increases), the range of parameter $\theta$ in which the domination occurs expands. Our results suggest that (i) adding an average quota with parameter $\theta\in[0,1]$, similar to demographic parity, can help increase long-term utilities (with respect to true values), and (ii) tuning parameter $\theta$ based on the bias in the signals can drastically amplify this effect. We further discuss this managerial insight in \Cref{sec:numerical-unintended}.

\begin{figure}[htb]
    \centering
    \begin{subfigure}[b]{0.32\textwidth}
        \centering
        \includegraphics[width=\textwidth]{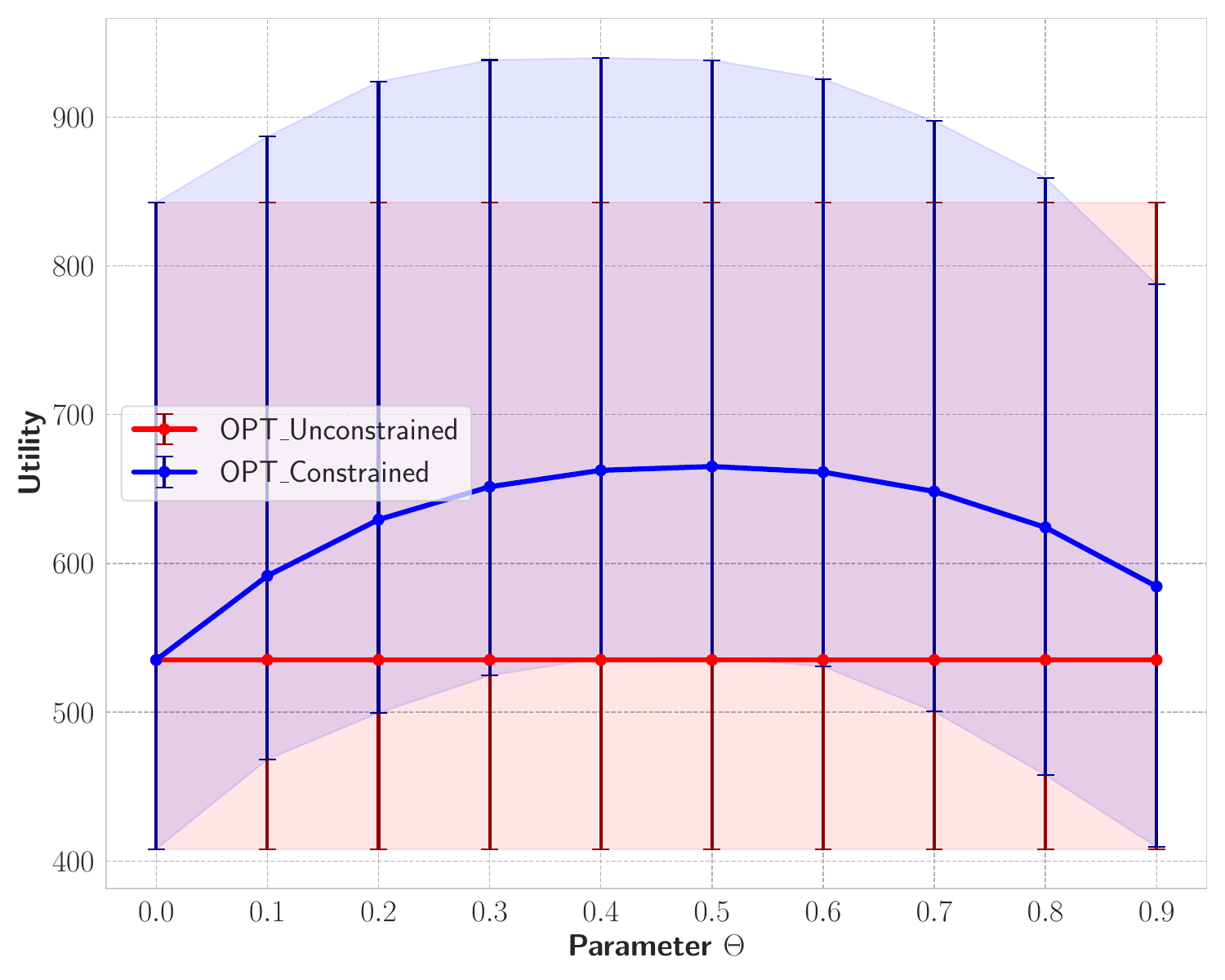}
        \caption{  Capacity = 10, bias factor = 0.1}
    \end{subfigure}
    \hfill
    \begin{subfigure}[b]{0.32\textwidth}
        \centering
        \includegraphics[width=\textwidth]{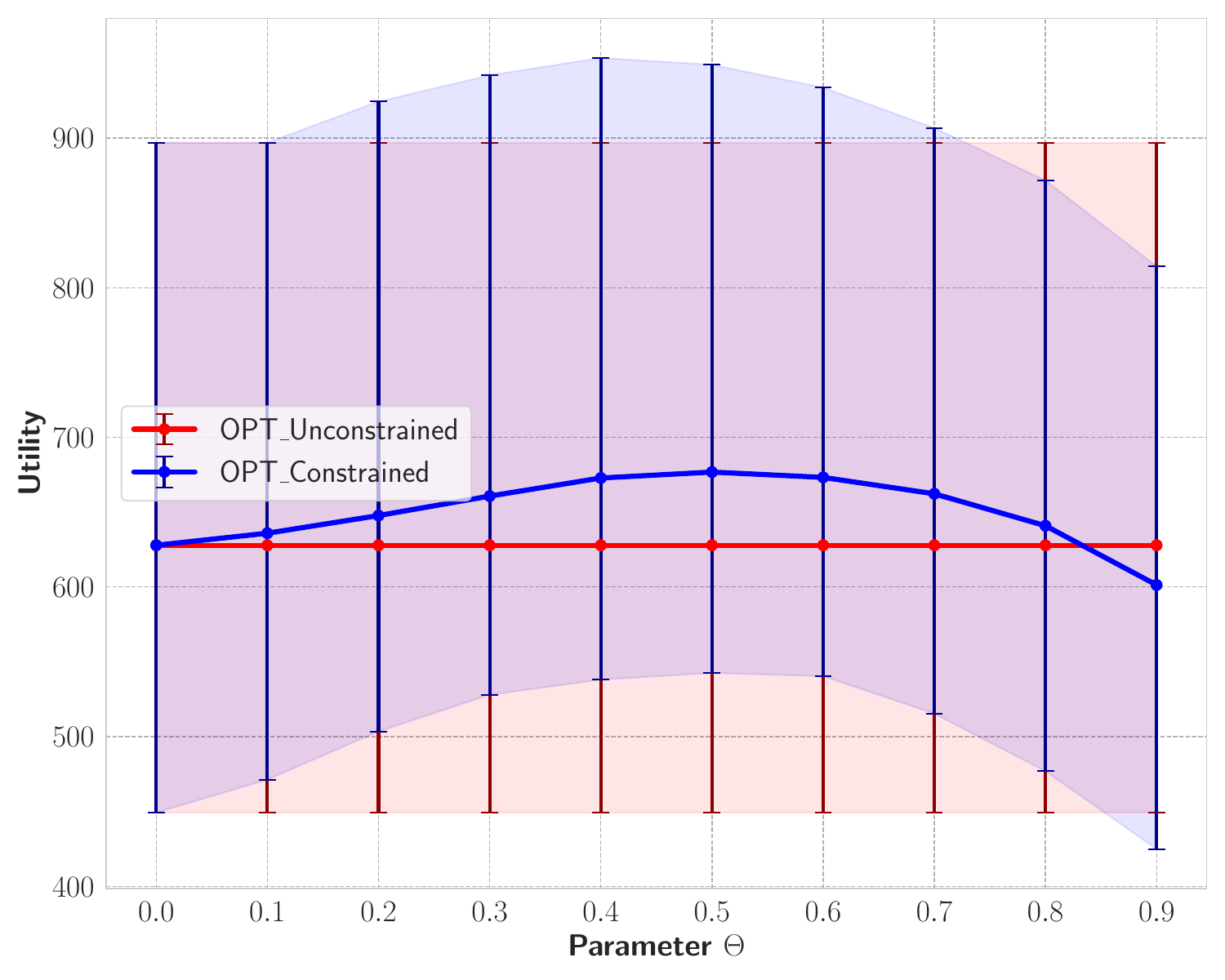}
        \caption{  Capacity = 10, bias factor = 0.5}
    \end{subfigure}
    \hfill
    \begin{subfigure}[b]{0.32\textwidth}
        \centering
        \includegraphics[width=\textwidth]{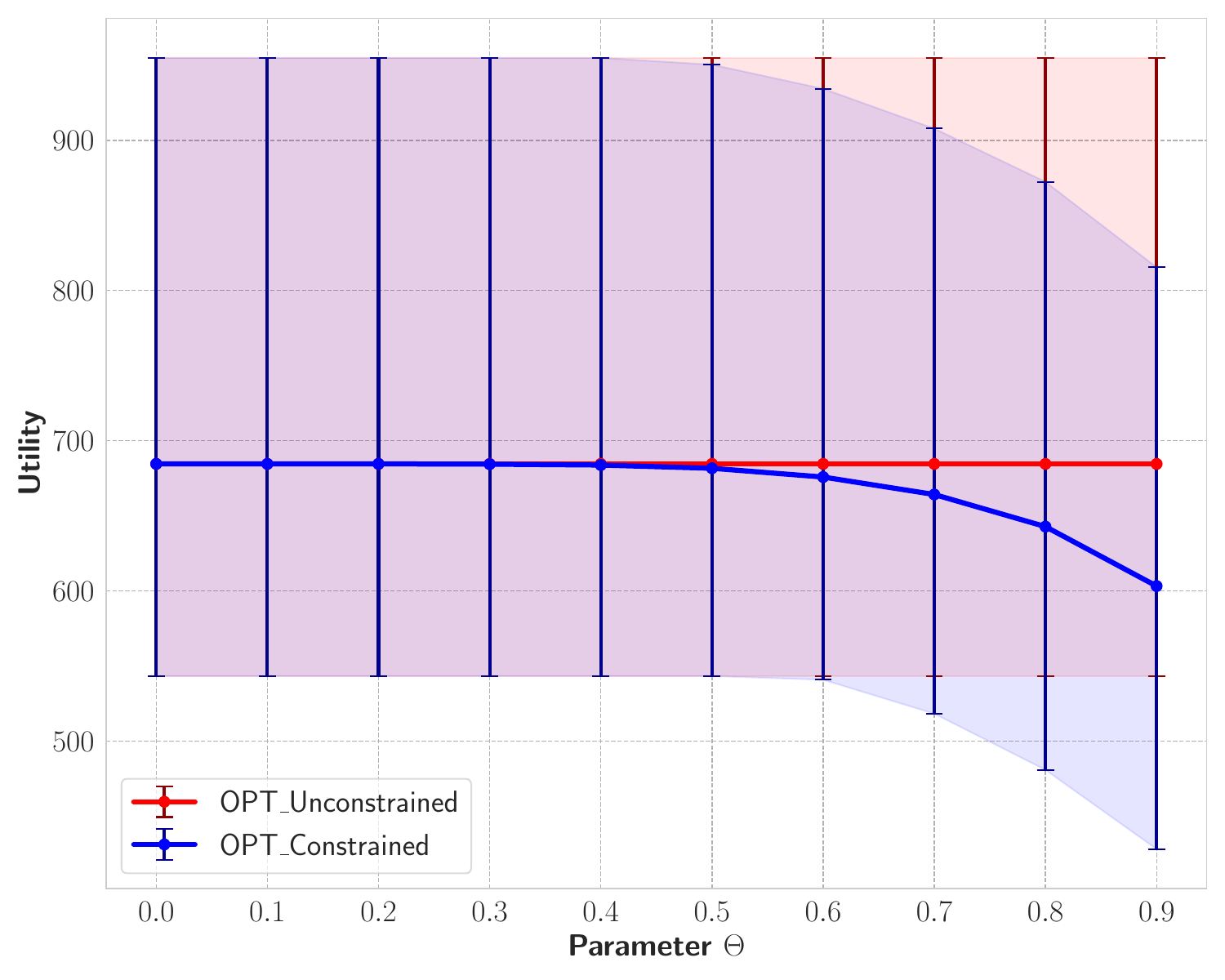}
        \caption{  Capacity = 10, bias factor = 1}
    \end{subfigure}
    \hfill
    \begin{subfigure}[b]{0.32\textwidth}
        \centering
        \includegraphics[width=\textwidth]{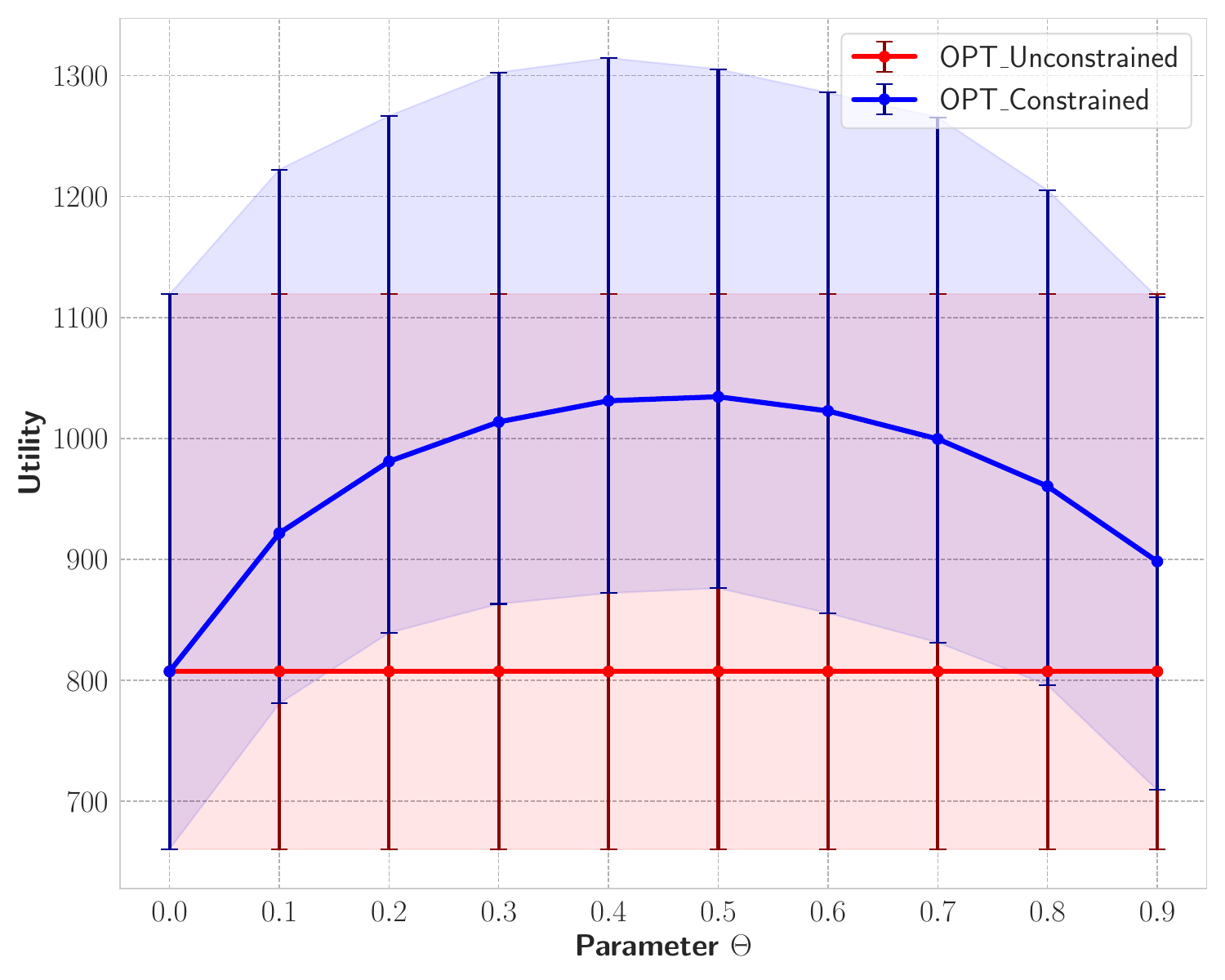}
        \caption{  Capacity = 20, bias factor = 0.1}
    \end{subfigure}
    \hfill
    \begin{subfigure}[b]{0.32\textwidth}
        \centering
        \includegraphics[width=\textwidth]{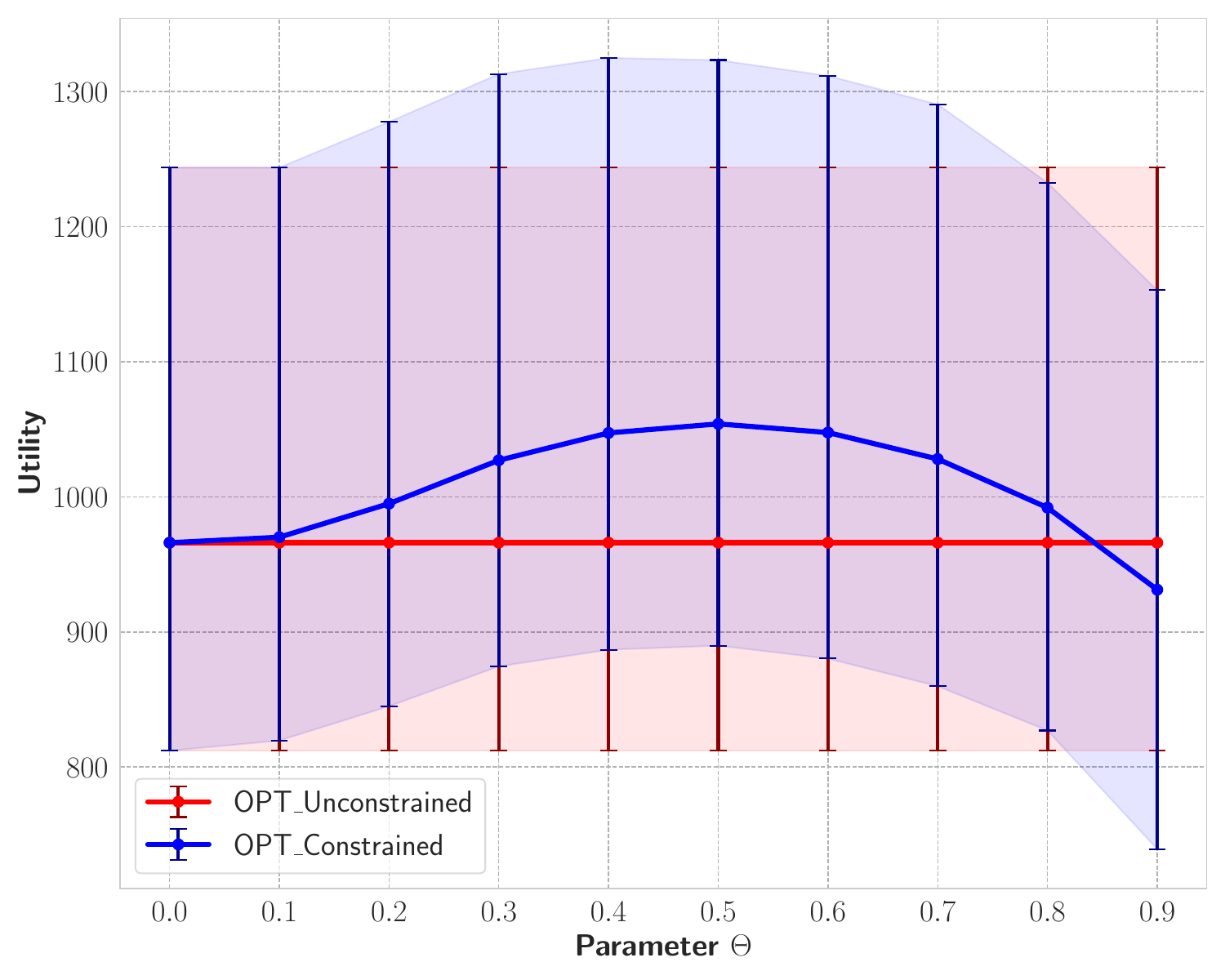}
        \caption{  Capacity = 20, bias factor = 0.5}
    \end{subfigure}
    \hfill
    \begin{subfigure}[b]{0.32\textwidth}
        \centering
        \includegraphics[width=\textwidth]{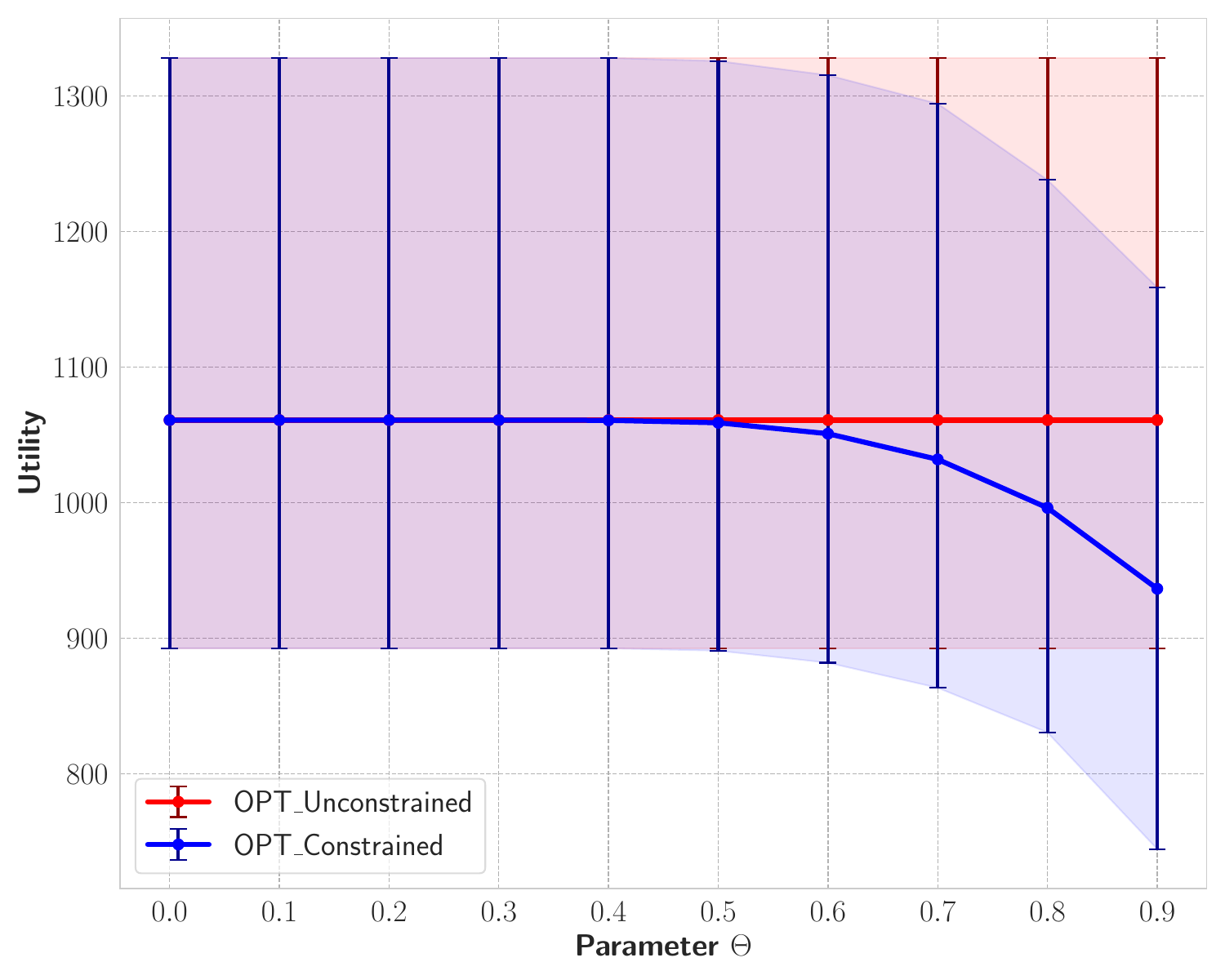}
        \caption{  Capacity = 20, bias factor = 1}
    \end{subfigure}
    \hfill
    \vspace{-2mm}
    \caption{Comparing the long-term utilities of optimal unconstrained (red) and constrained (blue) policies for the average quota in selection, for different capacities $\boldsymbol{k\in[1:20]}$ and bias factors $\boldsymbol{\rho\in\{0.1,\cdots,1\}}$; both curves are functions of the quota parameter $\boldsymbol{\theta\in[0,1]}$ in \eqref{eq:quota} ($\boldsymbol{\theta=0.5}$ corresponds to demographic parity).}
    \label{fig:v6_both_fair_unfair}
\end{figure}

\subsubsection{A few positions more.}
\label{sec:numerical-quota-more-cap}
We now consider the setting in \Cref{sec:numerical-few-pos}, but this time considering the quota selection constraint as in \eqref{eq:quota}, and study the effect of enhancing the capacity. See \Cref{fig:v5_dominated_capacity} for the effect on short-term utilities and \Cref{fig:v6_dominated_capacity} for the effect on long-term utilities. Comparing these two graphs with \Cref{fig:v1_dominated_capacity}, we observe that (i) for small values of $\theta$ (e.g., $\theta\leq 0.5$ for $\rho=0.9$ and $\theta\leq 0.3$ for $\rho=0.7$ ), increasing the capacity by less than $5\%$ is enough to ensure that the optimal constrained policy dominates the optimal unconstrained policy in terms of short-term utilities; (ii) the impact of capacity enhancement increases drastically when measuring the performance of policies based on their long-term utilities. For example, with $\theta\leq 0.7$ for $\rho=0.7$ and $\theta\leq 0.5$ for $\rho=0.9$, increasing the capacity by less than $2.5\%$ is enough to ensure that the long-term utility of the optimal constrained policy dominates the long-term utlitiy of the optimal unconstrained policy.
\begin{figure}[htb]
    \centering
    \begin{subfigure}[b]{0.32\textwidth}
        \centering
        \includegraphics[width=\textwidth]{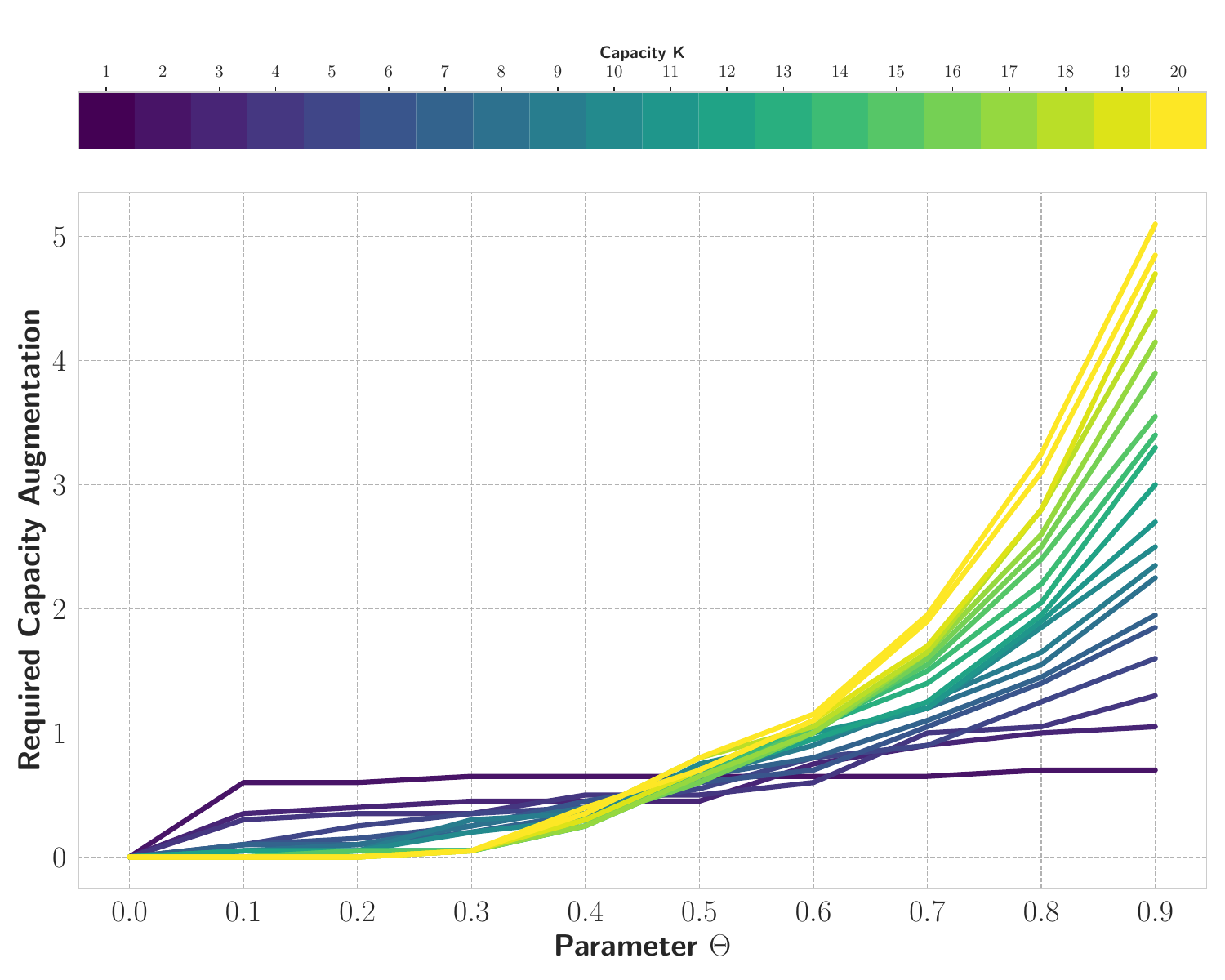}
        \caption{  Bias factor = 0.9}
    \end{subfigure}
    \begin{subfigure}[b]{0.32\textwidth}
        \centering
        \includegraphics[width=\textwidth]{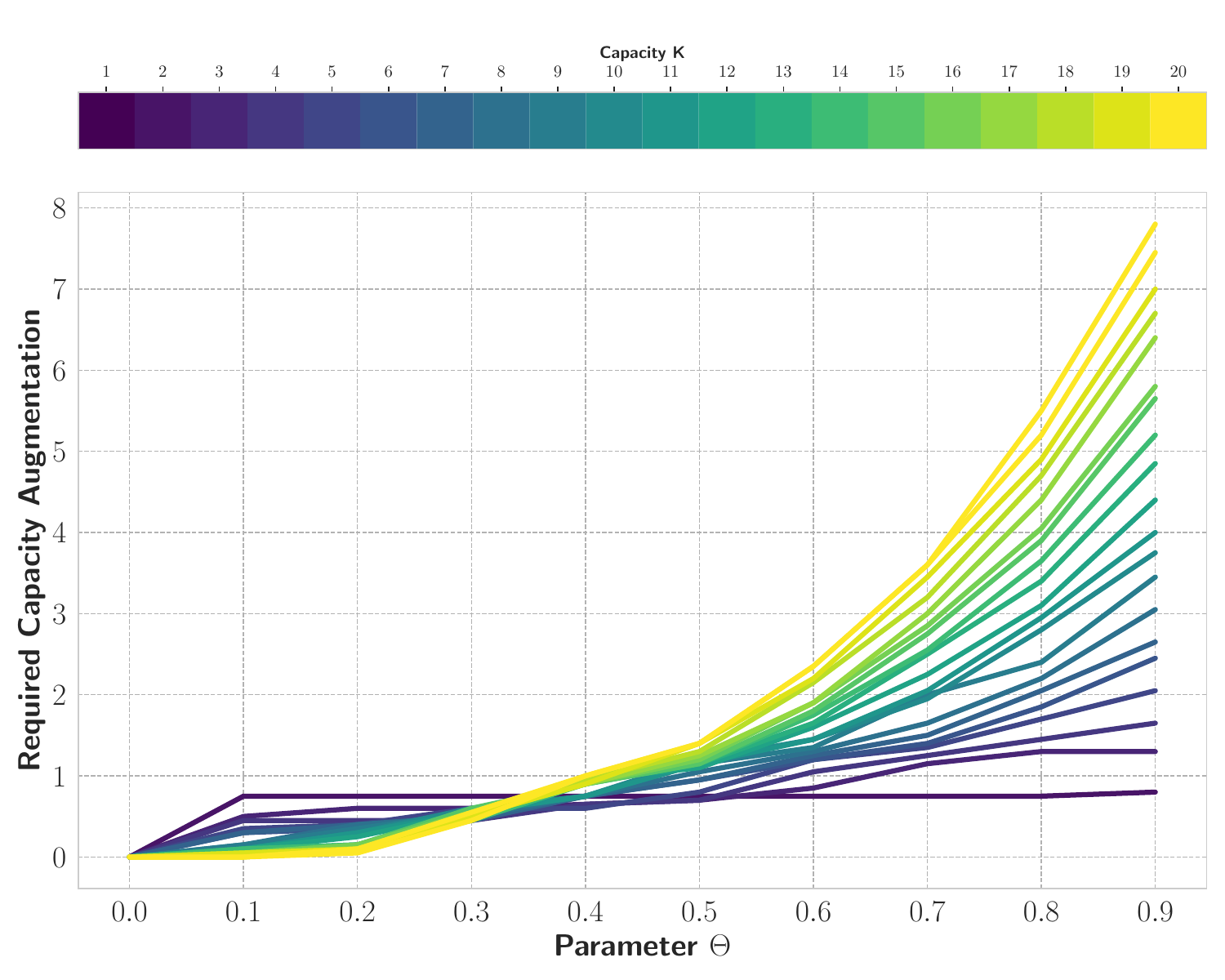}
        \caption{  Bias factor =  0.7}
    \end{subfigure}
    \begin{subfigure}[b]{0.32\textwidth}
        \centering
        \includegraphics[width=\textwidth]{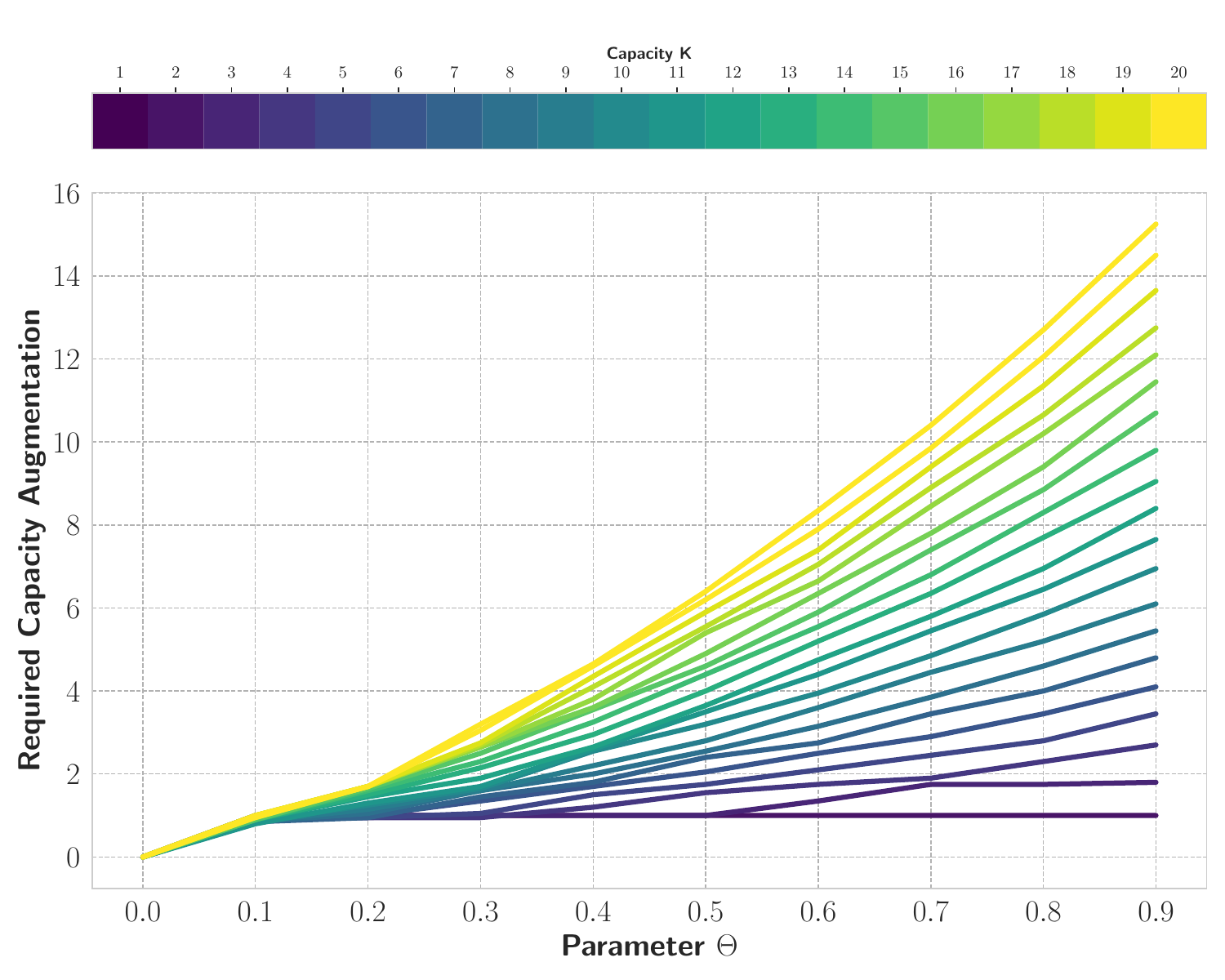}
        \caption{  Bias factor = 0.3}
    \end{subfigure}
    \hfill
    \caption{The required extra capacity to compensate for the short-term utility reduction due to imposing the average quota in selection, as a function of the quota parameter $\boldsymbol{\theta\in[0,1]}$ in \eqref{eq:quota} ($\boldsymbol{\theta=0.5}$ corresponds to demographic parity) for different values of bias factor $\boldsymbol{\rho}$ and capacity $\boldsymbol{k}$.}
    \label{fig:v5_dominated_capacity}
\end{figure}

\begin{figure}[htb]
    \centering
    \begin{subfigure}[b]{0.32\textwidth}
        \centering
        \includegraphics[width=\textwidth]{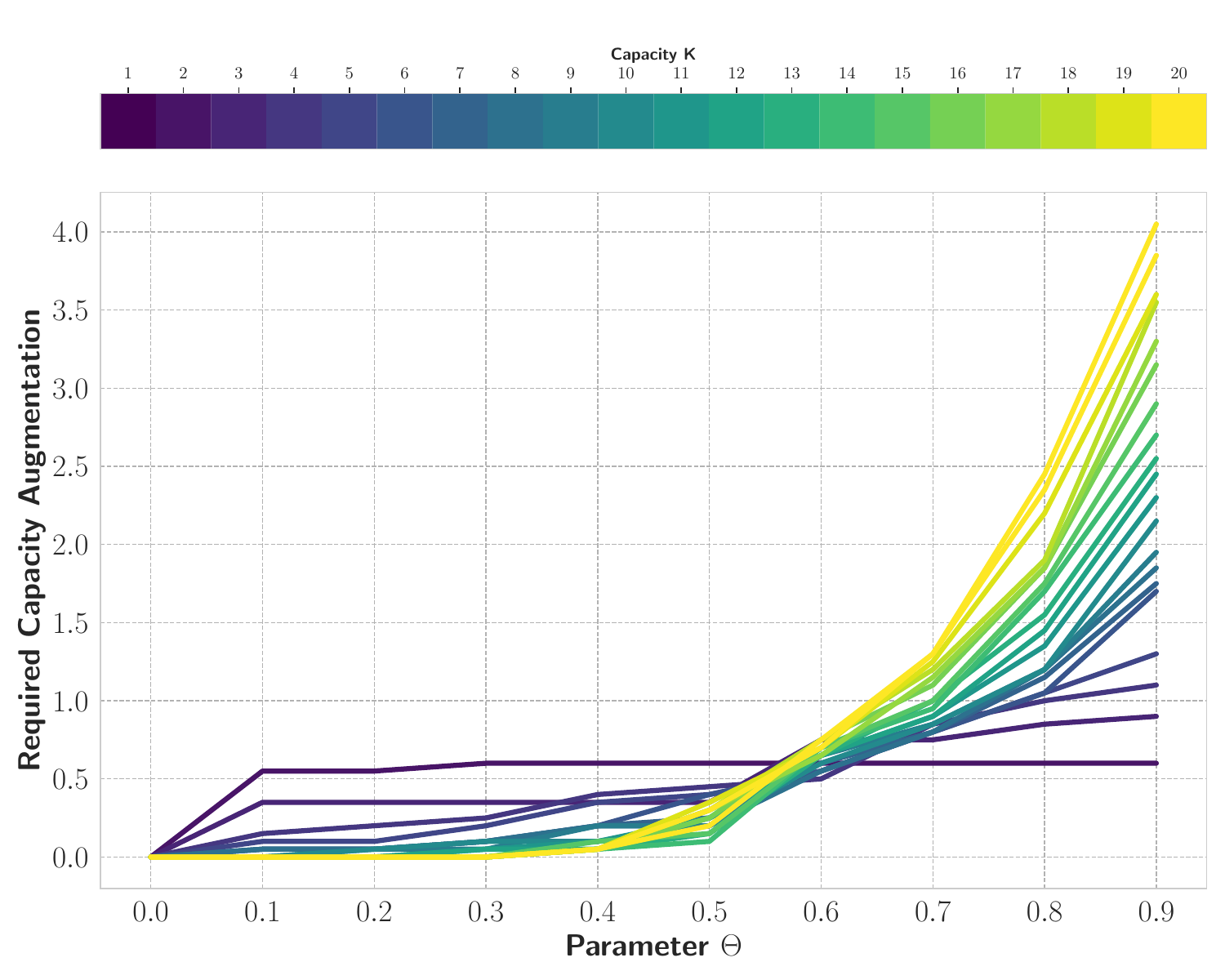}
        \caption{  Bias factor = 0.9}
    \end{subfigure}
    \hfill
    \begin{subfigure}[b]{0.32\textwidth}
        \centering
        \includegraphics[width=\textwidth]{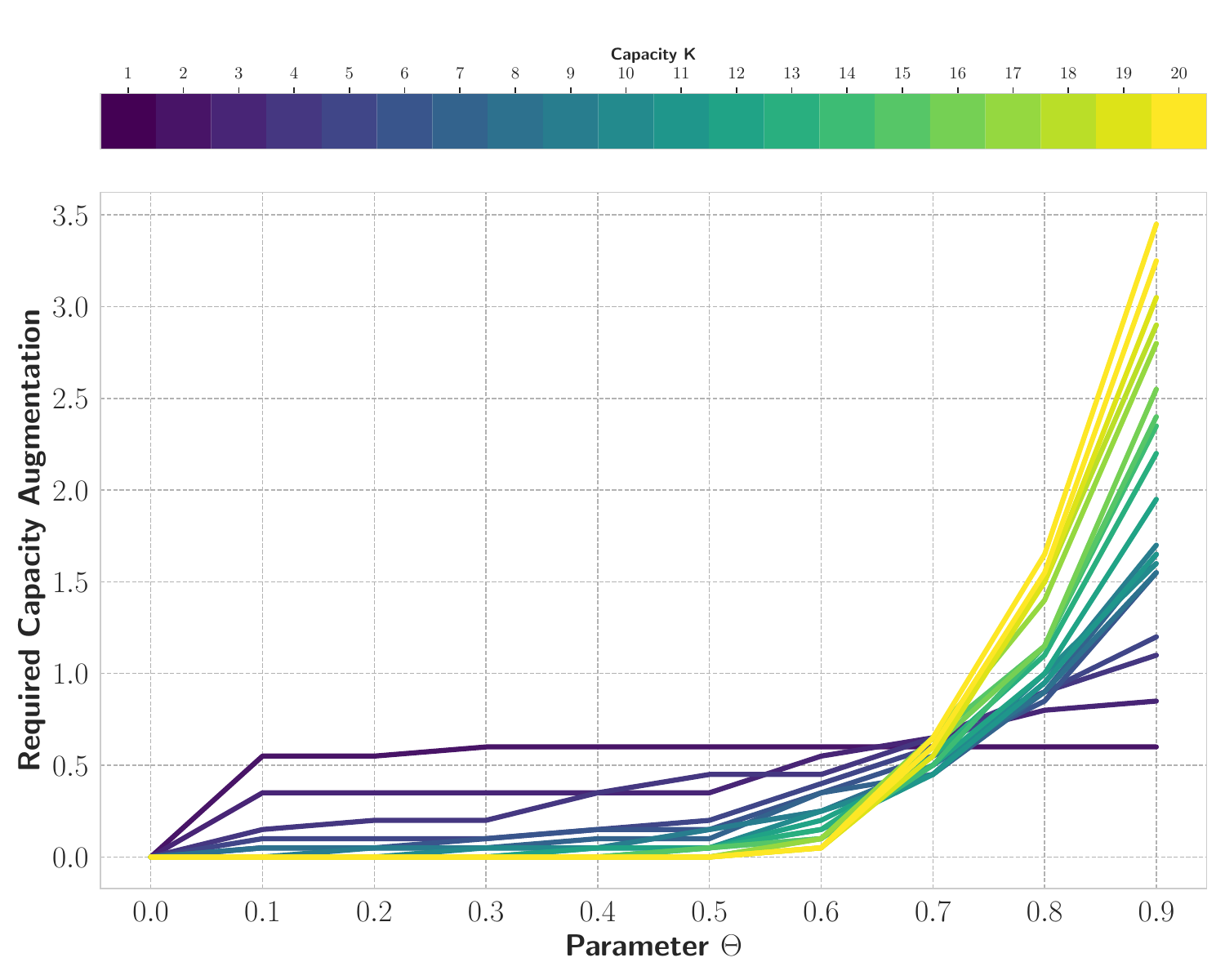}
        \caption{  Bias factor = 0.7}
    \end{subfigure}
    \hfill
    \begin{subfigure}[b]{0.32\textwidth}
        \centering
        \includegraphics[width=\textwidth]{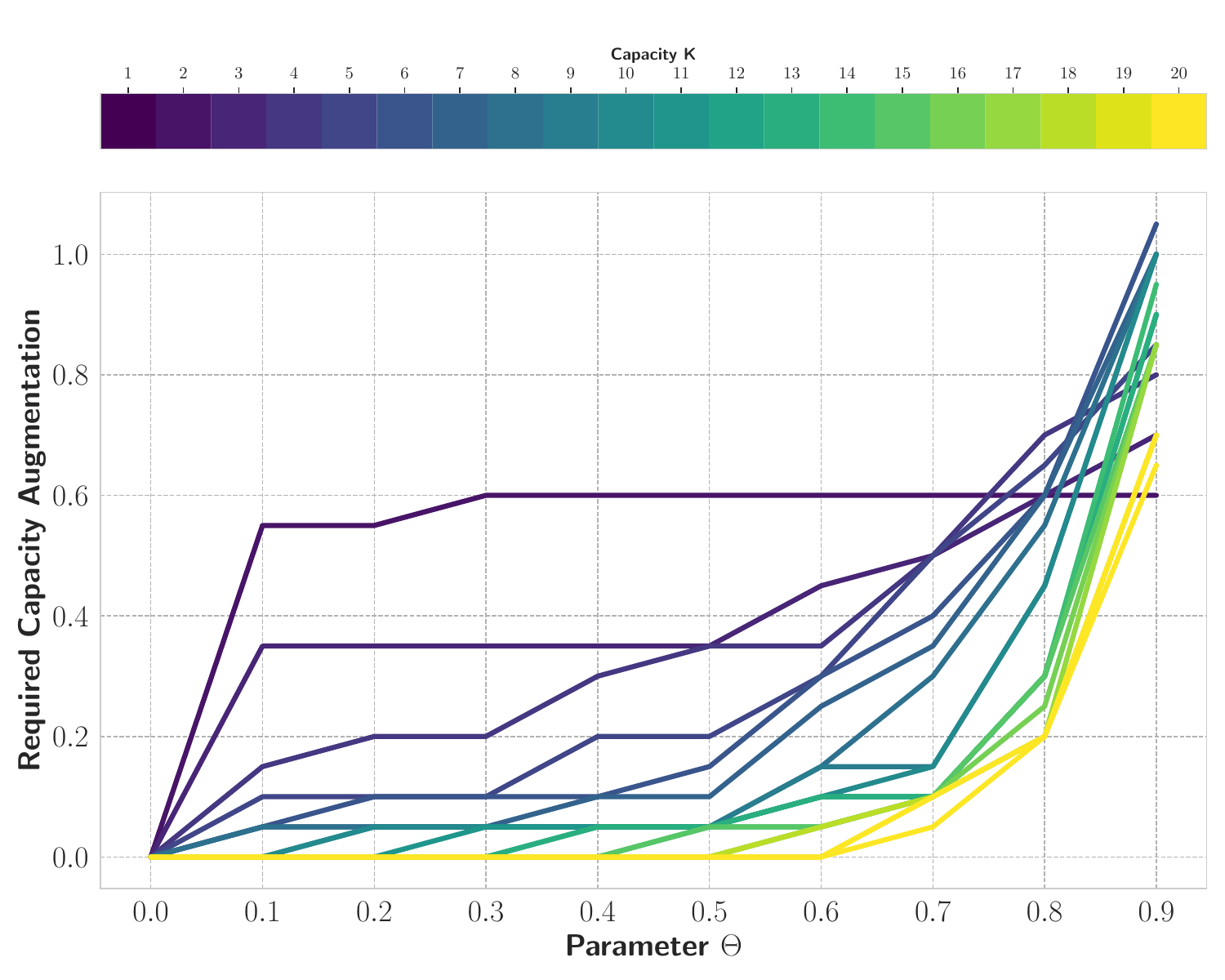}
        \caption{  Bias factor = 0.3}
    \end{subfigure}
    \hfill
    \caption{The required extra capacity to compensate for the long-term utility reduction due to imposing the average quota in selection, as a function of the quota parameter $\boldsymbol{\theta\in[0,1]}$ in \eqref{eq:quota} ($\boldsymbol{\theta=0.5}$ corresponds to demographic parity) for different values of bias factor $\boldsymbol{\rho}$ and capacity $\boldsymbol{k}$.}
    \label{fig:v6_dominated_capacity}
\end{figure}

\subsubsection{Additional Notes.}
\label{sec:numeric_pandora_quota_Additional Notes}
\revcolorm{Similar to \Cref{sec:numeric_pandora_parity_Additional Notes} we plot the normalized ex-post slack in \Cref{fig-apx:normal_v5_expost slack}, measured by the formula:
$$\frac{\theta(\sum_{i \in \ManSet} \select_i^{\policy})-(1-\theta)(\sum_{i \in \WomanSet} \select_i^{\policy})}{\theta(\sum_{i \in \ManSet} \select_i^{\policy})+(1-\theta)(\sum_{i \in \WomanSet} \select_i^{\policy})}.$$ 
As can be seen, for each $\theta$, the distribution has a low variance. Additionally, the reason the histogram of some of the smaller $\theta$ values are skewed toward left is due to the fact that the quota constraint becomes less binding as we decrease $\theta$, and the optimal unconstrained policy itself may be a feasible policy.
\begin{figure}[htb]
    \centering
    \includegraphics[width=0.6\textwidth]{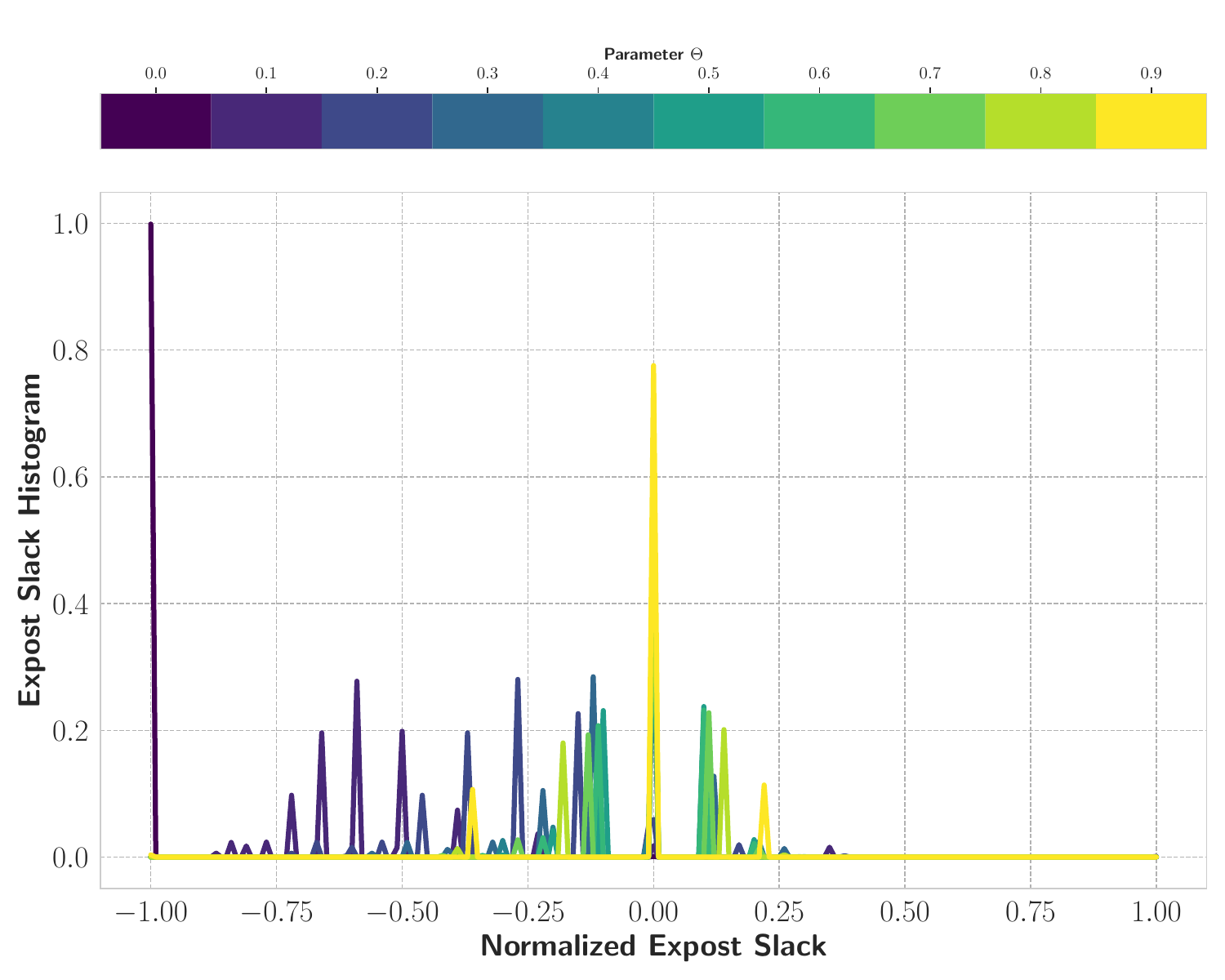}
    \caption{Histogram of normalized ex-post slack, Capacity = 20, Bias Factor = 0.7
    \label{fig-apx:normal_v5_expost slack}}
\end{figure}
}

\subsection{Average budget for subsidization}
\label{sec:numerical-budget}
We finally study the effect of the average budget to subsidize the hiring expenses of underprivileged applicants. To model this, we consider two groups of candidates $\WomanSet$ and $\ManSet$. Given an average budget $b$, we consider a variant of \eqref{eq:budget} in selection where the decision maker has to select no more than $b$ candidates from $\WomanSet$ in expectation, while it has to satisfy an overall ex-post capacity constraint $k$ among all individuals.


Clearly, any non-zero budget increases the search utility. Therefore, we define ``gain from budget'' as the ratio of the utility of the optimal constrained policy with the budget $b$ to that of the optimal policy without any budget --- which means that it cannot select anyone from group $\WomanSet$. See \Cref{fig:v4_CR_fair_unfair} to learn how the gain from the budget increases as a function of $b$, for various parameter choices for bias factor $\rho$ and capacity $k$. In \Cref{apx:numerical}, we further compare the utility of the optimal constrained policy with budget $b$ versus that of the optimal policy with unlimited budget (\Cref{fig-apx:v4_both_fair_unlimited}) as well as the optimal dual adjustment $\lambda^*$ (\Cref{fig-apx:v4_lambda_fair_unfair}), both as a function of the budget $b$.

\vspace{-1mm}
\begin{figure}[htb]
    \centering
    \begin{subfigure}[b]{0.32\textwidth}
        \centering
        \includegraphics[width=\textwidth]{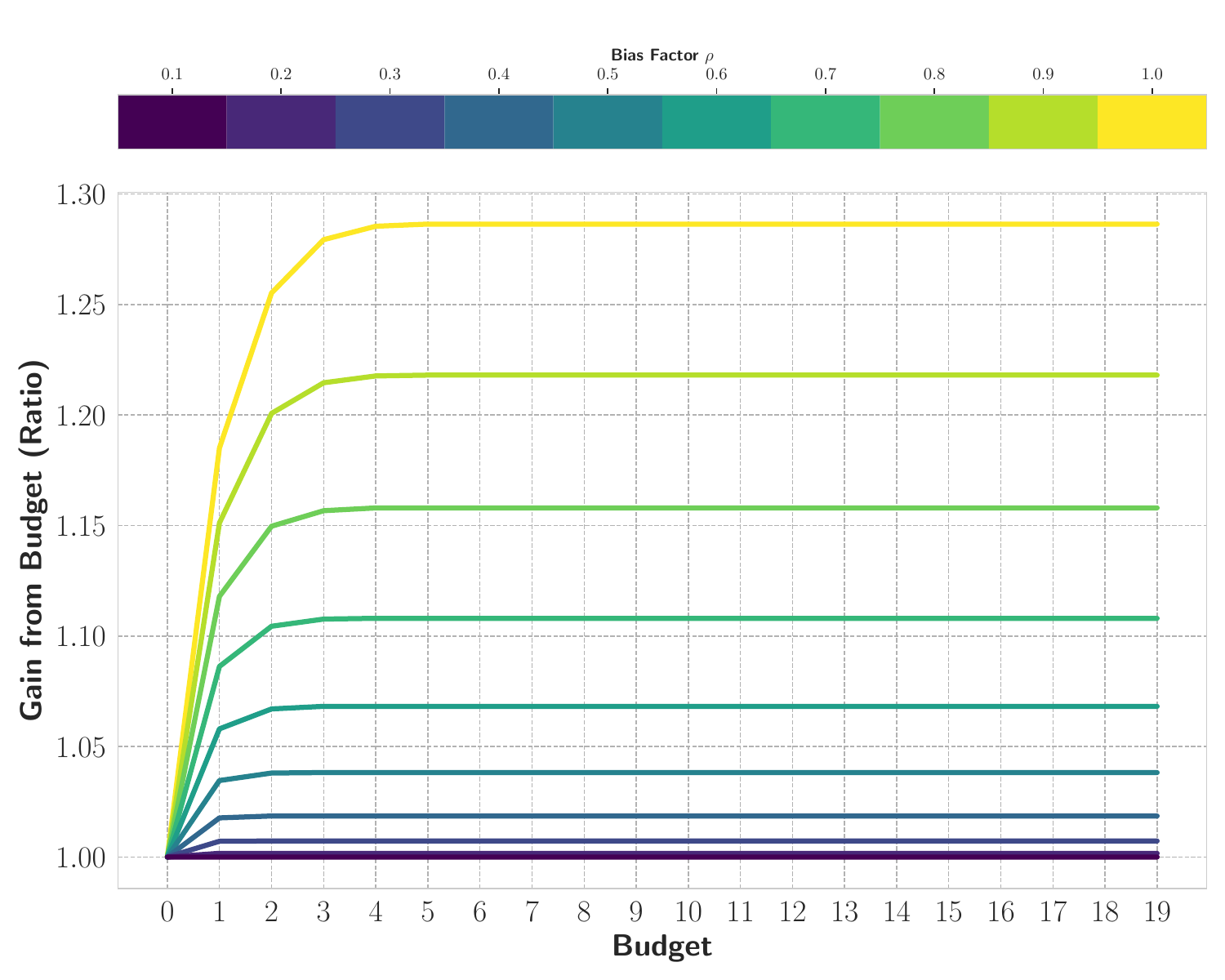}
        \caption{  Capacity = 5}
    \end{subfigure}
    \hfill
        \begin{subfigure}[b]{0.32\textwidth}
        \centering
        \includegraphics[width=\textwidth]{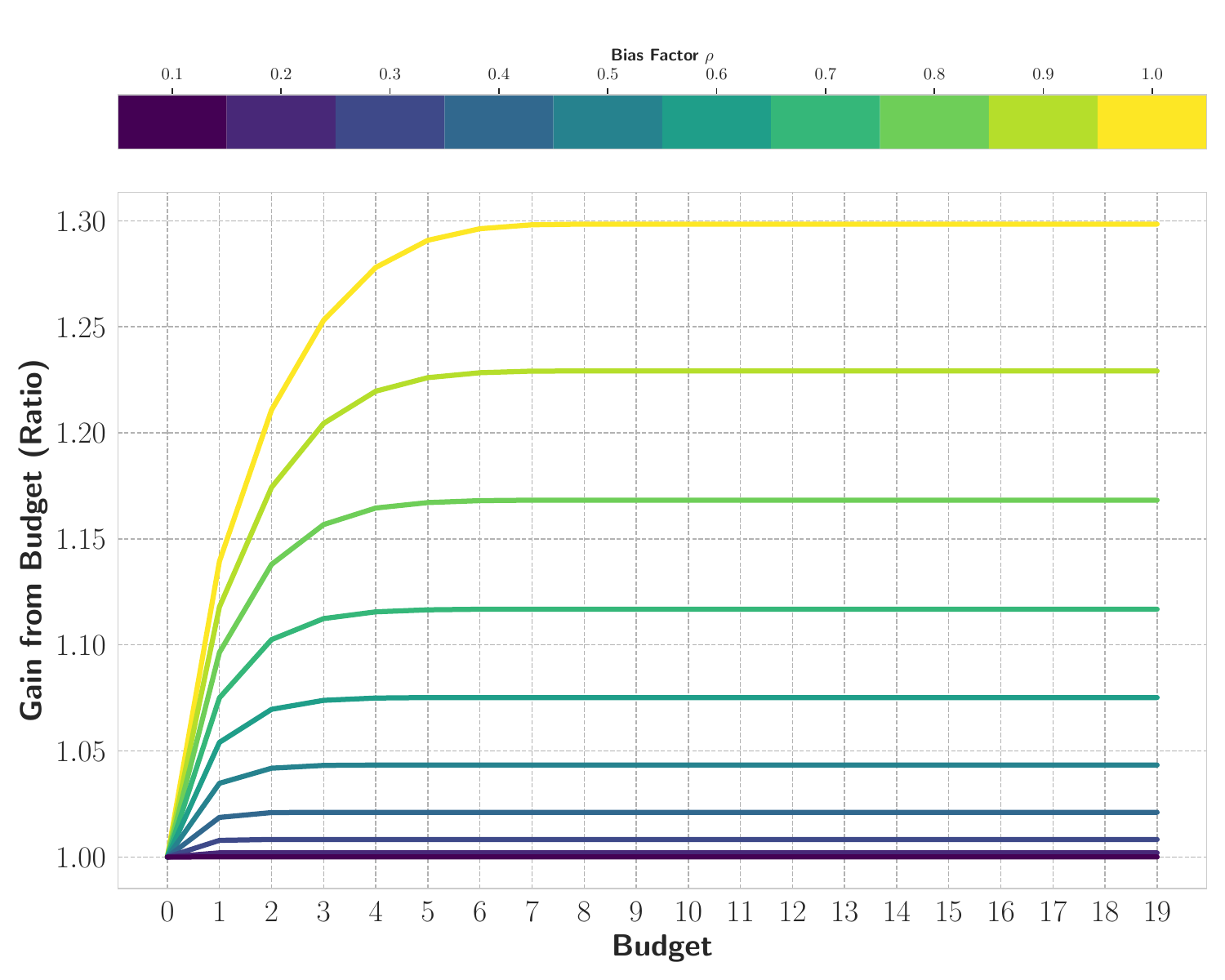}
        \caption{  Capacity = 10}
    \end{subfigure}
    \hfill
    \begin{subfigure}[b]{0.32\textwidth}
        \centering
        \includegraphics[width=\textwidth]{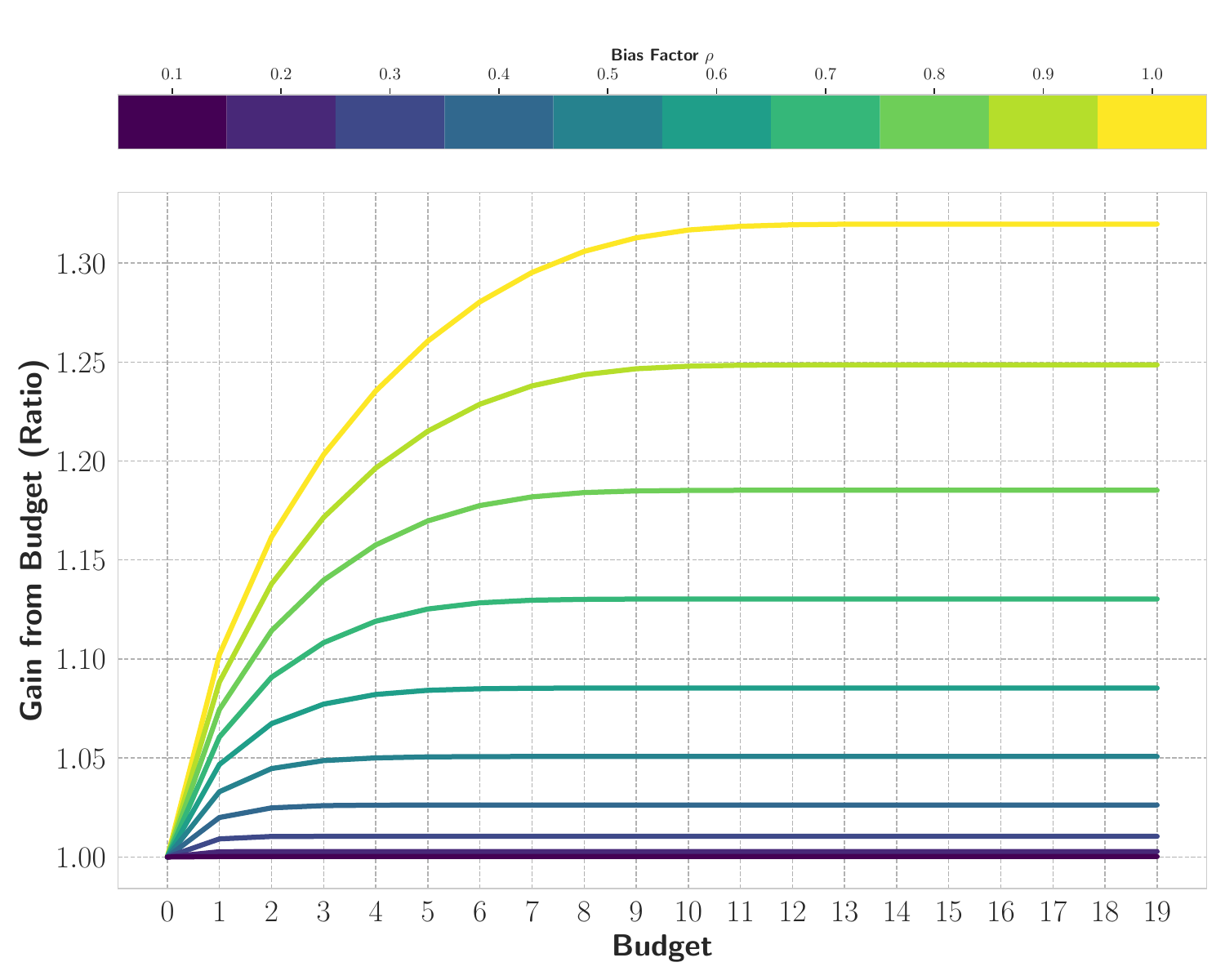}
        \caption{  Capacity = 20}
    \end{subfigure}
    \caption{The the gain from budget, i.e., the ratio of utilities of optimal constrained policy with a given average budget to the optimal policy with no budget,  for the average budget in selection subsidization (with biased values) problem, for different capacities $\boldsymbol{k}$ and bias factors $\boldsymbol{\rho}$, a function of average budget $\boldsymbol{b}$ on excepted selections from group $\boldsymbol{\WomanSet}$.}
    \label{fig:v4_CR_fair_unfair}
\end{figure}

Given the above setup, we can also ask an intriguing question. How valuable is the average budget? In particular, starting from the initial total capacity $k$ and the average budget $b$ for hiring from the group $\WomanSet$, how do we compare the utility gain from employing one additional unit of capacity with the utility gain derived from the increase of the average budget by one? In \Cref{tab:my_label} we try to answer this question. Each entry corresponds to a pair $(k,b)$, where $k$ is the capacity and $b$ is the average budget under the current system. The number written in each entry is the gain from a unit increase in the budget minus the gain from a unit increase the capacity, where gain is defined in terms of the ratio of utilities of optimal budget-constrained to optimal without budget. Our results suggest that (i) for small values of the current budget, the value of one extra unit of the average budget is considerably more than one additional unit of capacity, e.g., see the column corresponding to $b=0$ or $b=1$; (ii) the extra gain for each unit of budget is decreasing in $b$. Combining these two observations, we conclude that a little bit of average budget can go a long way --- not only does it help with more representation from the underprivileged group $\WomanSet$, but also it  allows selections from (potentially top) members of group $\WomanSet$ and leading to increasing the overall efficiency.

\vspace{-2mm}

\begin{table}[htb]
    \centering
     \footnotesize
    \begin{tabular}{|c||*{5}{c|}}\hline
        \backslashbox{$k$}{$b$}
        &\makebox[3em]{0}&\makebox[3em]{1}&\makebox[3em]{2}
        &\makebox[3em]{3}&\makebox[3em]{4}\\\hline\hline
        1 & 0.291 & 0.040 & 0.003 & 0.003 & 0.003\\ \hline
        2 & 0.251 & 0.067 & 0.009 & 0.001 & 0.001\\ \hline
        3 & 0.221 & 0.079 & 0.021 & 0.004 & 0.001\\ \hline
        4 & 0.201 & 0.081 & 0.028 & 0.007 & 0.001\\ \hline
        5 & 0.185 & 0.082 & 0.034 & 0.010 & 0.002\\ \hline
        6 & 0.173 & 0.082 & 0.039 & 0.015 & 0.003\\ \hline
        7 & 0.163 & 0.082 & 0.043 & 0.019 & 0.005\\ \hline
        8 & 0.154 & 0.081 & 0.047 & 0.024 & 0.009\\ \hline
        9 & 0.146 & 0.079 & 0.049 & 0.028 & 0.013\\ \hline
        10 & 0.139 & 0.077 & 0.049 & 0.031 & 0.016\\ \hline
    \end{tabular}
    \smallskip
    \caption{Gain of extra budget - Gain of extra capacity in terms of price of fairness.}
    \label{tab:my_label}
\end{table}


%

\subsubsection{Additional Notes.}
\label{sec:numeric_pandora_budget_Additional Notes}

\revcolorm{Similar to \Cref{sec:numeric_pandora_parity_Additional Notes} we plot the normalized ex-post slack in \Cref{fig-apx:normal_v4_expost slack}, measured by the formula:
$$\frac{\sum_{i \in \WomanSet} \select_i^{\policy}-b}{\sum_{i \in \WomanSet} \select_i^{\policy}+b}.$$ 
As can be seen, for each budget $b$, the distribution has a low variance. Additionally, the reason the histogram of some of the larger budgets $b$ are skewed toward left is due to the fact that the budget constraint becomes less binding as we increase the budget $b$, and the optimal unconstrained policy itself may be a feasible policy.
\begin{figure}[htb]
    \centering
    \includegraphics[width=0.55\textwidth]{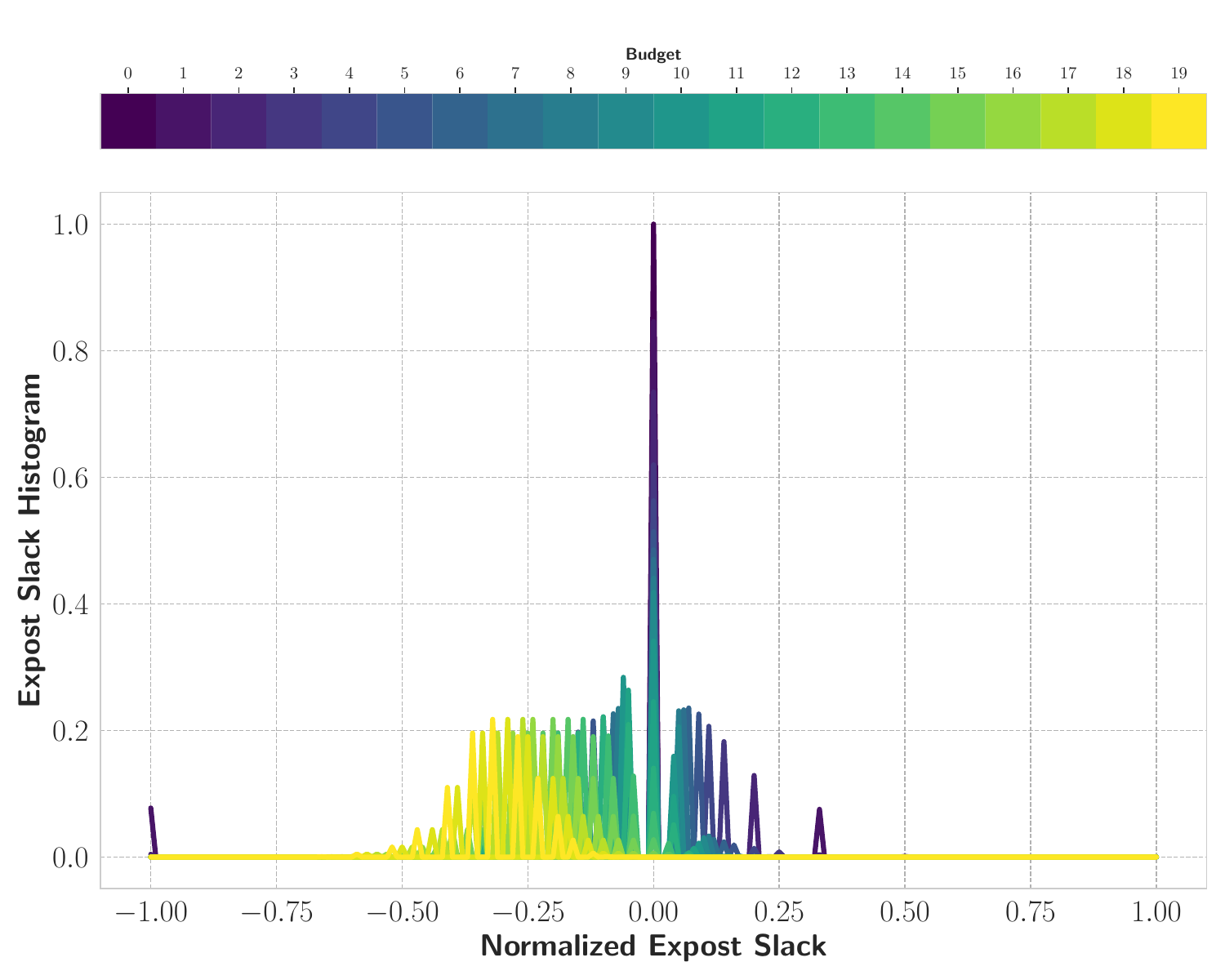}
    \caption{Histogram of normalized ex-post slack, Capacity = 20, Bias Factor = 1.0
    \label{fig-apx:normal_v4_expost slack}
    }
\end{figure}
}
\vspace{0mm}
\subsection{Missing Figures and Discussions of 
\texorpdfstring{\Cref{apx:numerical-main}}{}}
\label{apx:numerical}
We provide all the missing figures in our numerical simulations \revcolorm{for Pandora's box under normal values here}.
\begin{figure}[htb]
    \centering
    \begin{subfigure}[b]{0.32\textwidth}
        \centering
        \includegraphics[width=\textwidth]{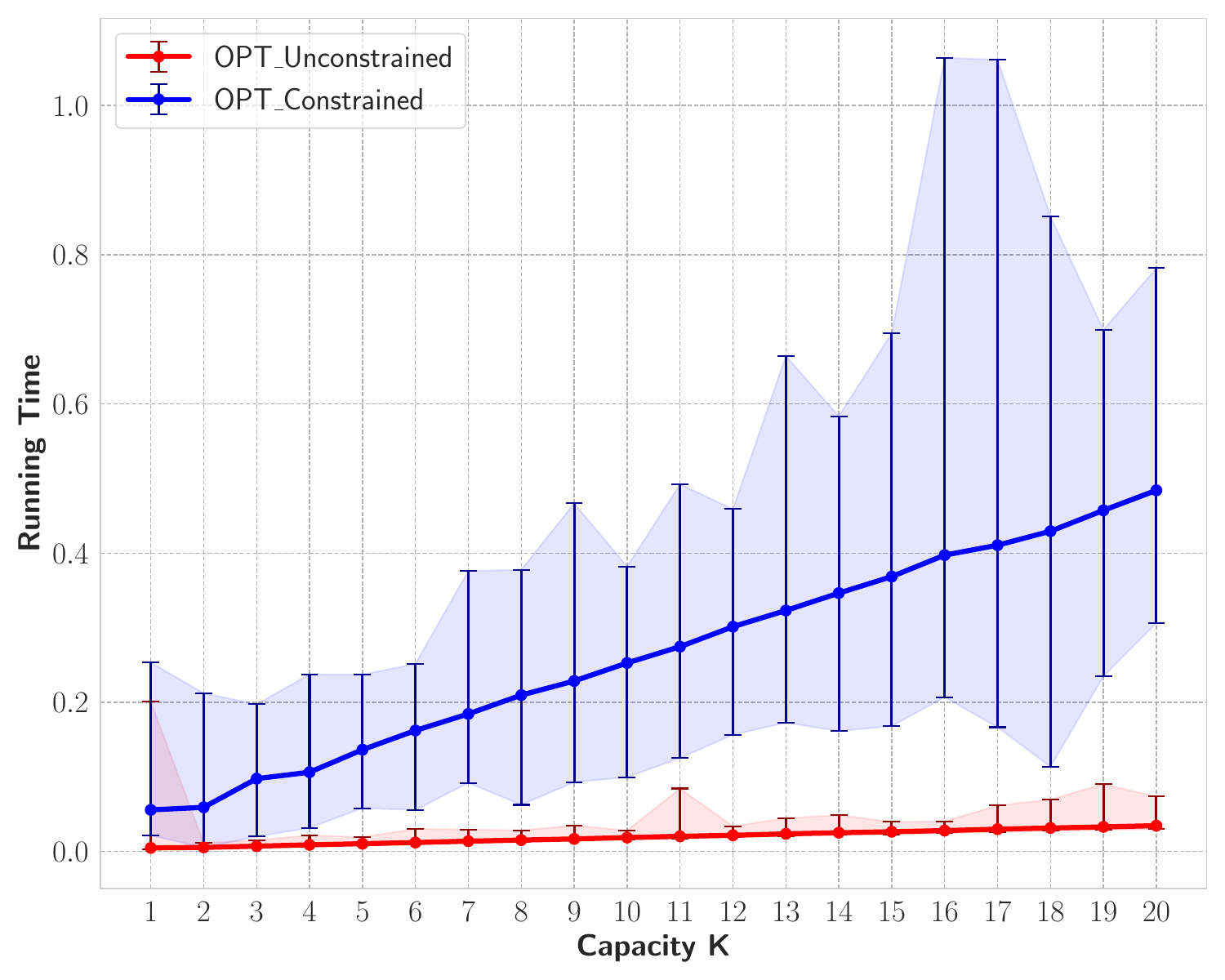}
        \caption{   Bias Factor = 0.3}
    \end{subfigure}
    \hfill
        \begin{subfigure}[b]{0.32\textwidth}
        \centering
        \includegraphics[width=\textwidth]{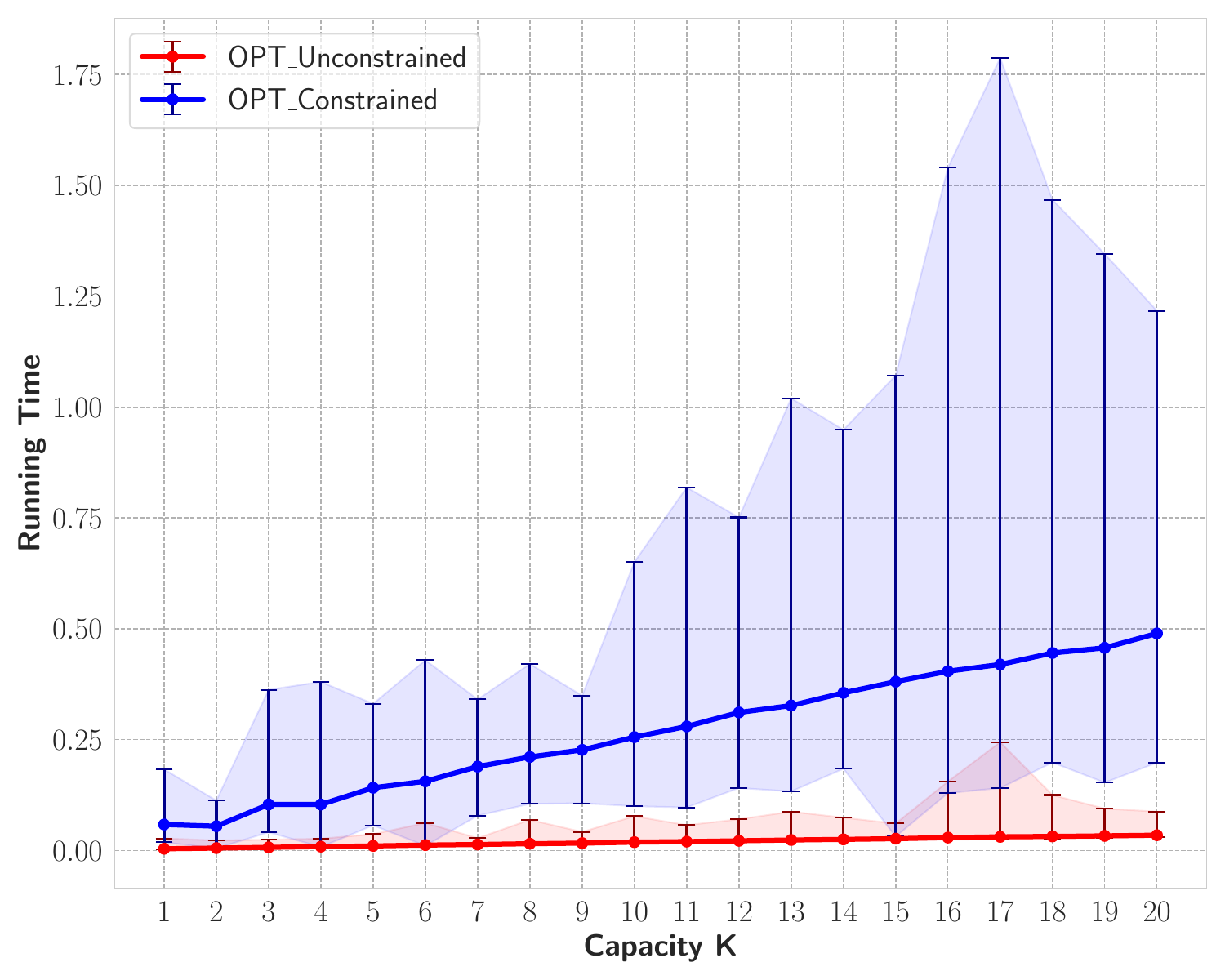}
        \caption{   Bias Factor = 0.7}
    \end{subfigure}
    \hfill
        \begin{subfigure}[b]{0.32\textwidth}
        \centering
        \includegraphics[width=\textwidth]{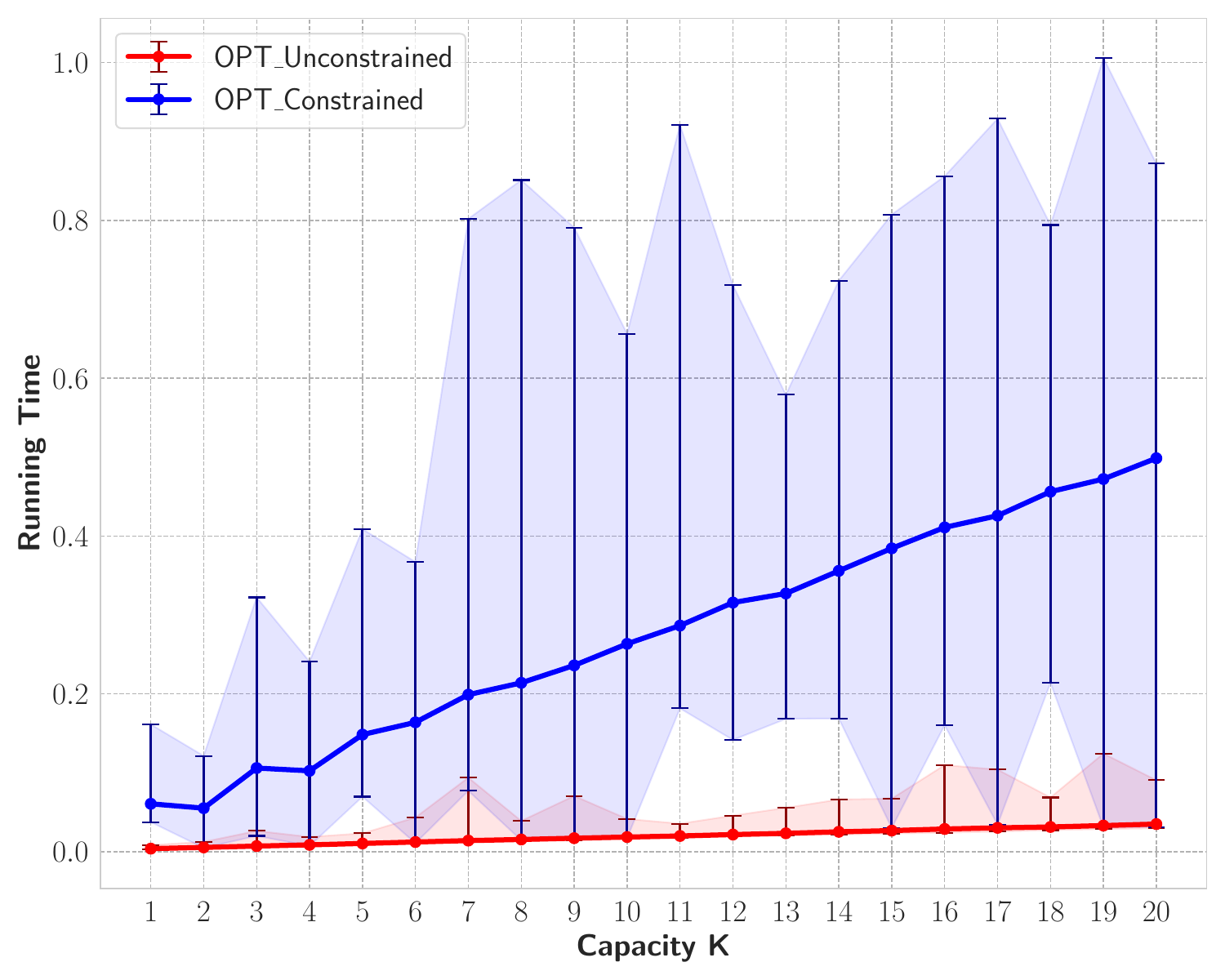}
        \caption{   Bias Factor = 1.0}
    \end{subfigure}
    \caption{Comparison of running times of optimal constrained and unconstrained policies (in seconds).}
    \label{fig-apx:v1_Running_Time}
\end{figure}

\vspace{-4mm}
\begin{figure}[htb]
    \centering
    \begin{subfigure}[b]{0.45\textwidth}
        \centering
        \includegraphics[width=\textwidth]{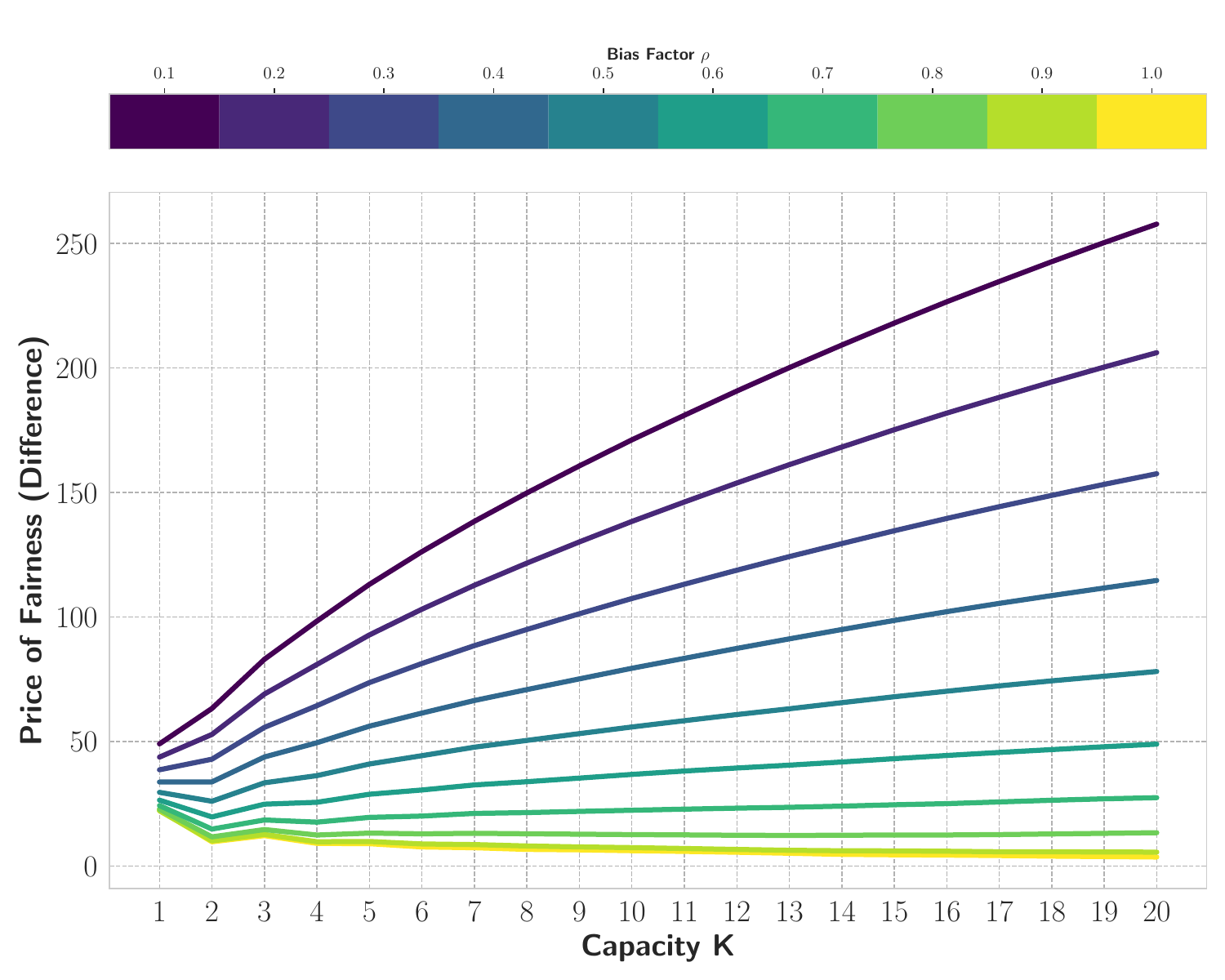}
         \caption{}
    \end{subfigure}
    \quad\quad\quad
    \begin{subfigure}[b]{0.45\textwidth}
        \centering
        \includegraphics[width=\textwidth]{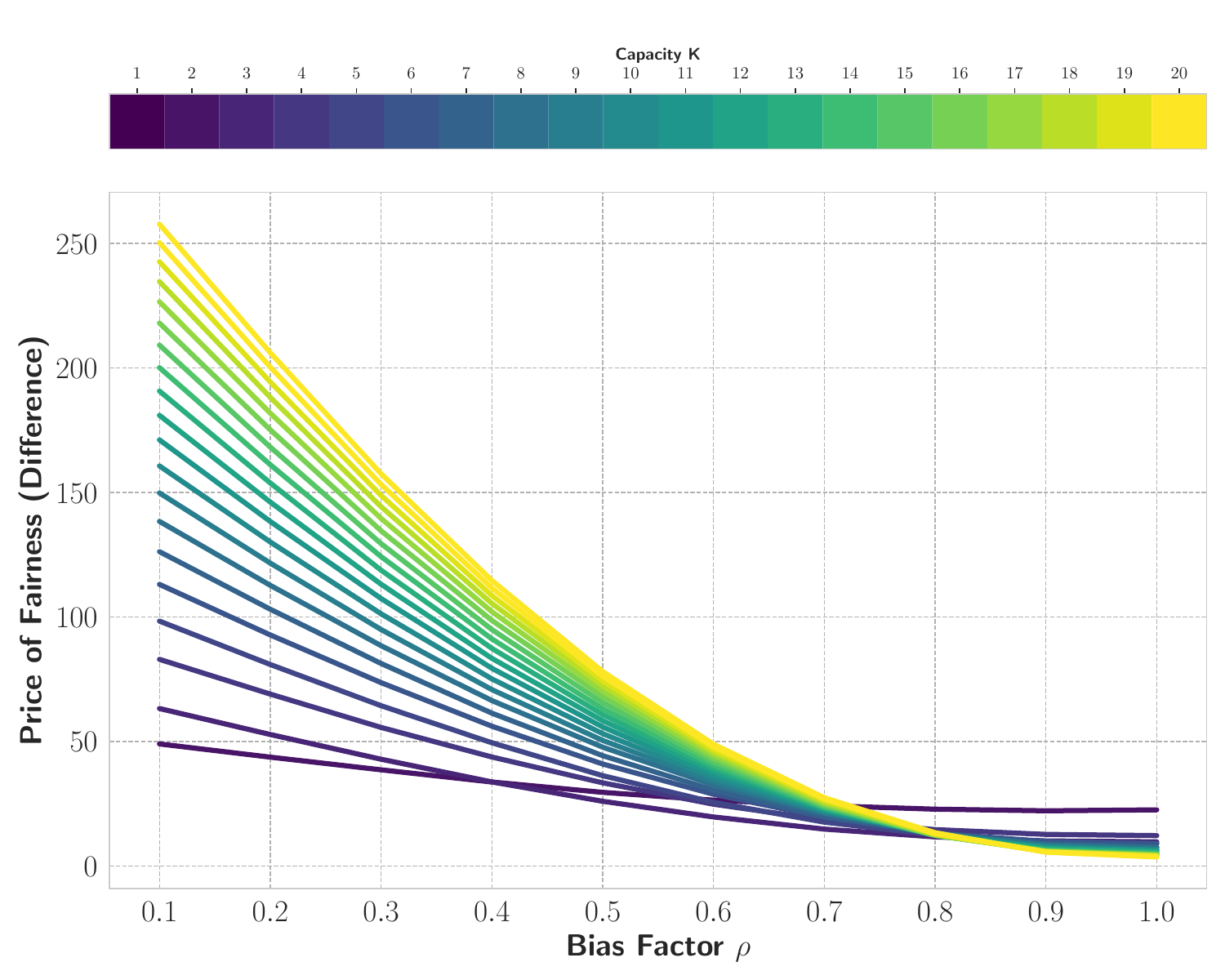}
         \caption{}
    \end{subfigure}
    \caption{Short-term price of fairness of demographic parity in selection in terms of utility differences; (a) as a function of capacity $\boldsymbol{k}$, and (b) as a function of bias factor $\boldsymbol{\rho}$.}
    \label{fig-apx:v1_diff_and_CR_fair_unfair}
\end{figure}
\begin{figure}[htb]
    \centering
    \begin{subfigure}{0.45\textwidth}
        \centering
        \includegraphics[width=\textwidth]{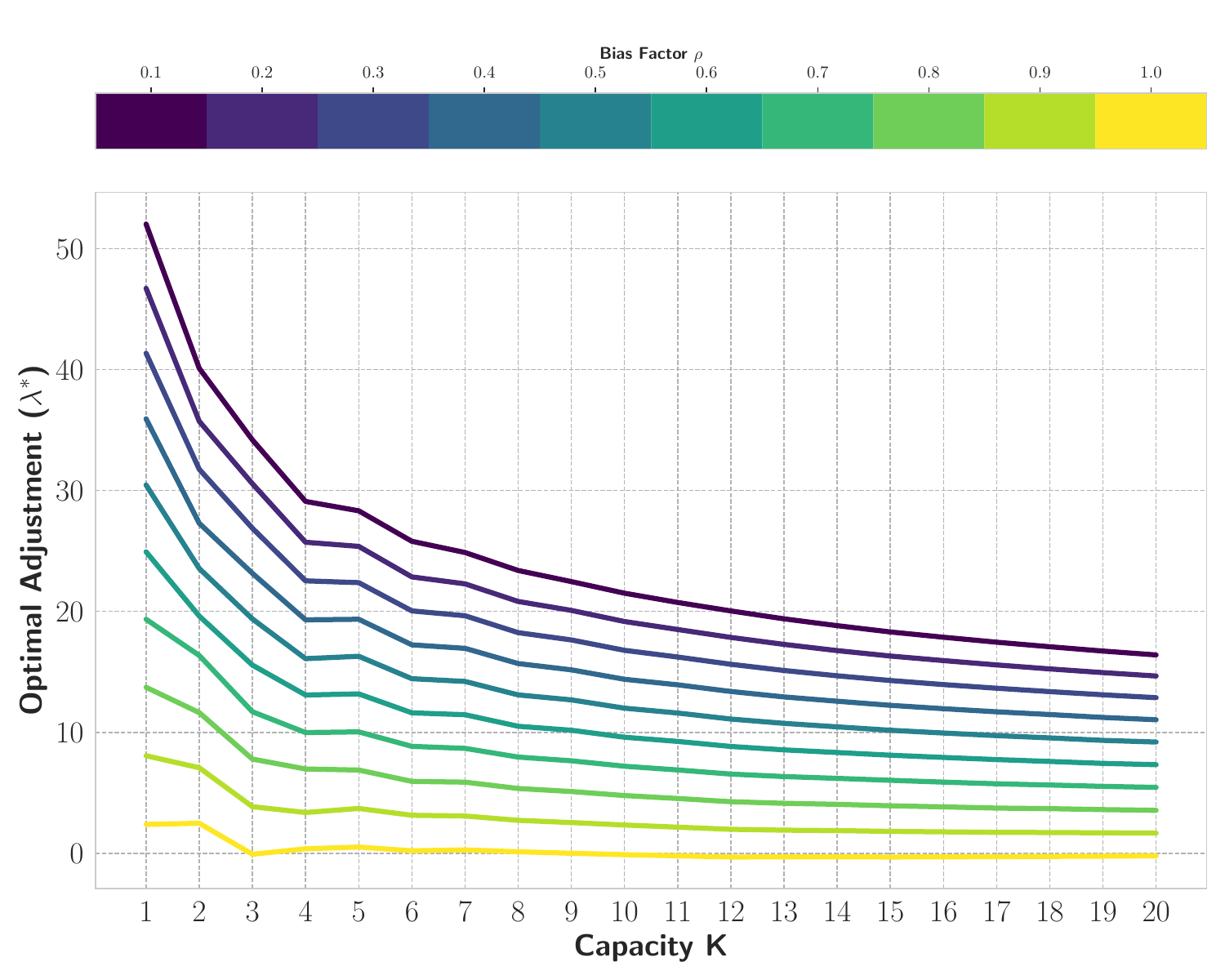}
        \caption{}
    \end{subfigure}
       \quad\quad\quad
    \begin{subfigure}{0.45\textwidth}
        \centering
        \includegraphics[width=\textwidth]{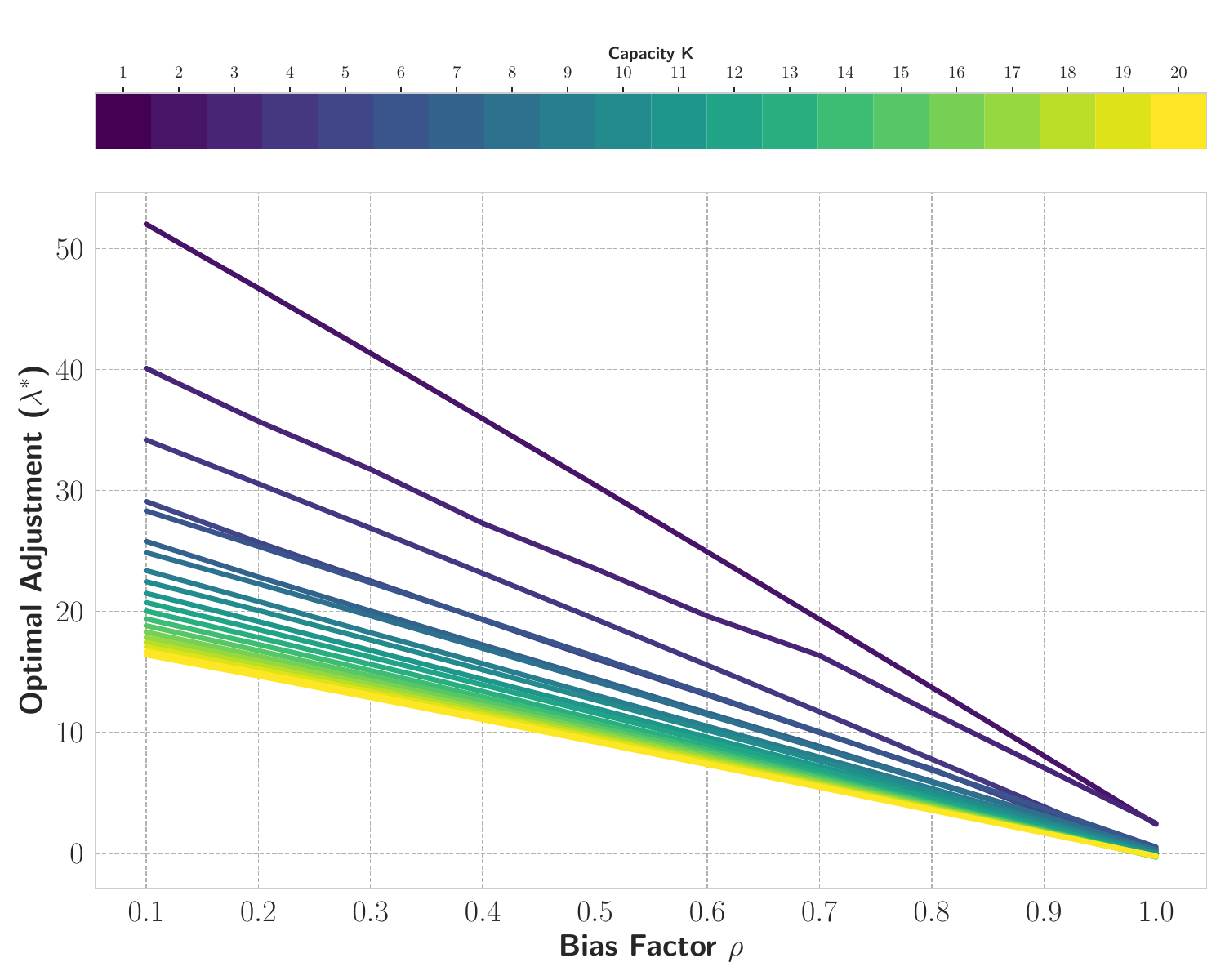}
        \caption{}
    \end{subfigure}
    \caption{The optimal dual adjustment $\boldsymbol{\lambda^*}$ in demographic parity in selection; (a) as a function of capacity $\boldsymbol{k}$, and (b) as a function of bias factors $\boldsymbol{\rho}$.}
    \label{fig-apx:v1_slack_lambda_fair_unfair}
\end{figure}


\begin{figure}[htb]
    \centering
    \begin{subfigure}[b]{0.32\textwidth}
        \centering
        \includegraphics[width=\textwidth]{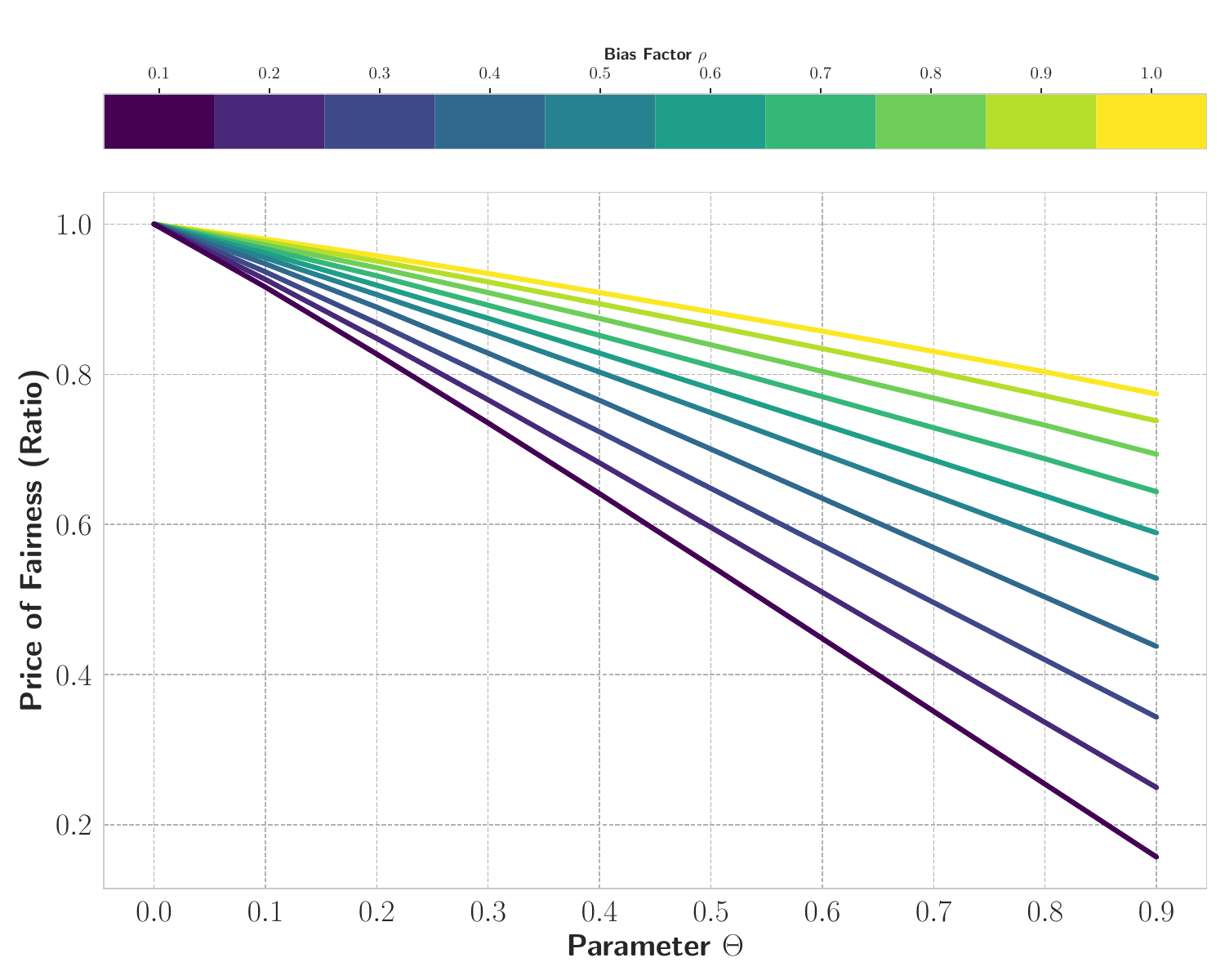}
        \caption{   Capacity = 1}
    \end{subfigure}
    \hfill
    \begin{subfigure}[b]{0.32\textwidth}
        \centering
        \includegraphics[width=\textwidth]{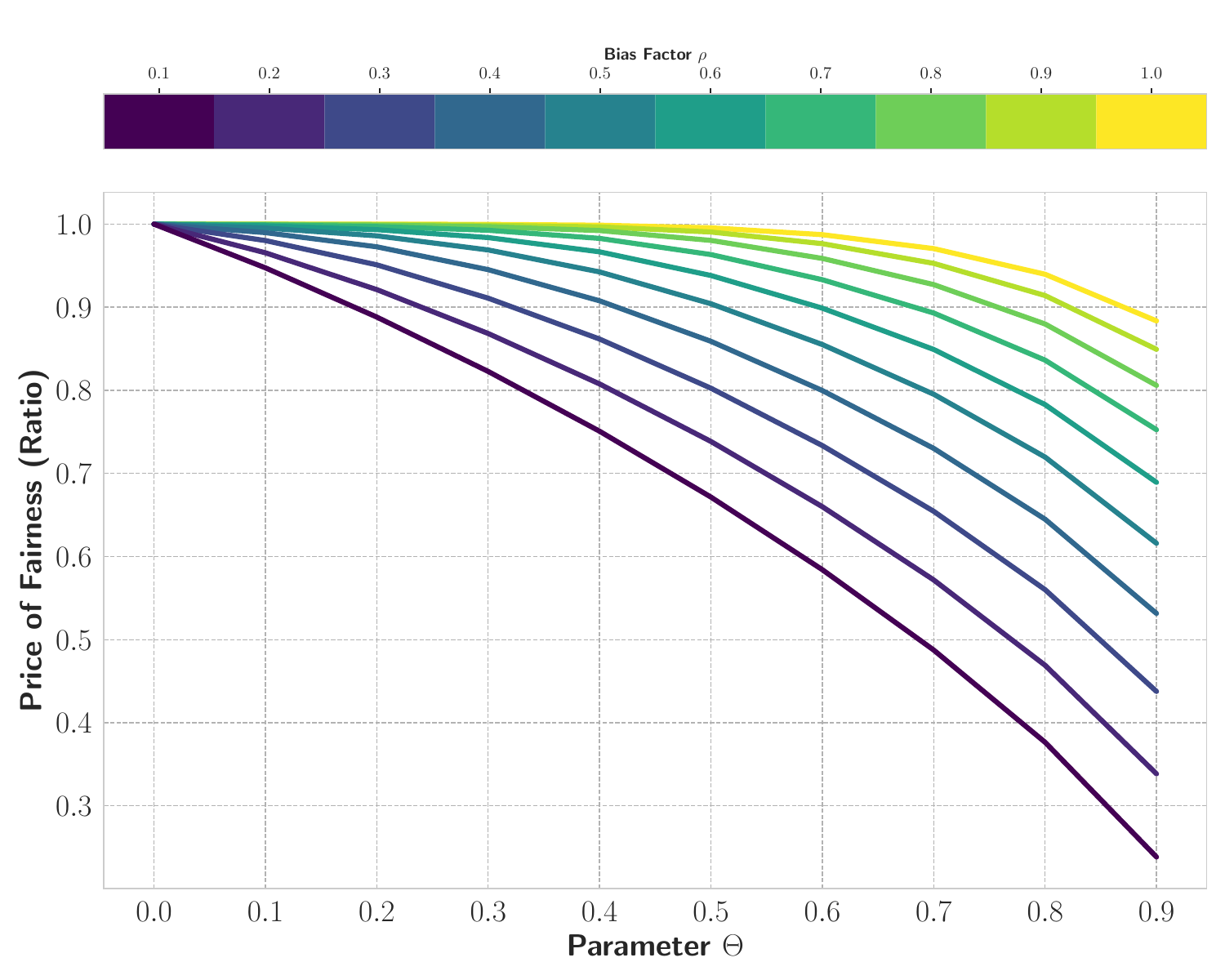}
        \caption{   Capacity = 10}
    \end{subfigure}
    \hfill
    \begin{subfigure}[b]{0.32\textwidth}
        \centering
        \includegraphics[width=\textwidth]{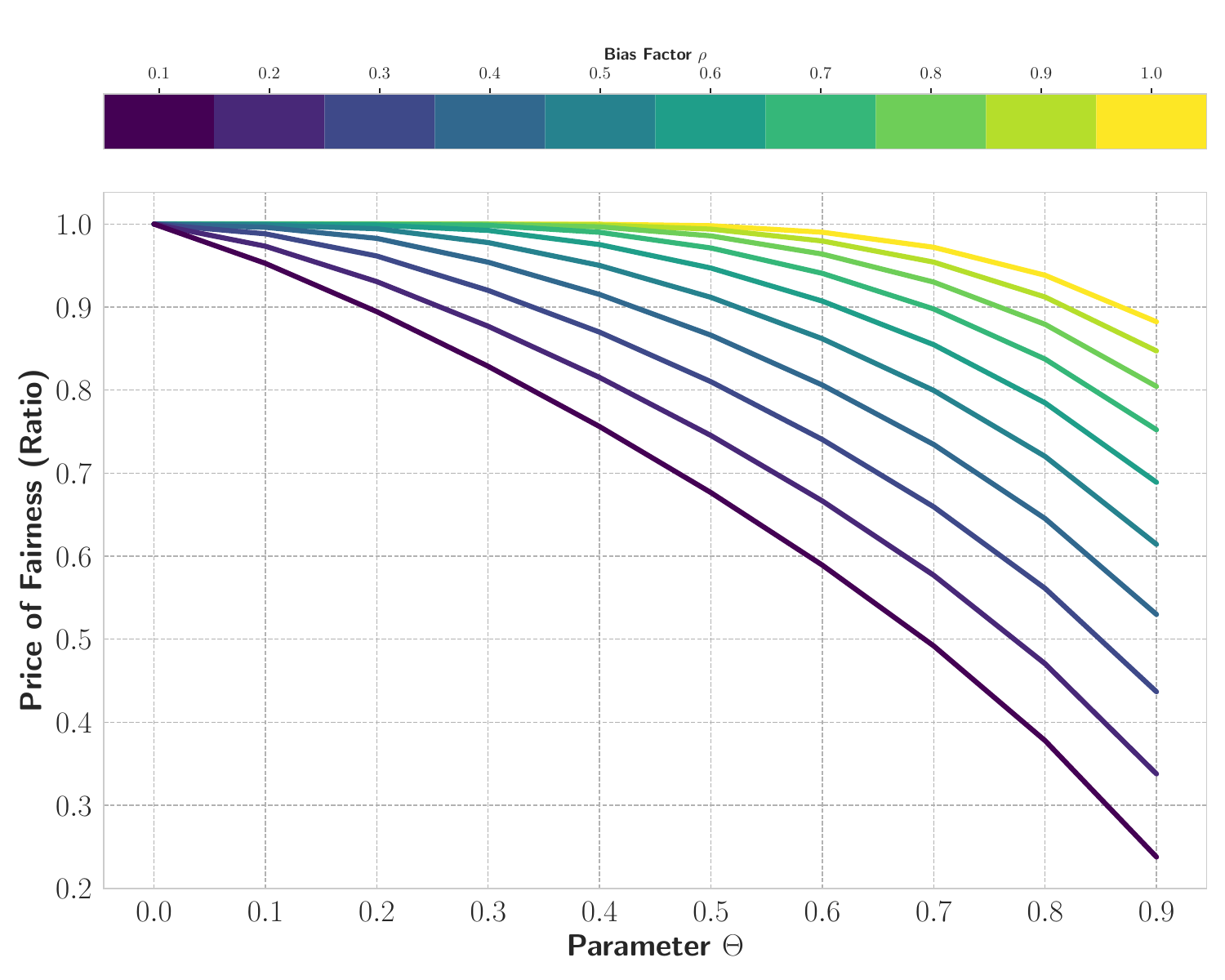}
        \caption{   Capacity = 20}
    \end{subfigure}
    \caption{Short-term price of fairness of average quota in selection  in terms of utility ratio, as a function of quota parameter $\boldsymbol{\theta\in[0,1]}$ in \eqref{eq:quota} ($\boldsymbol{\theta=0.5}$ corresponds to demographic parity).}
    \label{fig-apx:v5_CR_fair_unfair}
\end{figure}

\begin{figure}[htb]
    \centering
    \begin{subfigure}[b]{0.32\textwidth}
        \centering
        \includegraphics[width=\textwidth]{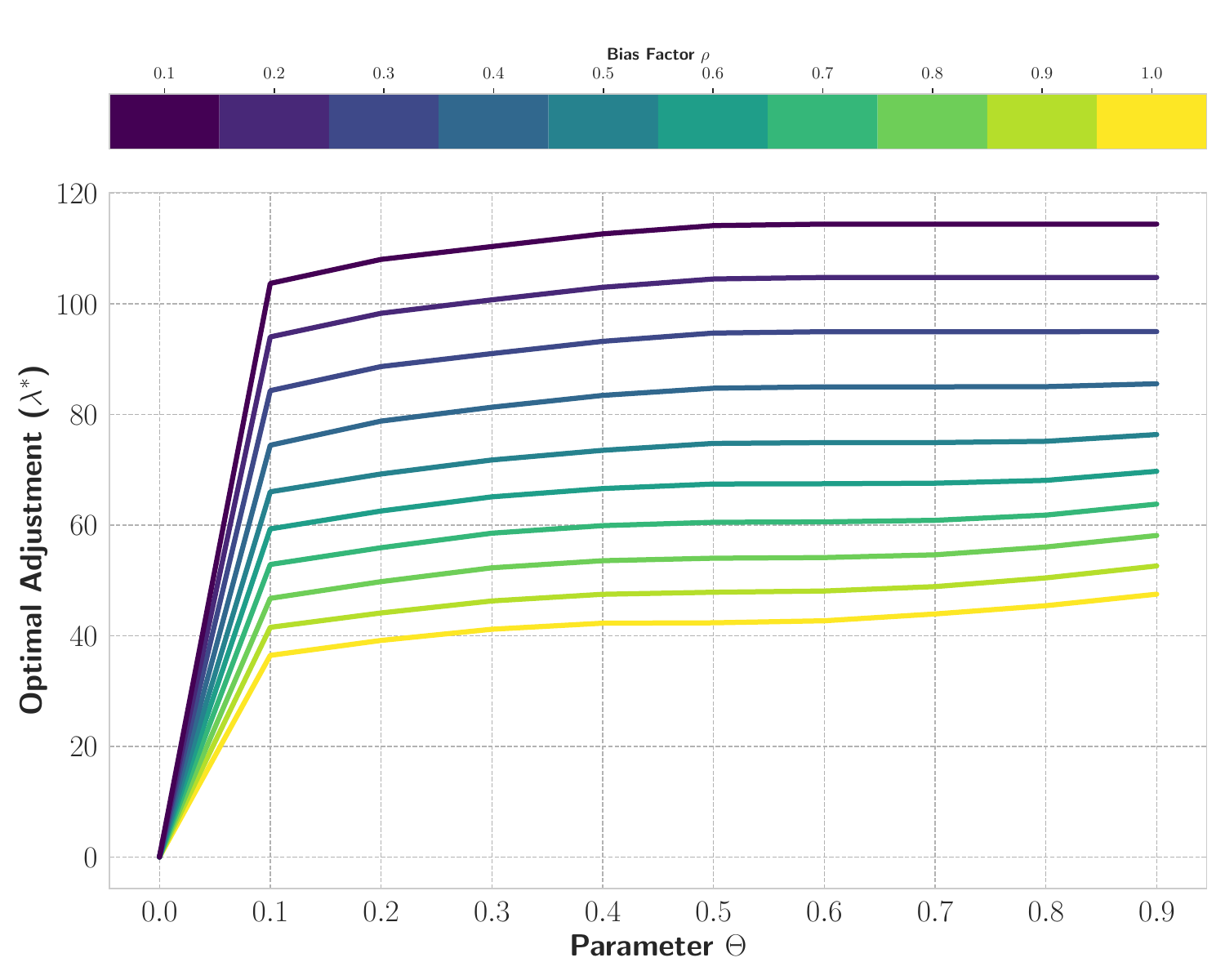}
        \caption{   Capacity = 1}
    \end{subfigure}
    \hfill
    \begin{subfigure}[b]{0.32\textwidth}
        \centering
        \includegraphics[width=\textwidth]{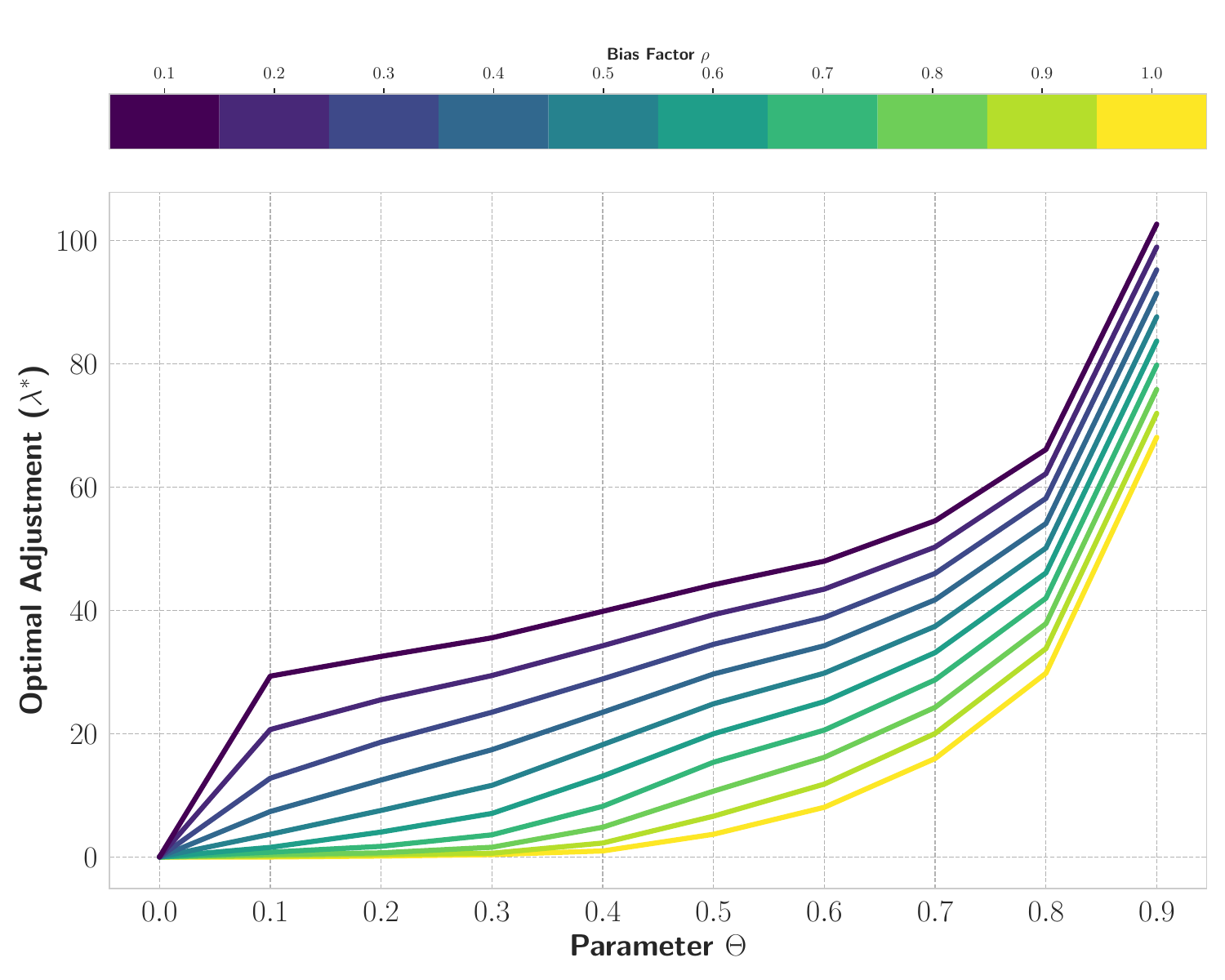}
        \caption{   Capacity = 10}
    \end{subfigure}
    \hfill
    \begin{subfigure}[b]{0.32\textwidth}
        \centering
        \includegraphics[width=\textwidth]{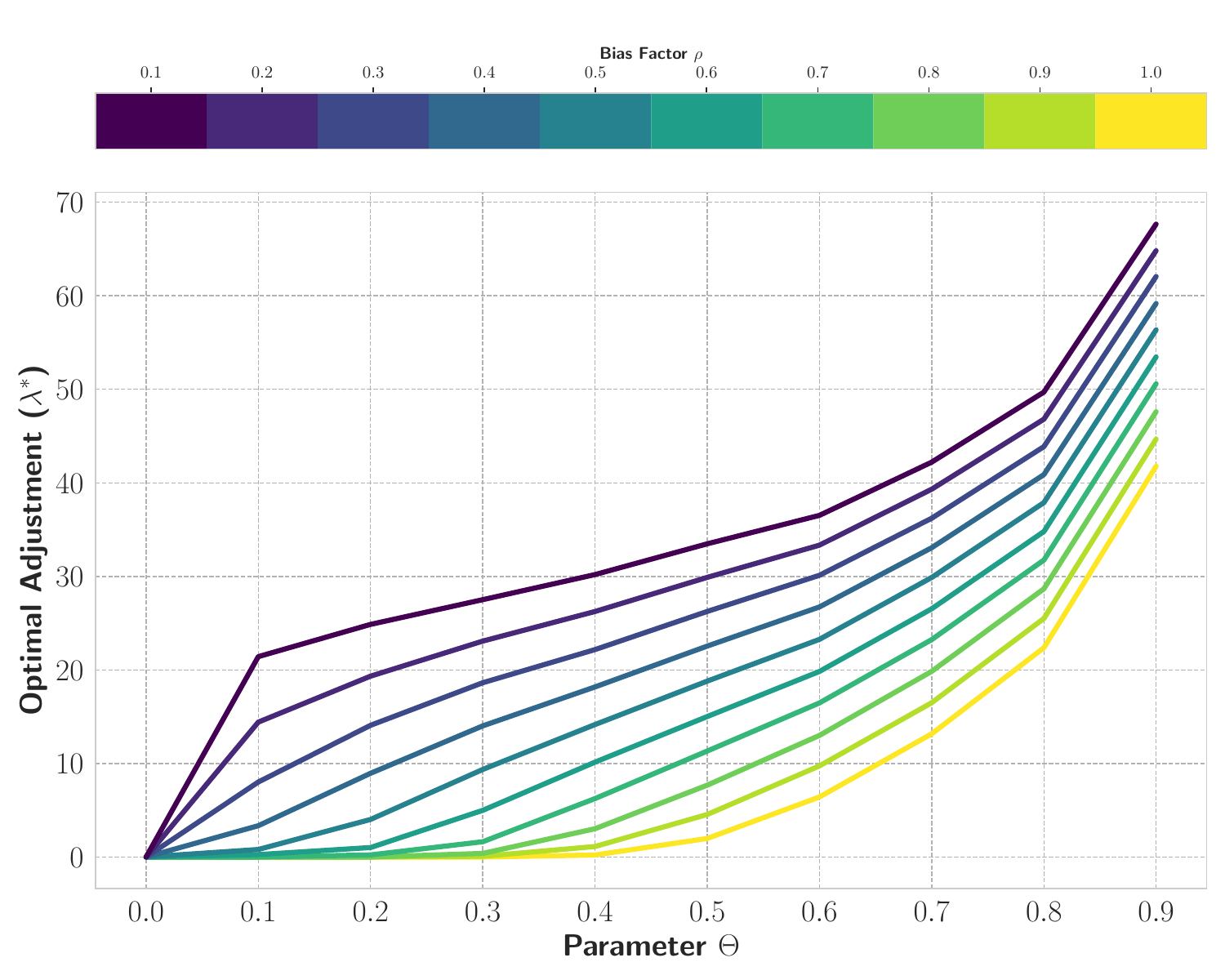}
        \caption{   Capacity = 20}
    \end{subfigure}
    \caption{The optimal dual adjustment $\boldsymbol{\lambda^*}$ for quota in selection constraint, as a function of quota parameter $\boldsymbol{\theta\in[0,1]}$ in \eqref{eq:quota} ($\boldsymbol{\theta=0.5}$ corresponds to demographic parity).}
    \label{fig-apx:v5_lambda_fair_unfair}
\end{figure}

\begin{figure}[htb]
    \centering
    \begin{subfigure}[b]{0.32\textwidth}
        \centering
        \includegraphics[width=\textwidth]{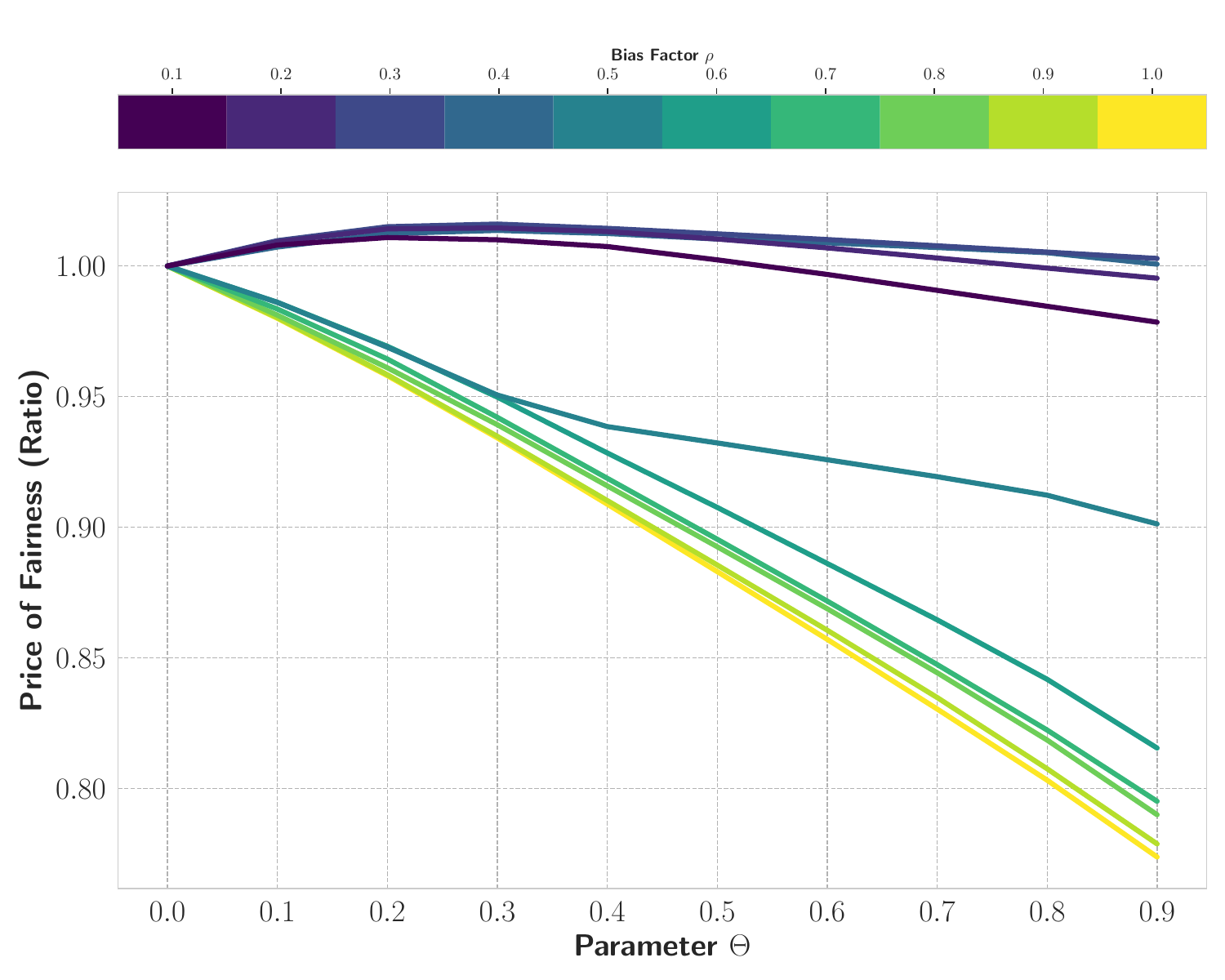}
        \caption{Capacity = 1}
    \end{subfigure}
    \begin{subfigure}[b]{0.32\textwidth}
        \centering
        \includegraphics[width=\textwidth]{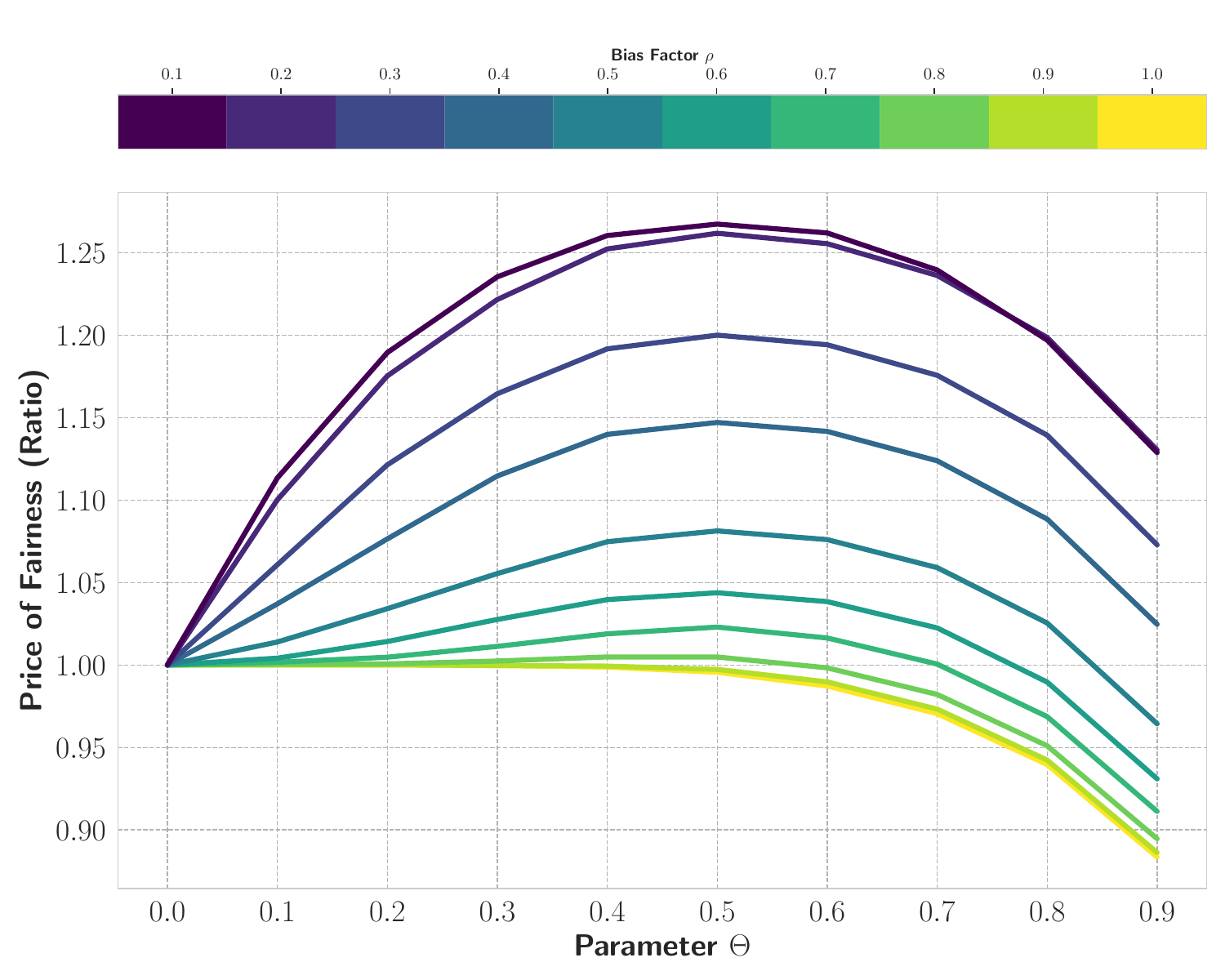}
        \caption{Capacity = 10}
    \end{subfigure}
    \begin{subfigure}[b]{0.32\textwidth}
        \centering
        \includegraphics[width=\textwidth]{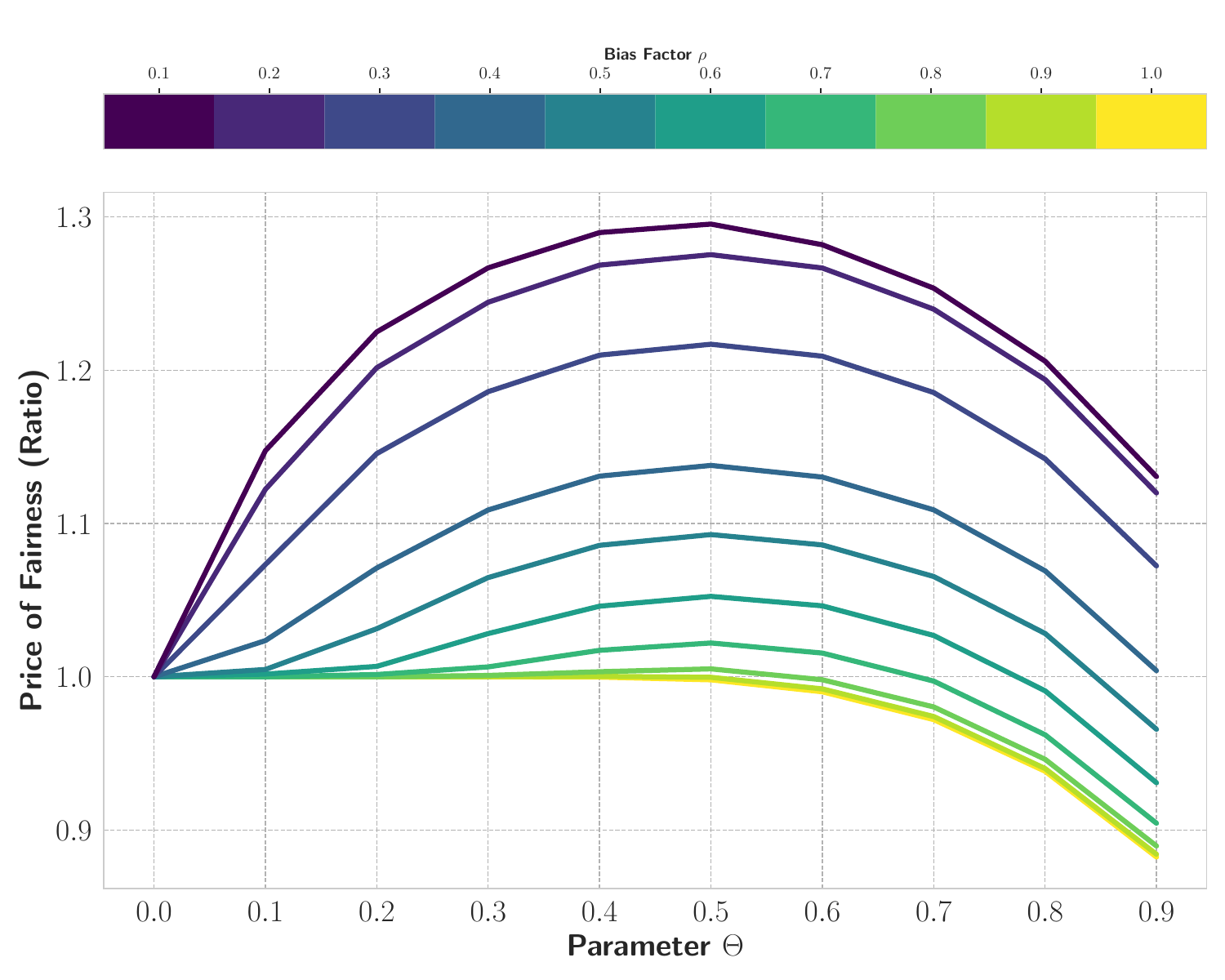}
        \caption{Capacity = 20}
    \end{subfigure}
    \caption{Long-term price of fairness of average quota in selection  in terms of utility ratio, as a function of quota parameter $\boldsymbol{\theta\in[0,1]}$ in \eqref{eq:quota} ($\boldsymbol{\theta=0.5}$ corresponds to demographic parity).}
    \label{fig-apx:v6_diff_CR_fair_unfair}
\end{figure}


\begin{figure}[htb]
    \centering
    \begin{subfigure}[b]{0.32\textwidth}
        \centering
        \includegraphics[width=\textwidth]{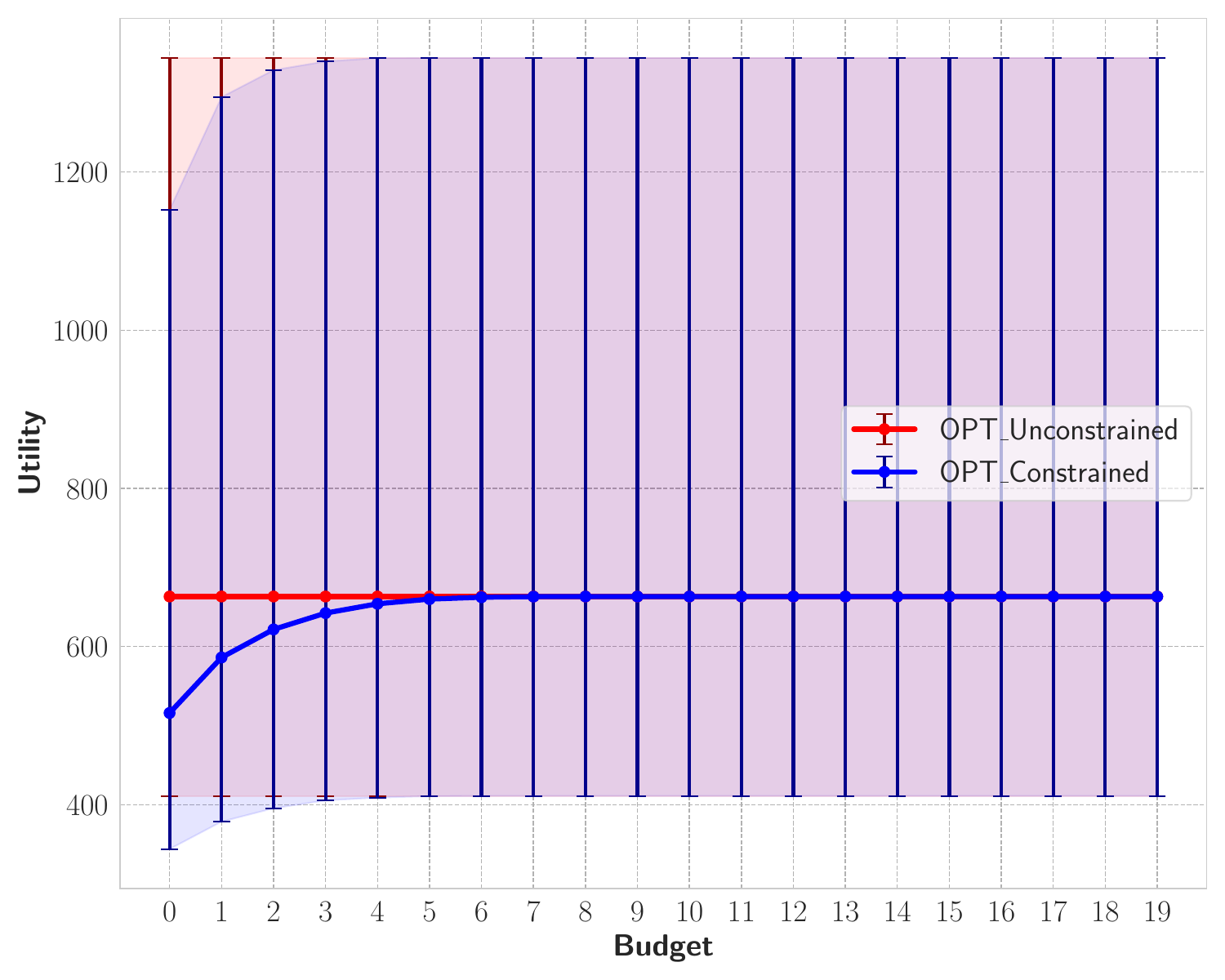}
        \caption{   Capacity = 10, bias factor = 1.0}
    \end{subfigure}
    \hfill
    \begin{subfigure}[b]{0.32\textwidth}
        \centering
        \includegraphics[width=\textwidth]{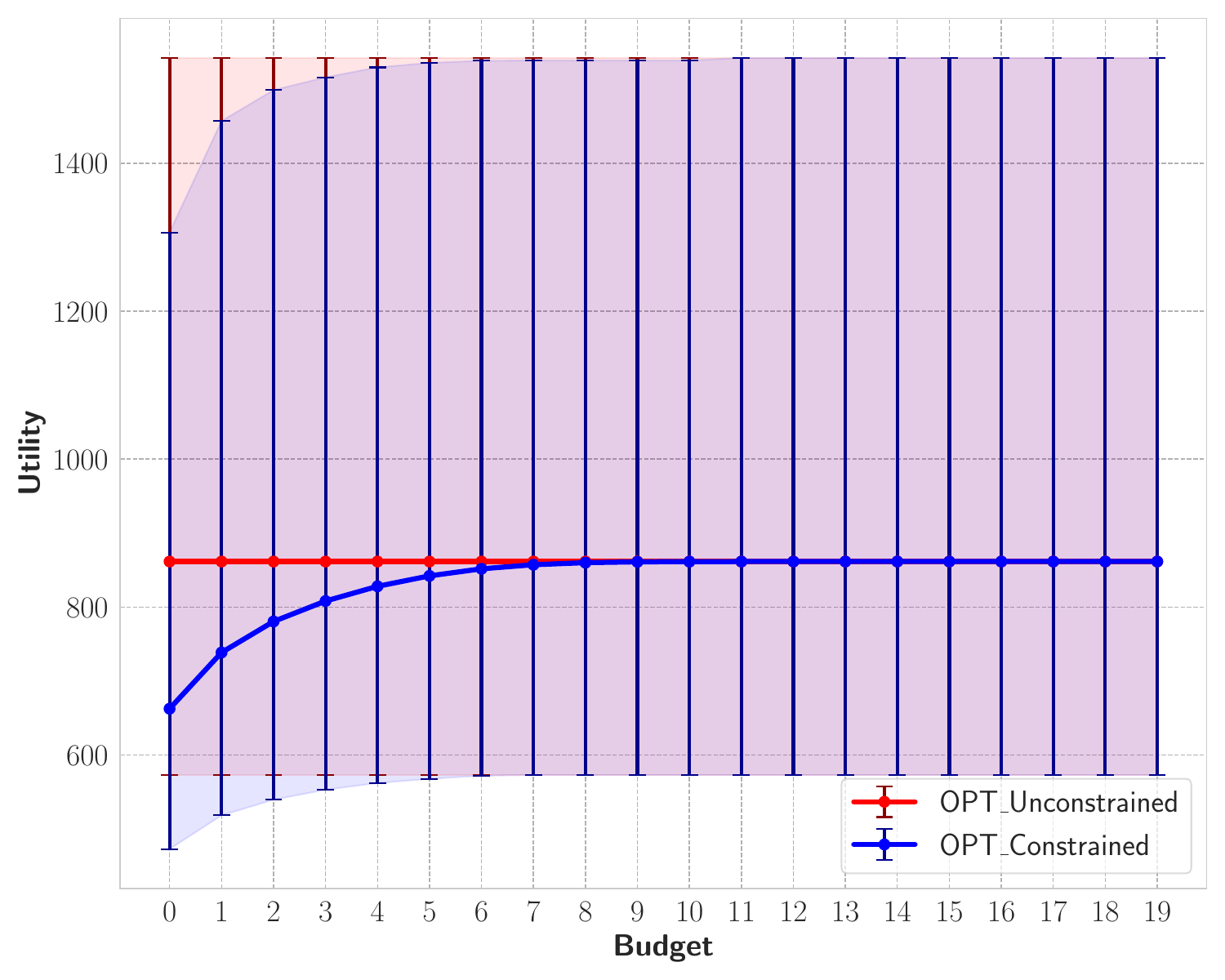}
        \caption{   Capacity = 15, bias factor = 1.0}
    \end{subfigure}
    \hfill
    \begin{subfigure}[b]{0.32\textwidth}
        \centering
        \includegraphics[width=\textwidth]{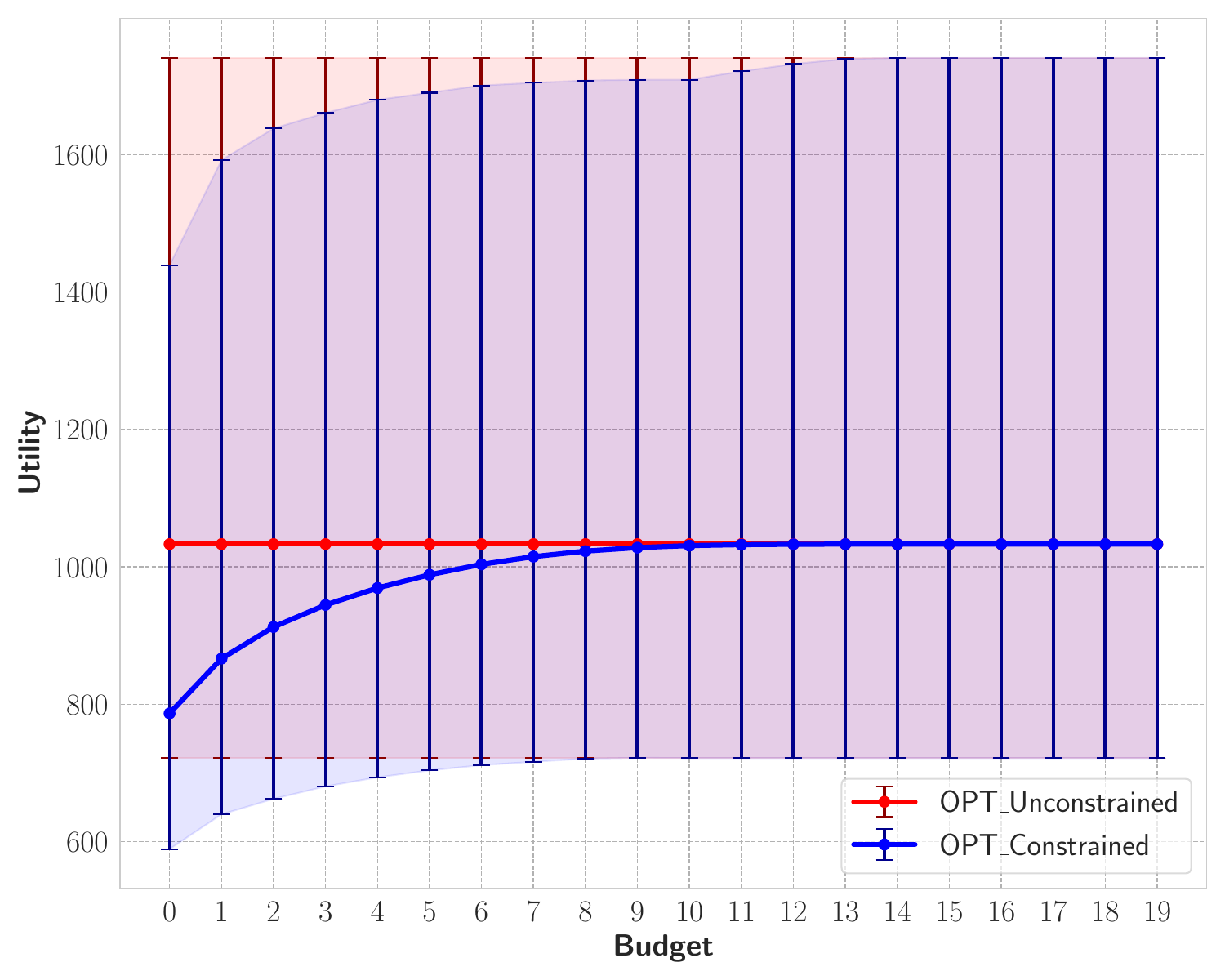}
        \caption{   Capacity = 20, bias factor = 1.0}
    \end{subfigure}
    \hfill
    \caption{Comparing the utilities of optimal unconstrained (red) and  optimal constrained (red) policies for the average budget in selection subsidization, for different capacities $\boldsymbol{k}$, as a function of average budget $\boldsymbol{b}$ on excepted selections from group $\boldsymbol{\WomanSet}$ (red curve corresponds to unlimited budget).}
    \label{fig-apx:v4_both_fair_unlimited}
\end{figure}

\begin{figure}[htb]
    \centering
    \begin{subfigure}[b]{0.32\textwidth}
        \centering
        \includegraphics[width=\textwidth]{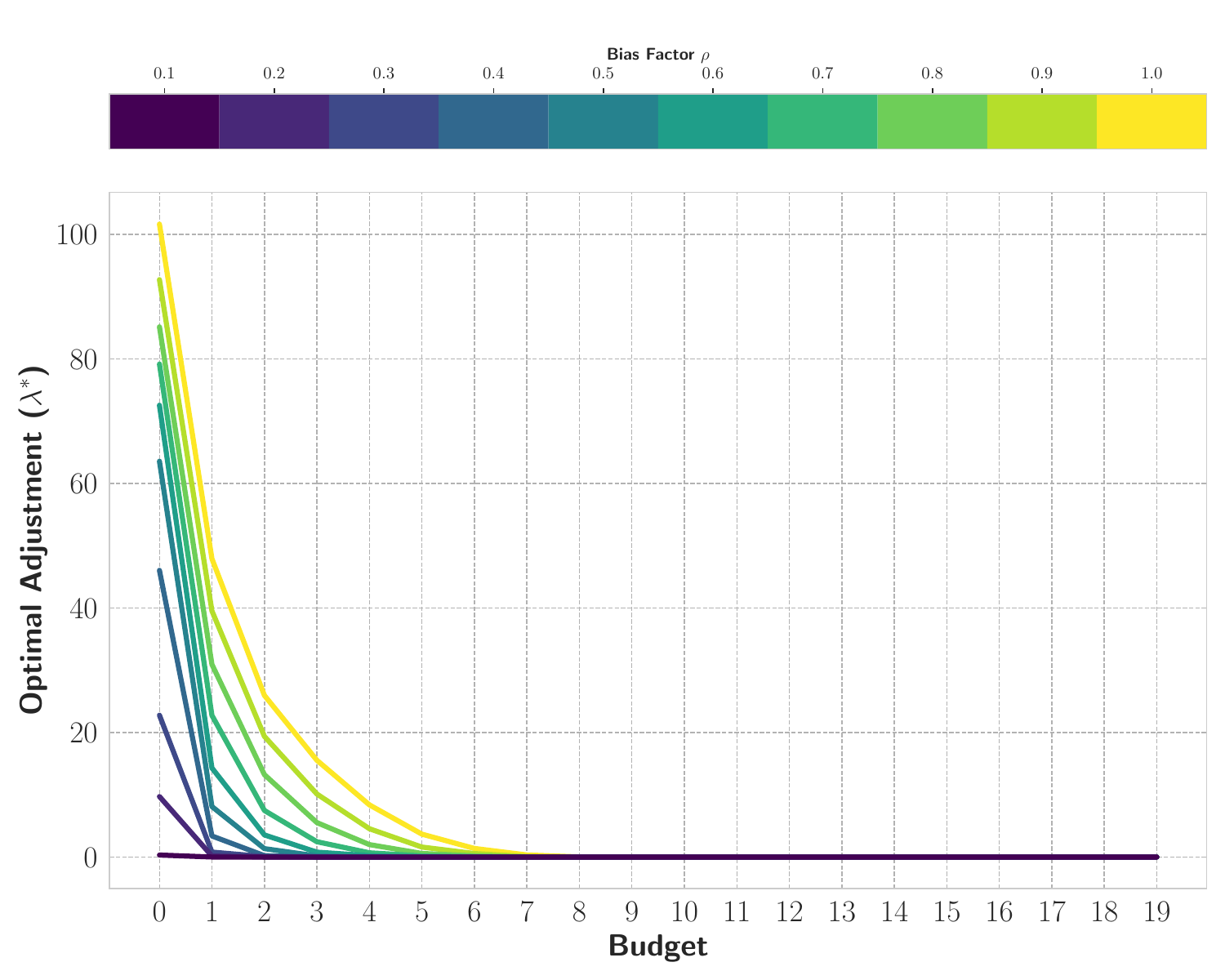}
        \caption{   Capacity = 10}
    \end{subfigure}
    \hfill
        \begin{subfigure}[b]{0.32\textwidth}
        \centering
        \includegraphics[width=\textwidth]{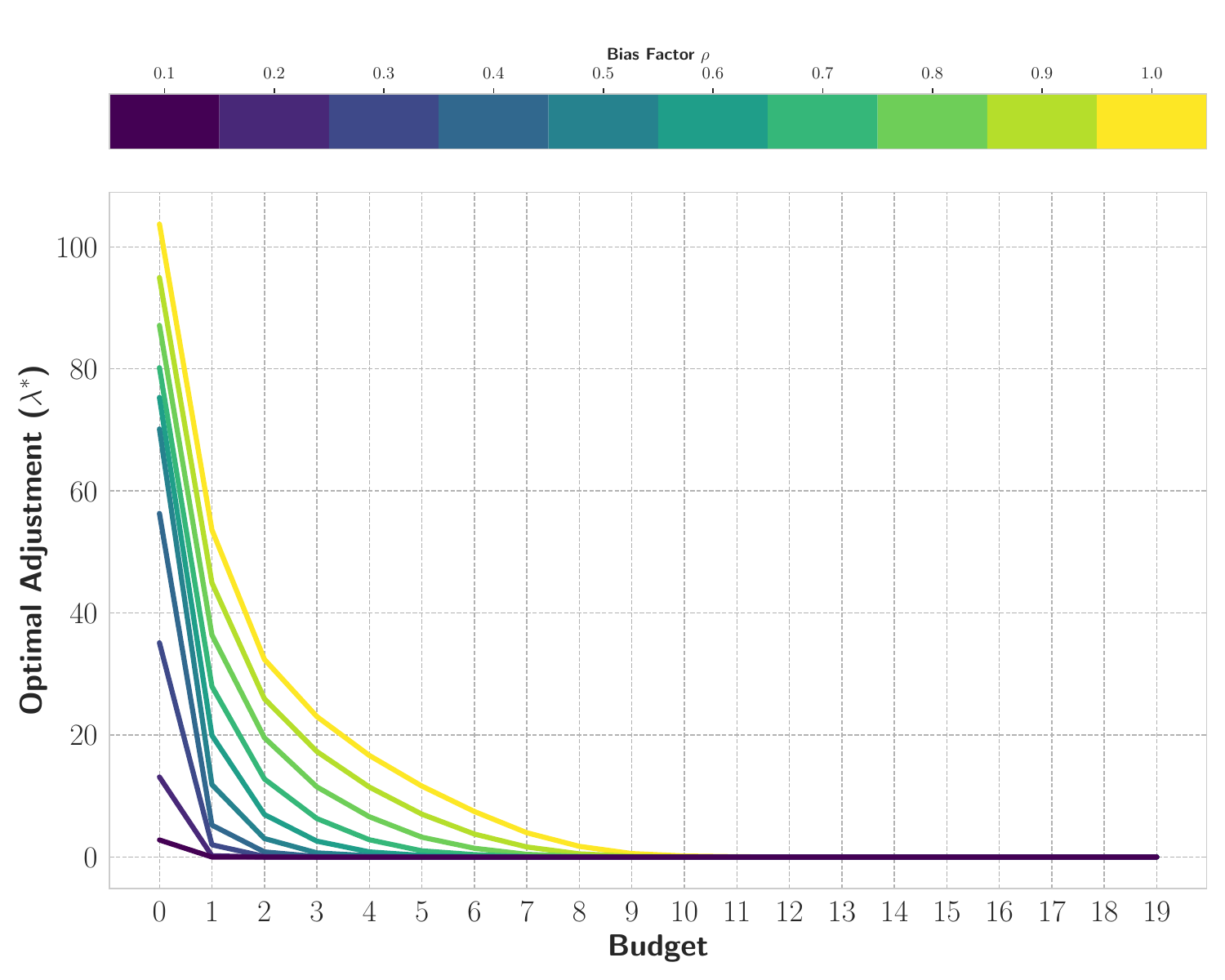}
        \caption{   Capacity = 15}
    \end{subfigure}
    \hfill
        \begin{subfigure}[b]{0.32\textwidth}
        \centering
        \includegraphics[width=\textwidth]{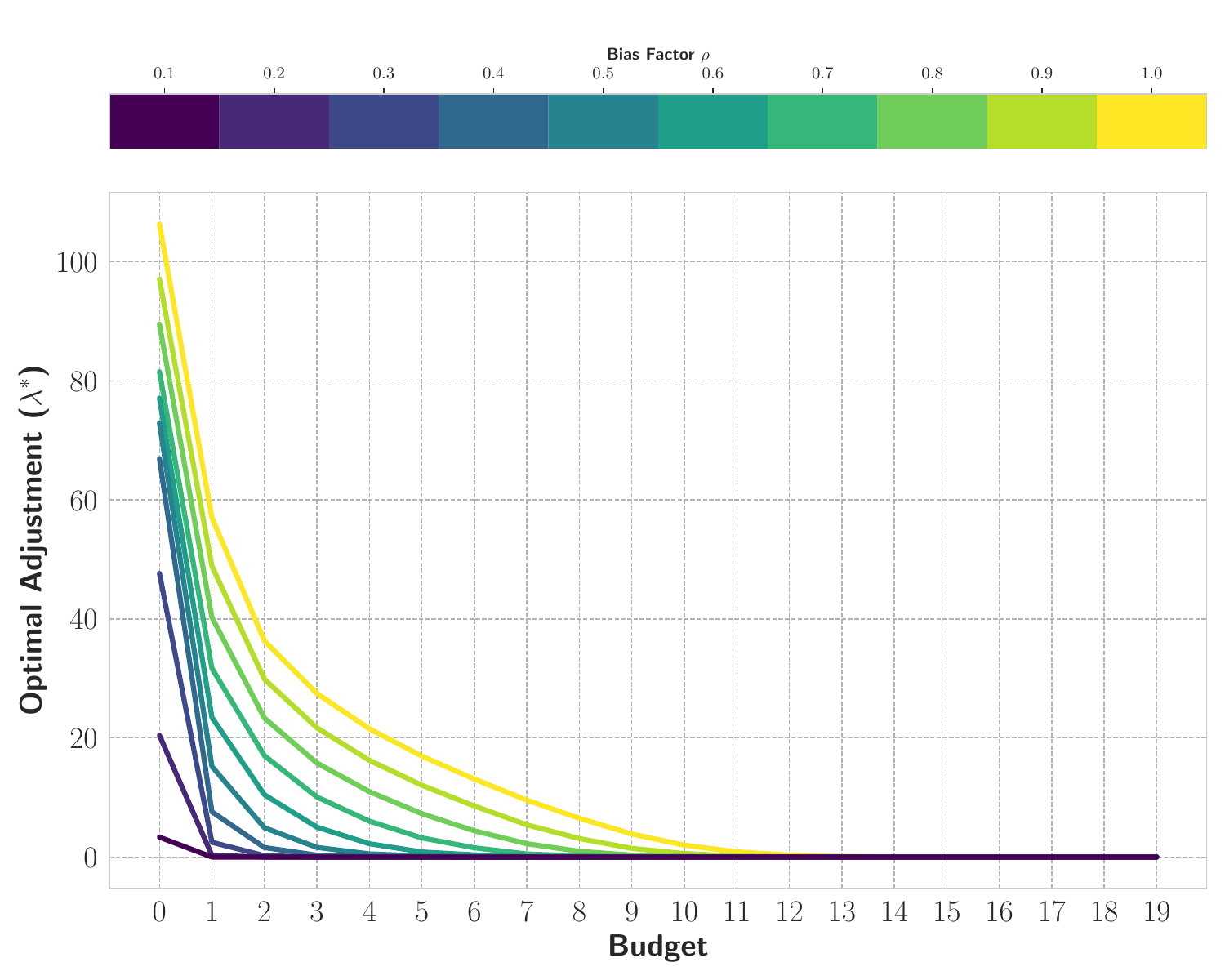}
        \caption{   Capacity = 20}
    \end{subfigure}
    \caption{The optimal dual adjustment $\boldsymbol{\lambda^*}$ for the average budget in selection subsidization, for different capacities $\boldsymbol{k}$ , as a function of average budget $\boldsymbol{b}$ on excepted number of selections from group $\boldsymbol{\WomanSet}$.}
    \label{fig-apx:v4_lambda_fair_unfair}
\end{figure}





~
\newpage
~
\newpage

\subsection{Numerical Simulations - Uniform Values}
\label{sec:numerical_uniform}

In this section we will demonstrate the robustness of the insights derived from our numerical simulations before, by conducting the same set of experiments but now under a different set of instances. As can be seen from the following figures, all of our previous managerial insights continue to hold, with some small changes to the exact numbers.

\smallskip
\noindent\textbf{Basic simulation setup:}  
The setup and structure of the problem is mainly as same as that of \Cref{sec:numerical}, except for the value distributions for the alternatives. In particular, this time we generate the mean values of the alternatives independently from a $\mathrm{Uniform}[20,60]$ distribution (as opposed to the $\mathrm{LogNormal}$ distribution before), and then add an independent $\mathrm{Uniform}[-20,20]$ noise on top of its mean (as opposed to the Gaussian noise before) to construct the value distribution for the corresponding alternatives. However, the rest of the setup, including the cost parameters, remain the same as before.

In \Cref{fig-apx:v1_Histogram_uniform}, you can observe the histogram of the generated values across all individuals, under three different bias levels. 
\vspace{-2mm}
\begin{figure}[htb]
    \centering
    \begin{subfigure}[b]{0.32\textwidth}
        \centering
        \includegraphics[width=\textwidth]{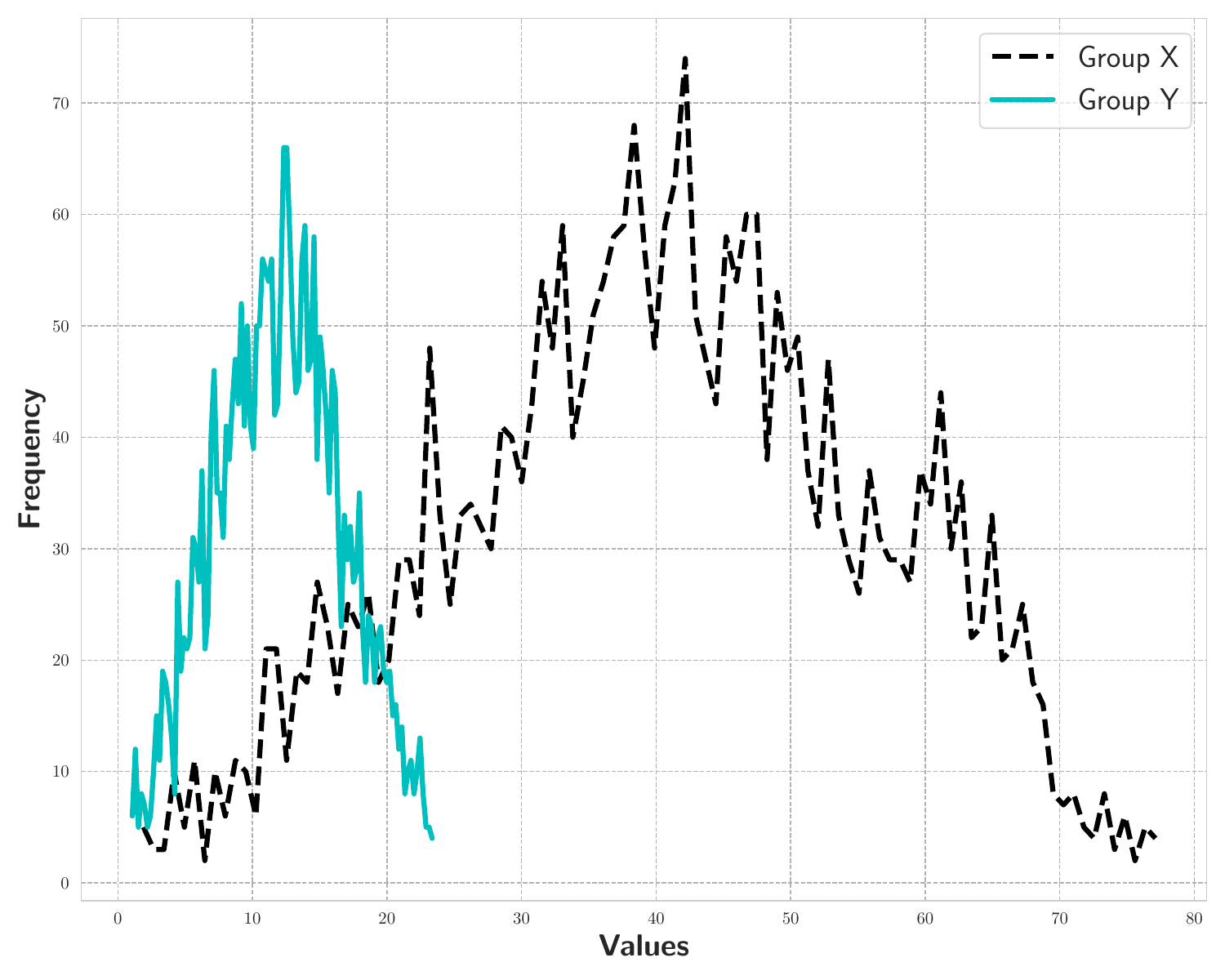}
        \caption{Bias factor = 0.3}
    \end{subfigure}
    \hfill
    \begin{subfigure}[b]{0.32\textwidth}
        \centering
        \includegraphics[width=\textwidth]{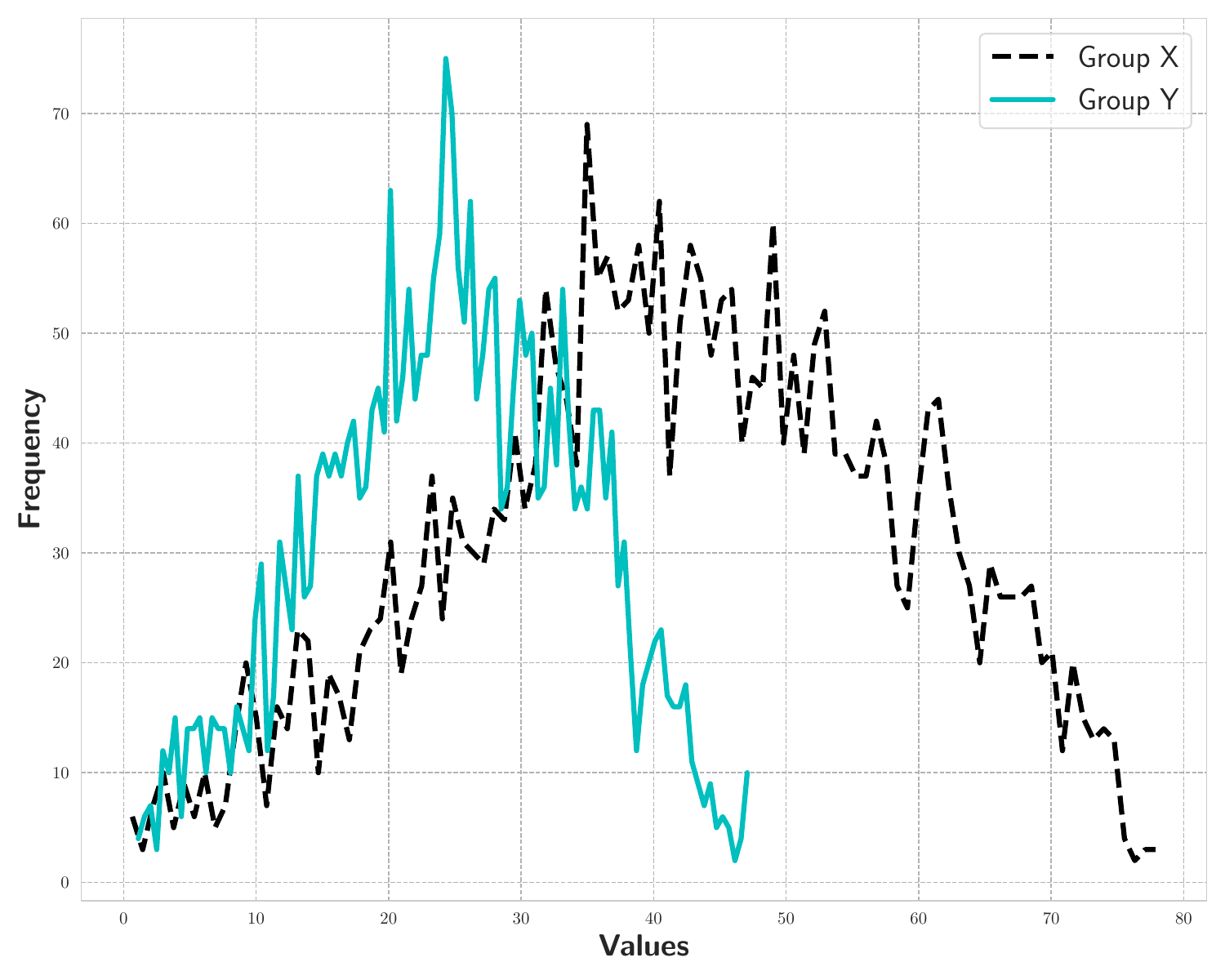}
        \caption{Bias factor = 0.6}
    \end{subfigure}
    \hfill
    \begin{subfigure}[b]{0.32\textwidth}
        \centering
        \includegraphics[width=\textwidth]{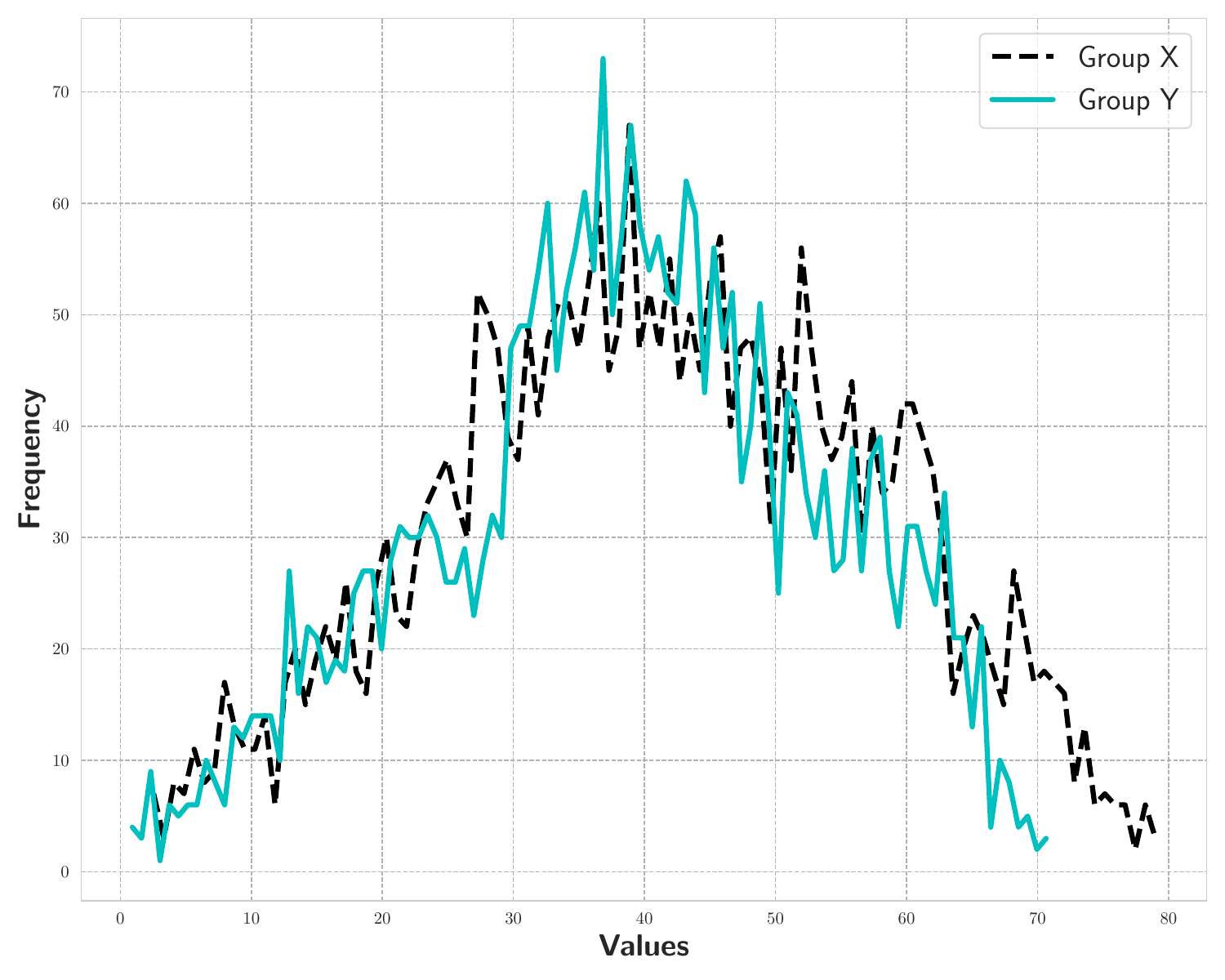}
        \caption{Bias factor = 0.9}
    \end{subfigure}
    \caption{Sample histograms of the generated values $\boldsymbol{\{{v}_i\}_{i\in[1:60]}}$ for the groups $\boldsymbol{\WomanSet}$ (cyan) and $\boldsymbol{\ManSet}$ (black).
    }
    \label{fig-apx:v1_Histogram_uniform}
\end{figure}

In the remainder of this part, we list the results of all the new simulations for this new instance. We encourage the reader to compare these results with \Cref{sec:numerical} and \Cref{apx:numerical-main}. As can be seen, while the resulting curves are (obviously) slightly different, there is no qualitative difference between these simulations and our previous set of simulations, indicating the robustness of our numerical results.

\subsubsection{Short-term {outcomes: (Surprisingly)} small utilitarian loss}
\label{sec:num-short-term_uniform}
\Cref{fig-rad:short-term_uniform_both_utilities} illustrates the short-term utilities of both optimal unconstrained policy and our proposed optimal constrained policy. Moreover, \Cref{fig-rad:short-term_uniform_price_of_fairness} shows the price of fairness together with the normalized slack of the optimal unconstrained policy.

\begin{figure}[htb]
    \centering
    \begin{subfigure}[b]{0.44\textwidth}
        \centering
        \includegraphics[width=\textwidth]{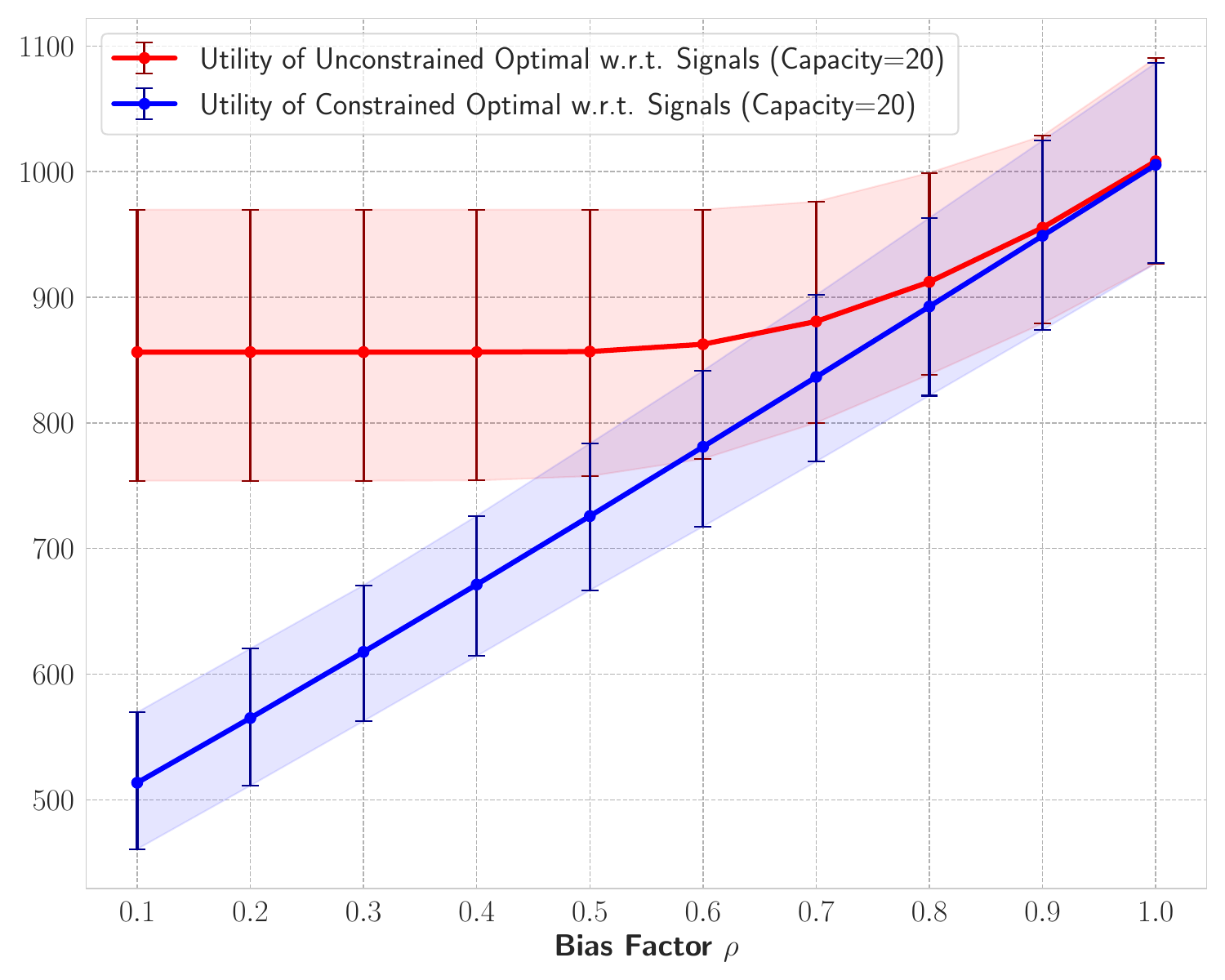}
        \caption{Expected utilities calculated based on signals $\{v_i\}_{i\in[n]}$ for the unconstrained optimal policy (red) and the constrained optimal policy (blue).}
        \label{fig-rad:short-term_uniform_both_utilities}
    \end{subfigure}
    \hfill
    \begin{subfigure}[b]{0.44\textwidth}
        \centering
        \includegraphics[width=\textwidth]{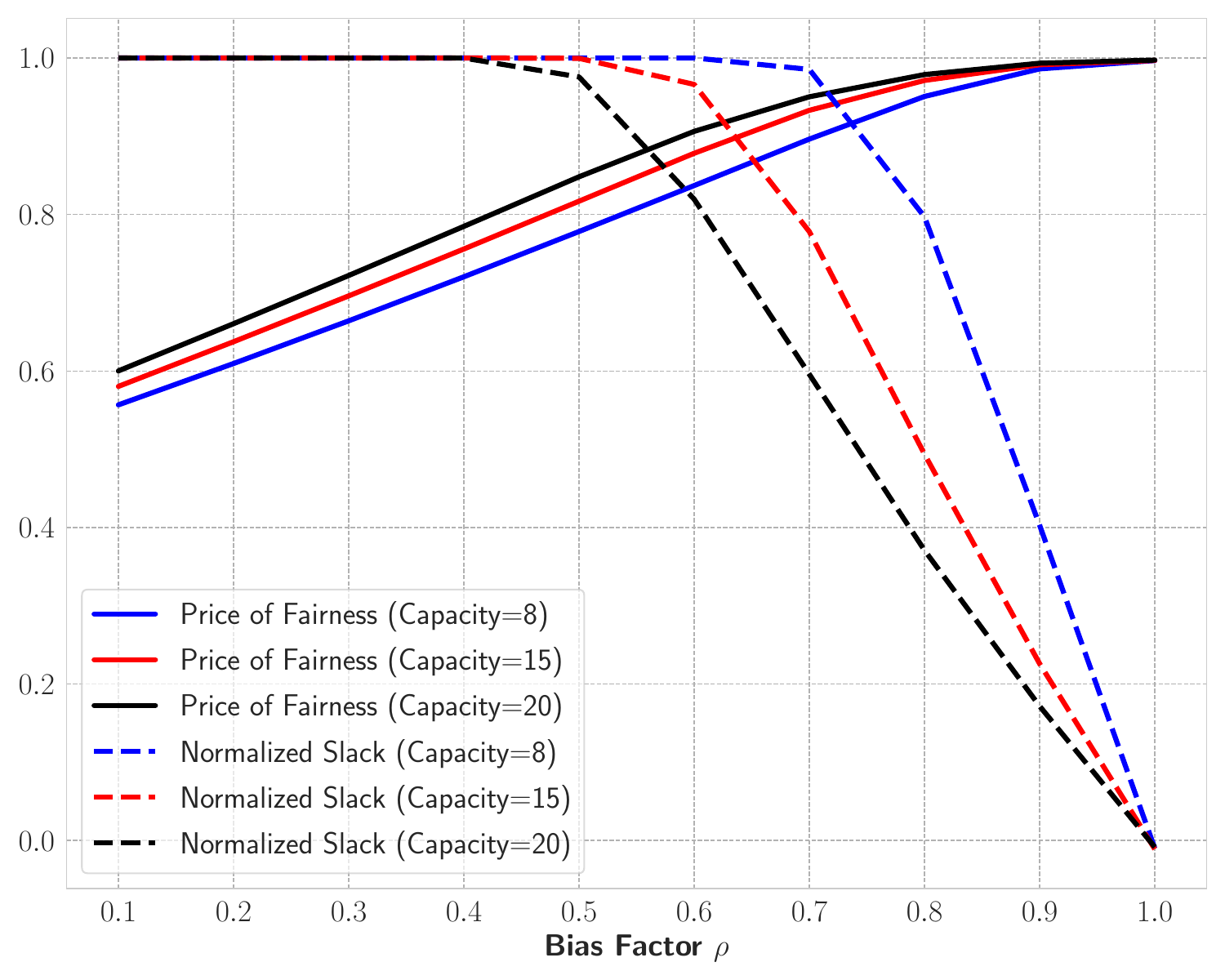}
        \caption{Price of fairness ratio calculated based on signals $\{v_i\}_{i\in[n]}$ (solid lines) and the normalized constraint slack of unconstrained optimal policy (dashed lines).}
        \label{fig-rad:short-term_uniform_price_of_fairness}
    \end{subfigure}
    \caption{Comparing the short-term outcomes of unconstrained and constrained optimal policies.}
    \label{fig-rad:short-term_uniform}
\end{figure}

\subsubsection{Long-term outcome: potential utilitarian gain}
\label{sec:num-long-term_uniform}
\Cref{fig-rad:long-term_uniform_both_utilities} illustrates the long-term utilities of both optimal unconstrained policy and our proposed optimal constrained policy. Moreover, \Cref{fig-rad:long-term_uniform_price_of_fairness} shows the price of fairness together with the normalized slack of the optimal unconstrained policy.

\begin{figure}[htb]
    \centering
    \begin{subfigure}[b]{0.44\textwidth}
        \centering
        \includegraphics[width=\textwidth]{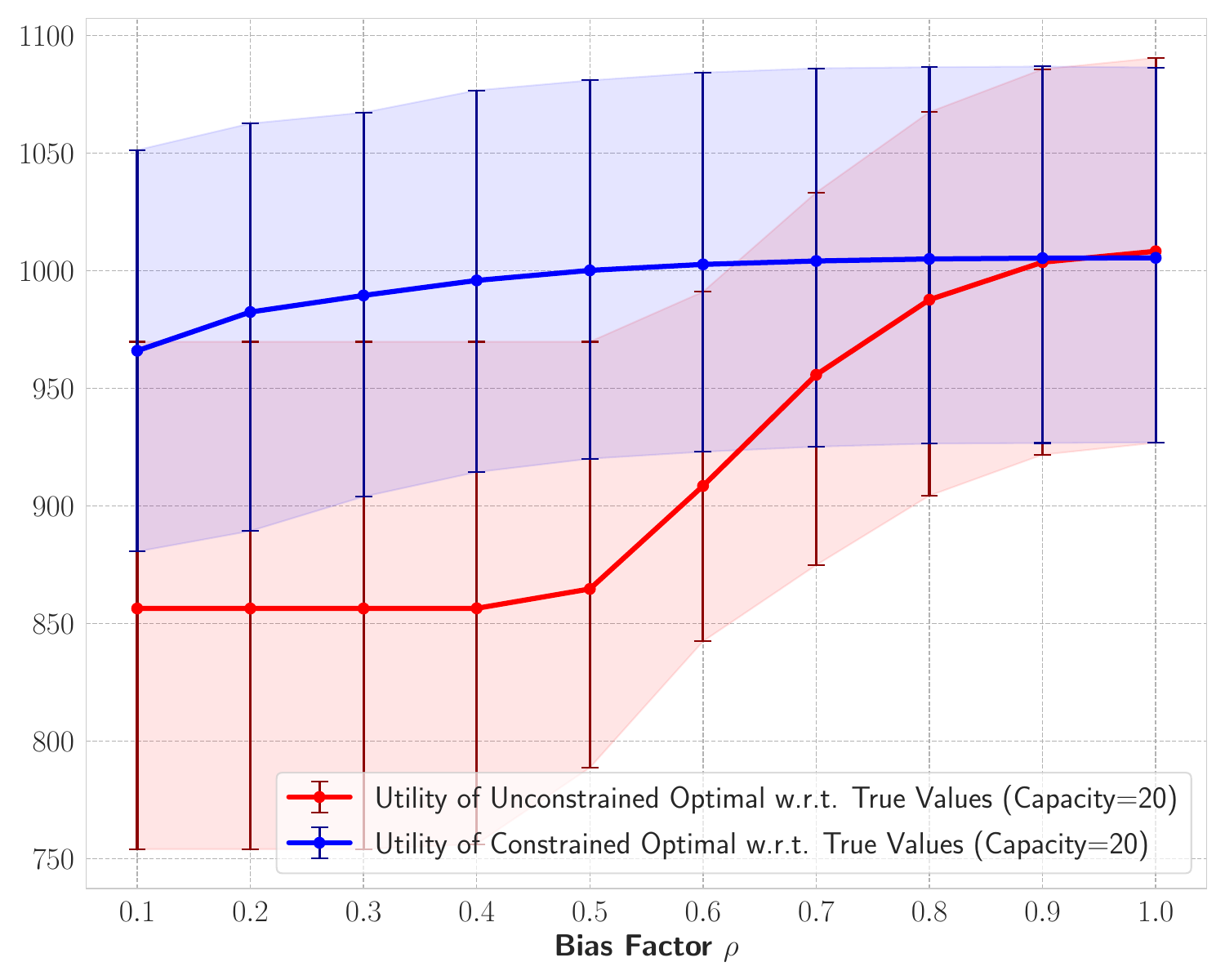}
        \caption{Expected utilities calculated based on true values $\{{v}^\dagger_i\}_{i\in[n]}$ for the unconstrained optimal policy (red) and the constrained optimal policy (blue).}
        \label{fig-rad:long-term_uniform_both_utilities}
    \end{subfigure}
    \hfill
        \begin{subfigure}[b]{0.44\textwidth}
        \centering
        \includegraphics[width=\textwidth]{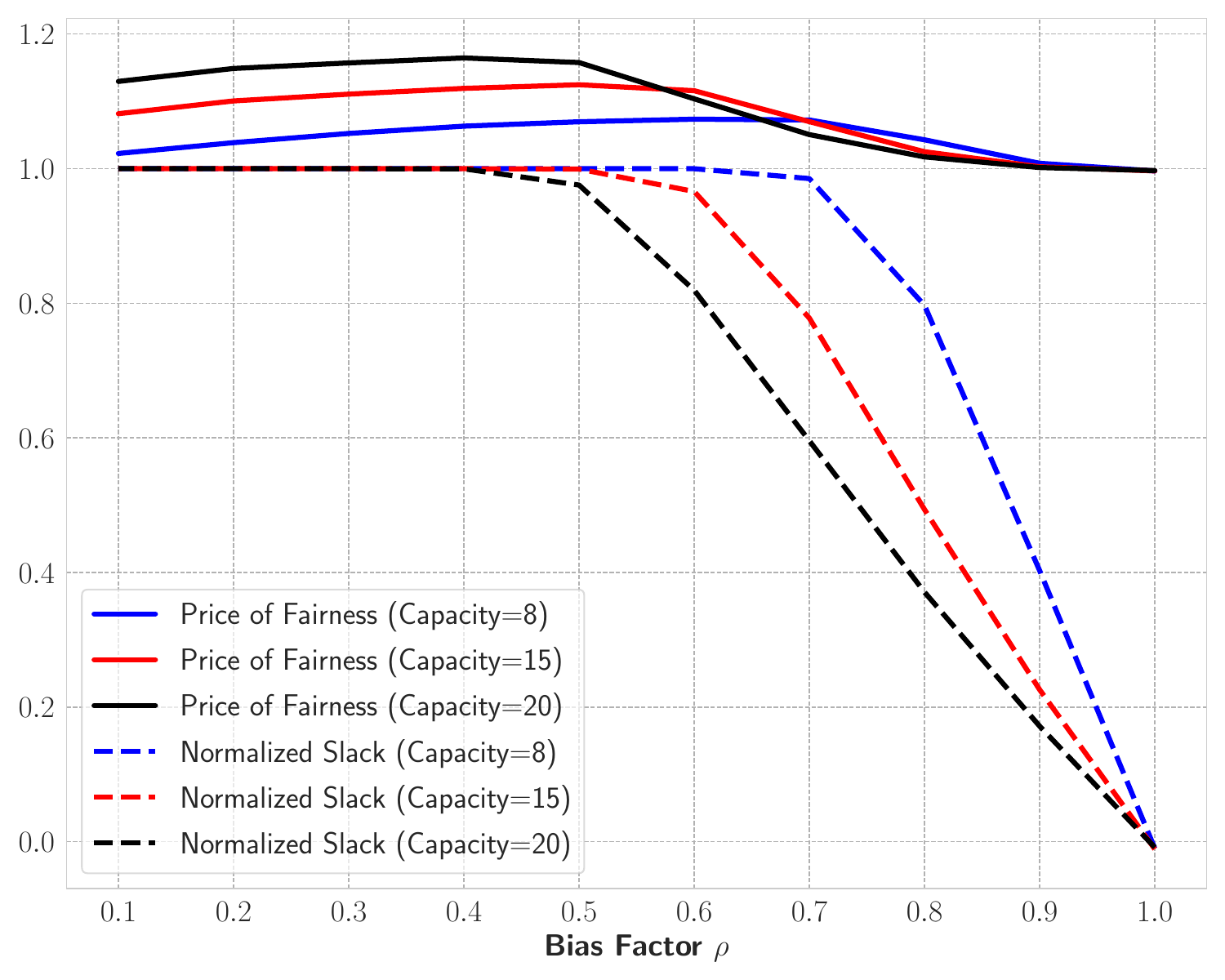}
        \caption{Price of fairness ratio calculated based on true values $\{v^\dagger_i\}_{i\in[n]}$ (solid lines) and the normalized constraint slack of unconstrained optimal policy (dashed lines).}
        \label{fig-rad:long-term_uniform_price_of_fairness}
    \end{subfigure}
    \caption{Comparing the long-term outcomes of unconstrained and constrained optimal policies}
    \label{fig-rad:long-term_uniform}
\end{figure}

\subsubsection{Minimum Quota Constraint}
\label{sec:num-long-term_uniform_quota}

\Cref{fig-rad:quota_uniform} illustrates short-term and long-term price of fairness across various quota parameters.

\begin{figure}[htb]
    \centering
    \begin{subfigure}[b]{0.44\textwidth}
        \centering
        \includegraphics[width=\textwidth]{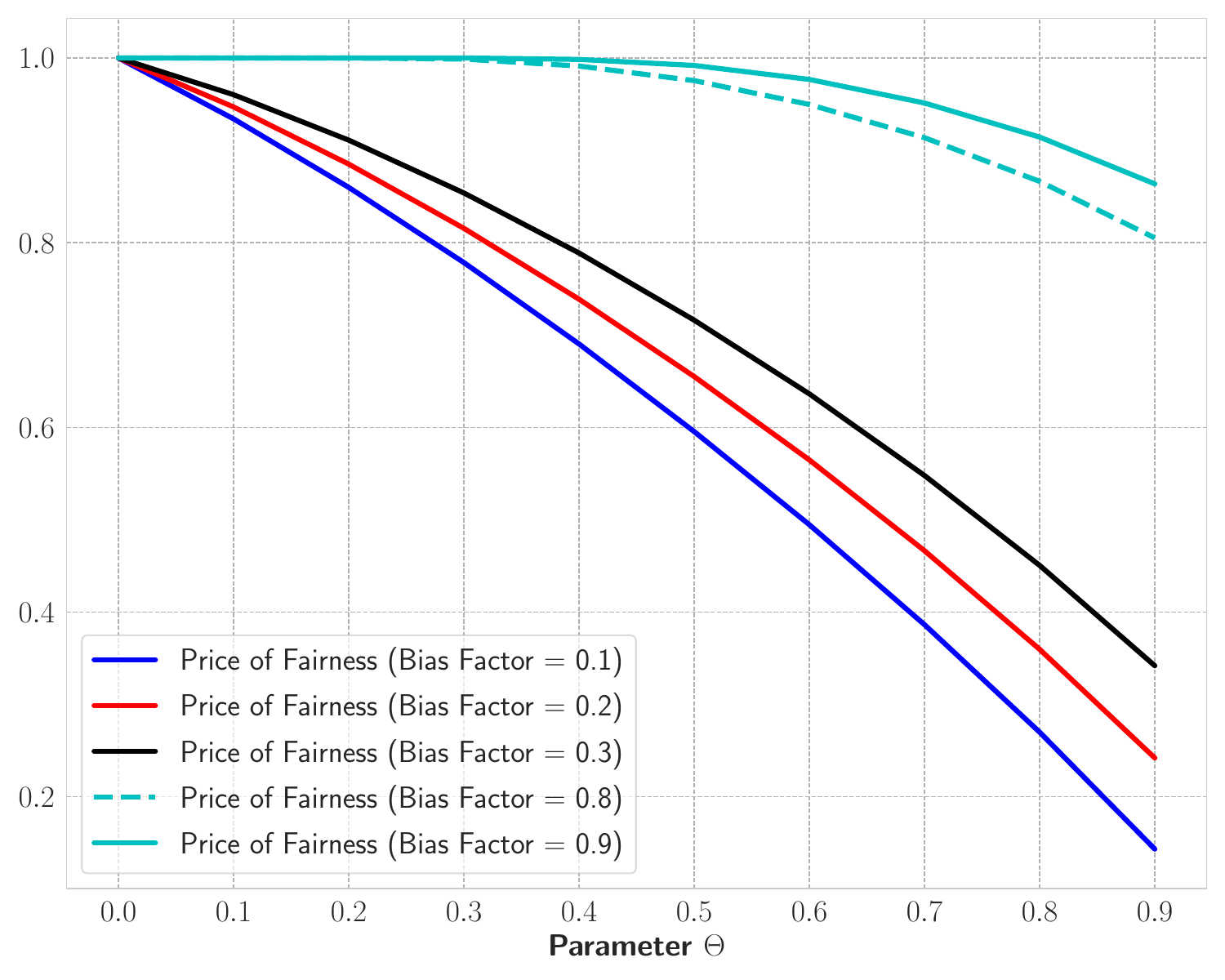}
        \caption{Short-term price of fairness}
    \end{subfigure}
    \hfill
        \begin{subfigure}[b]{0.44\textwidth}
        \centering
        \includegraphics[width=\textwidth]{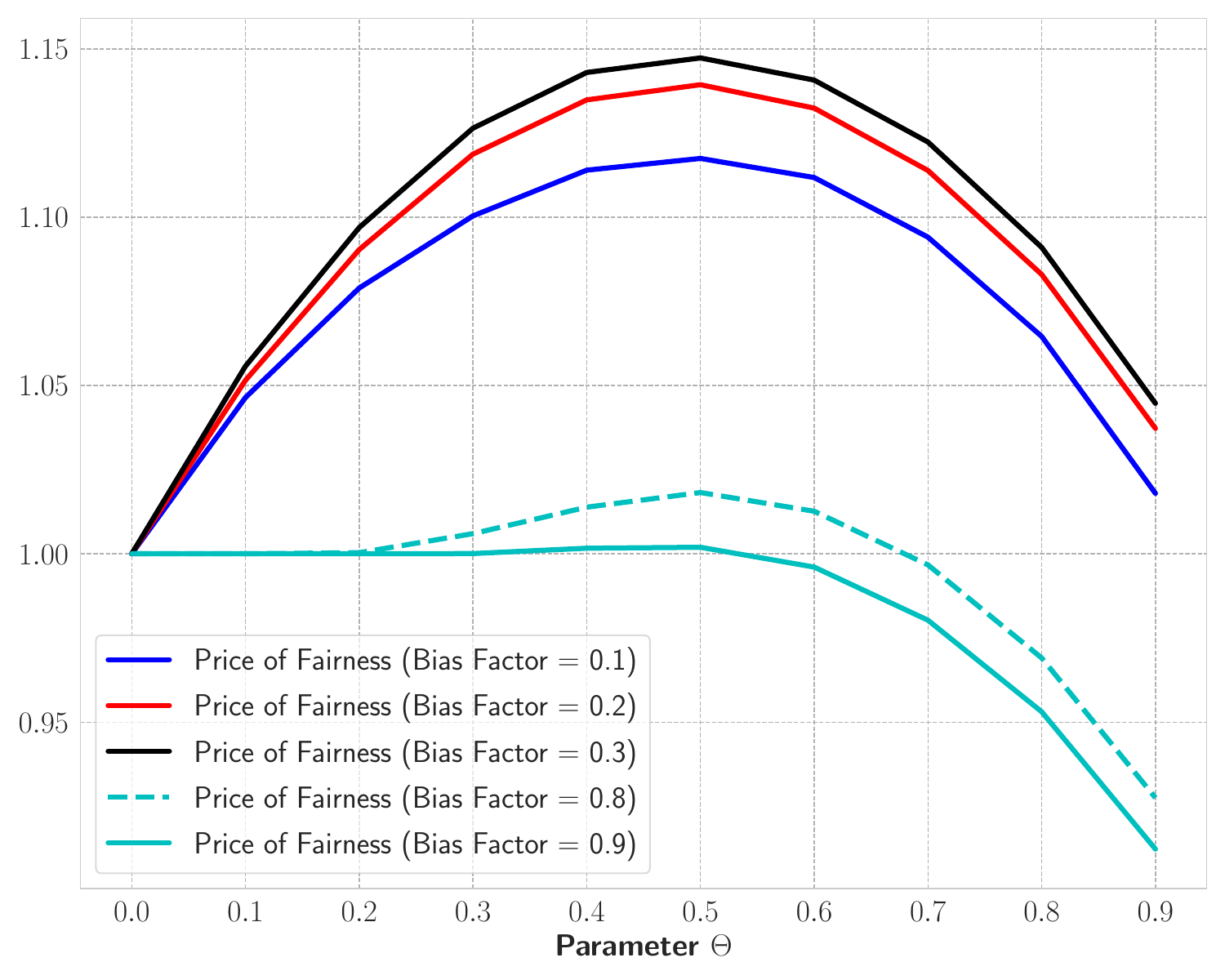}
        \caption{Long-term price of fairness}
    \end{subfigure}
    \caption{Performance of optimal constrained policy for \ref{eq:quota} in selection with parameter $\boldsymbol{\theta}$ ($\textbf{k=20}$).}
    \label{fig-rad:quota_uniform}
\end{figure}

\newpage

\subsection{Numerical Simulations - JMS}
\label{sec:numerical_JMS}
In this section, we study and analyze the performance of \Cref{alg:RAI} on a set of instances for the JMS problem. The main goal of our simulations in this section is to study the running time and convergence of \Cref{alg:RAI} as an iterative algorithm and  a FPTAS to the optimal policy, but we also study the utility of the search obtained by this policy (with respect to the observable signals).

More specifically, we consider the JMS instance provided in \Cref{ex:reject} (illustrated in \Cref{fig:MC-hiring}), which was a two-stage search with the possibility of rejection. We then consider three constraints that we would like to satisfy all at the same time. More specifically, we want to satisfy the \ref{eq:parity} in selection in all three stages of the search process, namely ``phone interviews", ``on-site interviews" and ``offers".

\smallskip
\noindent\textbf{Basic simulation setup:}  
The value distribution for each of the alternatives is, in fact, generated the same way as in \Cref{sec:numerical}. As for the costs and transition probabilities for the extra stages that are apparent in this problem, we use the following setup:
\begin{itemize}
    \item Cost of phone interview stage: $\mathrm{Uniform}[1,2]$
    \item Cost of onsite interview stage: $\mathrm{Uniform}[2,4]$
    \item Cost of offer to each individual = 3
    \item Probability of passing the phone interview = 80\%
    \item Probability offer getting accepted = 90\%
\end{itemize}
In order to evaluate the performance of the algorithm in expectation, we run a Monete-Carlo simulation with $20$ instances. 

\medskip

\noindent\textbf{Short-term price of fairness:} \Cref{fig:JMS_short-term_price_of_fairness} shows the short-term utilities of both our near-optimal constrained policy and unconstrained optimal policy, as well as their ratio (price of fairness), under different capacities.
As can be seen from the plots, the price of fairness is quite small, especially for $\rho \in [0.6, 1]$. This shows that the negative externalities due to fairness considerations are small and negligible in practical instances similar to those we consider here.

\begin{figure}[htb]
    \centering
    \begin{subfigure}[b]{0.48\textwidth}
        \centering
        \includegraphics[width=\textwidth]{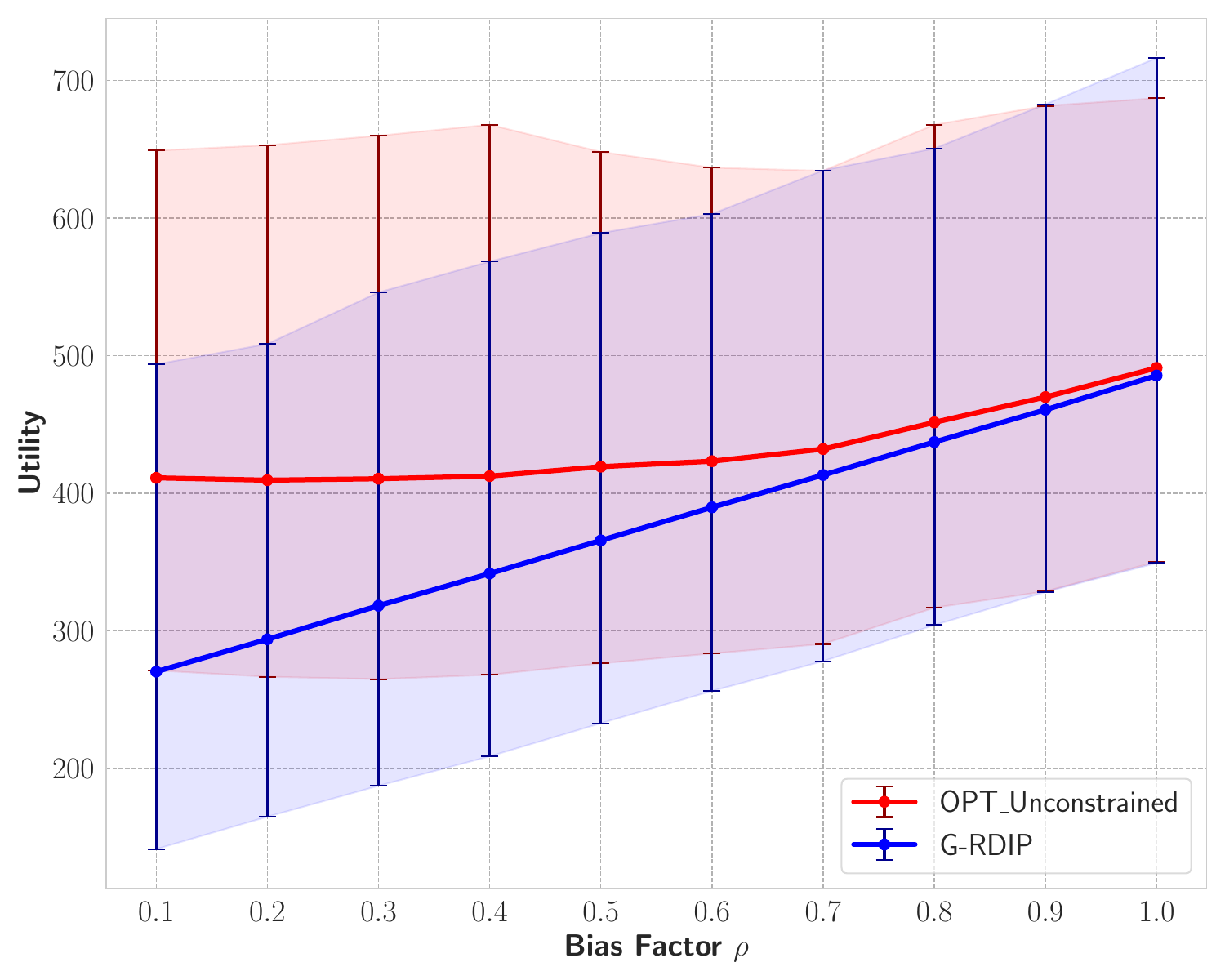}
        \caption{Capacity = 8}
    \end{subfigure}
    \hfill
    \begin{subfigure}[b]{0.48\textwidth}
        \centering
        \includegraphics[width=\textwidth]{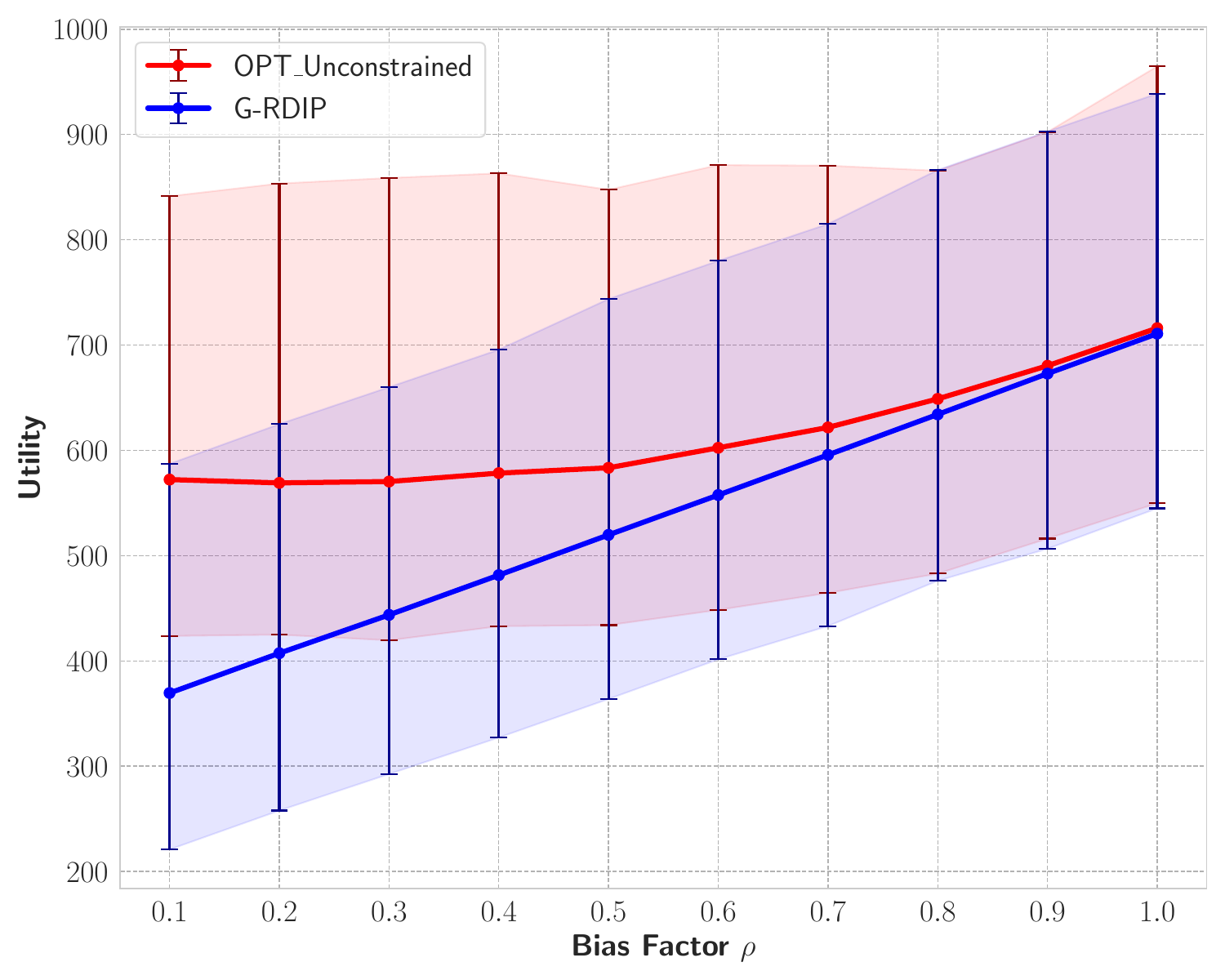}
        \caption{Capacity = 15}
    \end{subfigure}
    \hfill
    \begin{subfigure}[b]{0.48\textwidth}
        \centering
        \includegraphics[width=\textwidth]{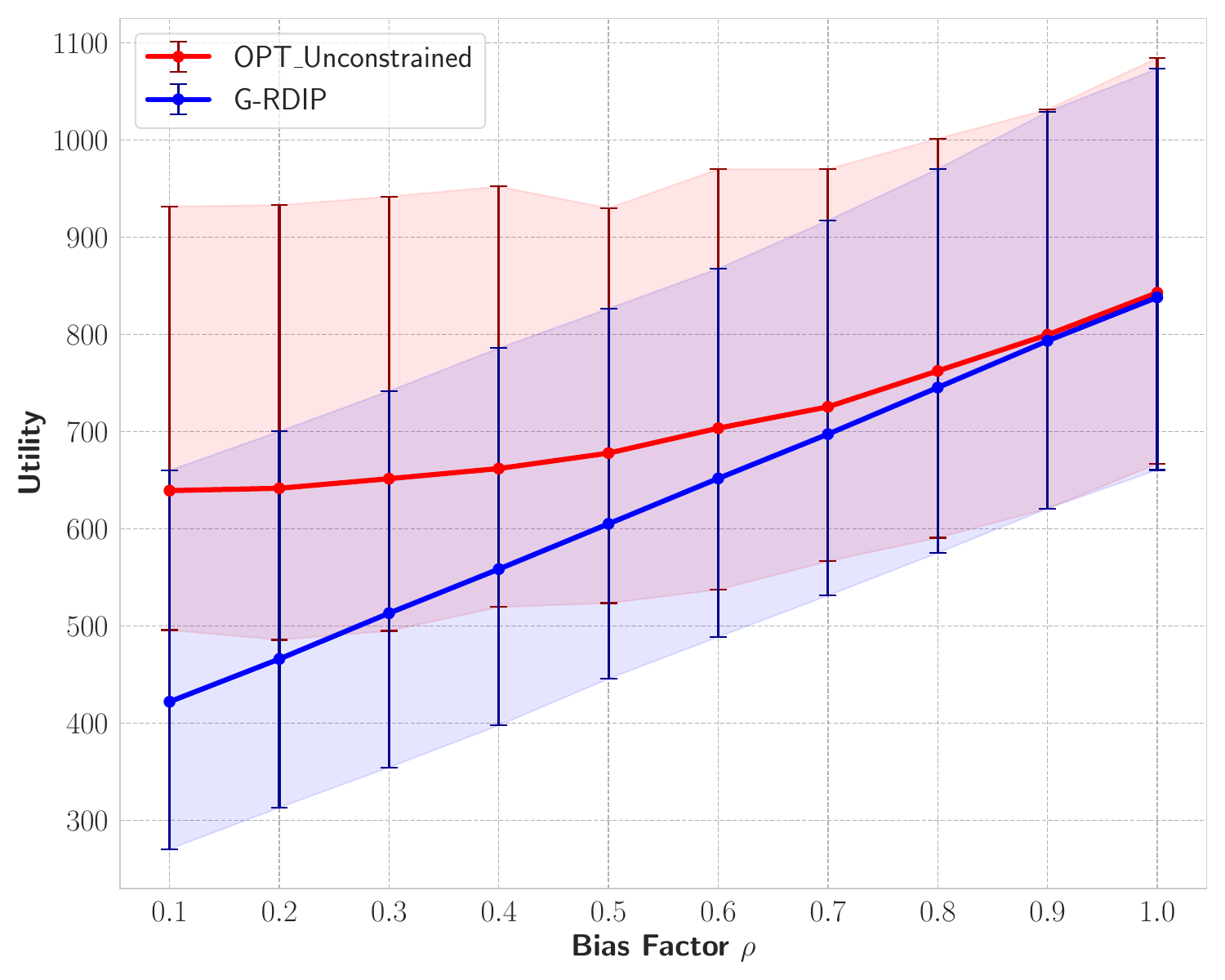}
        \caption{Capacity = 20}
    \end{subfigure}
        \hfill
    \begin{subfigure}[b]{0.48\textwidth}
        \centering
        \includegraphics[width=\textwidth]{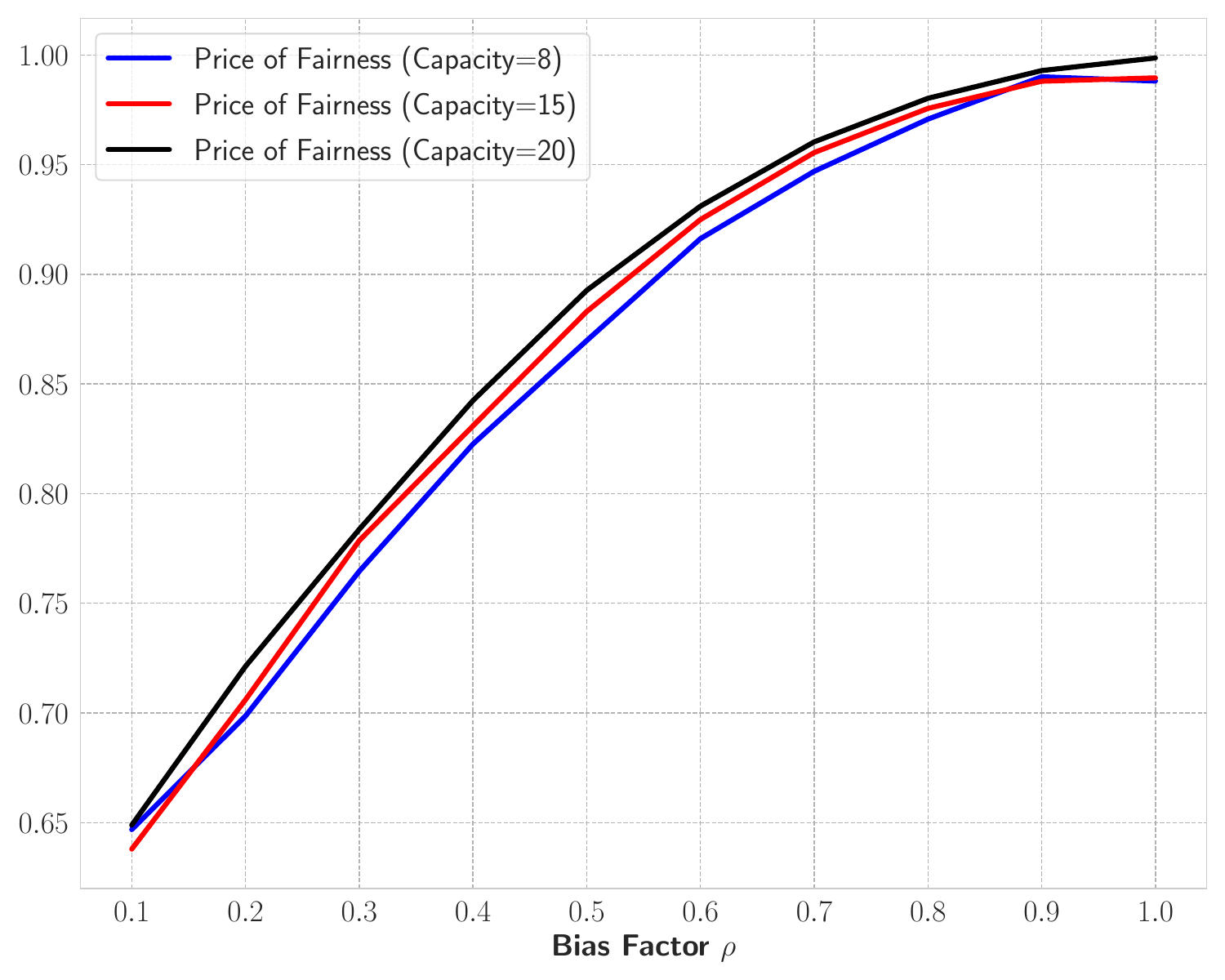}
        \caption{Price of fairness}
    \end{subfigure}
    
    \caption{(JMS simulation) (a,b,c) show the expected utilities calculated based on signals $\boldsymbol{\{v_i\}_{i\in[n]}}$ for the unconstrained optimal policy (red) and the constrained optimal policy (blue); and (d): Price of fairness ratio calculated based on signals $\boldsymbol{\{v_i\}_{i\in[n]}}$ (solid lines) and the normalized constraint slack of unconstrained optimal policy (dashed lines).}
    \label{fig:JMS_short-term_price_of_fairness}
\end{figure}

\medskip
\noindent\textbf{Convergence trajectory:} The following \Cref{fig:JMS_Trajectory_Lagrangian,fig:JMS_Trajectory_Mean_Lagrangian,fig:JMS_Trajectory_Mean_Slack_Offer,fig:JMS_Trajectory_Mean_Slack_Onsite,fig:JMS_Trajectory_Mean_Slack_Phone,fig:JMS_Trajectory_lambdas} illustrates the trajectory of Lagrangian, mean (over the past iterations) of Lagrangian, mean (over the past iterations) of slacks for each of the three constraints (each corresponding to the parity at one of the stages), as well as the dual adjustments $\boldsymbol{\lambda}$.
As shown by all the figures, we can see that all these metrics will converge to their goal in around $\OuterNum=20$ to $\OuterNum=40$ number of iterations. Note that these are only outer iterations of \Cref{alg:RAI}, as we do not have any convex constraints in this set of simulations and there is no need for the inner-loop. This demonstrates that, even though the theoretical number of iterations derived in \Cref{thm:RAI} can be quite large, the actual number of iterations need for convergence is quite small under practical instances.

\begin{figure}[htb]
    \centering
    \begin{subfigure}[b]{0.32\textwidth}
        \centering
        \includegraphics[width=\textwidth]{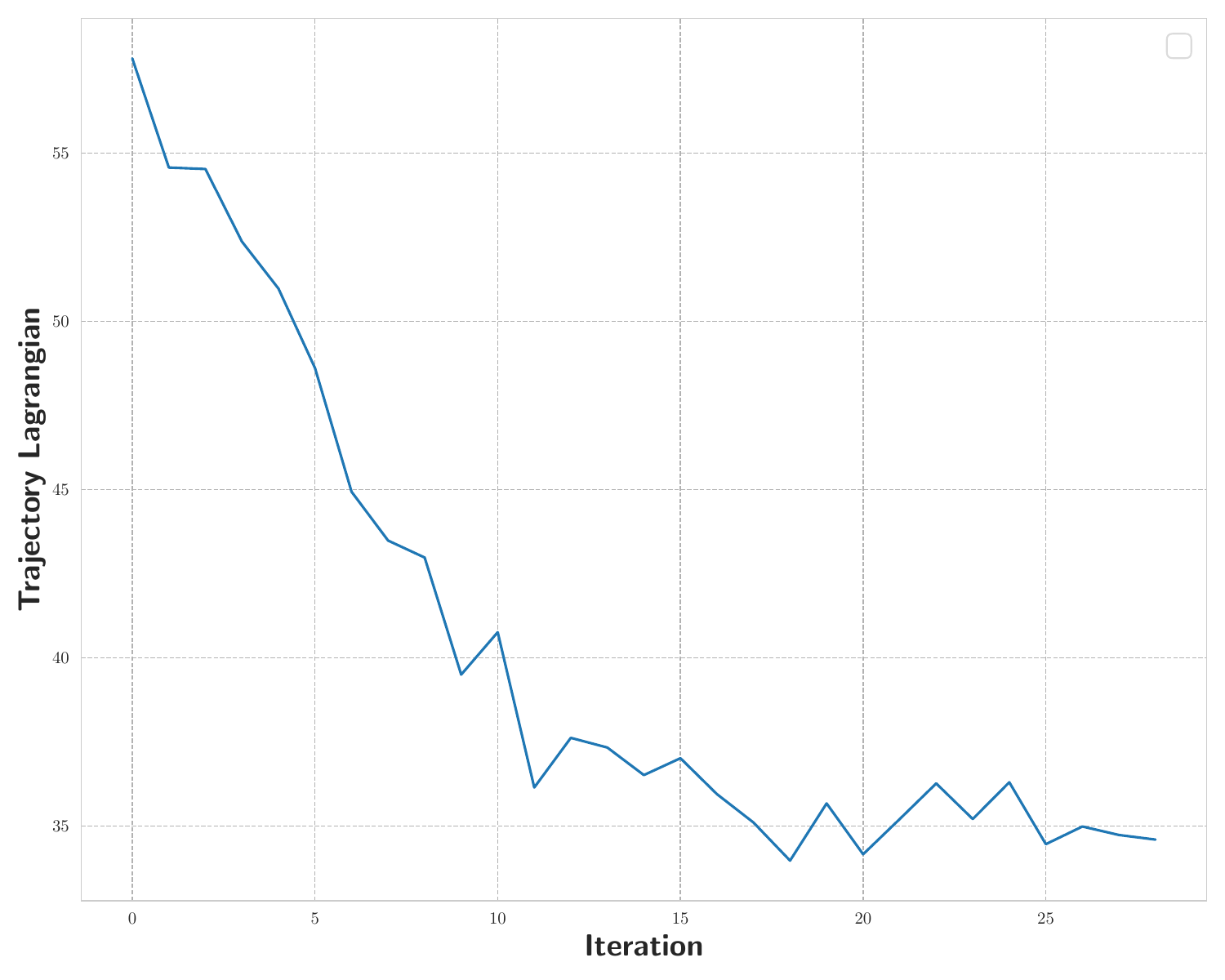}
        \caption{Bias factor = 0.3}
    \end{subfigure}
    \hfill
    \begin{subfigure}[b]{0.32\textwidth}
        \centering
        \includegraphics[width=\textwidth]{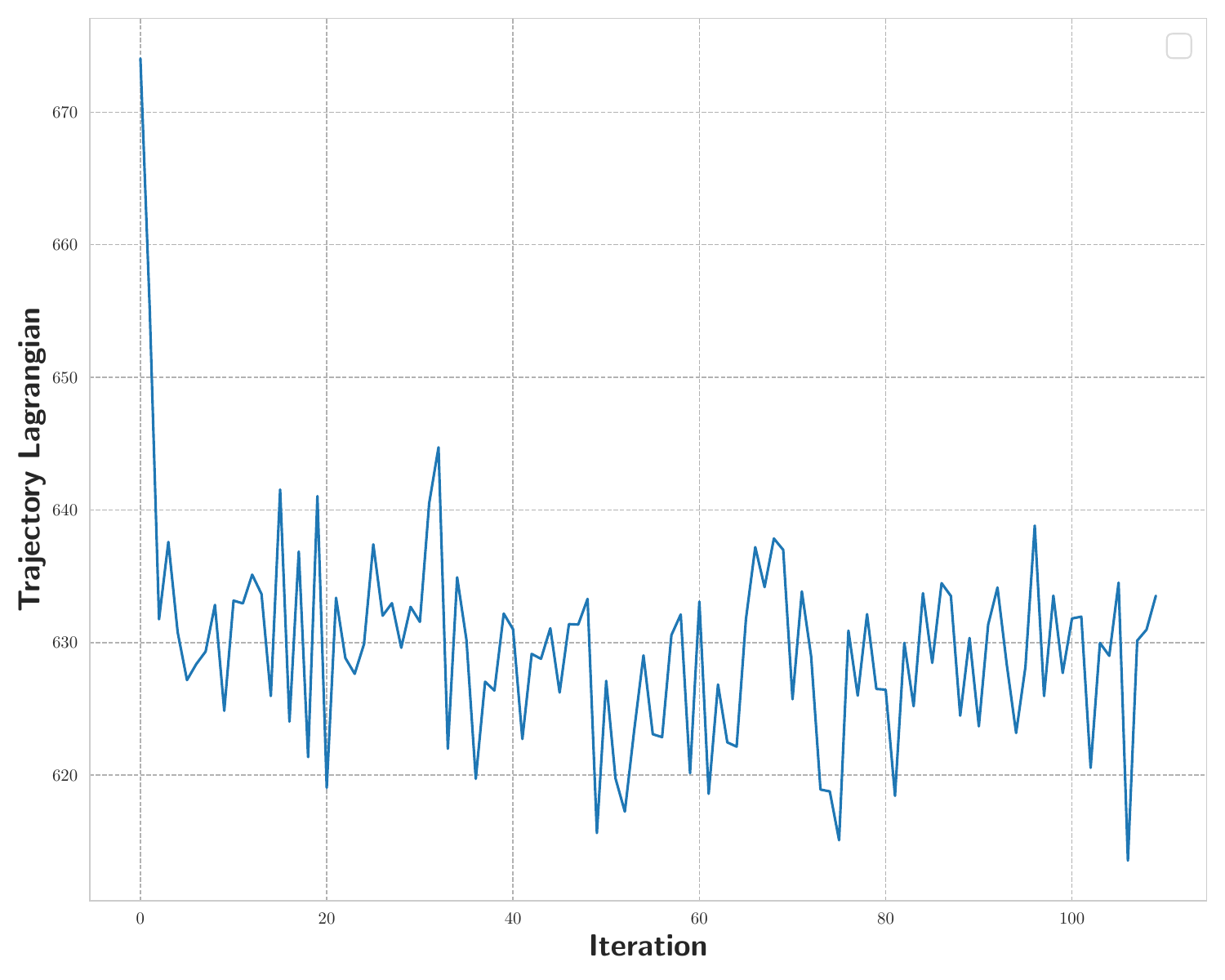}
        \caption{Bias factor = 0.7}
    \end{subfigure}
    \hfill
    \begin{subfigure}[b]{0.32\textwidth}
        \centering
        \includegraphics[width=\textwidth]{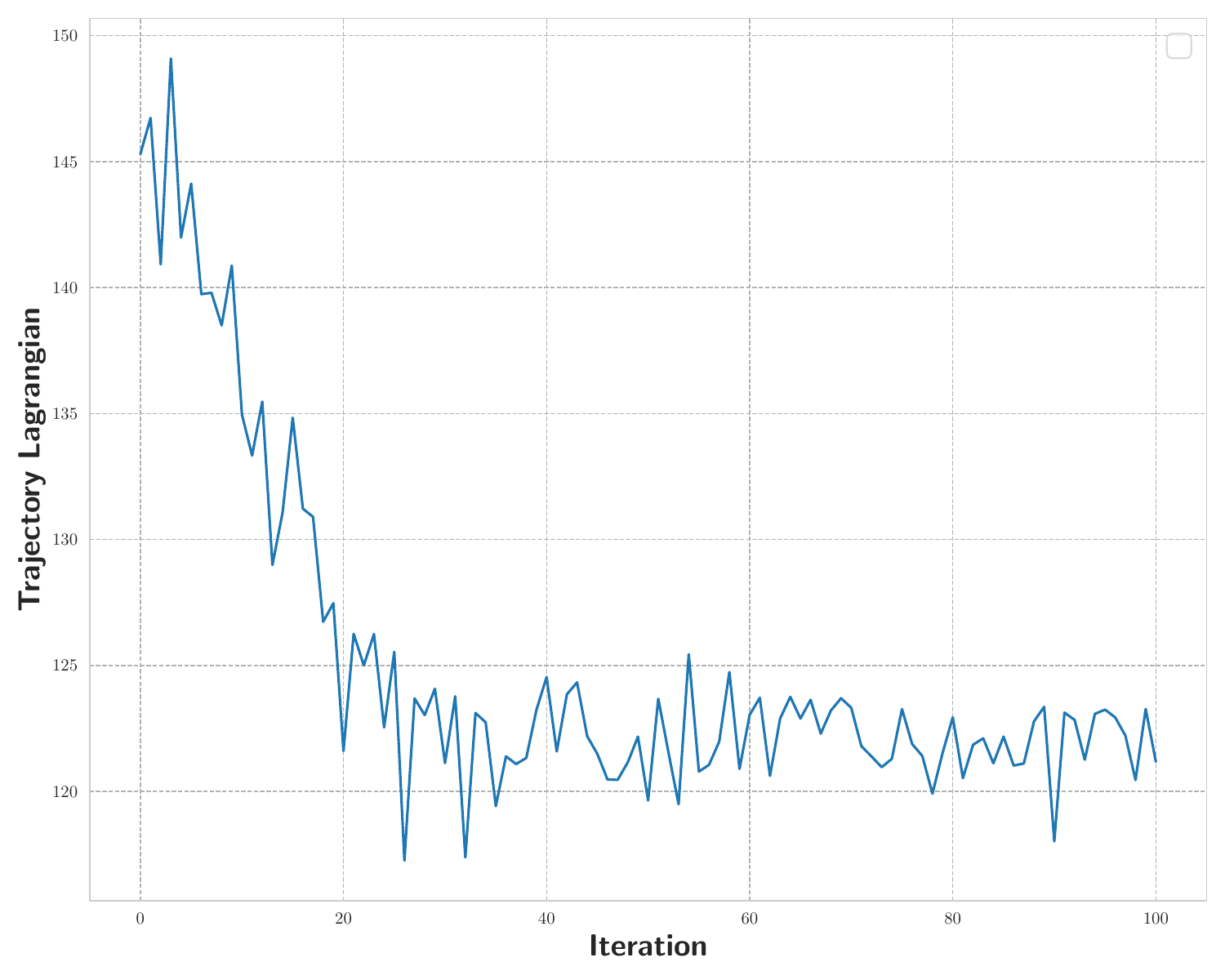}
        \caption{Bias factor = 0.9}
    \end{subfigure}
    
    \caption{(JMS simulation) Trajectory of Lagrangian under capacity $\boldsymbol{k= 20}$}
    \label{fig:JMS_Trajectory_Lagrangian}
\end{figure}

\begin{figure}[htb]
    \centering
    \begin{subfigure}[b]{0.32\textwidth}
        \centering
        \includegraphics[width=\textwidth]{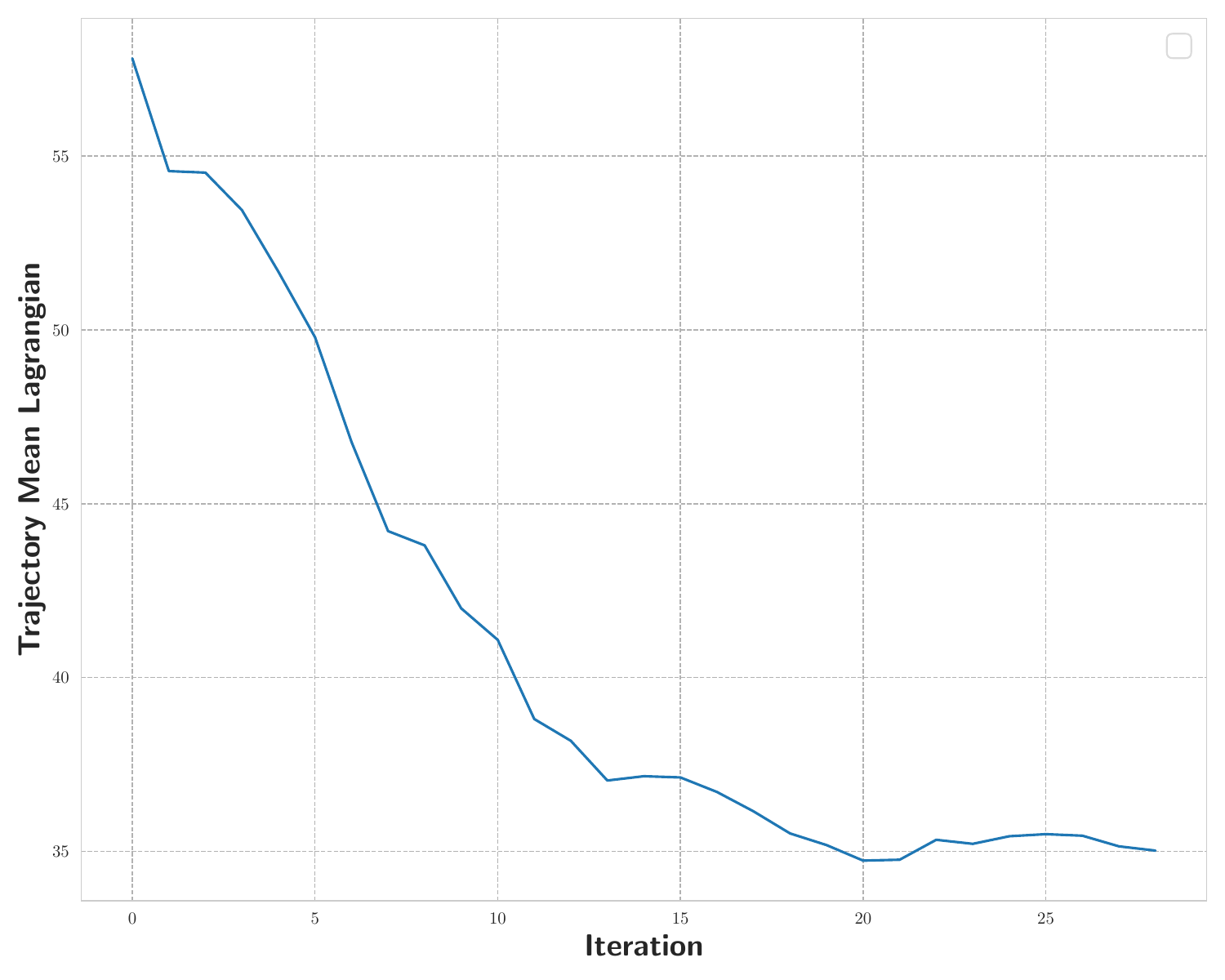}
        \caption{Bias factor = 0.3}
    \end{subfigure}
    \hfill
    \begin{subfigure}[b]{0.32\textwidth}
        \centering
        \includegraphics[width=\textwidth]{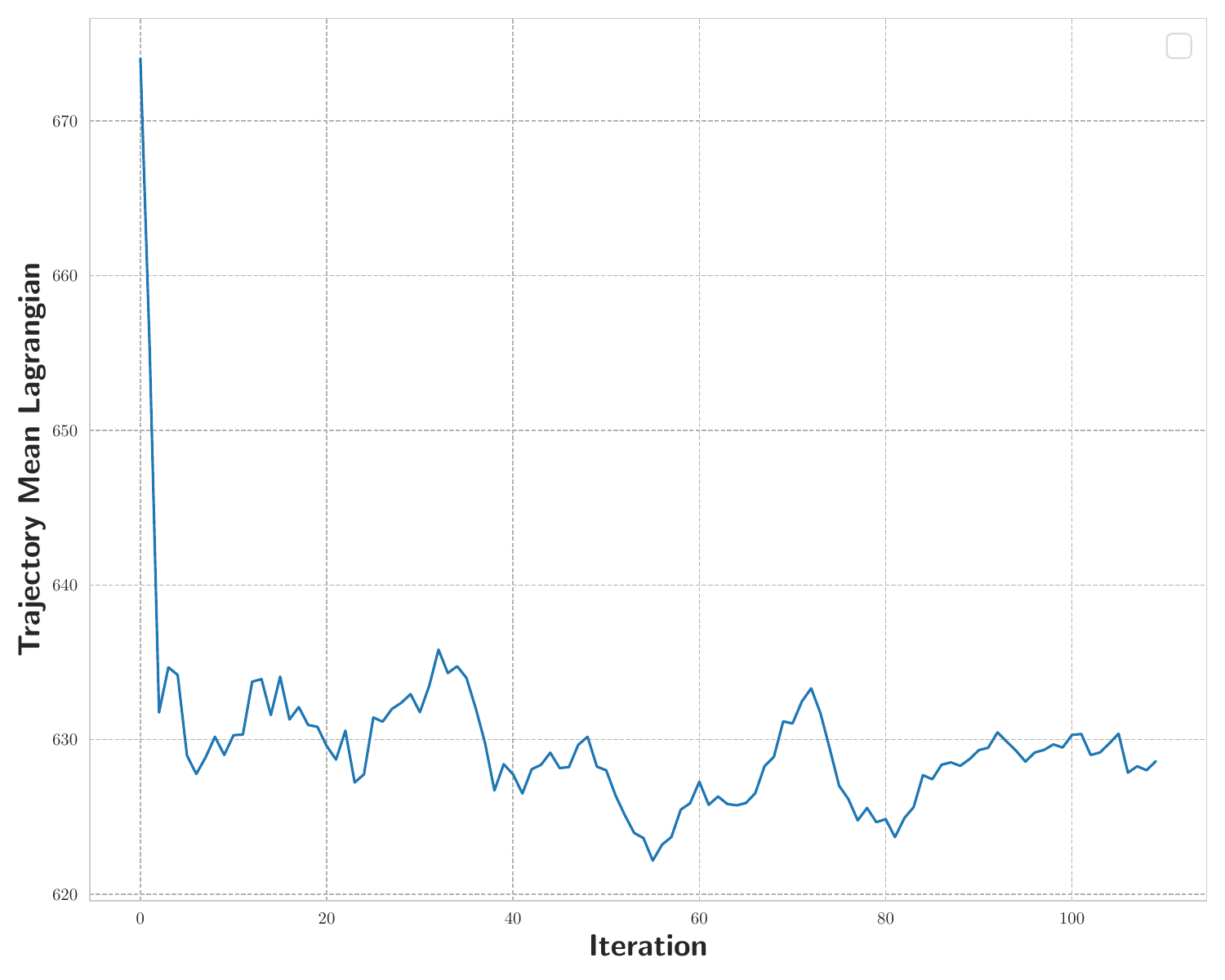}
        \caption{Bias factor = 0.7}
    \end{subfigure}
    \hfill
    \begin{subfigure}[b]{0.32\textwidth}
        \centering
        \includegraphics[width=\textwidth]{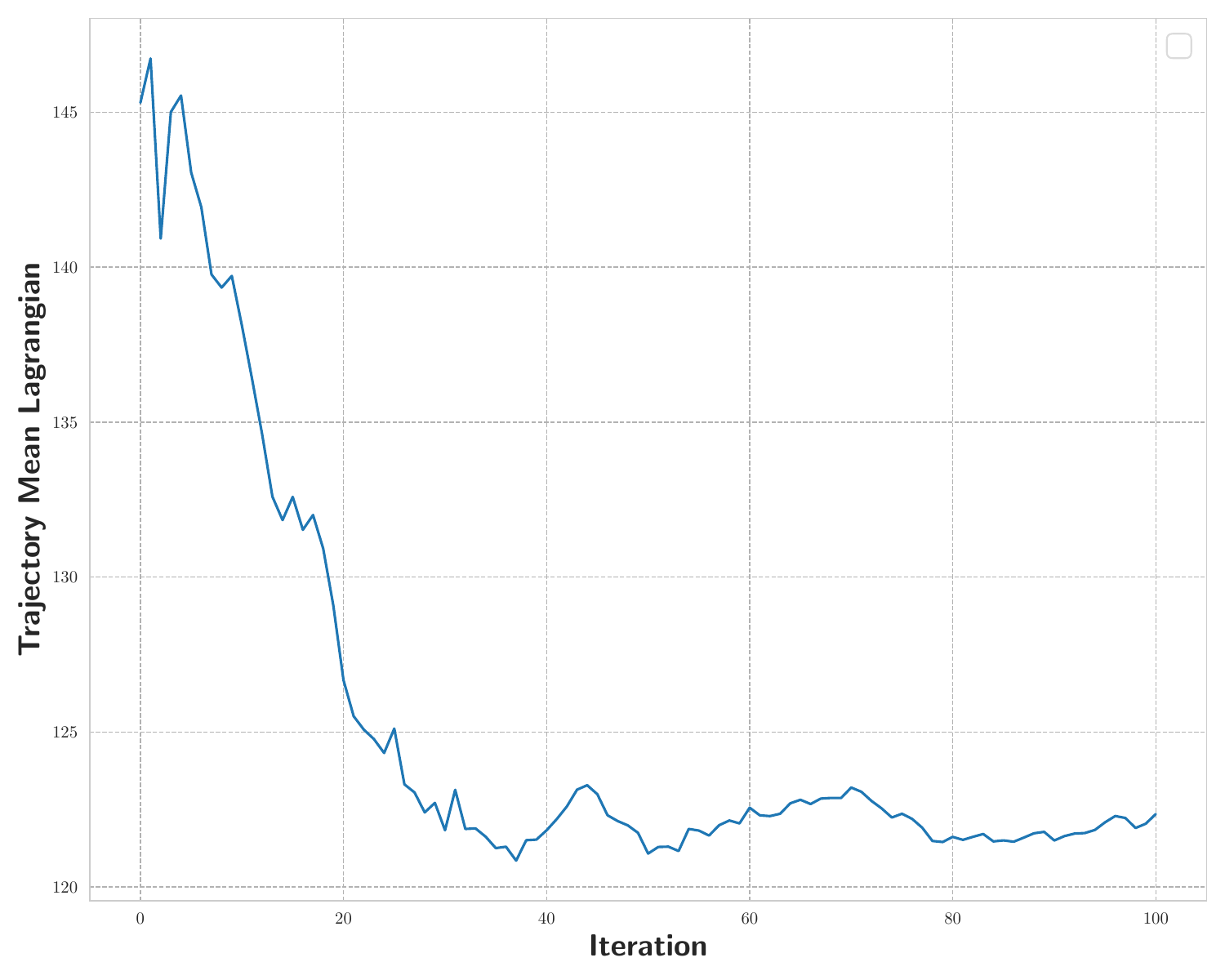}
        \caption{Bias factor = 0.9}
    \end{subfigure}
    
    \caption{(JMS simulation) Trajectory of mean Lagrangian under capacity $\boldsymbol{k= 20}$}
    \label{fig:JMS_Trajectory_Mean_Lagrangian}
\end{figure}

\begin{figure}[htb]
    \centering
    \begin{subfigure}[b]{0.32\textwidth}
        \centering
        \includegraphics[width=\textwidth]{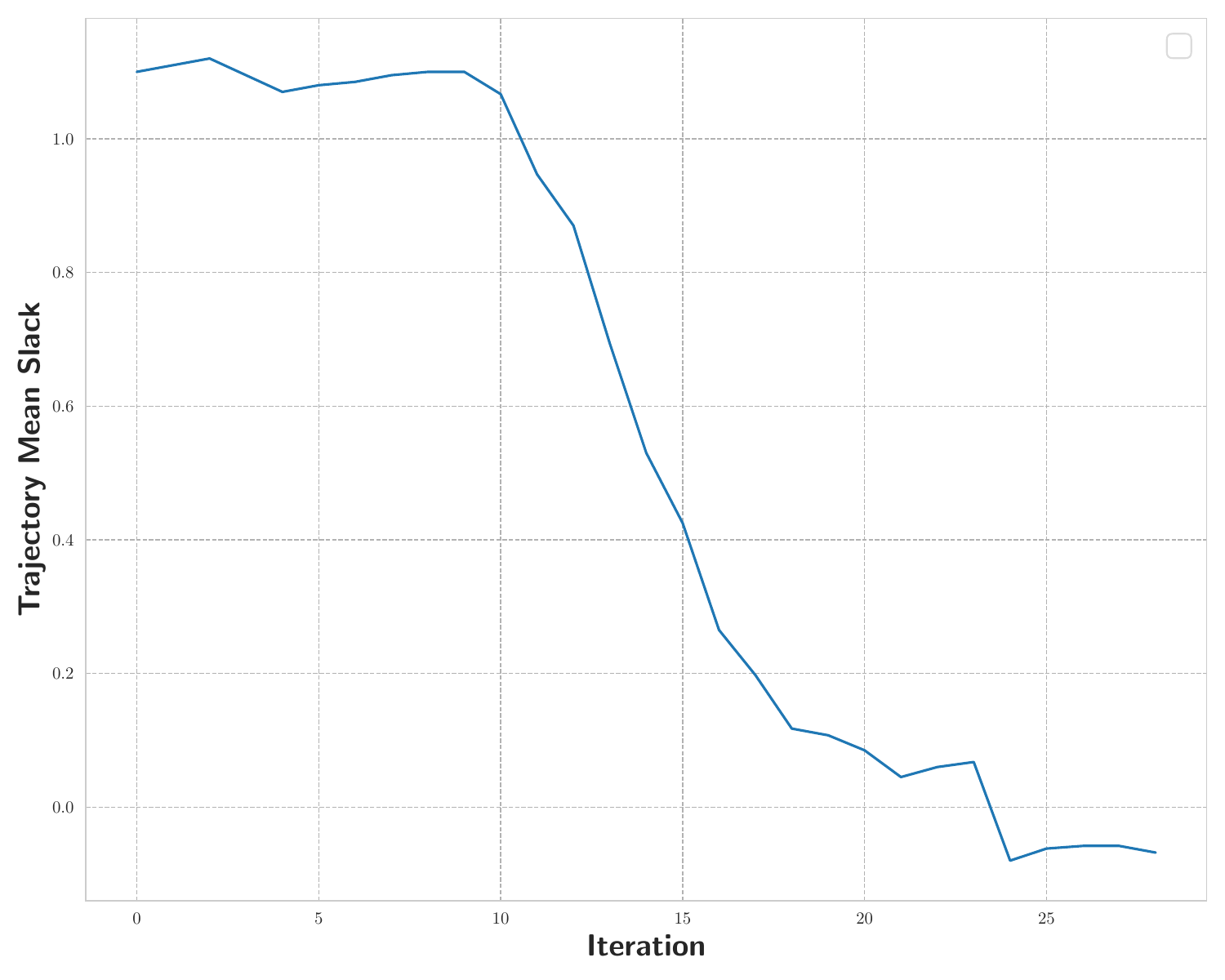}
        \caption{Bias factor = 0.3}
    \end{subfigure}
    \hfill
    \begin{subfigure}[b]{0.32\textwidth}
        \centering
        \includegraphics[width=\textwidth]{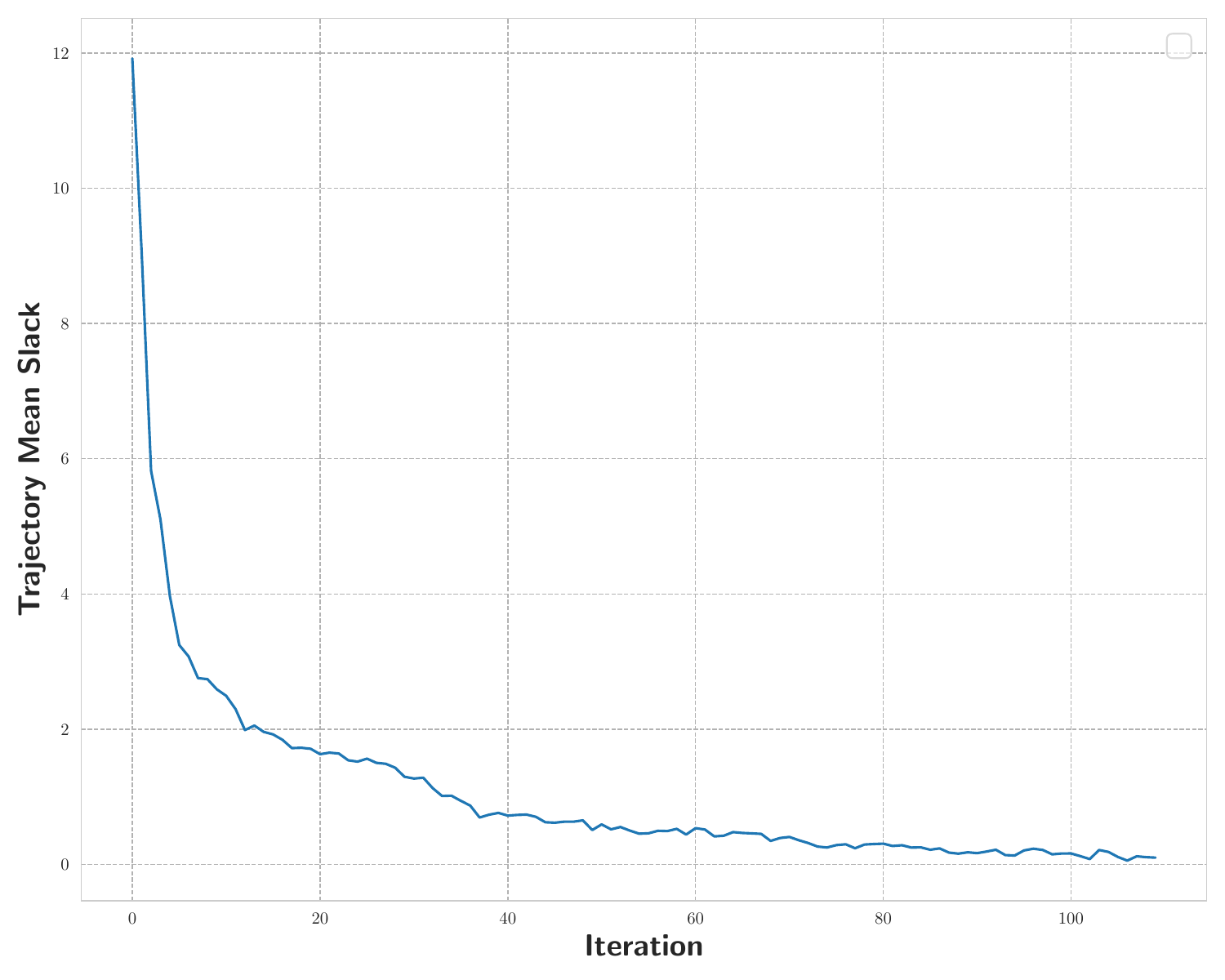}
        \caption{Bias factor = 0.7}
    \end{subfigure}
    \hfill
    \begin{subfigure}[b]{0.32\textwidth}
        \centering
        \includegraphics[width=\textwidth]{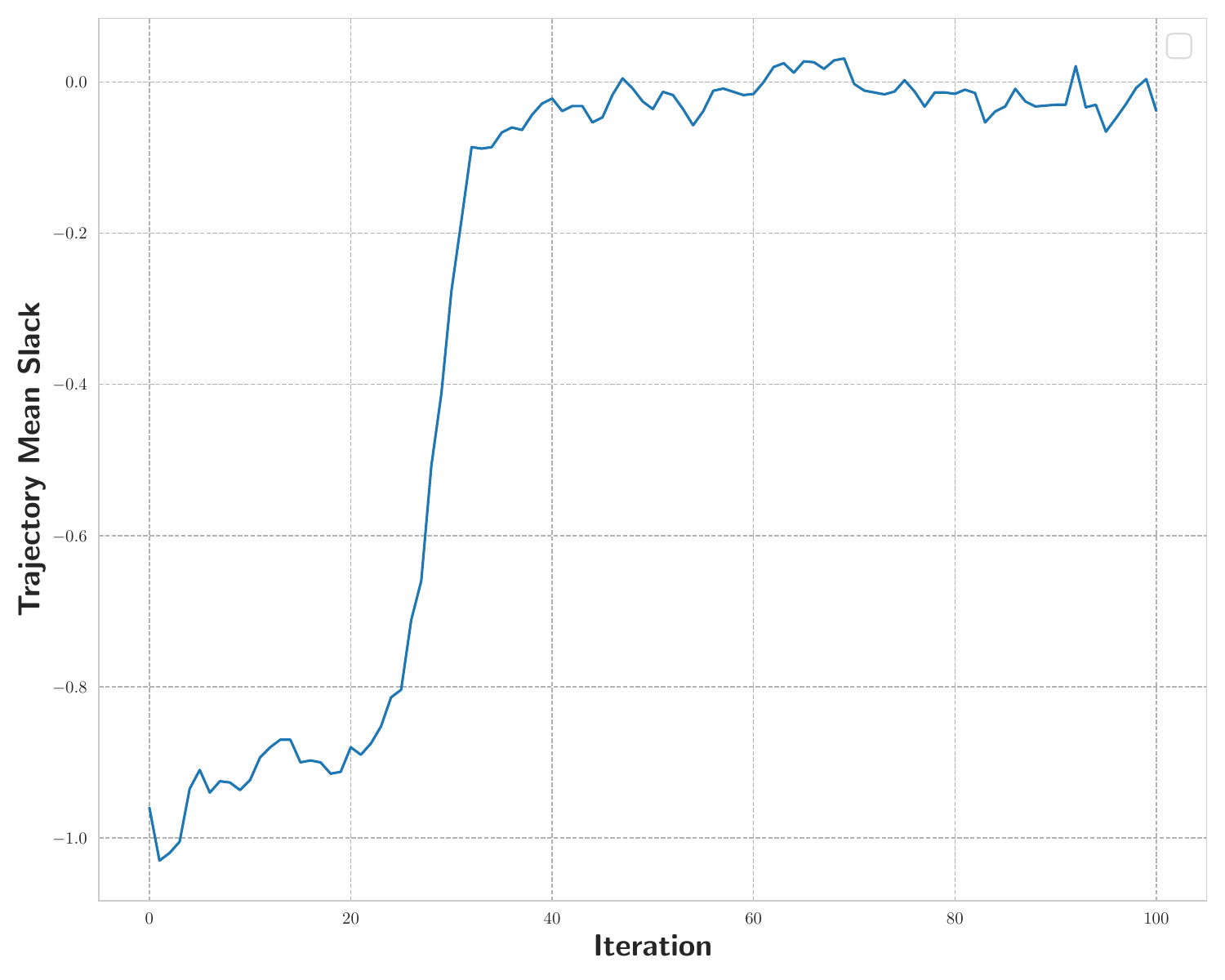}
        \caption{Bias factor = 0.9}
    \end{subfigure}
    
    \caption{(JMS simulation) Trajectory of mean slack for the parity at the offer stage under capacity $\boldsymbol{k= 20}$}
    \label{fig:JMS_Trajectory_Mean_Slack_Offer}
\end{figure}

\begin{figure}[htb]
    \centering
    \begin{subfigure}[b]{0.32\textwidth}
        \centering
        \includegraphics[width=\textwidth]{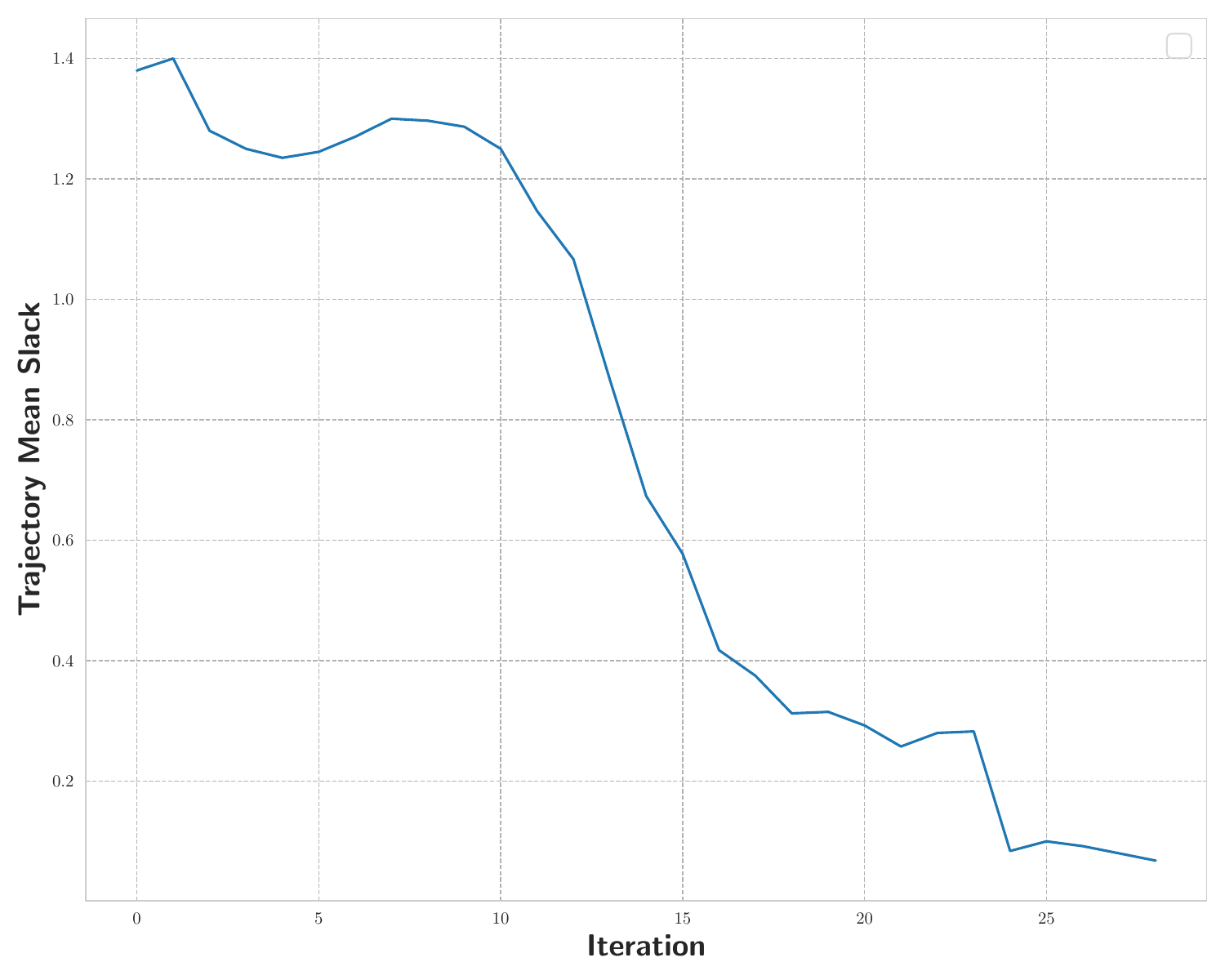}
        \caption{Bias factor = 0.3}
    \end{subfigure}
    \hfill
    \begin{subfigure}[b]{0.32\textwidth}
        \centering
        \includegraphics[width=\textwidth]{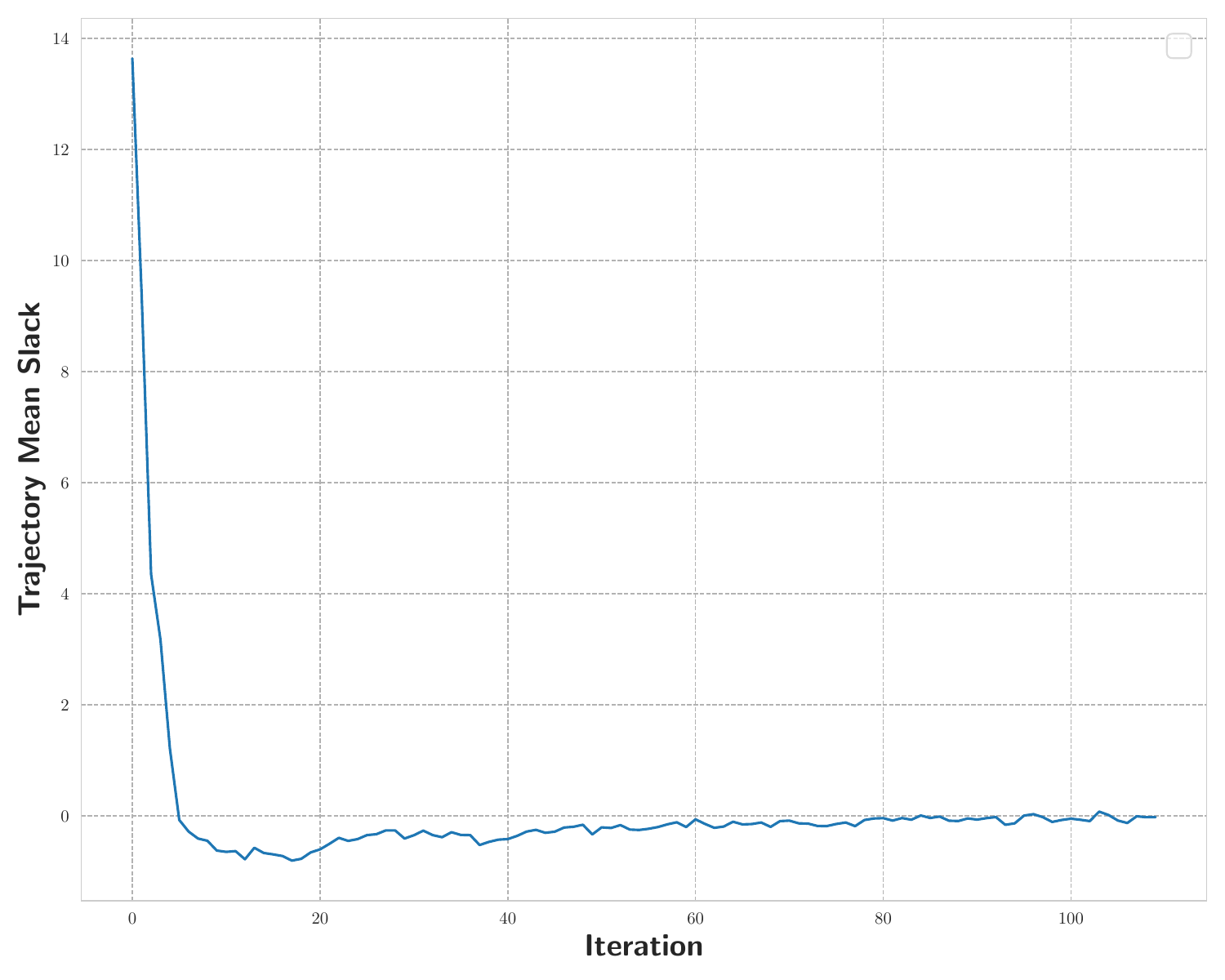}
        \caption{Bias factor = 0.7}
    \end{subfigure}
    \hfill
    \begin{subfigure}[b]{0.32\textwidth}
        \centering
        \includegraphics[width=\textwidth]{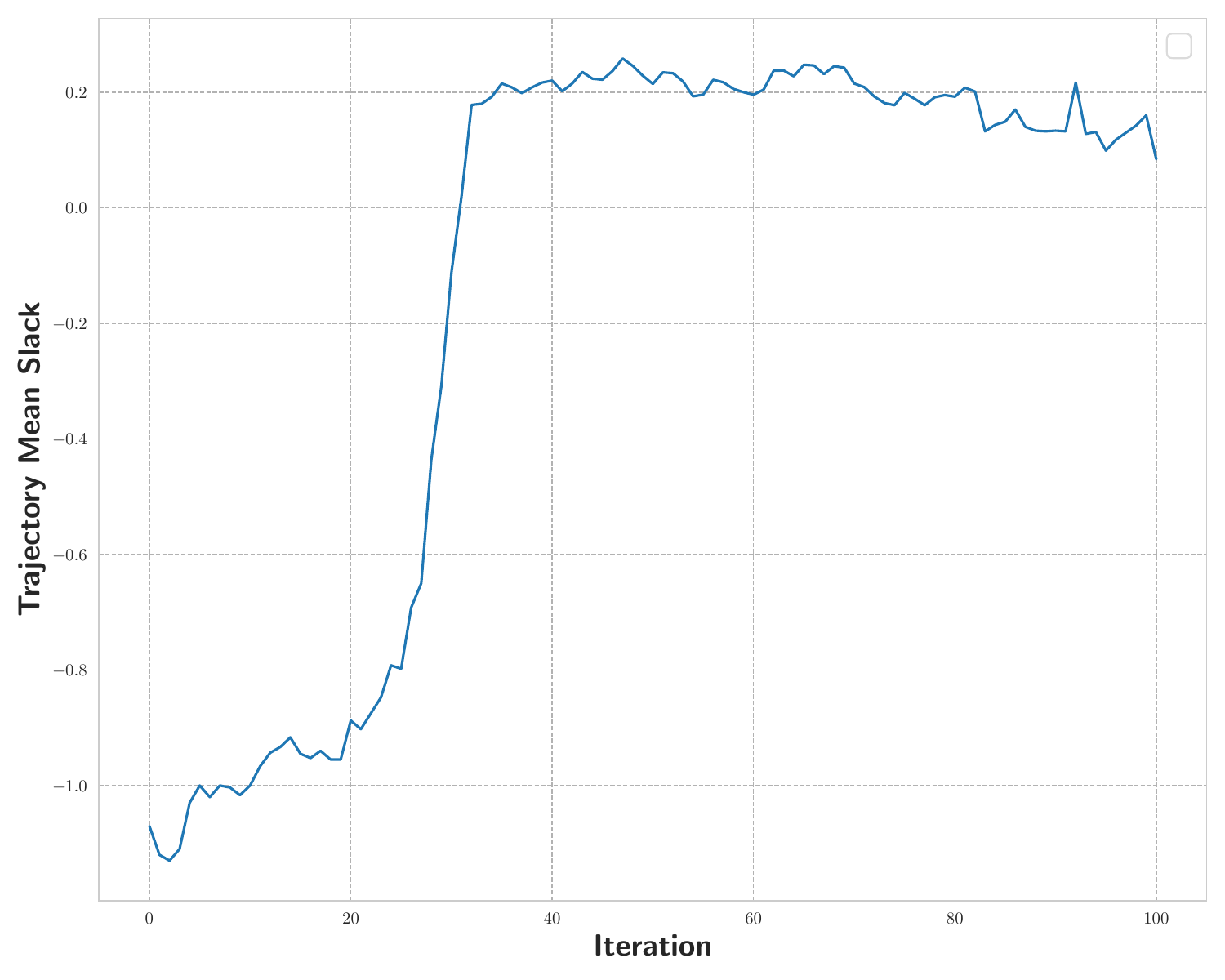}
        \caption{Bias factor = 0.9}
    \end{subfigure}
    
    \caption{(JMS simulation) Trajectory of mean slack for the parity at the onsite stage under capacity $\boldsymbol{k= 20}$}
    \label{fig:JMS_Trajectory_Mean_Slack_Onsite}
\end{figure}

\begin{figure}[htb]
    \centering
    \begin{subfigure}[b]{0.32\textwidth}
        \centering
        \includegraphics[width=\textwidth]{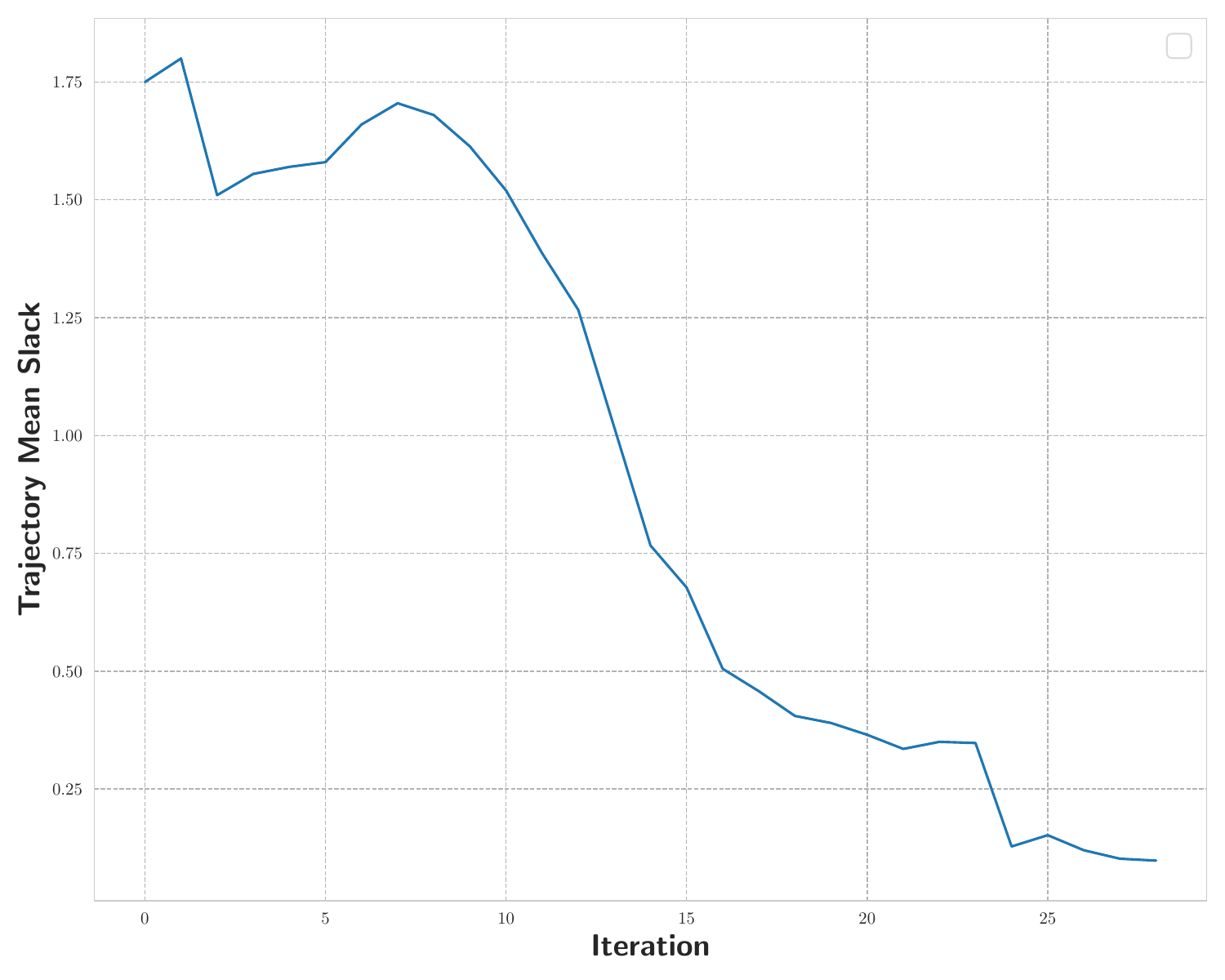}
        \caption{Bias factor = 0.3}
    \end{subfigure}
    \hfill
    \begin{subfigure}[b]{0.32\textwidth}
        \centering
        \includegraphics[width=\textwidth]{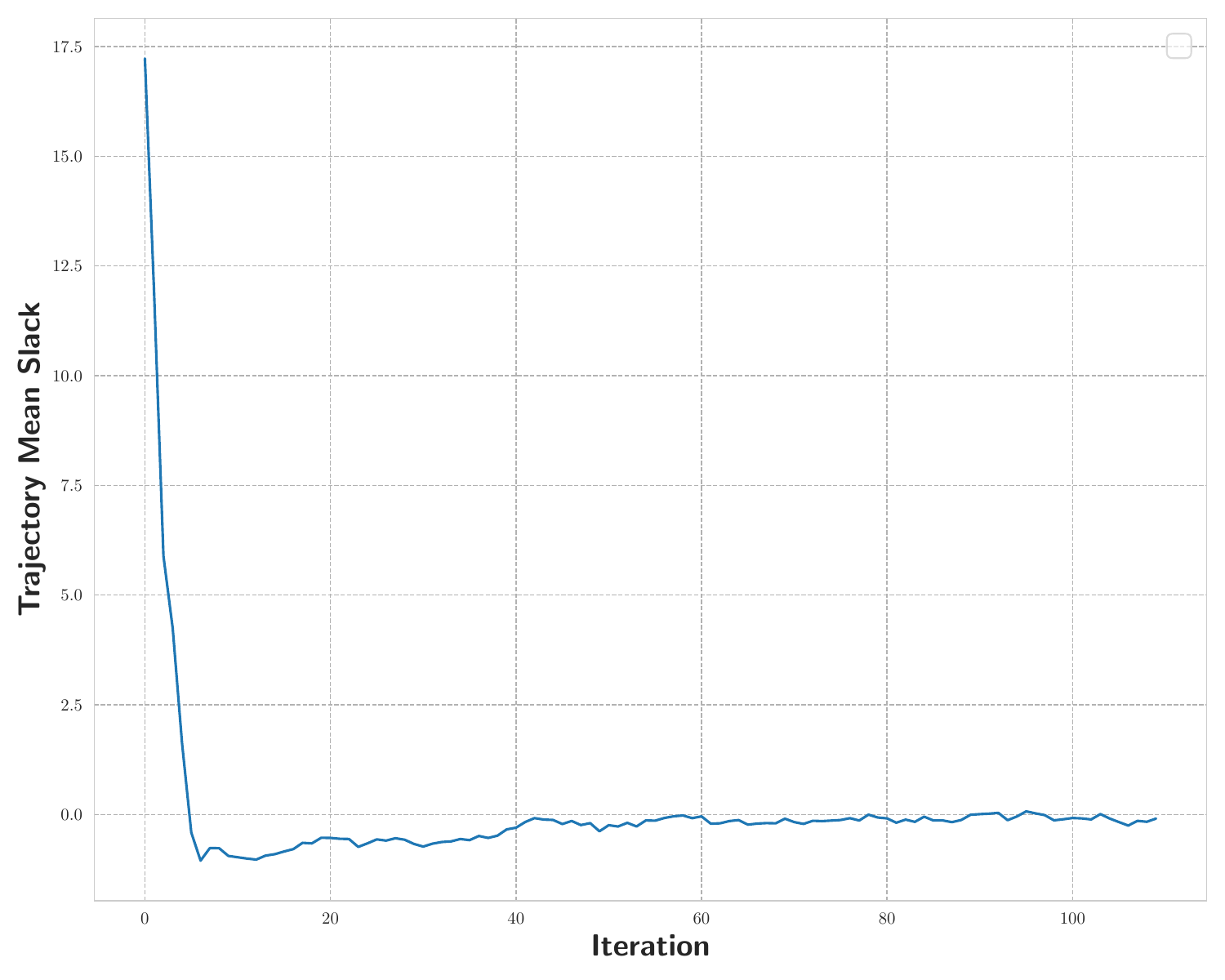}
        \caption{Bias factor = 0.7}
    \end{subfigure}
    \hfill
    \begin{subfigure}[b]{0.32\textwidth}
        \centering
        \includegraphics[width=\textwidth]{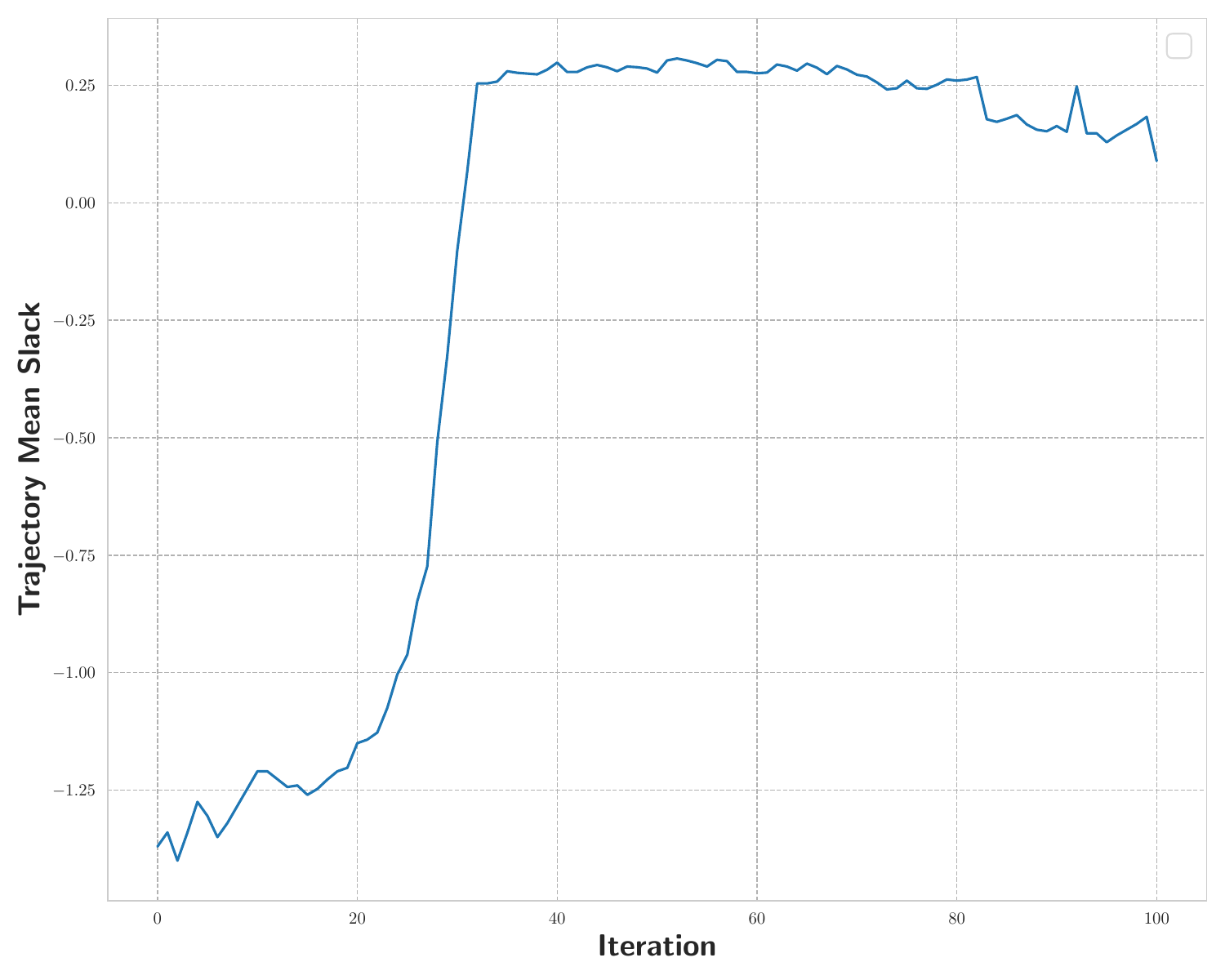}
        \caption{Bias factor = 0.9}
    \end{subfigure}
    
    \caption{(JMS simulation) Trajectory of mean slack for parity at phone interview stage, capacity $\boldsymbol{k= 20}$}
    \label{fig:JMS_Trajectory_Mean_Slack_Phone}
\end{figure}

\begin{figure}[htb]
    \centering
    \begin{subfigure}[b]{0.32\textwidth}
        \centering
        \includegraphics[width=\textwidth]{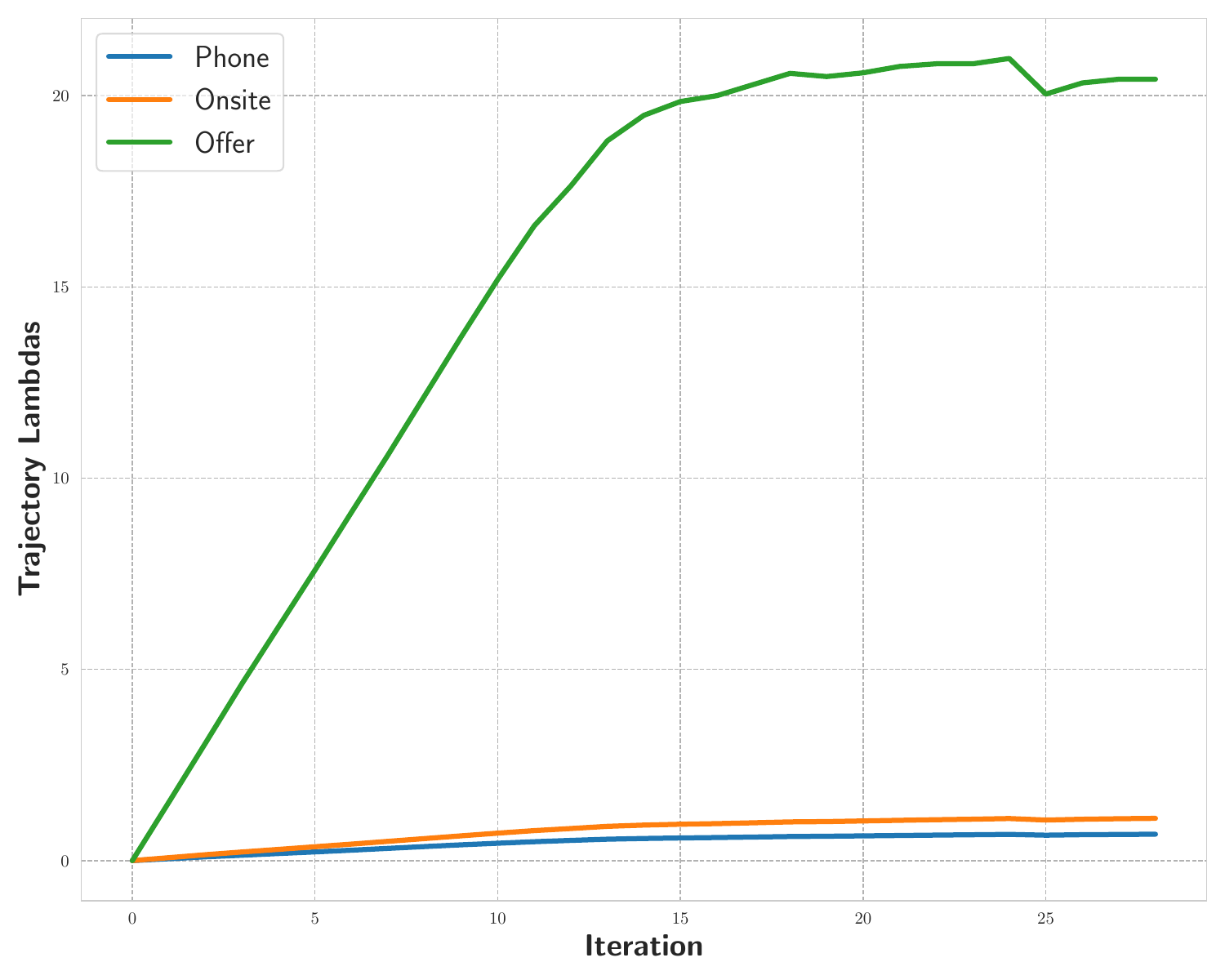}
        \caption{Bias factor = 0.3}
    \end{subfigure}
    \hfill
    \begin{subfigure}[b]{0.32\textwidth}
        \centering
        \includegraphics[width=\textwidth]{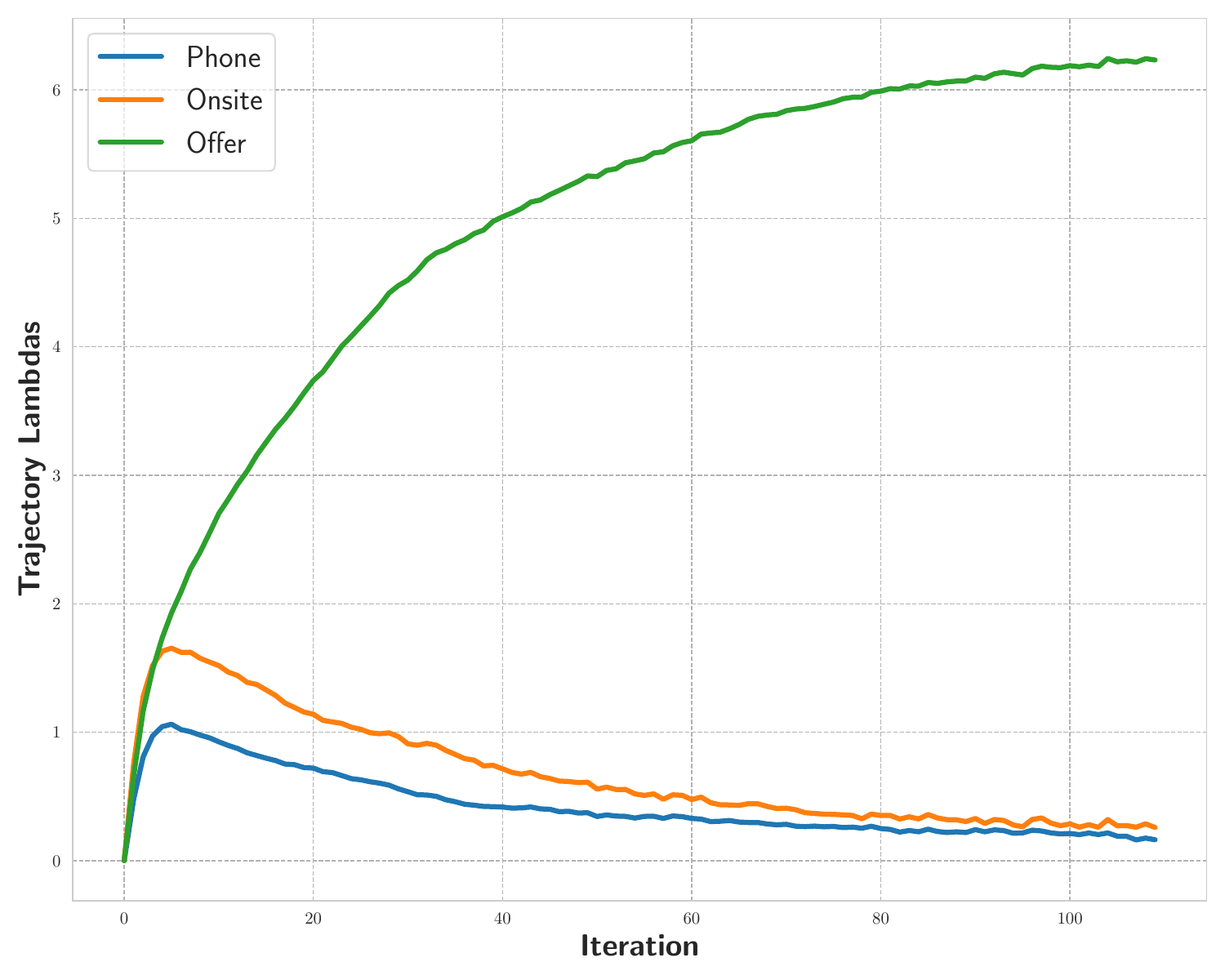}
        \caption{Bias factor = 0.7}
    \end{subfigure}
    \hfill
    \begin{subfigure}[b]{0.32\textwidth}
        \centering
        \includegraphics[width=\textwidth]{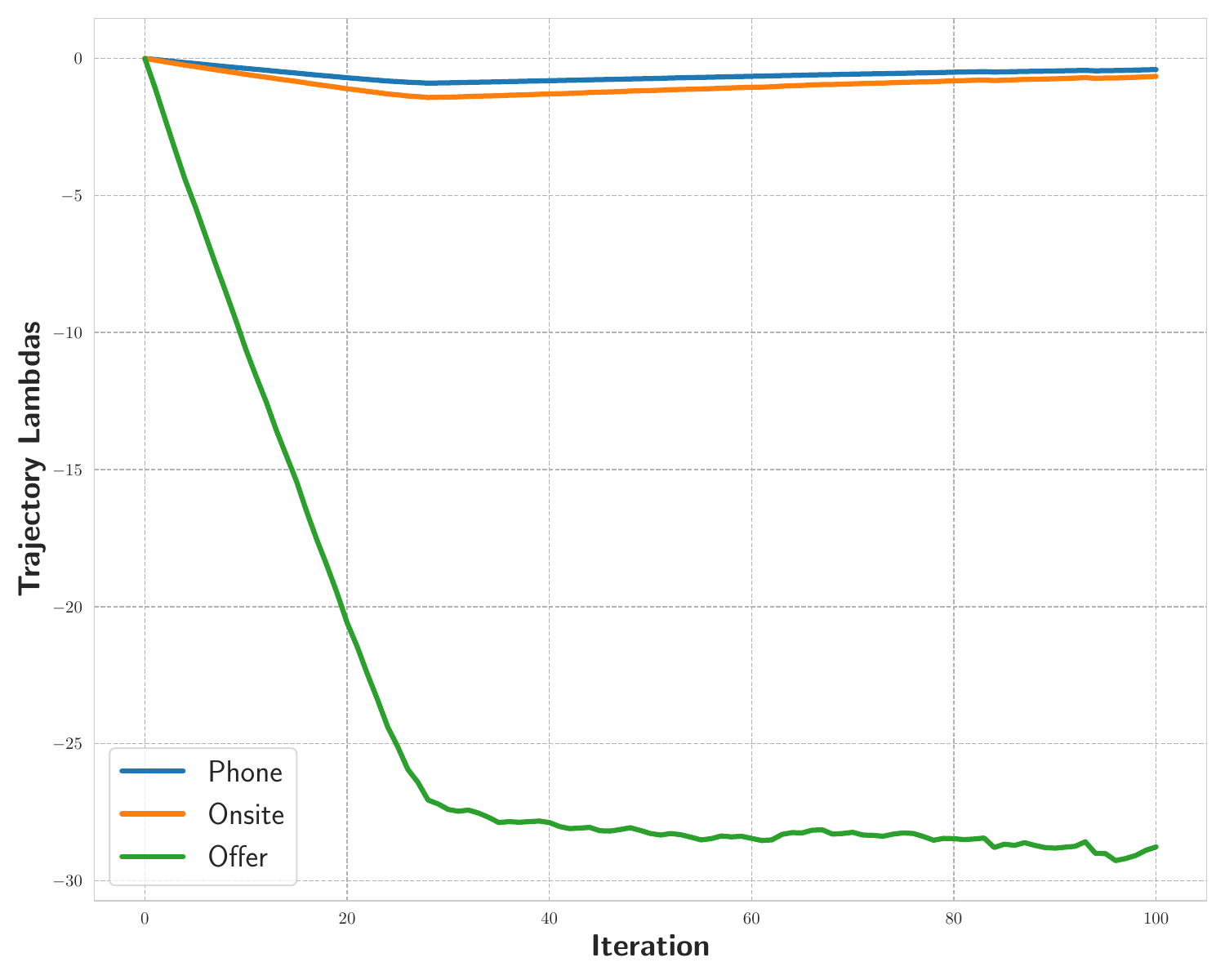}
        \caption{Bias factor = 0.9}
    \end{subfigure}
    
    \caption{Trajectory of dual adjustment $\boldsymbol{\lambda}$ under capacity $\boldsymbol{k= 20}$}
    \label{fig:JMS_Trajectory_lambdas}
\end{figure}

\noindent\textbf{Running times:}
\Cref{fig:JMS_Running Times} shows the running times of \Cref{alg:RAI}, at different bias levels.
Although the running times are considerably longer---compared to the results for \Cref{alg:const} when we had a single affine constraint and a simple single-stage search problem, as demonstrated in \Cref{fig-apx:v1_Running_Time}---they all take less than a minute for any given instance of the problem on the computer we used for our simulations (the same as the one we used for our earlier simulations).\footnote{\revcolor{We used a MacbookPro with 2.3 GHz Quad-core Intel Core i7 CPU, with 16GB of 3733 MHZ LDDR4X Memory for all of the simulations throughout the paper.}} This is especially important because we only need to run our algorithms once to find the (near-optimal) policy in any application, and after that the policy can be executed on each instantiation of the problem instance.
\begin{figure}[htb]
    \centering
    \begin{subfigure}[b]{0.32\textwidth}
        \centering
        \includegraphics[width=\textwidth]{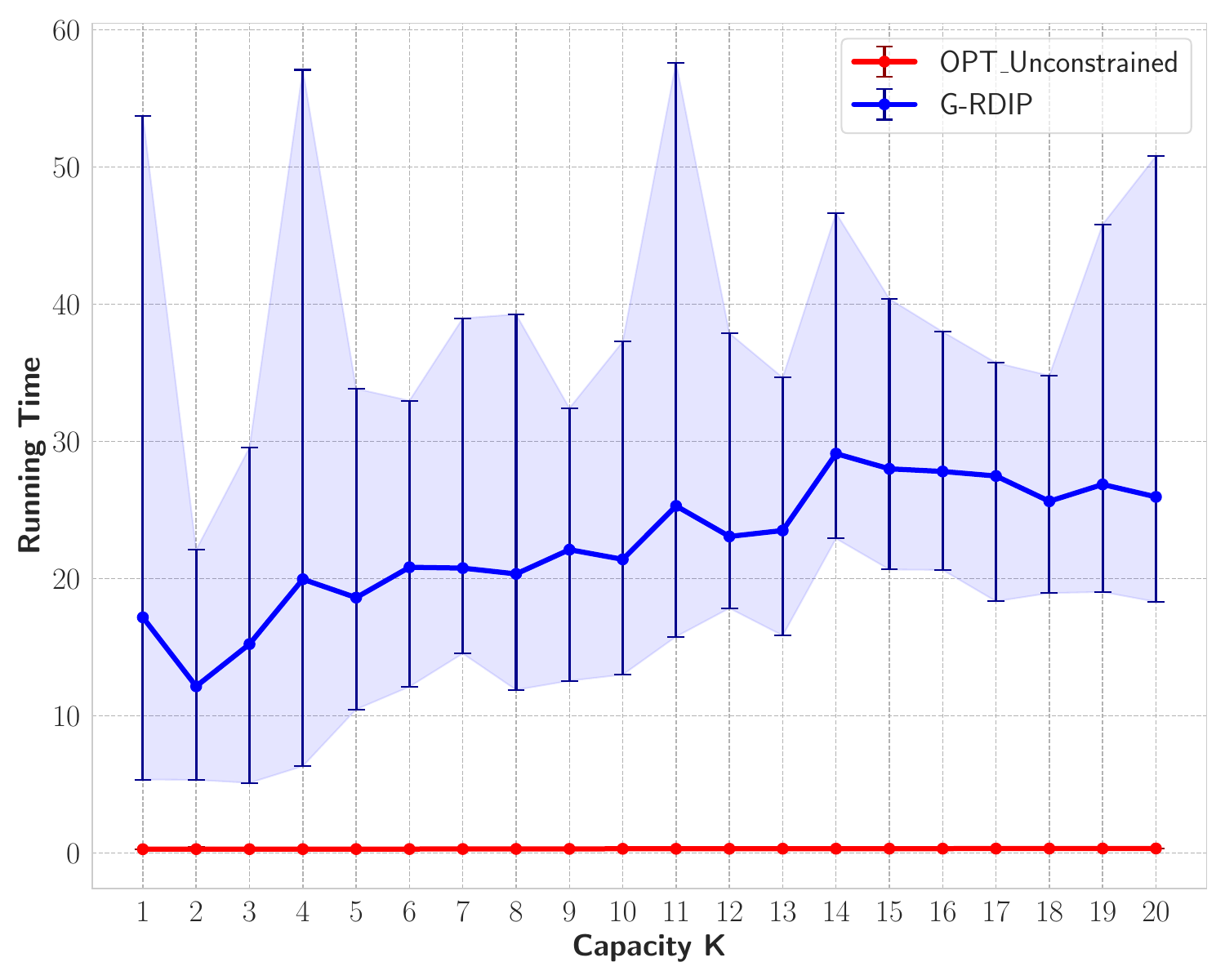}
        \caption{Bias factor = 0.3}
    \end{subfigure}
    \hfill
    \begin{subfigure}[b]{0.32\textwidth}
        \centering
        \includegraphics[width=\textwidth]{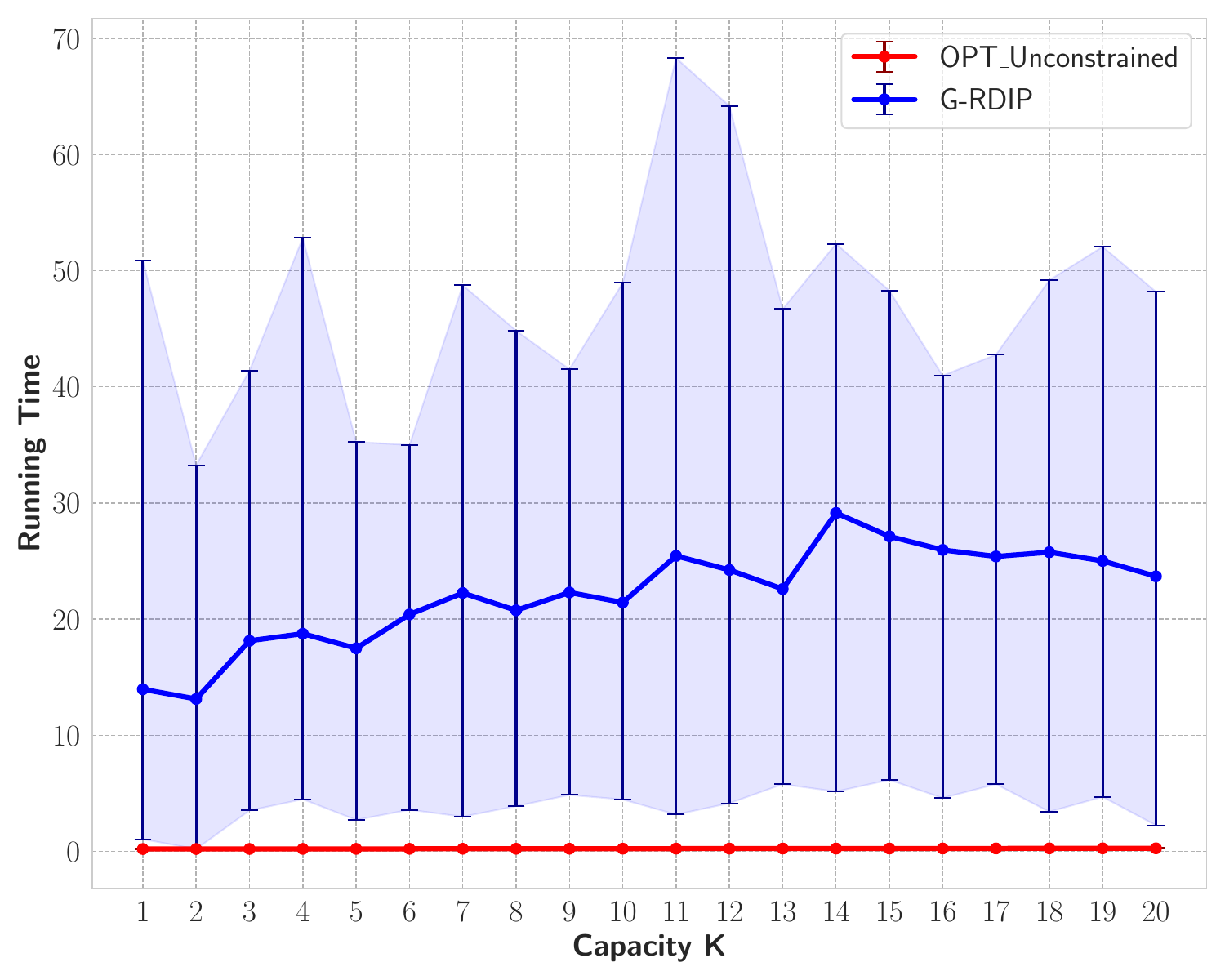}
        \caption{Bias factor = 0.7}
    \end{subfigure}
    \hfill
    \begin{subfigure}[b]{0.32\textwidth}
        \centering
        \includegraphics[width=\textwidth]{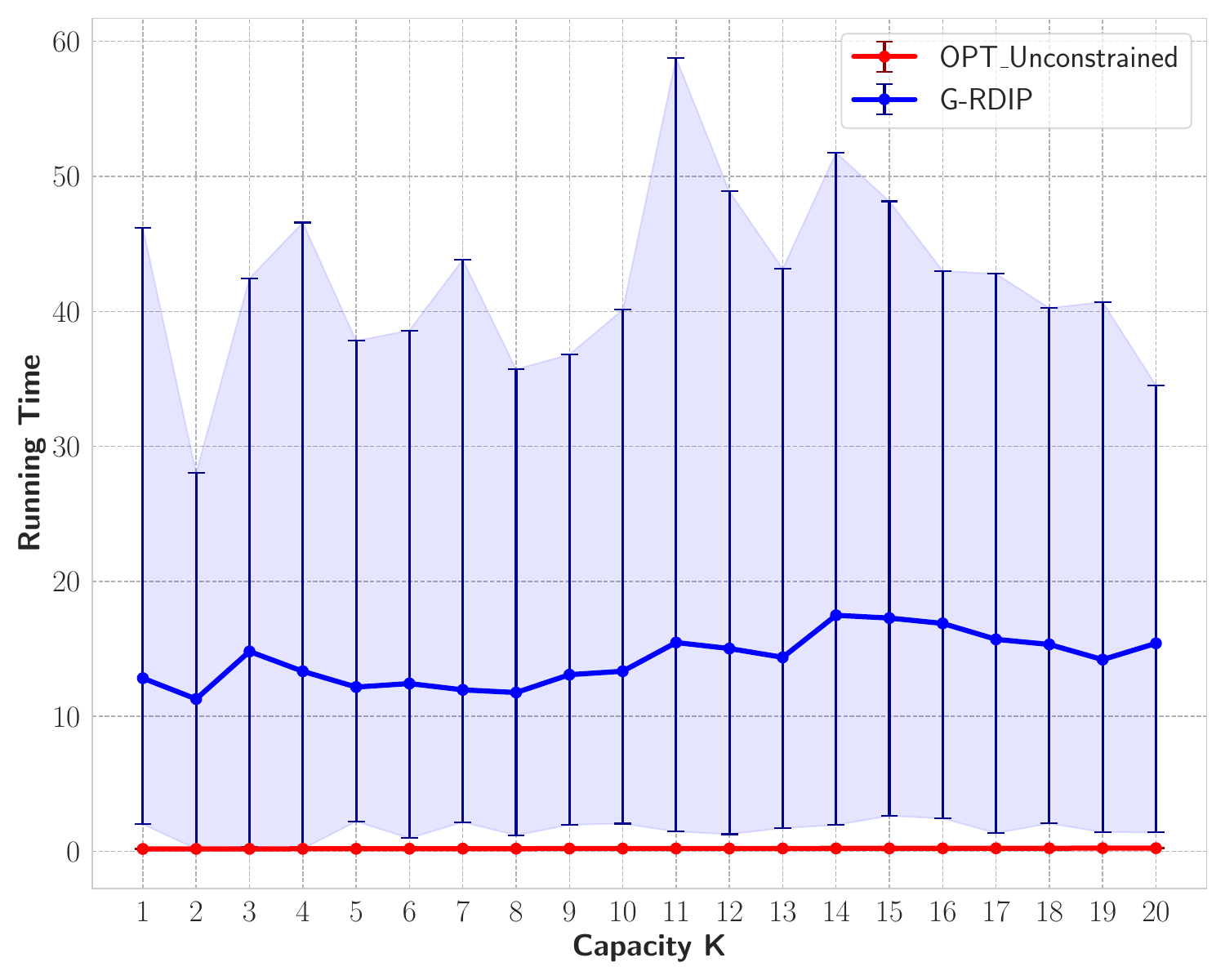}
        \caption{Bias factor = 1.0}
    \end{subfigure}
    
    \caption{The comparison of running times of near-optimal constrained (G-RDIP) and optimal unconstrained policies (in seconds).}
    \label{fig:JMS_Running Times}
\end{figure}

\end{document}